\UseRawInputEncoding
\pdfoutput=1
\documentclass[10pt]{article}
\usepackage{amsmath}
\usepackage{amssymb}
\usepackage{amsthm}
\usepackage{mathrsfs}
\usepackage{geometry}
\usepackage[colorlinks,linkcolor=blue,citecolor=blue]{hyperref}
\usepackage{xr}

\usepackage{microtype}
\usepackage{authblk}
\usepackage{graphicx}
\usepackage{epstopdf}
\usepackage{enumerate}
\usepackage{subfigure}
\usepackage{color}
\usepackage[all]{xy}

\usepackage{tikz}
\usetikzlibrary{matrix,positioning,arrows.meta,backgrounds,fit,shapes.geometric,calc,decorations.markings}

\allowdisplaybreaks

\ifdefined\final
\usepackage[disable]{todonotes}
\else
\usepackage[textsize=tiny]{todonotes}
\fi
\newtheorem{thm}{Theorem}[section]

\newtheorem{pro}{Proposition}[section]
\newtheorem{lem}{Lemma}[section]

\theoremstyle{definition}
\newtheorem{Ap}{Assumption}[section]
\newtheorem{remark}{Remark}[section]
\newtheorem{Def}{Definition}[section]

\newcommand{\mfh}{\mathfrak{h}}
\newcommand{\mfc}{\mathfrak{c}}
\newcommand{\mfa}{\mathfrak{a}}
\newcommand{\mfb}{\mathfrak{b}}

\newcommand{\mfm}{\mathfrak{m}}

\newcommand{\mfv}{\mathfrak{v}}
\newcommand{\mfC}{\mathfrak{C}}

\newcommand{\bbA}{\boldsymbol{A}}
\newcommand{\bbB}{\boldsymbol{B}}
\newcommand{\bbC}{\boldsymbol{C}}

\newcommand{\bbF}{\boldsymbol{F}}

\newcommand{\bbI}{\boldsymbol{I}}

\newcommand{\bbM}{\boldsymbol{M}}

\newcommand{\bbP}{\boldsymbol{P}}
\newcommand{\bbQ}{\boldsymbol{Q}}
\newcommand{\bbR}{\boldsymbol{R}}
\newcommand{\bbS}{\boldsymbol{S}}
\newcommand{\bbT}{\boldsymbol{T}}
\newcommand{\bbU}{\boldsymbol{U}}
\newcommand{\bbV}{\boldsymbol{V}}
\newcommand{\bbW}{\boldsymbol{W}}
\newcommand{\bbX}{\boldsymbol{X}}

\newcommand{\bba}{\boldsymbol{a}}
\newcommand{\bbb}{\boldsymbol{b}}
\newcommand{\bbc}{\boldsymbol{c}}

\newcommand{\bbe}{\boldsymbol{e}}
\newcommand{\bbf}{\boldsymbol{f}}
\newcommand{\bbg}{\boldsymbol{g}}
\newcommand{\bbh}{\boldsymbol{h}}

\newcommand{\bbm}{\boldsymbol{m}}

\newcommand{\bbq}{\boldsymbol{q}}
\newcommand{\bbr}{\boldsymbol{r}}

\newcommand{\bbt}{\boldsymbol{t}}
\newcommand{\bbu}{\boldsymbol{u}}
\newcommand{\bbv}{\boldsymbol{v}}
\newcommand{\bbw}{\boldsymbol{w}}
\newcommand{\bbx}{\boldsymbol{x}}
\newcommand{\bby}{\boldsymbol{y}}
\newcommand{\bbz}{\boldsymbol{z}}

\newcommand{\bbGa}{\boldsymbol{\Gamma}}
\newcommand{\bbLa}{\boldsymbol{\Lambda}}
\newcommand{\bbPhi}{\boldsymbol{\Phi}}
\newcommand{\bbPi}{\boldsymbol{\Pi}}

\newcommand{\bbTheta}{\boldsymbol{\Theta}}

\newcommand{\bbdel}{\boldsymbol{\delta}}
\newcommand{\bbve}{\boldsymbol{\varepsilon}}
\newcommand{\bbtau}{\boldsymbol{\tau}}

\newcommand{\mbC}{\mathbb{C}}
\newcommand{\mbE}{\mathbb{E}}
\newcommand{\mbP}{\mathbb{P}}

\newcommand{\mbN}{\mathbb{N}}

\newcommand{\mbR}{\mathbb{R}}
\newcommand{\mbS}{\mathbb{S}}

\newcommand{\mrO}{\mathrm{O}}
\newcommand{\mro}{\mathrm{o}}

\newcommand{\mcA}{\mathcal{A}}
\newcommand{\mcB}{\mathcal{B}}
\newcommand{\mcC}{\mathcal{C}}
\newcommand{\mcD}{\mathcal{D}}
\newcommand{\mcE}{\mathcal{E}}
\newcommand{\mcF}{\mathcal{F}}
\newcommand{\mcG}{\mathcal{G}}

\newcommand{\mcL}{\mathcal{L}}

\newcommand{\mcN}{\mathcal{N}}

\newcommand{\mcP}{\mathcal{P}}

\newcommand{\mcR}{\mathcal{R}}

\newcommand{\mcS}{\mathcal{S}}
\newcommand{\mcT}{\mathcal{T}}

\newcommand{\mcV}{\mathcal{V}}
\newcommand{\mcW}{\mathcal{U}}

\newcommand{\msB}{\mathscr{B}}
\newcommand{\msD}{\mathscr{D}}
\newcommand{\msO}{\mathscr{O}}

\newcommand{\tr}{\operatorname{Tr}}
\newcommand{\diag}{\operatorname{diag}}
\newcommand{\Cov}{\operatorname{Cov}}
\newcommand{\Var}{\operatorname{Var}}

\usepackage{xcolor}
\usepackage[T1]{fontenc}
\usepackage{babel}
\usepackage{a4wide}
\usepackage{pdfpages} 

\title{Alignment and matching tests for high-dimensional tensor signals by tensor contraction}
\author{Ruihan Liu\thanks{E-mail: rhliu@connect.hku.hk}\quad Zhenggang Wang\thanks{E-mail: zggwang@seu.edu.cn}\quad Jianfeng Yao\thanks{E-mail: jeffyao@cuhk.edu.cn}\\
{\small $^*$Department of Statistics and Actuarial Science, The University of Hong Kong.}\\
{\small $^\dagger$School of Statistics and Data Science, Southeast University.}\\
{\small $^\ddagger$School of Data Science, The Chinese University of Hong Kong (Shenzhen).}}

\date{}

\begin{document}
\maketitle
\begin{abstract}
  We consider two hypothesis testing problems for low-rank and
  high-dimensional tensor  signals, namely the tensor signal
  alignment and tensor signal matching problems. These problems are
  challenging due to the high dimension of tensors  and the lack of suitable
  test statistics. By exploiting a recent tensor contraction method, we propose and validate relevant test statistics using eigenvalues of a 
  data matrix resulting from the tensor contraction.
  The matrix entries exhibit long-range dependence, which  makes the
  analysis of the matrix challenging, involved,  and distinct from standard
  random matrix theory. Our approach provides a novel framework for addressing hypothesis testing problems in the context of high-dimensional tensor signals. 
\end{abstract}


\medskip


\section{Introduction}\label{sec:intro}

In the era of ``big data", the analysis of high-dimensional tensor
data has become increasingly important in various fields, including
genomics, economics, image analysis, and machine learning. High-order
tensor data often exhibit intrinsic low-rank structures
\cite{kolda2009tensor,udell2019big}. To capture these low-rank
structures, the
``signal plus noise"   tensor model has been widely adopted \cite{lesieur2017statistical,han2022tensor,huang2022power}.
Let $n_1,\ldots,n_d\in\mathbb{Z}^+$ denote $d$ dimension parameters, where $d\ge 3$, and let $N=n_1+\cdots+n_d$. The $d$-fold rank-$R$ spiked tensor model is defined as:
\begin{align}
  \bbT=\sum_{r=1}^R\beta_r\bbx^{(r,1)}\otimes\cdots\otimes\bbx^{(r,d)}+\frac{1}{\sqrt{N}}\bbX,
  \label{Main of Eq of general spiked tensor model}
\end{align}
where \(\beta_1\geq\cdots\geq\beta_R>0\) are the signal-to-noise
ratios (SNRs), 
\(\{\bbx^{(1,l)},\cdots,\bbx^{(R,l)}\}\) are mutually orthogonal unit
vectors in
\(\mbR^{n_l}\) for each \(1\leq l\leq d\) \cite{kolda2001orthogonal},
and
$\bbX=(X_{i_1\cdots i_d})_{n_1\times\cdots\times  n_d}\in\mbR^{n_1\times\cdots\times n_d}$ is a noise tensor with
independent and identically distributed (i.i.d.) entries, each having mean zero and unit variance.
Specifically, the rank-1 spiked tensor model \cite{richard2014statistical} is given by:
\begin{align}
\bbT=\beta\bbx^{(1)}\otimes\cdots\otimes\bbx^{(d)}+\frac{1}{\sqrt{N}}\bbX,\label{Main of Eq of spiky tensor model}
\end{align}
where $\beta>0$ is the single SNR of the model.

The primary focus of most of the existing literature is on recovering the signal vectors $\{\bbx^{(1,l)}, \ldots, \bbx^{(R,l)}\},~ 1 \leq l \leq d$ from the observed tensor $\bbT$, with a particular emphasis on the computational efficiency of recovery algorithms. In the case of the rank-1 model~\eqref{Main of Eq of spiky tensor model} with symmetric and i.i.d. Gaussian noise \(\bbX\), \cite{hillar2013most} showed that computing the maximum likelihood (ML) estimator of \(\beta\bbx^{(1)}\otimes\cdots\otimes\bbx^{(d)}\) is in general NP-hard, and \cite{arous2019landscape} provided a comprehensive discussion on the relationship between the computational complexity of the ML estimator and the value of the SNR $\beta$. To reduce the computational complexity, \cite{richard2014statistical} proposed  the use of  the power iteration method and approximate message passing (AMP) algorithms. These two methods  have been extensively investigated by \cite{lesieur2017statistical,chen2019phase,perry2020statistical,jagannath2020statistical,arous2020algorithmic} for AMP and by \cite{huang2022power} for power iteration. Moreover, \cite{richard2014statistical}  introduced the tensor unfolding method, which involves unfolding the tensor data $\bbT$ into matrices, enabling the recovery of signals through Principal Component Analysis (PCA).  \cite{arous2021long}  conducted a comprehensive study of  the tensor unfolding method for the general asymmetric model~\eqref{Main of Eq of spiky tensor model} under fairly general noise distribution assumptions.

However, when the SNRs fall below the phase transition threshold, these recovery methods often fail. In such cases, a more modest but achievable goal is to test the alignment of a signal in $\bbT$ with a given tensor $\bba^{(1)} \otimes \cdots \otimes \bba^{(d)}$, where ${\bba^{(l)}, 1 \leq l \leq d}$ are $d$ given directional unit vectors in $\mathbb{R}^{n_l}$. This leads to the following {\em tensor signal alignment test} between two hypotheses:
\begin{align}
  \begin{array}{l}
    H_0:\bba^{(l)}\perp\bbx^{(r,l)}\ \ {\rm for\ }1\leq l\leq d,\ 1\leq r\leq R.\\
    H_1:{\rm there\ exists\ at\ least\ one\ }1\leq l\leq d,\ 1\leq r\leq R\ {\rm such\ that\ \bba^{(l)}\not\perp\bbx^{(r,l)}}.
  \end{array}\label{Main of Eq of hypothesis test 1}
\end{align}
{Although the tensor signal alignment test appears more tractable than signal recovery, to the best of our knowledge, there is no rigorously justified procedure for addressing this problem. The main obstacle is the absence of a test statistic whose null distribution is tractable in the high-dimensional setting.}

{ In practice, the hypothesis test \eqref{Main of Eq of hypothesis test 1} arises naturally in tensor-based classification problems. A representative example is human action recognition from video data \cite{ji20123d,sun2015human}. Since a video inherently forms a 3D tensor (a spatial 2D field evolving along a temporal dimension), much of the existing literature (e.g., \cite{kim2008canonical,lu2008mpca,guha2011learning,yang2014sparse,yang2017tensor}) assumes that human actions admit a low-rank structure as in \eqref{Main of Eq of general spiked tensor model} and applies supervised tensor learning for classification \cite{tao2005supervised,camps2005kernel,lu2008mpca,kim2008canonical,guha2011learning,yang2014sparse}. This framework proceeds in two stages: first, low-rank signal structures are recovered from labeled training samples to obtain reference directions $\bba^{(l)}$; then, new samples are classified by testing whether their signals align with these learned references. This second stage corresponds precisely to the hypothesis test \eqref{Main of Eq of hypothesis test 1}. Beyond video analysis, the same tensor-based framework applies to other high-dimensional data. In neuroimaging, for instance, diagnosing brain disorders often relies on 3D functional magnetic resonance imaging (fMRI) data, where classification is performed by testing whether the low-rank structures of new scans align with those from labeled patient or control groups \cite{cichocki2011tensor,he2014dusk,sun2017store,he2017multi}.

Although deep neural networks have achieved strong empirical performance in tensor-based classification tasks \cite{makantasis2015deep,vakalopoulou2015building,sun2015human}, they are typically employed as ``black boxes'' \cite{benitez1997artificial} without explicit modeling assumptions such as \eqref{Main of Eq of general spiked tensor model}. This lack of interpretability can be a significant drawback in applications where understanding the underlying signal structure is essential \cite{makantasis2018tensor}. In contrast, classical statistical learning methods \cite{tan2012logistic,zhou2013tensor,sun2017store,makantasis2018tensor} offer interpretable models but may lack power when the signal-to-noise ratio is low. Our goal is to bridge this gap by developing a theoretically grounded test statistic for \eqref{Main of Eq of hypothesis test 1} that maintains good power even when the SNRs are relatively small.

Our strategy leverages random matrix theory by reducing the tensor to a lower-dimensional matrix whose spectral properties encode signal alignment. Specifically, we employ the tensor contraction operator $\bbPhi_d$ from \cite{Seddik2024203}, which maps a tensor $\bbT$ and reference directions $\{\bba^{(l)}\}$ to a symmetric 
$N \times N$ matrix by collapsing all but two dimensions along the given directions:}

\begin{align}
  &\bbPhi_d:\mbR^{n_1\times\cdots\times n_d}\times\mathbb{S}^{n_1-1}\times\cdots\times\mathbb{S}^{n_d-1}\longrightarrow\mbR^{N\times N},\notag\\
  &(\bbT,\boldsymbol{a}^{(1)},\cdots,\boldsymbol{a}^{(d)})\longmapsto
    \bbR=\left(\begin{array}{cccc}
            \boldsymbol{0}_{n_1\times n_1}&\bbT^{12}&\cdots&\bbT^{1d}\\(\bbT^{12})'&\boldsymbol{0}_{n_2\times n_2}&\cdots&\bbT^{2d}\\\vdots&\vdots&\ddots&\vdots\\(\bbT^{1d})'&(\bbT^{2d})'&\cdots&\boldsymbol{0}_{n_d\times n_d}
          \end{array}\right).\label{Main of Eq of tensor contraction}
\end{align}
Here, for a pair of indices  \(1\leq
j_1<j_2\leq d\),
$\bbT^{j_1j_2}$ is an $n_{j_1}\times n_{j_2}$ matrix, called the {\em  second order contraction
matrix of $\bbT$ along the directions
$\{\bba^{(j_1)},\bba^{(j_2)}\}$}, as introduced in \cite{lim2005singular}.
It is defined by: 
\begin{align}
  \bbT^{j_1j_2}=\Bigg[\sum_{i_j=1,j\neq j_1,j_2}^{n_j}T_{i_1\cdots
  i_d}\prod_{l=1,l\neq j_1,j_2}^da_{i_l}^{(l)}\Bigg]_{n_{j_1}\times
  n_{j_2}}. \label{Main of Eq of operator 1}
\end{align}
{  Since $\bbPhi_d$ is linear in $\bbT$, under \eqref{Main of Eq of general spiked tensor model},} we have 
\begin{align}
\bbR&=\bbPhi_d(\bbT,\bba^{(1)},\cdots,\bba^{(d)})\notag \\
    &=\sum_{r=1}^R\beta_r\bbPhi_d(\bbx^{(r,1)}\otimes\cdots\otimes\bbx^{(r,d)},\bba^{(1)},\cdots,\bba^{(d)})+\frac{1}{\sqrt{N}}\bbPhi_d(\boldsymbol{X},\bba^{(1)},\cdots,\bba^{(d)}),\notag\\
    & = \bbS + \bbM.  \label{Main of Eq of R and M}
\end{align}  {where $\bbS$ is the contracted signal matrix and $\bbM$ is the residual noise matrix. Under $H_0$, we have $\bbS = \mathbf{0}$ so that $\bbR = \bbM$, whereas under $H_1$, the signal component $\bbS \neq \mathbf{0}$ shifts $\bbR$ away from pure noise. Further, since the squared Frobenius norm $\|\bbS\|_2^2 =\sum_{i,j=1}^NS_{i,j}^2= \sum_{r=1}^R\beta_r^2\sum_{l=1}^d\langle\bbx^{(r,l)},\bba^{(l)}\rangle^2$ directly encodes the alignment terms in \eqref{Main of Eq of hypothesis test 1}, we use $\|\bbR\|_2^2$ as our test statistic. Moreover, $\|\bbR\|_2^2$ is a linear spectral statistic (LSS) of $\bbR$, enabling us to leverage random matrix theory for establishing its asymptotic distribution.
} Central limit theorems (CLT) for LSS of random matrices have been extensively studied; see \cite{bai2004clt,10.1214/09-AOP452,BJYZ09,PanZhou08,ZBY15a} for classical references.

When $d=2$, the tensor model \eqref{Main of Eq of general spiked tensor model} reduces to a spiked random matrix. In this case, the signal alignment test \eqref{Main of Eq of hypothesis test 1} can be viewed as a tensor extension of existing tests for detecting spikes along specified directions \cite{hallin2010optimal,silin2018bayesian,naumov2019bootstrap,silin2020hypothesis,bao2022statistical}.  However, for $d \geq 3$, a fundamental difference arises: the entries $\bbT^{j_1j_2}$ of the contracted matrix $\bbR$ become correlated. This substantially complicates the analysis of its eigenvalue distribution and the asymptotic behavior of $\widehat{T}_N^{(d)}$, necessitating the development of new techniques.

We first establish that the eigenvalue distribution of $\bbR$ has a limit $\nu$ when the $d$ dimensions $\{n_j\}$ grow to infinity at comparable rates. Based on this, we introduce the test statistic
\begin{align}
  \widehat{T}_N^{(d)}=\Vert\bbR\Vert_2^2-N\int_{-\infty}^{\infty}x^2\nu(dx).\label{Main of Eq of hat TN (d)}
\end{align}
We show that after proper centering and scaling,
\begin{align}
\frac{\widehat{T}_N^{(d)}-\xi_N^{(d)}}{\sigma_N^{(d)}}\overset{d}{\longrightarrow}\mcN(0,1) \quad \text{under } H_0,\label{Main of Eq of 1}
\end{align}
where $\xi_N^{(d)}$ and $\sigma_N^{(d)}$ are explicit parameters that can be computed numerically. Under the alternative $H_1$, a positive mean drift $\mcD^{(d)}/\sigma_N^{(d)}$ emerges:
\begin{align}
\frac{\widehat{T}_N^{(d)}-\xi_N^{(d)}}{\sigma_N^{(d)}}-\mcD^{(d)}/\sigma_N^{(d)}\overset{d}{\longrightarrow}\mcN(0,1);\label{Main of Eq of 2}
\end{align}
see \S\ref{sec of hypothesis} for details. The asymptotic normality in \eqref{Main of Eq of 1} enables test construction at any significance level $\alpha$, while \eqref{Main of Eq of 2} guarantees positive power depending on the magnitude of $\mcD^{(d)}/\sigma_N^{(d)}$.


The main contributions of this article are as follows.
\begin{enumerate}[(i)]
\item 
We conduct an in-depth analysis of the contracted data matrix $\bbR$, {whose entries display significant correlations and deviate from traditional random matrix models in which the elements of the noise matrix are typically assumed to be independent of one another,}
  including\\[1mm]
  (a) The characterization of its limiting spectral distribution (LSD) through a vector Dyson equation, along with entrywise behavior of the resolvent. \\[1mm]
  (b) The establishment of CLT for a broad class of its LSS.
  
\item We establish a rigorous procedure for  the tensor signal
  alignment test \eqref{Main of Eq of hypothesis test 1} by establishing
  the asymptotic normality of the test statistic  and deriving its power
  function under a general  alternative hypothesis.

\item { We also consider the setting where the prior information is given as tensor data rather than directional vectors. In this case, we test for signal matching between high-dimensional low-rank tensors against a reference tensor, requiring a mild condition on the reference signal strength. The procedure follows a two-step approach analogous to the tensor signal alignment test and is detailed in \S\ref{sec of generalized procedure}.}
\end{enumerate}

The contributions presented in this article are novel. One notable innovation is that our tensor signal model in \eqref{Main of Eq of general spiked tensor model} allows for non-Gaussian and non-symmetric signals. This sets our work apart from most existing literature on high-dimensional tensor data models, which typically assumes symmetry or Gaussianity for either the tensor signal, the tensor noise, or both.


The rest of this article is organized as follows. In $\S\ref{sec of LSD}$, we establish several important asymptotic spectral properties of the random matrix
$\bbR$, including its LSD, vector Dyson equation, and entrywise behaviors. For the sake of clarity, in $\S\ref{Sec of main d=3}$, we consider the case of 3-fold tensors and establish a CLT for linear spectral statistics (LSS) of the matrix $\bbR$. The corresponding CLT for the general case of $d$-fold tensors is  presented later in $\S\ref{Sec of general Main}$. In $\S\ref{sec of hypothesis}$, we establish new procedures
for testing tensor signal alignments and tensor signal matching. {A real data analysis is presented in \S\ref{Sec of real data}.} {Finally, numerical experiments evaluating the performance our CLT and two testing procedures introduced in $\S\ref{sec of hypothesis}$ and the proofs of all our results are included in the Supplementary Materials \cite{supplementary}.}

We end this section with some notations.
\begin{enumerate}[(i)]
\item Given \(z\in\mbC\), \(\Re(z)\) and \(\Im(z)\) are the real and imaginary part of \(z\), respectively.
\item We use a vector in \(\mbR^{n_1\times\cdots\times n_d}\) to represent the \(d\)-fold real tensor with size \(n_1\times\cdots\times n_d\).
\item Given a matrix \(A=[a_{i,j}]_{n\times n}\), \(\tr(A)=\sum_{i=1}^n a_{i,i}\) and \(A'\) denotes the transpose of \(A\) and \(\diag(A)\) is the diagonal matrix made with the main diagonal of \(A\). Moreover, \(\Vert A\Vert\) denotes the spectral norm of \(A\) and \(\Vert A\Vert_k=(\sum_{i,j}|a_{i,j}|^k)^{1/k}\) for \(k\in\mbN^+\). For any tensor or matrix, \(\bbT\in\mbR^{n_1\times\cdots\times n_d}\), \(\Vert\bbT\Vert_k=(\sum_{i_1\cdots i_d}|T_{i_1\cdots i_d}|^k)^{1/k}\).
\item Given an integrable random variable \(X\), \(X^c:=X-\mathbb{E}[X]\) denotes its centered version.
\item Given \(\eta>0\), define \(\mathbb{C}_{\eta}^+:=\{z\in\mathbb{C}:\Im(z)>\eta\}\) and \(\mathbb{C}^+:=\{z\in\mathbb{C}:\Im(z)>0\}\).
\item Given two matrices \(A,B\) of size \(m\times n\), when \(B_{ij}\neq0\) for all \(i,j\),
  \begin{align}
    \frac{A}{B}=[A_{ij}B_{ij}^{-1}]_{m\times n}.\label{Main of Eq of division of matrices}
  \end{align}
\item Let \(X=(X_n)\) and \(Y=(Y_n)\) be two sequences of nonnegative random variables. We say $Y$ stochastically dominates $X$ if for all (small) \(\epsilon>0\) and (large) \(D>0\),
  $$\mathbb{P}(X_n>n^{\epsilon}Y_n)\leq n^{-D}$$
  when \(n\geq n_0(\epsilon,D)\) is sufficiently large. This property is denoted by \(X\prec Y\) or \(X=\mrO_{\prec}(Y)\).
\end{enumerate}

\section{The limiting spectral distribution of the matrix \texorpdfstring{$\bbM$}{M}}\label{sec of LSD}

In this section, we explore several key spectral properties of the matrix \(\bbM\) (\ref{Main of Eq of R and M}). The matrix $\bbM$ is symmetric with mean-zero entries of variance $N^{-1}$, as in a standard Wigner matrix; however, unlike the Wigner case, the entries of $\bbM$ are generally correlated.
 For example, when \(d=3\), given \(\{i,j,k\}=\{1,2,3\}\), the \((i,j)\)-th block of \(\bbM\) is \(N^{-1/2}\bbX^{ij}\) by (\ref{Main of Eq of tensor contraction}), where the \((s_1,t_1)\)-th entry of \(\bbX^{ij}\) is \(X_{s_1,t_1}^{ij}=\sum_{i_k=1}^{n_k}X_{s_1t_1i_k}a_{i_k}^{(k)}\), we have
\begin{align}
  \Cov(X_{s_1,t_1}^{i j},X_{s_2,t_2}^{ik})=\delta_{s_1,s_2}a_{t_1}^{(j_1)}a_{t_2}^{(j_2)},\label{Main of Eq of cov of entries d=3}
\end{align}
where \(\delta_{s_1,s_2}\) is the Kronecker delta. Therefore, elements in the same row of \(\bbM\) are dependent. Likewise, when \(d\geq4\), it can be shown that
\begin{align}
  \Cov(X_{s_1,t_1}^{i_1 j_1},X_{s_2,t_2}^{i_2 j_2})=a_{s_1}^{(i_1)}a_{s_2}^{(i_2)}a_{t_1}^{(j_1)}a_{t_2}^{(j_2)},\label{Main of Eq of cov of entries d=4}
\end{align}
where \(1\leq i_1<j_1\leq d, 1\leq i_2<j_2\leq d\) such that \((i_1,j_1)\neq(i_2,j_2)\). In this case, the dependence is even more widespread.

{ Several recent works \cite{che2017universality,ajanki2019stability,alt2020correlated} have studied symmetric random matrices with correlated entries, typically under the assumption that correlations decay rapidly with the distance between indices. Our model $\bbM$, however, does not satisfy such a decay condition. Indeed, by \eqref{Main of Eq of cov of entries d=3} and \eqref{Main of Eq of cov of entries d=4}, if $\bba^{(j)} = (n_j^{-1/2}, \ldots, n_j^{-1/2})'$ for all $j$, then the entries of $\bbM$ exhibit long-range correlations.}

To study the LSD of $\bbM$, we start by examining its resolvent matrix:
\begin{align}
  \bbQ(z):=(\bbM-z\bbI_N)^{-1},\quad z\in\mbC^+.\label{Main of Eq of Q}
\end{align}
Similar to (\ref{Main of Eq of tensor contraction}), we split \(\bbQ(z)=[\bbQ^{ij}(z)]_{d\times d}\) into \(d\times d\) blocks where \(\bbQ^{ij}(z)\in\mbC^{n_i\times n_j}\). For \(i\)-th diagonal block \(\bbQ^{ii}(z), 1\leq i\leq d\), let
\begin{align}
  \left\{\begin{array}{ll}
           \rho_i(z):=N^{-1}\tr(\bbQ^{ii}(z)),&\rho(z):=\sum_{i=1}^d\rho_i(z),\\
           \mfm_i(z):=\mathbb{E}[\rho_i(z)],&\mfm(z):=\sum_{i=1}^d\mfm_i(z)=\mbE[\rho(z)].
         \end{array}\right.\label{Main of Eq of mi}
\end{align}
{
While the Stieltjes transform $\rho(z)$ characterizes the limiting spectral distribution, establishing our CLT requires finer control over the entire resolvent of $\bbM$. To achieve this, we analyze the vector Dyson equation induced by $\bbM$. Let $\bbm(z)=(\mfm_1(z),\ldots,\mfm_d(z))'$ where $\mfm_i(z)=\mbE[\rho_i(z)]$. We show that $\bbm(z)$ approximately satisfies
\begin{align}
-\frac{\mfc}{\bbm(z)}=z+\bbS_d\bbm(z)+\boldsymbol{\varepsilon}(z),\quad\bbS_d:=\boldsymbol{1}_{d\times d}-\bbI_d,\label{Main of Eq of approximation Dyson}
\end{align}
where the perturbation vector $\bbve(z)\to\mathbf{0}$ as $N\to\infty$. The limiting behavior is thus governed by the exact vector Dyson equation
\begin{align}
-\frac{\mfc}{\bbg(z)}=z+\bbS_d\bbg(z),\label{Main of Eq of MDE 3 order}
\end{align}
whose solution $\bbg(z)=(g_1(z),\ldots,g_d(z))'$ determines the LSD $\nu$ of $\bbM$. Beyond identifying $\nu$, our analysis establishes stability of \eqref{Main of Eq of MDE 3 order} under perturbations—a property essential for the CLT, as the inverse stability operator $\bbPi^{(3)}(z_1,z_2)^{-1}$ appears explicitly in the mean and variance expressions (Propositions~\ref{Main of Thm of Mean function} and~\ref{Main of Thm of Variance}). See Theorems~\ref{Main of Thm of Dyson} and~\ref{Main of Thm of approximation} for precise statements.}

{ Moreover, the resolvent $\mathbf{Q}(z)$ is approximated by a block-diagonal matrix whose $i$-th block is a diagonal matrix $\mathfrak{c}_i^{-1} g_i(z)\operatorname{I}_{n_i}$, so the vector Dyson equation solution $\mathbf{g}(z)$ directly determines the leading-order diagonal structure of the resolvent.}

The following assumptions are made for the general tensor model~\eqref{Main of Eq of general spiked tensor model}.

\begin{Ap}[Subexponential tails]\label{Main of Aq of general noise}
  The noise variables \(X_{i_1\cdots i_d}\) are i.i.d. with mean zero, unit variance and subexponential tails, that is, for some \(\theta>0\), 
  $$\limsup_{x\to\infty}e^{x^{\theta}}\mbP(|X_{i_1\cdots i_d}|\geq x)<\infty,$$
  
\end{Ap}

\begin{Ap}[High-dimensionality scheme]\label{Main of Aq of dimension}
  The tensor dimensions \(n_1,\cdots,n_d\) all tend to infinity in such a way that 
  $$\lim_{n_1,\cdots,n_d\to\infty}\frac{n_j}{n_1+\cdots+n_d}=\mfc_j\in(0,1),\quad 1\leq j\leq d.$$
  This limiting framework is denoted simply  as \(N:=n_1+\cdots+n_d\to\infty\) and we define
  \begin{align}
      \mfc=(\mfc_1,\cdots,\mfc_d)'.\label{Main of Eq of mfc}
  \end{align}
\end{Ap}\noindent

\begin{thm}\label{Main of Thm of Dyson}
  Under Assumptions {\rm \ref{Main of Aq of general noise}} and {\rm \ref{Main of Aq of dimension}}, we have
  \begin{enumerate}
  \item The vector Dyson equation {\rm (\ref{Main of Eq of MDE 3 order})} admits a unique analytic solution \(\bbg(z)=(g_1(z),\cdots,g_d(z))'\) on \(\mbC^+\).
  \item The function
    \begin{align}
      g(z)=\sum_{i=1}^d g_i(z).\label{Main of Eq of g(z)}
    \end{align}
    is the Stieltjes transform of a probability measure \(\nu\). Furthermore, the support of \(\nu\) is bounded, denoted by \([-\zeta,\zeta]\), where the edge $\zeta$ is nonnegative, finite, and determined by
    \begin{align}
      \zeta:=\inf\Big\{E>0:\lim_{\eta\downarrow0}\Im(g(E+{\rm i}\eta))=0\Big\}.\label{Main of Eq of support boundary}
    \end{align}
  \item The Stieltjes transform \(g(z)\) has a unique singularity at \(z=0\), i.e., \(\nu\) has a unique point mass at \(0\), if and only if \(\max\{\mfc_1,\cdots,\mfc_d\}\geq1/2\).
  \end{enumerate}
\end{thm}
The proof of Theorem \ref{Main of Thm of Dyson} is provided in \S\ref{Sec of Dyson} and \S\ref{Sec of LSD} of the supplement. The next theorem establishes the convergence of the LSD to the measure determined by  \eqref{Main of Eq of MDE 3 order}.
\begin{thm}\label{Main of Thm of approximation}
  Under Assumptions {\rm \ref{Main of Aq of general noise}} and {\rm \ref{Main of Aq of dimension}}, {let \(\mfv_d=2(d-1)\sum_{i=1}^d\sqrt{\mfc_i}\), we have
  $$\mbP(\Vert\bbM\Vert>\mfv_d+t)= \mro(N^{-l})$$
  for any \(t,l>0\). Next, for any \(\eta_0>0\), define the region
  \begin{align}
    \mcS_{\eta_0}:=\{z\in\mbC^+:|\Re(z)|,|\Im(z)|\leq\eta_0^{-1},{\rm dist}(z,[-\max\{\mfv_d,\zeta\},\max\{\mfv_d,\zeta\}])>\eta_0\},\label{Main of Eq of stable region}
  \end{align}}
  the vector function \(\bbve(z)\)
  \begin{align}
      \bbve(z):=\frac{\mfc}{\bbm(z)}+z+\bbS_d\bbm(z),\quad z\in\mcS_{\eta_0}\label{Main of Eq of perturbation}
  \end{align}
  satisfies that for any \(\omega\in(1/2-\delta,1/2)\), where $\delta>0$ is a sufficiently small constant
  $$\sup_{z\in\mcS_{\eta_0}}\Vert\bbve(z)\Vert_{\infty}=\mrO(\eta_0^{-11}N^{-2\omega}),$$
  where \(\bbm(z)\) and \(\bbS_d\) are defined in {\rm (\ref{Main of Eq of mi})} and {\rm (\ref{Main of Eq of approximation Dyson})}, then 
$$\lim_{N\to\infty}\sup_{z\in\mcS_{\eta_0}}\Vert\bbm(z)-\boldsymbol{g}(z)\Vert_{\infty}=0.$$
  Consequently, the measure \(\nu\) defined in Theorem {\rm \ref{Main of Thm of Dyson}} is the LSD of \(\bbM\).
\end{thm}
The proof of Theorem \ref{Main of Thm of approximation} is provided in \S\ref{sec of proof entrywise law d=3} of the supplement.


Theorem~\ref{Main of Thm of approximation} above expresses a crucial stability of the vector Dyson equation (\ref{Main of Eq of MDE 3 order}):   if a vector-valued function $\bbm(z)$ satisfies a perturbed version of the vector Dyson equation with a small perturbation term $\bbve(z)$  uniformly controlled over a given region $\mcS_{\eta_0}$ as in \eqref{Main of Eq of perturbation}, then the difference between $\bbm(z)$ and the solution $\boldsymbol{g}(z)$ of the original equation \eqref{Main of Eq of MDE 3 order} is also small uniformly over  $\mcS_{\eta_0}$. This stability immediately implies the asymptotic equivalence of  $\bbm(z)$ and $\boldsymbol{g}(z)$, and the measure $\nu$ associated with $\boldsymbol{g}(z)$ is the LSD of the matrix $\bbM$.


Furthermore, we provide the following approximation for general entries of the resolvent matrix $\bbQ(z)$.
\begin{thm}[Entrywise law]\label{Main of Thm of entrywise law d=3}
  Under Assumptions {\rm \ref{Main of Aq of general noise}} and {\rm \ref{Main of Aq of dimension}}, for any \(\eta_0>0\) and \(z\in\mcS_{\eta_0}\) in {\rm (\ref{Main of Eq of stable region})} and $\omega\in(1/2-\delta,1/2)$, where $\delta>0$ is a sufficiently small constant, and \(s,t\in\{1,\cdots,d\}\), 
we have
  \begin{align}\label{eq:entrywiselaw}
    \Bigg|Q_{i_si_t}^{st}(z)-\mathfrak{c}_s^{-1}g_s(z)\Bigg[\delta_{st}\delta_{i_si_t}+(a_{i_s}^{(s)})^2\sum_{k\neq s}^d(g(z)-g_s(z)-g_k(z))W_{sk}^{(d)}(z)\Bigg]\Bigg|\prec\mrO(\eta_0^{-21}N^{-\omega}),
  \end{align}
  where \(Q_{i_si_t}^{st}(z)\) is the \((i_s,i_t)\)-th entry of \(\bbQ^{st}\) and \(a_{i_s}^{(s)}\) is the \(i_s\)-th entry of \(\bba^{(s)}\), and $W_{sk}^{(d)}$ is defined later in \eqref{Main of Eq of bbW general}.	
\end{thm}
The proof of Theorem \ref{Main of Thm of entrywise law d=3} is provided in \S\ref{Sec of entrywise law} of the supplement.
\begin{remark}
  Note that the diagonal entries of the resolvent matrix $\bbQ(z)$ depend on both the vector $\bbg(z)$ and the given unit vectors $\bba^{(1)},\cdots,\bba^{(d)}$. The definition of localization is provided in \eqref{Main of Eq of localized}. For delocalized $\bba$, the entries will be close to $\mathfrak{c}_s^{-1}g_s(z)$, while localized $\bba$ will result in additional terms. For example, when $d=3$ and $\bba^{(1)}=(n_1^{-1/2},\cdots,n_1^{-1/2})'$ (a delocalized vector), then $(a_{i_1}^{(1)})^2=n_1^{-1}$ and we can show that $(a_{i_s}^{(s)})^2\big|\sum_{k\neq s}^d(g(z)-g_s(z)-g_k(z))W_{sk}^{(d)}(z)\big|\leq\mrO(\eta_0^{-2}N^{-1})$, then have
\begin{align*}
    \big|Q_{ii}^{11}(z)-\mathfrak{c}_1^{-1}g_1(z)\big|\prec\mrO(\eta_0^{-21}N^{-\omega}+\eta_0^{-3}N^{-1}),
  \end{align*}
for all $1\leq i\leq n_1$. In contrast, when $\bba^{(1)}=(1,0,\cdots,0)'$ (a localized vector), we have
 \begin{align*}
    \big|Q_{11}^{11}(z)-\mathfrak{c}_1^{-1}g_1(z)\big[1+g_2(z)W_{13}^{(3)}(z)+g_3(z)W_{12}^{(3)}(z)\big]\big|\prec\mrO(\eta_0^{-21}N^{-\omega}),
  \end{align*}
where an additional nonvanishing term $\mathfrak{c}_1^{-1}g_1(z)[g_2(z)W_{13}^{(3)}(z)+g_3(z)W_{12}^{(3)}(z)]$ appears.
\end{remark}
\begin{remark}\label{Rem of MLE}
  If the noise in  the rank-1 model (\ref{Main of Eq of spiky tensor model}) is Gaussian, \cite{Seddik2024203} showed that there exists a \(\beta_s>0\) such that when \(\beta\in(\beta_s,+\infty)\), the ML estimator of \(\beta\bbx^{(1)}\otimes\cdots\otimes\bbx^{(d)}\) in \eqref{Main of Eq of spiky tensor model},
  \begin{align}
		(\lambda_*,\bbu_*^{(1)},\cdots,\bbu_*^{(d)}):=\underset{\lambda\in\mathbb{R}^+,(\boldsymbol{u}^{(1)}\cdots\boldsymbol{u}^{(d)})\in\mathbb{S}^{n_1-1}\times\cdots\times\mathbb{S}^{n_d-1}}{\rm argmin}\Vert\boldsymbol{T}-\lambda\boldsymbol{u}^{(1)}\otimes\cdots\otimes\boldsymbol{u}^{(d)}\Vert_2^2,\label{Main of Eq of ML estimator}
	\end{align}
  satisfies that
  \begin{align}
    \left\{\begin{array}{l}
             \lambda_*\overset{a.s.}{\longrightarrow}\lambda^{\infty}(\beta)\\
             |\langle\bbx^{(i)},\bbu_*^{(i)}\rangle|\overset{a.s.}{\longrightarrow}q_i(\lambda^{\infty}(\beta))
           \end{array}\right.,\label{Main of Eq of limit tuple}
  \end{align}
  where 
  $$q_i(z):=\Bigg(\frac{\alpha_i(z)^{d-3}}{\prod_{j\neq i}\alpha_j(z)}\Bigg)^{\frac{1}{2d-4}}\quad{\rm and}\quad\alpha_i(z):=\frac{\beta}{z+g(z)-g_i(z)},$$
  and \(\lambda^{\infty}(\beta)\) satisfies \(f(\lambda^{\infty}(\beta),\beta)=0\), where \(f(z,\beta)=z+g(z)-\beta\prod_{i=1}^d q_i(z)\), and  \(\lambda^{\infty}(\beta)\) is a constant on \(\beta\in[0,\beta_s]\). This implies that when \(\beta<\beta_s\), no inference about \(\beta\) is possible based on \(\lambda_*\). It can be shown that (\ref{Main of Eq of limit tuple}) also holds for general non-Gaussian noises satisfying Assumption \ref{Main of Aq of general noise} by employing the same techniques used in \S\ref{Sec of entrywise law} of the supplementary document to analyze the ML estimator, though we do not pursue it in detail in this paper.
\end{remark}

	\section{CLT for linear spectral statistics of \texorpdfstring{$\bbM$}{M} when \texorpdfstring{$d=3$}{d=3}}\label{Sec of main d=3}

{ Under $H_0$ in \eqref{Main of Eq of hypothesis test 1}, the test statistic $\widehat{T}_N^{(d)}$ in \eqref{Main of Eq of hat TN (d)} reduces to a linear spectral statistic of $\bbM$. We now establish CLTs for general linear spectral statistics, focusing on $d=3$ for clarity; the general case appears in \S~\ref{Sec of general Main}.}


We first present our main results in \S \ref{sec of CLT LSS d=3}, and then provide the explicit formulas for the asymptotic mean and variance, which are relatively complex and tedious, in \S \ref{sec of mean variance}. Finally, a brief outline of the proofs is presented in \S \ref{sec of CLT}.
	
	Now consider \(d=3\) with dimensions \((n_1,n_2,n_3)\) satisfying Assumption \ref{Main of Aq of dimension}. Let \(\bba^{(1)}\in\mbR^{n_1},\bba^{(2)}\in\mbR^{n_2},\bba^{(3)}\in\mbR^{n_3}\) be three deterministic unit vectors, and \(\bbX=[X_{i_1i_2i_3}]\in\mbR^{n_1\times n_2\times n_3}\) be a random tensor satisfying Assumption \ref{Main of Aq of general noise}. Denote
	{\small \begin{align*}
		&\bbM=\frac{1}{\sqrt{N}}\boldsymbol{\Phi}_3(\boldsymbol{X},\bba^{(1)},\bba^{(2)},\bba^{(3)}),\quad\bbQ(z)=(\bbM-z\bbI_N)^{-1}=\left(\begin{array}{ccc}
			\bbQ^{11}(z)&\bbQ^{12}(z)&\bbQ^{13}(z)\\(\bbQ^{12}(z))'&\bbQ^{22}(z)&\bbQ^{23}(z)\\(\bbQ^{13}(z))'&(\bbQ^{23}(z))'&\bbQ^{33}(z)
		\end{array}\right),
	\end{align*}}\noindent
 where \(z\in\mbC^+\) and \(N=n_1+n_2+n_3\). As stated in Theorem \ref{Main of Thm of Dyson}, the support of the LSD $\nu$ is $[-\zeta,\zeta]$, defined in (\ref{Main of Eq of support boundary}), and $\Vert\bbM\Vert$ is stochastically bounded by $\mfv_3=4(\sqrt{\mfc_1}+\sqrt{\mfc_2}+\sqrt{\mfc_3})$. Define $v_B^{(3)}:=\max\{\zeta,\mfv_3\}$ and consider the class of functions
    \begin{align}
        &\mathfrak{F}_3:=\big\{f(z):f\ {\rm is\ analytic\ on\ an\ open\ set\ containing\ }\big[-v_B^{(3)},v_B^{(3)}\big]\big\}.\label{Main of Eq of mcU}
    \end{align}
    For \(f\in\mathfrak{F}_3\), consider an LSS of \(\bbM\) of the form:
    \begin{align}
		\mcL_{\bbM}(f):=\frac{1}{N}\sum_{l=1}^N f(\lambda_l)=\int_{\mbR}f(x)\nu_N(dx),\label{Main of Eq of LSS}
	\end{align}
    where \(\lambda_1,\cdots,\lambda_N\) are the eigenvalues of \(\bbM\) and \(\nu_N=N^{-1}\sum_{j=1}^N\delta_{\lambda_j}\) is the empirical spectral distribution (ESD) of \(\bbM\). By Theorem \ref{Main of Thm of approximation}, \(\nu_N\) converges to \(\nu\) almost surely, we thus consider
	\begin{align}
		G_N(f):=N\int_{-\infty}^{\infty}f(x)(\nu_N(dx)-\nu(dx))=N\Bigg(\mcL_{\bbM}(f)-\int_{-\infty}^{\infty}f(x)\nu(dx)\Bigg).\label{Main of Eq of Gf d=3}
	\end{align}
	Our main goal is to derive the asymptotic distribution of \(G_N(f)\).
	\subsection{Main results}\label{sec of CLT LSS d=3}
	Before presenting our main theorem, we introduce some auxiliary notation. First, we define the third and fourth cumulants of \(X_{i_1\cdots i_d}\) as follows:
\begin{equation}
    \begin{aligned}
\kappa_3:=\mbE[X_{i_1\cdots i_d}^3],\quad{\rm and}\quad\kappa_4:=\mbE[X_{i_1\cdots i_d}^4]-3.\label{Main of Eq of 3 4 cumulant}
\end{aligned}
\end{equation}	
	Given \(k\in\{1,\cdots,d\}\), define 
	\begin{align}
		\mfb_k^{(1)}:=\frac{1}{\sqrt{N}}\sum_{i_k=1}^{n_k}a_{i_k}^{(k)}.\label{Main of Eq of mfb}
	\end{align}
	Let $k_1, k_2, \dots, k_l$ be distinct integers in $\{1, \dots, d\}$. Define for \(r\in\mbN,r\geq2\)
	\begin{align}
		\mcA_{i_1\cdots i_d}^{(k_1,\cdots,k_l)}:=\prod_{j\neq k_1\cdots k_l}a_{i_j}^{(j)},\quad\mcB_{(r)}^{(k_1,\cdots,k_l)}:=\sum_{i_j=1,j\neq k_1\cdots k_l}^{n_j}(\mcA_{i_1\cdots i_d}^{(k_1,\cdots,k_l)})^r,\label{Main of Eq of mcB}
	\end{align}
	Moreover, we say the vector \(\bba^{(j)}\) is \emph{delocalized} if
	\begin{align}
		\lim_{n_j\to\infty}\Vert\bba^{(j)}\Vert_{\infty}=\lim_{n_j\to\infty}\max_{1\leq i_j\leq n_j}|a_{i_j}^{(j)}|=0.\label{Main of Eq of localized}
	\end{align}
	Otherwise, \(\bba^{(j)}\) is \emph{localized}.
	\begin{remark}\label{Rem of vectors}
		The purpose of defining \(\mfb_k^{(1)}\) and \(\mcB_{(4)}^{(k_1,k_2)}\) is that they will appear in the asymptotic mean and variance of the CLTs in the forthcoming Propositions \ref{Main of Thm of Mean function} and \ref{Main of Thm of Variance}. For example, when \(d=3\), we know that \(\mfb_k^{(1)}=N^{-1/2}\sum_{i_k=1}^{n_k}a_{i_k}^{(k)}\in[-\mfc_k,\mfc_k]\) and \(\mcB_{(4)}^{(1,2)}=\Vert\bba^{(3)}\Vert_4^4\in[0,1]\) due to \(\Vert\bba^{(k)}\Vert_2=1\). When \(d=4\), we have \(\mcB_{(4)}^{(1,2)}=\Vert\bba^{(3)}\Vert_4^4\times\Vert\bba^{(4)}\Vert_4^4\in[0,1]\). Notably, by (\ref{Main of Eq of localized}), if all \(\bba^{(l)}\) are delocalized, then \(\lim_{n_l\to\infty}\Vert\bba^{(l)}\Vert_4=0\), implying that all \(\lim_{N\to\infty}\mcB_{(4)}^{(k_1,k_2)}=0\).
	\end{remark}
	\begin{thm}\label{Main of Thm of CLT LSS d=3}
		Under Assumptions  \ref{Main of Aq of general noise} and \ref{Main of Aq of dimension}  with \(d=3\), for any \(f\in\mathfrak{F}_3\) in {\rm (\ref{Main of Eq of mcU})} and deterministic unit vectors \(\bba^{(1)}\in\mbR^{n_1},\bba^{(2)}\in\mbR^{n_2},\bba^{(3)}\in\mbR^{n_3}\), we have
		\begin{align*}
			\frac{G_N(f)-\xi_N^{(3)}}{\sigma_N^{(3)}}\overset{d}{\longrightarrow}\mcN(0,1),
		\end{align*}
		where
		\begin{align}
			\xi_N^{(3)}:&=-\frac{1}{2\pi{\rm i}}\oint_{\mfC_1}f(z)\mu_N^{(3)}(z;\kappa_3,\kappa_4,\bba^{(1)},\bba^{(2)},\bba^{(3)})dz,\label{Main of Eq of LSS mean d=3}\\
			(\sigma_N^{(3)})^2:&=-\frac{1}{4\pi^2}\oint_{\mfC_1}\oint_{\mfC_2}f(z_1)f(z_2)\mcC_N^{(3)}(z_1,z_2;\kappa_4,\bba^{(1)},\bba^{(2)},\bba^{(3)})dz_1dz_2,\label{Main of Eq of LSS variance d=3}
		\end{align}
		where \(\mfC_{1}\) and  \(\mfC_{2}\) are two disjoint rectangular contours with vertices \(\pm E_{1}\pm{\rm i}\eta_{1}\) and \(\pm E_{2}\pm{\rm i}\eta_{2}\), respectively, such that \(E_{1},E_{2}\geq v_B^{(3)}+t\), where \(t>0\) is a fixed constant and \(\eta_{1},\eta_{2}>0\). The mean function $\mu_N^{(3)}(z;\kappa_3,\kappa_4,\bba^{(1)},\bba^{(2)},\bba^{(3)})$ and the variance function $\mcC_N^{(3)}(z_1,z_2;\kappa_4,\bba^{(1)},\bba^{(2)},\bba^{(3)})$ are defined later in \eqref{Main of Eq of limiting mean} and \eqref{Main of Eq of limiting variance}, respectively.
	\end{thm}

 The proof of Theorem \ref{Main of Thm of CLT LSS d=3} is provided in \S\ref{Sec of mean and covariance} and \S\ref{Sec of CLT} of the supplement. For notational simplicity, we denote $\mu_N^{(3)}(z;\kappa_3,\kappa_4,\bba^{(1)},\bba^{(2)},\bba^{(3)})$ and $\mcC_N^{(3)}(z_1,z_2;\kappa_4,\bba^{(1)},\bba^{(2)},\bba^{(3)})$ by $\mu_N^{(3)}(z)$ and $\mcC_N^{(3)}(z_1,z_2)$, respectively. 

	\begin{remark}\label{Rem of LSS numerical}
		 {Despite their apparent complexity, the centering constant $\xi_N^{(3)}$ and variance $(\sigma_N^{(3)})^2$ in Theorem~\ref{Main of Thm of CLT LSS d=3} are readily computable. Both reduce to contour integrals of $\mu_N^{(3)}(z)$ and $\mcC_N^{(3)}(z_1,z_2)$, which are explicit functions of the Dyson equation solution $\bbg(z)$. Standard numerical quadrature yields accurate approximations; see \S\ref{sec of Existence and uniqueness Dyson equation} in the supplement for illustrations. In the balanced case $\mfc_1 = \mfc_2 = \mfc_3 = 1/3$, closed-form expressions are available in Remark \ref{Rem of close form} of the supplement.}
	\end{remark}
	\subsection{Mean and variance functions \texorpdfstring{\(\mu_N^{(3)}\)}{mu} and \texorpdfstring{\(\mcC_N^{(3)}\)}{sigma}}\label{sec of mean variance}
 {  The mean and variance functions \eqref{Main of Eq of LSS mean d=3} and \eqref{Main of Eq of LSS variance d=3} involve several auxiliary quantities that arise naturally from the resolvent analysis. We now define these quantities, which can be computed explicitly from the solution $\mathbf{g}(z)$ of the vector Dyson equation.}
    \begin{enumerate}
        \item For any \(z\in\mbC^+\), let 
		\begin{align*}
			\bbGa^{(3)}(z):=(z+g(z))\bbI_3-\diag(\bbg(z))+g(z)\bbS_3-\diag(\bbg(z))\bbS_3-\bbS_3\diag(\bbg(z)).
		\end{align*}
		And define
		\begin{align}
			\bbW^{(3)}(z):=-\bbGa^{(3)}(z)^{-1}.\label{Main of Eq of W limit d=3}
		\end{align}
		\item For any \(z_1,z_2\in\mbC^+\), let
		\begin{align}
			&\bbPi^{(3)}(z_1,z_2):=\boldsymbol{I}_3-{\rm diag}(\mfc^{-1}\circ\boldsymbol{g}(z_1)\circ\bbg(z_2))\boldsymbol{S}_3,\label{Main of Eq of invertible 2}
		\end{align}
		then define
		\begin{align}
			\bbV^{(3)}(z_1,z_2):=\bbPi^{(3)}(z_1,z_2)^{-1}\diag(\mfc^{-1}\circ\bbg(z_1)\circ\bbg(z_2)).\label{Main of Eq of bbV limit d=3}
		\end{align}
    \end{enumerate}
	\begin{pro}[Mean function \(\mu_N^{(3)}(z)\) for \(d=3\)]\label{Main of Thm of Mean function}
		Under Assumptions {\rm \ref{Main of Aq of general noise}} and {\rm \ref{Main of Aq of dimension}}, for any \(\eta_0>0\) and \(z\in\mcS_{\eta_0}\) in \eqref{Main of Eq of stable region}, let 
		$$\overrightarrow{M}_N^{(3)}(z)=\big(M_{1,N}^{(3)}(z),M_{2,N}^{(3)}(z),M_{3,N}^{(3)}(z)\big)',$$
		where for \(1\leq i\leq 3\)
		\begin{small}
		\begin{align}
			M_{i,N}^{(3)}(z):&=g_i(z)\sum_{r\neq i}^3\sum_{w\neq i,r}^3W_{rw}^{(3)}(z)+\sum_{l\neq i}^3\big[(g(z)-g_i(z)-g_l(z))W_{il}^{(3)}(z)+V_{il}^{(3)}(z,z)\big]\notag\\
			&-2\kappa_3(\mfc_1\mfc_2\mfc_3)^{-1}g_1(z)g_2(z)g_3(z)\mfb_1^{(1)}\mfb_2^{(1)}\mfb_3^{(1)}+\kappa_4\mfc_i^{-1}g_i(z)^2\sum_{l\neq i}^3\mcB_{(4)}^{(i,l)}\mfc_l^{-1}g_l(z)^2,\label{Main of Eq of MiN}
		\end{align}
		\end{small}\noindent
		and \(\mfb_k^{(1)},\mcB_{(4)}^{(i,l)},W_{jk}^{(3)}(z),V_{ij}^{(3)}(z,z)\) are defined in {\rm (\ref{Main of Eq of mfb}), (\ref{Main of Eq of mcB}), (\ref{Main of Eq of W limit d=3}), (\ref{Main of Eq of bbV limit d=3})}, respectively. Then we have
		\begin{align}
			\lim_{N\to\infty}\Vert N(\bbm(z)-\bbg(z))-\boldsymbol{\Pi}^{(3)}(z,z)^{-1}\diag(\boldsymbol{\mfc}^{-1}\circ\bbg(z))\overrightarrow{M}_N^{(3)}(z)\Vert=0,\label{Main of Eq of Mean function}
		\end{align}
		where \(\boldsymbol{\Pi}^{(3)}(z,z)\) is defined in {\rm (\ref{Main of Eq of invertible 2})}. Consequently, we obtain that
        $$\lim_{N\to\infty}\mbE[\tr(\bbQ(z))]-Ng(z)-\mu_N^{(3)}(z)=0,$$
        where
		\begin{align}
			\mu_N^{(3)}(z)&:=\boldsymbol{1}_3'\boldsymbol{\Pi}^{(3)}(z,z)^{-1}\diag(\boldsymbol{\mfc}^{-1}\circ\bbg(z))\overrightarrow{M}_N^{(3)}(z).\label{Main of Eq of limiting mean}
		\end{align}
	\end{pro}
	To introduce the covariance function \(\mcC_N^{(3)}(z_1,z_2)\), we need the functions \(\mcV_{st}^{(3)}(z_1,z_2)\) and \(\mcW_{st,N}^{(3)}(z_1,z_2)\) for \(1\leq s,t\leq3\) defined as follows:
    \begin{enumerate}
        \item For any \(s,t,r\in\{1,2,3\}\) and \(z_1,z_2\in\mbC_{\eta}^+\), define
        \begin{align}
			&\tilde{\bbV}_r^{(3)}(z_1,z_2)\label{Main of Eq of bbV_l d=3}\\
            &:=\bbPi^{(3)}(z_1,z_2)^{-1}\diag(\mfc^{-1}\circ\bbg(z_1))\diag(\bbV_{\cdot r}^{(3)}(z_2,z_2))(\bbI_3+\bbS_3\bbV^{(3)}(z_1,z_2)),\notag
		\end{align}
		where \(\bbV_{\cdot r}^{(3)}(z_1,z_2)\) is the \(r\)-th column of \(\bbV^{(3)}(z_1,z_2)\) defined in (\ref{Main of Eq of bbV limit d=3}). Let \(\tilde{V}_{str}^{(3)}(z_1,z_2)\) be the \((s,t)\)-th entry of \(\tilde{\bbV}_r^{(3)}(z_1,z_2)\), define
		\begin{align}
			\mcV_{st}^{(3)}(z_1,z_2)=\sum_{l\neq s}^3\tilde{V}_{slt}^{(3)}(z_1,z_2).\label{Main of Eq of mcV limiting d=3}
		\end{align}
		\item For any \(s,t\in\{1,2,3\}\) and \(z_1,z_2\in\mbC_{\eta}^+\), let
		\begin{align}
			\mathring{\bbV}^{(3)}(z_1,z_2):=\diag(\mfc^{-1}\circ\bbg(z_1))\bbV^{(3)}(z_1,z_2),\label{Main of Eq of bbV circ d=3}
		\end{align}
		and \(\mathring{V}_{st}^{(3)}(z_1,z_2)\) be the \((s,t)\)-th entry of \(\mathring{\bbV}^{(3)}(z_1,z_2)\), then define
		\begin{small}
		\begin{align}
			&\mcW_{st,N}^{(3)}(z_1,z_2)\label{Main of Eq of mcW limiting d=3}\\
			&:=\mfc_s^{-1}g_s(z_1)g_s(z_2)\sum_{l\neq s}^3\mcB_{(4)}^{(s,l)}\mathring{V}_{lt}^{(3)}(z_1,z_2)+\mathring{V}_{st}^{(3)}(z_1,z_2)\sum_{l\neq s}^3\mcB_{(4)}^{(s,l)}\mfc_l^{-1}g_l(z_1)g_l(z_2).\notag
		\end{align}
		\end{small}
    \end{enumerate}
	\begin{pro}[Covariance function \(\mcC_N^{(3)}(z_1,z_2)\) for \(d=3\)]\label{Main of Thm of Variance}
		Under Assumptions {\rm \ref{Main of Aq of general noise}} and {\rm \ref{Main of Aq of dimension}}, for any \(\eta_0>0\) and \(z_1,z_2\in\mcS_{\eta_0}\) in {\rm (\ref{Main of Eq of stable region})}, let
		\begin{small}
		\begin{align}
			\mcC_{st,N}^{(3)}(z_1,z_2):=\Cov(\tr(\bbQ^{ss}(z_1)),\tr(\bbQ^{tt}(z_2))),\quad\bbC_N^{(3)}(z_1,z_2):=[\mcC_{st,N}^{(3)}(z_1,z_2)]_{3\times3},\label{Main of Eq of mcC d=3}
		\end{align}
		\end{small}\noindent
		where \(s,t\in\{1,2,3\}\). Further define 
        \begin{small}
        \begin{align}
			\bbF_N^{(3)}(z_1,z_2)=[\mcF_{st,N}^{(3)}(z_1,z_2)]_{3\times3},\quad\mcF_{st,N}^{(3)}(z_1,z_2):=2\mcV_{st}^{(3)}(z_1,z_2)+\kappa_4\mcW_{st,N}^{(3)}(z_1,z_2),\label{Main of Eq of bbF d=3}
		\end{align}
        \end{small}\noindent
        where \(\mcV_{st}^{(3)}(z_1,z_2)\) and \(\mcW_{st,N}^{(3)}(z_1,z_2)\) are determined by the system of equations {\rm (\ref{Main of Eq of mcV limiting d=3})} and {\rm (\ref{Main of Eq of mcW limiting d=3})}, respectively. Then we have
		\begin{align}
			\lim_{N\to\infty}\Vert\bbC_N^{(3)}(z_1,z_2)-\bbPi^{(3)}(z_1,z_2)^{-1}\diag(\mfc^{-1}\circ\bbg(z_1))\bbF_N^{(3)}(z_1,z_2)\Vert=0,\label{Main of Eq of bbC d=3}
		\end{align}
		where \(\bbPi^{(3)}(z_1,z_2)\) is defined in {\rm (\ref{Main of Eq of invertible 2})}. Consequently, \(\Var(\tr(\bbQ(z)))\) is bounded by \(C_{\eta_0,\mfc}\) for any \(z\in\mcS_{\eta_0}\) and 
        $$\lim_{N\to\infty}\big|\Cov\big(\tr(\bbQ(z_1)),\tr(\bbQ(z_2))\big)-\mcC_N^{(3)}(z_1,z_2)\big|=0,$$
        where
		\begin{align}
			\mcC_N^{(3)}(z_1,z_2):=\boldsymbol{1}_3'\bbPi^{(3)}(z_1,z_2)^{-1}\diag(\mfc^{-1}\circ\bbg(z_1))\bbF_N^{(3)}(z_1,z_2)\boldsymbol{1}_3.\label{Main of Eq of limiting variance}
		\end{align}
	\end{pro}
	The proofs of Propositions \ref{Main of Thm of Mean function} and \ref{Main of Thm of Variance} are provided in \S\ref{Sec of mean and covariance} of the supplement.
	\begin{remark}
		The functions \(\mu_N^{(3)}(z)\) and \(\mcC_N^{(3)}(z_1,z_2)\), introduced in Propositions \ref{Main of Thm of Mean function} and \ref{Main of Thm of Variance}, involve the inverse of the matrix \(\bbPi^{(3)}(z_1,z_2)\) defined in (\ref{Main of Eq of invertible 2}). Consequently, it is necessary to establish the invertibility of \(\bbPi^{(3)}(z_1,z_2)\). Similarly, the invertibility of \(\bbGa(z)\) needs to be proven due to its appearance in (\ref{Main of Eq of W limit d=3}). The proofs of these invertibility results can be found in \S\ref{Sec of Stability operator} of the supplement.
	\end{remark}
	\begin{remark}\label{Rem of numerical}
		Given the Stieltjes transform \(g(z)\) and the vectors \(\bba^{(1)},\bba^{(2)},\bba^{(3)}\), the functions \(\bbW^{(3)}(z),\bbV^{(3)}(z_1,z_2),\mcV_{st}^{(3)}(z_1,z_2)\) and \(\mcW_{st,N}^{(3)}(z_1,z_2)\) can be calculated through (\ref{Main of Eq of W limit d=3}), (\ref{Main of Eq of bbV limit d=3}), (\ref{Main of Eq of mcV limiting d=3}) and (\ref{Main of Eq of mcW limiting d=3}), respectively. Combined with \(\kappa_3,\kappa_4\), we can further calculate the values of $\mu_N^{(3)}(z)$ and $\mcC_N^{(3)}(z_1,z_2)$. Furthermore, the Stieltjes transform $g(z)$ can be evaluated numerically using a fixed-point algorithm, see Lemma \ref{Lem of contraction} of the supplement for details.
	\end{remark}
	\begin{remark}\label{Rem of vectors type}
        By Propositions \ref{Main of Thm of Mean function} and \ref{Main of Thm of Variance}, \(\mu_N^{(3)}(z)\) depends on the third and fourth cumulants \(\kappa_3,\kappa_4\), and the unit vectors \(\bba^{(1)},\bba^{(2)},\bba^{(3)}\), while \(\mcC_N^{(3)}(z_1,z_2)\) depends on the fourth cumulant \(\kappa_4\) and the unit vectors \(\bba^{(1)},\bba^{(2)},\bba^{(3)}\). Precisely, for the mean function \(\mu_N^{(3)}(z)\), note that \(M_{i,N}^{(3)}(z)\) contains the functions
		\begin{align*}
			\kappa_3(\mfc_1\mfc_2\mfc_3)^{-1}g_1(z)g_2(z)g_3(z)\mfb_1^{(1)}\mfb_2^{(1)}\mfb_3^{(1)}\quad{\rm and}\quad\kappa_4\mfc_i^{-1}g_i(z)^2\sum_{l\neq i}^3\mcB_{(4)}^{(i,l)}\mfc_l^{-1}g_l(z)^2.
		\end{align*}
        For example, if \(\bba^{(l)}=(1,0,\cdots,0)'\) for some \(l\in\{1,2,3\}\), then \(\mfb_l^{(1)}=N^{-1/2}\sum_{i_l=1}^{n_l}a_{i_l}^{(l)}=\mrO(N^{-1/2})\) and \(\mu_N^{(3)}(z)\) will be independent of \(\kappa_3\); if all \(\bba^{(l)}\) are delocalized, then \(\mcB_{(4)}^{(i,l)}=\Vert\bba^{(k)}\Vert_4^4\to0\) by Remark \ref{Rem of vectors} and \eqref{Main of Eq of localized}, then \(\mu_N^{(3)}(z)\) will be also independent of \(\kappa_4\). Similarly, for the variance function \(\mcC_N^{(3)}(z_1,z_2)\), by (\ref{Main of Eq of mcW limiting d=3}), \(\mcW_{st,N}^{(3)}(z_1,z_2)\) depends on \(\mcB_{(4)}^{(s,l)}\), so \(\lim_{N\to\infty}|\mcW_{st,N}^{(3)}(z_1,z_2)|=0\) if all \(\bba^{(l)}\) are delocalized. By (\ref{Main of Eq of bbC d=3}) and (\ref{Main of Eq of bbF d=3}), we have
		\begin{align*}
			\bbC_N^{(3)}(z_1,z_2)=2\bbPi^{(3)}(z_1,z_2)^{-1}\diag(\mfc^{-1}\circ\bbg(z_1))\mcV^{(3)}(z_1,z_2)+\mro(1)\boldsymbol{1}_{3\times3},
		\end{align*}
		which is independent of \(\kappa_4\), so does \(\mcC_N^{(3)}(z_1,z_2)\) due to (\ref{Main of Eq of limiting variance}). 
	\end{remark}
   The following proposition follows from Remark \ref{Rem of vectors type}.
    \begin{pro}\label{Rem of comparison}
        \begin{enumerate}
            \item In general, the asymptotic mean \(\xi_N^{(3)}=\xi_N^{(3)}(\kappa_3,\kappa_4,\bba^{(1)},\bba^{(2)},\bba^{(3)})\) of LSS in \eqref{Main of Eq of LSS mean d=3} depends on \(\kappa_3,\kappa_4\) and \(\bba^{(1)},\bba^{(2)},\bba^{(3)}\).
            \item In general, the asymptotic variance \(\sigma_N^{(3)}=\sigma_N^{(3)}(\kappa_4,\bba^{(1)},\bba^{(2)},\bba^{(3)})\) of LSS in \eqref{Main of Eq of LSS variance d=3} depends on \(\kappa_4\) and \(\bba^{(1)},\bba^{(2)},\bba^{(3)}\).
            \item If the third and fourth cumulants of random noises are zero (e.g., the noise tensor \(\bbX\) is Gaussian), \(\xi_N^{(3)}\) and \(\sigma_N^{(3)}\) will be independent of \(\bba^{(1)},\bba^{(2)},\bba^{(3)}\).
            \item If all \(\bba^{(1)},\bba^{(2)},\bba^{(3)}\) are delocalized, \(\xi_N^{(3)}\) and \(\sigma_N^{(3)}\) will be independent of \(\kappa_4\).
        \end{enumerate}
    \end{pro}
   For further illustrations of these conclusions, readers may refer to the numerical experiments reported in Table \ref{Tab of LSS mean and variance} in \S\ref{Main of sec of experiment 1} of the supplement for more details.

{

\subsection{Outline of the proof}\label{sec of CLT}

The central analytical challenge of our theory arises from the intrinsic correlation structure of the contracted matrix $\bbM = N^{-1/2}\bbPhi_d(\bbX, \bba^{(1)}, \ldots, \bba^{(d)})$. Unlike classical random matrix ensembles with independent entries, the contraction operation introduces complex dependencies: each block $\bbM^{st}$ involves sums over shared tensor indices weighted by directional vectors $\bba^{(l)}$. This fundamentally distinguishes our setting from standard random matrix models.

We now outline the proof of Theorem \ref{Main of Thm of CLT LSS d=3}, emphasizing how the spectral-level results developed in preceding sections feed into the central limit theorem. The argument proceeds through a hierarchy of resolvent estimates, each building on the previous to achieve the fine control necessary for distributional convergence.

\textit{Reduction to resolvent analysis.}
Following the general framework of Chapter 9 in \cite{bai2010spectral}, consider the event $\mcE_{\bbM}:=\{\Vert\bbM\Vert\leq v_B^{(3)}+t\}$ for a fixed constant $t>0$. By Theorem \ref{Main of Thm of Dyson}, $\mbP(\mcE_{\bbM})\geq1-\mro(N^{-l})$ for any $l>0$, so outlier eigenvalues are negligible and $G_N(f)\mathbf{1}_{\mcE_{\bbM}}\overset{\mbP}{\longrightarrow}G_N(f)$. The Cauchy integral formula then reduces our analysis to
\begin{align*}
-\frac{1}{2\pi{\rm i}}\oint_{\mathfrak{C}}f(z)\{\tr(\bbQ(z))-Ng(z)\}\,dz,
\end{align*}
where $\mathfrak{C}$ is a rectangular contour with vertices $\pm E_0\pm{\rm i}\eta_0$ such that $E_0\geq v_B^{(3)}+t$ and $\eta_0>0$ is sufficiently small. The choice of $E_0$ together with the event $\mcE_{\bbM}$ ensures that the integrand is well-defined on $\mathfrak{C}$.

\textit{Hierarchy of resolvent estimates.}
The proof relies on three levels of resolvent control, established in increasing order of precision. Theorem \ref{Main of Thm of Dyson} provides a global spectral control, establishing concentration of $\|\bbM\|$ near the right edge $v_B^{(3)}$ and thereby confining eigenvalues to a bounded region with overwhelming probability. Theorem \ref{Main of Thm of approximation} then establishes averaged resolvent approximation: $N^{-1}\tr(\bbQ(z)) \approx g(z)$, yielding the deterministic centering $Ng(z)$ in our CLT. Finally, Theorem \ref{Main of Thm of entrywise law d=3} shows that the entrywise resolvent behavior is governed by coupled equations \eqref{eq:entrywiselaw}. This finest level of precision is essential for the cumulant expansion, enabling replacement of random resolvent entries with deterministic quantities at controlled error cost.

\begin{remark}[Limit theory Beyond LSD]\label{Rem of LSD to CLT}
It is worth clarifying why the machinery developed for the limiting spectral distribution is indispensable for the CLT. While Theorem \ref{Main of Thm of approximation} identifies the centering $Ng(z)$, establishing Gaussian fluctuations requires controlling error terms in cumulant expansions to precision $o(1)$. Each expansion step generates products of resolvent entries, and naive bounds using only $\|\bbQ(z)\| \leq \eta^{-1}$ are insufficient. Theorem \ref{Main of Thm of entrywise law d=3} provides the refined estimates needed: by approximating $Q_{ii}^{ss}$ with the deterministic term as in \eqref{eq:entrywiselaw} and exploiting the block structure, we identify which terms contribute to the mean and variance at order $O(1)$ versus those that vanish as $N \to \infty$. Furthermore, a critical step is establishing invertibility of auxiliary matrix $\bbPi^{(3)}(z_1, z_2)$ which appears over and over again in the mean and variance expressions of our CLT such as the covariance kernel $\mcC_N^{(3)}(z_1, z_2)$ appearing in the CLT. Thus,  rather than merely identifying the LSD, the full resolvent control via analyzing Dyson equation limiting behavior is essential for the fluctuation analysis and our inferential goals. See \S\ref{Sec of Dyson} and \S\ref{Sec of LSD} in the supplement for details.
\end{remark}

\textit{Cumulant expansion and term classification.}
With entrywise control in hand, we analyze the centered resolvent trace through systematic cumulant expansions following \cite{khorunzhy1996asymptotic}. The correlation structure of $\bbM = N^{-1/2}\bbPhi_3(\bbX, \bba^{(1)}, \bba^{(2)}, \bba^{(3)})$ necessitates expanding to fourth-order cumulants, since non-Gaussian noise introduces contributions from $\kappa_3$ and $\kappa_4$ to both the mean and variance. A key technical innovation is our classification scheme exploiting the $d \times d$ block structure $\bbQ(z) = [\bbQ^{st}(z)]_{1\le s,t\le d}$: terms containing off-diagonal products $Q_{i_{t_1} i_{t_2}}^{t_1 t_2}$ with $t_1 \neq t_2$ contribute at order $O(N^{-1/2}\|\bbQ\|^k)$, enabling systematic identification of the dominant ``major'' terms.

\textit{Mean, variance, and Gaussian convergence.}
The cumulant analysis yields explicit characterizations of the first two moments. Proposition \ref{Main of Thm of Mean function} identifies the mean function, revealing the bias of the empirical spectral distribution, while Proposition \ref{Main of Thm of Variance} computes the covariance structure with the stability operator $\bbPi^{(3)}(z_1,z_2)^{-1}$ appearing explicitly. For Gaussian convergence, we first establish tightness:
\begin{pro}\label{Main of Thm of Tightness}
			Under Assumptions {\rm \ref{Main of Aq of general noise}} and {\rm \ref{Main of Aq of dimension}}, for any \(\eta_0>0\), \(\tr(\boldsymbol{Q}(z))-\mathbb{E}[\tr(\boldsymbol{Q}(z))]\) is tight on \(\mcS_{\eta_0}\) in {\rm (\ref{Main of Eq of stable region})}, i.e., 
			$$\sup_{\substack{z_1,z_2\in\mcS_{\eta_0}\\z_1\neq z_2}}\frac{\mathbb{E}\left[|\tr(\boldsymbol{Q}(z_1)-\boldsymbol{Q}(z_2))-\mathbb{E}[\tr(\boldsymbol{Q}(z_1)-\boldsymbol{Q}(z_2))|^2\right]}{|z_1-z_2|^2}<C_{\eta_0}.$$
		\end{pro}
\noindent We then analyze  the joint characteristic function of the real part and imaginary part of \(\tr(\bbQ(z))-\mbE[\tr(\bbQ(z))]\), obtaining differential equations whose solutions converge to a Gaussian form:
\begin{pro}\label{Main of Thm of CLT}
			Under Assumptions {\rm \ref{Main of Aq of general noise}} and {\rm \ref{Main of Aq of dimension}}, for any sufficiently small constant \(\eta_0>0\), \(\tr(\boldsymbol{Q}(z))-\mathbb{E}[\tr(\boldsymbol{Q}(z))]\) weakly converges to a Gaussian random process on \(\mcS_{\eta_0}\).
		\end{pro}
\noindent The proofs of Propositions \ref{Main of Thm of Tightness} and \ref{Main of Thm of CLT} are provided in \S\ref{Sec of CLT} of the supplement.

\textit{Completing the proof.}
Decomposing the contour as $\mfC=\mfC^h\cup\mfC^v$, where $\mfC^h:=\{x\pm{\rm i}\eta_0:x\in[-E_0,E_0]\}$ and $\mfC^v:=\mfC\backslash\mfC^h$, we have $\mfC^h\subset\mcS_{\eta_0}$. By Propositions \ref{Main of Thm of Mean function}, \ref{Main of Thm of Variance}, \ref{Main of Thm of Tightness}, and \ref{Main of Thm of CLT},
$$-\frac{1}{2\pi{\rm i}\sigma_N^{(3)}}\oint_{\mfC^h}f(z)\{\tr(\bbQ(z))-Ng(z)\}dz-\xi_N^{(3)}/\sigma_N^{(3)}\overset{d}{\longrightarrow}\mcN(0,1).$$
The proof is completed by showing that the vertical contribution is asymptotically negligible:
$$\lim_{\eta_0\downarrow0}\limsup_{N\to\infty}\mbE\Bigg|\oint_{\mfC^v}f(z)\{\tr(\bbQ(z))-Ng(z)\}dz\Bigg|^2=0.$$

\begin{remark}\label{Rem of novelty}
Several aspects of this work distinguish it from existing random matrix theory literature. First, our block-structured cumulant expansion provides a systematic treatment of random matrices with dependent entries, departing from classical independent-entry assumptions. Second, we accommodate non-Gaussian noise satisfying only moment conditions and show that both the third and fourth cumulants $\kappa_3$ and $\kappa_4$ of the noise could contribute to both the mean and variance in our CLT (see Remark \ref{Rem of vectors type} and Proposition \ref{Rem of comparison}), different from the traditional RMT regime where $\kappa_3$ typically does not count. Third, our tensor model \eqref{Main of Eq of general spiked tensor model} allows asymmetric, non-Gaussian signals, in contrast to existing literature \cite{de2022random,Seddik2024203} that typically requires symmetry or Gaussianity. Fourth, while we present $d=3$ in detail for clarity, the combinatorial growth of cross-block interactions for general $d\geq 3$ requires careful treatments developed in \S\ref{Sec of general Main}. Finally, the connection between spectral properties of $\bbR$ and the alignment hypothesis \eqref{Main of Eq of hypothesis test 1} is novel, yielding asymptotically valid tests with explicit power calculations under alternatives.
\end{remark}
}

    \section{Tests for tensor signals when \texorpdfstring{\(d=3\)}{d=3}}\label{sec of hypothesis}

 { Having established the CLT for linear spectral statistics of $\mathbf{M}$, we now apply these results to construct hypothesis tests for signal alignments and signal matchings under mild conditions.} These problems are formulated in equations \eqref{Main of Eq of hypothesis test 1} and \eqref{Main of Eq of hypothesis test 2}, respectively.

    \subsection{Testing for tensor signal alignments}\label{sec of basic procedure}
	When \(d=3\), recall our spiked tensor model (\ref{Main of Eq of general spiked tensor model}):
	\begin{align*}
		\bbT=\sum_{r=1}^R\beta_r\bbx^{(r,1)}\otimes\bbx^{(r,2)}\otimes\bbx^{(r,3)}+\frac{1}{\sqrt{N}}\bbX.
	\end{align*}
    Given three unit vectors \(\bba^{(1)}\in\mbR^{n_1},\bba^{(2)}\in\mbR^{n_2},\bba^{(3)}\in\mbR^{n_3}\), we construct the following statistic:
	\begin{align}
		\widehat{T}_N^{(3)}:=\widehat{T}_N^{(3)}(\bba^{(1)},\bba^{(2)},\bba^{(3)}):=\Vert\bbR\Vert_2^2-N\int_{-\infty}^{\infty}x^2\nu(dx),\label{Main of Eq of test statistic 1}
	\end{align}
    where \(\nu\) is the LSD of \(\bbM\), and \(\bbR\) and \(\bbM\) are defined in (\ref{Main of Eq of R and M}). The following proposition is provided:
    \begin{pro}\label{Pro of test statistic CLT}
        Under Assumptions {\rm \ref{Main of Aq of general noise}} and {\rm \ref{Main of Aq of dimension}}, for the spiked tensor model {\rm \eqref{Main of Eq of general spiked tensor model}} and three unit vectors \(\bba^{(1)}\in\mbR^{n_1},\bba^{(2)}\in\mbR^{n_2},\bba^{(3)}\in\mbR^{n_3}\), the statistic \(\widehat{T}_N^{(3)}\) in {\rm \eqref{Main of Eq of test statistic 1}} satisfies that
        \begin{align*}
            \big(\widehat{T}_N^{(3)}-\xi_N^{(3)}-\mcD^{(3)}\big)/\sigma_N^{(3)}\overset{d}{\longrightarrow}\mcN(0,1),
        \end{align*}
        where
        \begin{align}
		  \mcD^{(3)}:=2\sum_{r=1}^R\beta_r^2\sum_{l=1}^3\langle\bbx^{(r,l)},\bba^{(l)}\rangle^2\geq0,\label{Main of Eq of mean drift}
	    \end{align}
        and \(\xi_N^{(3)},\sigma_N^{(3)}\) are derived from \eqref{Main of Eq of LSS mean d=3} and \eqref{Main of Eq of LSS variance d=3} by setting \(f(z)=z^2\), i.e.,
	    \begin{align*}
		  \xi_N^{(3)}&=-\frac{1}{2\pi{\rm i}}\oint_{\mfC_1}z^2\mu_N^{(3)}(z;\kappa_3,\kappa_4,\bba^{(1)},\bba^{(2)},\bba^{(3)})dz,\\
          (\sigma_N^{(3)})^2&=-\frac{1}{4\pi^2}\oint_{\mfC_1}\oint_{\mfC_2}z_1^2z_2^2\mcC_N^{(3)}(z_1,z_2;\kappa_4,\bba^{(1)},\bba^{(2)},\bba^{(3)})dz_1dz_2.
	    \end{align*}
    \end{pro}
    \begin{proof}
        Note that \(\bbR=\bbM+\sum_{r=1}^R\beta_r\Delta^{(r)}\), where
    \begin{small}
    $$\Delta^{(r)}=\bbU_r\left(\begin{array}{ccc}
		0&\langle\bba^{(3)},\bbx^{(r,3)}\rangle&\langle\bba^{(2)},\bbx^{(r,2)}\rangle\\
		\langle\bba^{(3)},\bbx^{(r,3)}\rangle&0&\langle\bba^{(1)},\bbx^{(r,1)}\rangle\\
		\langle\bba^{(2)},\bbx^{(r,2)}\rangle&\langle\bba^{(1)},\bbx^{(r,1)}\rangle&0
	\end{array}\right)\bbU_r',\quad\bbU_r:=\left(\begin{array}{ccc}
		\bbx^{(r,1)}&\boldsymbol{0}_{n_1}&\boldsymbol{0}_{n_2}\\
		\boldsymbol{0}_{n_2}&\bbx^{(r,2)}&\boldsymbol{0}_{n_2}\\
		\boldsymbol{0}_{n_3}&\boldsymbol{0}_{n_3}&\bbx^{(r,3)}
	\end{array}\right),$$
    \end{small}\noindent
    since \(\{\bbx^{(1,l)},\cdots,\bbx^{(R,l)}\}\) are orthogonal in \(\mbR^{n_l}\) for \(1\leq l\leq 3\), then we have \(\bbU_{r_1}\bbU_{r_2}'=\boldsymbol{0}_{N\times N}\) for \(r_1\neq r_2\) and
	$$\Vert\bbR\Vert_2^2=\Vert\bbM\Vert_2^2+\sum_{r=1}^R\beta_r^2\Vert\Delta^{(r)}\Vert_2^2+2\sum_{r=1}^R\beta_r\tr(\bbM\Delta^{(r)}).$$
    Moreover, for each \(1\leq r\leq R\), since
    \begin{align*}
		&\tr(\bbM\Delta^{(r)})=\frac{2}{\sqrt{N}}\big(\langle\bbx^{(r,3)},\bba^{(3)}\rangle(\bbx^{(r,1)})'\bbX(\bba^{(3)})\bbx^{(r,2)}+\langle\bbx^{(r,2)},\bba^{(2)}\rangle(\bbx^{(r,1)})'\bbX(\bba^{(2)})\bbx^{(r,3)}\notag\\
		&+\langle\bbx^{(r,1)},\bba^{(1)}\rangle(\bbx^{(r,2)})'\bbX(\bba^{(1)})\bbx^{(r,3)}\big)=\frac{2}{\sqrt{N}}\sum_{i_1,i_2,i_3=1}^{n_1,n_2,n_3}X_{i_1i_2i_3}\big(\langle\bbx^{(r,3)},\bba^{(3)}\rangle x_{i_1}^{(r,1)}x_{i_2}^{(r,2)}a_{i_3}^{(3)}\notag\\
		&+\langle\bbx^{(r,2)},\bba^{(2)}\rangle x_{i_1}^{(r,1)}a_{i_2}^{(2)}x_{i_3}^{(r,3)}+\langle\bbx^{(r,1)},\bba^{(1)}\rangle a_{i_1}^{(1)}x_{i_2}^{(r,2)}x_{i_3}^{(r,3)}\big),
	\end{align*}
    then we have \(\mbE[\tr(\bbM\Delta^{(r)})]=0\) and
	\begin{small}
	\begin{align}
		&\Var(\tr(\bbM\Delta^{(r)}))\notag\\
		&\leq\frac{4}{N}\sum_{i_1,i_2,i_3=1}^{n_1,n_2,n_3}\big(\langle\bbx^{(r,3)},\bba^{(3)}\rangle x_{i_1}^{(r,1)}x_{i_2}^{(r,2)}a_{i_3}^{(3)}+\langle\bbx^{(r,2)},\bba^{(2)}\rangle x_{i_1}^{(r,1)}a_{i_2}^{(2)}x_{i_3}^{(r,3)}+\langle\bbx^{(r,1)},\bba^{(1)}\rangle a_{i_1}^{(1)}x_{i_2}^{(r,2)}x_{i_3}^{(r,3)}\big)^2\notag\\
		&\leq\frac{12}{N}\sum_{i_1,i_2,i_3=1}^{n_1,n_2,n_3}\big[(x_{i_1}^{(r,1)}x_{i_2}^{(r,2)}a_{i_3}^{(3)})^2+(x_{i_1}^{(r,1)}a_{i_2}^{(2)}x_{i_3}^{(r,3)})^2+(a_{i_1}^{(1)}x_{i_2}^{(r,2)}x_{i_3}^{(r,3)})^2\big]\leq\frac{36}{N}.\notag
	\end{align}
	\end{small}\noindent
	Hence, by Chebyshev's inequality, \(\tr(\bbM\Delta^{(r)})\overset{\mbP}{\longrightarrow}0\) and
    \begin{align}
        \widehat{T}_N^{(3)}\overset{\mbP}{\longrightarrow}\Vert\bbM\Vert_2^2-N\int_{-\infty}^{\infty}x^2\nu(dx)+\sum_{r=1}^R\beta_r^2\Vert\Delta^{(r)}\Vert_2^2.\label{Main of Eq of widehat T P convergence}
    \end{align}
    According to Theorem \ref{Main of Thm of CLT LSS d=3}, since \(N^{-1}\Vert\bbM\Vert_2^2\) is an LSS of \(\bbM\), we have
    \begin{align}
		\frac{\Vert\bbM\Vert_2^2-N\int_{-\infty}^{\infty}x^2\nu(dx)-\xi_N^{(3)}}{\sigma_N^{(3)}}\overset{d}{\longrightarrow}\mcN(0,1).\notag
	\end{align}
    Finally, recall that \(|\langle\bbx^{(r_1,l)},\bbx^{(r_2,l)}\rangle|=\delta_{r_1,r_2}\) for \(1\leq l\leq 3\), we can derive that \(\Vert\Delta\Vert_2^2=2\sum_{l=1}^3\langle\bbx^{(r,l)},\bba^{(l)}\rangle^2\), which concludes this proposition.
    \end{proof}
    Here, let \(\tilde{\mcT}_N^{(3)}:=(\widehat{T}_N^{(3)}-\xi_N^{(3)})/\sigma_N^{(3)}\). Under \(H_0\), since \(\bba^{(l)}\perp\bbx^{(r,l)}\) for all \(1\leq r\leq R,1\leq l\leq d\), then \(\mcD^{(3)}=0\); under \(H_1\), since there exists at least one \(1\leq r\leq R\) and \(1\leq l\leq d\) such that \(\bba^{(l)}\not\perp\bbx^{(r,l)}\), then it implies that \(\mcD^{(3)}>0\). We conclude from Proposition \ref{Pro of test statistic CLT} that
	\begin{align}
		\left\{\begin{array}{ll}
		     \tilde{\mcT}_N^{(3)}\overset{d}{\longrightarrow}\mcN(0,1)&{\rm under\ }H_0,\\
		     \tilde{\mcT}_N^{(3)}-\mcD^{(3)}/\sigma_N^{(3)}\overset{d}{\longrightarrow}\mcN(0,1)&{\rm under\ }H_1. 
		\end{array}\right.\label{Main of Eq of statistic H0 H1}
	\end{align}
	Given a significance level \(\alpha\in(0,1)\), the rejection region of our test procedure is
	\begin{align}
		\big\{{\rm Reject\ }H_0\ {\rm if\ }\tilde{\mcT}_N^{(3)}>z_{\alpha}\big\},\label{Main of Eq of rejection region}
	\end{align}
	where \(z_{\alpha}\) is the \(\alpha\)-th upper quantile of the standard normal. Moreover, the asymptotic power of our test satisfies that
	\begin{align}
		\lim_{N\to\infty}\mbP(\tilde{\mcT}_N^{(3)}>z_{\alpha}|H_1)-1+\Phi(z_{\alpha}-\mcD^{(3)}/\sigma_N^{(3)})=0,\label{Main of Eq of power}
	\end{align}
	where \(\Phi(\cdot)\) is the cumulative distribution function of the standard normal.
	
	The implementation of the test procedure requires numerical values of \(\xi_N^{(3)},\sigma_N^{(3)}\), which can be found following Remarks \ref{Rem of LSS numerical} and \ref{Rem of numerical}. This procedure also needs the values of \(\kappa_3,\kappa_4\), we can estimate \(\kappa_3,\kappa_4\) using straightforward moment estimators:
 \begin{equation}\label{Main of Eq of estimate k3 k4}
     \begin{aligned}	\hat{\kappa}_3=&\frac{(\mfc_1\mfc_2\mfc_3)^{-1}}{N^{3/2}}\sum_{i_1,i_2,i_3=1}^{n_1,n_2,n_3}T_{i_1i_2i_3}^3,\\
\hat{\kappa}_4=&\frac{(\mfc_1\mfc_2\mfc_3)^{-1}}{N}\sum_{i_1,i_2,i_3=1}^{n_1,n_2,n_3}T_{i_1i_2i_3}^4-3.
	\end{aligned}
 \end{equation}
	
	We can show that \(\hat{\kappa}_3\overset{\mbP}{\longrightarrow}\kappa_3\) by the law of large numbers, and similarly for  \(\hat{\kappa}_4\). For \(\hat{\kappa}_3\), note that
	\begin{align*}
		T_{i_1i_2i_3}^3&=\frac{X_{i_1i_2i_3}^3}{N^{3/2}}+\frac{3X_{i_1i_2i_3}^2}{N}\sum_{r=1}^R\beta_rx_{i_1}^{(r,1)}x_{i_2}^{(r,2)}x_{i_3}^{(r,3)}\\
        &+\frac{3X_{i_1i_2i_3}}{N^{1/2}}\Bigg(\sum_{r=1}^R\beta_rx_{i_1}^{(r,1)}x_{i_2}^{(r,2)}x_{i_3}^{(r,3)}\Bigg)^2+\Bigg(\sum_{r=1}^R\beta_rx_{i_1}^{(r,1)}x_{i_2}^{(r,2)}x_{i_3}^{(r,3)}\Bigg)^3.
	\end{align*}
	Since \(\bbx^{(r,1)},\bbx^{(r,2)},\bbx^{(r,3)}\) are unit vectors for \(1\leq r\leq R\), then by H\"older's inequality, it yields that
	\begin{align*}
		&\frac{1}{N^{3/2}}\sum_{i_1,i_2,i_3=1}^{n_1,n_2,n_3}\Bigg|\sum_{r=1}^R\beta_rx_{i_1}^{(r,1)}x_{i_2}^{(r,2)}x_{i_3}^{(r,3)}\Bigg|^3\leq\frac{R^2}{N^{3/2}}\sum_{r=1}^R\beta_r^3\prod_{l=1}^3\Vert\bbx^{(r,l)}\Vert_3^3=\mrO(N^{-3/2}).
	\end{align*}
	Moreover, \(N^{-2}\sum_{i_1,i_2,i_3=1}^{n_1,n_2,n_3}X_{i_1i_2i_3}\big(\sum_{r=1}^R\beta_rx_{i_1}^{(r,1)}x_{i_2}^{(r,2)}x_{i_3}^{(r,3)}\big)^2\) have mean zero and variance of  
    $$N^{-4}\sum_{i_1,i_2,i_3=1}^{n_1,n_2,n_3}\big(\sum_{r=1}^R\beta_rx_{i_1}^{(r,1)}x_{i_2}^{(r,2)}x_{i_3}^{(r,3)}\big)^4\leq\mrO(N^{-4}).$$
    Similarly, the absolute mean of
    $$N^{-5/2}\sum_{i_1,i_2,i_3=1}^{n_1,n_2,n_3}X_{i_1i_2i_3}^2\sum_{r=1}^R\beta_rx_{i_1}^{(r,1)}x_{i_2}^{(r,2)}x_{i_3}^{(r,3)}$$
	is upper bounded by \(N^{-5/2}\sum_{r=1}^R\beta_r\sum_{i_1,i_2,i_3=1}^{n_1,n_2,n_3}|x_{i_1}^{(r,1)}x_{i_2}^{(r,2)}x_{i_3}^{(r,3)}|\leq\mrO(N^{-1})\), and its variance is equal to \(N^{-5}\sum_{i_1,i_2,i_3=1}^{n_1,n_2,n_3}\big(\sum_{r=1}^R\beta_rx_{i_1}^{(r,1)}x_{i_2}^{(r,2)}x_{i_3}^{(r,3)}\big)^2\leq\mrO(N^{-5})\). Therefore, we have
	\begin{align*}
		\frac{(\mfc_1\mfc_2\mfc_3)^{-1}}{N^{3/2}}\sum_{i_1,i_2,i_3=1}^{n_1,n_2,n_3}T_{i_1i_2i_3}^3\overset{\mbP}{\longrightarrow}\frac{1}{n_1n_2n_3}\sum_{i_1,i_2,i_3=1}^{n_1,n_2,n_3}X_{i_1i_2i_3}^3=\hat{\kappa}_3\overset{\mbP}{\longrightarrow}\kappa_3,
	\end{align*}
	so does \(\hat{\kappa}_4\). 
	
	To summarize, we propose the following test procedure:
	\begin{enumerate}[(i)]
		\item Given the observation \(\bbT\in\mbR^{n_1\times n_2\times n_3}\) such that the dimensions \((n_1,n_2,n_3)\) satisfy Assumption \ref{Main of Aq of dimension}, we first compute \(\hat{\kappa}_3,\hat{\kappa}_4\) using (\ref{Main of Eq of estimate k3 k4}).
		\item Based on \((\mfc_1,\mfc_2,\mfc_3)=N^{-1}(n_1,n_2,n_3)\), given any \(z\in\mbC^+\), solve \(\bbg(z)\) by the iterative method mentioned in Remark \ref{Rem of numerical}. After obtaining \(g(z)=\sum_{j=1}^3g_j(z)\), the LSD \(\nu(E)=\lim_{\eta\downarrow0}\pi^{-1}\Im(g(E+{\rm i}\eta))\), which allows us to compute \(\int_{-\infty}^{\infty}x^2\nu(dx)\) numerically.
		\item Based on \(\bbg(z)\) and \(\bba^{(1)},\bba^{(2)},\bba^{(3)}\), compute \(\bbW^{(3)}(z),\bbV^{(3)}(z,z),\mcV_{st}^{(3)}(z_1,z_2)\) and \(\mcW_{st,N}^{(3)}(z_1,z_2)\) by (\ref{Main of Eq of W limit d=3}), (\ref{Main of Eq of bbV limit d=3}), (\ref{Main of Eq of mcV limiting d=3}) and (\ref{Main of Eq of mcW limiting d=3}); combining with \(\hat{\kappa}_3,\hat{\kappa}_4\) in (\ref{Main of Eq of estimate k3 k4}), we further obtain \(\mu_N^{(3)}(z)\) and \(\mcC_N^{(3)}(z_1,z_2)\) by (\ref{Main of Eq of limiting mean}) and (\ref{Main of Eq of limiting variance}), then the asymptotic mean and variance of \(\widehat{T}_N^{(3)}\) can be numerically estimated by
		\begin{align*}
		&\hat{\xi}_N^{(3)}=-\frac{1}{2\pi{\rm i}}\oint_{\mfC_1}z^2\mu_N^{(3)}(z)dz\quad(\hat{\sigma}_N^{(3)})^2=-\frac{1}{4\pi^2}\oint_{\mfC_1}\oint_{\mfC_2}z_1^2z_2^2\mcC_N^{(3)}(z_1,z_2)dz_1dz_2,
		\end{align*}
        where the precise definitions of contours \(\mfC_1,\mfC_2\) are presented in Theorem \ref{Main of Thm of CLT LSS d=3}.
		\item Given a significance level \(\alpha\in(0,1)\), we reject \(H_0\) if 
		\begin{align*}
			&\tilde{\mcT}_N^{(3)}=\frac{\widehat{T}_N^{(3)}-\hat{\xi}_N^{(3)}}{\hat{\sigma}_N^{(3)}}>z_{\alpha}.
		\end{align*}
	\end{enumerate}

\subsection{Testing for tensor signal matching with a reference tensor}\label{sec of generalized procedure}
{In many applications, the reference directions $\bba^{(l)}$ in the hypothesis test \eqref{Main of Eq of hypothesis test 1} are not directly available. Instead, one may have access to a reference tensor that encodes these directions implicitly through its own signal components. A natural question then arises: given a new tensor observation, can we test whether its signal structure is related to that of the reference tensor?

To formalize this, consider two independent tensor observations \(\bbT^{(0)}\) and \(\bbT^{(1)}\) from the model~\eqref{Main of Eq of general spiked tensor model}:}
\begin{align}
  \left\{\begin{array}{l}
           \bbT^{(0)}=\sum_{r_0=1}^{R_0}\beta_{r_0,0}\bbx^{(r_0,1)}\otimes\bbx^{(r_0,2)}\otimes\bbx^{(r_0,3)}+\frac{1}{\sqrt{N}}\bbX^{(0)},\\
           \bbT^{(1)}=\sum_{r_1=1}^{R_1}\beta_{r_1,1}\bby^{(r_1,1)}\otimes\bby^{(r_1,2)}\otimes\bby^{(r_1,3)}+\frac{1}{\sqrt{N}}\bbX^{(1)}, 
         \end{array}\right.\label{Main of Eq of two spiked tensor}
\end{align}
where \(\bbX^{(0)}\) and \(\bbX^{(1)}\) are independent and \(\bbx^{(r_0,l)},\bby^{(r_1,l)}\in\mbR^{n_l}\) are deterministic unit vectors for \(1\leq l\leq 3\) and \(1\leq r_0\leq R_0,1\leq r_1\leq R_1\). {The goal is to test whether the signal components of $\bbT^{(1)}$ share any directional structure with those of $\bbT^{(0)}$. In this subsection, we will treat $\bbT^{(0)}$ as a reference tensor from which signal directions are inferred, while $\bbT^{(1)}$ is merely the observation to be tested. This asymmetry is both statistically and practically motivated (see Remark~\ref{Rem of asymmetric}). }
This leads to the following hypothesis test:
\begin{align}
  \begin{array}{l}
    H_0:\bbx^{(r_0,l)}\perp\bby^{(r_1,l)}\ {\rm for\ any\ }1\leq r_0\leq R_0,1\leq r_1\leq R_1\ {\rm and\ }1\leq l\leq 3,\\
    H_1:\text{there $\exists$ at least one } 1\leq r_0\leq R_0,1\leq r_1\leq R_1\ {\rm and\ }1\leq l\leq 3\ {\rm s.t.\  }\bbx^{(r_0,l)}\not\perp\bby^{(r_1,l)}.
  \end{array}\label{Main of Eq of hypothesis test 2}
\end{align}

To build our test procedure, define  for \(1\leq r_0\leq R_0\),  
    \begin{align*}
        &\bbR^{(r_0,1)}:=\bbPhi_d(\bbT^{(1)},\bbx^{(r_0,1)},\bbx^{(r_0,2)},\bbx^{(r_0,3)}),
    \end{align*}
    and
    \begin{align}
        &\widehat{T}_{r_0,N}^{(3)}=\widehat{T}_{r_0,N}^{(3)}(\bbx^{(r_0,1)},\bbx^{(r_0,2)},\bbx^{(r_0,3)}):=\Vert\bbR^{(r_0,1)}\Vert_2^2-N\int_{-\infty}^{\infty}x^2\nu(dx).\label{Main of Eq of widehat T r0}
    \end{align}
    Similar to Proposition \ref{Pro of test statistic CLT}, the following proposition establishes the asymptotic normality of the statistic $\widehat{T}_{r_0,N}^{(3)}$. 
    
    \begin{pro}\label{Pro of independent copy}
        Under Assumptions {\rm \ref{Main of Aq of general noise}} and {\rm \ref{Main of Aq of dimension}}, for two tensor data \(\bbT^{(0)}\) and \(\bbT^{(1)}\) in \eqref{Main of Eq of two spiked tensor} and \(1\leq r_0\leq R_0\), the statistic \(\widehat{T}_{r_0,N}^{(3)}\) \eqref{Main of Eq of widehat T r0} satisfies that
        \begin{align*}
            \big(\widehat{T}_{r_0,N}^{(3)}-\xi_N^{(r_0,3)}-\mcD^{(r_0,3)}\big)/\sigma_N^{(r_0,3)}\overset{d}{\longrightarrow}\mcN(0,1),
        \end{align*}
        where
        \begin{align*}
		  \mcD^{(r_0,3)}:=2\sum_{r_1=1}^{R_1}\beta_{r_1,1}^2\sum_{l=1}^3\langle\bbx^{(r_0,l)},\bby^{(r_1,l)}\rangle^2\geq0,
	    \end{align*}
        and \(\xi_N^{(r_0,3)},\sigma_N^{(r_0,3)}\) are derived from \eqref{Main of Eq of LSS mean d=3} and \eqref{Main of Eq of LSS variance d=3} by setting \(f(z)=z^2\), i.e.,
	    \begin{align*}
		  \xi_N^{(r_0,3)}&=-\frac{1}{2\pi{\rm i}}\oint_{\mfC_1}z^2\mu_N^{(3)}(z;\kappa_3,\kappa_4,\bbx^{(r_0,1)},\bbx^{(r_0,2)},\bbx^{(r_0,3)})dz,\\
          (\sigma_N^{(r_0,3)})^2&=-\frac{1}{4\pi^2}\oint_{\mfC_1}\oint_{\mfC_2}z_1^2z_2^2\mcC_N^{(3)}(z_1,z_2;\kappa_4,\bbx^{(r_0,1)},\bbx^{(r_0,2)},\bbx^{(r_0,3)})dz_1dz_2.
	    \end{align*}
    \end{pro}
    The proof of the above proposition is the same as Proposition \ref{Pro of test statistic CLT}. We omit the details here. Proposition~\ref{Pro of independent copy} provides the theoretical foundation for testing signal matching when the reference directions $\bbx^{(r_0,l)}$ are known. In practice, however, these directions must be estimated from the reference tensor $\bbT^{(0)}$. This introduces an additional layer of complexity, as the test must now be conducted using estimated directions $\hat{\bbx}^{(r_0,l)}$ rather than the true ones. Note that \(\bbX^{(0)}\) and \(\bbX^{(1)}\) are independent, then \(\bbT^{(0)}\) and \(\bbT^{(1)}\) are also independent. We can apply some existing algorithms in the literature, e.g., the tensor unfolding method \cite{richard2014statistical,arous2021long}, to estimate \(\bbx^{(r_0,1)}\otimes\bbx^{(r_0,2)}\otimes\bbx^{(r_0,3)}\). After obtaining the estimation \(\hat{\bbx}^{(r_0,1)}\otimes\hat{\bbx}^{(r_0,2)}\otimes\hat{\bbx}^{(r_0,3)}\) for \(1\leq r_0\leq R_0\), the statistic \(\widehat{T}_{r_0,N}^{(3)}\) in Proposition \ref{Pro of independent copy} yields the following empirical hypothesis test
    \begin{align}
        \begin{array}{l}
            H_0^{(r_0)}:\hat{\bbx}^{(r_0,l)}\perp\bby^{(r_1,l)}\ {\rm for\ any\ }1\leq r_1\leq R_1\ {\rm and\ }1\leq l\leq 3,\\
            H_1^{(r_0)}:{\rm there\ exists\ at\ least\ one\ }1\leq r_1\leq R_1\ {\rm and\ }1\leq l\leq 3\ {\rm such\ that\ }\hat{\bbx}^{(r_0,l)}\not\perp\bby^{(r_1,l)}, 
        \end{array}\label{Main of Eq of hypothesis test 2 empirical}
    \end{align}
    through the procedures in \S\ref{sec of basic procedure}. We accept \(H_0\) in (\ref{Main of Eq of hypothesis test 2}) only if we accept \(H_0^{(r_0)}\) for all \(1\leq r_0\leq R_0\). We address the validity of this substitution below.
    
    {  In general, the hypothesis tests \eqref{Main of Eq of hypothesis test 2} and \eqref{Main of Eq of hypothesis test 2 empirical} are not automatically equivalent, since the latter relies on estimated directions $\hat{\bbx}^{(r_0,l)}$ rather than the true $\bbx^{(r_0,l)}$. Specifically, we reject $H_0^{(r_0)}$ in \eqref{Main of Eq of hypothesis test 2 empirical} when 
    \begin{align*}
        \big(\widehat{T}_{r_0,N}^{(3)}(\hat{\bbx}^{(r_0,1)},\hat{\bbx}^{(r_0,2)},\hat{\bbx}^{(r_0,3)})-\hat{\xi}_N^{(r_0,3)}\big)/\hat{\sigma}_N^{(r_0,3)}>z_{\alpha},
    \end{align*}
    where $\hat{\xi}_N^{(r_0,3)}$ and $ \big(\hat{\sigma}_N^{(r_0,3)}\big)^2$ are "plug-in" estimators obtained by replacing  $\bbx^{(r_0,l)}$ with their empirical counterparts $\hat{\bbx}^{(r_0,l)}$ in \eqref{Main of Eq of LSS mean d=3} and \eqref{Main of Eq of LSS variance d=3}. The equivalence of these two tests hinges on the estimation error for the signal directions. Define
\begin{align}
    \mathfrak{r}_{r_0,l}:=1-|\langle\bbx^{(r_0,l)},\hat{\bbx}^{(r_0,l)}\rangle|, \quad 1\leq l\leq 3,\label{Main of Eq of estimation error}
\end{align}
which measures the misalignment between the true and estimated directions, with $\mathfrak{r}_{r_0,l}=0$ corresponding to perfect recovery (up to sign) and $\mathfrak{r}_{r_0,l}=1$ to orthogonality. A natural question is whether the test based on estimated directions has the same asymptotic properties as the oracle test using true directions. The following proposition shows that this holds when estimation is consistent.

\begin{pro}\label{Prop:equivalence}
    Under Assumptions {\rm \ref{Main of Aq of general noise}} and {\rm \ref{Main of Aq of dimension}}, if $\max_{1\leq l\leq 3}\mathfrak{r}_{r_0,l}\overset{a.s.}{\longrightarrow}0$, then 
    \begin{align}
        &\lim_{N\to\infty}\big|\mbP\big(\big(\widehat{T}_{r_0,N}^{(3)}(\hat{\bbx}^{(r_0,1)},\hat{\bbx}^{(r_0,2)},\hat{\bbx}^{(r_0,3)})-\hat{\xi}_N^{(r_0,3)}\big)/\hat{\sigma}_N^{(r_0,3)}>z_{\alpha}\big)\notag\\
        &\quad-\mbP\big(\big(\widehat{T}_{r_0,N}^{(3)}(\bbx^{(r_0,1)},\bbx^{(r_0,2)},\bbx^{(r_0,3)})-\xi_N^{(r_0,3)}\big)/\sigma_N^{(r_0,3)}>z_{\alpha}\big)\big|=0.\label{Main of Eq of consistent estimation}
    \end{align}
\end{pro}
    The proof is given in \S\ref{sec:additionalillust} of the supplement. Existing literature provides convergence rates for $\mathfrak{r}_{r_0,l}$ under various estimation methods. For instance, using the MLE approach of Remark \ref{Rem of MLE}, \cite{Seddik2024203} established that $\max_{1\leq l\leq 3}\mathfrak{r}_{r_0,l}=\mrO(\beta_{r_0,0}^{-2})$ almost surely when the reference signal strength $\beta_{r_0,0}$ is moderately large. Thus, if $\beta_{r_0,0}\to\infty$ as $N\to\infty$, the condition of Proposition \ref{Prop:equivalence} is satisfied.

    Moreover, even a moderately strong reference signal (e.g., $\beta_{1,0}=2.5$ or $3$) yields power comparable to the oracle test. See Figure \ref{Fig of 3} in \S\ref{Main of sec of experiment 3} of the supplement for numerical illustrations. We also discuss the power under imperfect estimation in Remark \ref{Rem of power imperfect} of the supplement.
}

{
\begin{remark}[On the asymmetric role of the reference tensor]\label{Rem of asymmetric}
From a statistical perspective, the matching hypothesis \eqref{Main of Eq of hypothesis test 2} concerns whether the \emph{true} signal directions $\{\bbx^{(r_0,l)}\}$ and $\{\bby^{(r_1,l)}\}$ are aligned. The symmetric case where both $\beta_{r_0,0}$ and $\beta_{r_1,1}$ are weak amounts to asking whether two patterns share structural similarity when neither is distinguishable from noise. 
However, when the reference signal $\beta_{r_0,0}$ falls far below the estimation threshold, the signal directions become statistically unidentifiable in the sense that no consistent estimator exists regardless of the method employed \cite{lesieur2017statistical,perry2020statistical}. The hypothesis then involves quantities that cannot be meaningfully inferred from data, rendering the testing problem quite challenging. This parallels assumptions in related high-dimensional testing literature, where minimum signal strength conditions are required for consistent inference \cite{mukherjee2015hypothesis}. When the reference tensor does carry extractable signal structure, it provides well-defined directions against which $\bbT^{(1)}$ can be tested, without requiring exact recovery or a strong target signal.

From a practical perspective, such asymmetry arises naturally in many application domains involving tensor data. In neuroimaging, group-level templates constructed by aggregating data from multiple subjects benefit from variance reduction and serve as reliable references for individual patient scans that are inherently noisier due to motion artifacts or limited acquisition time \cite{cichocki2011tensor}. In hyperspectral imaging, laboratory-calibrated spectral signatures obtained under controlled conditions serve as references against which field measurements---subject to variable lighting and sensor degradation---are compared for material identification. In chemometrics, reference tensors derived from controlled experiments with known compositions are used to analyze new samples measured under less favorable conditions. The common thread is that reference tensors often enjoy higher effective signal strength through aggregation, controlled measurement, or repeated observations, while the target tensor $\bbT^{(1)}$ represents a single noisy realization whose structural relationship to the reference is the quantity of inferential interest.
\end{remark}
}

\section{The general case of \texorpdfstring{$d$}{d}-fold tensors (\texorpdfstring{\(d\geq3\)}{d>3})}\label{Sec of general Main}

\subsection{CLT for LSS of \texorpdfstring{$\bbM$}{M}}

In this section, we extend Theorem \ref{Main of Thm of CLT LSS d=3} in \S\ref{Sec of main d=3} for general \(d\geq3\). The formulas below parallel those in \S\ref{Sec of main d=3}; we use superscript $(d)$ to distinguish the general case. Let
	\begin{align}
	    \bbM=\frac{1}{\sqrt{N}}\bbPhi_d(\bbX,\bba^{(1)},\cdots,\bba^{(d)})\quad{\rm and}\quad\bbQ(z)=(\bbM-z\bbI_N)^{-1},\label{Main of Eq of bbM d}
	\end{align}
	where \(\bba^{(i)}, 1\leq i\leq d\) are \(d\) deterministic unit vectors and \(N=\sum_{i=1}^dn_i\) where the dimension \(n_1,\cdots,n_d\) satisfy Assumption \ref{Main of Aq of dimension}. \(\bbX=[X_{i_1\cdots i_d}]_{n_1\times\cdots\times n_d}\in\mbR^{n_1\times\cdots\times n_d}\) is a random tensor with entries satisfying Assumption \ref{Main of Aq of general noise}. Similarly to \S\ref{sec of mean variance}, we define the mean function \(\mu_N^{(d)}(z)\) and covariance function \(\mcC_N^{(d)}(z_1,z_2)\) using the following functions, which are defined for any sufficiently small \(\eta>0\) and \(z,z_1,z_2\in\mbC_{\eta}^+\):
	\begin{enumerate}
		\item Let
		\begin{align}
			\bbGa^{(d)}(z):=(z+g(z))\bbI_d-\diag(\bbg(z))+g(z)\bbS_d-\diag(\bbg(z))\bbS_d-\bbS_d\diag(\bbg(z))\notag
		\end{align}
		and 
		\begin{align}
			\bbW^{(d)}(z)=[W_{st}^{(d)}(z)]_{d\times d}=-\bbGa^{(d)}(z)^{-1}.\label{Main of Eq of bbW general}
		\end{align}
		\item Let \(\bbPi^{(d)}(z_1,z_2)=\bbI_d-\diag(\mfc^{-1}\circ\bbg(z_1)\circ\bbg(z_2))\bbS_d\) and
		\begin{align}
			&\bbV^{(d)}(z_1,z_2):=\bbPi^{(d)}(z_1,z_2)^{-1}\diag(\mfc^{-1}\circ\bbg(z_1)\circ\bbg(z_2)).\label{Main of Eq of bbV}
		\end{align}
		\item Given \(r,k_1,k_2\in\{1,\cdots,d\}\), let
		\begin{align}
			\tilde{\bbV}_r^{(d)}(z_1,z_2):&=\bbPi^{(d)}(z_1,z_2)^{-1}\diag(\mfc^{-1}\circ\bbg(z_1))\diag(\bbV_{\cdot r}^{(d)}(z_2,z_2))(\bbI_d+\bbS_d\bbV^{(d)}(z_1,z_2)),\notag
		\end{align}
		where \(\bbV_{\cdot r}^{(d)}(z_1,z_2)\) is the \(r\)-th column of \(\bbV^{(d)}(z_1,z_2)\). Moreover, denote
		\begin{align}
			\mcV_{k_1k_2}^{(d)}(z_1,z_2):=\sum_{l\neq k_1}^d\tilde{V}_{k_1k_2l}^{(d)}(z_1,z_2).\label{Main of Eq of mcV limit general d}
		\end{align}
		\item Given \(k_1,k_2\in\{1,\cdots,d\}\), let \(\mathring{\bbV}^{(d)}(z_1,z_2):=\diag(\mfc^{-1}\circ\bbg(z_1))\bbV^{(d)}(z_2,z_2)\) and
		\begin{small}
		\begin{align}
			&\mcW_{k_1k_2,N}^{(d)}(z_1,z_2):=\label{Main of Eq of mcW limit general d}\\
			&\mfc_{k_1}^{-1}g_{k_1}(z_1)g_{k_1}(\bar{z}_2)\sum_{l\neq k_1}^d\mcB_{(4)}^{(k_1,l)}\mathring{V}_{lk_2}^{(d)}(z_1,z_2)+\mathring{V}_{k_1k_2}^{(d)}(z_1,z_2)\sum_{l\neq k_1}^d\mcB_{(4)}^{(k_1,l)}\mfc_l^{-1}g_l(z_1)g_l(\bar{z}_2),\notag
		\end{align}
		\end{small}\noindent
		where \(\mcB_{(4)}^{(k_1,l)}\) is defined in (\ref{Main of Eq of mcB}).
	\end{enumerate}
    Similar to \eqref{Main of Eq of mcU}, let \(v_B^{(d)}:=\max\{\mfv_d,\zeta\}\) and
    \begin{align*}
        &\mathfrak{F}_d:=\big\{f(z):f\ {\rm is\ analytic\ on\ an\ open\ set\ containing\ }\big[-v_B^{(d)},v_B^{(d)}\big]\big\}.
    \end{align*}
	Now, for any \(f\in\mathfrak{F}_d\), we present the extension of Theorem \ref{Main of Thm of CLT LSS d=3} as follows:
	\begin{thm}\label{Main of Thm of general CLT LSS}
		Under Assumptions  \ref{Main of Aq of general noise} and  \ref{Main of Aq of dimension}, for any \(f\in\mathfrak{F}_d\) and deterministic unit vectors \(\bba^{(1)}\in\mbR^{n_1},\cdots,\bba^{(d)}\in\mbR^{n_d}\), let 
        \begin{align*}
            G_N(f)=N\int_{-\infty}^{\infty}f(x)(\nu_N(dx)-\nu(dx)),
        \end{align*}
        where \(\nu_N\) and \(\nu\) are the ESD and LSD of \(\bbM\) in \eqref{Main of Eq of bbM d}, respectively. Then we have 
        \begin{align*}
			\frac{G_N(f)-\mu_N^{(d)}}{\sigma_N^{(d)}}\overset{d}{\longrightarrow}\mcN(0,1).
		\end{align*}
		where 
        \begin{align}
			\xi_N^{(d)}&:=-\frac{1}{2\pi{\rm i}}\oint_{\mfC_1}f(z)\mu_N^{(d)}(z)dz,\label{Main of Eq of mean general d}\\
			(\sigma_N^{(d)})^2&:=-\frac{1}{4\pi^2}\oint_{\mfC_1}\oint_{\mfC_2}f(z_1)f(z_2)\mcC_N^{(d)}(z_1,z_2)dz_1dz_2,\label{Main of Eq of variance general d}
		\end{align}
        where \(\mfC_{1}\) and  \(\mfC_{2}\) are two disjoint rectangular contours with vertices \(\pm E_{1}\pm{\rm i}\eta_{1}\) and \(\pm E_{2}\pm{\rm i}\eta_{2}\), respectively, such that \(E_{1},E_{2}\geq v_B^{(d)}+t\), where \(t>0\) is fixed and \(\eta_{1},\eta_{2}>0\). Here, the mean function \(\mu_N^{(d)}(z)\) is defined as follows:
		\begin{align}
			\mu_N^{(d)}(z):=\boldsymbol{1}_d'\bbPi^{(d)}(z,z)^{-1}\diag(\mfc^{-1}\circ g(z))\overrightarrow{M}_N^{(d)}(z),\notag
		\end{align}
		where \(\overrightarrow{M}_N^{(d)}(z)=(M_{1,N}^{(d)}(z),\cdots,M_{d,N}^{(d)}(z))'\) and for \(1\leq i\leq d\)
		\begin{align}
			&M_{i,N}^{(d)}(z):=g_i(z)\sum_{r\neq i}^d\sum_{w\neq i,r}^dW_{rw}^{(d)}(z)+\sum_{l\neq i}^d\big[(g(z)-g_i(z)-g_l(z))W_{il}^{(d)}(z)+V_{il}^{(d)}(z,z)\big]\notag\\
			&-2\kappa_3\sum_{l\neq i}^d\sum_{t\neq l,i}^d\mcB_{(3)}^{(i,l,t)}(\mfc_i\mfc_l\mfc_t)^{-1}g_i(z)g_l(z)g_t(z)\mfb_i^{(1)}\mfb_l^{(1)}\mfb_t^{(1)}+\kappa_4\mfc_i^{-1}g_i(z)^2\sum_{l\neq i}^d\mcB_{(4)}^{(i,l)}\mfc_l^{-1}g_l(z)^2,\notag
		\end{align}
        and \(\mfb_i^{(1)},\mcB_{(3)}^{(i,l,t)},W_{il}^{(d)}(z),V_{il}^{(d)}(z,z)\) are defined in {\rm (\ref{Main of Eq of mfb}), (\ref{Main of Eq of mcB}), (\ref{Main of Eq of bbW general})} and {\rm (\ref{Main of Eq of bbV})}, respectively. The variance function \(\mcC_N^{(d)}(z_1,z_2)\) is defined as follows:
		\begin{align}
			\mcC_N^{(d)}(z_1,z_2):=\boldsymbol{1}_d'\bbPi^{(d)}(z_1,z_2)^{-1}\diag(\mfc^{-1}\circ\bbg(z_1))\bbF_N^{(d)}(z_1,z_2)\boldsymbol{1}_d,\notag
		\end{align}
		where
		\begin{align*}
			\bbF_N^{(d)}(z_1,z_2)=[\mcF_{st,N}^{(d)}(z_1,z_2)]_{d\times d}\quad\mcF_{st,N}^{(d)}(z_1,z_2):=2\mcV_{st}^{(d)}(z_1,z_2)+\kappa_4\mcW_{st,N}^{(d)}(z_1,z_2),
		\end{align*}
        and \(\mcV_{st}^{(d)}(z_1,z_2),\mcW_{st,N}^{(d)}(z_1,z_2)\) are defined in {\rm (\ref{Main of Eq of mcV limit general d})} and {\rm (\ref{Main of Eq of mcW limit general d})}, respectively.
	\end{thm}
    The proof of Theorem \ref{Main of Thm of general CLT LSS} is provided in \S\ref{Sec of General cases} of the supplement.
	\subsection{Testing for tensor signals}
    Given the tensor data \(\bbT\) generated by \eqref{Main of Eq of general spiked tensor model} and \(d\) deterministic unit vectors \(\bba^{(l)}\in\mbR^{n_l}\) for \(1\leq l\leq d\), let's define
    $$\bbR=\bbR(\bba^{(1)},\cdots,\bba^{(d)})=\bbPhi_d(\bbT,\bba^{(1)},\cdots,\bba^{(d)}),$$
    and
    $$\widehat{T}_N^{(d)}=\widehat{T}_N^{(d)}(\bba^{(1)},\cdots,\bba^{(d)})=\Vert\bbR\Vert_2^2-N\int_{-\infty}^{\infty}x^2\nu(dx).$$
    Now, we present the generalization of Proposition \ref{Pro of test statistic CLT} for \(d\geq3\).
    \begin{pro}\label{Pro of test statistic CLT d}
        Under Assumptions {\rm \ref{Main of Aq of general noise}} and {\rm \ref{Main of Aq of dimension}}, for any deterministic unit vectors \(\bba^{(1)}\in\mbR^{n_1},\cdots,\bba^{(d)}\in\mbR^{n_d}\), we have
        \begin{align*}
            \big(\widehat{T}_N^{(d)}-\xi_N^{(d)}-\mcD^{(d)}\big)/\sigma_N^{(d)}\overset{d}{\longrightarrow}\mcN(0,1),
        \end{align*}
        where
        $$\mcD^{(d)}=\sum_{r=1}^R\beta_r^2\sum_{k\neq l}^d\prod_{j\neq k,l}^d\langle\bbx^{(r,j)},\bba^{(j)}\rangle^2\geq0,$$
        and \(\xi_N^{(d)},\sigma_N^{(d)}\) are derived from \eqref{Main of Eq of mean general d} and \eqref{Main of Eq of variance general d} as follows:
        \begin{align*}
            \xi_N^{(d)}&=-\frac{1}{2\pi{\rm i}}\oint_{\mfC_1}z^2\mu_N^{(d)}(z)dz,\\
			(\sigma_N^{(d)})^2&=-\frac{1}{4\pi^2}\oint_{\mfC_1}\oint_{\mfC_2}z_1^2z_2^2\mcC_N^{(d)}(z_1,z_2)dz_1dz_2.
        \end{align*}
    \end{pro}
    The proof of the above proposition follows the same steps as the proof of Proposition \ref{Pro of test statistic CLT} and is thus omitted.  The key step is to show that \(\bbR=\bbM+\sum_{r=1}^R\beta_r\bbU_r\bbB^{(r)}\bbU_r'\), where \(\bbB^{(r)}=[B_{k,l}^{(r)}]\in\mbR^{d\times d}\) such that \(B_{k,l}=(1-\delta_{k,l})\prod_{j\neq k,l}^d\langle\bbx^{(r,j)},\bba^{(j)}\rangle\) and 
	$$\bbU_r=\left(\begin{array}{cccc}
		\bbx^{(r,1)}&\boldsymbol{0}_{n_1}&\cdots&\boldsymbol{0}_{n_1}\\
		\boldsymbol{0}_{n_2}&\bbx^{(r,2)}&\cdots&\boldsymbol{0}_{n_2}\\
		\vdots&\ddots&\ddots&\vdots\\
        \boldsymbol{0}_{n_d}&\cdots&\boldsymbol{0}_{n_d}&\bbx^{(r,3)}
	\end{array}\right)\in\mbR^{N\times d}.$$
	Let \(\tilde{\mcT}_N^{(d)}:=\big(\widehat{T}_N^{(d)}-\xi_N^{(d)}\big)/\sigma_N^{(d)}\), Proposition \ref{Pro of test statistic CLT} implies that
	$$\left\{\begin{array}{lc}
		\tilde{\mcT}_N^{(d)}\overset{d}{\longrightarrow}\mcN(0,1)&{\rm under\ }H_0,\\
		\tilde{\mcT}_N^{(d)}-\mcD^{(d)}/\sigma_N^{(d)}\overset{d}{\longrightarrow}\mcN(0,1)&{\rm under\ }H_1.
	\end{array}\right.$$
 The test procedures for general \(d\geq3\) are nearly identical to those introduced in \S\ref{sec of basic procedure}, with the only difference being the estimations of \(\hat{\kappa}_3,\hat{\kappa}_4\), which now become:  
    \begin{align*}
		\left\{\begin{array}{l}
			\hat{\kappa}_3=\prod_{l=1}^d\mfc_l\times N^{3/2-d}\sum_{i_1\cdots i_d=1}^{n_1\cdots n_d}T_{i_1\cdots i_d}^3\\
			\hat{\kappa}_4=\prod_{l=1}^d\mfc_l\times N^{2-d}\sum_{i_1\cdots i_d=1}^{n_1\cdots n_d}T_{i_1\cdots i_d}^4-3
			\end{array}\right..
	\end{align*}
    One can show that  \(\hat{\kappa}_3\overset{\mbP}{\longrightarrow}\kappa_3\) and \(\hat{\kappa}_4\overset{\mbP}{\longrightarrow}\kappa_4\) using the same arguments as in \S\ref{sec of basic procedure}, we omit the details here. 

 {
\section{Real Data Analysis}\label{Sec of real data}

In this section, we apply the tensor signal matching test \eqref{Main of Eq of hypothesis test 2} to human action recognition in video data. This application provides a natural testbed for our methodology, as video data are inherently three-dimensional tensors (height $\times$ width $\times$ time), and action categories may be characterized by shared signal structures.

\subsection{Data and preprocessing}

We select six video samples from the UCF Sports Action Dataset \cite{rodriguez2008action}: three lifting videos (denoted $L_1, L_2, L_3$) and three diving videos (denoted $D_1, D_2, D_3$). Each video undergoes the following preprocessing steps: (i) resizing to $400 \times 400$ pixels per frame, (ii) conversion to grayscale, and (iii) centering along the time axis. After preprocessing, each video is represented as a tensor of dimension $400 \times 400 \times 55$. We model that each video tensor (e.g., $L_1$) admits a low-rank signal:
\[
L_1 = \sum_{k=1}^{K_{L_1}} \beta_{k,L_1} \, \bbx_{L_1}^{(k,1)} \otimes \bbx_{L_1}^{(k,2)} \otimes \bbx_{L_1}^{(k,3)},
\]
where $K_{L_1} \in \mathbb{N}^+$ denotes the rank.

Raw video data typically contain noise from heterogeneous sources, e.g., sensor imperfections, compression artifacts, lighting variations, whose statistical characteristics differ  from the homogeneous additive noise assumed in our theoretical framework. To address this mismatch, we adopt a controlled noise injection protocol: by adding synthetic Gaussian noise of sufficient magnitude, the injected noise dominates the original heterogeneous corruption, equalizing effective noise levels across videos and aligning the data with our model assumptions. This approach is well-established in the video analysis literature for evaluating methods under controlled conditions \cite{grushin2013robust}.

Formally, for each sample (e.g., $L_1$), we first rescale it so that
\[
\big\| L_1(\bbx_{L_1}^{(1,1)}, \bbx_{L_1}^{(1,2)}, \bbx_{L_1}^{(1,3)}) \big\|_2^2 = 2 z_{0.95} \hat{\sigma}_N^{(3)},
\]
where $\hat{\sigma}_N^{(3)}$ is defined in Proposition~\ref{Pro of independent copy} and $z_{0.95}$ denotes the upper $5\%$ quantile of the standard normal distribution. This rescaling calibrates the signal strength relative to the test's variance, ensuring that when testing a video against itself, the rejection rate is approximately $95\%$. We then construct noisy observations by adding Gaussian noise:
\[
\bbT_{L_1} := L_1 + N^{-1/2} \bbX_{L_1}, \quad \text{where } X_{i_1 i_2 i_3} \overset{\text{i.i.d.}}{\sim} \mathcal{N}(0,1) \text{ and } N = 855.
\]
Repeating this for all six videos produces noisy samples ${\bbT_{L_1}, \bbT_{L_2}, \bbT_{L_3}, \bbT_{D_1}, \bbT_{D_2}, \bbT_{D_3}}$.

\subsection{Hypothesis test}

For each pair of noisy videos, we test whether their leading signal components are aligned. For instance, given $\bbT_{L_1}$ and $\bbT_{L_2}$, we consider:
\begin{align}
\begin{array}{l}
H_0^{(L_1,L_2)}: \bbx_{l,L_1}^{(1)} \perp \bbx_{l,L_2}^{(k)} \ \text{ for all } l \in \{1,2,3\} \text{ and } k \in \{1,\ldots,K_{L_2}\}, \\[1mm]
H_1^{(L_1,L_2)}: \exists\, l \in \{1,2,3\} \text{ and } k \in \{1,\ldots,K_{L_2}\} \text{ such that } \bbx_{l,L_1}^{(1)} \not\perp \bbx_{l,L_2}^{(k)}.
\end{array}
\label{Main of Eq of real data test}
\end{align}
Intuitively, videos depicting the same action should share aligned signal components, leading to rejection of $H_0$. In contrast, videos of different actions should have non-matching signals, resulting in acceptance of $H_0$.

The testing procedure is as follows:
\begin{enumerate}
\item \textbf{Signal estimation:} Apply the tensor unfolding method to $\bbT_{L_1}$ to estimate $\hat{\beta}_{1,L_1} \hat{\bbx}_{L_1}^{(1,1)} \otimes \hat{\bbx}_{L_1}^{(1,2)} \otimes \hat{\bbx}_{L_1}^{(1,3)}$.

\item \textbf{Test statistic computation:} Compute $\hat{\xi}_N^{(3)}$ and $\hat{\sigma}_N^{(3)}$ using the formulas in Proposition~\ref{Pro of independent copy}, and construct the standardized statistic
\[
\widetilde{\mathcal{T}}_{L_1,L_2}^{(3)} := \frac{\widehat{T}_{L_1,L_2}^{(3)} - \hat{\xi}_N^{(3)}}{\hat{\sigma}_N^{(3)}},
\]
where
\[
\widehat{T}_{L_1,L_2}^{(3)} := \big\| \bbT_{L_2}(\hat{\bbx}_{L_1}^{(1,1)}, \hat{\bbx}_{L_1}^{(1,2)}, \hat{\bbx}_{L_1}^{(1,3)}) \big\|_2^2 - N \int_{-\infty}^{\infty} x^2 \, \nu_N(dx),
\]
and $\nu_N$ is the limiting spectral distribution derived from the dimensions $(400, 400, 55)$.

\item \textbf{Decision:} Reject $H_0^{(L_1,L_2)}$ if $\widetilde{\mathcal{T}}_{L_1,L_2}^{(3)} > z_{0.95}$.
\end{enumerate}
   
\subsection{Results}

We repeat the above procedure 100 times for each pair of videos and record the acceptance rates. The results are summarized in Table~\ref{Tab of real data analysis}, with a visual summary in Figure~\ref{Fig of histogram}.

\begin{table}[!hbtp]
		\centering
		\renewcommand\arraystretch{1.2}
		\caption{  {Empirical acceptance rates of the null hypothesis $H_0^{(\cdot,\cdot)}$ based on 100 independent repetitions. Each entry represents the proportion of times the null hypothesis (no signal matching) was accepted. Within-group pairs  show low acceptance rates, indicating that videos of the same action share aligned signal components. Between-group pairs (the off-diagonal blocks) exhibit high acceptance rates, confirming that videos of different actions have non-matching signal structures.}}
		\begin{tabular}{c|ccc|ccc}
             & $L_1$ & $L_2$ & $L_3$ & $D_1$ & $D_2$ & $D_3$ \\\hline
            $L_1$ &  & 0.18 & 0.26 & 0.78 & 0.74 & 0.70 \\
            $L_2$ & 0.15 &  & 0.22 & 0.78 & 0.77 & 0.72 \\
            $L_3$ & 0.18 & 0.23 &  & 0.72 & 0.77 & 0.84 \\\hline
            $D_1$ & 0.77 & 0.67 & 0.81 &  & 0.24 & 0.18 \\
            $D_2$ & 0.50 & 0.71 & 0.60 & 0.28 &  & 0.15 \\
            $D_3$ & 0.73 & 0.76 & 0.74 & 0.18 & 0.13 &  \\\hline
        \end{tabular}
		\label{Tab of real data analysis}
	\end{table}

\begin{figure}
        \centering
        \includegraphics[width=0.5\linewidth]{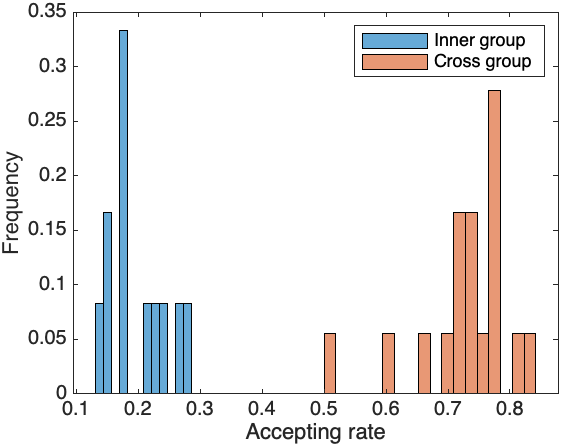}
        \caption{{Histogram of acceptance rates from Table \ref{Tab of real data analysis}, confirming that videos of different actions show distinct signal structures with our test.} }
        \label{Fig of histogram}
    \end{figure}

The results reveal a clear separation between the two types of comparisons. For within-group pairs (both videos from the same action category), the acceptance rates range from 0.13 to 0.28, while for between-group pairs (videos from different action categories), the rates range almost from 0.60 to 0.84. Although the within-group acceptance rates exceed the nominal 5\% level, which is expected and interpretable. Videos within the same action category are performed by different athletes with natural variations in execution speed, body posture, camera angle, and movement style. Consequently, exact signal alignment cannot be anticipated; rather, we expect partial alignment reflecting the common action pattern.  Nevertheless, the clear gap between within-group and between-group acceptance rates demonstrates that our test effectively distinguishes shared action structure from genuinely distinct motions. Within-group pairs consistently show much stronger evidence against the null hypothesis of non-matching signals than between-group pairs, confirming that videos of the same action share partially aligned signal components while videos of different actions do not.
}

\section*{Acknowledgments}

Jianfeng Yao was partially supported by NSFC RFIS Grant No. 12350710179 and Guangdong Province Program No. 2023JC10X022. 

\newpage
\appendix
\begin{center}
{\LARGE {\bf Supplementary Materials of the paper ``Alignment and matching tests for high-dimensional tensor signals
via tensor contraction''}}
\end{center}

{

This supplementary material provides complete proofs of all theoretical results stated in the main manuscript. Below, we outline the structure and explain the purpose of each section.

\medskip

\noindent\textbf{Additional figures and auxiliary results (\S\ref{sec:additionalillust}).} This section collects auxiliary remarks and formulas that support or illustrate the main theoretical developments.

\medskip

\noindent\textbf{Numerical experiments (\S\ref{sec of Numerical}).} This section conducts several numerical experiments to investigate the performance of our theorems and hypothesis tests.
\medskip

\noindent\textbf{Proofs of main results (\S\ref{sec of Basic settings}--\S\ref{Sec of General cases}).} The remainder of this supplement is devoted to proving all theorems and propositions from the main manuscript. These sections are written to be self-contained: all necessary notation, definitions, and assumptions are introduced in \S\ref{sec of Basic settings}, so readers need not refer back to the main text.

The proofs are organized around three core technical components, each building toward our main goal of establishing a central limit theorem (CLT) for linear spectral statistics (LSS) of the contracted tensor $\bbM$.

\begin{itemize}
    \item \textbf{The vector Dyson equation (\S\ref{Sec of Dyson}).} A central object in our analysis is a system of self-consistent equations--known as the vector Dyson equation--that characterizes the asymptotic behavior of $\bbM$'s resolvent. In this section, we establish three fundamental properties: the existence and uniqueness of solutions (\S\ref{sec of Existence and uniqueness Dyson equation}), the invertibility of an associated stability operator (\S\ref{Sec of Stability operator}), and the stability of solutions under perturbations (\S\ref{sec of Stability}). These results form the analytical foundation for all subsequent developments, including the characterization of $\bbM$'s limiting spectral distribution and the fine-grained control of its resolvent entries.

    \item \textbf{Spectral properties of $\bbM$ (\S\ref{Sec of LSD}).} Here we study the limiting spectral distribution (LSD) of $\bbM$, corresponding to Theorems 2.1 and 2.2 in the main manuscript. We first show that both the empirical and limiting spectral distributions have bounded support (\S\ref{Sec of Stable region ESD}--\S\ref{Sec of finite support}), then provide a necessary and sufficient condition for the LSD to have a point mass at zero (\S\ref{Sec of Singularity}). The boundedness of the spectral support is not merely a technical curiosity--it ensures that the contour integration techniques used to establish the CLT are well-defined.

    \item \textbf{Entrywise control of the resolvent (\S\ref{Sec of entrywise law}).} A key step toward the CLT is to show that individual entries of $\bbM$'s resolvent concentrate around their deterministic limits. This "entrywise local law" (Theorem 2.3 in the main manuscript) is established in this section for the case $d=3$. We focus on this case for clarity of exposition; the extension to general $d \geq 3$ follows the same approach but involves more elaborate notation. Section~\ref{sec of Preliminary Lemmas} contains auxiliary lemmas that will be reused throughout the subsequent analysis.

    \item \textbf{The CLT for linear spectral statistics (\S\ref{Sec of mean and covariance}--\S\ref{Sec of CLT}).} With the preceding machinery in place, we prove the main CLT (Theorem 3.1) in two stages. First, in \S\ref{Sec of mean and covariance}, we derive explicit formulas for the asymptotic mean and covariance of the LSS. These formulas arise from systematic equations whose solutions we compute in \S\ref{Sec of majors}. Second, in \S\ref{Sec of CLT}, we establish Gaussian convergence. The strategy is to first show that $\bbM$'s resolvent, viewed as a stochastic process, is tight and asymptotically Gaussian on contours enclosing the spectrum (\S\ref{Sec of Tightness}--\S\ref{Sec of CLT sub}). The CLT for LSS then follows by expressing these statistics as contour integrals of the resolvent (\S\ref{sec of proof CLT d=3}).

    \item \textbf{Extension to general tensor order (\S\ref{Sec of General cases}).} Finally, we extend the entrywise law and CLT to tensors of arbitrary order $d \geq 3$ (Theorems 2.3 and 5.1). Since the proof structure mirrors the $d=3$ case, we focus on highlighting the differences that arise in the calculations.
\end{itemize}
}

{
\section{Additional Illustrative Material and Auxiliary Results}\label{sec:additionalillust}
\setcounter{equation}{0}
\def\theequation{\thesection.\arabic{equation}}
\setcounter{subsection}{0}

\begin{remark}[Tensor Contraction: An Illustrative Examples]\label{rem:tensor_contraction}
	Let $\mathcal{T} \in \mathbb{R}^{n_1 \times n_2 \times \cdots \times n_d}$ be an $d$th-order tensor. The \emph{mode-$k$ product} of $\mathcal{T}$ with a vector $a^{(k)} \in \mathbb{R}^{n_k}$ computes a linear combination along mode $k$, reducing the tensor order by one:
	\begin{equation}
		\mathcal{T} \times_k a^{(k)} \quad \Longrightarrow \quad \text{$(d-1)$th-order tensor}.
	\end{equation}
	
	\medskip
	We illustrate the construction for the case $d = 3$, i.e., $\mathcal{T} \in \mathbb{R}^{n_1 \times n_2 \times n_3}$. The general framework extends naturally to arbitrary order $d$. For visual clarity, we depict the case where $n_1 < n_3 < n_2$.
	
	\begin{center}
		\begin{tikzpicture}[scale=0.85, every node/.style={font=\small}]
			
			\pgfmathsetmacro{\nOne}{1.0}    
			\pgfmathsetmacro{\nTwo}{2.0}    
			\pgfmathsetmacro{\nThree}{1.4}  
			\pgfmathsetmacro{\depthFactor}{0.5}   
			\pgfmathsetmacro{\depthX}{\depthFactor*\nThree}
			
			\begin{scope}[shift={(-5,0)}]
				\fill[blue!20, opacity=0.7] (0,0) -- (\nTwo,0) -- (\nTwo,\nOne) -- (0,\nOne) -- cycle;
				\fill[blue!40, opacity=0.7] (0,0) -- (\depthX,\depthX) -- (\depthX,{\nOne+\depthX}) -- (0,\nOne) -- cycle;
				\fill[blue!55, opacity=0.7] (0,\nOne) -- (\depthX,{\nOne+\depthX}) -- ({\nTwo+\depthX},{\nOne+\depthX}) -- (\nTwo,\nOne) -- cycle;
				\fill[blue!30, opacity=0.8] (\depthX,\depthX) -- ({\nTwo+\depthX},\depthX) -- ({\nTwo+\depthX},{\nOne+\depthX}) -- (\depthX,{\nOne+\depthX}) -- cycle;
				\fill[blue!45, opacity=0.7] (\nTwo,0) -- ({\nTwo+\depthX},\depthX) -- ({\nTwo+\depthX},{\nOne+\depthX}) -- (\nTwo,\nOne) -- cycle;
				\fill[blue!15, opacity=0.7] (0,0) -- (\nTwo,0) -- ({\nTwo+\depthX},\depthX) -- (\depthX,\depthX) -- cycle;
				
				\draw[thick] (0,0) -- (\nTwo,0) -- (\nTwo,\nOne) -- (0,\nOne) -- cycle;
				\draw[thick] (\depthX,\depthX) -- ({\nTwo+\depthX},\depthX) -- ({\nTwo+\depthX},{\nOne+\depthX}) -- (\depthX,{\nOne+\depthX}) -- cycle;
				\draw[thick] (0,0) -- (\depthX,\depthX);
				\draw[thick] (\nTwo,0) -- ({\nTwo+\depthX},\depthX);
				\draw[thick] (\nTwo,\nOne) -- ({\nTwo+\depthX},{\nOne+\depthX});
				\draw[thick] (0,\nOne) -- (\depthX,{\nOne+\depthX});
				
				\node at ({(\nTwo+\depthX)/2}, -1.25) {$\mathcal{T}\in\mathbb{R}^{n_1\times n_2\times n_3}$};
				
				\draw[<->, thick, gray] (-0.25,0) -- (-0.25,\nOne) node[midway, left] {$n_1$};
				\draw[<->, thick, gray] (0,-0.25) -- (\nTwo,-0.25) node[midway, below] {$n_2$};
				\draw[<->, thick, gray] ({\nTwo+0.15},-0.15) -- ({\nTwo+\depthX+0.15},{\depthX-0.15}) node[midway, below right=-2pt] {$n_3$};

			\end{scope}
			
			\draw[-{Stealth[scale=1.2]}, thick, blue!70!black] (-1.6,0.85) -- (0.4,0.85) 
			node[midway, above] {$\times_1 a^{(1)}$};
			
			\begin{scope}[shift={(1.0,0)}]
				\fill[green!35] (0,0) rectangle (\nTwo,\nThree);
				\draw[thick] (0,0) rectangle (\nTwo,\nThree);
				\foreach \x in {0.25,0.65,1.05,1.45,1.85} {
					\foreach \y in {0.25,0.7,1.15} {
						\fill (\x,\y) circle (1.8pt);
					}
				}
				\node[below] at ({\nTwo/2},-1) {$T^{23} \in \mathbb{R}^{n_2 \times n_3}$};
				\draw[<->, thick, gray] ({\nTwo+0.2},0) -- ({\nTwo+0.2},\nThree) node[midway, right] {$n_3$};
				\draw[<->, thick, gray] (0,-0.4) -- (\nTwo,-0.4) node[midway, below] {$n_2$};
			\end{scope}
		\end{tikzpicture}
	\end{center}
	
	\noindent The mode-1 product is defined elementwise as:
	\begin{equation}
		\left(\mathcal{T} \times_1 a^{(1)}\right)_{j,k} = \sum_{i=1}^{n_1} \mathcal{T}_{i,j,k} \cdot a_i^{(1)}, \qquad j \in [n_2],\; k \in [n_3].
	\end{equation}
	
	\vspace{0.5em}
	\noindent\textbf{Blockwise Contractions.} For a third-order tensor, contracting along each mode yields three matrices:
	
	\begin{center}
		\begin{tikzpicture}[scale=0.8, every node/.style={font=\small}]
			
			\pgfmathsetmacro{\nOne}{1.0}    
			\pgfmathsetmacro{\nTwo}{1.8}    
			\pgfmathsetmacro{\nThree}{1.3}  
			\pgfmathsetmacro{\depthFactor}{0.4}
			\pgfmathsetmacro{\depthX}{\depthFactor*\nThree}
			
			\begin{scope}[shift={(0,0)}]
				\fill[blue!20, opacity=0.7] (0,0) -- (\nTwo,0) -- (\nTwo,\nOne) -- (0,\nOne) -- cycle;
				\fill[blue!40, opacity=0.7] (0,0) -- (\depthX,\depthX) -- (\depthX,{\nOne+\depthX}) -- (0,\nOne) -- cycle;
				\fill[blue!55, opacity=0.7] (0,\nOne) -- (\depthX,{\nOne+\depthX}) -- ({\nTwo+\depthX},{\nOne+\depthX}) -- (\nTwo,\nOne) -- cycle;
				\fill[blue!30, opacity=0.8] (\depthX,\depthX) -- ({\nTwo+\depthX},\depthX) -- ({\nTwo+\depthX},{\nOne+\depthX}) -- (\depthX,{\nOne+\depthX}) -- cycle;
				\fill[blue!45, opacity=0.7] (\nTwo,0) -- ({\nTwo+\depthX},\depthX) -- ({\nTwo+\depthX},{\nOne+\depthX}) -- (\nTwo,\nOne) -- cycle;
				\draw[thick] (0,0) -- (\nTwo,0) -- (\nTwo,\nOne) -- (0,\nOne) -- cycle;
				\draw[thick] (\depthX,\depthX) -- ({\nTwo+\depthX},\depthX) -- ({\nTwo+\depthX},{\nOne+\depthX}) -- (\depthX,{\nOne+\depthX}) -- cycle;
				\draw[thick] (0,0) -- (\depthX,\depthX);
				\draw[thick] (\nTwo,0) -- ({\nTwo+\depthX},\depthX);
				\draw[thick] (\nTwo,\nOne) -- ({\nTwo+\depthX},{\nOne+\depthX});
				\draw[thick] (0,\nOne) -- (\depthX,{\nOne+\depthX});
				\node at ({(\nTwo+\depthX)/2}, {(\nOne+\depthX)/2}) {$\mathcal{T}$};
			\end{scope}
			
			\coordinate (TensorRight) at ({(\nTwo+\depthX)}, {(\nOne+\depthX)/2});
			\coordinate (TensorTop) at ({(\nTwo+\depthX)/2}, {(\nOne+\depthX)});
			\coordinate (TensorBottomRight) at ({\nTwo+\depthX*0.7}, {\depthX*0.3});
			
			\begin{scope}[shift={(5.5,0.1)}]
				\fill[green!40] (0,0) rectangle (\nTwo,\nThree);
				\draw[thick] (0,0) rectangle (\nTwo,\nThree);
				\foreach \x in {0.3,0.7,1.1,1.5} \foreach \y in {0.3,0.65,1.0} \fill (\x,\y) circle (1.5pt);
				\node[right] at ({\nTwo+0.15},{\nThree/2}) {$T^{23}$};
				\node[below, font=\footnotesize, green!50!black] at ({\nTwo/2},-0.2) {$n_2 \times n_3$};
			\end{scope}
			\draw[-{Stealth[scale=1.1]}, thick, green!60!black] (TensorRight) -- ({5.5},0.75) 
			node[midway, above] {$\times_1 a^{(1)}$};
			
			\begin{scope}[shift={(0.4,3.8)}]
				\fill[red!30] (0,0) rectangle (\nThree,\nOne);
				\draw[thick] (0,0) rectangle (\nThree,\nOne);
				\foreach \x in {0.25,0.65,1.05} \foreach \y in {0.35,0.7} \fill (\x,\y) circle (1.5pt);
				\node[above] at ({\nThree/2},{\nOne+0.1}) {$T^{13}$};
				\node[right, font=\footnotesize, red!50!black] at ({\nThree/2},-0.2) {$n_1 \times n_3$};
			\end{scope}
			\draw[-{Stealth[scale=1.1]}, thick, red!60!black] (TensorTop) -- (1.05,3.7) 
			node[midway, right] {$\times_2 a^{(2)}$};
			
			\begin{scope}[shift={(5.5,-2.8)}]
				\fill[orange!40] (0,0) rectangle (\nTwo,\nOne);
				\draw[thick] (0,0) rectangle (\nTwo,\nOne);
				\foreach \x in {0.3,0.7,1.1,1.5} \foreach \y in {0.35,0.7} \fill (\x,\y) circle (1.5pt);
				\node[right] at ({\nTwo+0.15},{\nOne/2}) {$T^{12}$};
				\node[below, font=\footnotesize, orange!60!black] at ({\nTwo/2},-0.2) {$n_1 \times n_2$};
			\end{scope}
			\draw[-{Stealth[scale=1.1]}, thick, orange!70!black] (TensorBottomRight) -- (5.4,-2.3) 
			node[midway, above, sloped] {$\times_3 a^{(3)}$};
			
		\end{tikzpicture}
	\end{center}
	
	\noindent The superscript $T^{ij}$ denotes the matrix retaining modes $i$ and $j$. Explicitly:
	\begin{align}
		T^{23}_{j,k} &= \textstyle\sum_{i=1}^{n_1} \mathcal{T}_{i,j,k} \cdot a_i^{(1)}, &
		T^{13}_{i,k} &= \textstyle\sum_{j=1}^{n_2} \mathcal{T}_{i,j,k} \cdot a_j^{(2)}, &
		T^{12}_{i,j} &= \textstyle\sum_{k=1}^{n_3} \mathcal{T}_{i,j,k} \cdot a_k^{(3)}.
	\end{align}
	
	\vspace{0.5em}
	\noindent\textbf{Full Contraction Matrix.} These contraction matrices assemble into a symmetric block matrix $R \in \mathbb{R}^{N \times N}$ with $N = n_1 + n_2 + n_3$:
	
	\begin{center}
		\begin{tikzpicture}[scale=0.8, every node/.style={font=\small}]
			
			\pgfmathsetmacro{\bOne}{1.2}    
			\pgfmathsetmacro{\bTwo}{2.2}    
			\pgfmathsetmacro{\bThree}{1.6}  
			\pgfmathsetmacro{\totalSize}{\bOne+\bTwo+\bThree}  
			
			\pgfmathsetmacro{\matX}{1.5}
			\pgfmathsetmacro{\matY}{0}
			
			\pgfmathsetmacro{\colOneX}{\matX + \bOne/2}
			\pgfmathsetmacro{\colTwoX}{\matX + \bOne + \bTwo/2}
			\pgfmathsetmacro{\colThreeX}{\matX + \bOne + \bTwo + \bThree/2}
			
			\pgfmathsetmacro{\rowOneY}{\matY + \bThree + \bTwo + \bOne/2}
			\pgfmathsetmacro{\rowTwoY}{\matY + \bThree + \bTwo/2}
			\pgfmathsetmacro{\rowThreeY}{\matY + \bThree/2}
			
			\begin{scope}[shift={(-5.5,0)}]
				
				\fill[orange!40] (0,3.5) rectangle (1.8,4.5);
				\draw[thick] (0,3.5) rectangle (1.8,4.5);
				\node at (0.9,4.0) {$T^{12}$};
				
				\fill[red!30] (0,1.8) rectangle (1.4,2.8);
				\draw[thick] (0,1.8) rectangle (1.4,2.8);
				\node at (0.7,2.3) {$T^{13}$};
				
				\fill[green!40] (0,0) rectangle (1.8,1.3);
				\draw[thick] (0,0) rectangle (1.8,1.3);
				\node at (0.9,0.65) {$T^{23}$};
				
				\node[above, font=\normalsize\bfseries, align=center] at (0.9,4.9) {Contraction\\[-2pt]Matrices};
			\end{scope}
			
			\coordinate (T12Src) at (-3.7, 4.0);
			\coordinate (T13Src) at (-4.1, 2.3);
			\coordinate (T23Src) at (-3.7, 0.65);
			
			\fill[red!20] (\matX,\matY) rectangle ({\matX+\bOne},{\matY+\bThree});
			\node at (\colOneX,\rowThreeY) {\footnotesize$(T^{13})^\top$};
			
			\fill[green!25] ({\matX+\bOne},\matY) rectangle ({\matX+\bOne+\bTwo},{\matY+\bThree});
			\node at (\colTwoX,\rowThreeY) {\footnotesize$(T^{23})^\top$};
			
			\fill[gray!20] ({\matX+\bOne+\bTwo},\matY) rectangle ({\matX+\totalSize},{\matY+\bThree});
			\node at (\colThreeX,\rowThreeY) {$0$};
			
			\fill[orange!25] (\matX,{\matY+\bThree}) rectangle ({\matX+\bOne},{\matY+\bThree+\bTwo});
			\node at (\colOneX,\rowTwoY) {\footnotesize$(T^{12})^\top$};
			
			\fill[gray!20] ({\matX+\bOne},{\matY+\bThree}) rectangle ({\matX+\bOne+\bTwo},{\matY+\bThree+\bTwo});
			\node at (\colTwoX,\rowTwoY) {$0$};
			
			\fill[green!40] ({\matX+\bOne+\bTwo},{\matY+\bThree}) rectangle ({\matX+\totalSize},{\matY+\bThree+\bTwo});
			\node at (\colThreeX,\rowTwoY) {$T^{23}$};
			
			\fill[gray!20] (\matX,{\matY+\bThree+\bTwo}) rectangle ({\matX+\bOne},{\matY+\totalSize});
			\node at (\colOneX,\rowOneY) {$0$};
			
			\fill[orange!40] ({\matX+\bOne},{\matY+\bThree+\bTwo}) rectangle ({\matX+\bOne+\bTwo},{\matY+\totalSize});
			\node at (\colTwoX,\rowOneY) {$T^{12}$};
			
			\fill[red!30] ({\matX+\bOne+\bTwo},{\matY+\bThree+\bTwo}) rectangle ({\matX+\totalSize},{\matY+\totalSize});
			\node at (\colThreeX,\rowOneY) {$T^{13}$};
			
			\draw[very thick] (\matX,\matY) rectangle ({\matX+\totalSize},{\matY+\totalSize});
			\draw[thick] ({\matX+\bOne},\matY) -- ({\matX+\bOne},{\matY+\totalSize});
			\draw[thick] ({\matX+\bOne+\bTwo},\matY) -- ({\matX+\bOne+\bTwo},{\matY+\totalSize});
			\draw[thick] (\matX,{\matY+\bThree}) -- ({\matX+\totalSize},{\matY+\bThree});
			\draw[thick] (\matX,{\matY+\bThree+\bTwo}) -- ({\matX+\totalSize},{\matY+\bThree+\bTwo});
			
			\node[left] at ({\matX-0.15},\rowOneY) {$n_1$};
			\node[left] at ({\matX-0.15},\rowTwoY) {$n_2$};
			\node[left] at ({\matX-0.15},\rowThreeY) {$n_3$};
			
			\node[below] at (\colOneX,{\matY-0.15}) {$n_1$};
			\node[below] at (\colTwoX,{\matY-0.15}) {$n_2$};
			\node[below] at (\colThreeX,{\matY-0.15}) {$n_3$};
			
			\node[above, font=\normalsize\bfseries] at ({\matX+\totalSize/2},{\matY+\totalSize+0.4}) {Unified Matrix $R$};
			
			\draw[-{Stealth}, thick, orange!70!black, dashed] 
			(T12Src) to[out=0, in=135] (\colTwoX, \rowOneY);
			
			\draw[-{Stealth}, thick, red!60!black, dashed] 
			(T13Src) to[out=0, in=150] (\colThreeX, \rowOneY);
			
			\draw[-{Stealth}, thick, green!60!black, dashed] 
			(T23Src) to[out=0, in=180] (\colThreeX, \rowTwoY);
			
		\end{tikzpicture}
	\end{center}
	
	\begin{equation}
		R = \Phi_3\left(\mathcal{T}; a^{(1)}, a^{(2)}, a^{(3)}\right) := 
		\begin{pmatrix}
			0_{n_1 \times n_1} & T^{12} & T^{13} \\[4pt]
			(T^{12})^\top & 0_{n_2 \times n_2} & T^{23} \\[4pt]
			(T^{13})^\top & (T^{23})^\top & 0_{n_3 \times n_3}
		\end{pmatrix} \in \mathbb{R}^{N \times N}.
	\end{equation}
	
	\noindent This block-symmetric structure with zero diagonal blocks encapsulates all pairwise mode contractions. The construction generalizes to $d$th-order tensors, yielding a symmetric block matrix with $\binom{d}{2}$ distinct off-diagonal blocks.
\end{remark}

\begin{remark}\label{Rem of close form}
    As mentioned in Remark 3.2 of the main manuscript, when $d=3$ and $\mfc_1=\mfc_2=\mfc_3=1/3$ and $\kappa_3=\kappa_4=0$, we have explicit expressions for the mean function $\mu_N^{(3)}(z)$ and the covariance function $\mcC_N^{(3)}(z_1,z_2)$ as follows:
\begin{footnotesize}
    \begin{align*}
        &\mu_N(z)=\frac{108g^9(z)+270g^7(z)-1782g^5(z)}{8g^{10}(z)+36g^8(z)-378g^6(z)-189g^4(z)+2430g^2(z)-2187},\\
        &\mcC_N^{(3)}(z_1,z_2)=\frac{-324g^4(z_1)g^4(z_2)+1944g^3(z_1)g^3(z_2)}{16g^6(z_1)g^6(z_2)-216g^4(z_1)g^4(z_2)+216g^3(z_1)g^3(z_2)+729g^2(z_1)g^2(z_2)-1458g(z_1)g(z_2)+729},
    \end{align*}
\end{footnotesize}\noindent
where $g(z)=\frac{3}{4}\left(\sqrt{z^2-\frac{8}{3}}-z\right)$.

\begin{tikzpicture}[
	scale=1.2,
	>={Stealth[length=5pt]},
	contour/.style={thick, blue!70!black, 
		decoration={markings, 
			mark=at position 0.12 with {\arrow{>}},
			mark=at position 0.37 with {\arrow{>}},
			mark=at position 0.62 with {\arrow{>}},
			mark=at position 0.87 with {\arrow{>}}
		}, postaction={decorate}},
	contour1/.style={thick, blue!70!black, 
		decoration={markings, 
			mark=at position 0.12 with {\arrow{>}},
			mark=at position 0.37 with {\arrow{>}},
			mark=at position 0.62 with {\arrow{>}},
			mark=at position 0.87 with {\arrow{>}}
		}, postaction={decorate}},
	contour2/.style={thick, red!70!black, 
		decoration={markings, 
			mark=at position 0.12 with {\arrow{>}},
			mark=at position 0.37 with {\arrow{>}},
			mark=at position 0.62 with {\arrow{>}},
			mark=at position 0.87 with {\arrow{>}}
		}, postaction={decorate}}
	]
	
	\begin{scope}[shift={(-3.2,0)}]
		\draw[->,gray] (-2.3,0) -- (2.3,0) node[right] {\small $\mathrm{Re}$};
		\draw[->,gray] (0,-1.6) -- (0,1.6) node[above] {\small $\mathrm{Im}$};
		
		\draw[line width=3pt, black!50!green] (-1.3,0) -- (1.3,0);
		\node[below=6pt, black!50!green, font=\footnotesize] at (0,0) {$\mathrm{supp}(\nu)$};
		
		\fill[black!50!green] (-1.3,0) circle (1.8pt);
		\fill[black!50!green] (1.3,0) circle (1.8pt);
		
		\draw[contour] (0,0) ellipse (1.9 and 1.1);
		\node[blue!70!black] at (2.15,0.9) {$\mathcal{C}$};
		
		\fill[blue!70!black] (1.35,0.78) circle (2pt) node[above right] {$z$};
		
		\node[font=\bfseries] at (0,2.1) {Mean};
		
		\node[align=center] at (0,-2.2) {
			$\displaystyle \mathbb{E}[X_f] = \frac{1}{2\pi \mathrm{i}} \oint_{\mathcal{C}} f(z)\, \mu_N(z)\, dz$
		};
	\end{scope}
	
	\begin{scope}[shift={(3.2,0)}]
		\draw[->,gray] (-2.3,0) -- (2.3,0) node[right] {\small $\mathrm{Re}$};
		\draw[->,gray] (0,-1.6) -- (0,1.6) node[above] {\small $\mathrm{Im}$};
		
		\draw[line width=3pt, black!50!green] (-1.3,0) -- (1.3,0);
		\node[below=6pt, black!50!green, font=\footnotesize] at (0,0) {$\mathrm{supp}(\nu)$};
		
		\fill[black!50!green] (-1.3,0) circle (1.8pt);
		\fill[black!50!green] (1.3,0) circle (1.8pt);
		
		\draw[contour1] (0,0) ellipse (1.9 and 1.1);
		\node[blue!70!black] at (2.15,0.9) {$\mathcal{C}_1$};
		
		\draw[contour2] (0,0) ellipse (1.55 and 0.75);
		\node[red!70!black] at (1.8,0.5) {$\mathcal{C}_2$};
		
		\fill[blue!70!black] (1.35,0.78) circle (2pt) node[above right] {$z_1$};
		\fill[red!70!black] (1.35,0.41) circle (2pt) node[right] {$z_2$};
		
		\node[font=\bfseries] at (0,2.1) {Covariance};
		
		\node[align=center] at (1,-2.2) {
			$\displaystyle \mathrm{Cov}(X_f, X_g) = \frac{-1}{4\pi^2} \oint_{\mathcal{C}_1} \oint_{\mathcal{C}_2} f(z_1) g(z_2)\, \mathcal{C}_N(z_1,z_2)\, dz_1 dz_2$
		};
	\end{scope}
	
	
	\draw[gray, dashed] (0,-2.8) -- (0,2.5);
	
\end{tikzpicture}

\end{remark}
\begin{pro}
    Under Assumptions 2.1 and 2.2 in the main manuscript, if $\max_{1\leq l\leq 3}\mathfrak{r}_{r_0,l}\overset{a.s.}{\longrightarrow}0$, then
    \begin{align*}
        \widehat{T}_{r_0,N}^{(3)}(\hat{\bbx}^{(r_0,1)},\hat{\bbx}^{(r_0,2)},\hat{\bbx}^{(r_0,3)})\overset{d}{\longrightarrow}\widehat{T}_{r_0,N}^{(3)}(\bbx^{(r_0,1)},\bbx^{(r_0,2)},\bbx^{(r_0,3)}).
    \end{align*}
\end{pro}
\begin{proof}
    Recall that $\mathfrak{r}_{r_0,l}=1-|\langle\bbx^{(r_0,l)},\hat{\bbx}^{(r_0,l)}\rangle|$ for $1\leq r_0\leq R$ and $1\leq l\leq3$, when $\max_{1\leq l\leq 3}\mathfrak{r}_{r_0,l}\overset{a.s.}{\longrightarrow}0$, it implies that
    \begin{align*}
        &\big|\langle\bby^{(r_1,l)},\bbx^{(r_0,l)}\rangle^2-\langle\bby^{(r_1,l)},\hat{\bbx}^{(r_0,l)}\rangle^2\big|\leq\big|\langle\bby^{(r_1,l)},\bbx^{(r_0,l)}-\hat{\bbx}^{(r_0,l)}\rangle\big|\cdot\big|\langle\bby^{(r_1,l)},\bbx^{(r_0,l)}+\hat{\bbx}^{(r_0,l)}\rangle\big|\\
        &\leq\sqrt{2}\min\{\Vert\bbx^{(r_0,l)}-\hat{\bbx}^{(r_0,l)}\Vert_2,\Vert\bbx^{(r_0,l)}+\hat{\bbx}^{(r_0,l)}\Vert_2\}\leq2\mathfrak{r}_{r_0,l}^{1/2},
    \end{align*}
    then
    \begin{align}
        &\big|\mcD^{(r_0,3)}-\hat{\mcD}^{(r_0,3)}\big|=2\Bigg|\sum_{r_1=1}^{R_1}\beta_{1,r_1}^2\sum_{l=1}^3\langle\bby^{(r_1,l)},\bbx^{(r_0,l)}\rangle^2-\sum_{r_1=1}^{R_1}\beta_{1,r_1}^2\sum_{l=1}^3\langle\bby^{(r_1,l)},\hat{\bbx}^{(r_0,l)}\rangle^2\Bigg|\notag\\
        &\leq2\sum_{r_1=1}^{R_1}\beta_{r_1}^2\sum_{l=1}^3\mathfrak{r}_{r_0,l}^{1/2}\leq6\max_{1\leq l\leq 3}\mathfrak{r}_{r_0,l}^{1/2}\sum_{r_1=1}^{R_1}\beta_{r_1}^2\overset{a.s.}{\longrightarrow}0.\label{Eq of mean drift convergence}
    \end{align}
    Moreover, conditional on $\{\hat{\bbx}^{(r_0,l)}:1\leq l\leq3\}$,
    \begin{align}
         &\big(\widehat{T}_{r_0,N}^{(3)}(\hat{\bbx}^{(r_0,1)},\hat{\bbx}^{(r_0,2)},\hat{\bbx}^{(r_0,3)})-\hat{\mcD}^{(r_0,3)}-\hat{\xi}_N^{(r_0,3)}\big)/\hat{\sigma}_N^{(r_0,3)}\label{Eq of empirical statistic CLT}\\
        &=\left(\Vert\bbT(\hat{\bbx}^{(r_0,1)},\hat{\bbx}^{(r_0,2)},\hat{\bbx}^{(r_0,3)})\Vert_2^2-N\int_{-\infty}^{\infty}x^2\nu(dx)-\hat{\mcD}^{(r_0,3)}-\hat{\xi}_N^{(r_0,3)}\right)/\hat{\sigma}_N^{(r_0,3)}\overset{d}{\longrightarrow}\mcN(0,1),\notag
    \end{align}
    where
    \begin{align*}
		\hat{\xi}_N^{(r_0,3)}&=-\frac{1}{2\pi{\rm i}}\oint_{\mfC_1}z^2\mu_N^{(3)}(z;\kappa_3,\kappa_4,\hat{\bbx}^{(r_0,1)},\hat{\bbx}^{(r_0,2)},\hat{\bbx}^{(r_0,3)})dz,\\
        \big(\hat{\sigma}_N^{(r_0,3)}\big)^2&=-\frac{1}{4\pi^2}\oint_{\mfC_1}\oint_{\mfC_2}z_1^2z_2^2\mcC_N^{(3)}(z_1,z_2;\kappa_4,\hat{\bbx}^{(r_0,1)},\hat{\bbx}^{(r_0,2)},\hat{\bbx}^{(r_0,3)})dz_1dz_2.
	\end{align*}
    By Propositions 3.1 and 3.2 in the main manuscript, the mean function $\mu_N^{(3)}$ and the covariance function $\mcC_N^{(3)}$ contain terms $\mfb_l^{(1)}=\langle\boldsymbol{1}_{n_l},\hat{\bbx}^{(r_0,l)}\rangle/\sqrt{N}$ and $\mcB_{(4)}^{(i,k)}=\Vert\hat{\bbx}^{(r_0,l)}\Vert_4^4$ for $\{i,k,l\}=\{1,2,3\}$ relating with $\{\hat{\bbx}^{(r_0,1)},\hat{\bbx}^{(r_0,2)},\hat{\bbx}^{(r_0,3)}\}$. Note that
    \begin{align*}
        &\big|\langle\boldsymbol{1}_{n_l},\hat{\bbx}^{(r_0,l)}-\bbx^{(r_0,l)}\rangle/\sqrt{N}\big|\leq\Vert\hat{\bbx}^{(r_0,l)}-\bbx^{(r_0,l)}\Vert_2\leq\sqrt{2}\mathfrak{r}_{r_0,l}^{1/2}\overset{a.s.}{\longrightarrow}0,\\
        &\big|\Vert\hat{\bbx}^{(r_0,l)}\Vert_4-\Vert\bbx^{(r_0,l)}\Vert_4\big|\leq\Vert\hat{\bbx}^{(r_0,l)}-\bbx^{(r_0,l)}\Vert_4\leq2^{1/4}\Vert\hat{\bbx}^{(r_0,l)}-\bbx^{(r_0,l)}\Vert_2^{1/2}\leq\sqrt{2}\mathfrak{r}_{r_0,l}^{1/4}\overset{a.s.}{\longrightarrow}0,
    \end{align*}
    which implies that 
    \begin{align}
        \hat{\xi}_N^{(r_0,3)}\overset{a.s.}{\longrightarrow}\xi_N^{(r_0,3)}\quad\text{and}\quad\hat{\sigma}_N^{(r_0,3)}\overset{a.s.}{\longrightarrow}\sigma_N^{(r_0,3)},\label{Eq of mean and variance convergence}
    \end{align}
    where
    \begin{align*}
		\xi_N^{(r_0,3)}&=-\frac{1}{2\pi{\rm i}}\oint_{\mfC_1}z^2\mu_N^{(3)}(z;\kappa_3,\kappa_4,\bbx^{(r_0,1)},\bbx^{(r_0,2)},\bbx^{(r_0,3)})dz,\\
        \big(\sigma_N^{(r_0,3)}\big)^2&=-\frac{1}{4\pi^2}\oint_{\mfC_1}\oint_{\mfC_2}z_1^2z_2^2\mcC_N^{(3)}(z_1,z_2;\kappa_4,\bbx^{(r_0,1)},\bbx^{(r_0,2)},\bbx^{(r_0,3)})dz_1dz_2.
	\end{align*}
    Finally, combined with \eqref{Eq of mean drift convergence}, \eqref{Eq of empirical statistic CLT} and \eqref{Eq of mean and variance convergence}, we have
    \begin{align*}
         \big(\widehat{T}_{r_0,N}^{(3)}(\hat{\bbx}^{(r_0,1)},\hat{\bbx}^{(r_0,2)},\hat{\bbx}^{(r_0,3)})-\mcD^{(r_0,3)}-\xi_N^{(r_0,3)}\big)/\sigma_N^{(r_0,3)}\overset{d}{\longrightarrow}\mcN(0,1),
    \end{align*}
    by Proposition 4.2 in the main manuscript, we complete our proof.
\end{proof}

\begin{remark}[Power under imperfect estimation]\label{Rem of power imperfect}
    Even when $\max_{1\leq l\leq 3}\mathfrak{r}_{r_0,l}$ does not vanish asymptotically, for instance, when the reference signal $\beta_{r_0,0}$ is moderate rather than diverging, the test based on estimated directions can still achieve substantial power. To illustrate, consider the rank-1 case with parallel signals:
    $$\bbT^{(0)}=\beta_{1,0}\bbx^{(1)}\otimes\bbx^{(2)}\otimes\bbx^{(3)}+\frac{1}{\sqrt{N}}\bbX^{(0)},\quad
    \bbT^{(1)}=\beta_{1,1}\bbx^{(1)}\otimes\bbx^{(2)}\otimes\bbx^{(3)}+\frac{1}{\sqrt{N}}\bbX^{(1)}.$$
    Let $\hat{\bbx}^{(l)}$ denote estimates obtained from $\bbT^{(0)}$. Under $H_1$ (signal matching), the test statistic satisfies
    $$\big(\widehat{T}_{1,N}^{(3)}(\hat{\bbx}^{(1)},\hat{\bbx}^{(2)},\hat{\bbx}^{(3)})-\hat{\xi}_N^{(1,3)}-\hat{\mcD}^{(1,3)}\big)/\hat{\sigma}_N^{(1,3)}\overset{d}{\longrightarrow}\mcN(0,1),$$
    where the effective drift is $\hat{\mcD}^{(1,3)}=2\beta_{1,1}^2\sum_{l=1}^3(1-\mathfrak{r}_{1,l})^2$. The asymptotic power is therefore $1-\Phi(z_{\alpha}-\hat{\mcD}^{(1,3)}/\hat{\sigma}_N^{(1,3)})$. This implies that even with nonvanishing estimation error, the power remains high provided $\beta_{1,1}$ is not too small. For example, using the bound $\mathfrak{r}_{1,l}\leq C\beta_{1,0}^{-2}$ from \cite{Seddik2024203}, we have
    $$\hat{\mcD}^{(1,3)}\geq 6\beta_{1,1}^2(1-C\beta_{1,0}^{-2})^2,$$
    which approaches the oracle drift $6\beta_{1,1}^2$ as $\beta_{1,0}$ increases. 
\end{remark}

\begin{remark}[On Statistical Motivation and High-Dimensional Geometry]
\label{rmk:motivation}
A natural question concerns testing signal alignment in high dimensions, given that random unit vectors tend to be nearly orthogonal to any fixed vector. We clarify that the reference directions $\mathbf{a}^{(l)}$ in our framework are not arbitrary vectors. They arise from scientific hypotheses, domain knowledge, or are estimated from a reference tensor containing signal. In the meantime, the tendency of random vectors toward orthogonality is precisely what gives our test its power: rejection of the null provides strong evidence that the observed alignment reflects genuine signal structure rather than chance. The ``near-orthogonality by default'' phenomenon serves as a natural baseline against which detected alignments become statistically meaningful.

On the other hand, unlike matrix singular vectors, CP tensor components lack orthogonality across ranks. This, combined with identifiability challenges, motivates directly testing alignment without requiring full decomposition--advantageous when signals are moderate. Complementary tests for parallelism among tensor singular vectors remain an interesting future direction.
\end{remark}

}

\section{Numerical Experiments}\label{sec of Numerical}
In this section, we conduct numerical experiments to investigate the performance of our theorems and hypothesis tests. First, we provide several examples to demonstrate the validity of our CLT results presented in Theorem \ref{Main of Thm of CLT LSS d=3}. 

{As discussed in \S\ref{sec:intro}, the tensor alignment test \eqref{Main of Eq of hypothesis test 1} is closely connected to tensor-based classification problems. In \S\ref{Main of sec of experiment 2}, we consider the setting where the reference directions $\bba^{(l)}$ are known, and examine the empirical power of our test statistic across varying signal strengths $\beta$. Finally, in \S\ref{Main of sec of experiment 3}, we turn to the more realistic scenario where prior information is available only in the form of a reference tensor rather than explicit signal directions. There, we demonstrate the performance of the tensor signal matching test developed in \S\ref{sec of generalized procedure}.}  

For simplicity of presentation, we focus on the case \(d=3\) and \(n_1=n_2=n_3=100\), i.e., \(\mfc_1=\mfc_2=\mfc_3=1/3\). The LSD $\nu(x)$ is obtained by solving  (\ref{Main of Eq of MDE 3 order}). We have
	\begin{align}
		g(z)=\frac{3}{4}\Bigg(\sqrt{z^2-\frac{8}{3}}-z\Bigg),\quad{\rm and}\quad\nu(x)=\frac{3}{4\pi}\sqrt{\frac{8}{3}-x^2}, \quad |x|\le\sqrt{\frac{8}{3}}. \label{Main of Eq of numerical LSD}
	\end{align}
{Additional numerical experiments under alternative settings are provided in \S A.1 of the supplement.}

 \subsection{Experiment 1: verification of the CLT}\label{Main of sec of experiment 1}

In this subsection,  we compare the empirical values of $\mathbb{E}[G_N(f)]$ and $\Var(G_N(f))$ with their theoretical limiting values given in Equations \eqref{Main of Eq of LSS mean d=3} and \eqref{Main of Eq of LSS variance d=3}, respectively. Additionally, we assess the normality of the statistics using quantile-quantile plots.

 To further illustrate the influence of unit vectors ${\bba^{(1)}, \bba^{(2)}, \bba^{(3)}}$ and the cumulants $\kappa_3$ and $\kappa_4$ of random noises on the asymptotic mean $\xi_N^{(d)}$ and variance $\sigma_N^{(3)}$ of the CLT for the LSS, as discussed in Proposition \ref{Rem of comparison} and Remark \ref{Rem of vectors type}, we consider several different test functions, two types of vector selection and two distributions for entries of  noise tensors.  The results are summarized in Table \ref{Tab of LSS mean and variance}.

Specifically, for vector selection, we consider $\bba^{(l)} = (1, 0, \ldots, 0)', l=1,2,3$, which we abbreviate as "localized" vectors, and $\bba^{(l)} = n_l^{-1/2}(1, \ldots, 1)',l=1,2,3$, abbreviated as "delocalized" vectors. When $\bba^{(l)} = (1, 0, \ldots, 0)'$, $l=1,2,3$, $\mu_N^{(3)}$ will be asymptotically independent of $\kappa_3$. If all $\bba^{(l)}$ are delocalized, both $\mu_N^{(3)}$ and $\sigma_N^{(3)}$ become independent of $\kappa_4$ as $N \to \infty$. For the noise  tensors, we consider those with elements following a standard normal distribution, $\mathcal{N}(0, 1)$, which has zero third and fourth cumulants, and those with elements uniformly distributed on $[-\sqrt{3}, \sqrt{3}]$, denoted as $\text{Unif}(\pm\sqrt{3})$, which have a third cumulant of 0 and a fourth cumulant of $-1.2$, { and those with elements following the centered and normalized binomial distribution $B(3,0.25)$, which have a third cumulant of $2/3$ and a fourth cumulant of $-2/9$.} When the noise follows a normal distribution, the influence of $\kappa_3$ and $\kappa_4$ on the asymptotic mean $\xi_N^{(d)}$ and variance $\sigma_N^{(3)}$ of LSS vanishes.

	\begin{table}[!hbtp]
		\centering
		\renewcommand\arraystretch{1}
        \caption{{Empirical mean and standard deviation of $G_N(f)$ from 1000 independent trials versus theoretical limits in (3.8) and (3.9), for $n_1=n_2=n_3=100$ with various test functions, noise distributions, and vector types.} }
		\begin{tabular}{ccc|cc|cc}
                &&& \multicolumn{2}{c}{\(\mbE[G_N(f)]\) } & \multicolumn{2}{c}{ \(\operatorname{Std}(G_N(f))\)}\\
			\(f(x)\)&Noise type&Vector type&Empirical &Limit&Empirical &Limit\\
			\hline
			\(x^2\)&\(\mathcal{N}(0,1)\)&all types&\(0.0240\)&\(0\) &\(1.6206\)&\(1.6218\)\\
			\(3^x\)&\(\mathcal{N}(0,1)\)&all types&\(0.1463\)&\(0.1442\)&\(1.3051\)&\(1.3092\)\\
			\(\cos(2x)\)&\(\mathcal{N}(0,1)\)&all types&\(0.7092\)&\(0.7247\)&\(1.2568\)&\(1.2722\)\\[1mm]\hline
			\(x^2\)&\({\rm Unif}(\pm\sqrt{3})\)&localized&\(-0.0132\)&\(0\)&\(1.0069\)&\(1.0259\)\\
			\(3^x\)&\({\rm Unif}(\pm\sqrt{3})\)&localized&\(0.0793\)&\(0.0872\)&\(0.8408\)&\(0.8541\)\\
			\(\cos(2x)\)&\({\rm Unif}(\pm\sqrt{3})\)&localized&\(0.4306\)&\(0.4200\)&\(0.8978\)&\(0.8903\)\\[1mm]\hline
			\(x^2\)&\({\rm Unif}(\pm\sqrt{3})\)&delocalized&\(0.0152\)&\(0\)&\(1.6297\)&\(1.6218\)\\
			\(3^x\)&\({\rm Unif}(\pm\sqrt{3})\)&delocalized&\(0.1617\)&\(0.1442\)&\(1.3242\)&\(1.3092\)\\
			\(\cos(2x)\)&\({\rm Unif}(\pm\sqrt{3})\)&delocalized&\(0.7101\)&\(0.7247\)&\(1.2593\)&\(1.2722\)\\[1mm]\hline
            \(x^2\)&\(B(1,0.25)\)&localized&\(-0.0204\)&\(0\)&\(1.3226\)&\(1.3243\)\\
			\(3^x\)&\(B(1,0.25)\)&localized&\(0.1008\)&\(0.1125\)&\(1.0953\)&\(1.0803\)\\
			\(\cos(2x)\)&\(B(1,0.25)\)&localized&\(0.5704\)&\(0.5554\)&\(1.0812\)&\(1.0706\)\\[1mm]\hline
            \(x^2\)&\(B(1,0.25)\)&delocalized&\(0.0121\)&\(0\)&\(1.6312\)&\(1.6218\)\\
			\(3^x\)&\(B(1,0.25)\)&delocalized&\(0.5141\)&\(0.5031\)&\(1.3513\)&\(1.3092\)\\
			\(\cos(2x)\)&\(B(1,0.25)\)&delocalized&\(0.7220\)&\(0.7247\)&\(1.2695\)&\(1.2722\)\\[1mm]\hline
            \(x^2\)&\(B(3,0.25)\)&localized&\(0.0072\)&\(0\)&\(1.5288\)&\(1.5291\)\\
			\(3^x\)&\(B(3,0.25)\)&localized&\(0.1446\)&\(0.1336\)&\(1.2409\)&\(1.2376\)\\
			\(\cos(2x)\)&\(B(3,0.25)\)&localized&\(0.6486\)&\(0.6683\)&\(1.2143\)&\(1.2106\)\\[1mm]\hline
            \(x^2\)&\(B(3,0.25)\)&delocalized&\(0.0151\)&\(0\)&\(1.5961\)&\(1.6218\)\\
			\(3^x\)&\(B(3,0.25)\)&delocalized&\(0.3756\)&\(0.3514\)&\(1.3152\)&\(1.3092\)\\
			\(\cos(2x)\)&\(B(3,0.25)\)&delocalized&\(0.7315\)&\(0.7247\)&\(1.2957\)&\(1.2722\)\\[1mm]\hline
            \(x^2\)&\(B(5,0.25)\)&localized&\(-0.0164\)&\(0\)&\(1.5630\)&\(1.5668\)\\
			\(3^x\)&\(B(5,0.25)\)&localized&\(0.1213\)&\(0.1378\)&\(1.2744\)&\(1.2667\)\\
			\(\cos(2x)\)&\(B(5,0.25)\)&localized&\(0.6690\)&\(0.6909\)&\(1.2657\)&\(1.2356\)\\[1mm]\hline
            \(x^2\)&\(B(5,0.25)\)&delocalized&\(-0.103\)&\(0\)&\(1.6554\)&\(1.6218\)\\
			\(3^x\)&\(B(5,0.25)\)&delocalized&\(0.3128\)&\(0.3047\)&\(1.3640\)&\(1.3092\)\\
			\(\cos(2x)\)&\(B(5,0.25)\)&delocalized&\(0.7534\)&\(0.7247\)&\(1.2696\)&\(1.2722\)\\[1mm]\hline
		\end{tabular}
		\label{Tab of LSS mean and variance}
	\end{table}
    {Table \ref{Tab of LSS mean and variance} shows close agreement between empirical and theoretical values. As predicted by Proposition 3.3, the limiting values under Gaussian noise are identical across vector types. The QQ plots in Figure \ref{Fig of 1} further validate the asymptotic normality of $G_N(f)$.}
    {Table \ref{Tab of LSS mean and variance} shows close agreement between empirical and theoretical values. As predicted by Proposition \ref{Rem of comparison}, the limiting values under Gaussian noise are identical across vector types. The QQ plots in Figure \ref{Fig of 1} further validate the asymptotic normality of $G_N(f)$.}
  
    \begin{figure}[!htbp]
		\centering
		\subfigure[\(f(x)=x^2,\bbX\sim\mcN(0,1)\), delocalized \(\bbx^{(i)}\).]{\includegraphics[width=0.49\linewidth]{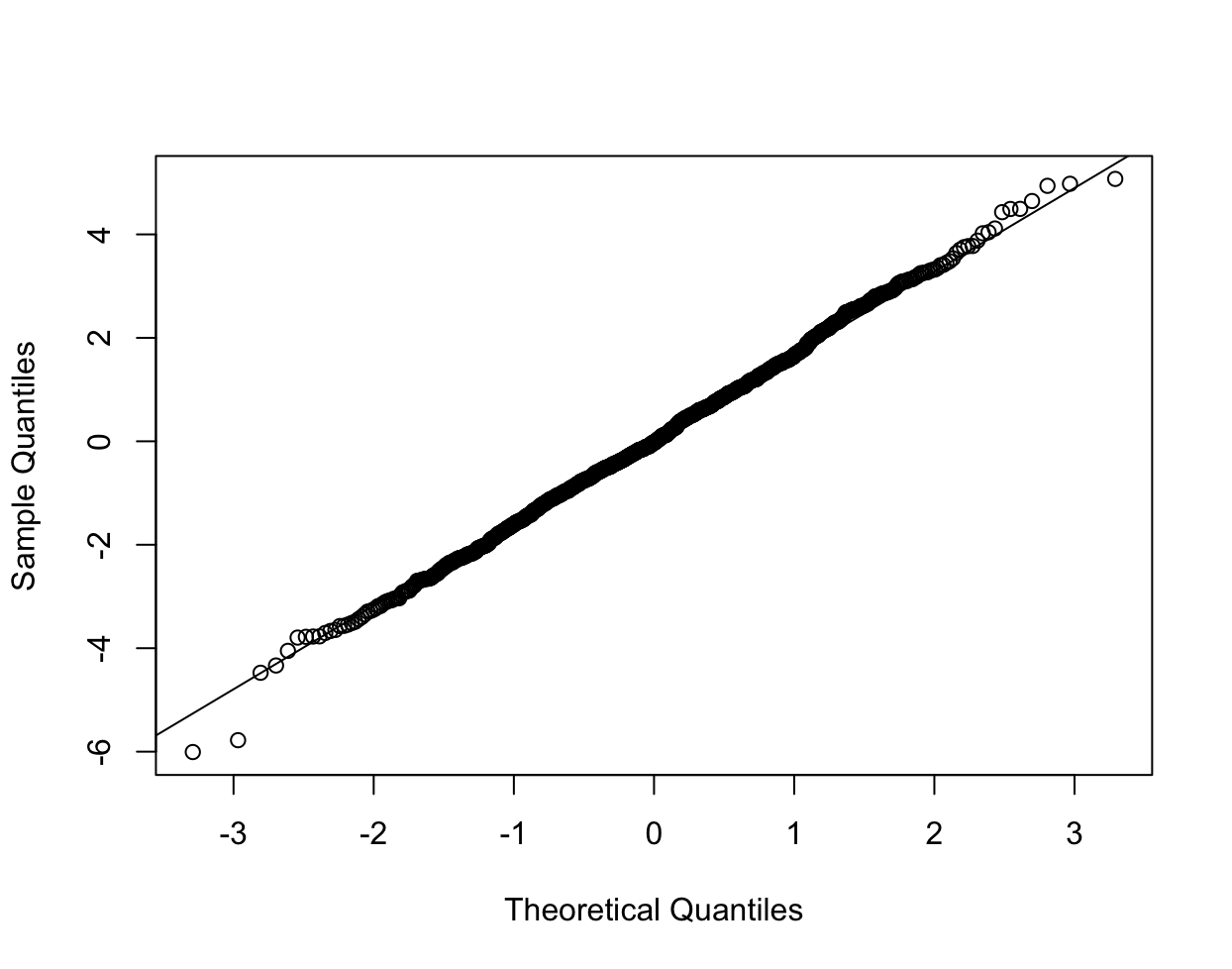}}
		\subfigure[\(f(x)=3^x,\bbX\sim{\rm Unif}(\pm\sqrt{3})\), localized \(\bbx^{(i)}\).]{\includegraphics[width=0.49\linewidth]{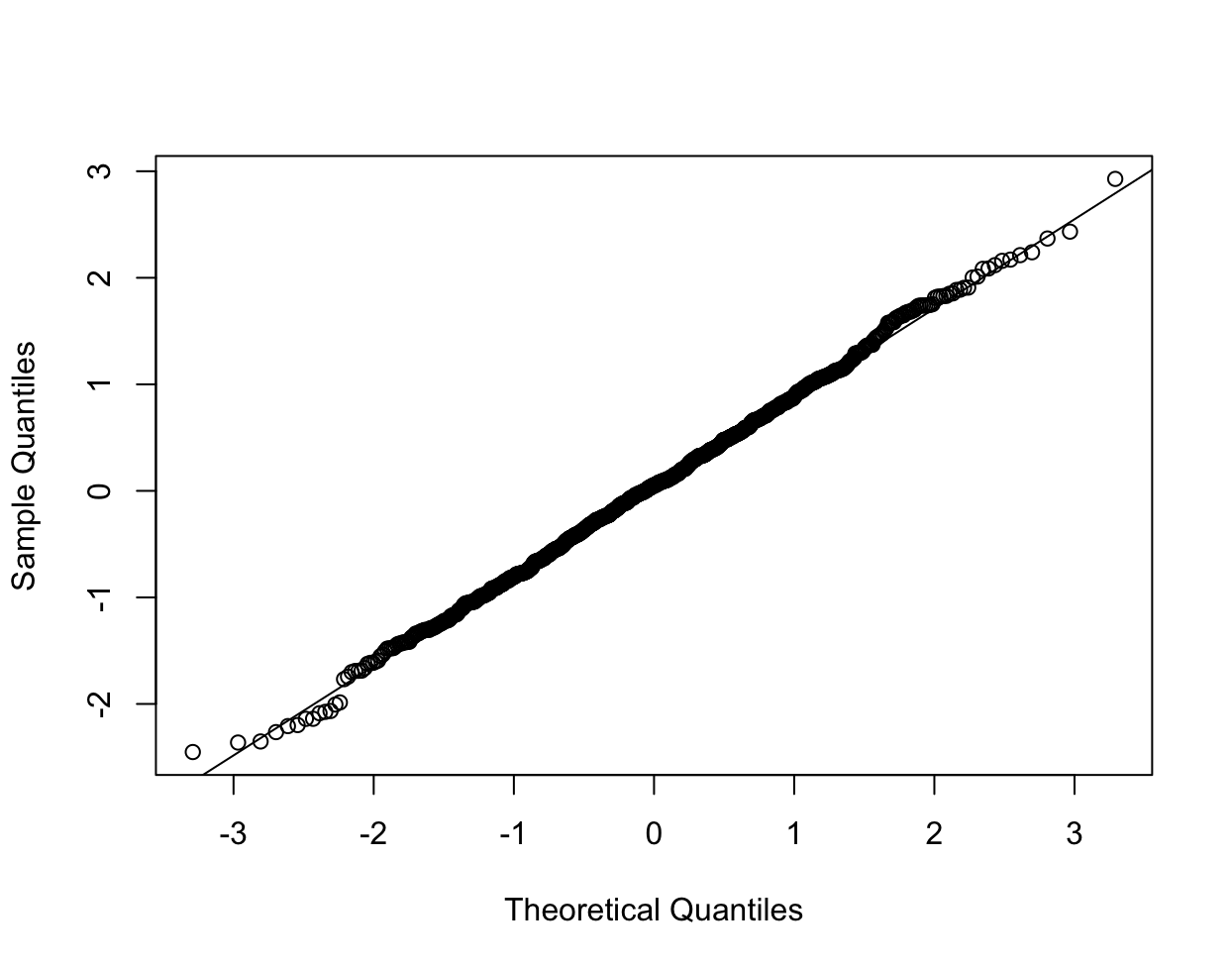}}
		\caption{QQ plots of \(G_N(f)\) from 1000 independent repetitions.}
		\label{Fig of 1}
	\end{figure}
	
\subsection{Experiment 2: tensor signal alignment test}\label{Main of sec of experiment 2}	
This experiment focuses on the tensor signal alignment test (\ref{Main of Eq of hypothesis test 1}). We generate the observation \(\bbT\) using (\ref{Main of Eq of spiky tensor model}) with varying values of $\beta$. We are particularly interested in the test's performance when the signal is below the phase transition threshold, i.e., \(\beta\in(0,\beta_s]\). { For the symmetric case $\mfc_1=\mfc_2=\mfc_3=1/3$, the phase transition threshold is $\beta_s=2/\sqrt{3}$ (see Corollary 3 of \cite{Seddik2024203}). This threshold characterizes the boundary below which consistent signal recovery via maximum likelihood estimation becomes impossible.}  According to (\ref{Main of Eq of statistic H0 H1}), we then have
	\begin{align*}
		\left\{\begin{array}{ll}
			\tilde{\mcT}_N^{(3)}(\bbx^{(1)},\bbx^{(2)},\bbx^{(3)})\overset{d}{\longrightarrow}\mcN(0,1),&{\rm under\ }H_0,\\
			\tilde{\mcT}_N^{(3)}(\bbx^{(1)},\bbx^{(2)},\bbx^{(3)})-6\beta^2/\sigma_N^{(3)}\overset{d}{\longrightarrow}\mcN(0,1),&{\rm under\ }H_1.
		\end{array}\right.
	\end{align*}
In Figure \ref{Fig of 2}, we use the same settings as in \S\ref{Main of sec of experiment 1}, with a significance level of $\alpha=0.05$.  We compute the test's empirical power for different $\beta$ values with 200 repetitions.


\begin{figure}
		\subfigure[\(\bbX\sim\mcN(0,1)\), delocalized \(\bbx^{(i)}\).]{\includegraphics[width=0.49\linewidth,height=3cm]{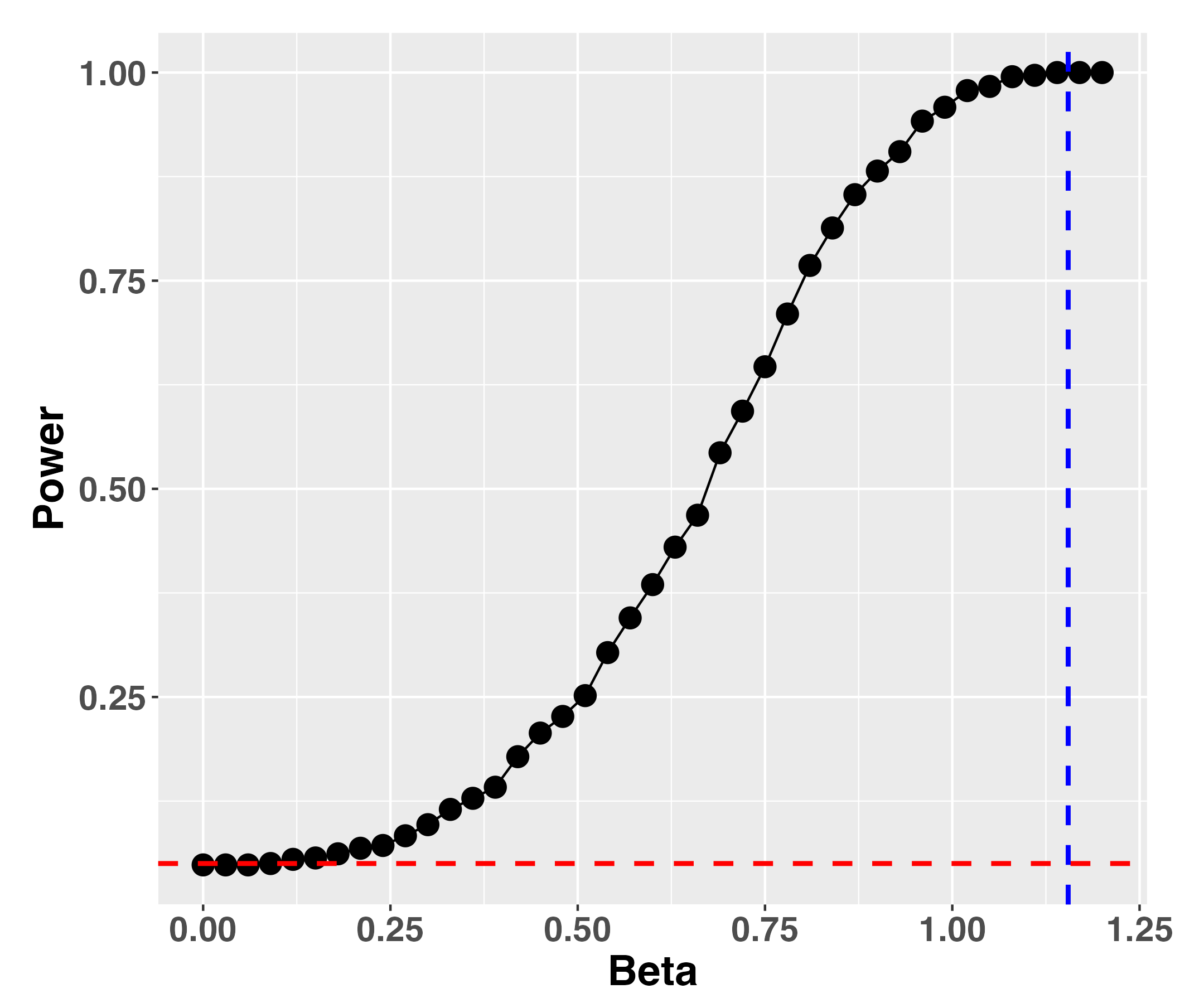}}
		\subfigure[\(\bbX\sim{\rm Unif}(\pm\sqrt{3})\), delocalized \(\bbx^{(i)}\).]{\includegraphics[width=0.49\linewidth,height=3cm]{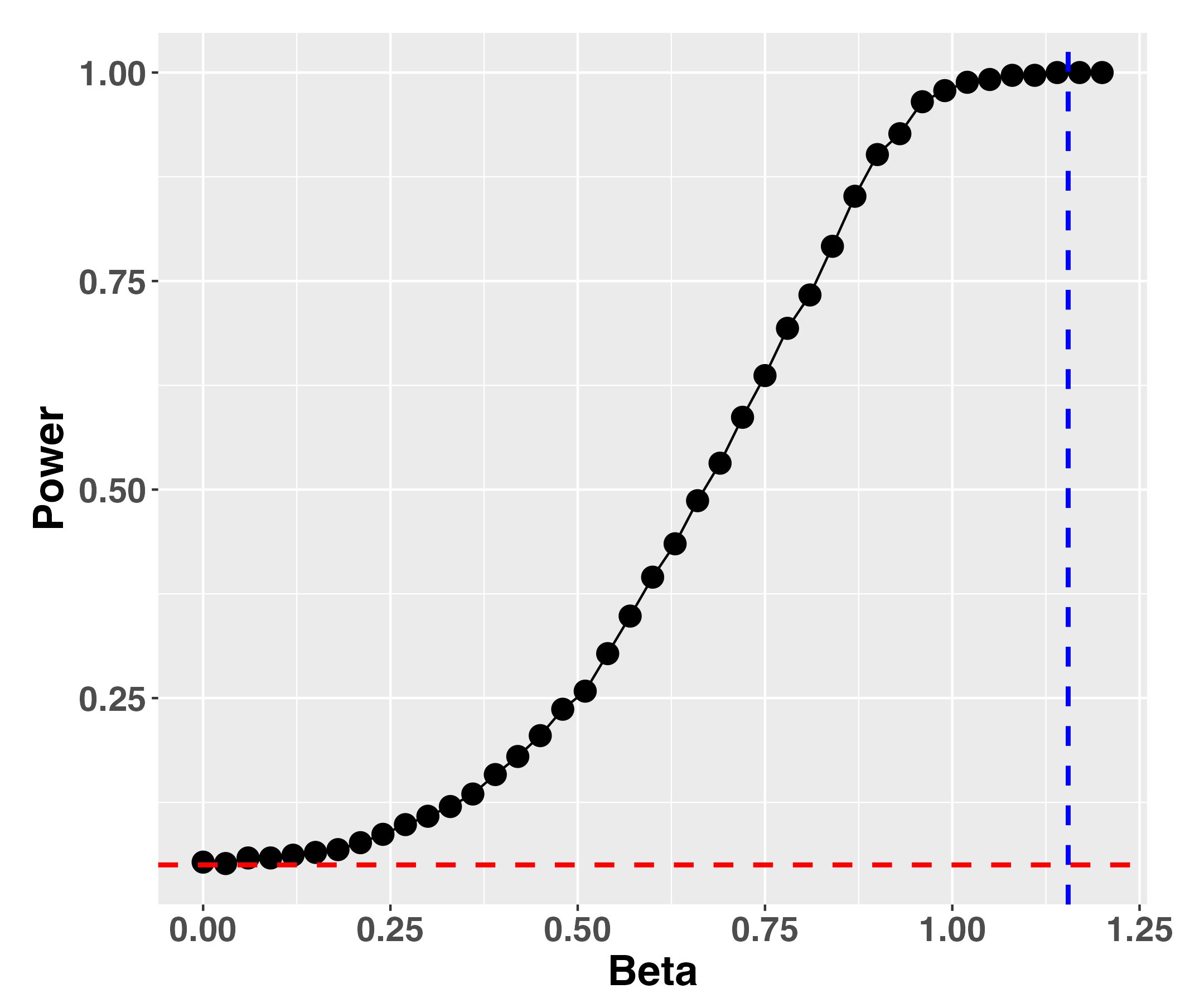}}
        \subfigure[\(\bbX\sim\mcN(0,1)\), localized \(\bbx^{(i)}\).]{\includegraphics[width=0.49\linewidth,height=3cm]{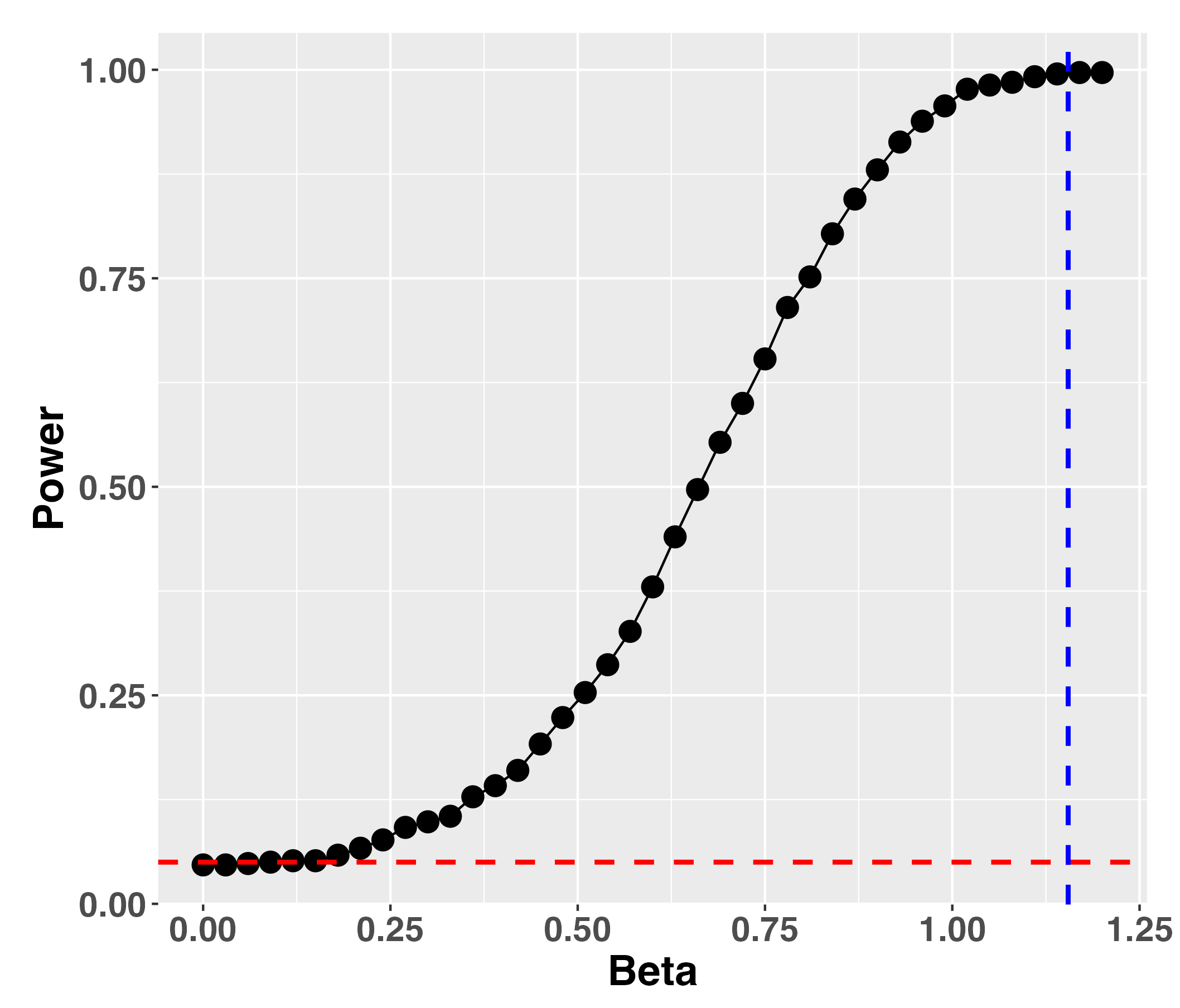}}
        \subfigure[\(\bbX\sim{\rm Unif}(\pm\sqrt{3})\), localized \(\bbx^{(i)}\).]{\includegraphics[width=0.49\linewidth,height=3cm]{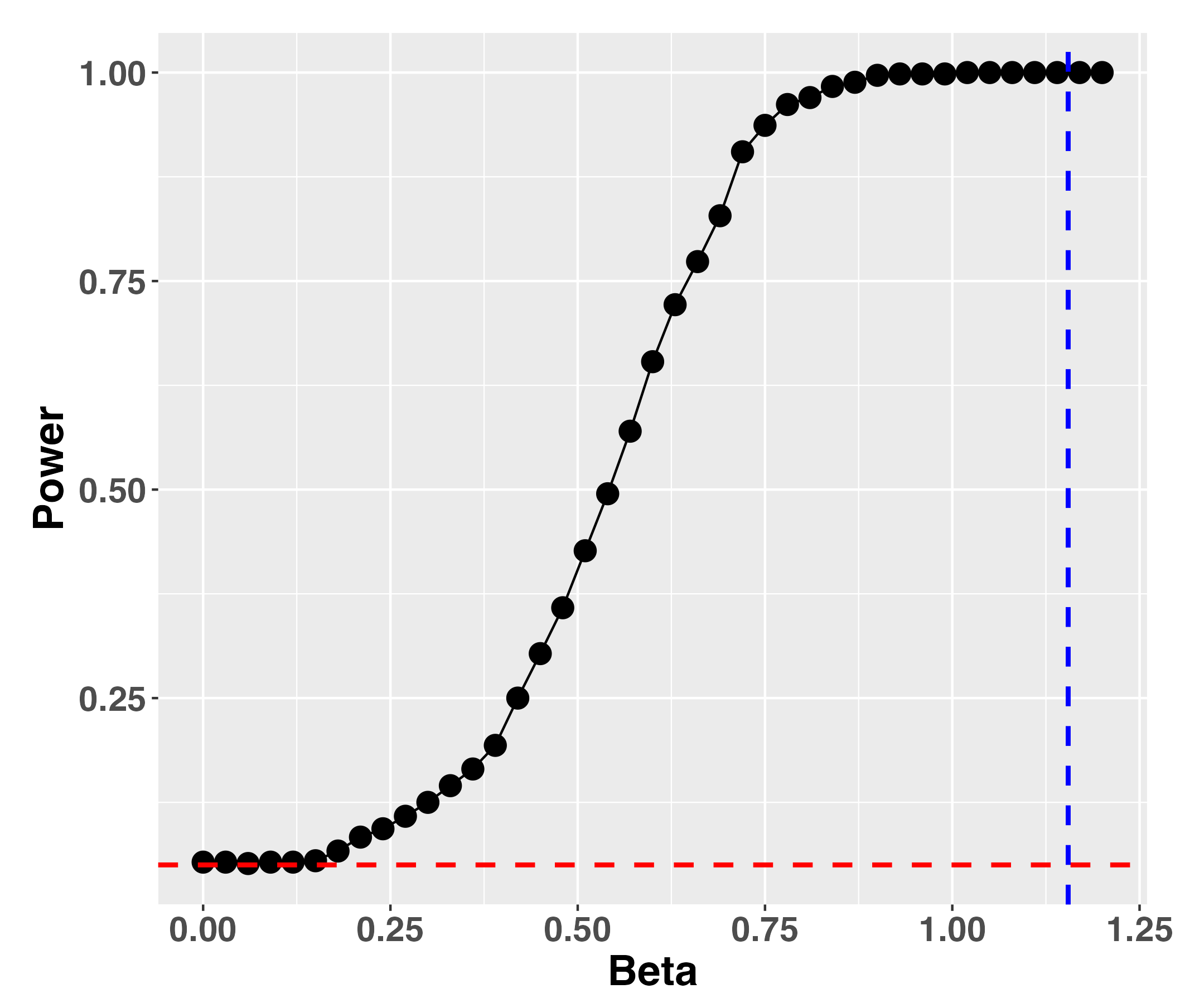}}
        \subfigure[\(\bbX\sim B(1,0.25)\), delocalized \(\bbx^{(i)}\).]{\includegraphics[width=0.49\linewidth,height=3cm]{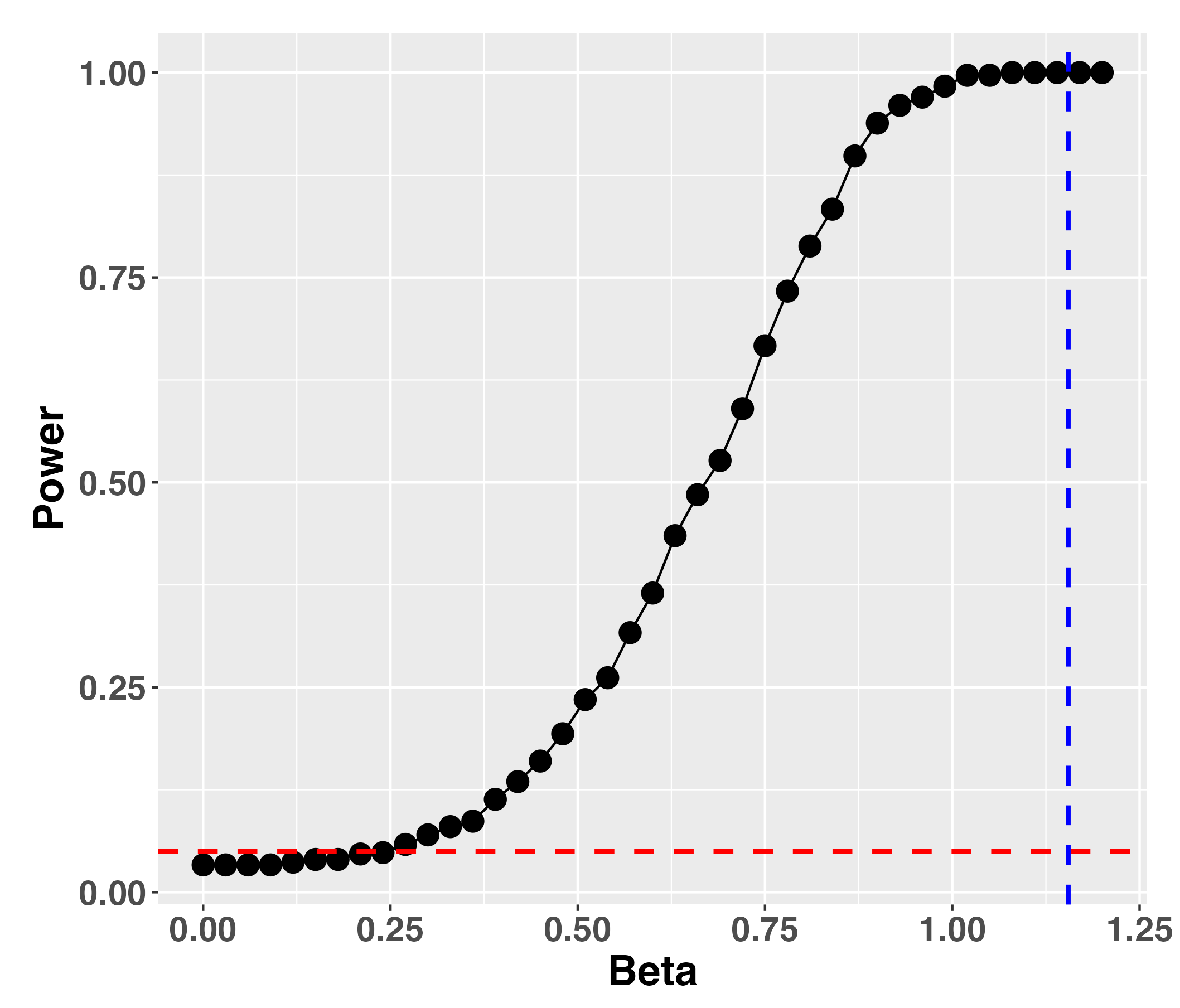}}
        \subfigure[\(\bbX\sim B(1,0.25)\), localized \(\bbx^{(i)}\).]{\includegraphics[width=0.49\linewidth,height=3cm]{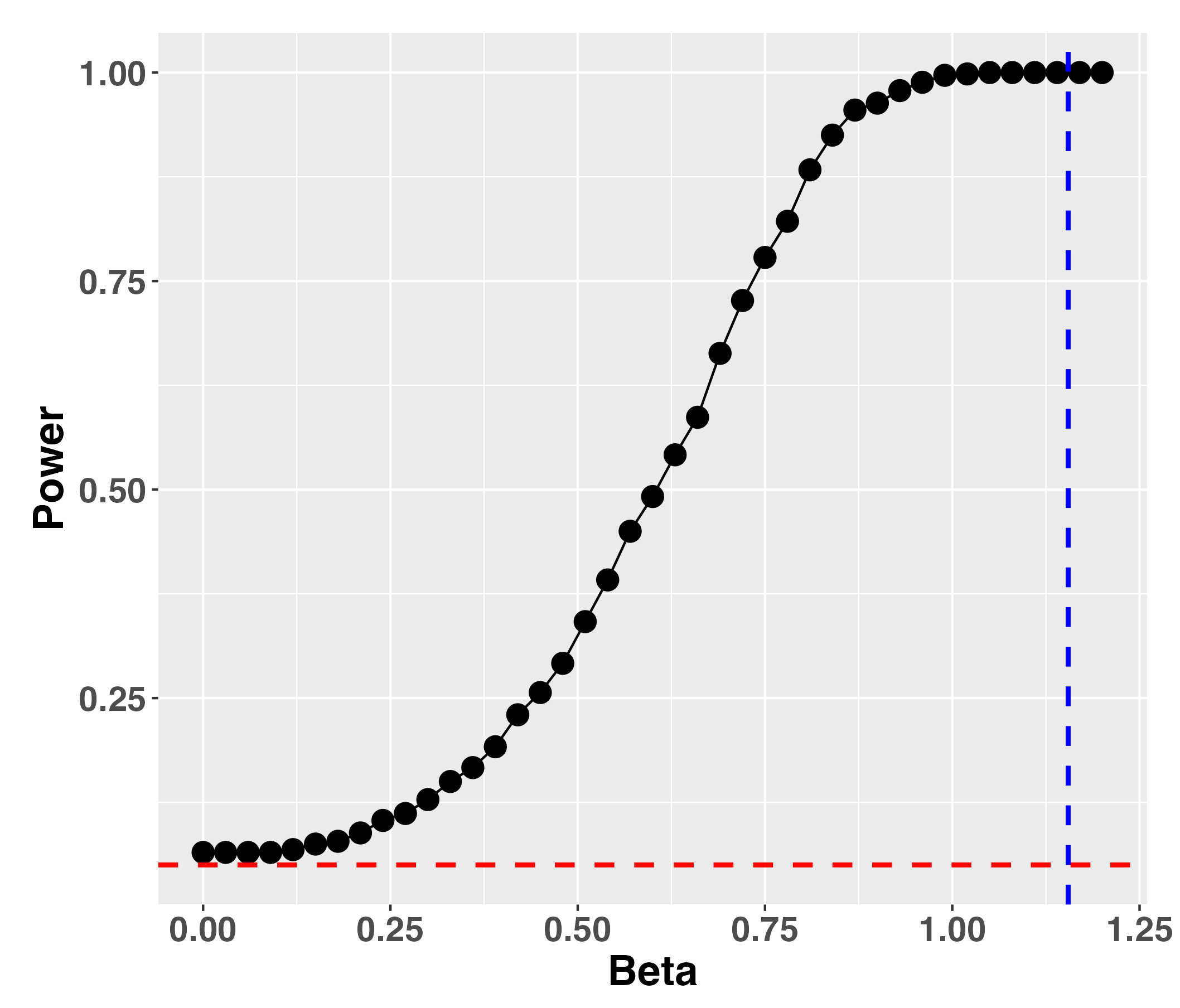}}
        \subfigure[\(\bbX\sim B(3,0.25)\), delocalized \(\bbx^{(i)}\).]{\includegraphics[width=0.49\linewidth,height=3cm]{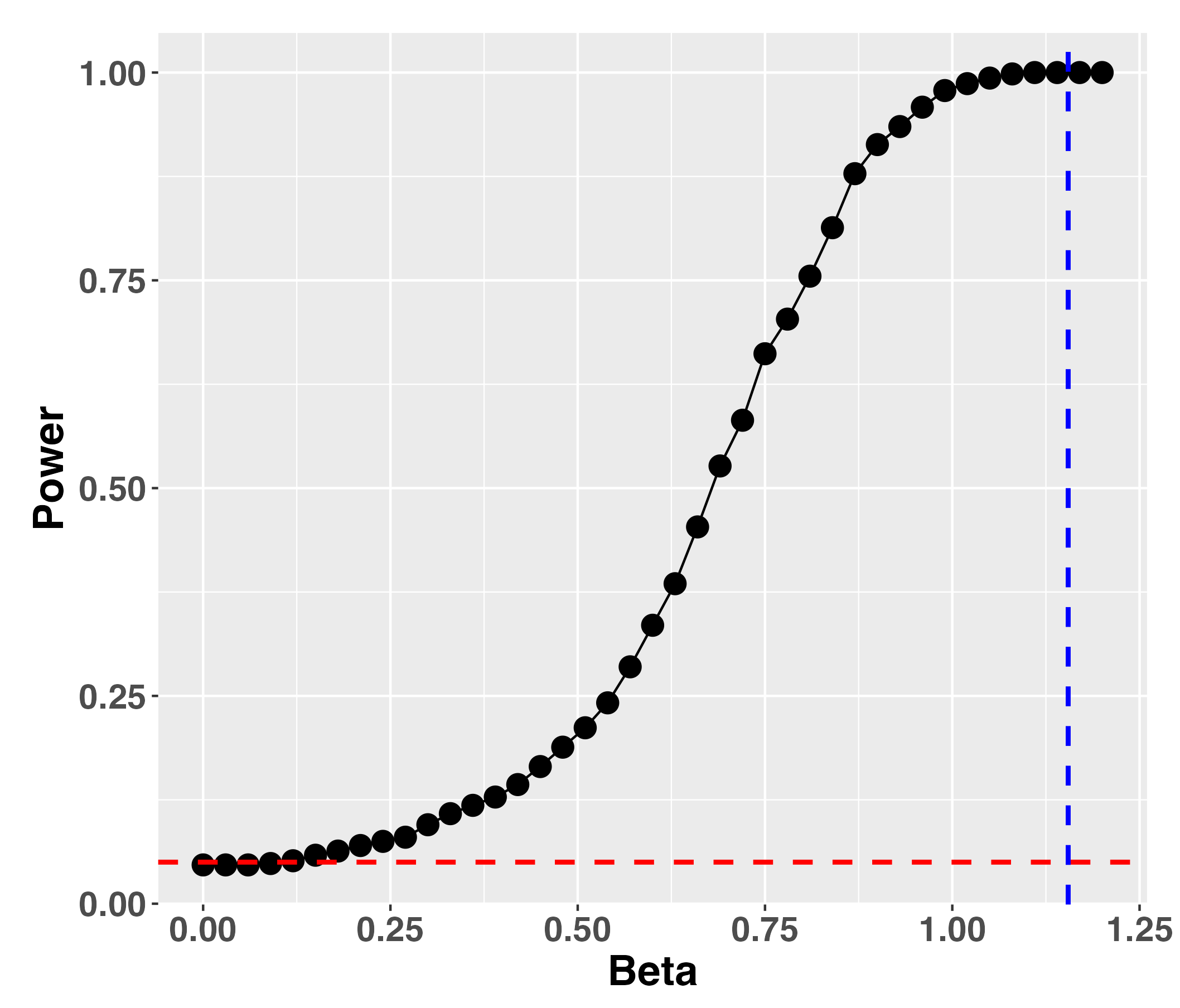}}
        \subfigure[\(\bbX\sim B(3,0.25)\), localized \(\bbx^{(i)}\).]{\includegraphics[width=0.49\linewidth,height=3cm]{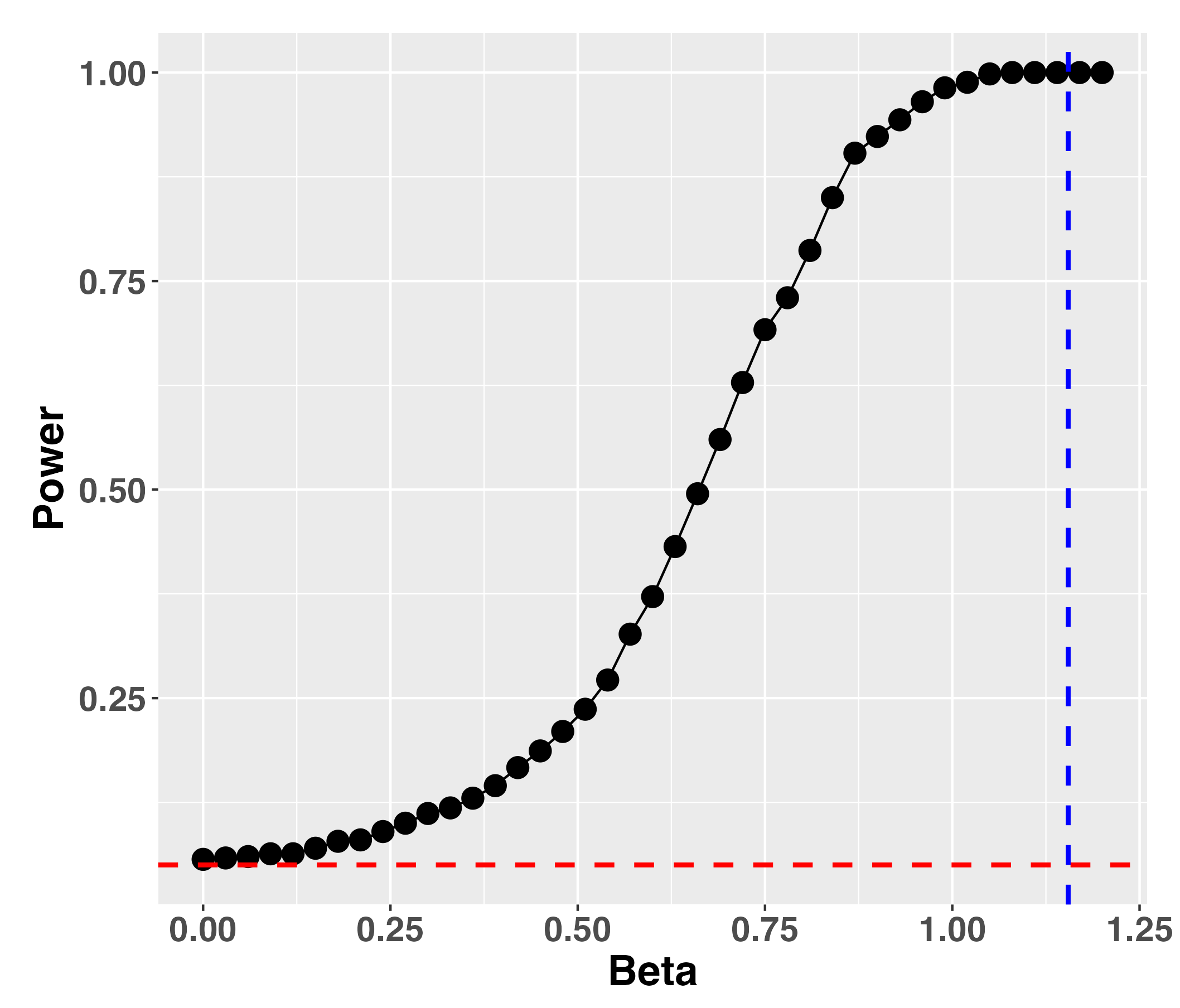}}
        \subfigure[\(\bbX\sim B(5,0.25)\), delocalized \(\bbx^{(i)}\).]{\includegraphics[width=0.49\linewidth,height=3cm]{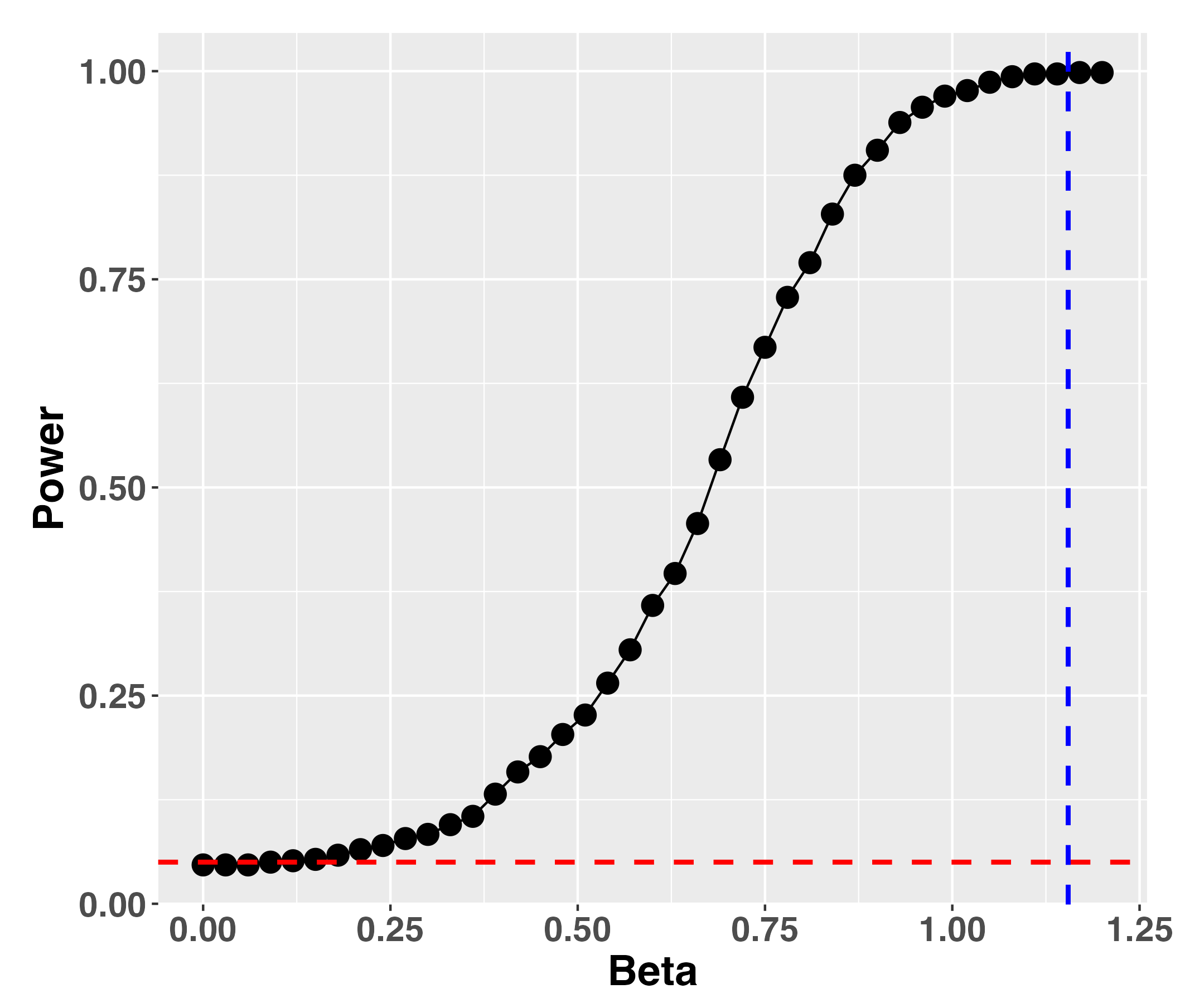}}
        \subfigure[\(\bbX\sim B(5,0.25)\), localized \(\bbx^{(i)}\).]{\includegraphics[width=0.49\linewidth,height=3cm]{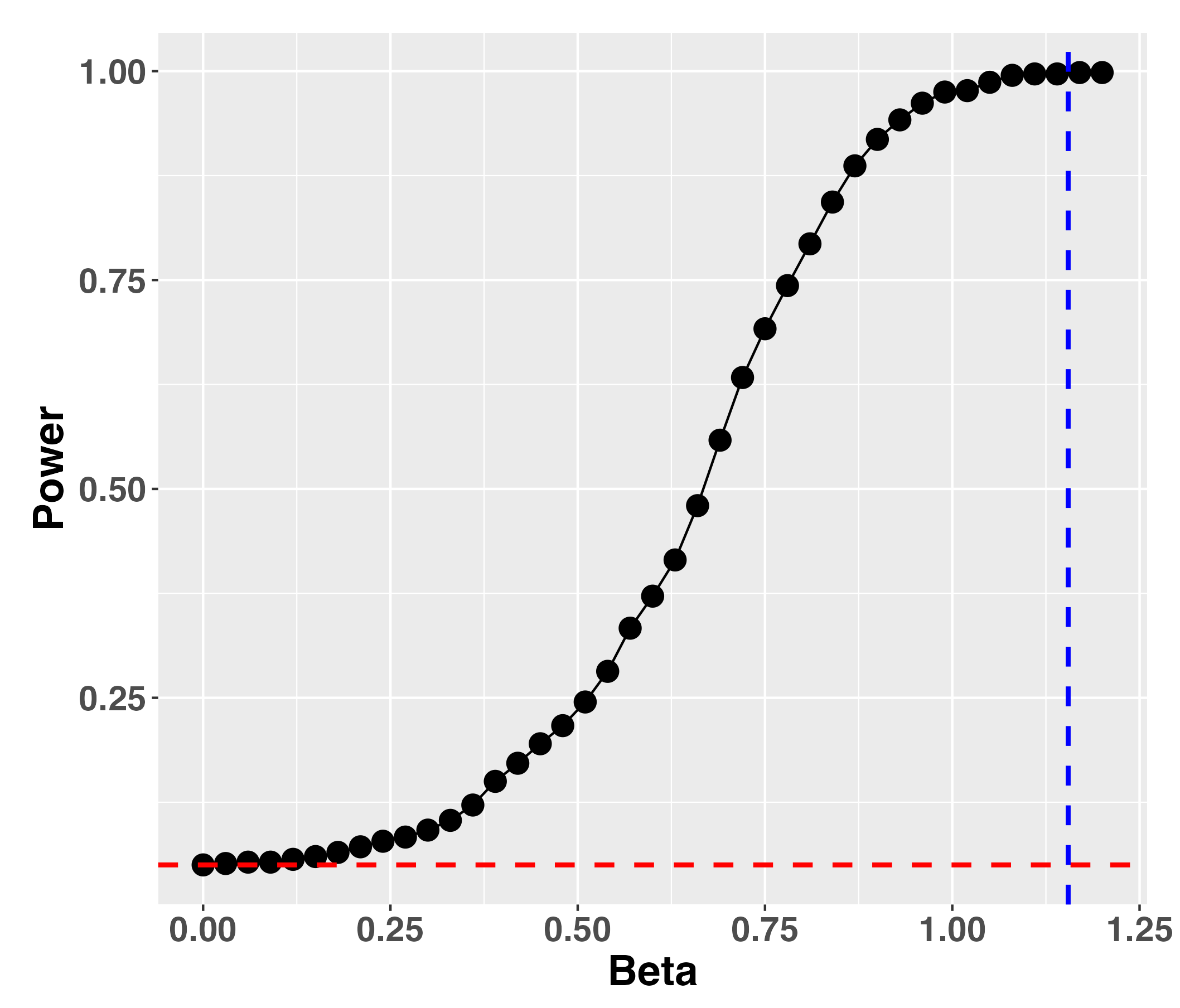}}
        \caption{Power plots of \(\tilde{\mcT}_N^{(3)}(\bbx^{(1)},\bbx^{(2)},\bbx^{(3)})\) under different \(\beta\)'s and types of noises \(\bbX\) and vectors \(\bbx^{(i)}\), where the dashed red line is the significance level \(\alpha=0.05\) and the dashed blue line is the threshold of phase transition.\label{Fig of 2}}
	\end{figure}

    Figure \ref{Fig of 2} provides several insights into the performance of the tensor signal alignment test. Firstly, the empirical sizes are close to the nominal level of 5\%. This suggests that the test maintains the desired significance level reasonably well. Secondly, as the SNR $\beta$ increases, the power of the test rapidly approaches 1 in all four scenarios. This indicates that the test is highly effective in detecting the presence of a signal when the SNR is sufficiently large. Most notably, even for $\beta$ values below the critical transition value $\beta_s=2/\sqrt{3}$, such as $\beta=1$, the test achieves a power close to one.  The test's ability to detect the presence of a signal in such challenging conditions highlights its sensitivity and effectiveness.

    \subsection{Experiment 3: tensor signal matching test}\label{Main of sec of experiment 3}
    In this subsection, we focus on the tensor signal matching test (\ref{Main of Eq of hypothesis test 2}). We generate two independent samples, $\bbT^{(0)}$ and  $\bbT^{(1)}$, using the following model:
    \begin{align*}
        \left\{\begin{array}{l}
             \bbT^{(0)}=\beta_0\bbx^{(1)}\otimes\bbx^{(2)}\otimes\bbx^{(3)}+\frac{1}{\sqrt{N}}\bbX^{(0)},\\
             \bbT^{(1)}=\beta_1\bbx^{(1)}\otimes\bbx^{(2)}\otimes\bbx^{(3)}+\frac{1}{\sqrt{N}}\bbX^{(1)},
        \end{array}\right.
    \end{align*}
    where the noise tensors \(\bbX^{(0)}\) and \(\bbX^{(1)}\) are independent,
    and the two rank-1 tensor signals are parallel but have different strengths. 
    Following the procedures described in \S\ref{sec of generalized procedure}, we first apply the tensor unfolding method to estimate \(\hat{\bbx}^{(1)}\otimes\hat{\bbx}^{(2)}\otimes\hat{\bbx}^{(3)}\) using the first tensor data \(\bbT^{(0)}\). Then, we test whether \(\bbT^{(1)}\) contains a signal along \(\hat{\bbx}^{(1)}\otimes\hat{\bbx}^{(2)}\otimes\hat{\bbx}^{(3)}\) or not. 
    
    The main objective of this experiment is to investigate how the values of \(\beta_0\) and \(\beta_1\) affect the power of (\ref{Main of Eq of hypothesis test 2}) and to compare it with the power of \(\tilde{\mcT}_N^{(3)}(\bbx^{(1)},\bbx^{(2)},\bbx^{(3)})\) when using known directional vectors. 
    
    We set \(\beta_0=2,2.5,3\) and estimate \(\hat{\bbx}^{(1)},\hat{\bbx}^{(2)},\hat{\bbx}^{(3)}\) for each \(\beta_0\). The rest of the setting is essentially the same as in \S\ref{Main of sec of experiment 2}, with the addition of \(\beta_1\in [0,1.2]\). We compute the empirical power of \(\tilde{\mcT}_N^{(3)}(\hat{\bbx}^{(1)},\hat{\bbx}^{(2)},\hat{\bbx}^{(3)})\) and present the power plots in Figure \ref{Fig of 3}.
    \begin{figure}
		\subfigure[\(\bbX\sim\mcN(0,1)\), delocalized \(\bbx^{(i)}\).]{\includegraphics[width=0.49\linewidth,height=3cm]{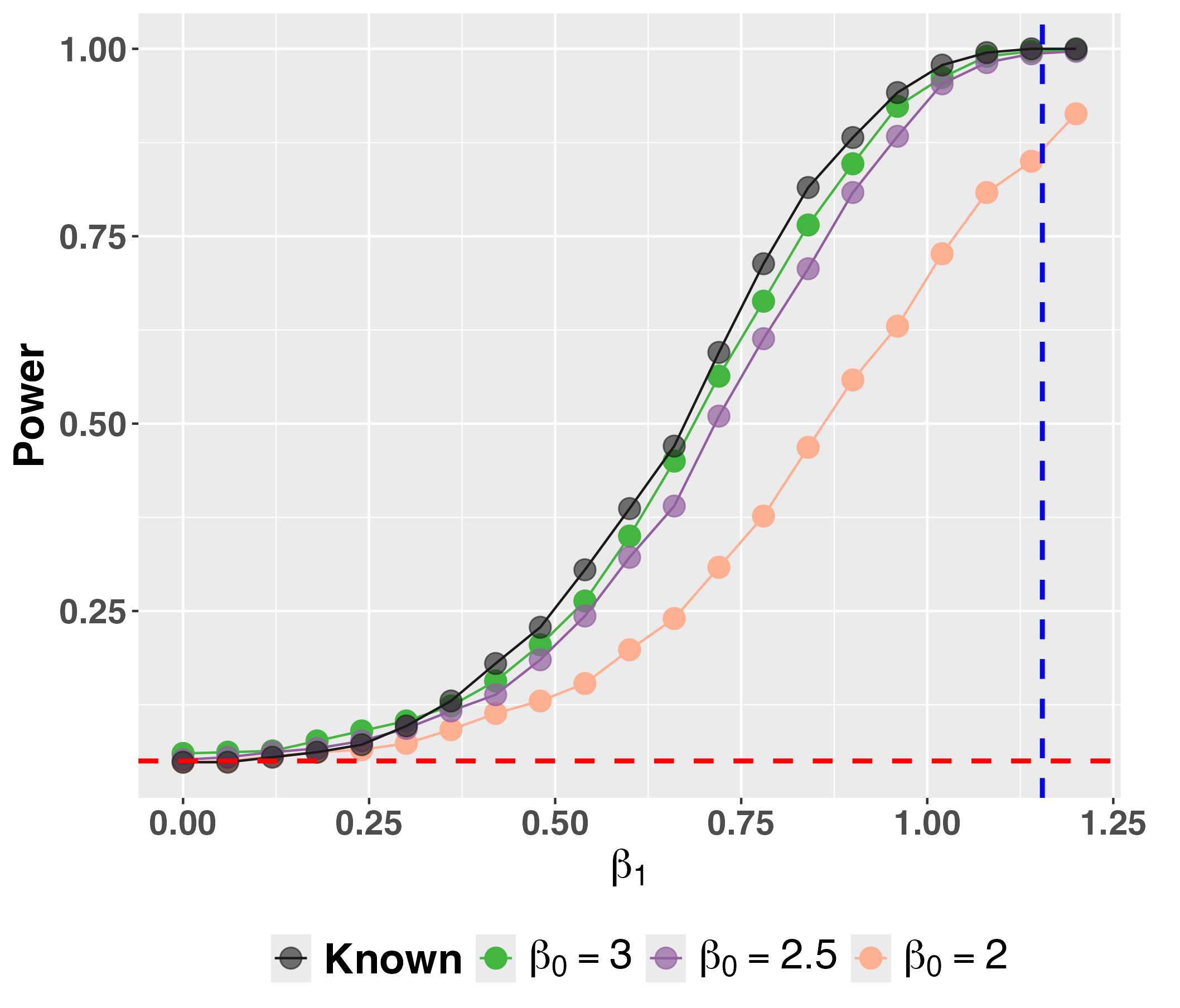}}
		\subfigure[\(\bbX\sim{\rm Unif}(\pm\sqrt{3})\), delocalized \(\bbx^{(i)}\).]{\includegraphics[width=0.49\linewidth,height=3cm]{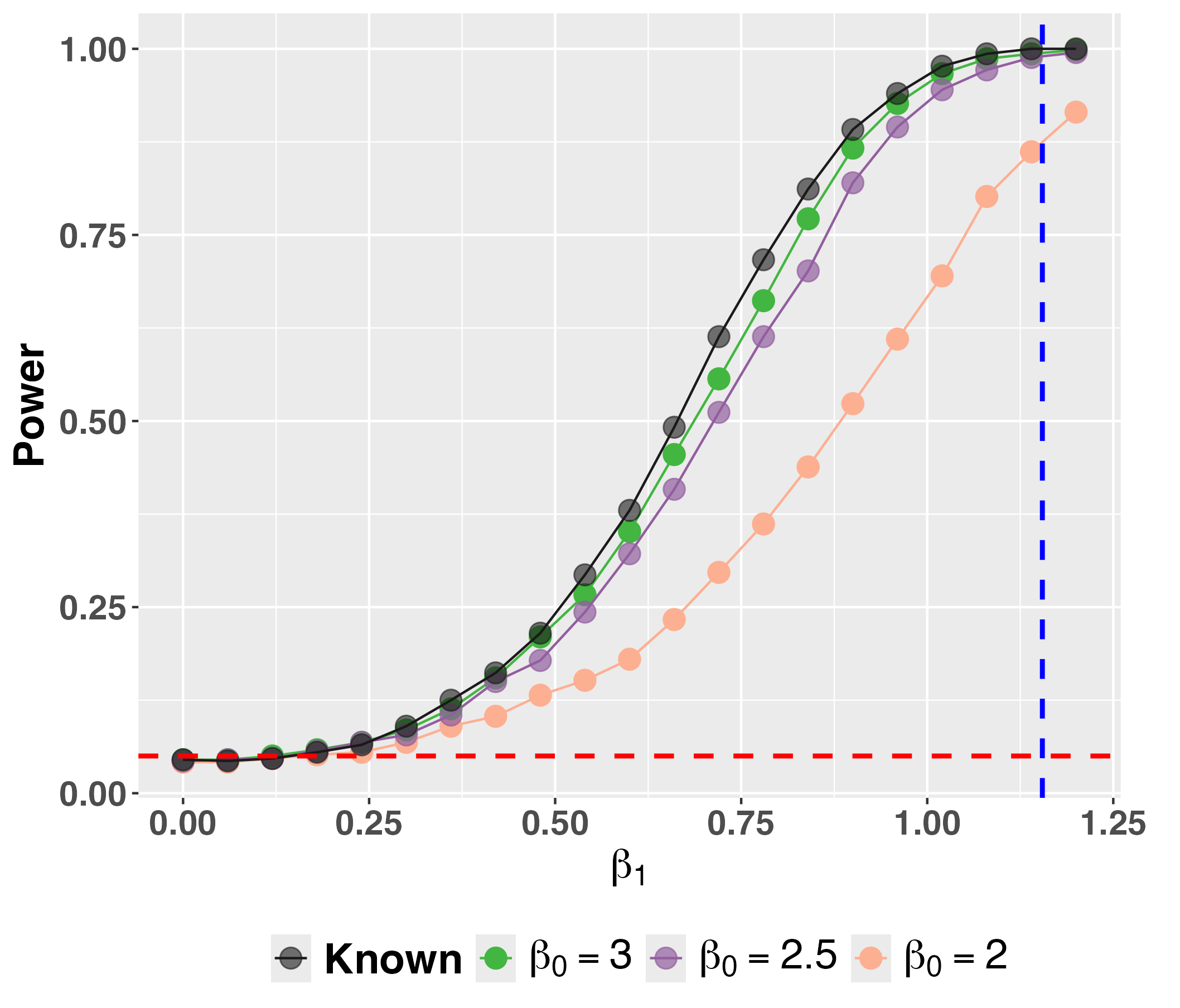}}
        \subfigure[\(\bbX\sim\mcN(0,1)\), localized \(\bbx^{(i)}\).]{\includegraphics[width=0.49\linewidth,height=3cm]{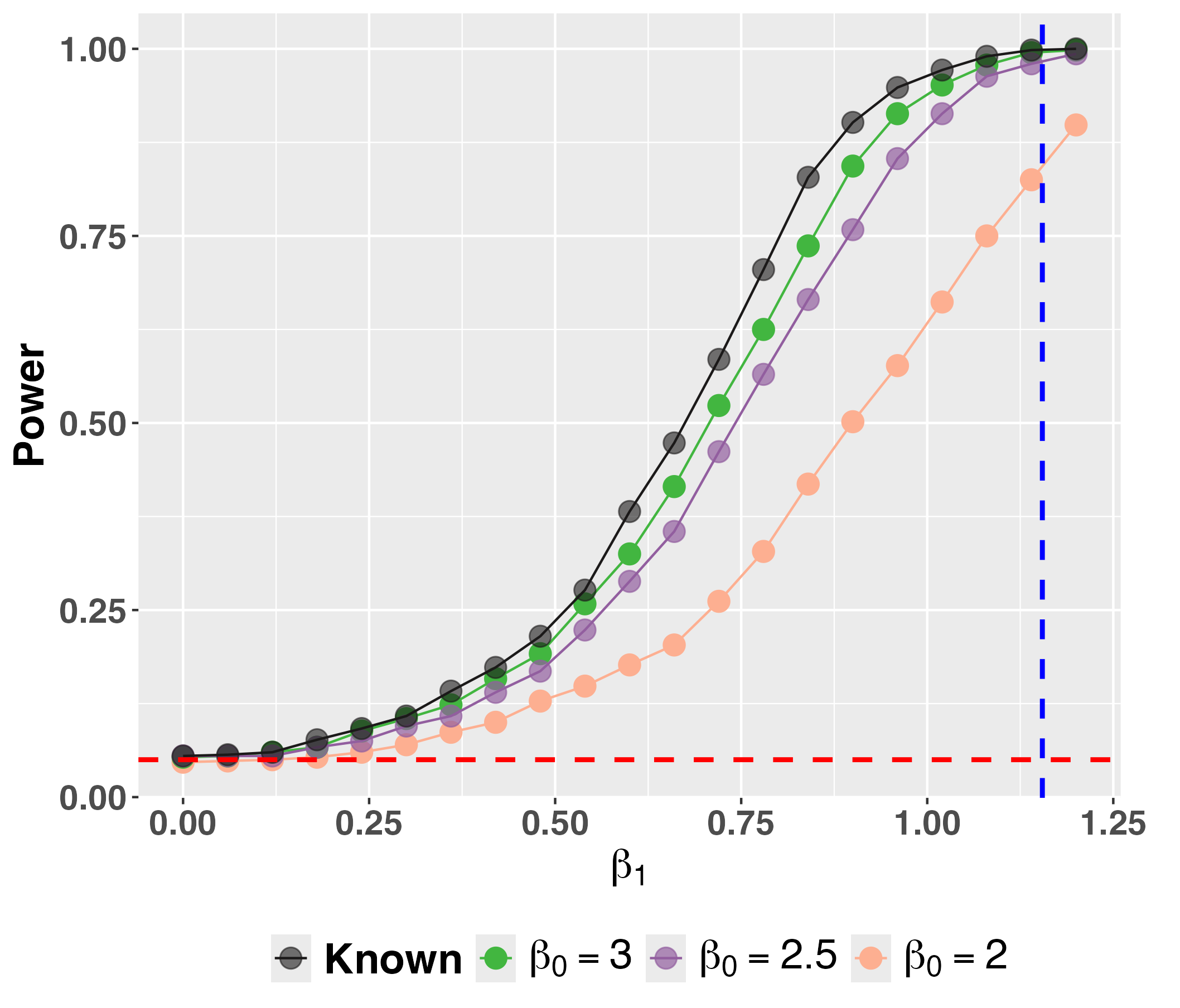}}
        \subfigure[\(\bbX\sim{\rm Unif}(\pm\sqrt{3})\), localized \(\bbx^{(i)}\).]{\includegraphics[width=0.49\linewidth,height=3cm]{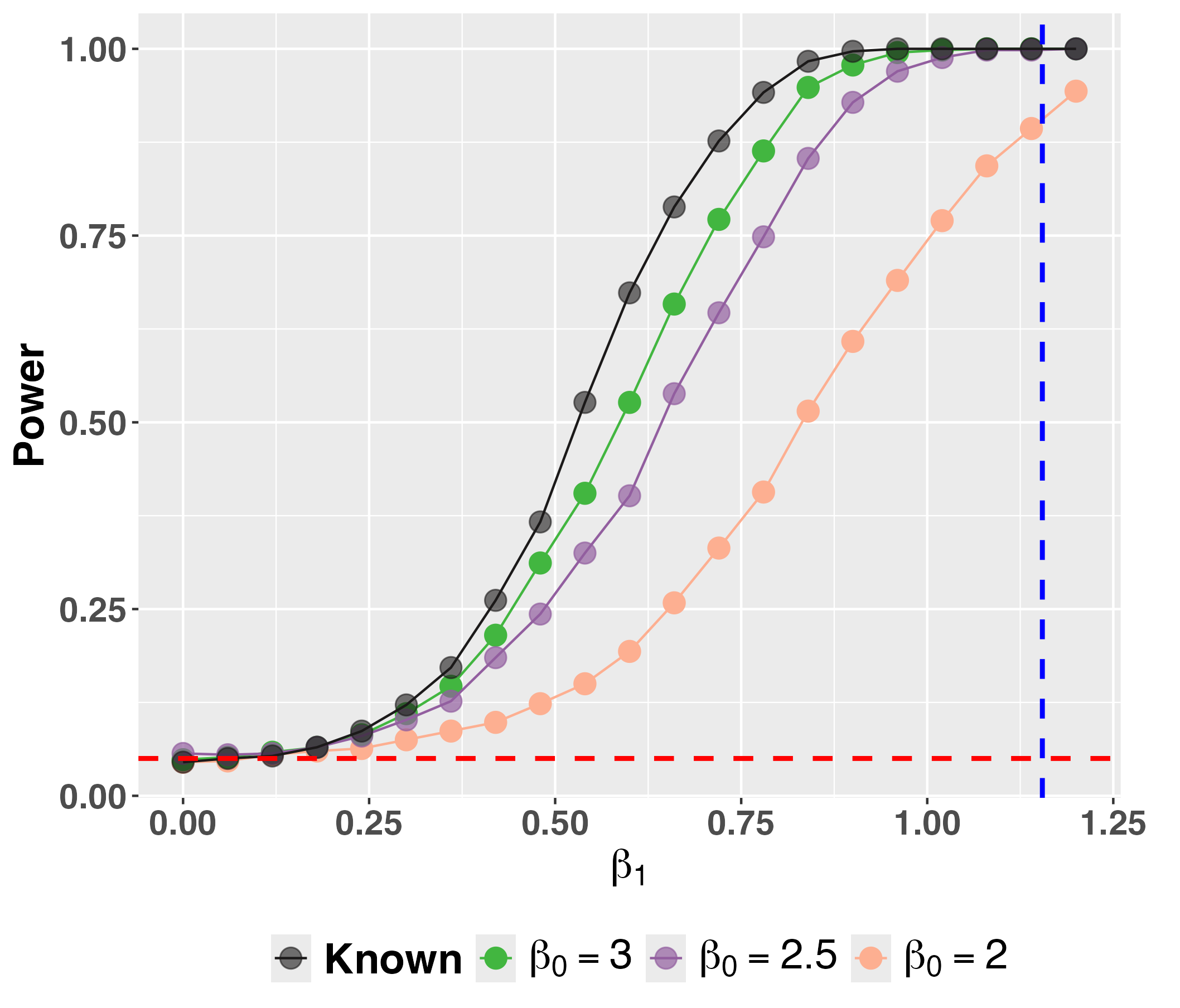}}
        \subfigure[\(\bbX\sim B(1,0.25)\), delocalized \(\bbx^{(i)}\).]{\includegraphics[width=0.49\linewidth,height=3cm]{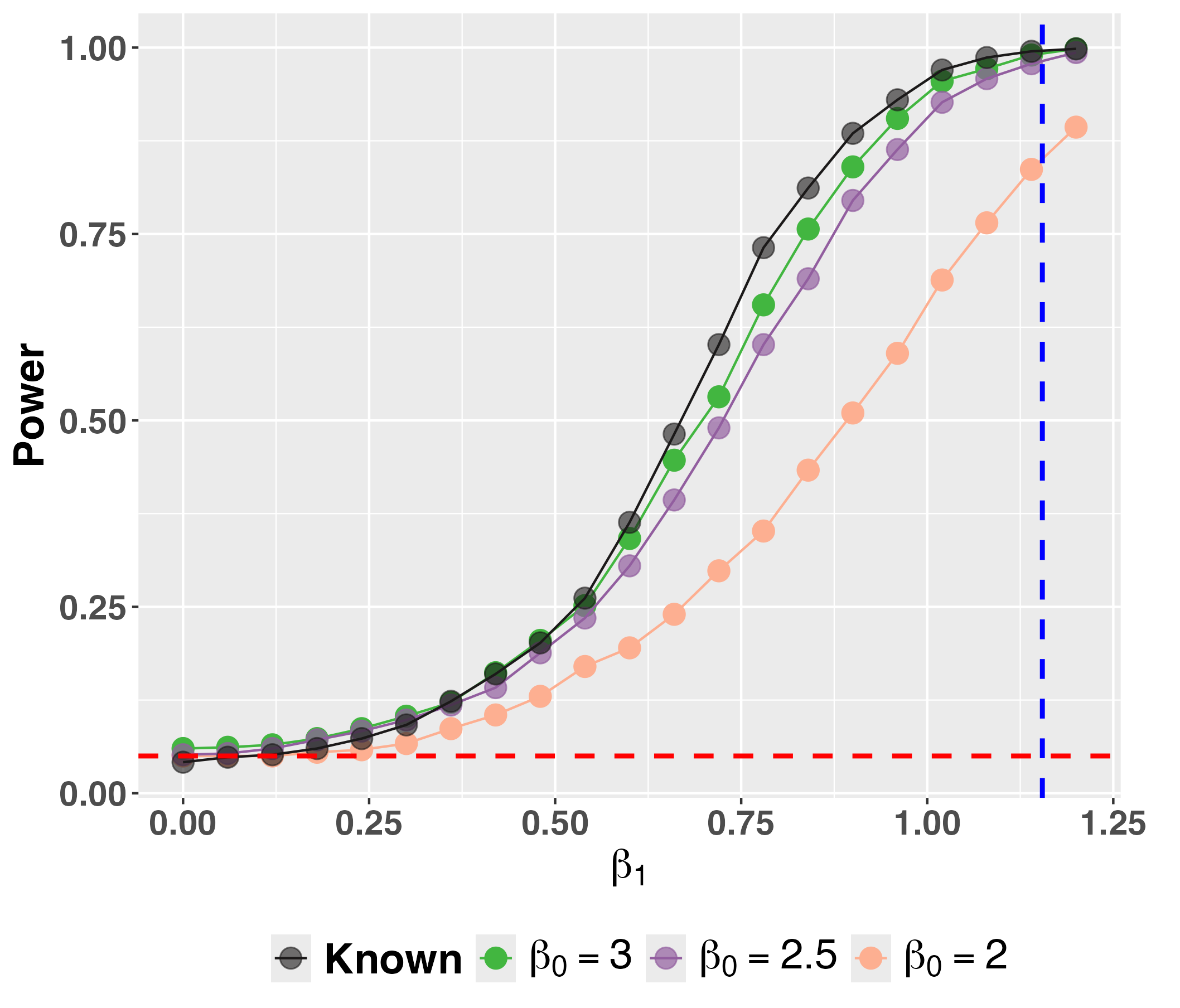}}
        \subfigure[\(\bbX\sim B(1,0.25)\), localized \(\bbx^{(i)}\).]{\includegraphics[width=0.49\linewidth,height=3cm]{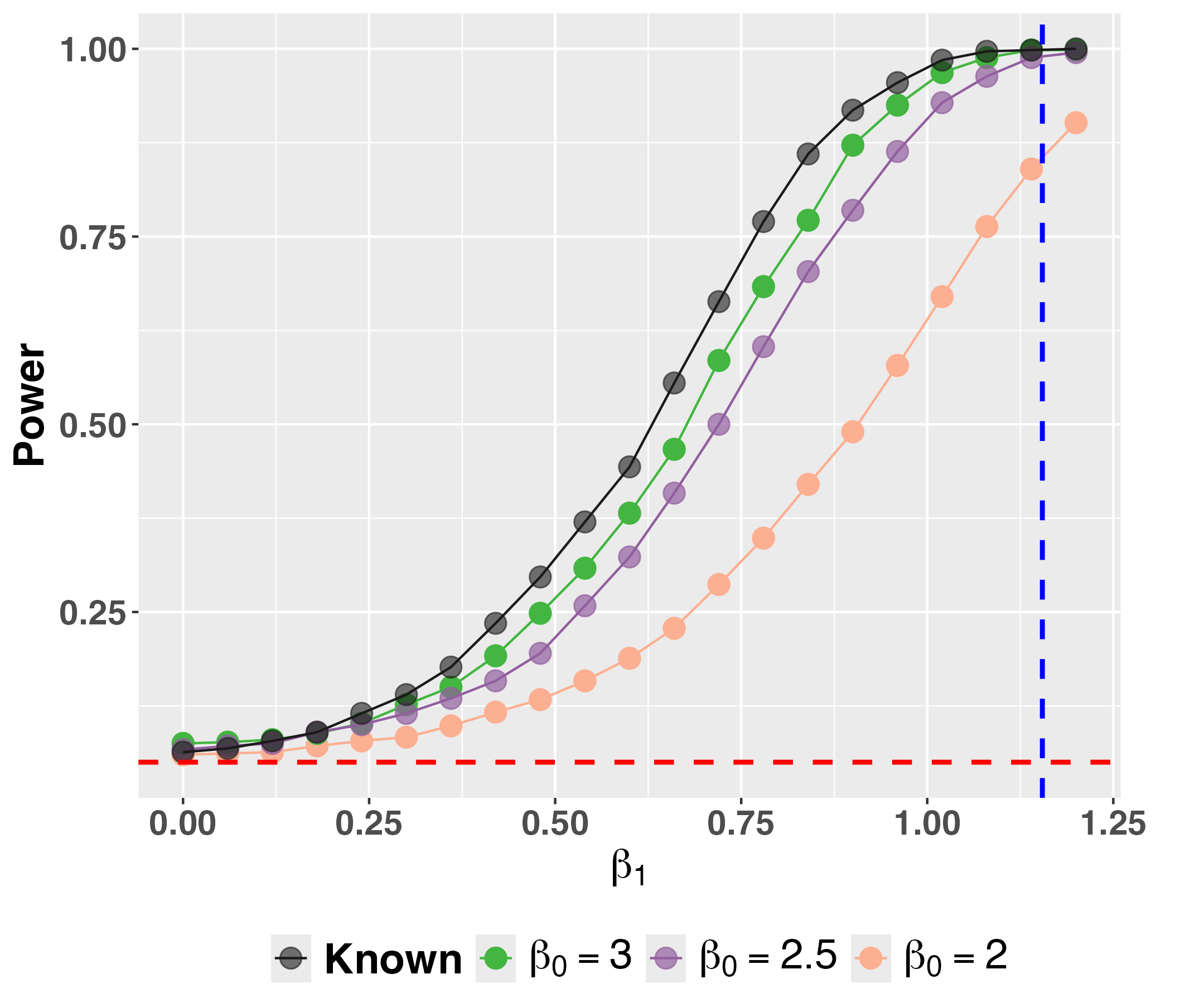}}
        \subfigure[\(\bbX\sim B(3,0.25)\), delocalized \(\bbx^{(i)}\).]{\includegraphics[width=0.49\linewidth,height=3cm]{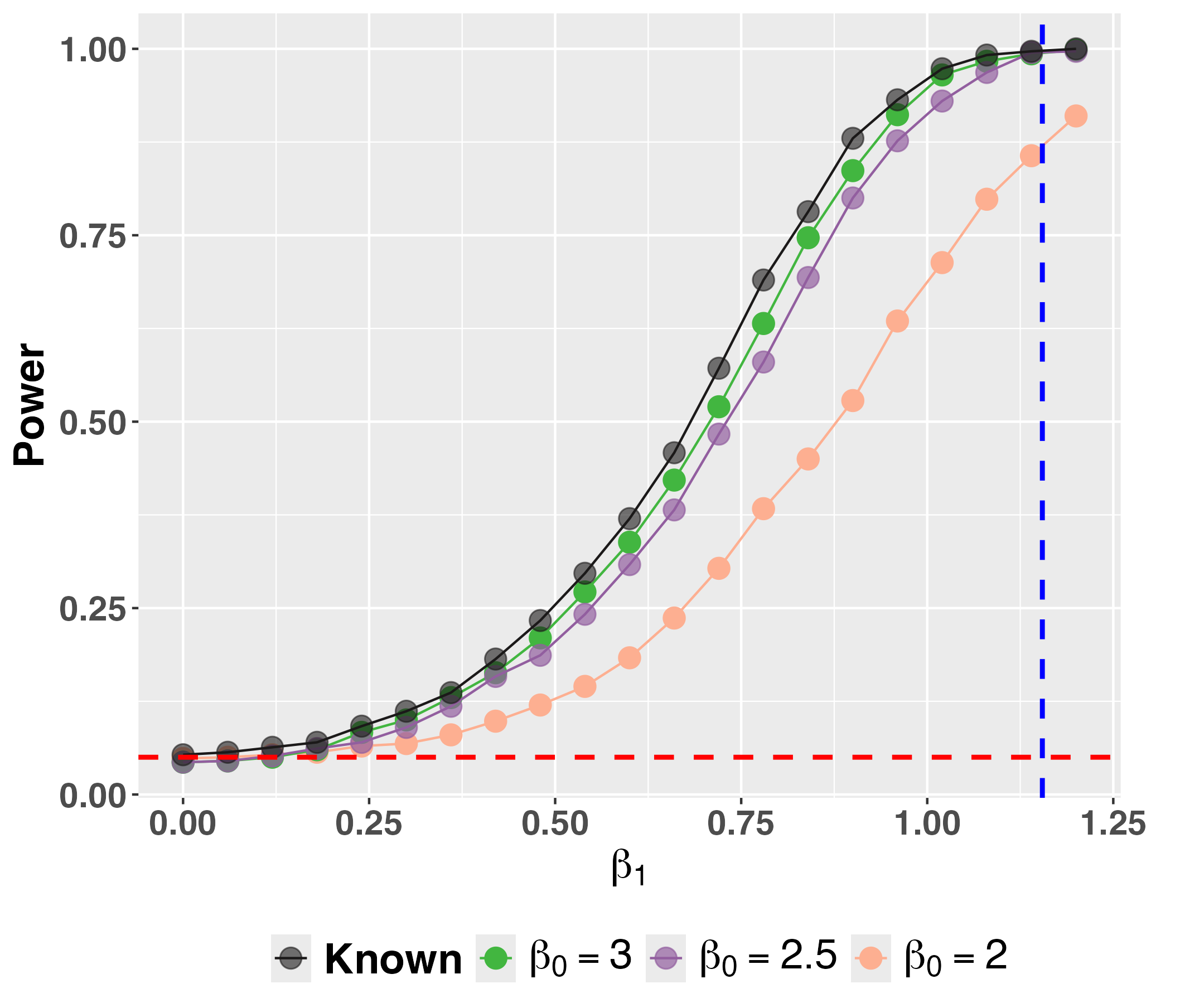}}
        \subfigure[\(\bbX\sim B(3,0.25)\), localized \(\bbx^{(i)}\).]{\includegraphics[width=0.49\linewidth,height=3cm]{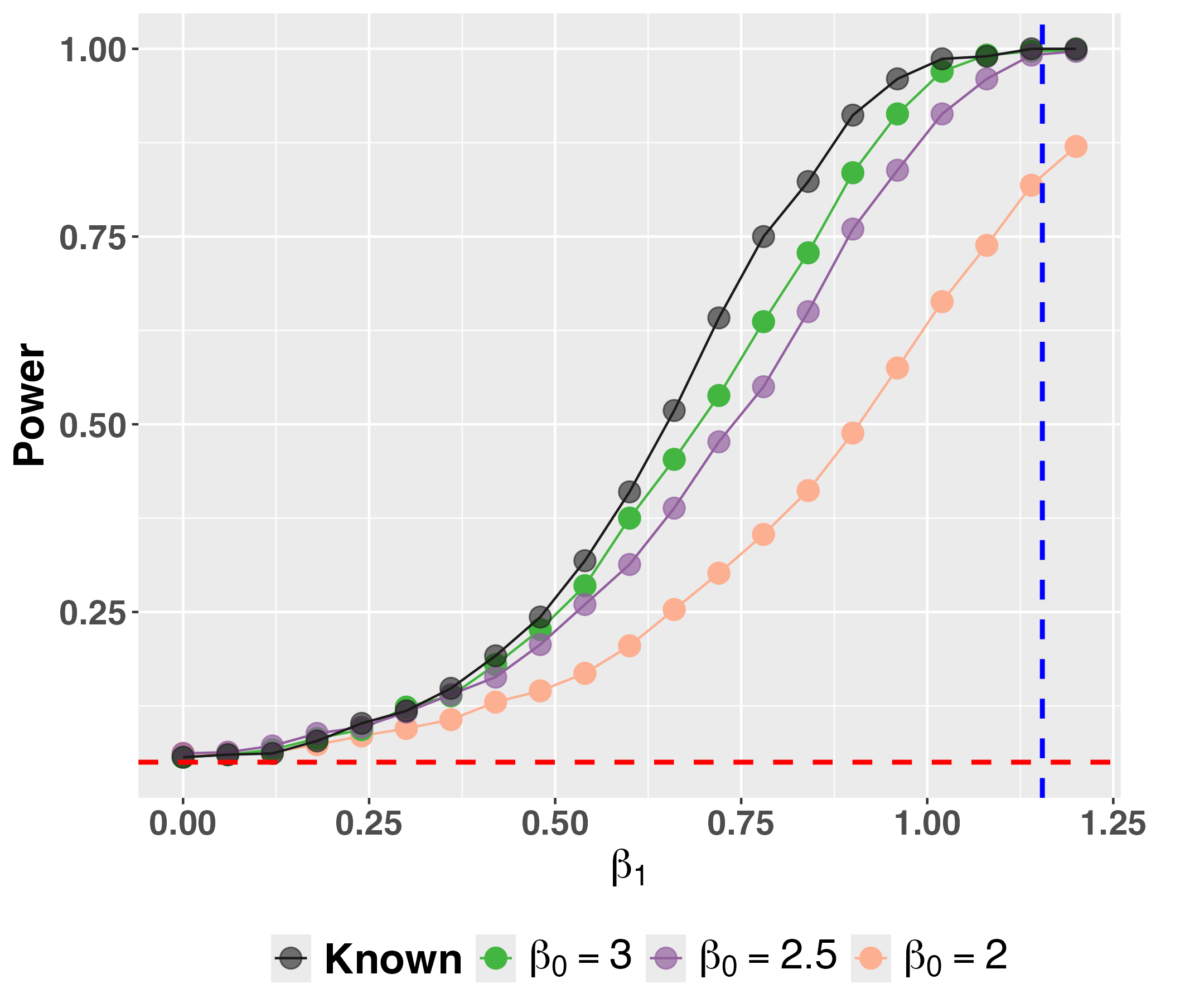}}
        \subfigure[\(\bbX\sim B(5,0.25)\), delocalized \(\bbx^{(i)}\).]{\includegraphics[width=0.49\linewidth,height=3cm]{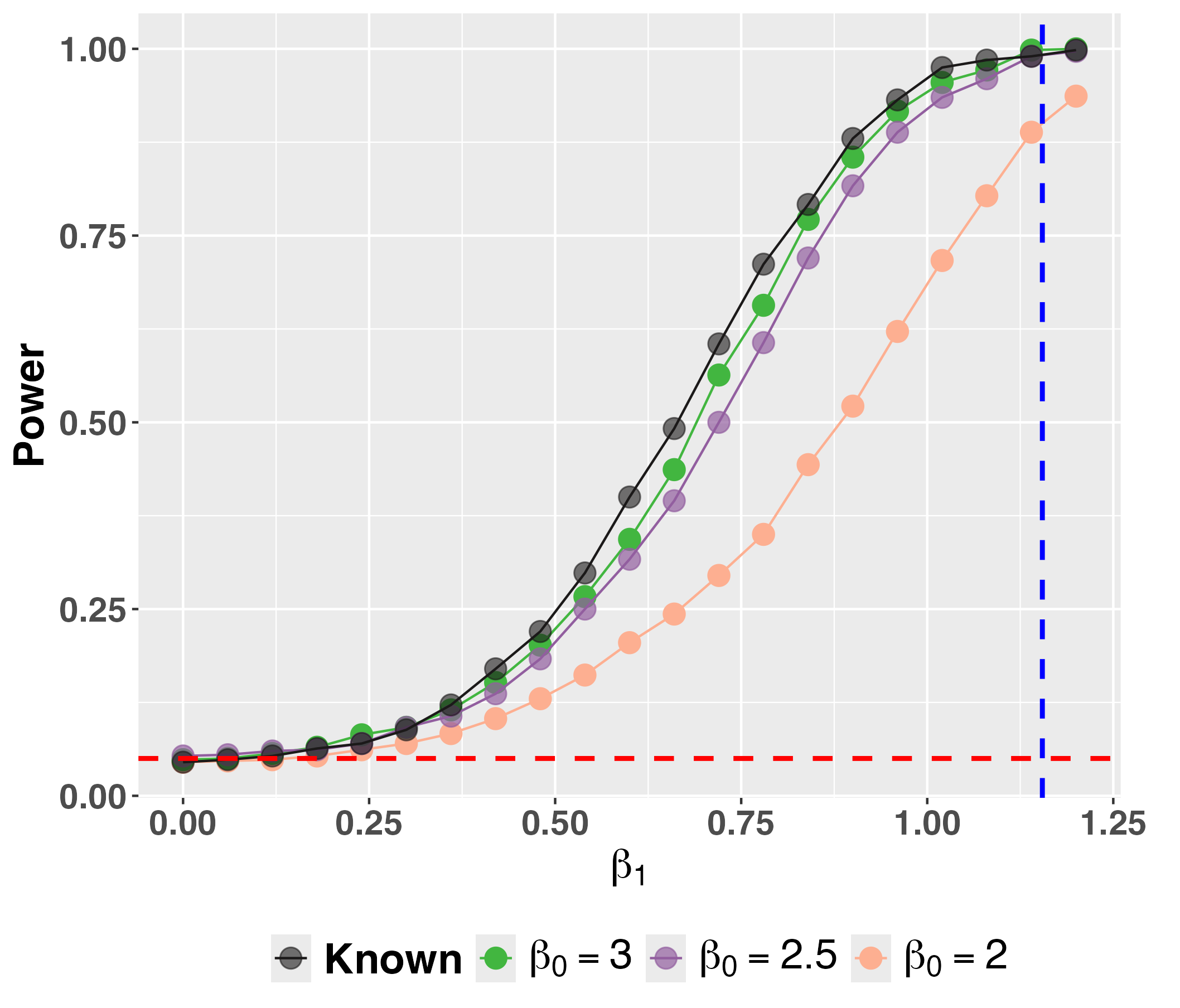}}
        \subfigure[\(\bbX\sim B(5,0.25)\), localized \(\bbx^{(i)}\).]{\includegraphics[width=0.49\linewidth,height=3cm]{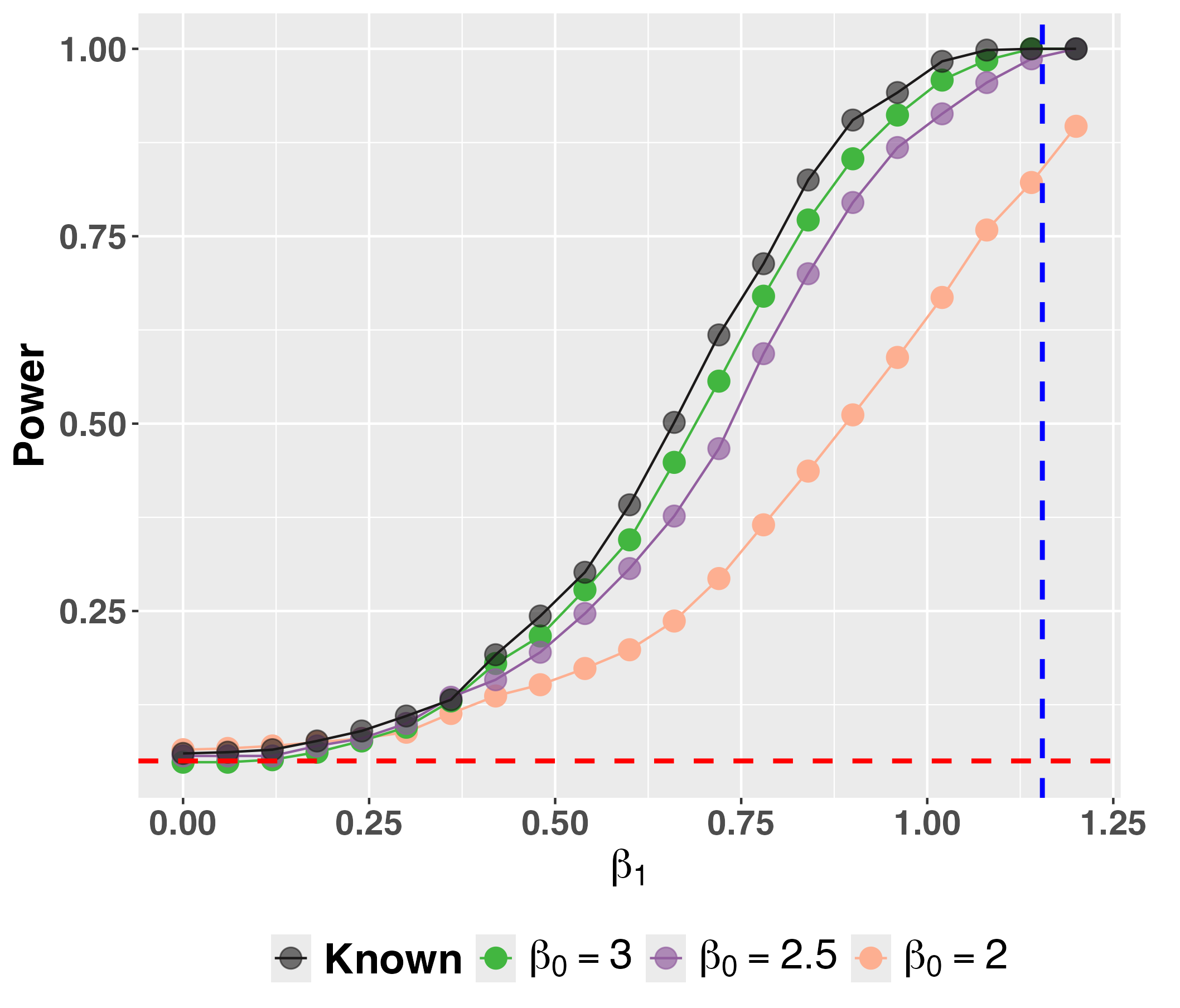}}
		\caption{Power plots of \(\tilde{\mcT}_N^{(3)}(\hat{\bbx}^{(1)},\hat{\bbx}^{(2)},\hat{\bbx}^{(3)})\) under different \(\beta_0,\beta_1\) and types of noises \(\bbX\) and vectors \(\bbx^{(i)}\). ``Known'' denotes the empirical power of \(\tilde{\mcT}_N^{(3)}(\bbx^{(1)},\bbx^{(2)},\bbx^{(3)})\), while ``Beta0=$a$'' represents the empirical power of \(\tilde{\mcT}_N^{(3)}(\hat{\bbx}^{(1)},\hat{\bbx}^{(2)},\hat{\bbx}^{(3)})\) when \(\beta_0=a\), $a=2,2.5,3$. The dashed red line and blue line indicate the significance level \(\alpha=0.05\) and the threshold of phase transition  $\beta_s=2/\sqrt 3=1.1547$, respectively.\label{Fig of 3}}
	\end{figure}


{ Figure \ref{Fig of 3} compares the empirical power of $\tilde{\mcT}_N^{(3)}$ using estimated directional vectors $(\hat{\bbx}^{(1)},\hat{\bbx}^{(2)},\hat{\bbx}^{(3)})$ versus known directional vectors $(\bbx^{(1)},\bbx^{(2)},\bbx^{(3)})$. Across all settings, the test based on estimated vectors exhibits lower power, as expected. However, when the signal strength in $\bbT^{(0)}$ is moderate (e.g., $\beta_0 = 2.5$), the two approaches still achieve comparable performance. This suggests that our two-step testing procedure remains effective under moderate signal conditions. Moreover, when the reference signal is moderately large ($\beta_0 \geq 2.5$), similar to the tensor alignment test, our matching test with reference maintains power close to 1 even for $\beta_1$ slightly below the phase transition threshold (e.g., $\beta_1 = 1$), demonstrating robust detection of signal matching even under weak target signal conditions.}


\section{Basic settings}\label{sec of Basic settings}
\setcounter{equation}{0}
\def\theequation{\thesection.\arabic{equation}}
\setcounter{subsection}{0}
For the sake of completeness and the readability of this supplement, we start by introducing some notations, definitions and assumptions, even though they may have been encountered earlier in the manuscript.
\begin{enumerate}[(i)]
	\item Given \(z\in\mbC\), \(\Re(z)\) and \(\Im(z)\) are the real and imaginary part of \(z\) respectively.
	\item We use an element in \(\mbR^{n_1\times\cdots\times n_d}\) to represent the \(d\)-fold real tensor of size \(n_1\times\cdots\times n_d\).
	\item Given \(A=[a_{ij}]_{n\times n}\), \(\tr(A)=\sum_{i=1}^n a_{ii}\) and \(A'\) denotes the transpose of \(A\) and \(\diag(A)\) is the diagonal matrix made with the main diagonal of \(A\). Moreover, \(\Vert A\Vert\) denotes the spectral norm of \(A\) and \(\Vert A\Vert_k=(\sum_{i,j}|a_{ij}|^k)^{1/k}\) for any \(k\in\mbN^+\).
    \item Given a matrix \(A=[a_{ij}]_{n\times n}\), \(A_{i\cdot}\) and \(A_{\cdot j}\) denote the \(i\)-th row and \(j\)-th column of \(A\), respectively.
	\item The \(n\)-dimensional unit sphere is defined as \(\mbS^{n-1}:=\{\bbx\in\mbR^n:\Vert\bbx\Vert_2=1\}\).
	\item \(C_{\eta}\) represents a positive constant which depends on some parameters \(\eta\).
	\item Given an integrable random variable \(X\), we define its centered version as \(X^c:=X-\mathbb{E}[X].\)
	\item Given \(\eta>0\), define \(\mathbb{C}_{\eta}^+:=\{z\in\mathbb{C}:\Im(z)>\eta\}\) and \(\mathbb{C}^+:=\{z\in\mathbb{C}:\Im(z)>0\}\).
	\item For a real sequence \(\{a_n\}\), \(a_n=\mro(n^{-r})\) for \(r\geq0\) means that \(\lim_{n\to\infty}a_n n^r=0\); \(a_n=\mrO(n^{-r})\) means \(a_nn^r\) is bounded.
	\item The asymptotic almost sure convergence, convergence in probability and in distribution are denoted by \(\overset{a.s.}{\longrightarrow},\overset{\mbP}{\longrightarrow}\) and \(\overset{d}{\longrightarrow}\), respectively.
	\item Given two matrices \(A,B\) of size \(m\times n\), when \(B_{ij}\neq0\) for all \(i,j\),
	\begin{align}
		\frac{\bbA}{\bbB}=[A_{ij}B_{ij}^{-1}]_{m\times n}.\label{Eq of division of matrices}
	\end{align}
	\item Let \(X=\{X_n\}\) and \(Y=\{Y_n\}\) be two sequences of nonnegative random variables. We say $Y$ stochastically dominates $X$ if for all (small) \(\epsilon>0\) and (large) \(D>0\),
	\begin{align}
		\mathbb{P}(X_n>n^{\epsilon}Y_n)\leq n^{-D}\label{Eq of stochastic domination}
	\end{align}
	for all \(n\geq n_0(\epsilon,D)\), which is denoted by \(X\prec Y\) or \(X\prec\mrO(Y)\).
\end{enumerate}
Let \(d\geq3\) be a positive integer, and let \(n_1,\cdots,n_d\in\mbN^+\) be \(d\) positive integers, the \(d\)-fold rank-\(R\) spiked tensor model is defined as:
\begin{align}
	\bbT=\sum_{r=1}^R\beta_r\bbx^{(r,1)}\otimes\cdots\otimes\bbx^{(r,d)}+\frac{1}{\sqrt{N}}\bbX,\label{Eq of spiky tensor model}
\end{align}
where \(\beta_r>0\), \(\bbx^{(r,i)}\in\mbS^{n_i-1}\), \(N:=\sum_{j=1}^dn_j\) and \(\bbX=[X_{i_1\cdots i_d}]_{n_1\times\cdots\times n_d}\in\mbR^{n_1\times\cdots\times n_d}\) is a random tensor whose entries \(X_{i_1\cdots i_d}\) are i.i.d. copies of a centered random variable \(X\) with unit variance and subexponential tails, i.e.,
\begin{Ap}\label{Ap of general noise}
	$$\limsup_{x\geq0}e^{x^{\theta}}\mbP(|X|\geq x)<\infty,$$
	where \(\theta>0\). Moreover, \(\mathbb{E}[X]=0,\Var(X)=1,\) and its third and fourth cumulants are denoted by
	$$ \kappa_3:=\mathbb{E}[X^3]\quad{\rm and}\quad\kappa_4:=\mathbb{E}[X^4]-3.$$
\end{Ap}
\begin{Ap}\label{Ap of dimension}
	The tensor dimensions \(n_1,\cdots,n_d\) all tend to infinity such that 
	$$\lim_{n_1,\cdots,n_d\to\infty}\frac{n_j}{n_1+\cdots+n_d}=\mfc_j\in(0,1),\quad 1\leq j\leq d.$$
	This limiting framework is simply denoted as \(N\to\infty\) (where \(N:=n_1+\cdots+n_d\)) and let
	$$\mfc=(\mfc_1,\cdots,\mfc_d)'.$$
\end{Ap}
Let \(\bba^{(1)}\in\mbS^{n_1-1},\cdots,\bba^{(d)}\in\mbS^{n_d-1}\) be \(d\) deterministic unit vectors such that the vector dimensions \(n_1,\cdots,n_d\) satisfy Assumption \ref{Ap of dimension}. Next, we further define several auxiliary notations as follows:
	\begin{itemize}
		\item Given \(k\in\{1,\cdots,d\}\), define
		\begin{align}
			\mfb_k^{(1)}:=\sum_{i_k=1}^{n_k}a_{i_k}^{(k)}.\label{Eq of mfb}
		\end{align}
		\item For any \(l\) pairwise distinct integers \(k_1,k_2,\cdots,k_l\in\{1,\cdots,d\}\), i.e., $k_i\neq k_j$ for all $i\neq j$, define
		\begin{align}
			\mcB_{(r)}^{(k_1,\cdots,k_l)}:=\sum_{i_j=1,j\neq k_1\cdots k_l}^{n_j}(\mcA_{i_1\cdots i_d}^{(k_1,\cdots,k_l)})^r,\label{Eq of mcB}
		\end{align}
		where \(r\geq2,r\in\mbN\) and
		\begin{align}
			\mcA_{i_1\cdots i_d}^{(k_1,\cdots,k_l)}:=\prod_{j\neq k_1\cdots k_l}a_{i_j}^{(j)}.\label{Eq of mcA}
		\end{align}
	\end{itemize}
	Moreover, we say \(\bba^{(j)}\) is \emph{delocalized} if
	\begin{align}
		\lim_{n_j\to\infty}\Vert\bba^{(j)}\Vert_{\infty}=\lim_{n_j\to\infty}\max_{1\leq i_j\leq n_j}|a_{i_j}^{(j)}|=0,\label{Eq of localized}
	\end{align}
	otherwise, \(\bba^{(j)}\) is \emph{localized}.
 
    As the core tool of this article, for any \(d\)-fold tensor \(\bbT\in\mbR^{n_1\times\cdots\times n_d}\) and vectors \(\bba^{(j)}=(a_1^{(j)},\cdots,a_{n_j}^{(j)})'\in\mbR^{n_j},1\leq j\leq d\) the \(d\)-fold \emph{blockwise tensor contraction operator} \(\bbPhi_d\) is defined by 
    \begin{align}
	&\bbPhi_d:\mbR^{n_1\times\cdots\times n_d}\times\mathbb{S}^{n_1-1}\times\cdots\times\mathbb{S}^{n_d-1}\longrightarrow\mbR^{N\times N},\notag\\
	&\bbPhi_d(\bbT,\boldsymbol{a}^{(1)},\cdots,\boldsymbol{a}^{(d)})\longrightarrow\left(\begin{array}{cccc}
		\boldsymbol{0}_{n_1\times n_1}&\bbT^{12}&\cdots&\bbT^{1d}\\(\bbT^{12})'&\boldsymbol{0}_{n_2\times n_2}&\cdots&\bbT^{2d}\\\vdots&\vdots&\ddots&\vdots\\(\bbT^{1d})'&(\bbT^{2d})'&\cdots&\boldsymbol{0}_{n_d\times n_d}
	\end{array}\right),\label{Eq of tensor contraction}
\end{align}
where
$$\bbT^{ij}=\bbT(\bba^{(1)},\cdots,\boldsymbol{a}^{(i-1)},\boldsymbol{a}^{(i+1)},\cdots,\boldsymbol{a}^{(j-1)},\boldsymbol{a}^{(j+1)},\cdots,\bba^{(d)})\in\mbR^{n_i\times n_j}\quad{\rm for\ }i<j.$$
and
\begin{align}
	\bbT(\{\bba^{(1)},\cdots,\bba^{(d)}\}\backslash\{\bba^{(j_1)},\bba^{(j_2)}\}):=\left[\sum_{i_j=1,j\neq j_1,j_2}^{n_j}T_{i_1,\cdots,i_d}\mcA_{i_1\cdots i_d}^{(j_1,j_2)}\right]_{n_{j_1}\times n_{j_2}}\label{Eq of operator 1}
\end{align}
is \emph{second order contraction matrix} for any \(1\leq j_1\neq j_2\leq d\). In this article, we will study the asymptotic spectral properties of
\begin{align}
	\bbM:=\frac{1}{\sqrt{N}}\Phi_d(\boldsymbol{X},\bba^{(1)},\cdots,\bba^{(d)})\quad{\rm and}\quad\bbQ(z)=(\bbM-z\bbI_N)^{-1},\label{Eq of N and Q}
\end{align}
where \(\bbQ(z)\) is the resolvent of \(\bbM\) for any \(z\in\mbC^+\). Similar to (\ref{Eq of tensor contraction}), we also split \(\bbQ(z)=[\bbQ^{ij}(z)]_{d\times d}\) into \(d\times d\) blocks such that \(\bbQ^{ij}(z)\in\mbC^{n_i\times n_j}\), then for each diagonal block, let
\begin{align}
	\rho_i(z):=N^{-1}\tr(\bbQ^{ii}(z)),\quad\rho(z):=\sum_{i=1}^d\rho_i(z),\quad\mfm_i(z):=\mathbb{E}\left[\rho_i(z)\right],\quad\mfm(z):=\mbE[\rho(z)],\label{Eq of mi}
\end{align}
and
\begin{align}
	\bbm(z):=(\mfm_1(z),\cdots,\mfm_d(z))'\quad{\rm and}\quad\mfc:=(\mfc_1,\cdots,\mfc_d)'.\label{Eq of bbm and mfc}
\end{align}
\section{Properties of vector Dyson equation induced by the matrix \emph{M}}\label{Sec of Dyson}
\setcounter{equation}{0}
\def\theequation{\thesection.\arabic{equation}}
\setcounter{subsection}{0}
In this section, we will investigate several important properties of the vector Dyson equation induced by \(\bbM\), which is defined as follows:
\begin{align}
	-\frac{\mfc}{\bbg(z)}=z+\bbS_d\bbg(z),\label{Eq of MDE 3 order}
\end{align}
where \(\bbg(z)=(g_1(z),\cdots,g_d(z))'\) is the solution of (\ref{Eq of MDE 3 order}) and ``\(\frac{\mfc}{\bbg(z)}\)'' is the entrywise division as in (\ref{Eq of division of matrices}) and
\begin{align}
	\bbS_d:=\boldsymbol{1}_{d\times d}-\bbI_d.\label{Eq of bbS d}
\end{align}
The main reason for studying the vector Dyson equation \eqref{Eq of MDE 3 order} is that the mean of the trace of resolvent \(\bbQ(z)\) satisfies that (see Theorem \ref{Thm of approximation} for more details of \eqref{Eq of Dyson example})
\begin{align}
    -\frac{\mfc}{\bbm(z)}=z+\bbS_d\bbm(z)+\bbdel(z),\label{Eq of Dyson example}
\end{align}
where \(\bbm(z)=(N^{-1}\mbE[\tr(\bbQ^{11}(z))],\cdots,N^{-1}\mbE[\tr(\bbQ^{dd}(z))])'\) defined in \eqref{Eq of bbm and mfc} and \(\bbdel(z)\) is a small perturbation term such that $\lim_{N\to\infty}\Vert\bbve(z)\Vert_{\infty}=0$. It is easy to see that \eqref{Eq of MDE 3 order} is the limiting form of \eqref{Eq of Dyson example}. Therefore, the vector Dyson equation \eqref{Eq of MDE 3 order} is an important tool to investigate the asymptotic properties of $N^{-1}\mbE[\tr(\bbQ(z))]$. To be precise, we investigate the following properties of the vector Dyson equation \eqref{Eq of MDE 3 order}:
\begin{enumerate}
    \item \eqref{Eq of MDE 3 order} admits a unique analytic solution on \(\mbC^+\);
    \item For any vector-valued analytic function \(\bbv(z):\mathbb{C}^+\rightarrow\mbC^d\) satisfying \(-\frac{\mfc}{\bbv(z)}=z+\bbS_d\bbv(z)+\bbve(z)\), \(\bbve(z)\) here is a small perturbation term uniformly controlled over a given region, then $\Vert\bbg(z)-\bbv(z)\Vert_{\infty}$ is also small uniformly over the same region.
\end{enumerate}
Particularly, Property 2 above is called the \emph{stability of the vector Dyson equation \eqref{Eq of MDE 3 order}}. This stability immediately implies the asymptotic equivalence of \(\bbm(z)\) and \(\bbg(z)\). Further combining the fact that there exists a probability measure \(\nu\) associated with \(\bbg(z)\) (see Theorem \ref{Lem of finite support}), we can determine \(\nu\) is indeed the limiting spectral distribution of the matrix \(\bbM\).

Technically, to establish the stability of the vector Dyson equation \eqref{Eq of MDE 3 order}, we prove that the stability operator of \eqref{Eq of MDE 3 order}, which is a \(d\times d\) complex matrix (see \eqref{Eq of stability operator} later), is invertible in \S\ref{Sec of Stability operator}. Moreover, this stability operator also appears in the asymptotic mean and variance of the linear spectral statistics of the matrix \(\bbM\).

For a comprehensive discussion of the Dyson equation of random matrices, readers can refer to \cite{ajanki2019quadratic}. Without loss of generality, we assume \(\mfc_1=\max_{1\leq l\leq d}\mfc_l\) in \S\ref{Sec of Dyson}.
\subsection{Existence and uniqueness for the solution of (\ref{Eq of MDE 3 order})}\label{sec of Existence and uniqueness Dyson equation}
\begin{thm}\label{Thm of Dyson}
	Under Assumption {\rm \ref{Ap of dimension}}, {\rm (\ref{Eq of MDE 3 order})} admits a unique analytic solution on \(\mbC^+\).
\end{thm}
First, we will show that (\ref{Eq of MDE 3 order}) has a unique solution within the domain
$$\mathscr{B}_{\eta_0}^d:=\left\{\bbu(z)\in\mathscr{B}_+^d:\Vert\bbu\Vert_{\infty}\leq\eta_0^{-1}\mfc_1,\quad\min_{1\leq i\leq d}\Im(u_i(z))\geq\frac{\eta_0^3\mfc_d^2\mfc_1^{-1}}{[1+\mfc_1(d-1)]^2}\right\},$$
where
$$\mathscr{B}_+^d:=\left\{\bbu(z)\in\mbC^d{\rm\ is\ analytic\ for\ }z\in\mbC^+\ {\rm and\ }\min_{1\leq i\leq d}\Im(u_i(z))>0\right\}.$$
Here, we introduce the following metric:
$$D_{\mbC^+}(z_1,z_2)=\frac{|z_1-z_2|^2}{\Im(z_1)\Im(z_2)}\quad{\rm for\ }\forall z_1,z_2\in\mbC^+.$$
Besides, we define a function mapping \(\boldsymbol{\Psi}_d:\mathscr{B}_+^d\to\mathscr{B}_+^d\) as follows:
$$\boldsymbol{\Psi}_d(z,\bbu)=-\frac{\mfc}{z+\boldsymbol{S}_d\bbu(z)}.$$
We have the following result:
\begin{lem}[\(\boldsymbol{\Psi}_d\) is a contraction mapping]\label{Lem of contraction}
	Under Assumption {\rm \ref{Ap of dimension}}, for any \(\eta_0>0\), let 
	\begin{align}
		\mathbb{H}_{\eta_0}:=\{z\in\mbC_{\eta_0}^+,|z|\leq\eta_0^{-1}\},\label{Eq of H0 region}
	\end{align}
	then \(\boldsymbol{\Psi}_d(z,\cdot)\) maps \(\mathscr{B}_{\eta_0}^d\) to itself such that
	$$\max_{1\leq j\leq d}D_{\mbC^+}(\boldsymbol{\Psi}_d(z,\bbu)_j,\boldsymbol{\Psi}_d(z,\bbw)_j)\leq(1+\eta_0^2\Vert\boldsymbol{S}_d\Vert^{-1})^{-2}\max_{1\leq j\leq d}D_{\mbC^+}(u_j(z),w_j(z)),$$
	for any \(z\in\mathbb{H}_{\eta_0}\) and \(\bbu,\bbw\in\mathscr{B}_{\eta_0}^+\), where \(\boldsymbol{\Psi}_d(z,\bbu)_j\) represents the \(j\)-th entry of \(\boldsymbol{\Psi}_d(z,\bbu)\).
\end{lem}
\begin{proof}
	First, notice that
	$$|\boldsymbol{\Psi}_d(z,\bbu)_j|\leq\Im(z+\boldsymbol{S}_d\bbu)_j^{-1}\max_{1\leq i\leq d}\mathfrak{c}_i\leq\eta_0^{-1}\mfc_1$$
	and
	$$|\boldsymbol{\Psi}_d(z,\bbu)_j|\geq\frac{\min_{1\leq i\leq d}\mathfrak{c}_i}{|z|+|(\boldsymbol{S}_d\bbu)_j|}\geq\frac{\eta_0\mfc_d}{1+\mfc_1(d-1)},$$
	where the last inequality is valid due to \(|z|\leq\eta_0^{-1}\) for \(z\in\mathbb{H}_{\eta_0}\), \(\Vert\bbu\Vert_{\infty}\leq\mfc_{(1)}\eta_0^{-1}\) and \(|(\boldsymbol{S}_d\bbu)_j|\leq\sum_{k=1,k\neq j}^d|u_k|\) for \(\bbu\in\msB_{\eta_0}^d\), which implies that
	$$\Im(\boldsymbol{\Psi}_d(z,\bbu)_j)=\frac{\mfc_j\Im(z+\boldsymbol{S}_d\bbu)_j}{|z+(\boldsymbol{S}_d\bbu)_j|^2}\geq\Im(z)|\boldsymbol{\Psi}_d(z,\bbu)_j|^2\mathfrak{c}_j^{-1}\geq \frac{\eta_0^3\mfc_d^2\mfc_1^{-1}}{[1+\mfc_1(d-1)]^2}.$$
	Hence, \(\boldsymbol{\Psi}_d(z,\cdot)\) maps \(\mathscr{B}_{\eta_0}^+\) to itself. Next, for any \(\bbu,\bbw\in\mathscr{B}_{\eta_0}^+\), we have
	\begin{align}
		&D_{\mbC^+}(\boldsymbol{\Psi}_d(z,\bbu)_j,\boldsymbol{\Psi}_d(z,\bbw)_j)=D_{\mbC^+}(z+(\boldsymbol{S}_d\bbu)_j,z+(\boldsymbol{S}_d\bbw)_j)\notag\\
		&=D_{\mbC^+}({\rm i\ Im}z+(\boldsymbol{S}_d\bbu)_j,{\rm i\ Im}z+(\boldsymbol{S}_d\bbw)_j)\notag\\
		&\leq\left(1+\frac{\Im(z)}{(\boldsymbol{S}_d\bbu)_j}\right)^{-1}\left(1+\frac{\Im(z)}{(\boldsymbol{S}_d\bbw)_j}\right)^{-1}D_{\mbC^+}((\boldsymbol{S}_d\bbu)_j,(\boldsymbol{S}_d\bbw)_j)\notag\\
		&\leq(1+\eta_0^2\Vert\boldsymbol{S}_d\Vert^{-1})^{-2}\max_{1\leq j\leq d}D_{\mbC^+}(u_j(z),w_j(z)),\notag
	\end{align}
	where we use some basic properties of \(D_{\mbC^+}(\cdot ,\cdot)\) in proving above inequalities, readers can refer to Lemma 4.2 in \cite{ajanki2019quadratic} for details.
\end{proof}
Now, the existence and uniqueness of (\ref{Eq of MDE 3 order}) for \(z\in\mathbb{H}_{\eta_0}\) can be proved by Lemma \ref{Lem of contraction} and Banach fixed-point theorem. Since \(\eta_0\) is an arbitrary positive number, we can extend this conclusion  to \(\mathscr{B}_+^d\) by letting \(\eta_0\to0\), which completes the proof of Theorem \ref{Thm of Dyson}.
\subsection{The invertibility of stability operators}\label{Sec of Stability operator}
We study the stability operator induced by the vector Dyson equation \eqref{Eq of MDE 3 order} in preparation for Theorem \ref{Thm of Stability}. To formalize, we first define the self-energy operator as follows:
\begin{align}
	\boldsymbol{F}^{(d)}:=\boldsymbol{F}^{(d)}(z)={\rm diag}(|\mfc^{-1}\circ\boldsymbol{g}(z)|)\boldsymbol{S}_d{\rm diag}(|\boldsymbol{g}(z)|)=[F_{ij}(z)]_{d\times d},\label{Eq of operator F}
\end{align}
where \(F_{ii}(z)\equiv0\) and \(F_{ij}(z)=\mathfrak{c}_i^{-1}|g_i(z)g_j(z)|\) for \(i\neq j\). Then the stability operator of \eqref{Eq of MDE 3 order} is defined as
\begin{align}
	\boldsymbol{B}^{(d)}:=\boldsymbol{B}^{(d)}(z)=\boldsymbol{I}_d-{\rm diag}(\mfc^{-1}\circ\boldsymbol{g}(z)^{\circ2})\boldsymbol{S}_d={\rm diag}(|\boldsymbol{g}|)(\boldsymbol{I}_d-{\rm diag}(e^{2{\rm i}\boldsymbol{q}})\boldsymbol{F}){\rm diag}(|\boldsymbol{g}|)^{-1}.\label{Eq of stability operator}
\end{align}
In this section, we first prove that
\begin{pro}\label{Pro of stability operator}
	Under Assumption {\rm \ref{Ap of dimension}}, for any \(\eta_0>0\) and \(z\in\mathbb{H}_{\eta_0}\) in \eqref{Eq of H0 region}, the stability operator \eqref{Eq of stability operator} is invertible.
\end{pro}
For simplicity, we simplify \(\bbF^{(d)}\) and \(\bbB^{(d)}\) by \(\bbF\) and \(\bbB\), respectively. Before proving the above proposition, we need some preliminaries. Notice that all entries of \(\boldsymbol{F}\) are non-negative, then according to the Perron-Frobenius theorem, there exists a positive vector \(\boldsymbol{f}:=\boldsymbol{f}(z)\) such that \(\boldsymbol{F}\boldsymbol{f}=\Vert\boldsymbol{F}\Vert\boldsymbol{f}\). In addition, taking the imaginary part of (\ref{Eq of MDE 3 order}), i.e.
\begin{align}
	{\rm (\ref{Eq of MDE 3 order})}\ \Rightarrow\ \frac{\mfc\circ\Im(\bbg)}{|\bbg|^2}=\Im(z)+\bbS_d\Im(\bbg),\label{Eq of VDE image}
\end{align}
which yields that
\begin{align}
	\sin\boldsymbol{q}=\Im(z)\mfc^{-1}\circ|\boldsymbol{g}|+\boldsymbol{F}\sin\boldsymbol{q},\label{Eq of MDE taking image}
\end{align}
where $\boldsymbol{g}=e^{{\rm i}\boldsymbol{q}}\circ|\boldsymbol{g}|$ and \(\sin\boldsymbol{q}=\frac{\Im(\boldsymbol{g})}{|\boldsymbol{g}|}\). Therefore, we can obtain
$$\langle\boldsymbol{f},\sin\boldsymbol{q}\rangle=\Im(z)\langle\boldsymbol{f},\mfc^{-1}\circ|\boldsymbol{g}|\rangle+\Vert\boldsymbol{F}\Vert\langle\boldsymbol{f},\sin\boldsymbol{q}\rangle,$$
i.e.
\begin{align}
	\Vert\boldsymbol{F}\Vert=1-\frac{\Im(z)\langle\boldsymbol{f},\mfc^{-1}\circ|\boldsymbol{g}|\rangle}{\langle\boldsymbol{f},\sin\boldsymbol{q}\rangle}<1,\ {\rm for\ }z\in\mathbb{C}^+.\label{Eq of F norm}
\end{align}
Hence, \(\boldsymbol{I}_d-\boldsymbol{F}\) is invertible. Moreover, by (\ref{Eq of stability operator}), we have
	$$\bbB(z)={\rm diag}(\mfc^{-1/2}\circ\boldsymbol{g}(z))(\bbI_d-{\rm diag}(\mfc^{-1/2}\circ\boldsymbol{g}(z))\boldsymbol{S}_d{\rm diag}(\mfc^{-1/2}\circ\boldsymbol{g}(z))){\rm diag}(\mfc^{-1/2}\circ\boldsymbol{g}(z))^{-1},$$
where
	\begin{align*}
		&\bbI_d-{\rm diag}(\mfc^{-1/2}\circ\boldsymbol{g}(z))\boldsymbol{S}_d{\rm diag}(\mfc^{-1/2}\circ\boldsymbol{g}(z))\\
		&=\diag(e^{{\rm i}\bbq(z)})(\diag(e^{-2{\rm i}\bbq(z)})-{\rm diag}(|\mfc^{-1/2}\circ\boldsymbol{g}(z)|)\boldsymbol{S}_d{\rm diag}(|\mfc^{-1/2}\circ\boldsymbol{g}(z)|))\diag(e^{{\rm i}\bbq(z)}).
	\end{align*}
to prove that \(\boldsymbol{B}(z)\) is invertible, it is enough to prove that 
$$\diag(e^{-2{\rm i}\bbq(z)})-{\rm diag}(|\mfc^{-1/2}\circ\boldsymbol{g}(z)|)\boldsymbol{S}_d{\rm diag}(|\mfc^{-1/2}\circ\boldsymbol{g}(z)|)$$
is invertible. In fact, we can prove the above matrix is invertible by showing the spectral gap of \({\rm diag}(|\mfc^{-1/2}\circ\boldsymbol{g}(z)|)\boldsymbol{S}_d{\rm diag}(|\mfc^{-1/2}\circ\boldsymbol{g}(z)|)\) is positive. To make this precise, we start with a definition:
\begin{Def}\label{Def of spectral gap}
	For any matrix \(\bbA\), the spectral gap \({\rm Gap}(\bbA)\) is the difference between the two largest eigenvalues of \(\sqrt{\bbA\bbA^*}\).
\end{Def}
Next, we need the following lemma:
\begin{lem}\label{Lem of spectral gap}
	Let \(\bbv=(v_1,\cdots,v_d)'\in\mathbb{R}^d\) such that \(0<v_d\leq\cdots\leq v_1\), where \(d\geq3\), then 
	$${\rm Gap}(\boldsymbol{v}\boldsymbol{v}'-{\rm diag}(\boldsymbol{v}^{\circ2}))>\sum_{i=3}^dv_i^2.$$
\end{lem}
\begin{proof}
	First, let \(t^{-1}\) be the eigenvalue of \(\boldsymbol{v}\boldsymbol{v}'-{\rm diag}(\boldsymbol{v}^{\circ2})\), by the matrix determinant lemma, it implies that
	$$0=\det(t(\boldsymbol{v}\boldsymbol{v}'-{\rm diag}(\boldsymbol{v}^{\circ2}))-\boldsymbol{I}_d)=(-1)^d\det(\boldsymbol{I}_d+t{\rm diag}(\boldsymbol{v}^{\circ2}))(1-t\boldsymbol{v}'(\boldsymbol{I}_d+t{\rm diag}(\boldsymbol{v}^{\circ2}))^{-1}\boldsymbol{v}).$$
	In fact, \(\bbv\bbv'-\diag(\bbv^{\circ2})\) always has one positive eigenvalue due to \(\boldsymbol{1}_d'(\bbv\bbv'-\diag(\bbv^{\circ2}))\boldsymbol{1}_d=\sum_{k\neq l}^dv_kv_l>0\). Suppose \(t>0\), that is, \(t^{-1}\) is a positive eigenvalue, we can obtain \(1-t\boldsymbol{v}'(\boldsymbol{I}_d+t{\rm diag}(\boldsymbol{v}^{\circ2}))^{-1}\boldsymbol{v}=0\), and the equation
	$$\sum_{i=1}^d\frac{1}{1+tv_i^2}=d-1$$
	has total \(d-1\) negative roots denoted by \(t_i\) such that \(0>t_2>\cdots>t_d\), where \(v_i^2<-t_i^{-1}<v_{i-1}^2\) for \(i=2,\cdots,d\) and one positive zero \(t_1\). Hence, we conclude that \(\boldsymbol{v}\boldsymbol{v}'-{\rm diag}(\boldsymbol{v}^{\circ2})\) only has one positive eigenvalue \(t_1^{-1}\). Let \(l=t^{-1}\), then we obtain
	$$\prod_{i=1}^d(l+v_i^2)-\sum_{i=1}^dv_i^2\prod_{j\neq i}^d(l+v_j^2)=0.$$
	Since the coefficient of \(l^{d-1}\) is zero, then \(\sum_{i=1}^dt_i^{-1}=0\). Next, suppose \(s^{-1}<0\) is a negative eigenvalue of \(\boldsymbol{v}\boldsymbol{v}'-{\rm diag}(\boldsymbol{v}^2)\) such that \(1-s\boldsymbol{v}'(\boldsymbol{I}_d+s{\rm diag}(\boldsymbol{v}^{\circ2}))^{-1}\boldsymbol{v}=0\), then we have 
	$$t_1^{-1}+s^{-1}\geq t_1^{-1}+t_2^{-1}=-\sum_{i=3}^d t_i^{-1}\geq\sum_{i=3}^dv_i^2.$$
	On the other hand, if \(\det(\boldsymbol{I}_d+s{\rm diag}(\boldsymbol{v}^{\circ2}))=0\), which implies that \(s=-v_i^{-2}\). Consider two possible cases. First, if \(v_1=v_2\), then
	$$t_1^{-1}+s^{-1}\geq-\sum_{i=2}^d t_i^{-1}-v_1^2\geq\sum_{i=2}^d v_i^2-v_1^2=\sum_{i=3}^dv_i^2.$$
	Second, if \(v_1>v_2\), then we claim that \(s\neq -v_1^{-2}\). Otherwise, there exists a nonzero \(\bbx\in\mathbb{R}^d\) such that
	$$(\boldsymbol{v}\boldsymbol{v}'-{\rm diag}(\boldsymbol{v}^{\circ2})+v_1^2\bbI_d)\bbx=0\ \Rightarrow\ v_1^2 x_k+v_k\sum_{j\neq k}^d v_j x_j=0.$$
	Let \(k=1\), it implies that \(\langle\boldsymbol{v},\bbx\rangle=0\). When \(k>1\), notice that
	\begin{align*}
		(v_1^2-v_k^2)x_k=-v_k\sum_{j=1}^dv_jx_j=0,
	\end{align*}
	since \(v_1>v_k\) for \(k>1\), we have \(x_k=0\) for \(k>1\), which further implies that \(x_1=0\) due to \(\langle\boldsymbol{v},\bbx\rangle=0\). It is a contradiction since \(\bbx\) is nonzero. As a result, we obtain that
	$$t_1^{-1}+s^{-1}\geq-\sum_{i=2}^d t_i^{-1}-v_2^2\geq\sum_{i=2}^d v_i^2-v_2^2=\sum_{i=3}^dv_i^2.$$
	By the Definition \ref{Def of spectral gap}, we have
	$${\rm Gap}(\boldsymbol{v}\boldsymbol{v}'-{\rm diag}(\boldsymbol{v}^2))=\min\{t_1^{-1}+s^{-1}:s<0\ {\rm and\ }s^{-1}\ {\rm is\ an\ eigenvalue\ of\ }\boldsymbol{v}\boldsymbol{v}'-{\rm diag}(\boldsymbol{v}^2)\}\geq\sum_{i=3}^dv_i^2,$$
	which completes our proof.
\end{proof}
\begin{remark}\label{Rem of bounded solutions}
	Since \(\boldsymbol{F}(z)\) in (\ref{Eq of operator F}) and \({\rm diag}(|\mfc^{-1/2}\circ\boldsymbol{g}(z)|)\boldsymbol{S}_d{\rm diag}(|\mfc^{-1/2}\circ\boldsymbol{g}(z)|)\) are similar, then the largest eigenvalue of \({\rm diag}(|\mfc^{-1/2}\circ\boldsymbol{g}(z)|)\boldsymbol{S}_d{\rm diag}(|\mfc^{-1/2}\circ\boldsymbol{g}(z)|)\) is the same as \(\bbF(z)\), which is strictly less than 1, then by Lemma \ref{Lem of spectral gap}, we have
	$$\sum_{i=1}^d\mathfrak{c}_i^{-1}|g_i(z)|^2-\max_{1\leq i\leq d}\mathfrak{c}_i^{-1}|g_i(z)|^2<\Vert\bbF(z)\Vert<1.$$
	Therefore, it implies that \(|g_i(z)|<\sqrt{\mathfrak{c}_i}\) for all \(i=1,\cdots,d\) except \(i=\arg\max_{1\leq i\leq d}\mathfrak{c}_i^{-1}|g_i(z)|^2\).
\end{remark}
Now, let us  prove Proposition \ref{Pro of stability operator} as follows:
\begin{proof}[Proof of Proposition \ref{Pro of stability operator}]
	By (\ref{Eq of stability operator}), since
	$$\bbB(z)={\rm diag}(\mfc^{-1/2}\circ\boldsymbol{g}(z))(\bbI_d-{\rm diag}(\mfc^{-1/2}\circ\boldsymbol{g}(z))\boldsymbol{S}_d{\rm diag}(\mfc^{-1/2}\circ\boldsymbol{g}(z))){\rm diag}(\mfc^{-1/2}\circ\boldsymbol{g}(z))^{-1},$$
	and
	\begin{align*}
		&\bbI_d-{\rm diag}(\mfc^{-1/2}\circ\boldsymbol{g}(z))\boldsymbol{S}_d{\rm diag}(\mfc^{-1/2}\circ\boldsymbol{g}(z))\\
		&=\diag(e^{{\rm i}\bbq(z)})(\diag(e^{-2{\rm i}\bbq(z)})-{\rm diag}(|\mfc^{-1/2}\circ\boldsymbol{g}(z)|)\boldsymbol{S}_d{\rm diag}(|\mfc^{-1/2}\circ\boldsymbol{g}(z)|))\diag(e^{{\rm i}\bbq(z)}),
	\end{align*}
	then \(\bbB(z)\) is invertible if and only if \(\diag(e^{-2{\rm i}\bbq(z)})-{\rm diag}(|\mfc^{-1/2}\circ\boldsymbol{g}(z)|)\boldsymbol{S}_d{\rm diag}(|\mfc^{-1/2}\circ\boldsymbol{g}(z)|)\) is invertible. Based on Lemma \ref{Lem of spectral gap}, we have
	$${\rm Gap}({\rm diag}(|\mfc^{-1/2}\circ\boldsymbol{g}(z)|)\boldsymbol{S}_d{\rm diag}(|\mfc^{-1/2}\circ\boldsymbol{g}(z)|))>\sum_{i=3}^{d}\mathfrak{c}_{(i)}^{-1}|g_{(i)}(z)|^2,$$
	where \(\mathfrak{c}_{(i)}^{-1}|g_{(i)}(z)|^2\) is the \(i\)-th largest entries in \(|\mfc^{-1/2}\circ\boldsymbol{g}(z)|^2\). Recall that \(|g_i(z)|\geq C_{\eta_0,d,\mfc}\) when \(z\in\mathbb{H}_{\eta_0}\), see Lemma \ref{Lem of contraction}, so the spectral gap of \({\rm diag}(|\mfc^{-1/2}\circ\boldsymbol{g}(z)|)\boldsymbol{S}_d{\rm diag}(|\mfc^{-1/2}\circ\boldsymbol{g}(z)|)\) is positive for \(z\in\mathbb{H}_{\eta_0}\). By the Remark \ref{Rem of bounded solutions}, we know that the largest eigenvalue of \({\rm diag}(|\mfc^{-1/2}\circ\boldsymbol{g}(z)|)\boldsymbol{S}_d{\rm diag}(|\mfc^{-1/2}\circ\boldsymbol{g}(z)|)\) is strictly smaller than 1, which further implies that
	$${\rm diag}(e^{-2{\rm i}\boldsymbol{q}})-{\rm diag}(|\mfc^{-1/2}\circ\boldsymbol{g}(z)|)\boldsymbol{S}_d{\rm diag}(|\mfc^{-1/2}\circ\boldsymbol{g}(z)|)$$
	is invertible for \(z\in\mathbb{C}^+\). In fact, denote \(\lambda^{\pm}\) to be the largest positive \((+)\) and smallest negative \((-)\) eigenvalue of \({\rm diag}(|\mfc^{-1/2}\circ\boldsymbol{g}(z)|)\boldsymbol{S}_d{\rm diag}(|\mfc^{-1/2}\circ\boldsymbol{g}(z)|)\), then
	$$\min_{1\leq i\leq d}|e^{-2{\rm i}q_i}-\lambda^-|>|-1-\lambda^-|>|\lambda^++\lambda^-|={\rm Gap}({\rm diag}(|\mfc^{-1/2}\circ\boldsymbol{g}(z)|)\boldsymbol{S}_d{\rm diag}(|\mfc^{-1/2}\circ\boldsymbol{g}(z)|))$$
	which suggests all \(e^{-2{\rm i}q_i}\) are not the eigenvalues of \(\bbF\).
\end{proof}
After establishing the invertibility of \(\bbB(z)\), we further need the following more general results.
\begin{pro}\label{Pro of invertible matrices}
	Under Assumption {\rm \ref{Ap of dimension}}, for any \(\eta_0>0\) and \(z_1,z_2\in\mathbb{H}_{\eta_0}\) in {\rm (\ref{Eq of H0 region})},
	\begin{align}
		\boldsymbol{\Lambda}^{(d)}(z_1,z_2)&:=\boldsymbol{I}_d-{\rm diag}(\mfc^{-1/2}\circ\boldsymbol{g}(z_1))\boldsymbol{S}_d{\rm diag}(\mfc^{-1/2}\circ\boldsymbol{g}(z_2))\label{Eq of invertible}\\
		\boldsymbol{\Pi}^{(d)}(z_1,z_2)&:=\boldsymbol{I}_d-{\rm diag}(\mfc^{-1}\circ\boldsymbol{g}(z_1)\circ\bbg(z_2))\boldsymbol{S}_d\label{Eq of invertible 2}
	\end{align}
	are invertible.
\end{pro}
\begin{remark}\label{Rem of invertible submatrix}
	In particular, when \(z_1=z_2\), we have \(\bbPi^{(d)}(z,z)=\bbB^{(d)}(z)\). The purpose of proving the above proposition is that \(\bbPi^{(d)}(z,z)^{-1}\) will appear in the asymptotic mean and variance of the LSS of the matrix \(\bbM\). 
\end{remark}
Similarly, we will simplify \(\bbPi^{(d)}(z,z)\) by \(\bbPi(z,z)\) in following proofs, as does others.
\begin{proof}[Proof of Proposition \ref{Pro of invertible matrices}]
	Notice that \(\bbLa^{(d)}(z_1,z_2)\) and \(\bbPi^{(d)}(z_1,z_2)\) are similar, so it is enough to prove that one of them is invertible. We have already shown that \({\rm diag}(e^{-2{\rm i}\boldsymbol{q}})-\boldsymbol{F}\) is invertible in Proposition \ref{Pro of stability operator}, which implies that 
	$${\rm diag}(e^{-2{\rm i}\boldsymbol{q}})-{\rm diag}(|\mfc^{-1/2}\circ\boldsymbol{g}(z)|)\boldsymbol{S}_d{\rm diag}(|\mfc^{-1/2}\circ\boldsymbol{g}(z)|)={\rm diag}(e^{-{\rm i}\boldsymbol{q}})\boldsymbol{\Lambda}(z,z){\rm diag}(e^{-{\rm i}\boldsymbol{q}})$$
	is also invertible. Next, let us  further consider the case of \(\boldsymbol{\Lambda}(z_1,z_2)\) for \(z_1\neq z_2\in\mathbb{H}_{\eta_0}\). Notice that 
	\begin{align}
		\bbLa(z_1,z_2)=\diag(\bbg(z_1))^{1/2}\diag(\bbg(z_2))^{-1/2}(\bbI_d-\bbGa(z_1,z_2))\diag(\bbg(z_2))^{1/2}\diag(\bbg(z_1))^{-1/2},\label{Eq of Gamma}
	\end{align}
	where
	$$\boldsymbol{\Gamma}(z_1,z_2):={\rm diag}(\mfc^{-1/2}\circ\sqrt{|\boldsymbol{g}(z_1)\circ\boldsymbol{g}(z_2)|})\boldsymbol{S}_d{\rm diag}(\mfc^{-1/2}\circ\sqrt{|\boldsymbol{g}(z_1)\circ\boldsymbol{g}(z_2)|}).$$
	Hence, \(\bbLa(z_1,z_2)\) is invertible if and only if \(\bbI_d-\bbGa(z_1,z_2)\) is invertible. For any unit vector \(\boldsymbol{x}\in\mathbb{R}^d\), we have
	\begin{align}
		\Vert\boldsymbol{\Gamma}(z_1,z_2)\boldsymbol{x}\Vert_2^2&=\sum_{i=1}^d\Big(\sum_{j\neq i}^d(\mathfrak{c}_i\mathfrak{c}_j)^{-1/2}|g_i(z_1)g_i(z_2)g_j(z_1)g_j(z_2)|^{1/2}x_j\Big)^2\notag\\
		&\leq\sum_{i=1}^d\Big(\sum_{j\neq i}^d(\mathfrak{c}_i\mathfrak{c}_j)^{-1/2}|g_i(z_1)g_j(z_1)x_j|\Big)\Big(\sum_{j\neq i}^d(\mathfrak{c}_i\mathfrak{c}_j)^{-1/2}|g_i(z_2)g_j(z_2)x_j|\Big)\notag\\
		&\leq\left[\sum_{i=1}^d\Big(\sum_{j\neq i}^d(\mathfrak{c}_i\mathfrak{c}_j)^{-1/2}|g_i(z_1)g_j(z_1)x_j|\Big)^2\times\sum_{i=1}^d\Big(\sum_{j\neq i}^d(\mathfrak{c}_i\mathfrak{c}_j)^{-1/2}|g_i(z_2)g_j(z_2)x_j|\Big)^2\right]^{1/2}\notag\\
		&=\Vert\boldsymbol{\Gamma}(z_1,z_1)|\bbx|\Vert_2\times\Vert\boldsymbol{\Gamma}(z_2,z_2)|\bbx|\Vert_2\leq\Vert\boldsymbol{\Gamma}(z_1,z_1)\Vert\times\Vert\boldsymbol{\Gamma}(z_2,z_2)\Vert.\notag
	\end{align}
	Since \(\bbF(z_1)\) in (\ref{Eq of operator F}) and
	$$\bbGa(z_1,z_1)={\rm diag}(|\mfc^{-1/2}\circ\boldsymbol{g}(z_1)|)\boldsymbol{S}_d{\rm diag}(|\mfc^{-1/2}\circ\boldsymbol{g}(z_1)|)$$
	are similar, their eigenvalues coincide. By (\ref{Eq of F norm}) and the symmetry of \(\bbGa(z_1,z_1)\), we conclude that \(\Vert\boldsymbol{\Gamma}(z_1,z_1)\Vert\leq\Vert\bbF(z_1)\Vert<1\), as does \(\Vert\boldsymbol{\Gamma}(z_2,z_2)\Vert\) and \(\Vert\boldsymbol{\Gamma}(z_1,z_2)\Vert\). Furthermore, by Lemma \ref{Lem of spectral gap}, we can also conclude that \({\rm Gap}(\boldsymbol{\Gamma}(z_1,z_2))>\sum_{i=3}^d\mathfrak{c}_{(i)}^{-1}|g_{(i)}(z_1)g_{(i)}(z_2)|\), so we can show that \(\boldsymbol{\Lambda}(z_1,z_2)\) is invertible for any \(z_1,z_2\in\mathbb{C}^+\) by the same arguments as those in Proposition \ref{Pro of stability operator}.
\end{proof}
Finally, we need the upper bound of the spectral norm of \(\bbPi^{(d)}(z_1,z_2)^{-1}\), i.e.
\begin{pro}\label{Pro of inverse norm}
	Under Assumption {\rm \ref{Ap of dimension}}, for any \(\eta_0>0\) and \(z_1,z_2\in\mathbb{H}_{\eta_0}\) in {\rm (\ref{Eq of H0 region})}, we have
	$$\Vert\bbPi^{(d)}(z_1,z_2)^{-1}\Vert,\Vert\bbLa^{(d)}(z_1,z_2)^{-1}\Vert\leq C_{d,\mfc}\eta_0^{-4}.$$
\end{pro}
\begin{proof}
	Let us  first prove \(\Vert\bbLa(z_1,z_2)^{-1}\Vert\leq C_{d,\boldsymbol{\mfc}}\eta_0^{-4}\). By \eqref{Eq of Gamma}, we know that 
    \begin{align*}
        \Vert\bbLa(z_1,z_2)^{-1}\Vert\leq C_{d,\boldsymbol{\mfc}}\eta_0^{-2}(1-\Vert\bbGa(z_1,z_2)\Vert)^{-1}.
    \end{align*}
    where we use the fact that \(C_{d,\boldsymbol{\mfc}}\eta_0\leq |g_i(z)|\leq\eta_0^{-1}\) by Lemma \ref{Lem of contraction}. Since we have shown that \(\Vert\bbGa(z_1,z_2)\Vert^2\leq\Vert\bbF(z_1)\Vert\cdot\Vert\bbF(z_2)\Vert\) in proofs of Proposition \ref{Pro of invertible matrices}, then 
    $$(1-\Vert\bbGa(z_1,z_2)\Vert)^{-1}\leq\max_{i=1,2}(1-\Vert\bbF(z_i)\Vert)^{-1}$$
    By \eqref{Eq of F norm}, we know that
    \begin{align*}
        (1-\Vert\bbF(z)\Vert)^{-1}=\frac{\langle\bbf,\sin\bbq\rangle}{\Im(z)\langle\bbf,\boldsymbol{\mfc}^{-1}\circ|\bbg|\rangle},
    \end{align*}
    since \(|g_i(z)|\geq C_{d,\boldsymbol{\mfc}}\eta_0\), then due to \(\bbf\) is a positive vector, we have
	$$\frac{\langle\bbf,\sin\bbq\rangle}{\Im(z)\langle\bbf,\boldsymbol{\mfc}^{-1}\circ\bbg\rangle}\leq\frac{\langle\bbf,\boldsymbol{1}_d\rangle}{\eta_0\min_{1\leq i\leq d}\mfc_i^{-1}|g_i(z)|\langle\bbf,\boldsymbol{1}_d\rangle}\leq C_{d,\boldsymbol{\mfc}}\eta_0^{-2},$$
    which implies that \(\Vert\bbLa(z_1,z_2)^{-1}\Vert\leq C_{d,\boldsymbol{\mfc}}\eta_0^{-4}\). Similarly, for \(\bbPi(z_1,z_2)\), since
	\begin{align*}
		\bbPi(z_1,z_2)&=\diag(\mfc^{1/2}\circ\bbg(z_2)^{-1/2}\circ\bbg(z_2)^{-1/2})(\bbI_d-\bbGa(z_1,z_2))\diag(\mfc^{-1/2}\circ\bbg(z_2)^{1/2}\circ\bbg(z_2)^{1/2}),
	\end{align*}
	we can complete our proof by repeating previous arguments.
\end{proof}
\subsection{Stability of the vector Dyson equation \texorpdfstring{\eqref{Eq of MDE 3 order}}{(C.1)}}\label{sec of Stability}

Roughly speaking, the stability of the vector Dyson equation \eqref{Eq of MDE 3 order} means that if a vector-valued function \(\bbv(z)\) satisfies a perturbed version of the vector Dyson equation with a small perturbation term \(\bbve(z)\) uniformly controlled over a given region \(\tilde{\mcS}_{\eta_0}\) as in \eqref{Eq of stability region 1} later, then the difference between \(\bbv(z)\) and the solution \(\bbg(z)\) of the original equation \eqref{Eq of MDE 3 order} is also small uniformly over \(\tilde{\mcS}_{\eta_0}\). This stability is a key tool to show that the empirical spectral distribution (ESD) of the matrix \(\bbM\) asymptotically converges to the measure \(\nu\) associated with \(\bbg(z)\), which we will prove to be the LSD of the matrix \(\bbM\) later. Here, we first define a region in the upper complex plane as follows:
\begin{align}
	\tilde{\mcS}_{\eta_0}:=\left\{z\in\mbC^+:{\rm dist}(z,[-\zeta,\zeta])\geq\eta_0,|\Re(z)|\leq\eta_0^{-1}\right\},\label{Eq of stability region 1}
\end{align}
where \(\zeta\) is the right and left boundary of the limiting spectral distribution \(\nu(\cdot)\) of the matrix \(\bbM\) in (\ref{Eq of support boundary}) later and
$${\rm dist}(z,[-\zeta,\zeta]):=\min\{|z-x|:x\in[-\zeta,\zeta]\}.$$
Here, we require that $\eta_0>0$ be sufficiently small so that \(\zeta<\eta_0^{-1}\). Next, let us  show that
\begin{thm}[Stability]\label{Thm of Stability}
	For any \(\eta_0>0\) and \(z\in\tilde{\mcS}_{\eta_0}\) in {\rm (\ref{Eq of stability region 1})}, let \(\bbv(z)=(v_1(z),\cdots,v_d(z))'\) be a \(d\)-dimensional analytic function on \(\mbC^+\) such that
	$$\bbve(z)=\frac{\mfc}{\bbv(z)}+z+\bbS_d\bbv(z),$$
	satisfies \(\sup_{z\in\tilde{\mcS}_{\eta_0}}\Vert\boldsymbol{\varepsilon}(z)\Vert_{\infty}=\mrO(\eta_0^{-\beta}N^{-\alpha})\) for some \(\alpha,\beta>0\),
	then we have  
	$$\sup_{z\in\tilde{\mcS}_{\eta_0}}\Vert\boldsymbol{g}(z)-\bbv(z)\Vert_{\infty}\leq\mrO(\eta_0^{-(\beta+4)}N^{-\alpha}),$$
	where \(\boldsymbol{g}(z)\) is the solution of {\rm (\ref{Eq of MDE 3 order})}.
\end{thm}
\begin{proof}
	First, let us  split the region \(\tilde{\mcS}_{\eta_0}\) into two parts, let us  define
	$$\tilde{\mcS}_{\eta_0}^1:=\{z\in\tilde{\mcS}_{\eta_0}:\Im(z)\leq\eta_0^{-1}\}\quad{\rm and}\quad\tilde{\mcS}_{\eta_0}^2:=\tilde{\mcS}_{\eta_0}\backslash\tilde{\mcS}_{\eta_0}^1.$$
	For \(z\in\tilde{\mcS}_{\eta_0}^2\), we have
	\begin{align}
		&|v_i(z)-g_i(z)|=\Big|\frac{\mfc_i}{z+\sum_{j\neq i}g_j(z)}-\frac{\mfc_i}{z+\sum_{j\neq i}v_j(z)+\varepsilon_i}\Big|\notag\\
		&=\mfc_i\Big|\frac{\sum_{j\neq i}v_j(z)-g_j(z)+\varepsilon_i}{(z+\sum_{j\neq i}g_j(z))(z+\sum_{j\neq i}v_j(z)+\varepsilon_i)}\Big|\notag\\
		&\leq\frac{\mfc_i}{\Im(z)^2}\sum_{j\neq i}|v_j(z)-g_j(z)|+\frac{\mfc_i|\varepsilon_i|}{\Im(z)^2}\notag\\
		&\leq\mfc_i\eta_0^2\sum_{j\neq i}|v_j(z)-g_j(z)|+\eta_0^2|\varepsilon_i|.\notag
	\end{align}
	In other words, we conclude that
	$$|\bbv(z)-\bbg(z)|\leq\eta_0^2{\rm diag}(\boldsymbol{\mfc})\bbS_d|\bbv(z)-\bbg(z)|+\eta_0^2\boldsymbol{\varepsilon}.$$
	Since \(\Vert{\rm diag}(\boldsymbol{\mfc})\bbS_d\Vert=C_{d,\boldsymbol{\mfc}}\), for sufficiently small \(\eta_0\) such that \(\eta_0^2\Vert{\rm diag}(\boldsymbol{\mfc})\bbS_d\Vert\ll1\), so \((\boldsymbol{I}_d-\eta_0^2{\rm diag}(\boldsymbol{\mfc})\bbS_d)^{-1}\) exists and 
	$$\Vert(\boldsymbol{I}_d-\eta_0^2{\rm diag}(\boldsymbol{\mfc})\bbS_d)^{-1}\Vert\leq(1-\eta_0^2\Vert{\rm diag}(\boldsymbol{\mfc})\bbS_d\Vert)^{-1}<2,$$
	Consequently, we have
	$$\Vert\bbv(z)-\bbg(z)\Vert_{\infty}<\Vert\bbv(z)-\bbg(z)\Vert_2<2\eta_0^2\sqrt{d}\Vert\boldsymbol{\varepsilon}\Vert_{\infty}=C_d\eta_0^{-\beta+2}N^{-\alpha}.$$
	Next, for \(z\in\tilde{\mcS}_{\eta_0}^1\), define \(\boldsymbol{h}(z):=\bbv(z)-\boldsymbol{g}(z)\), we have 
	$$(\diag(e^{-2{\rm i}\bbq})-\bbF)\frac{\bbh}{|\bbg|}=\mfc^{-1}\circ\big(e^{-{\rm i}\boldsymbol{q}}\circ\boldsymbol{h}\circ\boldsymbol{S}_d\boldsymbol{h}+(|\boldsymbol{g}|+e^{-{\rm i}\boldsymbol{q}}\circ\boldsymbol{h})\circ\boldsymbol{\varepsilon}\big),$$
	where \(\bbF=\diag(\mfc^{-1}\circ|\bbg|)\bbS_d\diag(|\bbg|)\) is defined in (\ref{Eq of operator F}). In fact, notice that
	\begin{align*}
		&e^{-{\rm i}\bbq}\circ\bbh\circ\bbS_d\bbh=e^{-{\rm i}\bbq}\circ(\bbv-\bbg)\circ\bbS_d(\bbv-\bbg)\\
		&=e^{-{\rm i}\bbq}\circ\bbv\circ\bbS_d(\bbv-\bbg)+|\bbg|\circ\bbS_d\bbg-|\bbg|\circ\bbS_d\bbv\\
		&=-e^{-{\rm i}\bbq}\circ\bbv\circ\bbve-e^{-{\rm i}\bbq}\circ\mfc+e^{-{\rm i}\bbq}\circ\mfc\circ\frac{\bbv}{\bbg}+|\bbg|\circ\bbS_d\bbg-|\bbg|\circ\bbS_d\bbv\\
		&=\mfc\circ e^{-2{\rm i}\bbq}\circ\frac{\bbh}{|\bbg|}-\mfc\circ\bbF\frac{\bbh}{|\bbg|}-e^{-{\rm i}\bbq}\circ\bbv\circ\bbve\\
		&=\mfc\circ(\diag(e^{-2{\rm i}\bbq})-\bbF)\frac{\bbh}{|\bbg|}-e^{-{\rm i}\bbq}\circ\bbh\circ\bbve-|\bbg|\circ\bbve,
	\end{align*}
	where we use the fact that
	$$-\frac{\mfc}{\bbv}-z-\bbve=\bbS_d\bbv\quad{\rm and}\quad-\frac{\mfc}{\bbg}-z=\bbS_d\bbg$$
	in the second equality. For simplicity, denote 
	\begin{align}
		\mathbb{B}:=\diag(e^{-2{\rm i}\bbq})-\bbF,\label{Eq of tiB}
	\end{align}
	so we derive that
	\begin{align*}
		\bbh=|\bbg|\circ\mathbb{B}^{-1}\big(\mfc^{-1}\circ\big(e^{-{\rm i}\boldsymbol{q}}\circ\boldsymbol{h}\circ\boldsymbol{S}_d\boldsymbol{h}+(|\boldsymbol{g}|+e^{-{\rm i}\boldsymbol{q}}\circ\boldsymbol{h})\circ\boldsymbol{\varepsilon}\big)\big).
	\end{align*}
	According to Proposition \ref{Pro of stability operator}, we know that \(\mathbb{B}\) is invertible and \(\Vert\mathbb{B}^{-1}\Vert\leq C_{d,\mfc}\eta_0^{-2}\) by Proposition \ref{Pro of inverse norm}, then we have
	\begin{align*}
		\Vert\boldsymbol{h}\Vert_{\infty}&\leq\mfc_d^{-1}\eta_0^{-1}\sqrt{d}\Vert\mathbb{B}^{-1}\Vert\left(\sqrt{d}\Vert\boldsymbol{S}_d\Vert\Vert\boldsymbol{h}\Vert_{\infty}^2+(\eta_0^{-1}+\Vert\boldsymbol{h}\Vert_{\infty})\Vert\bbve\Vert_{\infty}\right)\\
		&\leq C_{d,\mfc}^{(1)}\eta_0^{-3}\Vert\boldsymbol{h}\Vert_{\infty}^2+C_{d,\mfc}^{(3)}\eta_0^{-(3+\beta)}N^{-\alpha}\Vert\boldsymbol{h}\Vert_{\infty}+C_{d,\mfc}^{(2)}\eta_0^{-(4+\beta)}N^{-\alpha},
	\end{align*}
	where \(C_{d,\mfc}^{(l)}\) are three constants depending on \(d\) and \(\mfc\) for \(l=1,2,3\). Hence, we obtain that \(\Vert\boldsymbol{h}(z)\Vert_{\infty}\leq x^-\) or \(\Vert\boldsymbol{h}(z)\Vert_{\infty}\geq x^+\) for all \(z\in\tilde{\mcS}_{\eta_0}^1\), where
	\begin{align}
		x^{\pm}:=\frac{\eta_0^3-C_{d,\mfc}^{(3)}\eta_0^{-\beta}N^{-\alpha}\pm\sqrt{\big(\eta_0^3-C_{d,\mfc}^{(3)}\eta_0^{-\beta}N^{-\alpha}\big)^2-4C_{d,\mfc}^{(1)}C_{d,\mfc}^{(2)}\eta_0^{-(1+\beta)}N^{-\alpha}}}{2C_{d,\mfc}^{(1)}}.\notag
	\end{align}
	For any fixed \(\eta_0\), as \(N\to\infty\), it implies that \(\eta_0^{-\beta}N^{-\alpha},\eta_0^{-(\beta+1)}N^{-\alpha}=\mro(\eta_0^3)\) and
	$$x^+=\mrO(\eta_0^3).$$
	On the other hand, since
	\begin{align*}
		x^-=\frac{2C_{d,\mfc}^{(2)}\eta_0^{-(1+\beta)}N^{-\alpha}}{\eta_0^3-C_{d,\mfc}^{(3)}\eta_0^{-\beta}N^{-\alpha}+\sqrt{\big(\eta_0^3-C_{d,\mfc}^{(3)}\eta_0^{-\beta}N^{-\alpha}\big)^2-4C_{d,\mfc}^{(1)}C_{d,\mfc}^{(2)}\eta_0^{-(1+\beta)}N^{-\alpha}}},
	\end{align*}
	where the denominator has the same order as \(x^+=\mrO(\eta_0^3)\). Hence, it implies that
	$$x^-=\mrO(\eta_0^{-(\beta+4)}N^{-\alpha}).$$
	Since \(\Vert\boldsymbol{h}(z)\Vert_{\infty}\) is continuous on \(z\in\tilde{\mcS}_{\eta_0}\) and we have shown that \(\Vert\boldsymbol{h}(z)\Vert_{\infty}=\mrO(\eta_0^{-\beta+2}N^{-\alpha})\) for all \(z\in\tilde{\mcS}_{\eta_0}^2\), it implies that \(\Vert\boldsymbol{h}(z)\Vert_{\infty}\leq x^-\) for all \(z\in\tilde{\mcS}_{\eta_0}^1\), so \(\Vert\boldsymbol{h}(z)\Vert_{\infty}\leq x^-=\mrO(\eta_0^{-(\beta+4)}N^{-\alpha})\), which completes our proof.
\end{proof}

\section{Properties of spectral distribution}\label{Sec of LSD}
In this section, we will derive several important properties of the empirical spectral distribution (ESD) and limiting spectral distribution (LSD) of \(\bbM\). Without loss of generality, we  assume, as before, that \(\mfc_1=\max_{1\leq l\leq d}\mfc_l\).
\subsection{Support of the empirical spectral distribution}\label{Sec of Stable region ESD}
\begin{thm}\label{Thm of Extreme eigenvalue N d=3}
	Under Assumptions {\rm \ref{Ap of general noise}} and {\rm \ref{Ap of dimension}}, define \(\mfv_d:=2(d-1)\sum_{l=1}^d\sqrt{\mfc_l}\), then for any \(t,l>0\), we have
	\begin{align}
		\mbP(\Vert\bbM\Vert>\mfv_d+t)=\mro(N^{-l}).\label{Eq of Extreme eigenvalues d=3}
	\end{align}
\end{thm}
As preliminaries, we state the following two results:
\begin{lem}[Chapter 9.12.5 of \cite{bai2010spectral}]\label{Thm of Extreme eigenvalue}
	Let \(\bbX=[X_{ij}]_{p\times n}\) be a random matrix of size \(p\times n\), whose entries \(\{X_{ij}\}\) are i.i.d. complex random variables with mean zero, variance one, and finite fourth moments, and \(|X_{ij}|\leq n^{1/4}\). If \(\frac{p}{n}\to y\in(0,1)\), then for any \(x>(1+\sqrt{y})^2\) and \(l>0\), the spectral norm of \(\bbS_n=n^{-1}\bbX\bbX^*\) satisfies that
	$$\mathbb{P}\left(\Vert\bbS_n\Vert>x\right)=\mro(n^{-l}).$$
\end{lem}
\begin{lem}[\cite{gotze2021concentration}]\label{Pro of Bernstein}
	Let $X_1,\cdots,X_n$ be independent random variables with \(\Vert X_i\Vert_{\Psi_{\theta}}\leq M\) for some \(\theta\in(0,1]\). Let \(\bba=(a_1,\cdots,a_n)'\in\mathbb{R}^n\), then
	$$\mathbb{P}\left(\Big|\sum_{i=1}^na_i\left(X_i-\mathbb{E}[X_i]\right)\Big|>t\right)\leq2\exp\left(-C_{\theta}\min\left\{\frac{t^2}{M\Vert\bba\Vert_2^2},\frac{t^{\alpha}}{M^{\alpha}\max_{1\leq i\leq n}|a_i|^{\alpha}}\right\}\right),$$
	where \(\Vert\cdot\Vert_{\Psi_{\theta}}\) is the Orlicz norms with parameter \(\theta\).
\end{lem}
Now, let us  prove that
\begin{proof}[Proof of Theorem \ref{Thm of Extreme eigenvalue N d=3}]
	Since \(\Vert\bbM\Vert=\sup_{\Vert\bbx\Vert_2=1}\Vert\bbM\bbx\Vert_2\), let us  split \(\bbx\) into blocks, i.e. \(\bbx=(\bbx_1,\bbx_2,\bbx_3)'\) and \(\bbM\bbx=\sum_{i=1}^d\sum_{j\neq i}^d\bbM^{ji}\bbx_i\). By the Minkowski inequality, we have
	$$\Vert\bbM\bbx\Vert_2\leq\sum_{i=1}^d\sum_{j\neq i}^d\Vert\bbM^{ji}\bbx_i\Vert_2\quad\Longrightarrow\quad\Vert\bbM\Vert\leq\sum_{i=1}^d\sum_{j\neq i}^d\Vert\bbM^{ji}\Vert=\sum_{i\neq j}^d\Vert\bbM^{ij}\Vert.$$
	Next, we will show that \(\mbP(\Vert\bbM^{ij}\Vert>\sqrt{\mfc_i}+\sqrt{\mfc_j})=\mro(N^{-l})\) for any \(l>0\). Since the proof is the same for each block, we consider \(i=1,j=2\), we have
	$$M_{j_1j_2}^{12}=\frac{1}{\sqrt{N}}\sum_{i_3\cdots i_d}^{n_3\cdots n_d}X_{j_1j_2i_3\cdots i_d}\mcA_{i_1\cdots i_d}^{(1,2)},$$
    where \(\mcA_{i_1\cdots i_d}^{(1,2)}\) is defined in (\ref{Eq of mcA}). By Lemma \ref{Pro of Bernstein} and Assumption \ref{Ap of general noise}, we have
	\begin{align*}
		\mbP\Bigg(\Bigg|\sum_{i_3\cdots i_d}^{n_3\cdots n_d}X_{j_1j_2i_3\cdots i_d}\mcA_{i_1\cdots i_d}^{(1,2)}\Bigg|>N^{1/4}\Bigg)\leq2\exp(-C_{\theta}N^{1/4}),
	\end{align*}
	then 
	\begin{align*}
		&\mbP\Bigg(\Bigg|\sum_{i_3\cdots i_d}^{n_3\cdots n_d}X_{j_1j_2i_3\cdots i_d}\mcA_{i_1\cdots i_d}^{(1,2)}\Bigg|>N^{1/4}\Bigg|\exists\ j_1,j_2\Bigg)\leq\sum_{j_1,j_2}^{n_1,n_2}\mbP\Bigg(\Bigg|\sum_{i_3\cdots i_d}^{n_3\cdots n_d}X_{j_1j_2i_3\cdots i_d}\mcA_{i_1\cdots i_d}^{(1,2)}\Bigg|>N^{1/4}\Bigg)\\
		&\leq2n_1n_2\exp(-C_{\theta}N^{1/4})=\mro(N^{-l}).
	\end{align*}
	By Lemma \ref{Thm of Extreme eigenvalue}, for any $t>0$, we can show that
	\begin{align*}
		&\mbP\big(\Vert\bbM^{12}\Vert>\sqrt{\mfc_1}+\sqrt{\mfc_2}+t\big)<\mbP\big(\Vert\bbM^{12}\Vert>\sqrt{\mfc_1}+\sqrt{\mfc_2}+t\big|\forall\ |M_{j_1j_2}^{12}|<N^{-1/4}\big)\mbP(\forall\ |M_{j_1j_2}^{12}|<N^{-1/4})\\
		&+\mbP\big(\exists\ |M_{j_1j_2}^{12}|<N^{-1/4}\big)=\mro(N^{-l}).
	\end{align*}
	For other \(\bbM^{ij}\), we also have \(\mbP\big(\Vert\bbM^{ij}\Vert>\sqrt{\mfc_i}+\sqrt{\mfc_j}+t\big)=\mro(N^{-l})\), then we can conclude that
    \begin{align*}
        \mbP(\Vert\bbM\Vert\geq\mfv_d+d^2t)\leq\sum_{i\neq j}^d\mbP(\Vert\bbM^{ij}\Vert\geq\sqrt{\mfc_i}+\sqrt{\mfc_j}+t)=\mro(N^{-l}),
    \end{align*}
    note that $d$ is a fixed integer, then we can conclude (\ref{Eq of Extreme eigenvalues d=3}).
\end{proof}
\subsection{Support of the limit spectral distribution}\label{Sec of finite support}
\begin{thm}\label{Lem of finite support}
	Let \(\bbg(z)\) be the solution of {\rm (\ref{Eq of MDE 3 order})} on \(\mbC^+\) and \(g(z)=\sum_{i=1}^d g_i(z)\), then there exists a probability measure \(\nu\) with bounded support such that its Stieltjes transform is \(g(z)\).
\end{thm}
\begin{proof}
	Define
	\begin{align}
		g(z)=\sum_{j=1}^dg_j(z)\quad{\rm and}\quad\nu(E):=\lim_{\eta\downarrow0^+}\pi^{-1}g(E+{\rm i}\eta),\label{Eq of g(z)}
	\end{align}
	where \(g_j(z)\) is the \(j\)-th entry of \(\bbg(z)\) defined in (\ref{Eq of MDE 3 order}) and \(E\in\mbR\). We will show that \(g(z)\) is the Stieltjes transform of the probability measure \(\nu\), which has finite support. For a holomorphic function \(\phi:\mathbb{C}^+\to\mathbb{C}^+\) on the complex upper half plane, it is a Stieltjes transform of a measure \(\nu_{\phi}\) on the real line such that \(\nu_{\phi}(\mathbb{R})=\alpha>0\) if and only if \(|{\rm i}\eta\phi({\rm i}\eta)+\alpha|\to0\) as \(\eta\to\infty\), see Theorem B.9 in \cite{bai2010spectral}. Besides, \(\nu_{\phi}\) can be explicitly determined by \(\nu_{\phi}(x)=\pi^{-1}\Im(\phi(x))\) for \(x\in\mathbb{R}\). Therefore, let \(\bbg(z)=(g_1(z),\cdots,g_d(z))'\) be the unique analytic solution (\ref{Eq of MDE 3 order}) of on \(\mbC^+\), since \(\mfc_j+zg_j(z)=-g_j(z)\sum_{k\neq j}g_k(z)\), it implies that \(|\mfc_j+{\rm i}\eta g_j({\rm i}\eta)|\leq(d-1)\eta^{-2}\to0\) as \(\eta\to\infty\), i.e. the induced measure of \(g_j(z)\) has mass \(\mfc_j\). Hence, \(g(z)\) is a Stieltjes transform of a probability measure due to \(\sum_{i=1}^d\mfc_i=1\). Next, suppose \(\Vert\bbg(z)\Vert_{\infty}\leq\frac{|z|}{2(d-1)}\), then
	$$|g_i(z)|=\frac{\mfc_i}{|z+\sum_{j\neq i}g_j(z)|}\leq\frac{\mfc_i}{|z|-\sum_{j\neq i}|g_j(z)|}\leq\frac{2\mfc_j}{|z|}\Longrightarrow\Vert\bbg(z)\Vert_{\infty}\leq\min\{|z|/(2(d-1)),2/|z|\}.$$
	Let \(H_0:=\{z\in\mathbb{C}^+:|z|>2\sqrt{d-1}(1+\epsilon)\}\), where \(\epsilon>0\), then \(z\to\sqrt{d-1}\Vert\bbg(z)\Vert_{\infty}\) maps \(H_0\) into two disjoint regions, i.e. 
	$$z\to\sqrt{d-1}\Vert\bbg(z)\Vert_{\infty}:H_0\to[0,(1+\epsilon)^{-1}]\cup[1+\epsilon,+\infty].$$
	In fact, if we divide \(H_0\) into two disjoint parts, \(H_{0,1}:=\{z\in H_0:\Vert\bbg(z)\Vert_{\infty}\leq |z|/(2(d-1))\}\) and \(H_{0,2}:=H_0\backslash H_{0,1}\), then \(\sqrt{d-1}\Vert\bbg(z)\Vert_{\infty}>1+\epsilon\) for \(z\in H_{0,2}\) and \(\sqrt{d-1}\Vert\bbg(z)\Vert_{\infty}<(1+\epsilon)^{-1}\) for \(z\in H_{0,1}\). Since \(z\to\Vert\bbg(z)\Vert_{\infty}\) is continuous for \(z\in H_0\) and \(\Vert\bbg(z)\Vert_{\infty}\leq\Im(z)^{-1}\to0\) when \(\Im(z)\) is sufficiently large, then \(z\to\sqrt{d-1}\Vert\bbg(z)\Vert_{\infty}\) maps \(H_0\) to \([0,(1+\epsilon)^{-1}]\), i.e. \(\Vert\bbg(z)\Vert_{\infty}\leq2|z|^{-1}\) for \(z\in H_0\). Now, based on the (\ref{Eq of VDE image}), we have
	$$\Im(g_i(z))=\Im(z)\mfc_i^{-1}|g_i(z)|^2+\mfc_i^{-1}|g_i(z)|^2\sum_{j\neq i}\Im(g_j(z))\ \Leftrightarrow\ \Im(g_i(z))=\frac{\mfc_i^{-1}|g_i(z)|^2}{1+\mfc_i^{-1}|g_i(z)|^2}\Im(z+g(z)),$$
	which implies that
	$$\Im(g_i(z))\leq\frac{4\mfc_i}{|z|^2+4\mfc_i}\Im(z+g(z))<\frac{4\mfc_i}{|z|^2}\Im(z+g(z)),\quad{\rm for\ }\forall z\in H_0,$$
	where the first inequality is valid due to \(|g_i(z)|\leq\frac{2\mfc_i}{|z|}\). As a result, by summing all \(\Im(g_i(z))\), we have
	$$\Im(g(z))<\frac{4}{|z|^2}\Im(z+g(z))\quad\Rightarrow\quad\Im(g(z))<\frac{4\Im(z)}{|z|^2-4},$$
	where we use the fact that \(|z|>2+\epsilon\) for \(z\in H_0\), then
	$$\Im(g(z))<\frac{\Im(z)}{\epsilon(2+\epsilon)}\to0\quad{\rm as\ }\Im(z)\to0,$$
	Hence the induced measure \(\nu(E)\) of \(g(z)\) has  mass $0$ for all \(E>2\sqrt{d-1}\), i.e. \(\nu\) has finite support.
\end{proof}
Here, let
\begin{align}
	\zeta:=\inf\left\{E>0:\lim_{\eta\to0}\Im(g(E+{\rm i}\eta))=0,\eta>0\right\},\label{Eq of support boundary}
\end{align}
be the right boundary of the support of \(\nu\). Due to the symmetry of \(\nu\), the absolute value of the left boundary is the same as the right one. Now, combining (\ref{Eq of Extreme eigenvalues d=3}) and (\ref{Eq of support boundary}), we give the following stable region for the spectral distribution of \(\bbM\):
\begin{align}
	\mcS_{\eta_0}:=\{z\in\mbC^+:{\rm dist}(z,[-\max(\mfv_3,\zeta),\max(\mfv_3,\zeta)])\geq\eta_0,|\Re(z)|,|\Im(z)|\leq\eta_0^{-1}\},\label{Eq of stability region}
\end{align}
where \(\eta_0\) is sufficiently small such that \(\eta_0^{-1}>\max(\mfv_3,\zeta)\). Consequently, since \(\mbP(\Vert\bbM\Vert\leq\mfv_3)\geq1-\mro(N^{-l})\) for any \(l>0\), then
\begin{align}
	\mbP(\Vert\bbQ(z)\Vert\leq{\rm dist}(z,[-\max(\mfv_d,\zeta),\max(\mfv_d,\zeta)])^{-1}\leq\eta_0^{-1})\geq1-\mro(N^{-l}).\label{Eq of spectral norm mcE d=3}
\end{align}
Thus, without further specifications, we assume \(\Vert\bbQ(z)\Vert\leq\eta_0^{-1}\) for any \(z\in\mcS_{\eta_0}\) in the following contexts.
\subsection{Singularity of the limiting spectral distribution}\label{Sec of Singularity}
Recall that we assume \(\mfc_1=\max_{1\leq l\leq d}\mfc_l\), we will prove that
\begin{thm}\label{Thm of Singularity}
	\(\nu\) defined in {\rm (\ref{Eq of g(z)})} has a point mass at \(0\) (and this is its only point mass) if and only if \(\mfc_1\geq1/2\).
\end{thm}
By \eqref{Eq of g(z)}, we know that \(\nu\) has a point mass at \(E\in\mbR\) if and only if \(\lim_{\eta\to0^+}\Im(g(E+{\rm i}\eta))=\infty\). Thus, let us  show that 
\begin{lem}
	For any \(z=E+{\rm i}\eta\in\mbC^+\) and \(E\neq0\), then
	$$\lim_{\eta\to0^+}\Im(g(E+{\rm i}\eta))<\infty,$$
	where \(g(z)=\sum_{j=1}^dg_j(z)\) defined in {\rm (\ref{Eq of g(z)})}.
\end{lem}
\begin{proof}
	Taking the imaginary and real parts of (\ref{Eq of MDE 3 order}) respectively, we have
	\begin{align}
		\frac{\mfc\circ\Im(\bbg(z))}{|\bbg|^2}=\eta+\bbS_d\Im(\bbg(z))\quad{\rm and}\quad-\frac{\mfc\circ\Re(\bbg(z))}{|\bbg(z)|^2}=E+\bbS_d\Re(\bbg(z)).\label{Eq of Real part}
	\end{align}
	Suppose there exists an \(E\neq 0\) such that \(\lim_{\eta\to0^+}\Im(g(E+{\rm i}\eta))=\infty\), without loss of generality, assume \(\lim_{\eta\to0^+}\Im(g_1(E+{\rm i}\eta))=\infty\), which implies that \(\lim_{\eta\to0^+}|g_1(E+{\rm i}\eta)|=\infty\). By the second equation in (\ref{Eq of Real part}), we have
	$$|g_1(E+{\rm i}\eta)|=\frac{-\mfc_1\cos\theta_1}{E+\sum_{i\neq1}^d\Re(g_i(E+{\rm i}\eta))},$$
	where \(\cos\theta_j:=\Re(g_j(z))/|g_j(z)|\) and \(\sin\theta_j:=\Im(g_j(z))/|g_j(z)|\), so we obtain \(\lim_{\eta\to0^+}\sum_{i\neq1}^d\Re(g_i(E+{\rm i}\eta))=-E\). On the other hand, notice that for \(j\neq1\)
	$$|g_j(E+{\rm i}\eta)|=\frac{\mfc_j\sin\theta_j}{\eta+\sum_{i\neq j}^d\Im(g_i(E+{\rm i}\eta))}<\frac{1}{\Im(g_1(E+{\rm i}\eta))},$$
	it yields that \(\lim_{\eta\to0^+}|g_j(E+{\rm i}\eta)|=0\) and \(\lim_{\eta\to0^+}\Re(g_j(E+{\rm i}\eta))=0\). But it is a contradiction since
	$$\lim_{\eta\to0^+}\sum_{i\neq1}^d\Re(g_i(z))=-E\neq0,$$
	which proves our claim.
\end{proof}
By the above lemma, we only need to focus on the limiting behaviors of \(g({\rm i}\eta)\) as \(\eta\to0^+\). If we replace \(z\) by \(-\bar{z}\) in (\ref{Eq of MDE 3 order}) and take the imaginary part on the both sides of (\ref{Eq of MDE 3 order}) respectively, i.e.
$$-\frac{\mfc}{\bbg(-\bar{z})}=-\bar{z}+\bbS_d\bbg(-\bar{z})\quad\Longleftrightarrow\quad-\frac{\mfc}{-\overline{\bbg(z)}}=(-\bar{z})+\bbS_d(-\overline{\bbg(z)})$$
it implies that \(\boldsymbol{g}(-\bar{z})=-\overline{\boldsymbol{g}(z)}\) by Theorem \ref{Thm of Dyson}. Hence, \(\Re(g_i({\rm i\eta}))=0\) for \(\eta>0,i=1,\cdots,d\) and denote \(g_i({\rm i}\eta):={\rm i}\tilde{g}_i(\eta)\), then (\ref{Eq of MDE 3 order}) can be rewritten as
\begin{align}
	\frac{\boldsymbol{\mathfrak{c}}}{\tilde{\boldsymbol{g}}(\eta)}=\eta+\boldsymbol{S}_d\tilde{\boldsymbol{g}}(\eta).\label{Eq of MDE image part}
\end{align}
Before giving the proof of Theorem \ref{Thm of Singularity}, we need the following results:
\begin{lem}\label{Lem of gz}
	Let \(c>0,z=E+{\rm i}\eta\in\mathbb{C}^+,E\geq0\) and \(x_{1,2}=r_{1,2}\exp({\rm i}\theta_{1,2})\) be the solutions of \(x^2-zx-c=0\), where \(\theta_{1,2}\in[0,2\pi)\). Without loss of generality, let \(r_1\geq r_2\), then we have 
	\begin{itemize}
		\item \(r_2\leq\sqrt{c}\leq r_1\) and \(\theta_1,\theta_2\in[0,\pi],\theta_1+\theta_2=\pi\).
		\item \({\rm sign}(\Re(x_1))=-{\rm sign}(\Re(x_2))\) and \(\Im(x_1),\Im(x_2)>0\).
		\item Let \(\theta(E,\eta,c):=\min\{\theta_1,\theta_2\}\) and \(r_1:=r_1(E,\eta,c)\), then \(\partial_E \theta<0,\partial_{\eta} \theta>0\) and \(\partial_E r_1,\partial_{\eta} r_1>0\).
	\end{itemize}
\end{lem}
\begin{proof}
	 Since \(x_1x_2=-c\), it implies that \(r_2\leq\sqrt{c}\leq r_1\) and \(\theta_1+\theta_2=\pi\ {\rm or}\ 3\pi\). Notice that \(x_1+x_2=E+{\rm i}\eta\) and \(E\geq0,\eta>0\), we further conclude that \(\theta_1+\theta_2=\pi\) and
	$$(r_1-r_2)\cos\theta=E,\quad(r_1+r_2)\sin\theta=\eta,$$
	where \(\theta:=\min\{\theta_1,\theta_2\}=\theta_1\). Then we can solve that
	\begin{align}
		&\sin^2\theta=\frac{E^2+\eta^2+4c-\sqrt{(E^2+\eta^2+4c)^2-16c\eta^2}}{8c}:=s(E,\eta,c),\notag\\
		&\partial_E s(E,\eta,c)=\frac{E}{4c}\left(1-\frac{E^2+\eta^2+4c}{\sqrt{(E^2+\eta^2+4c)^2-16c\eta^2}}\right)<0,\notag\\
		&\partial_{\eta}s(E,\eta,c)=\frac{\eta}{4c}\left(1-\frac{E^2+\eta^2-4c}{\sqrt{(E^2+\eta^2-4c)^2+16cE^2}}\right)>0.\notag
	\end{align}
	Notice that \(\theta=\min\{\theta_1,\theta_2\}\in[0,\pi/2]\) and the monotonicity of \(\sin^2\theta\) and \(\theta\) are the same when \(\theta\in[0,\pi/2]\), which implies that \(\partial_E \theta<0,\partial_{\eta} \theta>0\). Besides, we also have
	$$r_1=\frac{1}{2}\left(\frac{E}{\cos\theta}+\frac{\eta}{\sin\theta}\right)=\frac{1}{2}\left(\frac{E}{\sqrt{1-s}}+\frac{\eta}{\sqrt{s}}\right)=\frac{1}{2\sqrt{2}}\left(\sqrt{M+\sqrt{M^2-16c\eta^2}}+\sqrt{N+\sqrt{N^2+16cE^2}}\right),$$
	where \(M:=E^2+\eta^2+4c,N:=E^2+\eta^2-4c\) and \(M^2-16c\eta^2=N^2+16cE^2\). Therefore, it implies that
	\begin{align}
		\partial_E r_1&=\frac{E}{2\sqrt{2}}\left(1+\frac{M}{\sqrt{M^2-16c\eta^2}}\right)\left[\left(M+\sqrt{M^2-16c\eta^2}\right)^{-1/2}+\left(N+\sqrt{N^2+16cE^2}\right)^{-1/2}\right]>0,\notag\\
		\partial_{\eta} r_1&=\frac{\eta}{2\sqrt{2}}\left(1+\frac{N}{\sqrt{N^2+16cE^2}}\right)\left[\left(M+\sqrt{M^2-16c\eta^2}\right)^{-1/2}+\left(N+\sqrt{N^2+16cE^2}\right)^{-1/2}\right]>0,\notag
	\end{align}
	which completes our proof.
\end{proof}
As a consequence of above lemma, we have that
\begin{lem}\label{Cor of max mfc}
	Recall that \(\mfc_1=\max_{1\leq i\leq d}\mfc_i\), when \(\eta\geq\max_{1\leq i\leq d}\mfc_i^{-1/2}\), we have 
	\begin{align*}
		\mfc_1^{-1/2}|g_1({\rm i}\eta)|=\arg\max_{1\leq i\leq d}\mathfrak{c}_i^{-1/2}|g_i({\rm i}\eta)|
	\end{align*}
	and
	\begin{align*}
		\max_{1\leq i\leq d}|g_i({\rm i}\eta)|\leq\frac{d}{2(d-1)}\sqrt{\eta^2+4(d-1)\mfc_1}-\frac{\eta}{d-1}.
	\end{align*}
\end{lem}
\begin{proof}
	By the symmetry \(\bbg(-\bar{z})=-\overline{\bbg(z)}\), it suffices to consider \(\Re(z)\geq0\) without loss of generality. Since \(g_i(z)\) is the solution of \(\mathfrak{c}_i+g_i(z)(z+g(z)-g_i(z))=0\) by (\ref{Eq of MDE 3 order}) and \(|g_i(z)|\leq\Im(z)^{-1}\leq\min_{1\leq i\leq d}\mfc_i^{1/2}\), we conclude that \(g_i(z)=r_2(z)\exp({\rm i}\theta_2(z))\) when \(\Im(z)\geq\max_{1\leq i\leq d}\mfc_i^{-1/2}\) by Lemma \ref{Lem of gz}. Let 
	\begin{align}
		h(z):=z+g(z)\quad{\rm and}\quad\mathfrak{g}_i(z):=g_i(z)/\sqrt{\mathfrak{c}_i},\label{Eq of h(z)}
	\end{align}
	then
	$$\mathfrak{g}_i^2(z)-\mathfrak{c}_i^{-1/2}h(z)\mathfrak{g}_i(z)-1=0\ \ {\rm and\ \ }|\mathfrak{g}_i(z)|^{-1}=r_1(\Re(h(z))/\sqrt{\mathfrak{c}_i},\Im(h(z))/\sqrt{\mathfrak{c}_i},1).$$
	Since \(r_1\) is increasing function of the real and imaginary part of \(h(z)/\sqrt{\mathfrak{c}_i}\) by Lemma \ref{Lem of gz}, then
	$$\arg\max_{1\leq i\leq d}\mathfrak{c}_i^{-1/2}|g_i(z)|=\arg\min_{1\leq i\leq d}|\mathfrak{g}_i(z)|^{-1}=\arg\max_{1\leq i\leq d}\mathfrak{c}_i,$$
	which implies that
	\begin{align*}
		\mfc_1^{-1/2}|g_1(z)|=\arg\max_{1\leq i\leq d}\mathfrak{c}_i^{-1/2}|g_i(z)|.
	\end{align*}
	Next, since \(h(z)=z+g(z)\) and \(z={\rm i}\eta\in\mbC^+\), then \(\Re(h({\rm i}\eta))=0\). Hence, let
	\begin{align}
	    h({\rm i}\eta)={\rm i}\mathfrak{h}(\eta),\quad{\rm where\ \ }\mathfrak{h}(\eta)>0,\label{Eq of mfh}
	\end{align}
	and the two solutions of \(\mathfrak{c}_i+g_i(z)(z+g(z)-g_i(z))=0\) are denoted by
	\begin{align}
	    g_i^{\pm}({\rm i}\eta)={\rm i}\tilde{g}_i^{\pm}(\eta):=\frac{h({\rm i}\eta)\pm\sqrt{h^2({\rm i}\eta)+4\mathfrak{c}_i}}{2}=\frac{{\rm i}\mfh(\eta)\pm\sqrt{4\mfc_i-\mfh^2(\eta)}}{2},\label{Eq of gi pm}
	\end{align}
	where \(|\tilde{g}_i^+(\eta)|\geq\sqrt{\mfc_i}\geq|\tilde{g}_i^-(\eta)|\) and \(|\tilde{g}_i^+(\eta)\tilde{g}_i^-(\eta)|=\mfc_i\) by Lemma \ref{Lem of gz}. Recall that \(\Re(g_i({\rm i}\eta))=0\), it implies that \(\mathfrak{h}(\eta)\geq2\max_{1\leq i\leq d}\sqrt{\mathfrak{c}_i}=2\sqrt{\mfc_1}\) whatever \(g_i({\rm i}\eta)={\rm i}\tilde{g}_i^+(\eta)\) or \(g_i({\rm i}\eta)={\rm i}\tilde{g}_i^-(\eta)\). Furthermore, when \(\eta\geq\max_{1\leq i\leq d}\mathfrak{c}_i^{-1/2}\), \(|g_i({\rm i}\eta)|\leq\eta^{-1}\leq\sqrt{\mathfrak{c}_i}\), hence
	\begin{align}
		g_i({\rm i}\eta)={\rm i}\tilde{g}_i^-(\eta)={\rm i}\frac{\mathfrak{h}(\eta)-\sqrt{\mathfrak{h}^2(\eta)-4\mathfrak{c}_i}}{2},\ {\rm for\ }i=1,\cdots,d,\label{Eq of image gi}
	\end{align}
	by Lemma \ref{Lem of gz}. Summing above equations for \(i=1,\cdots,d\), it has
	\begin{align}
		&2{\rm i}(\mfh(\eta)-\eta)=2\sum_{i=1}^dg_i({\rm i}\eta)={\rm i}\Big(d\mfh(\eta)-\sum_{i=1}^d\sqrt{\mathfrak{h}^2(\eta)-4\mathfrak{c}_i}\Big)\notag\\
		\Longrightarrow&2\eta+(d-2)\mathfrak{h}(\eta)=\sum_{i=1}^d\sqrt{\mathfrak{h}^2(\eta)-4\mathfrak{c}_i}>d\sqrt{\mathfrak{h}^2(\eta)-4\max_{1\leq i\leq d}\mathfrak{c}_i},\label{Eq of h image}
	\end{align}
	which implies that
	$$2\max_{1\leq i\leq d}\sqrt{\mathfrak{c}_i}\leq\mathfrak{h}(\eta)<\frac{d}{2(d-1)}\sqrt{\eta^2+4(d-1)\max_{1\leq i\leq d}\mathfrak{c}_i}+\frac{\eta(d-2)}{d-1}.$$
	Finally, since \(|g_i({\rm i}\eta)|=\tilde{g}_i(\eta)\) and \(\mfh(\eta)=\eta+\sum_{i=1}^d\tilde{g}_i(\eta)\), it implies that
	$$|g_i({\rm i}\eta)|\leq\mfh(\eta)-\eta<\frac{d}{2(d-1)}\sqrt{\eta^2+4(d-1)\mfc_1}-\frac{\eta}{d-1},$$
	which completes our proof.
\end{proof}
Based on proofs of above lemma, if \(g_i({\rm i}\eta)=g_i^-({\rm i}\eta)\) for all \(\eta>0\) and \(1\leq i\leq d\), we have 
$$\max_{1\leq i\leq d}|g_i({\rm i}\eta)|<\min\left\{\frac{d}{2(d-1)}\sqrt{\eta^2+4(d-1)\mfc_1}-\frac{\eta}{d-1},\eta^{-1}\right\},$$
In other words, \(g(z)=\sum_{i=1}^dg_i(z)\) will not have a singularity at \(z=0\). However, \(g_i({\rm i}\eta)=g_i^-({\rm i}\eta)\) will not hold under some certain conditions. Here, let \(\Pi_0:=\{\mfc\in\mbR^d:\sum_{l=1}^d\mfc_l=1,\mfc_l\geq0\}\) be the \(d\)-dimensional affine hyperplane, then we define the \emph{invariant branch region} as follows:
\begin{align}
    \Pi_1:=\Big\{\mfc\in\Pi_0:\sum_{i=1}^d\sqrt{\mfc_1-\mfc_i}\leq(d-2)\sqrt{\mfc_1}\Big\},\label{Eq of invariant branch region}
\end{align}
then we provide the following result.
\begin{pro}
    For any \(\mfc\in\Pi_0\backslash\Pi_1\) in \eqref{Eq of invariant branch region}, let \(\eta_1:=\sum_{i=1}^d\sqrt{\mathfrak{c}_1-\mathfrak{c}_i}-(d-2)\sqrt{\mathfrak{c}_1}>0\), then \label{Pro of solution branch}
	$$g_j({\rm i}\eta)=\left\{\begin{array}{ll}
		{\rm i}\tilde{g}_j^-(\eta)&\eta\in(\eta_1,+\infty),1\leq j\leq d,\\
        {\rm i}\tilde{g}_j^-(\eta)&\eta\in(0,\eta_1],2\leq j\leq d,\\
        {\rm i}\tilde{g}_j^+(\eta)&\eta\in(0,\eta_1],j=1,
	\end{array}\right.$$
    where \(g_j^{\pm}(\eta)\) are defined in {\rm (\ref{Eq of gi pm})}. On the other hand, for any \(\mfc\in\Pi_1\), \(g_j({\rm i}\eta)={\rm i}\tilde{g}_j^-(\eta)\) for any \(\eta>0\) and \(j=1,\cdots,d\).
\end{pro}
\begin{proof}
	First, taking the derivative of \(\eta\) in (\ref{Eq of MDE image part}), we obtain
	\begin{align}
		(\boldsymbol{I}_d+{\rm diag}(\boldsymbol{\mathfrak{c}}^{-1}\circ\tilde{\boldsymbol{g}}^{\circ2})\boldsymbol{S}_d)\tilde{\boldsymbol{g}}'=-\boldsymbol{\mathfrak{c}}^{-1}\circ\tilde{\boldsymbol{g}}^{\circ2}.\label{Eq of imag g derivative}
	\end{align}
	In fact, since
	$$\frac{\mfc_i}{\tilde{g}_i}=\eta+\sum_{j\neq i}\tilde{g}_j\ \ \Rightarrow\ \ -\tilde{g}_i'\frac{\mfc_i}{\tilde{g}_i^2}=1+\sum_{j\neq i}\tilde{g}_j'\ \ \Leftrightarrow\ \ \tilde{g}_i'+\mfc_i^{-1}\tilde{g}_i^2\sum_{j\neq i}\tilde{g}_j'=-\mfc_i^{-1}\tilde{g}_i^2,$$
	then we obtain (\ref{Eq of imag g derivative}) and 
	\begin{align}
	    \tilde{g}_i'+\mfc_i^{-1}\tilde{g}_i^2\sum_{j\neq i}\tilde{g}_j'=-\mfc_i^{-1}\tilde{g}_i^2\ \ \Rightarrow\ \ (\mfc_i/\tilde{g}_i^2)\tilde{g}_i'+(\tilde{g}'-\tilde{g}_i')=-1\ \ \Rightarrow\ \ (1-\mathfrak{c}_i/\tilde{g}_i^2(\eta))\tilde{g}_i'(\eta)=\mfh'(\eta),\label{Eq of derivative relation}
	\end{align}
    where \(\mfh(\eta)=\eta+\tilde{g}(\eta)\) in \eqref{Eq of mfh}. Here, we claim that \(\mfh'(\eta)>0\) when \(\eta>\max_{1\leq i\leq d}\mfc_i^{-1/2}\). Actually, if \(\eta>\max_{1\leq i\leq d}\mfc_i^{-1/2}\), \(\tilde{g}_i(\eta)<\sqrt{\mathfrak{c}_i}\) for all \(i=1,\cdots,d\). Suppose \(\mfh(\eta)\leq 0\), due to \(\tilde{g}_i(\eta)>\sqrt{\mathfrak{c}_i}\), we know that \(1-\mathfrak{c}_i/\tilde{g}_i^2(\eta)<0\), so it implies that \(\tilde{g}_i'(\eta)\geq0\) by \eqref{Eq of derivative relation} and \(\mfh'(\eta)=1+\sum_{i=1}^d\tilde{g}_i'(\eta)>0\), which is a contradiction. Next, we consider two cases. 
    
    \vspace{3mm}
    \noindent
    {\bf Case 1: \(\mfh'(\eta)>0\) for all \(\eta>0\).} It is easy to see that 
    $$\tilde{g}(0)\leq\mfh(0)\leq\mfh\left(\max_{1\leq i\leq d}\mfc_i^{-1/2}\right)\leq\max_{1\leq i\leq d}\mfc_i^{-1/2}+\sum_{i=1}^d\sqrt{\mfc_i},$$
    where \(g(0)={\rm i}\tilde{g}(0)\) is defined in \eqref{Eq of MDE image part}, it implies that \(g(z)\) in \eqref{Eq of g(z)} is nonsingular at \(0\). Moreover, all \(\tilde{g}_i(\eta)=\tilde{g}_i^-(\eta)\) in \eqref{Eq of gi pm} for \(\eta>0\). Otherwise, suppose \(\tilde{g}_i(\eta)=\tilde{g}_i^+(\eta)\) for some $i$, by Lemma \ref{Lem of gz}, we know that \(\tilde{g}_i^+(\eta)\geq\sqrt{\mfc_i}\), since \(\tilde{g}_i(\eta)\) is continuous and \(\tilde{g}_i(\eta)<\sqrt{\mfc_i}\) when \(\eta>\max_{1\leq i\leq d}\mfc_i^{-1/2}\), then there exists a \(\eta_0>0\) such that \(\tilde{g}_i^+(\eta_0)=\sqrt{\mfc_i}\), so \eqref{Eq of derivative relation} deduces that \(\mfh'(\eta_0)=0\), which is a contradiction.

    \vspace{3mm}
    \noindent
    {\bf Case 2: there exists an \(\eta_1\in(0,\max_{1\leq i\leq d}\mfc_i^{-1/2}]\) such that \(\mfh'(\eta)\leq0\).} Without loss of generality, let
    $$\eta_1:=\max\Big\{\eta\in\Big(0,\max_{1\leq i\leq d}\mfc_i^{-1/2}\Big]:\mfh'(\eta)=0\Big\},$$
    by Remark \ref{Rem of bounded solutions} and Lemma \ref{Lem of gz}, we know that \(\tilde{g}_i(\eta)<\sqrt{\mfc_i}\) and \(\tilde{g}_i(\eta)=\tilde{g}_i^-(\eta)\) in \eqref{Eq of gi pm} for \(i=2,\cdots,d\) and \(\eta>0\), combined with \eqref{Eq of derivative relation}, we obtain that
    \begin{align}
        \tilde{g}_i'(\eta_1)=0,\quad\text{for }i=2,\cdots,d,\quad\tilde{g}_1'(\eta_1)=\mfh'(\eta_1)-1-\sum_{i=2}^d\tilde{g}_i'(\eta_1)=-1,\label{Eq of tilde g_1 derivative}
    \end{align}
    which further implies that \(\tilde{g}_1(\eta_1)=\sqrt{\mfc_1}\) by \eqref{Eq of derivative relation} and \(\mfh(\eta_1)=\tilde{g}_1(\eta_1)+\mfc_1/\tilde{g}_1(\eta_1)=2\sqrt{\mfc_1}\). Let \(\mfh(\eta_1)=2\sqrt{\mfc_1}\) in \eqref{Eq of h image}, we deduce that
    $$\eta_1=\sum_{i=1}^d\sqrt{\mathfrak{c}_1-\mathfrak{c}_i}-(d-2)\sqrt{\mathfrak{c}_1}.$$
    Therefore, the case 2 is valid if and only if the above \(\eta_1>0\), i.e. \(\mfc\in\Pi_1\) in \eqref{Eq of invariant branch region}. Moreover, we claim that \(\mfh'(\eta)=0\) has a unique solution in \(\eta\in(0,\max_{1\leq i\leq d}\mfc_i^{-1/2}]\). Suppose there exists another \(\eta_2\in(0,\eta_1]\) such that \(\mfh'(\eta_2)=0\), we can still obtain that \(\tilde{g}_1(\eta_2)=\sqrt{\mfc_1}\) and \(\mfh(\eta_2)=2\sqrt{\mfc_1}\) by previous arguments, so \eqref{Eq of h image} further implies that 
    $$\eta_2=\sum_{i=1}^d\sqrt{\mathfrak{c}_1-\mathfrak{c}_i}-(d-2)\sqrt{\mathfrak{c}_1}=\eta_1,$$
    i.e. \(\eta_1\) is unique. Thus, suppose \(\eta_1>0\), \(\mfh'(\eta)\) is either positive or negative in \((0,\eta_1)\).
    Next, we will claim that
    \begin{align}
        \mfh'(\eta)<0\quad\text{for }\eta\in(0,\eta_1).\label{Eq of derivative claim}
    \end{align}
    In fact, recall that \(\tilde{g}_1'(\eta_1)=-1\) in \eqref{Eq of tilde g_1 derivative}, there exists an \(\epsilon>0\) such that \(\tilde{g}_1'(\eta)<0\) for \(\eta\in(\eta_1-\epsilon,\eta_1)\). Hence, \(\tilde{g}_1(\eta)>\sqrt{\mfc_1}\) in \((\eta_1-\epsilon,\eta_1)\) due to \(\tilde{g}_1(\eta_1)=\sqrt{\mfc_1}\). By \eqref{Eq of derivative relation}, we have
    $$\mfh'(\eta)=(1-\mfc_1/\tilde{g}_1^2(\eta))\tilde{g}_1'(\eta)<0\quad\text{for }\eta\in(\eta_1-\epsilon,\eta_1),$$
    which confirms our claim \eqref{Eq of derivative claim}. Finally, by Remark \ref{Rem of bounded solutions} and Lemma \ref{Lem of gz}, we have \(\tilde{g}_i(\eta)<\sqrt{\mfc_i}\) and \(\tilde{g}_i(\eta)=\tilde{g}_i^-(\eta)\) in \eqref{Eq of gi pm} for \(i=2,\cdots,d\) and \(\eta>0\), so \eqref{Eq of derivative relation} implies that \(g_i'({\rm i}\eta)>0\) for \(i=2,\cdots,d\), then
    $$\tilde{g}_1'(\eta)=\mfh'(\eta)-1-\sum_{i=2}^d\tilde{g}_i'(\eta)<0\quad\text{for }\eta\in(0,\eta_1),$$
    which further implies that \(\tilde{g}_1(\eta)>\sqrt{\mfc_1}\) in \((0,\eta_1)\). Thus, by Lemma \ref{Lem of gz}, we conclude that \(\tilde{g}_1(\eta)=\tilde{g}_1^+(\eta)\) in \((0,\eta_1)\).
\end{proof}
As a consequence of above proposition, we know that \(\tilde{g}_1(\eta)\) will have a \emph{branch change} at \(\eta_1\) if \(\mfc\in\Pi_0\backslash\Pi_1\). Under this situation, \(\tilde{g}_1(\eta)\) will increase as \(\eta\to0^+\). Therefore, to determine whether \(g(z)\) \eqref{Eq of g(z)} is singular or not at \(0\), it is enough to show that whether \(\lim_{\eta\downarrow0}\tilde{g}_1(\eta)\) is infinite or not. Here, let us  define a new function
$$F_d(x)=\sum_{i=2}^d\sqrt{x^2-4\mathfrak{c}_i}-\sqrt{x^2-4\mathfrak{c}_1}-(d-2)x,\ {\rm for\ }x\geq2\sqrt{\mathfrak{c}_1},$$
then we can provide a sufficient and necessary condition for the limiting behaviors of \(\tilde{g}_1(\eta)\) as \(\eta\downarrow0\) based on the number of solutions of \(F_d(x)=0\).
\begin{pro}\label{Pro of singularity}
	Suppose \(\eta_1\) in Proposition {\rm \ref{Pro of solution branch}} is strictly positive, then the solutions of \(F_d(x)=0\) on \([2\sqrt{\mathfrak{c}_1},\infty)\) are bounded \(\Longleftrightarrow\) \(\mathfrak{c}_1<0.5\).
\end{pro}
\begin{proof}
	Note that \(F_d(2\sqrt{\mathfrak{c}_1})=2\eta_1>0\), then \(F_d(x)=0\) has no bounded solution on \([2\sqrt{\mathfrak{c}_1},\infty)\) is equivalent to \(F_d(x)>0\) for all \(x\in[2\sqrt{\mathfrak{c}_1},\infty)\). Now, suppose \(\mathfrak{c}_1\geq0.5\), we can obtain that
	\begin{align}
		\sum_{i=2}^d\sqrt{x^2-4\mathfrak{c}_i}>(d-2)x+\sqrt{x^2-4(1-\mathfrak{c}_1)},\ {\rm for\ }x\geq2\sqrt{\mathfrak{c}_1}.\label{Eq of Fd 1}
	\end{align}
	In fact, since
	$$\sqrt{x^2-4\mathfrak{c}_2}+\sqrt{x^2-4\mathfrak{c}_3}>x+\sqrt{x^2-4(\mathfrak{c}_2+\mathfrak{c}_3)}\ \Leftrightarrow\ \mathfrak{c}_2\mathfrak{c}_3>0,$$
	we can easily conclude (\ref{Eq of Fd 1}). Therefore, if \(\mathfrak{c}_1\geq0.5\), it implies that
	$$F_d(x)>\sqrt{x^2-4(1-\mathfrak{c}_1)}-\sqrt{x^2-4\mathfrak{c}_1}\geq0,$$
	i.e. \(F_d(x)=0\) has no bounded solution. On the other hand, suppose \(\mathfrak{c}_1<0.5\), since
	$$F_d(x)\leq(d-1)\sqrt{x^2-\frac{4(1-\mathfrak{c}_1)}{(d-1)}}-\sqrt{x^2-4\mathfrak{c}_1}-(d-2)x:=\widetilde{F}_d(x),$$
	by the method of Lagrange multipliers. We only need to show that there exists a bounded \(x\in[2\sqrt{\mathfrak{c}_1},\infty)\) such that \(\widetilde{F}_d(x)\leq0\). Notice that
	$$\widetilde{F}_d(x)\leq0\ \Leftrightarrow\ x^2-x\sqrt{x^2-4\mathfrak{c}_1}\leq2\frac{(d-1)(1-\mathfrak{c}_1)-\mathfrak{c}_1}{d-2},$$
	and
	$$\mathfrak{c}_1<0.5\ \Longleftrightarrow\ 2\mathfrak{c}_1<2\frac{(d-1)(1-\mathfrak{c}_1)-\mathfrak{c}_1}{d-2},$$
    so we can choose a sufficiently small \(\epsilon>0\) such that
    $$2\mathfrak{c}_1+\epsilon<2\frac{(d-1)(1-\mathfrak{c}_1)-\mathfrak{c}_1}{d-2}.$$
    On the other hand, we can also choose a sufficiently large \(x=x(\epsilon)>0\) such that
	$$x^2-x\sqrt{x^2-4\mathfrak{c}_1}=\frac{4\mathfrak{c}_1x}{x+\sqrt{x^2-4\mathfrak{c}_1}}<2\mathfrak{c}_1+\epsilon,$$
	it implies that \(F_d(y)\leq\widetilde{F}_d(y)\leq0\) for all \(y\in[x(\epsilon),\infty)\), i.e. the solutions of \(F_d(x)=0\) are bounded if and only if \(\mathfrak{c}_1<0.5\). 
\end{proof}
Finally, let us  prove Theorem \ref{Thm of Singularity} as follows:
\begin{proof}[Proof of Theorem \ref{Thm of Singularity}]
	Recall that \(\mfc_1=\max_{1\leq i\leq d}\mfc_i\) and let \(\Pi_0:=\{\mfc\in\mathbb{R}^d:\sum_{i=1}^d\mfc_i=1,\mfc_i>0\}\) be the \(d\)-dimensional affine hyperplane, then we will show that the invariant branch region \(\Pi_1\) is a subset of the nonsingular region \(\Pi_2\), i.e.
	$$\Pi_1=\Big\{\mfc\in\Pi_0:\sum_{i=1}^d\sqrt{\mfc_1-\mfc_i}\leq(d-2)\sqrt{\mfc_1}\Big\}\subset\Pi_2=\{\mfc\in\Pi_0:\mfc_1<0.5\}.$$
	Suppose there is a \(\mfc\in\Pi_1\) such that \(\mfc\notin\Pi_2\), i.e. \(\mfc_1\geq0.5\), then according to (\ref{Eq of Fd 1}), we have
	$$(d-2)\sqrt{\mfc_1}\geq\sum_{i=1}^d\sqrt{\mfc_1-\mfc_i}>(d-2)\sqrt{\mfc_1}+\sqrt{2\mfc_1-1},$$
	it implies that \(0\leq(2\mfc_1-1)^{1/2}<0\), which is a contradiction. Finally, according to (\ref{Eq of h image}), if \(\eta\leq\eta_1\), \(\mathfrak{h}(\eta)\) is the solution of \(2\eta=F_d(\mathfrak{h}(\eta))\). By Proposition \ref{Pro of singularity}, if \(\mathfrak{c}_1\geq0.5\), \(F_d(x)=0\) has no finite solution. Notice that \(\lim_{x\to\infty}F_d(x)=0\), then \(\mathfrak{h}(\eta)\to+\infty\) as \(\eta\to0^+\), i.e. \(\tilde{g}_1(\eta)\to\infty\). Therefore, (\ref{Eq of g(z)}) is singular at \(z=0\). On the other hand, if \(\mathfrak{c}_1<0.5\), there exist a bounded \(x_0\) such that \(F_d(x)<0\) for all \(x>x_0\). Since \(F_d(\mfh(\eta_1))=2\eta_1\), then \(\mathfrak{h}(\eta)\leq x_0\) for all \(\eta\in(0,\eta_1]\), which suggests \(\tilde{g}_i(\eta)\) are all bounded for \(i=1,\cdots,d\). This completes the proof.
\end{proof}

\section{Entrywise law when \texorpdfstring{$d=3$}{d=3}}\label{Sec of entrywise law}
\setcounter{equation}{0}
\def\theequation{\thesection.\arabic{equation}}
\setcounter{subsection}{0}
In this section, we will establish the entrywise law for \(\bbQ(z)\) as follows: 
\begin{thm}\label{Thm of entrywise law d=3}
	Under Assumptions {\rm \ref{Ap of general noise}} and {\rm \ref{Ap of dimension}}, for any \(z\in\mcS_{\eta_0}\) in {\rm (\ref{Eq of stability region})} and \(\omega\in(1/2-\delta,1/2)\), where $\delta>0$ is a sufficiently small number, let 
    \begin{align}
        \bbW^{(3)}(z)=-((z+g(z))\bbI_3-\diag(\bbg(z))+g(z)\bbS_3-\diag(\bbg(z))\bbS_3-\bbS_3\diag(\bbg(z)))^{-1}.\label{Eq of bbW limiting d=3}
    \end{align}
    For \(s,t\in\{1,2,3\}\), we have 
	\begin{align*}
		\left|Q_{i_si_t}^{st}(z)-\mathfrak{c}_s^{-1}g_s(z)\left[\delta_{st}\delta_{i_si_t}+(a_{i_s}^{(s)})^2\sum_{k\neq s}^3(g(z)-g_s(z)-g_k(z))W_{sk}^{(3)}(z)\right]\right|\prec\mrO(\eta_0^{-21}N^{-\omega}),
	\end{align*}
	where \(Q_{i_si_t}^{st}(z)\) is the \((i_s,i_t)\)-th entry of \(\bbQ^{st}\) and \(a_{i_s}^{(s)}\) is the \(i_s\)-th entry of \(\bba^{(s)}\), as does \(W_{sk}^{(3)}(z)\).
\end{thm}
The existence of \(\bbW^{(3)}(z)\) on \(\mcS_{\eta_0}\) is established in Lemma \ref{Lem of invertible bbGa} later. For simplicity, we rewrite the three deterministic unit vectors as follows:
\begin{align}
	\bba:=\bba^{(1)}\in\mbS^{m-1},\quad\bbb:=\bba^{(2)}\in\mbS^{n-1},\quad\bbc:=\bba^{(3)}\in\mbS^{p-1},\label{Eq of vector notation d=3}
\end{align}
and
$$\bbM=\frac{1}{\sqrt{N}}\bbPhi_3(\bbX,\bba,\bbb,\bbc),\quad\bbQ(z)=(\bbM-z\bbI_N)^{-1}=\left(\begin{array}{ccc}
	\bbQ^{11}(z)&\bbQ^{12}(z)&\bbQ^{13}(z)\\\bbQ^{12}(z)'&\bbQ^{22}(z)&\bbQ^{23}(z)\\\bbQ^{13}(z)'&\bbQ^{23}(z)'&\bbQ^{33}(z)
\end{array}\right),$$
where \(z\in\mbC^+\) and \(N=m+n+p\). In addition, we will need the following lemmas:
\begin{lem}[\cite{mcdiarmid1989method}]\label{Lem of Bounded Differences Inequality}
	Let \(\Omega\subset\mathbb{R}\) and \(f:\Omega^n\to\mbC\) such that
	$$\sup_{x_1,\cdots,x_n,x_i'\in\Omega}|f(\cdots,x_i,\cdots)-f(\cdots,x_i',\cdots)|\leq M_i.$$
	where \(M_i\) are bounded positive constants. Then 
	\begin{align}
		\mathbb{P}\left(|f(X_1,\cdots,X_n)^c|\geq t\right)\leq4\exp\left(-\frac{t^2}{\sum_{i=1}^nM_i^2}\right).\label{Eq of Bounded Differences Inequality}
	\end{align}
\end{lem}
\begin{lem}[\cite{khorunzhy1996asymptotic}]\label{Lem of Stein’s equation}
	For any real-valued random variable \(\xi\) with \(\mathbb{E}[|\xi|^{K+2}]<\infty\) and complex-valued function \(g(z)\) with continuous and bounded \(K+1\) derivatives, then
	\begin{align}
		\mathbb{E}\left[\xi g(\xi)\right]=\sum_{l=0}^K\frac{\kappa_{l+1}}{l!}\mathbb{E}\big[g^{(l)}(\xi)\big]+\epsilon_{(K+1)},\label{Eq of cumulant expansion}
	\end{align}
	where \(\kappa_l\) is the \(l\)-th cumulant of $\xi$, and
	$$|\epsilon_{(K+1)}|\leq C_K\sup_{z\in\mathbb{C}}\big|g^{(K+1)}(z)\big|\mathbb{E}\big[|\xi|^{K+2}\big].$$
\end{lem}
Here, the \(l\)-th cumulant of \(\xi\) is defined via
$$\log\mbE[e^{{\rm i}x\xi}]=\sum_{l=1}^{\infty}\kappa_l\frac{({\rm i}x)^l}{l!},\quad x\in\mbR.$$
\subsection{Preliminary Lemmas}\label{sec of Preliminary Lemmas}
To prove Theorem \ref{Thm of entrywise law d=3}, we need to deal with quadratic forms of \(\bbQ(z)\) and \(\diag(\bbQ(z))\). Actually, the $(s,t)$-th entry of \(\bbQ(z)\) itself can be reduced to a special case of such quadratic forms. Therefore, we present several lemmas in this section to deal with such quadratic forms. Here, we use a simple example to illustrate the main purpose of the lemmas in \S\ref{sec of Preliminary Lemmas}. Note that \(N^{-1}\tr(\bbQ(z))=N^{-1}\boldsymbol{1}_N'\diag(\bbQ(z))\boldsymbol{1}_N\) is a quadratic form of \(\diag(\bbQ(z))\), and the general procedures for calculating \(N^{-1}\tr(\bbQ(z))\) have two steps:
\begin{enumerate}[1.]
    \item Show that \(N^{-1}\tr(\bbQ(z))\overset{a.s.}{\longrightarrow}N^{-1}\mbE[\tr(\bbQ(z))]\);
    \item Compute \(N^{-1}\mbE[\tr(\bbQ(z))]\).
\end{enumerate}
For the first step, we need Lemma \ref{Thm of Entrywise almost sure convergence} in \S\ref{ssec of quadratic forms}. This lemma establishes the convergence rate of quadratic forms of \(\bbQ(z)\) and \(\diag(\bbQ(z))\) to their mean, where \(N^{-1}\tr(\bbQ(z))=N^{-1}\boldsymbol{1}_N'\diag(\bbQ(z))\boldsymbol{1}_N\) is a quadratic form of \(\diag(\bbQ(z))\). For the second step, by the definition of \(\bbQ(z)\) in \eqref{Eq of N and Q}, we know that \(\bbM\bbQ(z)-z\bbQ(z)=\bbI_N\), i.e. \(\bbQ(z)=z^{-1}(\bbM\bbQ(z)-\bbI_N)\), so we obtain (e.g.)
\begin{align}
	\frac{1}{N}\mbE[\tr(\bbQ^{11}(z))]=\frac{1}{N}\sum_{i=1}^m\mbE[Q_{ii}^{11}(z)]=\frac{z^{-1}}{N^{3/2}}\sum_{i,j,k=1}^{m,n,p}\mbE[X_{ijk}(c_kQ_{ij}^{12}(z)+b_jQ_{ik}^{13}(z))]-z^{-1}\mfc_1.\notag
\end{align}
To compute \(\mbE[X_{ijk}(c_kQ_{ij}^{12}(z)+b_jQ_{ik}^{13}(z))]\), we will use the cumulant expansion \eqref{Eq of cumulant expansion}. To be precise, the definition of \(Q_{st}(z)\) allows us to treat it as a smooth function of \(\boldsymbol{X}\). Consequently, we can compute its expectation using the cumulant expansion (\ref{Eq of cumulant expansion}).
Next, define
\begin{align}
	\partial_{ijk}^{(l)}:=\frac{\partial^l}{\partial X_{ijk}^l}\quad l\in\mbN^+,\label{Eq of partial operator}
\end{align}
then we have
$$\big(\partial_{ijk}^{(1)}\bbM\big)\bbQ(z)+(\bbM-z\bbI_N)\partial_{ijk}^{(1)}\bbQ(z)=0\quad\Rightarrow\quad\partial_{ijk}^{(1)}\bbQ(z)=-\bbQ(z)\big(\partial_{ijk}^{(1)}\bbM\big)\bbQ(z).$$
By the notations in (\ref{Eq of vector notation d=3}), we rewrite (\ref{Eq of mcA}) as follows:
\begin{align}
	\mcA_{ijk}^{(j_1,j_2)}=\left\{\begin{array}{cc}
		a_i,&(j_1,j_2)=(1,2)\ {\rm or\ }(2,1)\\b_j,&(j_1,j_2)=(1,3)\ {\rm or\ }(3,1)\\c_k,&(j_1,j_2)=(2,3)\ {\rm or\ }(3,2)
	\end{array}\right.\quad\tilde{t}_\alpha=\left\{\begin{array}{ll}
		i,&t_\alpha=1\\j,&t_\alpha=2\\k,&t_\alpha=3
	\end{array}\right.,\label{Eq of nij}
\end{align}
so we obtain that
\begin{align}
	\partial_{ijk}^{(1)}Q_{i_1i_2}^{j_1j_2}(z)=-N^{-1/2}\sum_{\substack{t_1, t_2 \\ t_1 \neq t_2}}Q_{i_1\tilde{t}_1}^{j_1t_1}(z)\mcA_{ijk}^{(t_1,t_2)}Q_{\tilde{t}_2i_2}^{t_2j_2}(z),\label{Eq of partial Q 1}
\end{align}
where the summations of \(t_1\) and \(t_2\) are over \(\{1,2,3\}\). For brevity, we write \(\boldsymbol{Q}\) for \(\boldsymbol{Q}(z)\) henceforth and refer to its entries without the explicit dependency on \(z\), unless otherwise specified. By the cumulant expansion \eqref{Eq of cumulant expansion} and \eqref{Eq of partial Q 1}, we have
\begin{align}
    &\frac{1}{N}\sum_{i=1}^m\mbE[Q_{ii}^{11}]=\frac{z^{-1}}{N^{3/2}}\sum_{i,j,k=1}^{m,n,p}\mbE[\partial_{ijk}^{(1)}(c_kQ_{ij}^{12}+b_jQ_{ik}^{13})]+\frac{z^{-1}}{N^{3/2}}\sum_{i,j,k=1}^{m,n,p}\epsilon_{ijk}^{(2)}-z^{-1}\mfc_1\label{Eq of Q11 cumulant expansion example}\\
    &=-z^{-1}\mfc_1-\frac{z^{-1}}{N^2}\mbE\big[\tr(\bbQ^{11})\tr(\bbQ^{22}+\bbQ^{33})+\tr(\bbQ^{12}\bbQ^{21}+\bbQ^{13}\bbQ^{31})+\bba'\bbQ^{12}\bbQ^{23}\bbc+\bba'\bbQ^{13}\bbQ^{32}\bbb\notag\\
    &+2\bbb'\bbQ^{21}\bbQ^{13}\bbc+2\bbb'\bbQ^{23}\bbc\tr(\bbQ^{11})+\bba'\bbQ^{13}\bbc\tr(\bbQ^{22})+\bba'\bbQ^{12}\bbb\tr(\bbQ^{33})\big]+\frac{z^{-1}}{N^{3/2}}\sum_{i,j,k=1}^{m,n,p}\epsilon_{ijk}^{(2)}.\notag
\end{align}
The lemmas in \S\ref{ssec of quadratic forms} and \S\ref{ssec of major terms} address the following computational challenges that arise when evaluating \eqref{Eq of Q11 cumulant expansion example}.
\begin{enumerate}
    \item Lemma \ref{Thm of Entrywise almost sure convergence} in \S\ref{ssec of quadratic forms} establishes the convergence rate of quadratic forms of \(\bbQ\) and \(\diag(\bbQ)\) to their mean, e.g. \(\bbb'\bbQ^{23}\bbc\) and \(N^{-1}\tr(\bbQ^{11})=N^{-1}\boldsymbol{1}_m'\diag(\bbQ^{11})\boldsymbol{1}_m\), so we can conclude that
    $$\lim_{N\to\infty}\big|\mbE[N^{-1}\tr(\bbQ^{11})N^{-1}\tr(\bbQ^{22})]-\mbE[N^{-1}\tr(\bbQ^{11})]\mbE[N^{-1}\tr(\bbQ^{22})]\big|=0$$
    and 
    $$\lim_{N\to\infty}\big|\mbE[N^{-1}\tr(\bbQ^{11})\bbb'\bbQ^{23}\bbc]-\mbE[N^{-1}\tr(\bbQ^{11})]\mbE[\bbb'\bbQ^{23}\bbc]\big|=0$$
    in \eqref{Eq of Q11 cumulant expansion example}.
    \item To prove that the remainder \(N^{-3/2}\sum_{i,j,k=1}^{m,n,p}\epsilon_{ijk}^{(2)}\) in \eqref{Eq of Q11 cumulant expansion example} vanishes  as \(N\to\infty\), we need lemmas in \S\ref{ssec of major terms}. Actually, for later calculations of the asymptotic mean and variance of the LSS of the matrix \(\bbM\), see \S\ref{Sec of mean and covariance}, we need to compute \(\sum_{i,j,k=1}^{m,n,p}\mbE[\partial_{ijk}^{(l)}(c_kQ_{ij}^{12}+b_jQ_{ik}^{13})]\) for \(l=2,3,4\), and there will appear lots of different complicated terms as those in \eqref{Eq of Q11 cumulant expansion example}. The lemmas in \S\ref{ssec of major terms} help us determine which terms in \(\sum_{i,j,k=1}^{m,n,p}\mbE[\partial_{ijk}^{(l)}(c_kQ_{ij}^{12}+b_jQ_{ik}^{13})]\) vanish as \(N\to\infty\). We refer to these as minor terms. By distinguishing between major and minor terms, we can concentrate on the terms that significantly contribute to the asymptotic mean and variance of the LSS of the matrix \(\bbM\).
\end{enumerate}
Consequently, the lemmas in \S\ref{ssec of quadratic forms} and \S\ref{ssec of major terms} will simplify calculations of \eqref{Eq of Q11 cumulant expansion example} as follows:
$$N^{-1}\mbE[\tr(\bbQ^{11})](z+N^{-1}\mbE[\tr(\bbQ^{22}+\bbQ^{33})])+\mfc_1=\mro(1),$$
i.e. \(-\frac{\mfc_1}{\mfm_1(z)}=z+\mfm_2(z)+\mfm_3(z)+\mro(1)\), and the limiting form of this equation is just the vector Dyson equation \eqref{Eq of MDE 3 order}. For more details, readers can refer to \S\ref{sec of proof entrywise law d=3} later.
\subsubsection{Almost sure convergence of quadratic forms}\label{ssec of quadratic forms}
\begin{lem}\label{Thm of Entrywise almost sure convergence}
	When \(d=3\), for any \(K\in\mathbb{N}^+,z_1,\cdots,z_K\in\mathbb{C}_\eta^+\) and $\omega\in(1/2-\delta,1/2)$, where $\delta>0$ is a sufficiently small number, let \(s_i\in\{1,2,3\}\) for \(i=1,\cdots,K\) such that \(s_{2j}=s_{2j+1}\), then for any two deterministic vectors \(\bbx\in\mbC^{n_{s_1}},\bby\in\mbC^{n_{s_{K+1}}}\) with bounded $L^2$ norms, we have
	\begin{align}
		\Big|\boldsymbol{x}'\prod_{i=1}^K\boldsymbol{Q}^{s_i s_{i+1}}(z_i)\boldsymbol{y}-\mathbb{E}\Big[\boldsymbol{x}'\prod_{i=1}^K\boldsymbol{Q}^{s_i s_{i+1}}(z_i)\boldsymbol{y}\Big]\Big|\prec C_K\eta^{-(K+4)}N^{-\omega}.\label{Eq of as 1}
	\end{align}
	Moreover, if \(s_1=s_{K+1}\), we have
	\begin{align}
		\Big|\boldsymbol{x}'{\rm diag}\Big(\prod_{i=1}^K\boldsymbol{Q}^{s_i s_{i+1}}(z_i)\Big)\boldsymbol{y}-\mathbb{E}\Big[\boldsymbol{x}'{\rm diag}\Big(\prod_{i=1}^K\boldsymbol{Q}^{s_i s_{i+1}}(z_i)\Big)\boldsymbol{y}\Big]\Big|\prec C_K\eta^{-(K+4)}N^{-\omega}\label{Eq of as 2}.
	\end{align}
\end{lem}
Lemma \ref{Thm of Entrywise almost sure convergence} implies that \(\bba'\bbQ^{12}(z)\bbb\overset{a.s.}{\longrightarrow}\mbE[\bba'\bbQ^{12}(z)\bbb]\) and
$$\frac{1}{N}\tr(\bbQ(z))\overset{a.s.}{\longrightarrow}\frac{1}{N}\mbE[\tr(\bbQ(z))],\quad\frac{1}{N}\tr(\bbQ^{12}(z_1)\bbQ^{21}(z_2))\overset{a.s.}{\longrightarrow}\frac{1}{N}\mbE[\tr(\bbQ^{12}(z_1)\bbQ^{21}(z_2))],$$
and all the quadratic forms above will appear in the asymptotic covariance and mean functions of \(\tr(\bbQ(z))-Ng(z)\). 

Moreover, we have given the formula of \(\partial_{ijk}^{(1)}\bbQ\) in \eqref{Eq of partial Q 1}. For higher derivatives, we can show that \(\partial_{ijk}^{(l)}\bbQ=(-1)^ll!\big(\bbQ\partial_{ijk}^{(1)}\bbM\big)^l\bbQ\) for \(l\geq2\) and
\begin{align}
	\partial_{ijk}^{(l)}Q_{i_1i_2}^{j_1j_2}=(-N^{-1/2})^ll!\sum_{t_1\cdots t_{2l}}Q_{i_1\tilde{t}_1}^{j_1t_1}\left(\prod_{\alpha=1}^{l-1}\mcA_{ijk}^{(t_{2\alpha-1},t_{2\alpha})}Q_{\tilde{t}_{2\alpha}\tilde{t}_{2\alpha+1}}^{t_{2\alpha}t_{2\alpha+1}}\right)\mcA_{ijk}^{(t_{2l-1},t_{2l})}Q_{\tilde{t}_{2l}i_2}^{t_{2l}j_2},\label{Eq of partial Q}
\end{align}
where the summations of all \(t_\alpha\) are over \(\{1,2,3\}\) such that \({t_{2\alpha-1}\neq t_{2\alpha}}\) for all \(\alpha=1,\cdots,l\).
\begin{remark}
	Let us  use a simple example to demonstrate the structure of  (\ref{Eq of partial Q}), since
	\begin{align}
		\partial_{ijk}^{(1)}\bbM=\frac{1}{\sqrt{N}}\left(\begin{array}{ccc}
			\boldsymbol{0}_{m\times m}&c_k\boldsymbol{e}_i^m(\boldsymbol{e}_j^n)'&b_j\boldsymbol{e}_i^m(\boldsymbol{e}_k^p)'\\c_k\boldsymbol{e}_j^n(\boldsymbol{e}_i^m)'&\boldsymbol{0}_{n\times n}&a_i\boldsymbol{e}_j^n(\boldsymbol{e}_k^p)'\\b_j\boldsymbol{e}_k^p(\boldsymbol{e}_i^m)'&a_i\boldsymbol{e}_k^p(\boldsymbol{e}_j^n)'&\boldsymbol{0}_{p\times p}
		\end{array}\right),\label{Eq of partial N}
	\end{align}
	where \(\boldsymbol{e}_i^m\) is a \(m\) dimensional vector whose \(i\)-th entry is \(1\) and others are \(0\), as does \(\boldsymbol{e}_j^n,\boldsymbol{e}_k^p\). Consider \(\partial_{ijk}^{(1)}Q_{11}^{11}\), which is indeed the \((1,1)\) entry in the \((1,1)\) block of \(\partial_{ijk}^{(1)}\bbQ=-\bbQ(\partial_{ijk}^{(1)}\bbM)\bbQ\), so we first consider \(\partial_{ijk}^{(1)}\bbQ^{11}\), i.e.
	$$\partial_{ijk}^{(1)}\bbQ^{11}=-\sum_{t_2,t_3}\bbQ^{1t_2}(\partial_{ijk}^{(1)}\bbM^{t_2t_3})\bbQ^{t_31},$$
	notice that the diagonal blocks of \(\partial_{ijk}^{(1)}\bbM\) are zero, then this implies that \(t_2\neq t_3\). Besides, for each \(\partial_{ijk}^{(1)}\bbM^{t_2t_3},t_2\neq t_3\), it only has one nonzero entry with value of \(\mcA_{ijk}^{t_2t_3}\). Hence, we show that for any two adjacent  \(Q_{\tilde{t}_{2\alpha-1}\tilde{t}_{2\alpha}}^{t_{2\alpha-1}t_{2\alpha}}\), i.e. \(Q_{\tilde{t}_{2\alpha-1}\tilde{t}_{2\alpha}}^{t_{2\alpha-1}t_{2\alpha}}\mcA_{ijk}^{(t_{2\alpha},t_{2\alpha+1})}Q_{\tilde{t}_{2\alpha+1}\tilde{t}_{2\alpha+2}}^{t_{2\alpha+1}t_{2\alpha+2}}\), it has \(t_{2\alpha}\neq t_{2\alpha+1}\).
\end{remark}
We say \(Q_{\tilde{t}_{2\alpha-1}\tilde{t}_{2\alpha}}^{t_{2\alpha-1}t_{2\alpha}}\) comes from {\bf diagonal} blocks if \(t_{2\alpha-1}=t_{2\alpha}\), otherwise from {\bf off-diagonal} blocks. Since we will apply the bounded differences inequality (\ref{Eq of Bounded Differences Inequality}) to prove Lemma \ref{Thm of Entrywise almost sure convergence}, we need the upper bound for the summation of higher order derivatives over all \(i,j,k\). Hence, we state the following lemma.
\begin{lem}\label{Lem of minor terms 1}
	When \(d=3\), for any \(K\in\mbN^+\) and \(z\in\mbC_{\eta}^+\), let \(\bbx,\bby\in\mbC^N\) be two deterministic vectors with bounded $L^2$ norms, then we have
	$$\sum_{i,j,k=1}^{m,n,p}\big|\boldsymbol{x}'\partial_{ijk}^{(l)}\Big(\prod_{k=1}^K\bbQ(z)\Big)\boldsymbol{y}\big|^2<\left\{\begin{array}{ll}
		C_{l,K}\Vert\bbQ(z)\Vert^{2(l+K)}N^{-1}&l=1,2,\\
		C_{l,K}\Vert\bbQ(z)\Vert^{2(l+K)}N^{-2}&l=3.
	\end{array}\right.$$
\end{lem}
\begin{proof}
	Recall the notations in (\ref{Eq of vector notation d=3}), assume that \(\boldsymbol{x}=(\boldsymbol{x}^1,\boldsymbol{x}^2,\boldsymbol{x}^3)'\) such that \(\boldsymbol{x}^1\in\mathbb{C}^m,\boldsymbol{x}^2\in\mathbb{C}^n,\boldsymbol{x}^3\in\mathbb{C}^p\) and \(\Vert\boldsymbol{x}^1\Vert_2=\Vert\boldsymbol{x}^2\Vert_2=\Vert\boldsymbol{x}^3\Vert_2=1\), as does \(\boldsymbol{y}=(\boldsymbol{y}^1,\boldsymbol{y}^2,\boldsymbol{y}^3)'\). It suffices to show that for given \(s_1,\cdots,s_{K+1}\in\{1,2,3\}\) and \(l=1,2,3\), we have
	\begin{align}
		\sum_{i,j,k=1}^{m,n,p}\Big|(\boldsymbol{x}^{s_1})'\partial_{ijk}^{(l)}\Big(\prod_{l=1}^K\bbQ^{s_ls_{l+1}}\Big)\bby^{s_{K+1}}\Big|^2<C_l\eta^{-2(l+1)}N^{-1}.\label{Eq of entrywise partial}
	\end{align}
	For simplicity, let 
	\begin{align}
		\boldsymbol{x}^{(i_0)}=\left\{\begin{array}{ll}
			\prod_{i=1}^{i_0-1}\boldsymbol{Q}^{s_is_{i+1}}(z_i)\boldsymbol{x}^{s_1}&i_0>1,\\
			\bbx&i_0=1.
		\end{array}\right.\quad\boldsymbol{y}^{(i_0)}=\left\{\begin{array}{ll}
			\prod_{i=i_0+1}^K\boldsymbol{Q}^{s_is_{i+1}}(z_i)\boldsymbol{y}^{s_{K+1}}&i_0<K,\\
			\bby&i_0=K.
		\end{array}\right.\label{Eq of xi0 yi0}
	\end{align} 
	According to (\ref{Eq of partial Q 1}) and (\ref{Eq of partial Q}), we have
	\begin{align}
		&(\boldsymbol{x}^{j_1})'\partial_{ijk}^{(l)}\boldsymbol{Q}^{j_1j_2}\boldsymbol{y}^{j_2}\label{Eq of inner product partial}\\
		&=\left\{\begin{array}{ll}
			(-N^{-1/2})^ll!\sum_{t_1\cdots t_{2l}}(\boldsymbol{x}^{j_1})'Q_{\cdot\tilde{t}_1}^{j_1t_1}\left(\prod_{\alpha=1}^{l-1}\mcA_{ijk}^{(t_{2\alpha-1},t_{2\alpha})}Q_{\tilde{t}_{2\alpha}\tilde{t}_{2\alpha+1}}^{t_{2\alpha}t_{2\alpha+1}}\right)\mcA_{ijk}^{(t_{2l-1},t_{2l})}Q_{\tilde{t}_{2l}\cdot}^{t_{2l}j_2}\boldsymbol{y}^{j_2}&l\geq2\\
			-N^{-1/2}\sum_{t_1,t_2}(\boldsymbol{x}^{j_1})'Q_{\cdot\tilde{t}_2}^{j_1t_1}\mcA_{ijk}^{(t_1,t_2)}Q_{\tilde{t}_2\cdot}^{t_2j_2}\boldsymbol{y}^{j_2}&l=1
		\end{array}\right.,\notag
	\end{align}
	where \(t_1,\cdots,t_{2l}\in\{1,2,3\}\) such that \(t_{2\alpha-1}\neq t_{2\alpha}\) for \(\alpha=1,\cdots,l\) and \(\tilde{t}_{\alpha}\) is defined in (\ref{Eq of nij}).

    \vspace{5mm}
    \noindent
    {\bf First derivatives:} When \(l=1\), since
		\begin{align*}
			&\sum_{i,j,k=1}^{m,n,p}\Big|(\bbx^{s_1})'\partial_{ijk}^{(1)}\Big(\prod_{l=1}^K\bbQ^{s_ls_{l+1}}\Big)\bby^{s_{K+1}}\Big|^2\leq\sum_{l_0=1}^K\sum_{i,j,k=1}^{m,n,p}\Big|(\bbx^{(l_0)})'\partial_{ijk}^{(1)}\bbQ^{s_{l_0}s_{l_0+1}}\bby^{(l_0)}\Big|^2\\
			&=\sum_{l_0=1}^K\sum_{i,j,k=1}^{m,n,p}\Big|\sum_{t_1\neq t_2}N^{-1/2}(\bbx^{(l_0)})'Q_{\cdot\tilde{t}_1}^{s_{l_0}t_1}\mcA_{ijk}^{(t_1,t_2)}Q_{\tilde{t}_2\cdot}^{t_2s_{l_0+1}}\bby^{(l_0)}\Big|^2\\
			&\leq6\sum_{l_0=1}^K\sum_{t_1\neq t_2}\sum_{i,j,k=1}^{m,n,p}N^{-1}\Big|(\bbx^{(l_0)})'Q_{\cdot\tilde{t}_1}^{s_{l_0}t_1}\mcA_{ijk}^{(t_1,t_2)}Q_{\tilde{t}_2\cdot}^{t_2s_{l_0+1}}\bby^{(l_0)}\Big|^2:=6\sum_{l_0=1}^K\sum_{t_1\neq t_2}\mcP_{t_1t_2},
		\end{align*}
		where we use Cauchy's inequalityin the third step and \(t_1,t_2\in\{1,2,3\}\). To conclude Lemma \ref{Lem of minor terms 1} for \(l=1\), it is enough to show that each \(\mcP_{t_1t_2}<N^{-1}\Vert\bbQ\Vert^{2(K+1)}\). For example,
		$$\mcP_{t_1t_2}=\sum_{i,j,k=1}^{m,n,p}|(\bbx^{(l_0)})'Q_{\cdot j}^{s_{l_0}2}a_iQ_{k\cdot}^{3s_{l_0+1}}\bby^{(l_0)}|^2=N^{-1}\Vert \bbQ^{t_1s_{l_0}}\bbx^{(l_0)}\Vert_2^2\times\Vert \bbQ^{t_2s_{l_0+1}}\bby^{(l_0)}\Vert_2^2\leq N^{-1}\Vert\bbQ\Vert^{2(K+1)},$$
		where we use the fact that \(\Vert \bbQ^{t_1s_{l_0}}\bbx^{(l_0)}\Vert_2\leq\Vert\bbQ^{t_1s_{l_0}}\Vert\cdot\Vert\bbx^{(l_0)}\Vert_2\leq\Vert\bbQ\Vert^{l_0}\) and \(\Vert \bbQ^{t_2s_{l_0+1}}\bby^{(l_0)}\Vert_2\leq\Vert\bbQ\Vert^{K-l_0+1}\). The arguments for the others are the same, the details are omitted for brevity.

        \vspace{5mm}
        \noindent
        {\bf Second derivatives:}  For the second derivatives, i.e. 
        \begin{align}
            &(\bbx^{s_1})'\partial_{ijk}^{(2)}\Big(\prod_{l=1}^K\bbQ^{s_ls_{l+1}}\Big)\bby^{s_{K+1}}=\sum_{l_0=1}^K(\bbx^{(l_0)})'\partial_{ijk}^{(2)}\bbQ^{s_{l_0}s_{l_0+1}}\bby^{(l_0)}\label{Eq of a.s. Lemma 2nd 1}\\
            &+2\sum_{l_0<l_1}^K(\bbx^{(l_0)})'\partial_{ijk}^{(1)}\bbQ^{s_{l_0}s_{l_0+1}}\bbP^{l_0l_1}\partial_{ijk}^{(1)}\bbQ^{s_{l_1}s_{l_1+1}}\bby^{(l_1)},\label{Eq of a.s. Lemma 2nd 2}
        \end{align}
        where \(\bbx^{(l_0)},\bby^{(l_0)}\) are defined in \eqref{Eq of xi0 yi0} and
        \begin{align}
			\bbP^{l_0l_1}:=\left\{\begin{array}{ll}
				\prod_{i=l_0+1}^{l_1}\boldsymbol{Q}^{s_is_{i+1}}&l_0+1<l_0\\\boldsymbol{I}_{n_{s_{l_0+1}}}&l_0+1=l_1
			\end{array}\right..\label{Eq of bbP}
		\end{align}
       Consider the following two possible scenarios. 
        
        \textbf{Case 1} \eqref{Eq of a.s. Lemma 2nd 1}: by Cauchy's inequality  and (\ref{Eq of inner product partial}), for any \(l_0\in\{1,\cdots,K\}\), consider
		\begin{align*}
			&\mcR_{(2,2)}^{(l_0)}:=\sum_{i,j,k=1}^{m,n,p}\Big|(\bbx^{(l_0)})'\partial_{ijk}^{(2)}\bbQ^{s_{l_0}s_{l_0+1}}\bby^{(l_0)}\Big|^2\\
			&=\sum_{i,j,k=1}^{m,n,p}\Big|2N^{-1}\sum_{t_1\neq t_2}\sum_{t_3\neq t_4}(\bbx^{(l_0)})'Q_{\cdot\tilde{t}_1}^{s_{l_0}t_1}\mcA_{ijk}^{(t_1,t_2)}Q_{\tilde{t}_2\tilde{t}_3}^{t_2t_3}\mcA_{ijk}^{(t_3,t_4)}Q_{\tilde{t}_4\cdot}^{t_4s_{l_0+1}}\bby^{(l_0)}\Big|^2\\
			&\leq C\sum_{t_1\neq t_2}\sum_{t_3\neq t_4}\sum_{i,j,k=1}^{m,n,p}N^{-2}\Big|(\bbx^{(l_0)})'Q_{\cdot\tilde{t}_1}^{s_{l_0}t_1}\mcA_{ijk}^{(t_1,t_2)}Q_{\tilde{t}_2\tilde{t}_3}^{t_2t_3}\mcA_{ijk}^{(t_3,t_4)}Q_{\tilde{t}_4\cdot}^{t_4s_{l_0+1}}\bby^{(l_0)}\Big|^2\\
			&:=C\sum_{t_1\neq t_2}\sum_{t_3\neq t_4}\mcP_{t_1\cdots t_4},
		\end{align*}
		where \(t_1,\cdots,t_4\in\{1,2,3\}\). Therefore, it is enough to show that each \(\mcP_{t_1\cdots t_4}\leq N^{-1}\Vert\bbQ\Vert^{2(K+2)}\). 
First, consider the case where \(Q_{\tilde{t}_2\tilde{t}_3}^{t_2t_3}\) is an element of the off-diagonal blocks, which implies that \(t_2\neq t_3\). In this scenario, there are two possible subcases:
		\begin{itemize}
			\item Both \(\mcA_{ijk}^{(t_1,t_2)}\) and \(\mcA_{ijk}^{(t_3,t_4)}\) do not contain \(a_{i_{t_2}}^{(t_2)}\) and \(a_{i_{t_3}}^{(t_3)}\), it implies that \(t_1=t_3\) and \(t_2=t_4\), so
			\begin{align*}
				&\mcP_{t_1\cdots t_4}=\sum_{i,j,k=1}^{m,n,p}N^{-2}\Big|(\bbx^{(l_0)})'Q_{\cdot\tilde{t}_1}^{s_{l_0}t_1}\mcA_{ijk}^{(t_1,t_2)}Q_{\tilde{t}_2\tilde{t}_1}^{t_2t_1}\mcA_{ijk}^{(t_1,t_2)}Q_{\tilde{t}_2\cdot}^{t_2s_{l_0+1}}\bby^{(l_0)}\Big|^2\\ 
				&\leq N^{-2}(|\bbx^{(l_0)}|^{\circ2})'|\bbQ^{s_{l_0}t_1}|^{\circ2}|\bbQ^{t_1t_2}|^{\circ2}|\bbQ^{t_2s_{l_0+1}}|^{\circ2}|\bby^{(l_0)}|^{\circ2}\leq N^{-2}\Vert\bbQ\Vert^{2(K+2)},
			\end{align*}
			where we use the fact that \(\Vert|\boldsymbol{Q}|^{\circ2}\Vert=\Vert\boldsymbol{Q}\circ\overline{\bbQ}\Vert\leq\Vert\boldsymbol{Q}\Vert^2\) and all \(\bba^{(i)}\) are unit vectors.
			\item Otherwise, at least one of \(\mcA_{ijk}^{(t_1,t_2)}\) and \(\mcA_{ijk}^{(t_3,t_4)}\) contains \(a_{i_{t_2}}^{(t_2)}\) or \(a_{i_{t_3}}^{(t_3)}\). Without loss of generality, assume \(a_{i_{t_2}}^{(t_2)}\) exists, then
			\begin{align*}
				&\mcP_{t_1\cdots t_4}\leq N^{-2}\Vert \bbQ^{t_1s_{l_0}}\bbx^{(l_0)}\Vert_2^2\times\Vert \bbQ^{t_4s_{l_0+1}}\bby^{(l_0)}\Vert_2^2\times\Vert \bbQ^{t_3t_2}\bba^{(t_2)}\Vert_2^2\leq N^{-2}\Vert\bbQ\Vert^{2(K+2)}.
			\end{align*}
		\end{itemize}
		Second, suppose \(Q_{\tilde{t}_2\tilde{t}_3}^{t_2t_3}\) comes from the diagonal blocks, i.e. \(t_2=t_3\). Similarly to the previous case, this scenario can be further divided into two subcases:
		\begin{itemize}
			\item If \(t_1=t_4\), we have
			\begin{align*}
				&\mcP_{t_1\cdots t_4}=\sum_{i,j,k=1}^{m,n,p}N^{-2}\Big|(\bbx^{(l_0)})'Q_{\cdot\tilde{t}_1}^{s_{l_0}t_1}\mcA_{ijk}^{(t_1,t_2)}Q_{\tilde{t}_2\tilde{t}_2}^{t_2t_2}\mcA_{ijk}^{(t_2,t_1)}Q_{\tilde{t}_1\cdot}^{t_1s_{l_0+1}}\bby^{(l_0)}\Big|^2\\
				&\leq N^{-2}\tr(|\bbQ^{t_2t_2}|^{2\circ})\cdot(|\bbx^{(l_0)}|^{\circ2})'|\bbQ^{s_{l_0}t_1}|^{\circ2}|\bbQ^{t_1s_{l_0+1}}|^{\circ2}|\bby^{(l_0)}|^{\circ2}\leq N^{-1}\Vert\bbQ\Vert^{2(K+2)},
			\end{align*}
			\item if \(t_1\neq t_4\), we have
			\begin{align*}
				&\mcP_{t_1\cdots t_4}=\sum_{i,j,k=1}^{m,n,p}N^{-2}\Big|(\bbx^{(l_0)})'Q_{\cdot\tilde{t}_1}^{s_{l_0}t_1}\mcA_{ijk}^{(t_1,t_2)}Q_{\tilde{t}_2\tilde{t}_2}^{t_2t_2}\mcA_{ijk}^{(t_2,t_4)}Q_{\tilde{t}_4\cdot}^{t_4s_{l_0+1}}\bby^{(l_0)}\Big|^2\\
				&\leq N^{-2}\tr(|\bbQ^{t_2t_2}|^{2\circ})\cdot(|\bbx^{(l_0)}|^{\circ2})'|\bbQ^{s_{l_0}t_1}|^{\circ2}|\bba^{(t_1)}|^{\circ2}\cdot(|\bba^{(t_4)}|^{\circ2})'|\bbQ^{t_1s_{l_0+1}}|^{\circ2}|\bby^{(l_0)}|^{\circ2}\leq N^{-1}\Vert\bbQ\Vert^{2(K+2)}.
			\end{align*}
		\end{itemize}
		\textbf{Case 2} \eqref{Eq of a.s. Lemma 2nd 2}: for any \(l_0,l_1\in\{1,\cdots,K\}\) such that \(l_0<l_1\), consider
		\begin{align*}
			&\mcR_{(2,1)}^{(l_0,l_1)}:=\sum_{i,j,k=1}^{m,n,p}\Big|(\bbx^{(l_0)})'\partial_{ijk}^{(1)}\bbQ^{s_{l_0}s_{l_0+1}}\bbP^{l_0l_1}\partial_{ijk}^{(1)}\bbQ^{s_{l_1}s_{l_1+1}}\bby^{(l_0)}\Big|^2\\
			&\leq C\sum_{t_1\neq t_2}\sum_{t_3\neq t_4}\sum_{i,j,k=1}^{m,n,p}N^{-2}\Big|(\bbx^{(l_0)})'Q_{\cdot\tilde{t}_1}^{s_{l_0}t_1}\mcA_{ijk}^{(t_1,t_2)}Q_{\tilde{t}_2\cdot}^{t_2s_{l_0+1}}\bbP^{l_0l_1}Q_{\cdot\tilde{t}_3}^{s_{l_1}t_3}\mcA_{ijk}^{(t_3,t_4)}Q_{\tilde{t}_4\cdot}^{t_4s_{l_0+1}}\bby^{(l_0)}\Big|^2\\
			&:=C\sum_{t_1\neq t_2}\sum_{t_3\neq t_4}\mcP_{t_1\cdots t_4},
		\end{align*}
		where we use (\ref{Eq of inner product partial}) and Cauchy's inequality  again. In this situation, we can also show that each \(\mcP_{t_1\cdots t_4}\leq N^{-1}\Vert\bbQ\Vert^{2(K+2)}\) by the same method as in Case 1, so the details are omitted for brevity. Now, notice that
		\begin{align*}
			&\sum_{i,j,k=1}^{m,n,p}\Big|(\bbx^{s_1})'\partial_{ijk}^{(2)}\Big(\prod_{l=1}^K\bbQ^{s_ls_{l+1}}\Big)\bby^{s_{K+1}}\Big|^2\leq\sum_{l_0=1}^K\mcR_{(2,2)}^{(l_0)}+\sum_{l_0\neq l_1}^K\mcR_{(2,1)}^{(l_0,l_1)},
		\end{align*}
		then we can conclude Lemma \ref{Lem of minor terms 1} for \(l=2\).

 \vspace{5mm}
\noindent
 {\bf Third derivatives:} Similar to \eqref{Eq of a.s. Lemma 2nd 1} and \eqref{Eq of a.s. Lemma 2nd 2}, for the third derivatives, we have
 \begin{align}
     &(\bbx^{(l_0)})'\partial_{ijk}^{(3)}\bbQ^{s_{l_0}s_{l_0+1}}\bby^{(l_0)}=\sum_{l_0=1}^K(\bbx^{(l_0)})'\partial_{ijk}^{(3)}\bbQ^{s_{l_0}s_{l_0+1}}\bby^{(l_0)}\notag\\
     &+3\sum_{l_0<l_1}^K(\bbx^{(l_0)})'\partial_{ijk}^{(2)}\bbQ^{s_{l_0}s_{l_0+1}}\bbP^{l_0l_1}\partial_{ijk}^{(1)}\bbQ^{s_{l_1}s_{l_1+1}}\bby^{(l_1)}+3\sum_{l_0<l_1}^K(\bbx^{(l_0)})'\partial_{ijk}^{(1)}\bbQ^{s_{l_0}s_{l_0+1}}\bbP^{l_0l_1}\partial_{ijk}^{(2)}\bbQ^{s_{l_1}s_{l_1+1}}\bby^{(l_1)}\notag\\
     &+6\sum_{l_0<l_1<l_2}^K(\bbx^{(l_0)})'\partial_{ijk}^{(1)}\bbQ^{s_{l_0}s_{l_0+1}}\bbP^{l_0l_1}\partial_{ijk}^{(1)}\bbQ^{s_{l_1}s_{l_1+1}}\bbP^{l_1l_2}\partial_{ijk}^{(1)}\bbQ^{s_{l_2}s_{l_2+1}}\bby^{(l_2)}.\label{Eq of Lem of minor terms 1 3rd derivative}
 \end{align}
 Here, we only present the detailed calculation procedures for \((\bbx^{(l_0)})'\partial_{ijk}^{(3)}\bbQ^{s_{l_0}s_{l_0+1}}\bby^{(l_0)}\), since the arguments for the others are the same. By Cauchy's inequality  and (\ref{Eq of inner product partial}), for any \(l_0\in\{1,\cdots,K\}\), we have
		\begin{align}
			&\mcR_{(3,3)}^{(l_0)}:=\sum_{i,j,k=1}^{m,n,p}\Big|(\bbx^{(l_0)})'\partial_{ijk}^{(3)}\bbQ^{s_{l_0}s_{l_0+1}}\bby^{(l_0)}\Big|^2\label{Eq of minor 3rd}\\
			&\leq C\sum_{t_1\neq t_2}\sum_{t_3\neq t_4}\sum_{t_5\neq t_6}\sum_{i,j,k=1}^{m,n,p}N^{-3}\Big|(\bbx^{(l_0)})'Q_{\cdot\tilde{t}_1}^{s_{l_0}t_1}\mcA_{ijk}^{(t_1,t_2)}Q_{\tilde{t}_2\tilde{t}_3}^{t_2t_3}\mcA_{ijk}^{(t_3,t_4)}Q_{\tilde{t}_4\tilde{t}_5}^{t_4t_5}\mcA_{ijk}^{(t_5,t_6)}Q_{\tilde{t}_6\cdot}^{t_6s_{l_0+1}}\bby^{(l_0)}\Big|^2\notag\\
			&:=C\sum_{t_1\neq t_2}\sum_{t_3\neq t_4}\sum_{t_5\neq t_6}\mcP_{t_1\cdots t_6}\notag,
		\end{align}
		where \(t_1,\cdots,t_6\in\{1,2,3\}\). First, if both \(Q_{\tilde{t}_2\tilde{t}_3}^{t_2t_3},Q_{\tilde{t}_4\tilde{t}_5}^{t_4t_5}\) come from the off-diagonal blocks, i.e. \(t_2\neq t_3,t_4\neq t_5\), we claim that \(\mcP_{t_1\cdots t_6}\leq N^{-3}\Vert\bbQ\Vert^{2(K+3)}\). Consider the following two subcases.
		\begin{itemize}
			\item If there is no \(a_{i_{t_2}}^{(t_2)}\) and \(a_{i_{t_3}}^{(t_3)}\) associating with \(Q_{\tilde{t}_2\tilde{t}_3}^{t_2t_3}\) or \(a_{i_{t_4}}^{(t_4)}\) and \(a_{i_{t_5}}^{(t_5)}\) associating with \(Q_{\tilde{t}_4\tilde{t}_5}^{t_4t_5}\). Without loss of generality, assume there is no \(a_{i_{t_2}}^{(t_2)}\) and \(a_{i_{t_3}}^{(t_3)}\) in all \(\mcA_{ijk}^{(t_1,t_2)},\mcA_{ijk}^{(t_3,t_4)},\mcA_{ijk}^{(t_5,t_6)}\), then \(\mcA_{ijk}^{(t_1,t_2)}=\mcA_{ijk}^{(t_3,t_4)}=\mcA_{ijk}^{(t_5,t_6)}\), i.e. \(t_1=t_3=t_5,t_2=t_4=t_6\) while \(t_1\neq t_2\), then it implies that
			\begin{align*}
				&\mcP_{t_1\cdots t_6}=N^{-3}(|\bbx^{(l_0)}|^{\circ2})'|\bbQ^{s_{l_0}t_1}|^{\circ2}|\bbQ^{t_1t_2}|^{\circ2}|\bbQ^{t_2t_3}|^{\circ2}|\bbQ^{t_3s_{s_{l_0+1}}}|^{\circ2}|\bby^{(l_0)}|^{\circ2}\leq N^{-3}\Vert\bbQ\Vert^{2(K+3)}.
			\end{align*}
			\item Otherwise, both \(Q_{\tilde{t}_2\tilde{t}_3}^{t_2t_3}\) and \(Q_{\tilde{t}_4\tilde{t}_5}^{t_4t_5}\) have at least one of \(a_{i_{t_2}}^{(t_2)},a_{i_{t_3}}^{(t_3)}\) and \(a_{i_{t_4}}^{(t_4)},a_{i_{t_5}}^{(t_5)}\) associating with itself, respectively. Since the case when \((t_2,t_3)=(t_4,t_5)\) is solved previously, Consider the situation when there is one common \(t_i\) in \((t_2,t_3)\) and \((t_4,t_5)\), i.e. \(t_2=t_5\) while \(t_3\neq t_4\), then we have \(t_2\neq t_3\) and \(t_2\neq t_4\) due to both \(Q_{\tilde{t}_2\tilde{t}_3}^{t_2t_3}\) and \(Q_{\tilde{t}_4\tilde{t}_5}^{t_4t_5}\) come from the off-diagonal blocks, so \(\mcA_{ijk}^{(t_3,t_4)}\) must contain \(a_{i_{t_2}}^{(t_2)}\). In this case, if \(\mcA_{ijk}^{(t_1,t_2)}\) or \(\mcA_{ijk}^{(t_2,t_6)}\) contains \(a_{i_{t_3}}^{(t_3)}\) or \(a_{i_{t_4}}^{(t_4)}\), without loss of generality, assume there exists \(a_{i_{t_3}}^{(t_3)}\), then
			\begin{align*}
				&\mcP_{t_1\cdots t_6}=N^{-3}\sum_{i,j,k=1}^{m,n,p}\Big|(\bbx^{(l_0)})'Q_{\cdot\tilde{t}_1}^{s_{l_0}t_1}\mcA_{ijk}^{(t_1,t_2)}Q_{\tilde{t}_2\tilde{t}_3}^{t_2t_3}\mcA_{ijk}^{(t_3,t_4)}Q_{\tilde{t}_4\tilde{t}_2}^{t_4t_2}\mcA_{ijk}^{(t_2,t_6)}Q_{\tilde{t}_6\cdot}^{t_6s_{l_0+1}}\bby^{(l_0)}\Big|^2\\
				&\leq N^{-3}\Vert \bbQ^{t_1s_{l_0}}\bbx^{(l_0)}\Vert_2^2\times\Vert \bbQ^{t_6s_{l_0+1}}\bby^{(l_0)}\Vert_2^2\times\Vert \bbQ^{t_2t_3}\bba^{(t_3)}\Vert_2^2\times\Vert \bbQ^{t_4t_2}\bba^{(t_2)}\Vert_2^2\leq N^{-3}\Vert\bbQ\Vert^{2(K+3)}.
			\end{align*}
			If not, then \(t_1=t_3\neq t_4=t_6\) or \(t_1=t_6\neq t_3=t_4\), so we have (e.g.)
			\begin{align*}
				&\mcP_{t_1\cdots t_6}=N^{-3}\sum_{i,j,k=1}^{m,n,p}\Big|(\bbx^{(l_0)})'Q_{\cdot\tilde{t}_1}^{s_{l_0}t_1}\mcA_{ijk}^{(t_1,t_2)}Q_{\tilde{t}_2\tilde{t}_1}^{t_2t_1}\mcA_{ijk}^{(t_1,t_4)}Q_{\tilde{t}_4\tilde{t}_2}^{t_4t_2}\mcA_{ijk}^{(t_2,t_4)}Q_{\tilde{t}_4\cdot}^{t_4s_{l_0+1}}\bby^{(l_0)}\Big|^2\\
				&\leq N^{-3}(|\bbx^{(l_0)}|^{\circ2})'|\bbQ^{s_{l_0}t_1}|^{\circ2}|\bbQ^{t_1t_2}|^{\circ2}|\bbQ^{t_2t_4}|^{\circ2}|\bbQ^{t_4s_{s_{l_0+1}}}|^{\circ2}|\bby^{(l_0)}|^{\circ2}\leq N^{-3}\Vert\bbQ\Vert^{2(K+3)}.
			\end{align*}
			Finally, when there is no common \(t_i\) in \((t_2,t_3)\) and \((t_4,t_5)\), since both \(Q_{\tilde{t}_2\tilde{t}_3}^{t_2t_3}\) and \(Q_{\tilde{t}_4\tilde{t}_5}^{t_4t_5}\) have at least one of \(a_{i_{t_2}}^{(t_2)},a_{i_{t_3}}^{(t_3)}\) and \(a_{i_{t_4}}^{(t_4)},a_{i_{t_5}}^{(t_5)}\) associating with itself, respectively, we have (e.g.)
				\begin{align*}
				&\mcP_{t_1\cdots t_6}=N^{-3}\sum_{i,j,k=1}^{m,n,p}\Big|(\bbx^{(l_0)})'Q_{\cdot\tilde{t}_1}^{s_{l_0}t_1}\mcA_{ijk}^{(t_1,t_2)}Q_{\tilde{t}_2\tilde{t}_3}^{t_2t_3}\mcA_{ijk}^{(t_3,t_4)}Q_{\tilde{t}_4\tilde{t}_5}^{t_4t_5}\mcA_{ijk}^{(t_5,t_6)}Q_{\tilde{t}_6\cdot}^{t_6s_{l_0+1}}\bby^{(l_0)}\Big|^2\\
				&\leq N^{-3}\Vert \bbQ^{t_1s_{l_0}}\bbx^{(l_0)}\Vert_2^2\times\Vert \bbQ^{t_6s_{l_0+1}}\bby^{(l_0)}\Vert_2^2\times\Vert \bbQ^{t_2t_3}\bba^{(t_3)}\Vert_2^2\times\Vert \bbQ^{t_4t_5}\bba^{(t_5)}\Vert_2^2\leq N^{-3}\Vert\bbQ\Vert^{2(K+3)}.
			\end{align*}
		\end{itemize}
		Next, if one of \(Q_{\tilde{t}_2\tilde{t}_3}^{t_2t_3},Q_{\tilde{t}_4\tilde{t}_5}^{t_4t_5}\) comes from the diagonal blocks, e.g. \(t_2=t_3\) without loss of generality, consider the following two subcases. 
		\begin{itemize}
			\item First, if \(\mcA_{ijk}^{(t_5,t_6)}\) does not contain \(a_{i_{t_2}}^{(t_2)}\), then \(t_5=t_2\) or \(t_6=t_2\). Suppose \(t_6=t_2\), it implies that \(t_4\neq t_2\) and \(t_5\neq t_2\), since \(t_4\neq t_5\), \(\mcA_{ijk}^{(t_1,t_2)}\) must contains \(a_{i_{t_4}}^{(t_4)}\) or \(a_{i_{t_5}}^{(t_5)}\), then we have (e.g.)
			\begin{align*}
				&\mcP_{t_1\cdots t_6}=\sum_{i,j,k=1}^{m,n,p}N^{-3}\Big|(\bbx^{(l_0)})'Q_{\cdot\tilde{t}_1}^{s_{l_0}t_1}\mcA_{ijk}^{(t_1,t_2)}Q_{\tilde{t}_2\tilde{t}_2}^{t_2t_2}\mcA_{ijk}^{(t_2,t_4)}Q_{\tilde{t}_4\tilde{t}_5}^{t_4t_5}\mcA_{ijk}^{(t_5,t_2)}Q_{\tilde{t}_2\cdot}^{t_2s_{l_0+1}}\bby^{(l_0)}\Big|^2\\
				&\leq N^{-3}\Vert \bbQ^{t_1s_{l_0}}\bbx^{(l_0)}\Vert_2^2\cdot\Vert \bbQ^{t_5t_4}\bba^{(t_4)}\Vert_2^2\cdot\boldsymbol{1}'\diag(|\bbQ^{t_2t_2}|^{\circ2})|\bbQ^{t_2s_{l_0+1}}|^{\circ2}|\bby^{(l_0)}|^{\circ2}\leq N^{-2}\Vert\bbQ\Vert^{2(K+3)}.
			\end{align*}
			Otherwise, \(t_5=t_2\), if \(\mcA_{ijk}^{(t_1,t_2)}\) and \(\mcA_{ijk}^{(t_2,t_6)}\) do not contain \(a_{i_{t_4}}^{(t_4)}\), then \(t_1=t_4=t_6\) and
			\begin{align*}
				&\mcP_{t_1\cdots t_6}=\sum_{i,j,k=1}^{m,n,p}N^{-3}\Big|(\bbx^{(l_0)})'Q_{\cdot\tilde{t}_1}^{s_{l_0}t_1}\mcA_{ijk}^{(t_1,t_2)}Q_{\tilde{t}_2\tilde{t}_2}^{t_2t_2}\mcA_{ijk}^{(t_2,t_1)}Q_{\tilde{t}_1\tilde{t}_2}^{t_1t_2}\mcA_{ijk}^{(t_2,t_1)}Q_{\tilde{t}_1\cdot}^{t_1s_{l_0+1}}\bby^{(l_0)}\Big|^2\\
				&\leq N^{-3}(|\bbx^{(l_0)}|^{\circ2})'|\bbQ^{s_{l_0}t_1}|^{\circ2}\diag(\boldsymbol{1}'\diag(|\bbQ^{t_2t_2}|^{\circ2})|\bbQ^{t_2t_1}|^{\circ2})|\bbQ^{t_1s_{l_0+1}}|^{\circ2}|\bby^{(l_0)}|^{\circ2}\leq N^{-2}\Vert\bbQ\Vert^{2(K+3)}.
			\end{align*}
			Finally, if \(\mcA_{ijk}^{(t_1,t_2)}\) or \(\mcA_{ijk}^{(t_2,t_6)}\) contains \(a_{i_{t_4}}^{(t_4)}\), then
			\begin{align*}
				&\mcP_{t_1\cdots t_6}\leq N^{-3}\Vert \bbQ^{t_1s_{l_0}}\bbx^{(l_0)}\Vert_2^2\cdot\Vert \bbQ^{t_2s_{l_0+1}}\bby^{(l_0)}\Vert_2^2\cdot\Vert\diag(\bbQ^{t_2t_2})\bbQ^{t_2t_4}\bba^{(t_4)}\Vert_2^2\leq N^{-3}\Vert\bbQ\Vert^{2(K+3)}.
			\end{align*}
			\item Next, if \(\mcA_{ijk}^{(t_5,t_6)}\) contains \(a_{i_{t_2}}^{(t_2)}\), i.e. \(t_5\neq t_2\) and \(t_6\neq t_2\). Notice that \(t_4\neq t_2\) and \(t_4\neq t_5\) due to \(Q_{\tilde{t}_4\tilde{t}_5}^{t_4t_5}\) comes from off-diagonal blocks, then \(\mcA_{ijk}^{(t_1,t_2)}\) must contain \(a_{i_{t_4}}^{(t_4)}\) or \(a_{i_{t_5}}^{(t_5)}\), so
			\begin{align*}
				&\mcP_{t_1\cdots t_6}=\sum_{i,j,k=1}^{m,n,p}N^{-3}\Big|(\bbx^{(l_0)})'Q_{\cdot\tilde{t}_1}^{s_{l_0}t_1}\mcA_{ijk}^{(t_1,t_2)}Q_{\tilde{t}_2\tilde{t}_2}^{t_2t_2}\mcA_{ijk}^{(t_2,t_4)}Q_{\tilde{t}_4\tilde{t}_5}^{t_4t_5}\mcA_{ijk}^{(t_5,t_6)}Q_{\tilde{t}_6\cdot}^{t_6s_{l_0+1}}\bby^{(l_0)}\Big|^2\\
				&\leq N^{-3}\Vert \bbQ^{t_1s_{l_0}}\bbx^{(l_0)}\Vert_2^2\cdot\Vert \bbQ^{t_2s_{l_0+1}}\bby^{(l_0)}\Vert_2^2\cdot\Vert\bbQ^{t_5t_4}\bba^{(t_4)}\Vert_2^2\cdot\boldsymbol{1}'\diag(|\bbQ^{t_2t_2}|^{\circ2})|\bba^{(t_2)}|^{\circ2}\leq N^{-2}\Vert\bbQ\Vert^{2(K+3)}.
			\end{align*}
		\end{itemize}
		Finally, if both of \(Q_{\tilde{t}_2\tilde{t}_3}^{t_2t_3},Q_{\tilde{t}_4\tilde{t}_5}^{t_4t_5}\) comes from the diagonal blocks, i.e. \(t_2=t_3,t_4=t_5\). We have three subcases.
		\begin{itemize}
			\item Since \(t_2\neq t_4\), if both \(\mcA_{ijk}^{(t_1,t_2)}\) and \(\mcA_{ijk}^{(t_5,t_6)}=\mcA_{ijk}^{(t_4,t_6)}\) does not contain \(a_{i_{t_2}}^{(t_2)}\) and \(a_{i_{t_4}}^{(t_4)}\), then \(\mcA_{ijk}^{(t_1,t_2)}=\mcA_{ijk}^{(t_3,t_4)}=\mcA_{ijk}^{(t_2,t_4)}=\mcA_{ijk}^{(t_5,t_6)}=\mcA_{ijk}^{(t_4,t_6)}\) and
			\begin{align*}
				&\mcP_{t_1\cdots t_6}=\sum_{i,j,k=1}^{m,n,p}N^{-3}\Big|(\bbx^{(l_0)})'Q_{\cdot\tilde{t}_4}^{s_{l_0}t_4}\mcA_{ijk}^{(t_4,t_2)}Q_{\tilde{t}_2\tilde{t}_2}^{t_2t_2}\mcA_{ijk}^{(t_2,t_4)}Q_{\tilde{t}_4\tilde{t}_4}^{t_4t_4}\mcA_{ijk}^{(t_4,t_2)}Q_{\tilde{t}_2\cdot}^{t_2s_{l_0+1}}\bby^{(l_0)}\Big|^2\\
				&\leq N^{-3}(|\bbx^{(l_0)}|^{\circ2})'|\bbQ^{s_{l_0}t_4}|^{\circ2}\diag(|\bbQ^{t_4t_4}|^{\circ2})\boldsymbol{1}\cdot\boldsymbol{1}'\diag(|\bbQ^{t_2t_2}|^{\circ2})|\bbQ^{t_2s_{l_0+1}}|^{\circ2}|\bby^{(l_0)}|^{\circ2}\leq N^{-2}\Vert\bbQ\Vert^{2(K+3)}.
			\end{align*}
			\item Otherwise, if there exists one of \(a_{i_{t_2}}^{(t_2)}\) and \(a_{i_{t_4}}^{(t_4)}\), without loss generality, assume \(a_{i_{t_2}}^{(t_2)}\) exists, then \(t_6\neq t_2\) and \(t_1=t_4\). Since \(t_6\neq t_4\), then \(\mcA_{ijk}^{(t_2,t_4)}\) must contain \(a_{i_{t_6}}^{(t_6)}\) and we have
			\begin{align*}
				&\mcP_{t_1\cdots t_6}=\sum_{i,j,k=1}^{m,n,p}N^{-3}\Big|(\bbx^{(l_0)})'Q_{\cdot\tilde{t}_4}^{s_{l_0}t_4}\mcA_{ijk}^{(t_4,t_2)}Q_{\tilde{t}_2\tilde{t}_2}^{t_2t_2}\mcA_{ijk}^{(t_2,t_4)}Q_{\tilde{t}_4\tilde{t}_4}^{t_4t_4}\mcA_{ijk}^{(t_4,t_6)}Q_{\tilde{t}_6\cdot}^{t_6s_{l_0+1}}\bby^{(l_0)}\Big|^2\\
				&\leq N^{-3}\boldsymbol{1}\diag(|\bbQ^{t_2t_2}|^{\circ2})|\bba^{(t_2)}|^{\circ2}\cdot(|\bbx^{(l_0)}|^{\circ2})'|\bbQ^{s_{l_0}t_4}|^{\circ2}\diag(|\bbQ^{t_4t_4}|^{\circ2})\boldsymbol{1}\cdot(|\bba^{(t_6)}|^{\circ2})'|\bbQ^{t_6s_{l_0+1}}|^{\circ2}|\bby^{(l_0)}|^{\circ2}\\
				&\leq N^{-2}\Vert\bbQ\Vert^{2(K+3)}.
			\end{align*}
			\item Finally, when both \(a_{i_{t_2}}^{(t_2)}\) and \(a_{i_{t_4}}^{(t_4)}\) exists, we have
			\begin{align*}
				&\mcP_{t_1\cdots t_6}\leq N^{-3}\Vert \bbQ^{t_1s_{l_0}}\bbx^{(l_0)}\Vert_2^2\cdot\Vert \bbQ^{t_2s_{l_0+1}}\bby^{(l_0)}\Vert_2^2\cdot\boldsymbol{1}\diag(|\bbQ^{t_2t_2}|^{\circ2})|\bba^{(t_2)}|^{\circ2}\\
				&\cdot\boldsymbol{1}\diag(|\bbQ^{t_4t_4}|^{\circ2})|\bba^{(t_4)}|^{\circ2}\leq N^{-2}\Vert\bbQ\Vert^{2(K+3)}.
			\end{align*}
		\end{itemize}
		Now, we have shown that each \(\mcP_{t_1\cdots t_6}\) in (\ref{Eq of minor 3rd}) is bounded by \(N^{-2}\Vert\bbQ\Vert^{2(K+3)}\). Similar to the previous arguments, for any \(l_0,l_1\in\{1,\cdots,K\}\) such that \(l_0<l_1\) or \(l_0,l_1,l_2\in\{1,\cdots,K\}\) such that \(l_0<l_1<l_2\), we can show that 
		\begin{align}
			&\mcR_{(3,2)}^{(l_0,l_1)}=\sum_{i,j,k=1}^{m,n,p}\Big|(\bbx^{(l_0)})'\partial_{ijk}^{(2)}\bbQ^{s_{l_0}s_{l_0+1}}\bbP^{l_0l_1}\partial_{ijk}^{(1)}\bbQ^{s_{l_1}s_{l_1+1}}\bby^{(l_1)}\Big|^2\leq N^{-2}\Vert\bbQ\Vert^{2(K+3)},\label{Eq of minor 3rd 2}
		\end{align}
		and
		\begin{align}
			&\mcR_{(3,1)}^{(l_0,l_1,l_2)}:=\sum_{i,j,k=1}^{m,n,p}\Big|(\bbx^{(l_0)})'\partial_{ijk}^{(1)}\bbQ^{s_{l_0}s_{l_0+1}}\bbP^{l_0l_1}\partial_{ijk}^{(1)}\bbQ^{s_{l_1}s_{l_1+1}}\bbP^{l_1l_2}\partial_{ijk}^{(1)}\bbQ^{s_{l_2}s_{l_2+1}}\bby^{(l_2)}\Big|^2\leq N^{-2}\Vert\bbQ\Vert^{2(K+3)},\label{Eq of minor 3rd 3}
		\end{align}
		here we omit the details for clarity. Finally, combining \eqref{Eq of Lem of minor terms 1 3rd derivative}, \eqref{Eq of minor 3rd}, \eqref{Eq of minor 3rd 2} and \eqref{Eq of minor 3rd 3}, we have
		\begin{align*}
			&\sum_{i,j,k=1}^{m,n,p}\Big|(\bbx^{s_1})'\partial_{ijk}^{(3)}\Big(\prod_{l=1}^K\bbQ^{s_ls_{l+1}}\Big)\bby^{s_{K+1}}\Big|^2\leq\sum_{l_0=1}^K\mcR_{(3,3)}^{(l_0)}+\sum_{l_0\neq l_2}^K\mcR_{(3,2)}^{(l_0,l_1)}+\sum_{l_0\neq l_1\neq l_2}^K\mcR_{(3,1)}^{(l_0,l_1,l_2)}\leq\mrO(N^{-2}\Vert\bbQ\Vert^{2(K+3)}).
		\end{align*}
		This completes the proof of Lemma \ref{Lem of minor terms 1} for \(l=3\).
\end{proof}
Now, we provide the proof of Lemma \ref{Thm of Entrywise almost sure convergence} as follows:
\begin{proof}[Proof of Lemma \ref{Thm of Entrywise almost sure convergence}]
	We will demonstrate the proof for equation (\ref{Eq of as 1}), as the approach for proving (\ref{Eq of as 2}) follows an identical strategy. Without loss of generality, we assume that \(\Vert\bbx\Vert_2=\Vert\bby\Vert_2=1\). Notice that 
	$$\boldsymbol{x}'\prod_{i=1}^K\boldsymbol{Q}(z_i)\boldsymbol{y}=\sum_{s_1\cdots s_{K+1}}\boldsymbol{x}^{s_1}\prod_{i=1}^K\boldsymbol{Q}^{s_is_{i+1}}(z_i)\boldsymbol{y}^{s_{K+1}},$$
	where \(s_i\in\{1,2,3\}\) for \(i=1,\cdots,K+1\). Hence, we only need to show that for each given \((s_1,\cdots,s_{K+1})\), we have that
	$$\Big|(\boldsymbol{x}^{s_1})'\prod_{i=1}^K\boldsymbol{Q}^{s_is_{i+1}}(z_i)\boldsymbol{y}^{s_{K+1}}-\mathbb{E}\Big[(\boldsymbol{x}^{s_1})'\prod_{i=1}^K\boldsymbol{Q}^{s_is_{i+1}}(z_i)\boldsymbol{y}^{s_{K+1}}\Big]\Big|\prec C_K\eta_0^{-(K+4)}N^{-\omega}.$$
    For a sufficiently small \(\delta>0\), we split the following probability into two parts:
    \begin{align}
		&\mathbb{P}\left(\Big|\boldsymbol{x}^{s_1}\prod_{i=1}^K\boldsymbol{Q}^{s_is_{i+1}}(z_i)\boldsymbol{y}^{s_{K+1}}-\mathbb{E}\Big[\boldsymbol{x}^{s_1}\prod_{i=1}^K\boldsymbol{Q}^{s_is_{i+1}}(z_i)\boldsymbol{y}^{s_{K+1}}\Big]\Big|\geq t\right)\notag\\
		&\leq\mathbb{P}\left(\Big|\boldsymbol{x}^{s_1}\prod_{i=1}^K\boldsymbol{Q}^{s_is_{i+1}}(z_i)\boldsymbol{y}^{s_{K+1}}-\mathbb{E}\Big[\boldsymbol{x}^{s_1}\prod_{i=1}^K\boldsymbol{Q}^{s_is_{i+1}}(z_i)\boldsymbol{y}^{s_{K+1}}\Big]\Big|\geq t,\forall X_{ijk}\leq N^{\delta}\right)\label{Eq of a.s. convergence major probability}\\
		&+\mathbb{P}\left(\Big|\boldsymbol{x}^{s_1}\prod_{i=1}^K\boldsymbol{Q}^{s_is_{i+1}}(z_i)\boldsymbol{y}^{s_{K+1}}-\mathbb{E}\Big[\boldsymbol{x}^{s_1}\prod_{i=1}^K\boldsymbol{Q}^{s_is_{i+1}}(z_i)\boldsymbol{y}^{s_{K+1}}\Big]\Big|\geq t,\exists X_{ijk}> N^{\delta}\right).\label{Eq of a.s. convergence tail probability}
	\end{align}
    For the second part \eqref{Eq of a.s. convergence tail probability}, by Assumption \ref{Ap of general noise}, the tail probability satisfies that
    \begin{align}
        \eqref{Eq of a.s. convergence tail probability}\leq\sum_{i,j,k=1}^{m,n,p}\mathbb{P}\left(|X_{ijk}|>N^{\delta}\right)\leq N^3\exp(-N^{\delta\theta}).\label{Eq of a.s. convergence tail probability 1}
    \end{align}
    For the first part \eqref{Eq of a.s. convergence major probability}, since all \(X_{ijk}\) are bounded by \(N^{\delta}\), then we will apply the bounded differences inequality \eqref{Eq of Bounded Differences Inequality} to compute \eqref{Eq of a.s. convergence major probability}. Note that \(\boldsymbol{Q}^{s_is_{i+1}}(z)\) is a differentiable function of \(\boldsymbol{X}\), denoted as \(\boldsymbol{Q}^{s_is_{i+1}}(z,\boldsymbol{X})\), let \(\boldsymbol{X}\) and \(\boldsymbol{X}^{(ijk)}\) be two random tensors that are identical for all elements except at position $(i,j,k)$, where \(X_{ijk}\) and \(X_{ijk}^{(ijk)}\) are independent and identically distributed (one can refer to Lemma \ref{Lem of Bounded Differences Inequality}).
    By the bounded differences inequality (\ref{Eq of Bounded Differences Inequality}), we have
    \begin{align*}
        &\eqref{Eq of a.s. convergence major probability}\leq4\exp\left(-\frac{t^2}{\sum_{i,j,k=1}^{m,n,p}\Delta_{ijk}^2}\right),
    \end{align*}
    where
    \begin{align*}
        &\Delta_{ijk}:=\sup_{|X_{ijk}|,|X_{ijk}^{(ijk)}|\leq N^{\delta}}\Bigg|(\boldsymbol{x}^{s_1})'\prod_{i=1}^K\boldsymbol{Q}^{s_is_{i+1}}(z_i,\boldsymbol{X}^{(ijk)})\boldsymbol{y}^{s_{K+1}}-(\boldsymbol{x}^{s_1})'\prod_{i=1}^K\boldsymbol{Q}^{s_is_{i+1}}(z_i,\boldsymbol{X})\boldsymbol{y}^{s_{K+1}}\Bigg|.
    \end{align*}
	By the Taylor expansion, we have
	\begin{align}
		\Delta_{ijk}&=\sup_{|X_{ijk}|,|X_{ijk}^{(ijk)}|\leq N^{\delta}}\Bigg|\sum_{l=1}^{\infty}(l!)^{-1}(\boldsymbol{x}^{s_1})'\partial_{ijk}^{(l)}\Big\{\prod_{i=1}^K\boldsymbol{Q}^{s_is_{i+1}}(z_i,\boldsymbol{X})\Big\}\boldsymbol{y}^{s_{K+1}}(X_{ijk}^{(ijk)}-X_{ijk})^l\Bigg|\notag\\
        &\leq\sup_{|X_{ijk}|,|X_{ijk}^{(ijk)}|\leq N^{\delta}}\sum_{l=1}^{\infty}(l!)^{-1}\Bigg|(\boldsymbol{x}^{s_1})'\partial_{ijk}^{(l)}\Big\{\prod_{i=1}^K\boldsymbol{Q}^{s_is_{i+1}}(z_i,\boldsymbol{X})\Big\}\boldsymbol{y}^{s_{K+1}}(X_{ijk}^{(ijk)}-X_{ijk})^l\Bigg|\notag\\
        &\leq\sup_{|X_{ijk}|\leq N^{\delta}}\sum_{l=1}^{\infty}(l!)^{-1}2^lN^{\delta l}\Bigg|(\boldsymbol{x}^{s_1})'\partial_{ijk}^{(l)}\Big\{\prod_{i=1}^K\boldsymbol{Q}^{s_is_{i+1}}(z_i,\boldsymbol{X})\Big\}\boldsymbol{y}^{s_{K+1}}\Bigg|,\label{Eq of Taylor}
	\end{align}
    where we use the fact that all \(|X_{ijk}|,|X_{ijk}^{(ijk)}|\leq N^{\delta}\) in \eqref{Eq of Taylor}. Notice that \(\partial_{ijk}^{(l)}\big\{\prod_{i=1}^K\boldsymbol{Q}^{s_is_{i+1}}(z_i,\boldsymbol{X})\big\}\) in (\ref{Eq of Taylor}) does not involve \(\bbX^{ijk}\), we simplify \(\bbQ^{s_is_{i+1}}(z,\bbX)\) by \(\bbQ^{s_is_{i+1}}(z)\), as does their entries. Next, we separate the above Taylor expansion (\ref{Eq of Taylor}) into the following two cases, higher derivatives (\(l\geq4\)) and lower derivatives (\(l=1,2,3\)).

    \vspace{5mm}
\noindent
    {\bf Case 1:}  When \(l\geq4\), recall that \(\partial_{ijk}^{(l)}\boldsymbol{Q}(z)=(-1)^ll!(\boldsymbol{Q}(z)\partial_{ijk}^{(1)}\bbM)^l\boldsymbol{Q}(z)\) and \(\partial_{ijk}^{(1)}\bbM\) is defined in (\ref{Eq of partial N}), which implies that \(\Vert\partial_{ijk}^{(1)}\bbM\Vert\leq 3N^{-1/2}\) and
	\begin{align}
		\Vert\partial_{ijk}^{(l)}\boldsymbol{Q}(z)\Vert\leq l!\Vert\boldsymbol{Q}(z)\Vert^{l+1}(3N^{-1/2})^l\leq l!\eta_0^{-1}(3\eta_0^{-1}N^{-1/2})^l.\label{Eq of bound for partial Q}
	\end{align}
    by (\ref{Eq of bound for partial Q}), we have
		\begin{align}
			&(l!)^{-1}\Bigg|(\boldsymbol{x}^{s_1})'\partial_{ijk}^{(l)}\Big\{\prod_{i=1}^K\boldsymbol{Q}^{s_is_{i+1}}(z_i,\boldsymbol{X})\Big\}\boldsymbol{y}^{s_{K+1}}\Bigg|\leq(l!)^{-1}\Big\Vert\partial_{ijk}^{(l)}\Big\{\prod_{i=1}^K\boldsymbol{Q}^{s_is_{i+1}}(z_i)\Big\}\Big\Vert\notag\\
			&\leq(l!)^{-1}\sum_{l_1+\cdots+l_K=l}\binom{l}{l_1,\cdots,l_K}\prod_{i=1}^K\Vert\partial_{ijk}^{(l_i)}\boldsymbol{Q}^{s_is_{i+1}}(z_i)\Vert\leq\eta_0^{-K}\binom{l+K-1}{K-1}(3N^{-1/2}\eta_0^{-1})^l,\notag
		\end{align}
		where \(l_1,\cdots,l_K\) such that \(l_1+\cdots+l_K=l\) and \(l_i\geq0\) for \(i=1,\cdots,K\), so we have
		$$(l!)^{-1}2^lN^{\delta l}\Big|(\boldsymbol{x}^{s_1})'\partial_{ijk}^{(l)}\Big\{\prod_{i=1}^K\boldsymbol{Q}^{s_is_{i+1}}(z_i)\Big\}\boldsymbol{y}^{s_{d+1}}\Big|\leq \eta_0^{-K}\binom{l+K-1}{K-1}(6N^{-1/2+\delta}\eta_0^{-1})^l.$$
		Let \(q:=6N^{-1/2+\delta}\eta_0^{-1}\), it implies that
		\begin{align}
			&\sum_{l=4}^{\infty}\binom{l+K-1}{K-1}q^l=\frac{1}{(K-1)!}\sum_{l=4}^{\infty}\frac{\partial^{K-1}}{\partial q^{K-1}}q^{l+K-1}=\frac{1}{(K-1)!}\frac{\partial^{K-1}}{\partial q^{K-1}}[q^{3+K}(1-q)^{-1}]\notag\\
			&=\sum_{r=0}^{K-1}[r!(K-1-r)!]^{-1}\frac{\partial^r q^{3+K}}{\partial q^r}\frac{\partial^{(K-r-1)}(1-q)^{-1}}{\partial q^{(K-r-1)}}\leq C_Kq^4,\notag
		\end{align}
        so we conclude that
        \begin{align*}
            &\sup_{|X_{ijk}|\leq N^{\delta}}\Big|\sum_{l=4}^{\infty}(l!)^{-1}2^lN^{\delta l}(\boldsymbol{x}^{s_1})'\partial_{ijk}^{(l)}\Big\{\prod_{i=1}^K\boldsymbol{Q}^{s_is_{i+1}}(z_i)\Big\}\boldsymbol{y}^{s_{K+1}}\Big|\leq C_KN^{-2+4\delta}\eta_0^{-4-K},
        \end{align*}
		and
		\begin{align}
			\sum_{i,j,k=1}^{m,n,p}\sup_{|X_{ijk}|\leq N^{\delta}}\Big|\sum_{l=4}^{\infty}(l!)^{-1}(\boldsymbol{x}^{s_1})'\partial_{ijk}^{(l)}\Big\{\prod_{i=1}^K\boldsymbol{Q}^{s_is_{i+1}}(z_i)\Big\}\boldsymbol{y}^{s_{K+1}}\Big|^2\leq C_K\eta_0^{-2(K+4)}N^{-1+8\delta}.\label{Eq of Taylor 4}
		\end{align}

    \vspace{5mm}
    \noindent
    {\bf Case 2:} When \(l=3\), by Lemma \ref{Lem of minor terms 1}, we have shown that
		$$\sum_{i,j,k=1}^{m,n,p}\big|(\boldsymbol{x}^{s_1})'\partial_{ijk}^{(3)}\Big\{\prod_{i=1}^K\boldsymbol{Q}^{s_is_{i+1}}(z_i)\Big\}\boldsymbol{y}^{s_{K+1}}\big|^2<C_l\eta_0^{-2(l+K)}N^{-2},$$
		so
		\begin{align}
			\sum_{i,j,k=1}^{m,n,p}\Big|(\boldsymbol{x}^{s_1})'\partial_{ijk}^{(3)}\Big\{\prod_{i=1}^K\boldsymbol{Q}^{s_is_{i+1}}(z_i)\Big\}\boldsymbol{y}^{s_{K+1}}(X_{ijk}^{(ijk)}-X_{ijk})^3\Big|^2\leq C_K\eta_0^{-2(K+3)}N^{-2+6\delta}.\label{Eq of Taylor 3}
		\end{align}
		Similarly, for \(l=1,2\), we have
		\begin{align}
			\sum_{i,j,k=1}^{m,n,p}\Big|(\boldsymbol{x}^{s_1})'\partial_{ijk}^{(l)}\Big\{\prod_{i=1}^K\boldsymbol{Q}^{s_is_{i+1}}(z_i)\Big\}\boldsymbol{y}^{s_{K+1}}(X_{ijk}^{(ijk)}-X_{ijk})^l\Big|^2\leq C_K\eta_0^{-2(K+l)}N^{-1+2l\delta}.\label{Eq of Taylor 12}
		\end{align}
	Finally, by Cauchy's inequality    and (\ref{Eq of Taylor 12}), (\ref{Eq of Taylor 3}) and (\ref{Eq of Taylor 4}), we have
	\begin{align}
		&\sum_{i,j,k=1}^{m,n,p}\Delta_{ijk}^2\leq 4\sum_{l=1}^3\sum_{i,j,k=1}^{m,n,p}\sup_{|X_{ijk}|\leq N^{\delta}}\Big|(l!)^{-1}2^lN^{\delta l}(\boldsymbol{x}^{s_1})'\partial_{ijk}^{(l)}\Big\{\prod_{i=1}^K\boldsymbol{Q}^{s_is_{i+1}}(z_i,\boldsymbol{X})\Big\}\boldsymbol{y}^{s_{K+1}}\Big|^2\notag\\
		&+4\sum_{i,j,k=1}^{m,n,p}\sup_{|X_{ijk}|\leq N^{\delta}}\Big|\sum_{l=4}^{\infty}(l!)^{-1}2^lN^{\delta l}(\boldsymbol{x}^{s_1})'\partial_{ijk}^{(l)}\Big\{\prod_{i=1}^K\boldsymbol{Q}^{s_is_{i+1}}(z_i,\boldsymbol{X})\Big\}\boldsymbol{y}^{s_{K+1}}\Big|^2\notag\\
		&\leq C_K\eta_0^{-2(K+4)}N^{-1+8\delta}.\notag
	\end{align}
	Thus, we have
	\begin{align}
		&\eqref{Eq of a.s. convergence major probability}\leq4\exp\left(-\frac{t^2}{\sum_{i,j,k=1}^{m,n,p}\Delta_{ijk}^2}\right)\leq4\exp\left(-C_K\eta_0^{2(K+4)}N^{1-8\delta}t^2\right),\notag
	\end{align}
    combined with \eqref{Eq of a.s. convergence major probability}, \eqref{Eq of a.s. convergence tail probability} and \eqref{Eq of a.s. convergence tail probability 1}, we conclude that
    \begin{align*}
        &\mathbb{P}\left(\Big|\boldsymbol{x}^{s_1}\prod_{i=1}^K\boldsymbol{Q}^{s_is_{i+1}}(z_i)\boldsymbol{y}^{s_{K+1}}-\mathbb{E}\Big[\boldsymbol{x}^{s_1}\prod_{i=1}^K\boldsymbol{Q}^{s_is_{i+1}}(z_i)\boldsymbol{y}^{s_{K+1}}\Big]\Big|\geq t\right)\\
        &\leq4\exp\left(-C_K\eta_0^{2(K+4)}N^{1-8\delta}t^2\right)+ N^3\exp(-N^{\delta\theta}),
    \end{align*}
	then choose any \(t=\eta_0^{-(K+4)}N^{-1/2+4\delta+\epsilon}\), where \(\epsilon\in(0,1/2-4\delta)\) is a sufficiently small positive number, we can show the almost sure convergence for (\ref{Eq of as 1}) by the Borel–Cantelli lemma. Furthermore, for (\ref{Eq of as 2}), since
	$$\boldsymbol{x}'{\rm diag}\Big\{\prod_{i=1}^K\boldsymbol{Q}(z_i)\Big\}\boldsymbol{y}=\sum_{s_1\cdots s_K}\boldsymbol{x}^{s_1}{\rm diag}\Big\{\prod_{i=1}^K\boldsymbol{Q}^{s_is_{i+1}}(z_i)\Big\}\boldsymbol{y}^{s_{K+1}},$$
	where \(s_1=s_{K+1},s_i\in\{1,2,3\}\) for \(i=1,\cdots,K\). We only need to show
	$$\Big|(\boldsymbol{x}^{s_1})'{\rm diag}\Big\{\prod_{i=1}^K\boldsymbol{Q}^{s_is_{i+1}}(z_i)\Big\}\boldsymbol{y}^{s_{K+1}}-\mathbb{E}\Big[(\boldsymbol{x}^{s_1})'{\rm diag}\Big\{\prod_{i=1}^K\boldsymbol{Q}^{s_is_{i+1}}(z_i)\Big\}\boldsymbol{y}^{s_{K+1}}\Big]\Big|\prec C_K\eta_0^{-(4+K)}N^{-\omega}.$$
	The proof follows identically as those for (\ref{Eq of as 1}), i.e. separate the following Taylor expansion
	\begin{align}
		&(\boldsymbol{x}^{s_1})'{\rm diag}\Big\{\prod_{i=1}^K\boldsymbol{Q}^{s_is_{i+1}}(z_i,\boldsymbol{X}^{ijk})\Big\}\boldsymbol{y}^{s_{K+1}}-(\boldsymbol{x}^{s_1})'{\rm diag}\Big\{\prod_{i=1}^K\boldsymbol{Q}^{s_is_{i+1}}(z_i,\boldsymbol{X})\Big\}\boldsymbol{y}^{s_{K+1}}\notag\\
		&=\sum_{l=1}^{\infty}(l!)^{-1}(\boldsymbol{x}^{s_1})'{\rm diag}\Big\{\partial_{ijk}^{(l)}\prod_{i=1}^K\boldsymbol{Q}^{s_is_{i+1}}(z_i,\boldsymbol{X})\Big\}\boldsymbol{y}^{s_{K+1}}(X_{ijk}^{ijk}-X_{ijk})^l\notag
	\end{align}
	into \(l\geq4\) and \(l=1,2,3\), and we can obtain the same conclusion as (\ref{Eq of as 1}), so the details are omitted for brevity.
\end{proof}
\subsubsection{Systematic treatment for minor terms in cumulant expansions}\label{ssec of major terms}
As we have mentioned before, to derive the asymptotic mean of the LSS of the matrix \(\bbM\), we need to compute \(N^{-1/2}\sum_{i,j,k=1}^{m,n,p}\mbE[\partial_{ijk}^{(l)}(c_kQ_{ij}^{12}+b_jQ_{ik}^{13})]\) for \(l=2,3,4\). By \eqref{Eq of partial Q}, we know that
\begin{align*}
    c_k\partial_{ijk}^{(l)}Q_{ij}^{12}=(-N^{-1/2})^ll!c_k\sum_{t_1\cdots t_{2l}}Q_{i\tilde{t}_1}^{1t_1}\left(\prod_{\alpha=1}^{l-1}\mcA_{ijk}^{(t_{2\alpha-1},t_{2\alpha})}Q_{\tilde{t}_{2\alpha}\tilde{t}_{2\alpha+1}}^{t_{2\alpha}t_{2\alpha+1}}\right)\mcA_{ijk}^{(t_{2l-1},t_{2l})}Q_{\tilde{t}_{2l}j}^{t_{2l}2},
\end{align*}
where \(t_1\cdots t_{2l}\in\{1,2,3\}\) such that \(t_{2\alpha-1}\neq t_{2\alpha}\) for \(\alpha=1,\cdots,l\). To compute \(\sum_{i,j,k=1}^{m,n,p}c_k\partial_{ijk}^{(l)}Q_{ij}^{12}\), it is essential to determine which terms in the above equation vanish as \(N\to\infty\). The following lemma will provide a criterion for distinguishing major and minor terms.
\begin{lem}\label{Lem of minor terms}
	For any \(z\in\mbC^+\) and \(l\in\mbN^+\), let \(s_1,\cdots,s_{2(l+1)}\in\{1,2,3\}\) such that \(s_{2\alpha}\neq s_{2\alpha+1}\) and \(s_1\neq s_{2(l+1)}\) for \(1\leq \alpha\leq l\), consider the following two equations:
	\begin{align}
		\left\{\begin{array}{l}
			\sum_{i,j,k=1}^{m,n,p}\mcA_{ijk}^{(s_1,s_{2l+2})}Q_{\tilde{s}_1\tilde{s}_2}^{s_1s_2}(z)\left(\prod_{\alpha=1}^l\mcA_{ijk}^{(s_{2\alpha},s_{2\alpha+1})}Q_{\tilde{s}_{2\alpha+1}\tilde{s}_{2\alpha+2}}^{s_{2\alpha+1}s_{2\alpha+2}}(z)\right),\\
			\sum_{i,j,k=1}^{m,n,p}\mcA_{ijk}^{(s_{2l},s_{2l+1})}Q_{\tilde{s}_{2l+1}\cdot}^{s_{2l+1}s_1}(z)Q_{\cdot\tilde{s}_2}^{s_1s_2}(z)\left(\prod_{\alpha=1}^{l-1}\mcA_{ijk}^{(s_{2\alpha},s_{2\alpha+1})}Q_{\tilde{s}_{2\alpha+1}\tilde{s}_{2\alpha+2}}^{s_{2\alpha+1}s_{2\alpha+2}}(z)\right),
		\end{array}\right.\label{Eq of minor terms}
	\end{align}
	where \(Q_{i\cdot},Q_{\cdot i}\) means the \(i\)-th row and column of \(\bbQ\), \(\mcA_{ijk}^{(s_{2l},s_{2l+1})}\) is defined in {\rm (\ref{Eq of nij})}. If there is at least one term in
	$$\left\{Q_{\tilde{s}_{2\alpha-1}\tilde{s}_{2\alpha}}^{s_{2\alpha-1}s_{2\alpha}}:\alpha=1,\cdots,l+1\right\}\ {\rm or\ }\left\{Q_{\tilde{s}_{2\alpha-1}\tilde{s}_{2\alpha}}^{s_{2\alpha-1}s_{2\alpha}},Q_{\tilde{s}_{2l+1}\cdot}^{s_{2l+1}s_1}Q_{\cdot\tilde{s}_2}^{s_1s_2}:\alpha=2,\cdots,l\right\}$$
    coming from the off-diagonal block, then the norms of {\rm (\ref{Eq of minor terms})} are bounded by \(\mrO(\Vert\boldsymbol{Q}\Vert^{l+1} N)\).
\end{lem}
In particular, we say $Q_{\tilde{s}_{2l+1}\cdot}^{s_{2l+1}s_1}Q_{\cdot\tilde{s}_2}^{s_1s_2}$ in the second equation of \eqref{Eq of minor terms} comes from the off-diagonal blocks if \(s_{2l+1}\neq s_2\); otherwise, it comes from the diagonal block. For the second equation in \eqref{Eq of minor terms}, it appears in \(\partial_{ijk}^{(l)}\tr(\bbQ)\). For example,
\begin{align*}
    \partial_{ijk}^{(1)}\tr(\bbQ^{11})=\sum_{i=1}^m\partial_{ijk}^{(1)}Q_{ii}^{11}=-N^{-1/2}\sum_{i=1}^m\sum_{t_1,t_2}Q_{i\tilde{t}_1}^{1t_1}\mcA_{ijk}^{(t_1,t_2)}Q_{\tilde{t}_2i}^{t_21}=-N^{-1/2}\sum_{t_1,t_2}\mcA_{ijk}^{(t_1,t_2)}Q_{\tilde{t}_2\cdot}^{t_21}Q_{\cdot\tilde{t}_1}^{1t_1},
\end{align*}
where \(t_1,t_2\in\{1,2,3\}\) and \(t_1\neq t_2\).
\begin{proof}[Proof of Lemma \ref{Lem of minor terms}]
	In fact, the two cases in (\ref{Eq of minor terms}) essentially coincide. Since \(Q_{\tilde{s}_{2l+1}\cdot}^{s_{2l+1}s_1}Q_{\cdot\tilde{s}_2}^{s_1s_2}\) is just the \((\tilde{s}_{2l+1},\tilde{s}_2)\)-th entry of \(\bbQ^{s_{2l+1}s_1}\bbQ^{s_1s_2}\), whose spectral norm is bounded by \(\Vert\bbQ\Vert^2\). Therefore, we only consider the first case in (\ref{Eq of minor terms}) and rewrite it into the following form
	\begin{align}
		\sum_{i,j,k=1}^{m,n,p}(a_i)^{n_a}(b_j)^{n_b}(c_k)^{n_c}(Q_{ii}^{11})^{n_{11}}(Q_{jj}^{22})^{n_{22}}(Q_{kk}^{33})^{n_{33}}(Q_{ij}^{12})^{n_{12}}(Q_{ik}^{13})^{n_{13}}(Q_{jk}^{23})^{n_{23}},\label{Eq of minor terms 1}
	\end{align}
	where \(n_a\) is the number of \(a_i\) appearing in $\big\{\mcA_{ijk}^{(s_2,s_3)},\cdots,\mcA_{ijk}^{(s_{2l},s_{2l+1})},\mcA_{ijk}^{(s_{2l+2},s_1)}\big\}$, as does \(n_b\) and \(n_c\). Similarly, \(n_{12}\) is the number of \(Q_{ij}^{12}\) appearing in $\big\{Q_{\tilde{s}_{2\alpha-1}\tilde{s}_{2\alpha}}^{s_{2\alpha-1}s_{2\alpha}}:\alpha=1,\cdots,l+1\big\}$, so dose $n_{ij}$ for $1\leq i\leq j\leq3$. By definitions, we have \(n_a+n_b+n_c=l+1\) and \(\sum_{1\leq i\leq j\leq3}n_{ij}=l+1\). Next, based on the number of nonzero terms in \(\{n_{12},n_{13},n_{23}\}\), consider the following three situations.

    \vspace{5mm}
    \noindent
    {\bf Case 1:} Suppose all \(n_{12},n_{13},n_{23}\) are nonzero, then we claim that at least two of \(n_a,n_b,n_c\) are nonzero. Otherwise, if \(n_a=n_b=0\) without loss of generality, then all \(Q_{\tilde{s}_{2\alpha-1}\tilde{s}_{2\alpha}}^{s_{2\alpha-1}s_{2\alpha}}\) come from block \(\boldsymbol{Q}^{12}\), \(\boldsymbol{Q}^{11}\) or \(\boldsymbol{Q}^{22}\), which is a contradiction. Therefore, suppose \(n_a,n_b\geq1\) without loss of generality, then the norm of (\ref{Eq of minor terms 1}) is bounded by 
		$$(|\boldsymbol{a}|^{\circ n_a})'{\rm diag}\left(|\boldsymbol{Q}^{11}|^{\circ n_{11}}\right)\left(|\boldsymbol{Q}^{12}|^{\circ n_{12}}\circ\left(|\boldsymbol{Q}^{13}|^{\circ n_{13}}{\rm diag}\left(|\boldsymbol{Q}^{33}|^{\circ n_{33}}\right)|\boldsymbol{Q}^{32}|^{\circ n_{23}}\right)\right){\rm diag}\left(|\boldsymbol{Q}^{22}|^{\circ n_{22}}\right)|\boldsymbol{b}|^{\circ n_b},$$
		which is smaller than \(\Vert\boldsymbol{Q}\Vert^{l+1}\).

    \vspace{5mm}
    \noindent
    {\bf Case 2:} If only two of \(\{n_{12},n_{13},n_{23}\}\) are nonzero, without loss of generality, suppose \(n_{23}=0\) and \(n_{12},n_{13}>0\), then at least two of \(n_a,n_b,n_c\) are nonzero; otherwise, as the arguments in Case 1, there will only be one type off-diagonal block. So we assume \(n_a,n_b>0\), and the norm of (\ref{Eq of minor terms}) is bounded by
		$$(|\bbb|^{\circ n_b})'{\rm diag}(|\boldsymbol{Q}^{22}|^{\circ n_{22}})|\boldsymbol{Q}^{21}|^{\circ n_{12}}{\rm diag}(|\bba|^{\circ n_a}){\rm diag}(|\boldsymbol{Q}^{11}|^{\circ n_{11}})|\boldsymbol{Q}^{13}|^{\circ n_{13}}{\rm diag}(|\boldsymbol{Q}^{33}|^{\circ n_{33}})\boldsymbol{1}_k,$$
		which is smaller than \(\Vert\boldsymbol{Q}\Vert^{l+1}N^{1/2}\).

    \vspace{5mm}
    \noindent
    {\bf Case 3:} If there is only one term in (\ref{Eq of minor terms}) coming from the off-diagonal block, suppose \(n_{12}>0,n_{13}=n_{23}=0\), then if \(n_a,n_c>0\) or \(n_b,n_c>0\), the norm of (\ref{Eq of minor terms}) is bounded by (e.g. \(n_a,n_c>0\))
		$$(|\boldsymbol{a}|^{\circ n_a})'{\rm diag}(|\boldsymbol{Q}^{11}|^{\circ n_{11}})|\boldsymbol{Q}^{12}|^{\circ n_{12}}{\rm diag}(|\boldsymbol{Q}^{22}|^{\circ n_{22}})\boldsymbol{1}_n\times\boldsymbol{1}_p'|\boldsymbol{Q}^{33}|^{\circ n_{33}}|\boldsymbol{c}|^{\circ n_c},$$
		which is smaller than \(\Vert\boldsymbol{Q}\Vert^{l+1}N\). Otherwise, if \(n_c=0\), i.e. \(n_a,n_b>0\); the norm of (\ref{Eq of minor terms}) is bounded by
		$$(|\boldsymbol{a}|^{\circ n_a})'{\rm diag}(|\boldsymbol{Q}^{11}|^{\circ n_{11}})|\boldsymbol{Q}^{12}|^{\circ n_{12}}{\rm diag}(|\boldsymbol{Q}^{22}|^{\circ n_{22}})|\boldsymbol{b}|\times\tr|\boldsymbol{Q}^{33}|^{\circ n_{33}},$$
		which is smaller than \(\Vert\boldsymbol{Q}\Vert^{l+1}N\). Finally, if only one of \(\{n_a,n_b,n_c\}\) is nonzero, the only possible case is that \(n_c>0,n_a=n_b=0\). Otherwise, if \(n_a>0,n_b=n_c=0\), then all \(\mcA^{(t_{2\alpha},t_{2\alpha+1})},\mcA^{(s_{2\gamma},s_{2\gamma+1})}=a_i\), it implies that \((t_{2\alpha},t_{2\alpha+1})=(2,3)\) or \((3,2)\), as does \((s_{2\gamma},s_{2\gamma+1})\). Hence, the only possible off-diagonal block is \(Q_{jk}^{23}\), i.e. \(n_{23}>0\), which is a contradiction. Now, since \(n_c>0,n_a=n_b=0\), then \(n_{11},n_{22}\geq0\) and \(n_{33}=0\). If \(n_{11},n_{22}>0\), (\ref{Eq of minor terms}) is bounded by
		$$\boldsymbol{1}_m'{\rm diag}(|\boldsymbol{Q}^{11}|^{\circ n_{11}})|\boldsymbol{Q}^{12}|^{\circ n_{12}}{\rm diag}(|\boldsymbol{Q}^{22}|^{\circ n_{22}})\boldsymbol{1}_n\leq \Vert\boldsymbol{Q}\Vert^{l+1}N.$$
	   This completes the proof of Lemma \ref{Lem of minor terms}.
\end{proof}
For simplicity, we define two operators \(\mathscr{D},\mathscr{O}\) as follows:
\begin{align}
	&\mathscr{D}\left(\partial_{ijk}^{(l)}Q_{i_1i_2}^{j_1j_2}\right):=(-1)^ll!N^{-l/2}\sum_{\substack{t_1\cdots t_l\\t_{l+1}=j_2}}Q_{i_1i_{t_2}}^{j_1j_1}\mcA_{ijk}^{(j_1,t_1)}\Big(\prod_{\alpha=1}^{l-1}Q_{i_{t_{\alpha}}i_{t_{\alpha}}}^{t_{\alpha}t_{\alpha}}\mcA_{ijk}^{(t_{\alpha},t_{\alpha+1})}\Big)Q_{i_{j_2}i_2}^{j_2j_2},\label{Eq of operator D}\\
	&\mathscr{O}\left(\partial_{ijk}^{(l)}Q_{i_1i_2}^{j_1j_2}\right):=\partial_{ijk}^{(l)}Q_{i_1i_2}^{j_1j_2}-\mathscr{D}\left(\partial_{ijk}^{(l)}Q_{i_1i_2}^{j_1j_2}\right),\label{Eq of operator O}
\end{align}
The operator \(\msD\) selects the summation terms in \(\partial_{ijk}^{(l)}Q_{i_1i_2}^{j_1j_2}\) that only contains diagonal terms. According to Lemma \ref{Lem of minor terms}, when \(l\geq2\), for any \(z\in\mbC_{\eta_0}^+\), we can conclude that
$$N^{-1/2}\sum_{i,j,k=1}^{m,n,p}\Big|\mathbb{E}\left[\mathscr{O}\left(c_kQ_{ij}^{12}+b_jQ_{ik}^{13}\right)\right]\Big|\leq\mrO(\eta_0^{-(l+1)}N^{-(l-1)/2}),$$
thus, the major terms will only appear in \(N^{-1/2}\sum_{i,j,k=1}^{m,n,p}\mathbb{E}\big[\mathscr{D}\big(c_kQ_{ij}^{12}+b_jQ_{ik}^{13}\big)\big]\).

Finally, when calculating the asymptotic variance of the LSS of the matrix \(\bbM\), we need to compute
$$\sum_{i,j,k=1}^{m,n,p}\mbE\left[\partial_{ijk}^{(\alpha)}(c_kQ_{ij}^{12}+b_jQ_{ik}^{13})\partial_{ijk}^{(l-\alpha)}\{\tr(\bbQ^{11})\}\right],\quad l=2,3,4;\alpha=0,1,\cdots,l,$$
see \eqref{Eq of covariance 0} later for an example. To further determine the major terms in the above equation, we need the following result.
\begin{lem}\label{Cor of minor terms}
	For any \(z\in\mathbb{C}_{\eta}^+\) and \(l_1,l_2\in\mbN\) such that \(l_1+l_2\geq2\), let \(t_{\alpha},s_{\gamma}\in\{1,2,3\}\) such that \(t_{2\alpha}\neq t_{2\alpha+1},s_{2\gamma}\neq s_{2\gamma+1}\) for \(1\leq \alpha\leq l_1,1\leq\gamma\leq l_2\) and \(t_1\neq t_{2l_1+2},s_1\neq s_{2l_2+2}\), then define
	$$\left\{\begin{array}{l}
		P_1(z):=\mcA_{ijk}^{(t_1,t_{2l_1+2})}Q_{i_{t_1}i_{t_2}}^{t_1t_2}(z)\left(\prod_{\alpha=1}^{l_1}\mcA_{ijk}^{(t_{2\alpha},t_{2\alpha+1})}Q_{i_{t_{2\alpha+1}}i_{t_{2\alpha+2}}}^{t_{2\alpha+1}t_{2\alpha+2}}(z)\right)\\
		P_2(z):=\mcA_{ijk}^{(s_1,s_{2l_2+2})}Q_{i_{s_1}i_{s_2}}^{s_1s_2}(z)\left(\prod_{\gamma=1}^{l_2}\mcA_{ijk}^{(s_{2\gamma},s_{2\gamma+1})}Q_{i_{s_{2\gamma+1}}i_{s_{2\gamma+2}}}^{s_{2\gamma+1}s_{2\gamma+2}}(z)\right)
	\end{array}\right..$$
	If there are at least one term in
	$$\left\{Q_{i_{t_{2\alpha+1}}i_{t_{2\alpha+2}}}^{t_{2\alpha+1}t_{2\alpha+2}}(z):\alpha=1,\cdots,l_1+1\right\}\quad{\rm or}\quad\left\{Q_{i_{s_{2\gamma+1}}i_{s_{2\gamma+2}}}^{s_{2\gamma+1}s_{2\gamma+2}}(z):\gamma=1,\cdots,l_2+1\right\}$$
    coming from the off-diagonal block, then the norm of \(\sum_{i,j,k=1}^{m,n,p}P_1(z)P_2(z)\) is bounded by \(\mrO(\Vert\bbQ\Vert^{l_1+l_2+2}N)\).
\end{lem}
Before proving the above lemma, we need one preliminary result. By Lemma \ref{Lem of minor terms}, if there exists at least one off-diagonal term in (\ref{Eq of minor terms}), then
\begin{small}
\begin{align}
		\left\{\begin{array}{ll}
			\sum_{i,j,k=1}^{m,n,p}\left|\mcA_{ijk}^{(s_1,s_{2l+2})}Q_{\tilde{s}_1\tilde{s}_2}^{s_1s_2}(z)\left(\prod_{\alpha=1}^l\mcA_{ijk}^{(s_{2\alpha},s_{2\alpha+1})}Q_{\tilde{s}_{2\alpha+1}\tilde{s}_{2\alpha+2}}^{s_{2\alpha+1}s_{2\alpha+2}}(z)\right)\right|^2\leq\Vert\bbQ\Vert^{2(l+1)}N,&l\geq1,\\
			\sum_{i,j,k=1}^{m,n,p}\left|\mcA_{ijk}^{(s_{2l},s_{2l+1})}Q_{\tilde{s}_{2l+1}\cdot}^{s_{2l+1}s_1}(z)Q_{\cdot\tilde{s}_2}^{s_1s_2}(z)\left(\prod_{\alpha=1}^{l-1}\mcA_{ijk}^{(s_{2\alpha},s_{2\alpha+1})}Q_{\tilde{s}_{2\alpha+1}\tilde{s}_{2\alpha+2}}^{s_{2\alpha+1}s_{2\alpha+2}}(z)\right)\right|^2\leq\Vert\bbQ\Vert^{2(l+1)}N,&l\geq2.
		\end{array}\right.\label{Rem of minor terms}
	\end{align}
\end{small}\noindent
The proofs of the above two inequalities are the same as those in {\rm Lemma \ref{Lem of minor terms}}, since we can rewrite them into the following forms:
	$$\sum_{i,j,k=1}^{m,n,p}(a_i)^{2n_a}(b_j)^{2n_b}(c_k)^{2n_c}|Q_{ii}^{11}|^{2n_{11}}|Q_{jj}^{22}|^{2n_{22}}|Q_{kk}^{33}|^{2n_{33}}|Q_{ij}^{12}|^{2n_{12}}|Q_{ik}^{13}|^{2n_{13}}|Q_{jk}^{23}|^{2n_{23}}.$$
	For example, when \(n_{12}>0,n_{13}=n_{23}=0\) and \(n_c>0,n_a=n_b=0\), we can show that the above term equals 
	$$\boldsymbol{1}_n'|\bbQ^{21}|^{\circ n_{12}}{\rm diag}(|\bbQ^{11}|^{\circ2n_{11}})|\bbQ^{12}|^{\circ n_{12}}\boldsymbol{1}_n\leq\Vert\boldsymbol{Q}\Vert^{2(l+1)}N.$$
    We omit the details here.
\begin{proof}[Proof of Lemma \ref{Cor of minor terms}]
	First, if both \(P_1(z)\) and \(P_2(z)\) contain off-diagonal terms, by Cauchy's inequality   and (\ref{Rem of minor terms}), we have
	$$\Big|\sum_{i,j,k=1}^{m,n,p}P_1(z)P_2(z)\Big|^2\leq\sum_{i,j,k=1}^{m,n,p}|P_1(z)|^2\times\sum_{i,j,k=1}^{m,n,p}|P_1(z)|^2\leq C\Vert\bbQ\Vert^{2(l_1+l_2+2)}N^2.$$
	Therefore, we only need to consider the case only \(P_1(z)\) contains off-diagonal terms. Similarly to Lemma \ref{Lem of minor terms}, we can rewrite \(P_1(z)P_2(z)\) as the following form:
	\begin{align}
		\sum_{i,j,k=1}^{m,n,p}(a_i)^{n_a}(b_j)^{n_b}(c_k)^{n_c}(Q_{ii}^{11})^{n_{11}}(Q_{jj}^{22})^{n_{22}}(Q_{kk}^{33})^{n_{33}}(Q_{ij}^{12})^{n_{12}}(Q_{ik}^{13})^{n_{13}}(Q_{jk}^{23})^{n_{23}},\notag
	\end{align}
	where \(n_a,n_b,n_c,n_{ij}\in\mathbb{N}\), \(n_a+n_b+n_c=\sum_{1\leq i\leq j\leq3}n_{ij}=l_1+l_2+2\). Similar to proofs of Lemma \ref{Lem of minor terms},  consider three situations based on the number of nonzero terms in \(\{n_{12},n_{23},n_{13}\}\). Actually, we can repeat the proofs of Lemma \ref{Lem of minor terms} for Case 1, 2 and 3 to derive the same conclusion, so we omit the details here. 
\end{proof}
\subsection{Proof of Theorem \ref{Thm of entrywise law d=3}}\label{sec of proof entrywise law d=3}
In this section, we will prove the entrywise law. First, let's  show that \(\lim_{N\to\infty}|N^{-1}\mbE[\tr(\bbQ(z))]-g(z)|=0\), where \(g(z)\) is defined in \eqref{Eq of g(z)}.
\begin{thm}\label{Thm of approximation}
	Under Assumptions {\rm \ref{Ap of general noise}} and {\rm \ref{Ap of dimension}}, for any \(z\in\mathcal{S}_{\eta_0}\) in {\rm (\ref{Eq of stability region})} and \(\omega\in(1/2-\delta,1/2)\), where $\delta>0$ is a sufficiently small number, let 
    $$\bbve(z)=\frac{\mfc}{\bbm(z)}+z+\bbS_3\bbm(z),$$
    where \(\bbS_d\) and \(\mfm_i(z)\) are defined in {\rm (\ref{Eq of bbS d})} and {\rm (\ref{Eq of mi})}, then we have 
    $$\sup_{z\in\mcS_{\eta_0}}\Vert\bbve(z)\Vert_{\infty}=\mrO(\eta_0^{-11}N^{-2\omega}).$$
    Consequently, by Theorem {\rm \ref{Thm of Stability}}, we obtain
	$$\sup_{z\in\mcS_{\eta_0}}\Vert\bbg(z)-\bbm(z)\Vert_{\infty}=\mrO(\eta_0^{-15}N^{-2\omega}).$$
\end{thm}
\begin{proof}
	Without loss of generality, we only prove 
	\begin{align}
		-\frac{\mfc_1}{\mfm_1(z)}=z+\mfm_2(z)+\mfm_3(z)+\mrO(\eta_0^{-11}N^{-2\omega}),\label{Eq of approximation 1}
	\end{align}
	since the same arguments apply to the other cases. By \(\bbM\bbQ(z)-z\bbQ(z)=\bbI_N\) and the cumulant expansion (\ref{Eq of cumulant expansion}), we have
	\begin{align}
		&zN^{-1}\mbE[\tr(\bbQ^{11}(z))]=N^{-3/2}\sum_{i,j,k=1}^{m,n,p}\mbE[X_{ijk}(c_kQ_{ij}^{12}(z)+b_jQ_{ik}^{13}(z))]-\mfc_1\notag\\
		&=N^{-3/2}\sum_{i,j,k=1}^{m,n,p}\left(\mbE\big[\partial_{ijk}^{(1)}(c_kQ_{ij}^{12}(z)+b_jQ_{ik}^{13}(z))\big]+\epsilon_{ijk}^{(2)}\right)-\mfc_1,\notag
	\end{align}
    where the remainder \(\epsilon_{ijk}^{(2)}\) satisfies that
    $$|\epsilon_{ijk}^{(2)}|\leq C_{\kappa_3}\sup_{z\in\mcS_{\eta_0}}\big|\partial_{ijk}^{(2)}(c_kQ_{ij}^{12}(z)+b_jQ_{ik}^{13}(z))\big|.$$
	Let us  first show that \(\sum_{i,j,k=1}^{m,n,p}|\epsilon_{ijk}^{(2)}|\) is a minor term. By the definition of \(\mathscr{O}\) in (\ref{Eq of operator O}) and Lemma \ref{Lem of minor terms}, we know that
	$$N^{-3/2}\sum_{i,j,k=1}^{m,n,p}\Big|c_k\mathscr{O}\big(\partial_{ijk}^{(2)}Q_{ij}^{12}(z)\big)\Big|\leq\mrO(\eta_0^{-3}N^{-3/2}).$$
	On the other hand, based on the definition of \(\mathscr{D}\), we have
	$$N^{-3/2}\sum_{i,j,k=1}^{m,n,p}\Big|c_k\mathscr{D}\big(\partial_{ijk}^{(2)}Q_{ij}^{12}(z)\big)\Big|=N^{-5/2}\sum_{i,j,k=1}^{m,n,p}|a_ib_jc_kQ_{ii}^{11}(z)Q_{jj}^{22}(z)Q_{kk}^{33}(z)|\leq\mrO(\eta_0^{-3}N^{-1}).$$
	Next, by direct calculation, we have
	\begin{align}
		&N^{-3/2}\sum_{i,j,k=1}^{m,n,p}\mbE\big[\partial_{ijk}^{(1)}(c_kQ_{ij}^{12}(z)+b_jQ_{ik}^{13}(z))\big]=-N^{-2}\mbE[\tr(\bbQ^{11}(z))\tr(\bbQ^{22}(z)+\bbQ^{33}(z))]\notag\\
		&-N^{-2}\mbE[\tr(\bbQ^{12}(z)\bbQ^{21}(z)+\bbQ^{13}(z)\bbQ^{31}(z))+2\tr(\bbQ^{11}(z))\bbb'\bbQ^{23}(z)\bbc+2\bbb'\bbQ^{21}(z)\bbQ^{13}(z)\bbc]\notag\\
		&-N^{-2}\mbE[\bba'\bbQ^{13}(z)\bbc\tr(\bbQ^{22}(z))+\bba'\bbQ^{12}(z)\bbb\tr(\bbQ^{33}(z))+\bba'\bbQ^{12}(z)\bbQ^{23}(z)\bbc+\bba'\bbQ^{13}(z)\bbQ^{32}(z)\bbb]\notag\\
		&=-\mbE[\rho_1(z)(\rho_2(z)+\rho_3(z))]+\mrO(\eta_0^{-2}N^{-1}),\notag
	\end{align}
	where \(\rho_i(z)=N^{-1}\tr(\bbQ^{ii}(z))\). Next, by Lemma \ref{Thm of Entrywise almost sure convergence}, we know that \(N^{-1}|\tr(\bbQ^{ii}(z))^c|\prec\mrO(\eta_0^{-5}N^{-\omega})\), combined with the fact that \(|\rho_i(z)|\leq\eta_0^{-1}\), it yields that
	$$|\operatorname{Cov}(\rho_1(z),\rho_2(z))|\leq\mbE[|\rho_1(z)^c||\rho_2(z)^c|]\leq\eta_0^{-10}N^{-2\omega}+\eta_0^{-2}\exp(-CN^{1-2\omega})=\mrO(\eta_0^{-10}N^{-2\omega}),$$
	so
	$$\mfc_1+\mfm_1(z)(z+\mfm_2(z)+\mfm_3(z))=\mrO(\eta_0^{-10}N^{-2\omega}).$$
	By the definition of \(\mcS_{\eta_0}\) in (\ref{Eq of stability region}), we know that \(|\mfm_1(z)|\geq c(\eta_0)\), for some constant $c>0$, so we can conclude (\ref{Eq of approximation 1}) by dividing by \(\mfm_1(z)\).
\end{proof}
Finally, let us  prove the entrywise law for \(d=3\) as follows:
\begin{proof}[Proof of Theorem \ref{Thm of entrywise law d=3}]
	First, let us  focus on the diagonal terms, e.g. \(Q_{il}^{11}(z)\), by cumulant expansion (\ref{Eq of cumulant expansion}), we have
	\begin{align}
		\mathbb{E}[Q_{il}^{11}(z)]&=\frac{z^{-1}}{\sqrt{N}}\sum_{j,k=1}^{n,p}\mathbb{E}[X_{ijk}\left(c_kQ_{lj}^{12}(z)+b_jQ_{lk}^{13}(z)\right)]-\delta_{il}z^{-1}\notag\\
		&=\frac{z^{-1}}{\sqrt{N}}\sum_{j,k=1}^{n,p}\left(\mathbb{E}\left[c_k\partial_{ijk}^{(1)}Q_{lj}^{12}(z)+b_j\partial_{ijk}^{(1)}Q_{lk}^{13}(z)\right]+\epsilon_{ijk}^{(2)}\right)-\delta_{il}z^{-1},\notag
	\end{align}
	where \(|\epsilon_{ijk}^{(2)}|\leq C_{\kappa_3}\sup_{z\in\mcS_{\eta_0}}\big|\partial_{ijk}^{(2)}(c_kQ_{lj}^{12}(z)+b_jQ_{lk}^{13}(z))\big|\). We will show that 
	\begin{align}
		N^{-1/2}\Big|\sum_{j,k=1}^{n,p}\epsilon_{ijk}^{(2)}\Big|=\mrO(\eta_0^{-3}N^{-1/2}(a_i+N^{-1/2})).\label{Eq of ve(2) d=3}
	\end{align}
	Since \(\partial_{ijk}^{(2)}\boldsymbol{Q}(z)=2(\boldsymbol{Q}(z)\partial_{ijk}\bbM(z))^2\boldsymbol{Q}(z)\), without loss of generality, consider
	\begin{align}
		N^{-1/2}\sum_{j,k=1}^{n,p}c_k\partial_{ijk}^{(2)}Q_{lj}^{12}=N^{-3/2}\sum_{t_1\neq t_2,t_3\neq t_4}^3\sum_{j,k=1}^{n,p}c_kQ_{l\tilde{t}_2}^{1t_1}\mcA_{ijk}^{(t_1,t_2)}Q_{\tilde{t}_2\tilde{t}_3}^{t_2t_3}\mcA_{ijk}^{(t_3,t_4)}Q_{\tilde{t}_4j}^{t_42},\label{Eq of ve(2) derivatives d=3}
	\end{align}
	where \(\mcA_{ijk}^{(t_1,t_2)}\) is defined in (\ref{Eq of nij}). If \(\mcA_{ijk}^{(t_1,t_2)},\mcA_{ijk}^{(t_3,t_4)}\neq a_i\), suppose \(\mcA_{ijk}^{(t_1,t_2)}=\mcA_{ijk}^{(t_3,t_4)}=c_k\), then all \(Q_{l\tilde{t}_2}^{1t_1},Q_{\tilde{t}_2\tilde{t}_3}^{t_2t_3},Q_{\tilde{t}_4j}^{t_42}\) must come from \(\bbQ^{11},\bbQ^{22}\) or \(\bbQ^{12}\) and there must exist at least one off-diagonal term, which comes from \(\bbQ^{12}\). Here, we have two possible situations: first, all these three terms come from \(\bbQ^{12}\), then
	\begin{align*}
		N^{-3/2}\sum_{j,k=1}^{n,p}\big|c_k^3Q_{lj}^{12}Q_{ij}^{12}Q_{ij}^{12}\big|\leq N^{-3/2}|Q_{i\cdot}^{12}|^{\circ2}|Q_{\cdot l}^{21}|=\mrO(\eta_0^{-3}N^{-3/2}).
	\end{align*}
	Next, three terms come from \(\bbQ^{11},\bbQ^{22}\) and \(\bbQ^{12}\), respectively, then
	\begin{align*}
		N^{-3/2}\sum_{j,k=1}^{n,p}\big|c_k^3Q_{li}^{11}Q_{ij}^{12}Q_{jj}^{22}\big|\leq N^{-3/2}|Q_{l\cdot}^{11}||\bbQ^{12}|\diag|\bbQ^{22}|\boldsymbol{1}_n=\mrO(\eta_0^{-3}N^{-1}).
	\end{align*}
	Otherwise, one of \(\mcA_{ijk}^{(t_1,t_2)},\mcA_{ijk}^{(t_3,t_4)}\) equals \(c_k\) while the other is \(b_j\). First, if \(\mcA_{ijk}^{(t_1,t_2)}=b_j,\mcA_{ijk}^{(t_3,t_4)}=c_k\), then all possible situations are presented as follows:
	\begin{align*}
		&N^{-3/2}\sum_{j,k=1}^{n,p}\big|c_k^2b_jQ_{li}^{11}Q_{ki}^{31}Q_{jj}^{22}\big|\leq N^{-3/2}|Q_{l\cdot}^{11}||\bbQ^{13}||\bbc|^{\circ2}\cdot\boldsymbol{1}_n\diag|\bbQ^{22}||\bbb|=\mrO(\eta_0^{-3}N^{-1}).\\
		&N^{-3/2}\sum_{j,k=1}^{n,p}\big|c_k^2b_jQ_{lk}^{13}Q_{ii}^{11}Q_{jj}^{22}\big|\leq N^{-3/2}|Q_{ii}^{11}||Q_{l\cdot}^{13}||\bbc|^{\circ2}\cdot\boldsymbol{1}_n\diag|\bbQ^{22}||\bbb|=\mrO(\eta_0^{-3}N^{-1}).\\
		&N^{-3/2}\sum_{j,k=1}^{n,p}\big|c_k^2b_jQ_{lk}^{13}Q_{ij}^{12}Q_{ij}^{12}\big|\leq N^{-3/2}|Q_{l\cdot}^{13}||\bbc|^{\circ2}\cdot|Q_{i\cdot}^{12}|^{\circ2}|\bbb|=\mrO(\eta_0^{-3}N^{-3/2}).\\
		&N^{-3/2}\sum_{j,k=1}^{n,p}\big|c_k^2b_jQ_{li}^{11}Q_{kj}^{32}Q_{ij}^{12}\big|\leq N^{-3/2}|Q_{il}^{11}||Q_{i\cdot}^{12}|\diag(|\bbb|)|\bbQ^{23}||\bbc|=\mrO(\eta_0^{-3}N^{-3/2}).
	\end{align*}
	Finally, when \(\mcA_{ijk}^{(t_1,t_2)}=\mcA_{ijk}^{(t_3,t_4)}=b_j\), we can also derive that (\ref{Eq of ve(2) derivatives d=3}) is no more than \(\mrO(\eta_0^{-3}N^{-1})\), e.g.
	\begin{align*}
		&N^{-3/2}\sum_{j,k=1}^{n,p}\big|c_kb_j^2Q_{li}^{11}Q_{kk}^{33}Q_{ij}^{12}\big|\leq N^{-3/2}|Q_{il}^{11}||Q_{i\cdot}^{12}||\bbb|^{\circ2}\cdot\boldsymbol{1}_p'\diag(|\bbQ^{33}|)|\bbc|=\mrO(\eta_0^{-3}N^{-1}).
	\end{align*}
	Here, we omit the other situations; one can verify that their norms are also bounded by \(\mrO(\eta_0^{-3}N^{-1})\) by similar calculations as the above equation. Now, suppose one of \(\mcA_{ijk}^{(t_1,t_2)},\mcA_{ijk}^{(t_3,t_4)}\) equals \(a_i\). First, if \(\mcA_{ijk}^{(t_1,t_2)}=\mcA_{ijk}^{(t_3,t_4)}=a_i\), we have the following three cases:
    \begin{align*}
        &N^{-3/2}\sum_{j,k=1}^{n,p}\big|c_ka_i^2Q_{lj}^{12}Q_{kk}^{33}Q_{jj}^{22}\big|,\ N^{-3/2}\sum_{j,k=1}^{n,p}\big|c_ka_i^2Q_{lj}^{12}Q_{kj}^{32}Q_{kj}^{32}\big|\leq\mrO(a_i^2\eta_0^{-3}N^{-1/2}),\\
        &N^{-3/2}\sum_{j,k=1}^{n,p}\big|c_ka_i^2Q_{lk}^{13}Q_{jk}^{23}Q_{jj}^{22}\big|\leq\mrO(a_i^2\eta_0^{-3}N^{-1}).
    \end{align*}
    Next, if \(\mcA_{ijk}^{(t_1,t_2)}=a_i,\mcA_{ijk}^{(t_3,t_4)}=c_k\), we have
    \begin{align*}
        &N^{-3/2}\sum_{j,k=1}^{n,p}\big|c_k^2a_iQ_{lj}^{12}Q_{ki}^{31}Q_{jj}^{22}\big|,\ N^{-3/2}\sum_{j,k=1}^{n,p}\big|c_k^2a_iQ_{lk}^{13}Q_{ji}^{21}Q_{jj}^{22}\big|\leq\mrO(a_i\eta_0^{-3}N^{-1})\\
        &N^{-3/2}\sum_{j,k=1}^{n,p}\big|c_k^2a_iQ_{lj}^{12}Q_{kj}^{32}Q_{ij}^{12}\big|\leq\mrO(a_i\eta_0^{-3}N^{-3/2}).
    \end{align*}
    For other situations as \(\mcA_{ijk}^{(t_1,t_2)}=c_k,\mcA_{ijk}^{(t_3,t_4)}=a_i\) and \(\mcA_{ijk}^{(t_1,t_2)}=a_i,\mcA_{ijk}^{(t_3,t_4)}=b_j\) and \(\mcA_{ijk}^{(t_1,t_2)}=b_j,\mcA_{ijk}^{(t_3,t_4)}=a_i\), we can also show that 
    $$\sum_{j,k=1}^{n,p}\big|c_kQ_{l\tilde{t}_2}^{1t_1}\mcA_{ijk}^{(t_1,t_2)}Q_{\tilde{t}_2\tilde{t}_3}^{t_2t_3}\mcA_{ijk}^{(t_3,t_4)}Q_{\tilde{t}_4j}^{t_42}\big|\leq\mrO(a_i\eta_0^{-3}N^{-1/2}),$$
    we omit the details. 
    Now, we obtain that
    $$N^{-1/2}\sum_{j,k=1}^{n,p}\left|c_k\partial_{ijk}^{(2)}Q_{lj}^{12}\right|\leq\mrO(\eta_0^{-3}N^{-1/2}(a_i+N^{-1/2})).$$
	Similarly, we can show the above conclusion is valid for \(\partial_{ijk}^{(2)}Q_{lk}^{13}\) by the same argument, so we conclude that (\ref{Eq of ve(2) d=3}). Notice that
	\begin{align}
		&N^{-1/2}\sum_{j,k=1}^{n,p}\mathbb{E}\left[c_k\partial_{ijk}^{(1)}Q_{lj}^{12}(z)+b_j\partial_{ijk}^{(1)} Q_{lk}^{13}(z)\right]=\notag\\
		&-N^{-1}\mathbb{E}\left[Q_{li}^{11}\tr\left(\boldsymbol{Q}^{22}+\boldsymbol{Q}^{33}\right)+a_iQ_{l\cdot}^{12}\boldsymbol{b}\tr\left(\boldsymbol{Q}^{33}\right)+a_iQ_{l\cdot}^{13}\boldsymbol{c}\tr\left(\boldsymbol{Q}^{22}\right)\right]\notag\\
		&-N^{-1}\mathbb{E}\left[a_iQ_{l\cdot}^{12}\boldsymbol{Q}^{23}\boldsymbol{c}+a_iQ_{l\cdot}^{13}\boldsymbol{Q}^{32}\boldsymbol{b}+Q_{i\cdot}^{12}\boldsymbol{b}Q_{l\cdot}^{13}\boldsymbol{c}+Q_{i\cdot}^{12}Q_{\cdot l}^{21}+Q_{l\cdot}^{12}\boldsymbol{b}Q_{i\cdot}^{13}\boldsymbol{c}+Q_{i\cdot}^{13}Q_{\cdot l}^{31}+2Q_{il}^{11}\boldsymbol{b}'\boldsymbol{Q}^{23}\boldsymbol{c}\right]\notag\\
		&=-\mathbb{E}\left[Q_{li}^{11}(\rho_2(z)+\rho_3(z))+a_i(Q_{l\cdot}^{12}\boldsymbol{b}\rho_3(z)+Q_{l\cdot}^{13}\boldsymbol{c}\rho_2(z))\right]+\mrO(\eta_0^{-2}N^{-1}),\label{Eq of Qii11}
	\end{align}
	where $Q_{i\cdot}^{12}$ and $Q_{\cdot l}^{21}$ is the \(i\)-th row and \(l\)-th column of \(\bbQ^{12}\) and $\bbQ^{21}$, respectively. We also use the fact that
	$$\left|a_lQ_{l\cdot}^{12}\boldsymbol{Q}^{23}\boldsymbol{c}\right|\leq\Vert\boldsymbol{Q}\Vert\Vert Q_{l\cdot}^{12}\Vert\leq\Vert\boldsymbol{Q}\Vert^2,\ \ \left|Q_{i\cdot}^{13}\boldsymbol{c}\right|\leq\Vert\boldsymbol{Q}^{13}\boldsymbol{c}\Vert\leq\Vert\boldsymbol{Q}\Vert,\ {\rm and\ }\Vert Q_{i\cdot}^{12}\Vert^2\leq\Vert\boldsymbol{Q}^{12}\boldsymbol{Q}^{21}\Vert\leq\Vert\boldsymbol{Q}\Vert^2,$$
	so we obtain that
	$$z\mathbb{E}\left[Q_{il}^{11}\right]=-\mathbb{E}\left[Q_{li}^{11}(\rho_2(z)+\rho_3(z))+a_i(Q_{l\cdot}^{12}\boldsymbol{b}\rho_3(z)+Q_{l\cdot}^{13}\boldsymbol{c}\rho_2(z))\right]-\delta_{il}+\mrO(\eta_0^{-3}N^{-1/2}(a_i+N^{-1/2})),$$
	i.e.
	\begin{align*}
		\mathbb{E}\left[(z+\rho_2(z)+\rho_3(z))Q_{il}^{11}\right]=-\mathbb{E}\left[a_i(Q_{l\cdot}^{12}\boldsymbol{b}\rho_3(z)+Q_{l\cdot}^{13}\boldsymbol{c}\rho_2(z))\right]-\delta_{il}+\mrO(\eta_0^{-3}N^{-1/2}(a_i+N^{-1/2})),
	\end{align*}
	where \(\rho_l(z)=N^{-1}\tr(\bbQ^{ll}(z))\). Similarly, we also have
	\begin{align}
		z\mathbb{E}\left[Q_{jl}^{22}\right]&=-\mathbb{E}\left[Q_{lj}^{22}(\rho_1(z)+\rho_3(z))+b_j(Q_{l\cdot}^{21}\boldsymbol{a}\rho_3(z)+Q_{l\cdot}^{23}\boldsymbol{c}\rho_1(z))\right]-\delta_{il}+\mrO(\eta_0^{-3}N^{-1/2}(b_j+N^{-1/2})).\notag\\
		z\mathbb{E}\left[Q_{kl}^{33}\right]&=-\mathbb{E}\left[Q_{lk}^{33}(\rho_1(z)+\rho_2(z))+c_k(Q_{l\cdot}^{31}\boldsymbol{a}\rho_2(z)+Q_{l\cdot}^{32}\boldsymbol{b}\rho_1(z))\right]-\delta_{il}+\mrO(\eta_0^{-3}N^{-1/2}(c_k+N^{-1/2})).\notag
	\end{align}
	Next, for off-diagonal blocks such as \(Q_{il}^{12}(z)\), we can repeat the previous argument to show that
	\begin{align}
		z\mathbb{E}\left[Q_{il}^{12}\right]&=-\mathbb{E}\left[Q_{il}^{12}(\rho_2(z)+\rho_3(z))+a_i(Q_{l\cdot}^{23}\boldsymbol{c}\rho_2(z)+Q_{l\cdot}^{22}\boldsymbol{b}\rho_3(z))\right]+\mrO(\eta_0^{-3}N^{-1/2}(a_i+N^{-1/2})).\notag\\
		z\mathbb{E}\left[Q_{il}^{13}\right]&=-\mathbb{E}\left[Q_{il}^{13}(\rho_2(z)+\rho_3(z))+a_i(Q_{l\cdot}^{33}\boldsymbol{c}\rho_2(z)+Q_{l\cdot}^{32}\boldsymbol{b}\rho_3(z))\right]+\mrO(\eta_0^{-3}N^{-1/2}(b_j+N^{-1/2})).\notag\\
		z\mathbb{E}\left[Q_{jl}^{23}\right]&=-\mathbb{E}\left[Q_{il}^{23}(\rho_1(z)+\rho_3(z))+b_j(Q_{l\cdot}^{33}\boldsymbol{c}\rho_1(z)+Q_{l\cdot}^{31}\boldsymbol{a}\rho_3(z))\right]+\mrO(\eta_0^{-3}N^{-1/2}(c_k+N^{-1/2})).\notag
	\end{align}
	Finally, by similar argument again, we have
	$$\mathbb{E}\left[(z+\rho_2(z)+\rho_3(z))Q_{i\cdot}^{12}\boldsymbol{b}\right]=-a_i\mathbb{E}\left[\boldsymbol{b}'\boldsymbol{Q}^{23}\boldsymbol{c}\rho_2(z)+\boldsymbol{b}'\boldsymbol{Q}^{22}\boldsymbol{b}\rho_3(z)\right]+\mrO(\eta_0^{-3}N^{-1/2}(a_i+N^{-1/2})).$$
	According to Lemma \ref{Thm of Entrywise almost sure convergence}, we have that \(|(Q_{l\cdot}^{13}\bbc)^c|,|(Q_{l\cdot}^{12}\bbb)^c|,|(\bbb'\bbQ^{23}\bbc)^c|,|(\bbb'\bbQ^{22}\bbb)^c|,|\rho_2(z)^c|,|\rho_3(z)^c|\prec\eta_0^{-5}N^{-\omega}\), then
	$$|{\rm Cov}(\boldsymbol{b}'\boldsymbol{Q}^{22}\boldsymbol{b},\rho_3(z))|\leq\mathbb{E}\left[|\rho_3(z)^c(\bbb'\bbQ^{22}\bbb)^c|\right]\leq\mrO(\eta_0^{-10}N^{-2\omega}+\Vert\bbQ\Vert^2\exp(-CN^{1-2\omega})),$$
	as does others, hence we obtain that
	\begin{small}
	\begin{align*}
        &(z+\mfm_2(z)+\mfm_3(z))\mbE[Q_{il}^{11}]=-a_i\big(\mfm_2(z)\mbE\big[Q_{l\cdot}^{13}\bbc\big]+\mfm_3(z)\mbE\big[Q_{l\cdot}^{12}\bbb\big]\big)-\delta_{il}+\mrO(\eta_0^{-10}N^{-2\omega}+\eta_0^{-3}N^{-1/2}(a_i+N^{-1/2})),\\
		&(z+\mfm_2(z)+\mfm_3(z))\mathbb{E}\left[Q_{i\cdot}^{12}\boldsymbol{b}\right]=-a_i\big(\mfm_2(z)W_{23,N}^{(3)}+\mfm_3(z)W_{22,N}^{(3)}\big)+\mrO(\eta_0^{-10}N^{-2\omega}+\eta_0^{-3}N^{-1/2}(a_i+N^{-1/2})),
	\end{align*}
	\end{small}\noindent
	where
    \begin{align}
        W_{st,N}^{(3)}(z)=\mbE\big[(\bba^{(s)})'\bbQ^{st}(z)\bba^{(t)}\big],\label{Eq of W}
    \end{align} 
    for \(1\leq s,t\leq 3\). By Theorem \ref{Thm of approximation} and \((z+g_2(z)+g_3(z))^{-1}=-\mfc_1^{-1}g_1(z)\), it yields that
	\begin{align}
	    &\mbE[Q_{il}^{11}]=\mfc_1^{-1}g_1(z)\big(\delta_{il}+a_i\big(\mfm_2(z)\mbE\big[Q_{l\cdot}^{13}\bbc\big]+\mfm_3(z)\mbE\big[Q_{l\cdot}^{12}\bbb\big]\big)\big)+\mrO(\eta_0^{-17}N^{-2\omega}+\eta_0^{-5}N^{-1/2}(a_i+N^{-1/2})),\notag\\
        &a_i\mathbb{E}\left[Q_{i\cdot}^{12}\boldsymbol{b}\right]=a_i^2\mfc_1^{-1}g_1(z)\big(g_2(z)W_{23,N}^{(d)}+g_3(z)W_{22,N}^{(d)}\big)+a_i\mrO(\eta_0^{-17}N^{-2\omega}+\eta_0^{-5}N^{-1/2}(a_i+N^{-1/2})).\label{Eq of W12 approximation 1}
	\end{align}
	Summing all \(1\leq i\leq n\) for \(a_i\mathbb{E}\left[Q_{i\cdot}^{12}\boldsymbol{b}\right]\), we have
	\begin{align}
	    W_{12,N}^{(3)}=\mfc_1^{-1}g_1(z)\big(g_2(z)W_{23,N}^{(d)}+g_3(z)W_{22,N}^{(d)}\big)+\mrO(\eta_0^{-17}N^{-2\omega+1/2}+\eta_0^{-5}N^{-1/2}),\label{Eq of W12 approximation 2}
	\end{align}
	since \(\omega\in(1/2-\delta,1/2)\) and $\delta>0$ is sufficiently small, then \(2\omega-1/2\in(1/2-2\delta,1/2)\), combining \eqref{Eq of W12 approximation 1} and \eqref{Eq of W12 approximation 2}, we have 
	$$\mathbb{E}\left[Q_{i\cdot}^{12}\boldsymbol{b}\right]=a_iW_{12,N}^{(3)}+a_i\mrO(\eta_0^{-17}N^{-2\omega+1/2}),$$
	and we can derive the same result for \(\mathbb{E}\left[Q_{i\cdot}^{13}\boldsymbol{c}\right]\), it concludes that
	\begin{align*}
		\mathbb{E}\left[Q_{il}^{11}\right]&=\mfc_1^{-1}g_1(z)\big[a_i\big(g_3(z)\mbE[Q_{i\cdot}^{12}\bbb]+g_2(z)\mbE[Q_{i\cdot}^{13}\bbc]\big)+\delta_{il}\big]+\mrO(\eta_0^{-17}N^{-2\omega}+\eta_0^{-5}N^{-1/2}(a_i+N^{-1/2}))\\
		&=\mfc_1^{-1}g_1(z)\big[a_i^2\big(g_3(z)W_{12,N}^{(3)}+g_2(z)W_{13,N}^{(3)}\big)+\delta_{il}\big]+\mrO(\eta_0^{-19}N^{-2\omega+1/2}).
	\end{align*}
	Finally, by Lemma \ref{Thm of Entrywise almost sure convergence}, i.e. \(|(Q_{ii}^{11})^{\circ}|\prec\eta_0^{-5}N^{-\omega}\), and
    $$|W_{st,N}^{(3)}(z)-W_{st}^{(3)}(z)|\leq\mrO(\eta_0^{-17}N^{-\omega}),$$
    the proof of the above equation is deferred to \eqref{Eq of solve W} for clarity, we can conclude Theorem \ref{Thm of entrywise law d=3} for $Q_{il}^{11}$. For other cases, since the proof arguments are the same, the details are omitted for brevity.
\end{proof}

\section{Mean and covariance functions when \texorpdfstring{$d=3$}{d=3}}\label{Sec of mean and covariance}
\setcounter{equation}{0}
\def\theequation{\thesection.\arabic{equation}}
\setcounter{subsection}{0}
In this section, we will derive the mean function \(\mbE[\tr(\bbQ(z))]-Ng(z)\) and variance function \(\Var(\tr(\bbQ(z)))\), respectively, and these two functions will be used to calculate the asymptotic mean and variance of the LSS in \S\ref{Sec of CLT}. Recall that \(\omega\in(1/2-\delta,1/2)\) in Lemma \ref{Thm of Entrywise almost sure convergence} is a fixed constant which can be sufficiently close to \(1/2\); we will not repeat its definition throughout this section. Next, by (\ref{Eq of nij}), the notations in (\ref{Eq of mfb}) and (\ref{Eq of mcB}) are equivalent to
\begin{align}
	\mfb_1^{(1)}=\frac{1}{\sqrt{N}}\sum_{i=1}^ma_i,\quad\mfb_2^{(1)}=\frac{1}{\sqrt{N}}\sum_{j=1}^nb_j,\quad\mfb_3^{(1)}=\frac{1}{\sqrt{N}}\sum_{k=1}^pc_k,\label{Eq of Ap L1 d=3}
\end{align}
and
\begin{align}
	\mcB_{(4)}^{(2,3)}=\Vert\bba\Vert_4^4,\quad\mcB_{(4)}^{(1,3)}=\Vert\bbb\Vert_4^4,\quad\mcB_{(4)}^{(1,2)}=\Vert\bbc\Vert_4^4,\label{Eq of Ap L4 d=3}
\end{align}
where we use the notations in (\ref{Eq of vector notation d=3}).
\subsection{Covariance function}\label{Sec of Covariance}
\begin{thm}\label{Thm of Variance}
	Under Assumptions {\rm \ref{Ap of general noise}} and {\rm \ref{Ap of dimension}}, for any \(\eta_0>0\), \(z_1,z_2\in\mathcal{S}_{\eta_0}\) in {\rm (\ref{Eq of stability region})}, let
	\begin{align}
		\mcC_{st,N}^{(3)}(z_1,z_2):=\Cov(\tr(\bbQ^{ss}(z_1)),\tr(\bbQ^{tt}(z_2)))\quad{\rm and}\quad\bbC_N^{(3)}(z_1,z_2):=[\mcC_{st,N}^{(3)}(z_1,z_2)]_{3\times3},\label{Eq of mcC d=3}
	\end{align}
    where \(s,t\in\{1,2,3\}\). Further define
    \begin{align*}
		\bbF_N^{(3)}(z_1,z_2)=[\mcF_{st,N}^{(3)}(z_1,z_2)]_{3\times3},\quad\mcF_{st,N}^{(3)}(z_1,z_2):=2\mcV_{st}^{(3)}(z_1,z_2)+\kappa_4\mcW_{st,N}^{(3)}(z_1,z_2),
	\end{align*}
    where the precise definitions of \(\mcV_{st}^{(3)}(z_1,z_2)\) and \(\mcW_{st,N}^{(3)}(z_1,z_2)\) are postponed to {\rm (\ref{Eq of mcV limiting d=3})} and {\rm (\ref{Eq of mcW limiting d=3})}, respectively. Then we have
	\begin{align}
		\lim_{N\to\infty}\Vert\bbC_N^{(3)}(z_1,z_2)-\bbPi^{(3)}(z_1,z_2)^{-1}\diag(\mfc^{-1}\circ\bbg(z_1))\bbF_N^{(3)}(z_1,z_2)\Vert=0,\label{Eq of bbC d=3}
	\end{align}
	where \(\bbPi^{(3)}(z_1,z_2)\) is defined in {\rm (\ref{Eq of invertible 2})}. Consequently, \(\Var(\tr(\bbQ(z)))\) is bounded by \(C_{\eta_0,\mfc}\) for any \(z\in\mcS_{\eta_0}\) and
    \begin{align*}
        \lim_{N\to\infty}\big|\Cov\big(\tr(\bbQ(z_1),\tr(\bbQ(z_2))\big)-\mcC_N^{(3)}(z_1,z_2)\big|=0,
    \end{align*}
    where
    \begin{align}
    \mcC_N^{(3)}(z_1,z_2):=\Cov(\tr(\bbQ(z_1)),\tr(\bbQ(z_2)))=\boldsymbol{1}_3'\bbC_N^{(3)}(z_1,z_2)\boldsymbol{1}_3\label{Eq of covariance function d=3}
\end{align}
\end{thm}
\begin{proof}
	Without loss of generality, assume \(\mathcal{C}_{ii,N}^{(3)}(z,z)>1\) for \(i=1,2,3\), otherwise \(\mathcal{C}_{ii,N}^{(3)}(z,z)\) is already bounded. Here, we present the detailed proof only for \(\mcC_{11,N}^{(3)}(z_1,z_2)\); the arguments for other cases are identical. Note that
	$$\mathcal{C}_{11,N}^{(3)}(z,z)=\mathbb{E}\left[\tr(\boldsymbol{Q}^{11}(z))(\tr(\boldsymbol{Q}^{11}(\bar{z}))-\mathbb{E}[\tr(\boldsymbol{Q}^{11}(\bar{z}))])\right],$$
	and \(\boldsymbol{Q}(z)(\bbM-z\boldsymbol{I}_N)=\boldsymbol{I}_N\), we have
	\begin{align}
		z\mathcal{C}_{11,N}^{(3)}(z,z)=\frac{1}{\sqrt{N}}\sum_{i,j,k=1}^{m,n,p}\mathbb{E}\left[X_{ijk}F_{ijk}^1(z)\tr(\boldsymbol{Q}^{11}(\bar{z}))^c\right],\notag
	\end{align}
	where the superscript ``\(c\)'' represents the centering operator and \(F_{ijk}^1(z):=c_kQ_{ij}^{12}(z)+b_jQ_{ik}^{13}(z)\), then by cumulant expansion (\ref{Eq of cumulant expansion}), we have
	\begin{align}
		z\mathcal{C}_{11,N}^{(3)}(z,z)=\frac{1}{\sqrt{N}}\sum_{i,j,k=1}^{m,n,p}\left(\sum_{l=1}^3\frac{\kappa_{l+1}}{l!}\mathbb{E}\left[\partial_{ijk}^{(l)}\big\{F_{ijk}^1(z)\tr(\boldsymbol{Q}^{11}(\bar{z}))^c\big\}\right]+\epsilon_{ijk}^{(4)}\right),\label{Eq of covariance 0}
	\end{align}
    where the remainder term \(\epsilon_{ijk}^{(4)}\) satisfies that
    $$|\epsilon_{ijk}^{(4)}|\leq C_{\kappa_5}\sup_{z\in\mcS_{\eta_0}}\big|\partial_{ijk}^{(4)}\big\{F_{ijk}^1(z)\tr(\boldsymbol{Q}^{11}(\bar{z}))^c\big\}\big|.$$
    \noindent
    {\bf First derivatives:} When \(l=1\), by direct computations, we obtain
		\begin{align}
			&N^{-1/2}\sum_{i,j,k=1}^{m,n,p}\mathbb{E}\left[\partial_{ijk}^{(1)}\big\{F_{ijk}^1(z)\tr(\boldsymbol{Q}^{11}(\bar{z}))^c\big\}\right]=\notag\\
			&-N^{-1}{\rm Cov}\left(\tr(\boldsymbol{Q}^{11}(z))\tr(\boldsymbol{Q}^{22}(z)+\boldsymbol{Q}^{33}(z)),\tr(\boldsymbol{Q}^{11}(z))\right)\notag\\
			&-N^{-1}{\rm Cov}\left(\tr(\boldsymbol{Q}^{12}(z)\boldsymbol{Q}^{21}(z)+\boldsymbol{Q}^{13}(z)\boldsymbol{Q}^{31}(z)),\tr(\boldsymbol{Q}^{11}(z))\right)\notag\\
			&-N^{-1}{\rm Cov}\left(2\boldsymbol{b}'\boldsymbol{Q}^{23}(z_1)\boldsymbol{c}\tr(\boldsymbol{Q}^{11}(z))+\boldsymbol{a}'\boldsymbol{Q}^{13}(z)\boldsymbol{c}\tr(\boldsymbol{Q}^{22}(z))+\boldsymbol{a}'\boldsymbol{Q}^{12}(z)\boldsymbol{b}\tr(\boldsymbol{Q}^{33}(z)),\tr(\boldsymbol{Q}^{11}(z))\right)\notag\\
			&-N^{-1}{\rm Cov}\left(2\boldsymbol{b}'\boldsymbol{Q}^{21}(z)\boldsymbol{Q}^{13}(z)\boldsymbol{c}+\boldsymbol{a}'\boldsymbol{Q}^{12}(z)\boldsymbol{Q}^{23}(z)\boldsymbol{c}+\boldsymbol{a}'\boldsymbol{Q}^{13}(z)\boldsymbol{Q}^{32}(z)\boldsymbol{b},\tr(\boldsymbol{Q}^{11}(z))\right)\notag\\
			&-2N^{-1}\mathbb{E}\left[\tr(\boldsymbol{Q}^{11}(\bar{z})\boldsymbol{Q}^{12}(\bar{z})\boldsymbol{Q}^{21}(z)+\boldsymbol{Q}^{11}(\bar{z})\boldsymbol{Q}^{13}(\bar{z})\boldsymbol{Q}^{31}(z))\right]+\mrO(\eta^{-3}N^{-1}).\notag
		\end{align}
        Here, we claim that, except 
        $$\left\{\begin{array}{l}
             N^{-1}{\rm Cov}\left(\tr(\boldsymbol{Q}^{11}(z))\tr(\boldsymbol{Q}^{22}(z)+\boldsymbol{Q}^{33}(z)),\tr(\boldsymbol{Q}^{11}(z))\right),\\
             N^{-1}\mathbb{E}\left[\tr(\boldsymbol{Q}^{11}(\bar{z})\boldsymbol{Q}^{12}(\bar{z})\boldsymbol{Q}^{21}(z)+\boldsymbol{Q}^{11}(\bar{z})\boldsymbol{Q}^{13}(\bar{z})\boldsymbol{Q}^{31}(z))\right],
        \end{array}\right.$$
        all other terms in the above equation are bounded by $\mrO\big(\eta^{-6}N^{-\omega}\mcC_{11,N}^{(3)}(z,z)\big)$. By Lemma \ref{Thm of Entrywise almost sure convergence}, we have
		$$N^{-1}\big|\tr(\boldsymbol{Q}^{ii}(z))^c\big|\prec \eta^{-5}N^{-\omega}\ \ {\rm and\ \ }N^{-1}\big|\tr(\boldsymbol{Q}^{ij}(z)\boldsymbol{Q}^{ji}(z))^c\big|\prec \eta^{-6}N^{-\omega},$$
		with probability of \(1-C\exp(-CN^{1-2\omega})\), then we can imply
		\begin{align}
			&N^{-1}\big|{\rm Cov}\left(\tr(\boldsymbol{Q}^{12}(z)\boldsymbol{Q}^{21}(z)),\tr(\boldsymbol{Q}^{11}(z))\right)\big|=N^{-1}\big|\mathbb{E}\left[\tr(\boldsymbol{Q}^{12}(z)\boldsymbol{Q}^{21}(z))^c\tr(\boldsymbol{Q}^{11}(\bar{z}))^c\right]\big|\notag\\
			&=N^{-1}\Big|\mathbb{E}\left[\tr(\boldsymbol{Q}^{12}(z)\boldsymbol{Q}^{21}(z))^c\tr(\boldsymbol{Q}^{11}(\bar{z}))^c\big(1_{N^{-1}|\tr(\bbQ^{12}(z)\bbQ^{21}(z))^c|<\eta^{-6}N^{-\omega}}+1_{N^{-1}|\tr(\bbQ^{12}(z)\bbQ^{21}(z))^c|\geq\eta^{-6}N^{-\omega}}\big)\right]\Big|\notag\\
			&\leq\eta^{-6}N^{-\omega}\mcC_{11,N}^{(3)}(z,z)^{1/2}+\eta^{-3}N\mathbb{P}(N^{-1}|\tr(\bbQ^{12}(z)\bbQ^{21}(z))^c|\geq\eta^{-6}N^{-\omega})\leq \mrO\big(\eta^{-6}N^{-\omega}\mcC_{11,N}^{(3)}(z,z)\big),\notag
		\end{align}
		where we use the fact that \(\mcC_{11,N}^{(3)}(z,z)>1\). Similarly, we can also show that
		\begin{align}
			&N^{-1}{\rm Cov}(\boldsymbol{b}'\boldsymbol{Q}^{21}(z)\boldsymbol{Q}^{13}(z)\boldsymbol{c},\tr(\boldsymbol{Q}^{11}(z))),N^{-1}{\rm Cov}(\boldsymbol{b}'\boldsymbol{Q}^{23}(z)\boldsymbol{c}\tr(\boldsymbol{Q}^{11}(z)),\tr(\boldsymbol{Q}^{11}(z)))\leq\mrO\big(\eta^{-6}N^{-\omega}\mcC_{11,N}^{(3)}(z,z)\big).\notag
		\end{align}
		Moreover, note that
		\begin{align}
			&N^{-1}{\rm Cov}\left(\tr(\boldsymbol{Q}^{11}(z))\tr(\boldsymbol{Q}^{22}(z)),\tr(\boldsymbol{Q}^{11}(z))\right)\notag\\
			&=N^{-1}\mathbb{E}\Big[\big(\tr(\boldsymbol{Q}^{11}(z))\tr(\boldsymbol{Q}^{22}(z))-\mathbb{E}[\tr(\boldsymbol{Q}^{11}(z))]\tr(\boldsymbol{Q}^{22}(z))\notag\\
			&+\mathbb{E}[\tr(\boldsymbol{Q}^{11}(z))]\tr(\boldsymbol{Q}^{22}(z))-\mathbb{E}[\tr(\boldsymbol{Q}^{11}(z))]\mathbb{E}[\tr(\boldsymbol{Q}^{22}(z))]\notag\\
			&+\mathbb{E}[\tr(\boldsymbol{Q}^{11}(z))]\mathbb{E}[\tr(\boldsymbol{Q}^{22}(z))]-\mathbb{E}[\tr(\boldsymbol{Q}^{11}(z))\tr(\boldsymbol{Q}^{22}(z))]\big)\tr(\boldsymbol{Q}^{11}(\bar{z}))^c\Big]\notag\\
			&=\mfm_1(z){\rm Cov}\left(\tr(\boldsymbol{Q}^{22}(z)),\tr(\boldsymbol{Q}^{11}(z))\right)+N^{-1}\mathbb{E}\left[\tr(\boldsymbol{Q}^{22}(z))|\tr(\boldsymbol{Q}^{11}(z))^c|^2\right],\label{Eq of covariance trick}
		\end{align}
		where \(\mfm_i(z)=N^{-1}\mathbb{E}[\tr(\boldsymbol{Q}^{ii}(z))],\rho_i(z)=N^{-1}\tr(\bbQ^{ii}(z))\) are defined in (\ref{Eq of mi}), and \(|\rho_2(z)^c|\leq \eta^{-5}N^{-\omega}\) with probability of \(1-C\exp(-CN^{1-2\omega})\) by Lemma \ref{Thm of Entrywise almost sure convergence}, then we have
		\begin{align}
			&\big|N^{-1}\mathbb{E}\left[\tr(\boldsymbol{Q}^{22}(z))|\tr(\boldsymbol{Q}^{11}(z))^c|^2\right]-\mfm_2(z)\mcC_{11,N}^{(3)}(z,z)\big|=\big|\mathbb{E}\left[\rho_2(z)^c|\tr(\boldsymbol{Q}^{11}(z))^c|^2\right]\big|\notag\\
			&\leq \eta^{-5}N^{-\omega}\mcC_{11,N}^{(3)}(z,z)+\mathbb{E}\left[|\rho_2(z)^c|1_{|\rho_2(z)^c|\geq \eta_0^{-4}N^{-\omega}}|\tr(\boldsymbol{Q}^{11}(z))^c|^2\right]\notag\\
			&\leq \eta^{-5}N^{-\omega}\mcC_{11,N}^{(3)}(z,z)+\eta^{-3}N^2\exp(-CN^{1-2\omega})\leq\mrO(\eta^{-5}N^{-\omega}\mcC_{11,N}^{(3)}(z,z)).\notag
		\end{align}
		In summary, we obtain that
		\begin{align}
			&N^{-1/2}\sum_{i,j,k=1}^{m,n,p}\mathbb{E}\left[\partial_{ijk}^{(1)}\big\{F_{ijk}^1(z)\tr(\boldsymbol{Q}^{11}(\bar{z}))^c\big\}\right]=-2\mathcal{V}_{11,N}^{(3)}(z,\bar{z})+\mrO(\eta^{-3}N^{-1})\notag\\
			&-\mfm_1(z)[\mathcal{C}_{2,1}(z,z)+\mathcal{C}_{3,1}(z,z)]-(\mfm_2(z)+\mfm_3(z)+\mrO(\eta^{-5}N^{-\omega}))\mcC_{11,N}^{(3)}(z,z),\label{Eq of covariance 1}
		\end{align}
		where
		\begin{align}
			\mathcal{V}_{ij,N}^{(3)}(z_1,z_2):=N^{-1}\sum_{l\neq i}^3\mbE[\tr(\boldsymbol{Q}^{ij}(z_2)\boldsymbol{Q}^{jl}(z_2)\boldsymbol{Q}^{li}(z_1))]\label{Eq of capital V}
		\end{align}
		for \(i,j\in\{1,2,3\}\). Readers can refer to \S\ref{Sec of majors} for proofs of \(\lim_{N\to\infty}\mathcal{V}_{ij,N}^{(3)}(z_1,z_2)=\mathcal{V}_{ij}^{(3)}(z_1,z_2)\) in (\ref{Eq of mcV limiting d=3}).

    \vspace{5mm}
    \noindent
    {\bf Second derivatives:} When \(l=2\), since
		\begin{align}
			\partial_{ijk}^{(2)}\big\{F_{ijk}^1(z)\tr(\boldsymbol{Q}^{11}(\bar{z}))^c\big\}=\sum_{\alpha=0}^2\binom{2}{\alpha}\partial_{ijk}^{(2-\alpha)}\tr(\boldsymbol{Q}^{11}(\bar{z}))^c\partial_{ijk}^{(\alpha)}F_{ijk}^1(z),\notag
		\end{align}
		and \(\partial_{ijk}^{(l)}\tr(\boldsymbol{Q}^{11}(\bar{z}))\) have the following expression by Lemma \ref{Lem of minor terms}:
		$$\partial_{ijk}^{(l)}\tr(\bbQ^{11})=(-1)^ll!N^{-l/2}\mcA_{ijk}^{(t_{2l},t_{2l+1})}Q_{\tilde{t}_{2l+1}\cdot}^{t_{2l+1}1}Q_{\cdot\tilde{t}_2}^{1t_2}\prod_{\alpha=2}^l\mcA_{ijk}^{(t_{2\alpha-2},t_{2\alpha-1})}Q_{\tilde{t}_{2\alpha-1}\tilde{t}_{2\alpha}}^{t_{2\alpha-1}t_{2\alpha}},\ {\rm for\ }l\geq2,$$
		and
		\begin{align}
			\partial_{ijk}^{(1)}\tr(\bbQ^{11})&=-2N^{-1/2}(a_iQ_{j\cdot}^{21}Q_{\cdot k}^{13}+b_jQ_{i\cdot}^{11}Q_{\cdot k}^{13}+c_kQ_{i\cdot}^{11}Q_{\cdot j}^{12}):=-2N^{-1/2}(a_iP_{jk}^{23}+b_jP_{ik}^{13}+c_kP_{ij}^{12}),\notag
		\end{align}
		where \(Q_{\cdot i}\) and \(Q_{i\cdot}\) represent the \(i\)-th column and row of \(\bbQ\) and \(P_{\tilde{t}_{2l+1}\tilde{t}_2}^{t_{2l+1}t_{2}}:=Q_{\tilde{t}_{2l+1}\cdot}^{t_{2l+1}1}Q_{\cdot\tilde{t}_2}^{1t_2}\), it implies that \(\partial_{ijk}^{(1)}\tr(\bbQ^{11})\) only contains off-diagonal parts. By the definitions of operators \(\mathscr{D},\mathscr{O}\) in (\ref{Eq of operator D}), (\ref{Eq of operator O}) and Lemma \ref{Lem of minor terms}, we know that \(|F_{ijk}^1(z)|\leq\mrO(\eta^{-1})\) and
		\begin{align}
			&N^{-1/2}\sum_{i,j,k=1}^{m,n,p}\mathbb{E}\left[\big|F_{ijk}^1(z)\mathscr{O}\big(\partial_{ijk}^{(2)}\tr(\bbQ^{11}(\bar{z}))^c\big)\big|\right]\notag\\
            &\leq\eta^{-1}N^{-1/2}\sum_{i,j,k=1}^{m,n,p}\mathbb{E}\left[\big|\mathscr{O}\big(\partial_{ijk}^{(2)}\tr(\bbQ^{11}(\bar{z}))^c\big)\big|\right]\leq\mrO(\eta^{-4}N^{-1/2}).\label{Eq of covariance l=2 alpha=0}
		\end{align}
		Next, for
		\begin{align}
			&N^{-1/2}\sum_{i,j,k=1}^{m,n,p}\mathbb{E}\left[\big|F_{ijk}^1(z)\mathscr{D}\big(\partial_{ijk}^{(2)}\tr(\bbQ^{11}(\bar{z}))^c\big)\big|\right]\notag\\
			&=6N^{-3/2}\sum_{i,j,k=1}^{m,n,p}\mathbb{E}\left[\big|F_{ijk}^1(z)(a_i^2P_{jj}^{22}(z)Q_{kk}^{33}(z)+b_j^2P_{ii}^{11}(z)Q_{kk}^{33}(z)+c_k^2P_{ii}^{11}(z)Q_{jj}^{22}(z))\big|\right],\notag
		\end{align}
		we can show that the above equation is bounded by \(\mrO(\eta^{-4}N^{-1/2})\). For example, we have
		\begin{align}
			&N^{-3/2}\sum_{i,j,k=1}^{m,n,p}\mathbb{E}\left[\big|F_{ijk}^1(z)b_j^2P_{ii}^{11}(z)Q_{kk}^{33}(z)\big|\right]\notag\\
			&\leq N^{-3/2}\left(\boldsymbol{1}_m'{\rm diag}(|\bbP^{11}|)|\bbQ^{13}|{\rm diag}(|\bbQ^{33}|)\boldsymbol{1}_k+\boldsymbol{1}_m'{\rm diag}(|\bbP^{11}|)|\bbQ^{12}||\bbb|\times\boldsymbol{1}_k{\rm diag}(|\bbQ^{13}|)|\bbc|\right)\notag\\
			&\leq\mrO(\eta^{-4}N^{-1/2}),\label{Eq of covariance l=2 alpha=0 D}
		\end{align}
        the calculations for other terms are the same, the details are omitted for brevity. For the \(\partial_{ijk}^{(1)}F_{ijk}^1(z)\partial_{ijk}^{(1)}\tr(\bbQ(\bar{z}))^c\), since we know that \(|\partial_{ijk}^{(1)}\tr(\bbQ(\bar{z}))|=\mrO(\eta^{-2}N^{-1/2})\), then by Lemma \ref{Lem of minor terms}, it implies that
		\begin{align}
			&N^{-1/2}\sum_{i,j,k=1}^{m,n,p}\mathbb{E}\left[\big|\partial_{ijk}^{(1)}\tr(\bbQ(\bar{z}))^c\mathscr{O}\big(\partial_{ijk}^{(1)}F_{ijk}^1(z)\big)\big|\right]\notag\\
            &\leq \eta^{-2}N^{-1}\sum_{i,j,k=1}^{m,n,p}\mathbb{E}\left[\big|\mathscr{O}\big(\partial_{ijk}^{(1)}F_{ijk}^1(z)\big)\big|\right]\leq\mrO(\eta^{-4}N^{-1/2}).\label{Eq of covariance l=2 alpha=1}
		\end{align}
		By the same arguments as those in (\ref{Eq of covariance l=2 alpha=0 D}), we can derive that
		$$N^{-1/2}\sum_{i,j,k=1}^{m,n,p}\mathbb{E}\left[\big|\partial_{ijk}^{(1)}\tr(\bbQ(\bar{z}))^c\mathscr{D}\big(\partial_{ijk}^{(1)}F_{ijk}^1(z)\big)\big|\right]\leq\mrO(\eta^{-4}N^{-1/2}).$$
		Finally, let us focus on
		\begin{align}
			&N^{-1/2}\sum_{i,j,k=1}^{m,n,p}\mathbb{E}\left[\tr(\boldsymbol{Q}^{11}(\bar{z}))^c\mathscr{D}\big(\partial_{ijk}^{(2)}F_{ijk}^1(z)\big)\right]\label{Eq of covariance l=2 alpha=2}\\
			&=4N^{-3/2}\mathbb{E}\left[\boldsymbol{1}_m'{\rm diag}(\boldsymbol{Q}^{11}(z))\boldsymbol{a}\boldsymbol{1}_n'{\rm diag}(\boldsymbol{Q}^{22}(z))\boldsymbol{b}\boldsymbol{1}_p'{\rm diag}(\boldsymbol{Q}^{33}(z))\boldsymbol{c}\times\tr(\boldsymbol{Q}^{11}(\bar{z}))^c\right]\notag\\
			&=4N^{-3/2}\mathbb{E}\left[\{\boldsymbol{1}_m'{\rm diag}(\boldsymbol{Q}^{11}(z))\boldsymbol{a}\boldsymbol{1}_n'{\rm diag}(\boldsymbol{Q}^{22}(z))\boldsymbol{b}\boldsymbol{1}_p'{\rm diag}(\boldsymbol{Q}^{33}(z))\boldsymbol{c}\}^c\times\tr(\boldsymbol{Q}^{11}(\bar{z}))\right].\notag
		\end{align}
		For simplicity, let \(H_N^{(1)}(z):=N^{-1/2}(\boldsymbol{1}_m'{\rm diag}(\boldsymbol{Q}^{11}(z))\boldsymbol{a})^c\) and define \(H_N^{(2)}(z),H_N^{(3)}(z)\) analogously, then 
		\begin{align}
			(H_N^{(1)}H_N^{(2)}H_N^{(3)})^c&=(H_N^{(1)})^cH_N^{(2)}H_N^{(3)}+\mbE[H_N^{(1)}](H_N^{(2)}H_N^{(3)})^c+\mbE[H_N^{(1)}(H_N^{(2)}H_N^{(3)})^c],\notag\\
			(H_N^{(2)}H_N^{(3)})^c&=(H_N^{(2)})^cH_N^{(3)}+\mbE[H_N^{(2)}](H_N^{(3)})^c+\mbE[H_N^{(3)}(H_N^{(3)})^c].\notag
		\end{align}
		By Lemma \ref{Thm of Entrywise almost sure convergence}, \(H_N^{(1)}(z),H_N^{(2)}(z),H_N^{(3)}(z)\prec\mrO(\eta^{-5}N^{-\omega})\), so \(\big|(H_N^{(1)}H_N^{(2)}H_N^{(3)})^c\big|\prec\mrO(\eta^{-7}N^{-\omega})\) and
		$$N^{-1/2}\Big|\sum_{i,j,k=1}^{m,n,p}\mathbb{E}\left[\tr(\boldsymbol{Q}^{11}(\bar{z}))^c\mathscr{D}\big(\partial_{ijk}^{(2)}F_{ijk}^1(z)\big)\right]\Big|\leq\mrO(\eta^{-7}N^{-\omega}\mcC_{1,1}(z,z)).$$
		Similarly, by the previous arguments, we can also show that
		$$N^{-1/2}\Big|\sum_{i,j,k=1}^{m,n,p}\mathbb{E}\left[\tr(\boldsymbol{Q}^{11}(\bar{z}))^c\mathscr{O}(\partial_{ijk}^{(2)}F_{ijk}^1(z))\right]\Big|\leq\mrO(\eta^{-7}N^{-\omega}\mcC_{1,1}(z,z)),$$
		here we omit the details for clarity. Finally, we obtain
		\begin{align}
			N^{-1/2}\sum_{i,j,k=1}^{m,n,p}\mathbb{E}\left[\partial_{ijk}^{(2)}\big\{F_{ijk}^1(z)\tr(\boldsymbol{Q}^{11}(\bar{z}))\big\}\right]=\mrO(\eta^{-7}N^{-\omega}\mcC_{11,N}^{(3)}(z,z)+\eta^{-4}N^{-1/2}).\label{Eq of covariance 2}
		\end{align}

    \vspace{5mm}
    \noindent
    {\bf Third derivatives:} When \(l=3\), consider
		$$N^{-1/2}\sum_{i,j,k=1}^{m,n,p}\mathbb{E}\left[\partial_{ijk}^{(3-\alpha)}\tr(\boldsymbol{Q}^{11}(\bar{z}))^c\partial_{ijk}^{(\alpha)}F_{ijk}^1(z)\right],\ {\rm where\ }\alpha=0,1,\cdots,3.$$
		We claim that the major terms appear only when \(\alpha=1\). First, when \(\alpha=3\), similar to (\ref{Eq of covariance l=2 alpha=2}), we can show that
		\begin{align}
			&N^{-1/2}\sum_{i,j,k=1}^{m,n,p}\mathbb{E}\left[\tr(\boldsymbol{Q}^{11}(\bar{z}))^c\mathscr{D}(\partial_{ijk}^{(3)}F_{ijk}^1(z))\right]\notag\\
			&=-N^{-2}\mathbb{E}\left[\tr(\boldsymbol{Q}^{11}(z)^{\circ2})\left(\Vert\boldsymbol{c}\Vert_4^4\tr(\boldsymbol{Q}^{22}(z)^{\circ2})+\Vert\boldsymbol{b}\Vert_4^4\tr(\boldsymbol{Q}^{33}(z)^{\circ2})\right)\tr(\boldsymbol{Q}^{11}(\bar{z}))^c\right]+\mrO(\eta^{-5}N^{-1/2})\notag\\
			&=\mrO(\eta^{-8}N^{-\omega}\mcC_{11,N}^{(3)}(z,z))+\mrO(\eta^{-5}N^{-1/2}),\notag
		\end{align}
		where \(N^{-1}\tr(\boldsymbol{Q}^{ll}(z)^{\circ2})^c\prec \eta^{-6}N^{-\omega}\) by Lemma \ref{Thm of Entrywise almost sure convergence}, and as does
		$$N^{-1/2}\Big|\sum_{i,j,k=1}^{m,n,p}\mathbb{E}\left[\tr(\boldsymbol{Q}^{11}(\bar{z}))^c\mathscr{O}(\partial_{ijk}^{(3)}F_{ijk}^1(z))\right]\Big|\leq\mrO(\eta^{-8}N^{-\omega}\mcC_{11,N}^{(3)}(z,z)).$$
		Besides, when \(\alpha=0\) and \(\alpha=2\), similar to (\ref{Eq of covariance l=2 alpha=0}) and (\ref{Eq of covariance l=2 alpha=1}), we can show that
		$$N^{-1/2}\sum_{i,j,k=1}^{m,n,p}\mathbb{E}\left[\big|F_{ijk}^1(z)\mathscr{O}\big(\partial_{ijk}^{(3)}\tr(\bbQ^{11}(\bar{z}))^c\big)\big|\right],N^{-1/2}\sum_{i,j,k=1}^{m,n,p}\mathbb{E}\left[\big|\partial_{ijk}^{(1)}\tr(\bbQ^{11}(\bar{z}))^c\mathscr{O}\big(\partial_{ijk}^{(2)}F_{ijk}^1(z)\big)\big|\right]\leq\mrO(\eta^{-5}N^{-1}).$$
		Moreover, since \(|F_{ijk}^1(z)|\leq\mrO(\eta^{-1}),|\partial_{ijk}^{(1)}\tr(\bbQ(\bar{z}))|\leq\mrO(\eta^{-2}N^{-1/2})\), then we have
		\begin{align}
			&N^{-1/2}\sum_{i,j,k=1}^{m,n,p}\mathbb{E}\left[\big|F_{ijk}^1(z)\mathscr{D}\big(\partial_{ijk}^{(3)}\tr(\bbQ^{11}(\bar{z}))^c\big)\big|\right]\notag\\
			&\leq \eta^{-1}N^{-2}\sum_{i,j,k=1}^{m,n,p}\mathbb{E}\left[|a_ib_jc_k(Q_{ii}^{11}(\bar{z}))^2Q_{jj}^{22}(\bar{z})Q_{kk}^{33}(\bar{z})|\right]=\mrO(\eta^{-5}N^{-1/2}),\label{Eq of covariance l=3 alpha=0}
		\end{align}
		and
		\begin{align}
			&N^{-1/2}\sum_{i,j,k=1}^{m,n,p}\mathbb{E}\left[\big|\partial_{ijk}^{(1)}\tr(\bbQ^{11}(\bar{z}))^c\mathscr{D}\big(\partial_{ijk}^{(2)}F_{ijk}^1(z)\big)\big|\right]\notag\\
			&\leq2\eta^{-2}N^{-2}\sum_{i,j,k=1}^{m,n,p}\mathbb{E}\left[\big|a_ib_jc_kQ_{ii}^{11}(z)Q_{jj}^{22}(z)Q_{kk}^{33}(z)\big|\right]=\mrO(\eta^{-5}N^{-1/2}).\label{Eq of covariance l=3 alpha=1}
		\end{align}
		Therefore, we only need to consider \(\alpha=1\), i.e. \(\partial_{ijk}^{(2)}\tr(\boldsymbol{Q}^{11}(\bar{z}))\partial_{ijk}^{(1)}F_{ijk}^1(z)\), by Lemma \ref{Cor of minor terms}, we have
		\begin{align}
			&N^{-1/2}\sum_{i,j,k=1}^{m,n,p}\mathbb{E}\left[\partial_{ijk}^{(2)}\tr(\boldsymbol{Q}^{11}(\bar{z}))\partial_{ijk}^{(1)}F_{ijk}^1(z)\right]=\notag\\
			&N^{-1/2}\sum_{i,j,k=1}^{m,n,p}\mathbb{E}\left[\mathscr{D}\big(\partial_{ijk}^{(2)}\tr(\boldsymbol{Q}^{11}(\bar{z}))\big)\mathscr{D}\big(\partial_{ijk}^{(1)}F_{ijk}^1(z)\big)\right]+\mrO(\eta^{-5}N^{-1/2})=\notag\\
			&-2N^{-2}\mathbb{E}\left[\tr(\boldsymbol{Q}^{11}(z){\rm diag}(\boldsymbol{Q}^{11}(\bar{z})^2))\left(\Vert\boldsymbol{c}\Vert_4^4\tr({\rm diag}(|\boldsymbol{Q}^{22}(z)|^{\circ2}))+\Vert\boldsymbol{b}\Vert_4^4\tr({\rm diag}(|\boldsymbol{Q}^{33}(z)|^{\circ2}))\right)\right]\notag\\
			&-2N^{-2}\mathbb{E}\left[\Vert\boldsymbol{c}\Vert_4^4\tr(\boldsymbol{Q}^{22}(z){\rm diag}(\boldsymbol{Q}^{21}(\bar{z})\boldsymbol{Q}^{12}(\bar{z})))\tr(\boldsymbol{Q}^{11}(z){\rm diag}(\boldsymbol{Q}^{11}(\bar{z})))\right]\notag\\
			&-2N^{-2}\mathbb{E}\left[\Vert\boldsymbol{b}\Vert_4^4\tr(\boldsymbol{Q}^{33}(z){\rm diag}(\boldsymbol{Q}^{31}(\bar{z})\boldsymbol{Q}^{13}(\bar{z})))\tr(\boldsymbol{Q}^{11}(z){\rm diag}(\boldsymbol{Q}^{11}(\bar{z})))\right]+\mrO(\eta^{-5}N^{-1/2}).\notag
		\end{align}
		Therefore, we finally obtain
		\begin{align}
			&\frac{\kappa_4}{6\sqrt{N}}\sum_{i,j,k=1}^{m,n,p}\mathbb{E}\left[\partial_{ijk}^{(3)}\big\{F_{ijk}^1(z)\tr(\boldsymbol{Q}^{11}(\bar{z}))\big\}\right]\notag\\
			&=-\kappa_4\widetilde{\mcW}_{11,N}^{(3)}(z,\bar{z})+\mrO(\eta^{-5}N^{-1/2}+\eta^{-8}N^{-\omega}\mcC_{11,N}^{(3)}(z,z)),\label{Eq of covariance 3}
		\end{align}
		where \(\{i,j,k\}=\{1,2,3\}\) and \(\bba^{(l)}\) is defined in (\ref{Eq of vector notation d=3}) and
		\begin{align}
			\widetilde{\mcW}_{ij,N}^{(3)}(z_1,z_2):=&N^{-2}\Vert\bba^{(k)}\Vert_4^4\mbE[\tr(\boldsymbol{Q}^{ii}(z_1)\circ\boldsymbol{Q}^{ii}(z_2))\tr(\boldsymbol{Q}^{jj}(z_1)\circ\boldsymbol{Q}^{jj}(z_2)\boldsymbol{Q}^{jj}(z_2))]\label{Eq of capital W}\\
			+&N^{-2}\Vert\bba^{(k)}\Vert_4^4\mbE[\tr(\boldsymbol{Q}^{jj}(z_1)\circ\boldsymbol{Q}^{jj}(z_2))\tr(\boldsymbol{Q}^{ii}(z_1)\circ(\boldsymbol{Q}^{ij}(z_2)\boldsymbol{Q}^{ji}(z_2)))]\notag\\
			+&N^{-2}\Vert\bba^{(j)}\Vert_4^4\mbE[\tr(\boldsymbol{Q}^{ii}(z_1)\circ\boldsymbol{Q}^{ii}(z_2))\tr(\boldsymbol{Q}^{kk}(z_1)\circ\boldsymbol{Q}^{kj}(z_2)\boldsymbol{Q}^{jk}(z_2))]\notag\\
			+&N^{-2}\Vert\bba^{(j)}\Vert_4^4\mbE[\tr(\boldsymbol{Q}^{kk}(z_1)\circ\boldsymbol{Q}^{kk}(z_2))\tr(\boldsymbol{Q}^{ii}(z_1)\circ(\boldsymbol{Q}^{ij}(z_2)\boldsymbol{Q}^{ji}(z_2)))].\notag
		\end{align}
    Readers can further refer to \S\ref{Sec of majors} for proofs of \(\lim_{N\to\infty}\widetilde{\mcW}_{ij,N}^{(3)}(z_1,z_2)-\mcW_{ij,N}^{(3)}(z_1,z_2)=0\) in (\ref{Eq of mcW limiting d=3}).
    
    \vspace{5mm}
    \noindent
    {\bf Remainders:} We claim the following:
		\begin{align}
			N^{-1/2}\Big|\sum_{i,j,k=1}^{m,n,p}\mbE\left[\partial_{ijk}^{(\alpha)}F_{ijk}^1(z)\partial_{ijk}^{(4-\alpha)}\tr(\bbQ^{11}(\bar{z}))^c\right]\Big|\leq\mrO(\eta^{-6}N^{-1/2}),\ \alpha=0,1,\cdots,4.\label{Eq of covariance 4}
		\end{align}
		First, when \(\alpha=4\), since \(|\tr(\bbQ^{11}(\bar{z}))|\leq\mrO(\eta^{-1}N)\) and \(\sum_{i,j,k=1}^{m,n,p}|\partial_{ijk}^{(4)}F_{ijk}^1(z)|\leq\mrO(\eta^{-5}N^{-3/2})\) by (\ref{Eq of minor remainder}) later, we have 
		$$N^{-1/2}\Big|\sum_{i,j,k=1}^{m,n,p}\tr(\bbQ^{11}(\bar{z}))^c\partial_{ijk}^{(4)}F_{ijk}^1(z)\Big|\leq\mrO(\eta^{-6}N^{-1}).$$
		Next, when \(\alpha=0\ {\rm and\ }3\), by repeating the arguments used for (\ref{Eq of covariance l=2 alpha=0}), \eqref{Eq of covariance l=2 alpha=0 D}, (\ref{Eq of covariance l=2 alpha=1}), (\ref{Eq of covariance l=3 alpha=0}) and (\ref{Eq of covariance l=3 alpha=1}), we can show that 
        \begin{align*}
            &N^{-1/2}\Big|\sum_{i,j,k=1}^{m,n,p}\mbE\big[\partial_{ijk}^{(3)}F_{ijk}^1(z)\partial_{ijk}^{(1)}\tr(\bbQ(\bar{z}))^c\big]\Big|,N^{-1/2}\Big|\sum_{i,j,k=1}^{m,n,p}\mbE\big[F_{ijk}^1(z)\partial_{ijk}^{(4)}\tr(\bbQ(\bar{z}))^c\big]\Big|\leq\mrO(\eta^{-6}N^{-1/2}).
        \end{align*}
        Finally, when \(\alpha=1\ {\rm and\ }2\), by Lemma \ref{Cor of minor terms}, it is enough to show the following terms are minor:
		\begin{align}
			N^{-1/2}\sum_{i,j,k=1}^{m,n,p}\mbE\big[\mathscr{D}\big(\partial_{ijk}^{(2)}F_{ijk}^1(z)\big)\mathscr{D}\big(\partial_{ijk}^{(2)}\tr(\boldsymbol{Q}^{11}(\bar{z}))\big)\big],\quad N^{-1/2}\sum_{i,j,k=1}^{m,n,p}\mbE\big[\mathscr{D}\big(\partial_{ijk}^{(1)}F_{ijk}^1(z)\big)\big(\partial_{ijk}^{(3)}\tr(\boldsymbol{Q}^{11}(\bar{z}))\big)\big].\notag
		\end{align}
		Recall that in Lemma \ref{Lem of minor terms}, if \(n_a,n_b,n_c>0\), then
		\begin{align}
			N^{-5/2}\boldsymbol{1}_m'{\rm diag}(|\boldsymbol{Q}^{11}|^{\circ n_{11}})|\boldsymbol{a}|\boldsymbol{1}_n'{\rm diag}(|\boldsymbol{Q}^{22}|^{\circ n_{11}})|\boldsymbol{b}|\boldsymbol{1}_p'{\rm diag}(|\boldsymbol{Q}^{33}|^{\circ n_{11}})|\boldsymbol{c}|\leq\mrO(\eta^{-6}N^{-1}).\notag
		\end{align}
		Actually, by direct calculations, \(\mathscr{D}(\partial_{ijk}^{(2)}F_{ijk}^1(z)),\mathscr{D}(\partial_{ijk}^{(3)}\tr(\boldsymbol{Q}^{11}(\bar{z})))\) satisfy \(n_a,n_b,n_c>0\). Hence, \(N^{-1/2}\sum_{i,j,k=1}^{m,n,p}|\epsilon_{ijk}^{(4)}|\) is a minor term.

    \vspace{5mm}
    \noindent
	Now, combining (\ref{Eq of covariance 0}), (\ref{Eq of covariance 1}), (\ref{Eq of covariance 2}), (\ref{Eq of covariance 3}) and (\ref{Eq of covariance 4}), we obtain
	\begin{align}
		&(z+\mfm_2(z)+\mfm_3(z)+\mrO(\eta^{-8}N^{-\omega}))\mcC_{11,N}^{(3)}(z,z)=-\mfm_1(z)[\mcC_{21,N}^{(3)}(z,z)+\mcC_{31,N}^{(3)}(z,z)]\notag\\
		&-2\mcV_{11,N}^{(3)}(z,\bar{z})-\kappa_4\mcW_{11,N}^{(3)}(z,\bar{z})+\mrO(\eta^{-6}N^{-1/2}).\notag
	\end{align}
	Similarly, for any \(s,t\in\{1,2,3\}\), we can obtain
	\begin{align}
		(z+\mfm-\mfm_s)\mathcal{C}_{st,N}^{(3)}+\mrO(\eta^{-8}N^{-\omega})\mathcal{C}_{tt,N}^{(3)}=-\mfm_s\sum_{l\neq s}^3\mcC_{lt,N}^{(3)}-(2\mcV_{st,N}^{(3)}(z,\bar{z})+\kappa_4\mcW_{st,N}^{(3)}(z,\bar{z}))+\mrO(\eta^{-6}N^{-1/2}),\notag
	\end{align}
	where we omit the \((z,z)\) in \(\mathcal{C}_{ij,N}^{(3)}(z,z)\) and \((z)\) in \(\mfm_i(z)\) for convenience. Next, define two matrices \(\bbTheta_N^{(3)}(z_1,z_2)=[\Theta_{ij,N}^{(3)}(z_1,z_2)]_{3\times3}\in\mbC^{3\times3}\) and \(\bbF_N^{(3)}(z_1,z_2)=[\mcF_{ij,N}^{(3)}(z_1,z_2)]_{3\times3}\in\mbC^{3\times3}\) such that 
	\begin{align}
		\Theta_{ij,N}^{(3)}(z_1,z_2)=\left\{\begin{array}{cc}
			z_1+\mfm(z_1)-\mfm_i(z_1)&i=j\\
			\mfm_i(z_2)&i\neq j
		\end{array}\right.,\label{Eq of bbTheta d=3}
	\end{align}
	and
	\begin{align}
		\mcF_{ij,N}^{(3)}(z_1,z_2):=2\mcV_{ij,N}^{(3)}(z_1,z_2)+\kappa_4\mcW_{ij,N}^{(3)}(z_1,z_2)\label{Eq of mcF d=3}
	\end{align}
	then we obtain
	\begin{align}
		\bbTheta_N^{(3)}(z,z)\bbC_N^{(3)}(z,z)=-\bbF_N^{(3)}(z,z)+\mrO(\eta^{-8}N^{-\omega})\diag(\bbC_N^{(3)}(z,z))+\mrO(\eta_0^{-6}N^{-1/2})\boldsymbol{1}_{3\times3}.\label{Eq of covariance system equation 1}
	\end{align}
	By Theorem \ref{Thm of approximation}, \(\Vert\bbm(z)-\bbg(z)\Vert_{\infty}=\mrO(\eta_0^{-15}N^{-2\omega})\), it implies that
	$$\Vert\bbTheta_N^{(3)}(z,z)+\diag(\boldsymbol{\mfc}\circ\bbg(z)^{-1})\bbPi^{(3)}(z,z)\Vert_{\infty}\leq\mrO(\eta_0^{-15}N^{-2\omega}),$$
	where \(\bbPi(z,z)\) is defined in (\ref{Eq of invertible 2}). By Lemma \ref{Lem of contraction}, we have
	$$\Vert\diag(\boldsymbol{\mfc}\circ\bbg(z)^{-1})\bbPi^{(3)}(z,z)\Vert\geq3^{-1/2}\Vert\diag(\boldsymbol{\mfc}\circ\bbg(z)^{-1})\bbPi^{(3)}(z,z)\Vert_F\geq C_{\boldsymbol{\mfc}}\eta_0^3,$$
	then \(\diag(\boldsymbol{\mfc}\circ\bbg(z)^{-1})\bbPi^{(3)}(z,z)\) is the dominating term as \(N\to\infty\). Moreover, \(\bbPi^{(3)}(z,z)\) is invertible by the Remark \ref{Rem of invertible submatrix}, which implies that \(\bbTheta_N^{(3)}(z,z)\) is also invertible as \(N\to\infty\). Hence,
	\begin{align}
		\bbC_N^{(3)}(z,z)&=-\bbTheta_N^{(3)}(z,z)^{-1}\bbF_N^{(3)}(z,z)\label{Eq of bbC_N d=3}\\
		&+\mrO(\eta^{-8}N^{-\omega})\bbTheta_N^{(3)}(z,z)^{-1}\diag(\bbC_N^{(3)}(z,z))+\mrO(\eta_0^{-6}N^{-1/2})\boldsymbol{1}_{3\times3}.\notag
	\end{align}
	Let \(\Delta^{(3)}:=\bbTheta_N^{(3)}(z,z)+\diag(\boldsymbol{\mfc}\circ\bbg(z)^{-1})\bbPi^{(3)}(z,z)\), where \(\Vert\Delta^{(3)}\Vert_{\infty}=\mrO(\eta_0^{-15}N^{-2\omega})\). Since
	\begin{align}
		&\bbTheta_N^{(3)}(z,z)^{-1}=-(\diag(\boldsymbol{\mfc}\circ\bbg(z)^{-1})\bbPi^{(3)}(z,z)-\Delta^{(3)})^{-1}\notag\\
		&=-\bbPi^{(3)}(z,z)^{-1}\diag(\boldsymbol{\mfc}^{-1}\circ\bbg(z))-\bbPi^{(3)}(z,z)^{-1}\diag(\boldsymbol{\mfc}^{-1}\circ\bbg(z))\Delta^{(3)}\bbTheta_N^{(3)}(z,z)^{-1},\notag
	\end{align}
	then
	\begin{align*}
		\Vert\bbTheta_N^{(3)}(z,z)^{-1}\Vert\leq(1-\Vert\bbPi^{(3)}(z,z)^{-1}\diag(\boldsymbol{\mfc}^{-1}\circ\bbg(z))\Vert\cdot\Vert\Delta^{(3)}\Vert)^{-1}\Vert\bbPi^{(3)}(z,z)^{-1}\diag(\boldsymbol{\mfc}^{-1}\circ\bbg(z))\Vert.
	\end{align*}
	By Proposition \ref{Pro of inverse norm} and Lemma \ref{Lem of contraction}, it implies that \(\Vert\bbPi^{(3)}(z,z)^{-1}\diag(\boldsymbol{\mfc}^{-1}\circ\bbg(z))\Vert\leq C_{\mfc}\eta_0^{-5}\). Moreover, since \(\lim_{N\to\infty}\Vert\Delta^{(3)}\Vert=0\), we obtain that \(\Vert\bbTheta_N^{(3)}(z,z)^{-1}\Vert\leq C_{\mfc}\eta_0^{-5}\). By (\ref{Eq of bbC_N d=3}), it further gives that
	\begin{align*}
		\Vert\bbC_N^{(3)}(z,z)\Vert\leq(1-\mrO(\eta_0^{-8}N^{-\omega})\Vert\bbTheta_N^{(3)}(z,z)^{-1}\Vert)^{-1}(\Vert\bbTheta_N^{(3)}(z,z)^{-1}\bbF_N^{(3)}(z,z)\Vert+\mrO(\eta_0^{-6}N^{-1/2})),
	\end{align*}
	then we obtain that 
    \begin{align}
        \Vert\bbC_N^{(3)}(z,z)\Vert\leq C_{\mfc}\eta_0^{-10},\label{Eq of covariance bound d=3}
    \end{align}
    so do its entries. Moreover, since
	\begin{align}
		&\Vert\bbTheta_N^{(3)}(z,z)^{-1}+\bbPi^{(3)}(z,z)^{-1}\diag(\boldsymbol{\mfc}^{-1}\circ\bbg(z))\Vert\notag\\
		&\leq\Vert\bbPi^{(3)}(z,z)^{-1}\diag(\boldsymbol{\mfc}^{-1}\circ\bbg(z))\Delta^{(3)}\bbTheta_N^{(3)}(z,z)^{-1}\Vert\leq C_{\mfc}\eta_0^{-10}\Vert\Delta^{(3)}\Vert,\notag
	\end{align}
	then we replace all \(\bbTheta_N^{(3)}(z,z)^{-1}\) in (\ref{Eq of bbC_N d=3}) by \(-\bbPi^{(3)}(z,z)^{-1}\diag(\boldsymbol{\mfc}^{-1}\circ\bbg(z))\) and derive that
	\begin{align}
		\lim_{N\to\infty}\Vert\bbC_N^{(3)}(z,z)-\bbPi^{(3)}(z,z)^{-1}\diag(\mfc^{-1}\circ\bbg(z))\bbF_N^{(3)}(z,z)\Vert=0.\label{Eq of covariance system equation 2}
	\end{align}
	Finally, for \(z_1\neq z_2\in\mathcal{S}_{\eta_0}\), we can repeat the previous arguments to derive that
	\begin{align*}
		&(z_1+\mfm(z_1)-\mfm_s(z_1))\mathcal{C}_{st,N}^{(3)}(z_1,z_2)\\
    &=-\mfm_s(z_1)\sum_{l\neq s}^3\mcC_{lt,N}^{(3)}(z_1,z_2)-\mcF_{st,N}^{(3)}(z_1,\bar{z}_2)+\mrO(C_{\eta_0}N^{-\omega})\mcC_{tt,N}^{(3)}(z_2,z_2)+\mrO(C_{\eta_0}N^{-\omega}),
	\end{align*}
	since we have shown that \(\mcC_{tt,N}^{(3)}(z_2,z_2)\) is bounded by \(C_{\eta_0,\mfc,d}\), then we can repeat the previous arguments to derive that
	\begin{align}
		\lim_{N\to\infty}\Vert\bbC_N^{(3)}(z_1,z_2)-\bbPi^{(3)}(z_1,z_2)^{-1}\diag(\mfc^{-1}\circ\bbg(z_1))\bbF_N^{(3)}(z_1,\bar{z}_2)\Vert=0,\notag
	\end{align}
	here we omit the details.
\end{proof}
\subsection{Mean function}\label{ssec of mean function}
In this subsection, we derive the limiting form of \(\mbE[\tr(\bbQ(z))]-Ng(z)\). For convenience, we omit \((z)\) in \(\bbQ(z)\) and its entries.
\begin{thm}\label{Thm of Mean function}
	Under Assumptions {\rm \ref{Ap of general noise}} and {\rm \ref{Ap of dimension}}, for any \(\eta_0>0\) and \(z\in\mathcal{S}_{\eta_0}\) in {\rm (\ref{Eq of stability region})}, let 
	$$\overrightarrow{M}_N^{(3)}(z)=(M_{1,N}^{(3)}(z),M_{2,N}^{(3)}(z),M_{3,N}^{(3)}(z))'$$
	where for \(1\leq i\leq3\)
	\begin{align}
		&M_{i,N}^{(3)}(z):=g_i(z)\sum_{r\neq i}^3\sum_{w\neq i,r}^3W_{rw}^{(3)}(z)+\sum_{l\neq i}^3\big[(g(z)-g_i(z)-g_l(z))W_{il}^{(3)}(z)+V_{il}^{(3)}(z,z)\big]\notag\\
		&-2\kappa_3G_N^{(3)}(z)+\kappa_4H_{i,N}^{(3)}(z,z),\label{Eq of MiN}
	\end{align}
	where \(W_{jk}^{(3)}(z),V_{ij}^{(3)}(z,z),G_N^{(3)}(z),H_{i,N}^{(3)}(z,z)\) are defined in {\rm (\ref{Eq of bbW limiting d=3}), (\ref{Eq of bbV limit d=3}), (\ref{Eq of H2 d=3}), (\ref{Eq of H3 d=3})}. Then we have
    $$\lim_{N\to\infty}\Vert N(\bbm(z)-\bbg(z))-\boldsymbol{\Pi}^{(3)}(z,z)^{-1}\diag(\boldsymbol{\mfc}^{-1}\circ\bbg(z))\overrightarrow{M}_N^{(3)}(z)\Vert=0.$$
    Consequently, we obtain that 
    $$\lim_{N\to\infty}\mbE[\tr(\bbQ(z))]-Ng(z)-\mu_N^{(3)}(z)=0,$$
    where
	\begin{align}
		\mu_N^{(3)}(z):=\boldsymbol{1}_3'\boldsymbol{\Pi}^{(3)}(z,z)^{-1}\diag(\boldsymbol{\mfc}^{-1}\circ\bbg(z))\overrightarrow{M}_N^{(3)}(z),\label{Eq of Mean function}
	\end{align}
	and \(\boldsymbol{\Pi}^{(3)}(z,z)\) is defined in {\rm (\ref{Eq of invertible 2})}.
\end{thm}
Before proving the above theorem, we first give the explicit forms of major terms in cumulant expansions of \(\mbE[\tr(\bbQ]\). It suffices to calculate \(\mbE[\tr(\bbQ^{ii}]\) for \(i=1,2,3\). Without loss of generality, we only calculate \(\mbE[\tr(\bbQ^{11})]\) in detail. By notations in (\ref{Eq of vector notation d=3}), the trick of \(z\bbQ=\bbQ\bbM-\bbI_N\) and the cumulant expansion (\ref{Eq of cumulant expansion}), we have
\begin{align*}
	&z\mbE[\tr(\bbQ^{11}]=\frac{1}{\sqrt{N}}\sum_{i,j,k=1}^{m,n,p}\mbE[X_{ijk}(b_jQ_{ik}^{13}+c_kQ_{ij}^{12})]-m\\
	&=\frac{1}{\sqrt{N}}\sum_{i,j,k=1}^{m,n,p}\Big(\sum_{\alpha=0}^3\frac{\kappa_{l+1}}{l!}\mbE\big[b_j\partial_{ijk}^{(l)}Q_{ik}^{13}+c_k\partial_{ijk}^{(l)}Q_{ij}^{12}\big]+\epsilon_{ijk}^{(4)}\Big)-m,
\end{align*}
where \(|\epsilon_{ijk}^{(4)}|\leq C_{\kappa_5}\sup_{z\in\mcS_{\eta_0}}\big|b_j\partial_{ijk}^{(4)}Q_{ik}^{13}+c_k\partial_{ijk}^{(4)}Q_{ij}^{12}\big|\). By Lemma \ref{Lem of minor terms}, for \(2\leq l\leq4\), we have 
\begin{align}
    N^{-1/2}\sum_{i,j,k=1}^{m,n,p}\Big|\mathscr{O}\left(b_j\partial_{ijk}^{(l)}Q_{ik}^{13}+c_k\partial_{ijk}^{(l)}Q_{ij}^{12}\right)\Big|\leq\mrO(\eta_0^{-(l+1)}N^{-(l+1)/2+1}),\label{Eq of off diagonal minor}
\end{align}
so it is enough to focus on
\begin{align}
	N^{-1/2}\sum_{i,j,k=1}^{m,n,p}\mathscr{D}\left(b_j\partial_{ijk}^{(l)}Q_{ik}^{13}+c_k\partial_{ijk}^{(l)}Q_{ij}^{12}\right),\label{Eq of mean majors}
\end{align}
where the operators ``\(\mathscr{D},\mathscr{O}\)'' are defined in (\ref{Eq of operator D}) and (\ref{Eq of operator O}).

\vspace{5mm}
\noindent
{\bf First derivatives:} When \(l=1\), by direct calculations, we can obtain that
	\begin{align}
		&N^{-1/2}\sum_{i,j,k=1}^{m,n,p}\mathbb{E}\left[b_j\partial_{ijk}^{(1)}Q_{ik}^{13}+c_k\partial_{ijk}^{(1)}Q_{ij}^{12}\right]\notag\\
		=&-N^{-1}\mathbb{E}\big[\tr(\boldsymbol{Q}^{11})\tr(\boldsymbol{Q}^{22})+\tr(\boldsymbol{Q}^{11})\tr(\boldsymbol{Q}^{33})\big]-\big(V_{12,N}^{(3)}(z,z)+V_{13,N}^{(3)}(z,z)\big)\notag\\
		-&N^{-1}\mathbb{E}\big[2\bbb'\bbQ^{23}\bbc\tr(\bbQ^{11})+\boldsymbol{a}'\boldsymbol{Q}^{12}\boldsymbol{b}\tr(\boldsymbol{Q}^{33})+\boldsymbol{a}'\boldsymbol{Q}^{13}\boldsymbol{c}\tr(\boldsymbol{Q}^{22})\big]+\mathrm{O}(\eta_0^{-2}N^{-1}),\notag
	\end{align}
	where  
	\begin{align}
		V_{ij,N}^{(3)}(z_1,z_2):=N^{-1}\mathbb{E}[\operatorname{tr}(\bbQ^{ij}(z_1)\bbQ^{ji}(z_2))],\ 1\leq i\leq j\leq3.\label{Eq of V}
	\end{align}
{\bf Second derivatives:} When \(l=2\), we have
	\begin{align}
		&N^{-1/2}\sum_{i,j,k=1}^{m,n,p}\mathbb{E}\left[\mathscr{D}\left(b_j\partial_{ijk}^{(2)}Q_{ik}^{13}+c_k\partial_{ijk}^{(2)}Q_{ij}^{12}\right)\right]\label{Eq of mean function 2 d=3}\\
		&=4N^{-3/2}\sum_{i,j,k=1}^{m,n,p}a_ib_jc_k\mathbb{E}\big[Q_{ii}^{11}Q_{jj}^{22}Q_{kk}^{33}\big]+\mathrm{O}(\eta_0^{-3}N^{-1/2})\notag\\
		&=4N^{-3/2}\mathbb{E}\left[\left(\boldsymbol{1}_m'{\rm diag}(\boldsymbol{Q}^{11})\boldsymbol{a}\right)\left(\boldsymbol{1}_n'{\rm diag}(\boldsymbol{Q}^{22})\boldsymbol{b}\right)\left(\boldsymbol{1}_p'{\rm diag}(\boldsymbol{Q}^{33})\boldsymbol{c}\right)\right]+\mathrm{O}(\eta_0^{-3}N^{-1/2}).\notag
	\end{align}
	By Lemma \ref{Thm of Entrywise almost sure convergence} and Theorem \ref{Thm of entrywise law d=3}, we can show that
	$$\Big|N^{-1/2}\boldsymbol{1}'\diag(\bbQ^{kk}(z))\bba^{(k)}-\mfc_k^{-1}g_k(z)N^{-1/2}\sum_{j=1}^{n_k}a_j^{(k)}\Big|\prec\mrO(C_{\eta_0}N^{-\omega}),$$
	by (\ref{Eq of Ap L1 d=3}), we define
	\begin{align}
		G_N^{(3)}(z):=(\mfc_1\mfc_2\mfc_3)^{-1}g_1(z)g_2(z)g_3(z)\mfb_1^{(1)}\mfb_2^{(1)}\mfb_3^{(1)},\label{Eq of H2 d=3}
	\end{align}
	and \((\ref{Eq of mean function 2 d=3})=6G_N^{(3)}(z)+\mrO(C_{\eta_0}N^{-\omega})\).

\vspace{5mm}
\noindent
{\bf Third derivatives:} When \(l=3\), there are only two possible situations in (\ref{Eq of mean majors}). Without loss of generality, we use \(\mathscr{D}(c_k\partial_{ijk}^{(3)}Q_{ij}^{12})\) as an example. First, if \(n_{11},n_{22},n_{33}>0\), then the terms \(a_i^2c_k^2Q_{ii}^{11}(Q_{jj}^{22})^2Q_{kk}^{33}\) and \(b_j^2c_k^2(Q_{ii}^{11})^2Q_{jj}^{22}Q_{kk}^{33}\) will appear in \(\mathscr{D}(c_k\partial_{ijk}^{(3)}Q_{ij}^{12})\), and we can conclude that
	$$N^{-2}\sum_{i,j,k=1}^{m,n,p}a_i^2c_k^2Q_{ii}^{11}(Q_{jj}^{22})^2Q_{kk}^{33}\leq\bba'{\rm diag}(\bbQ^{11})\bba\times\bbc'{\rm diag}(\bbQ^{33})\bbc\times\tr((\bbQ^{22})^{\circ 2})$$
	which is bounded by \(N^{-1}\Vert\boldsymbol{Q}\Vert^4\), as does $N^{-2}\sum_{i,j,k=1}^{m,n,p}a_i^2c_k^2Q_{ii}^{11}(Q_{jj}^{22})^2Q_{kk}^{33}$. And the same conclusion is also valid for \(b_j\partial_{ijk}^{(3)}Q_{ik}^{13}\). On the other hand, if only two of \(n_{11},n_{22},n_{33}\) are nonzero, then the only possible case in \(\mathscr{D}(c_k\partial_{ijk}^{(3)}Q_{ij}^{12})\) is \(c_k^4(Q_{ii}^{11})^2(Q_{jj}^{22})^2\), and we finally obtain
	\begin{align}
		&N^{-2}\sum_{i,j,k=1}^{m,n,p}\mathbb{E}\left[\mathscr{D}\left(b_j\partial_{ijk}^{(3)}Q_{ik}^{13}+c_k\partial_{ijk}^{(3)}Q_{ij}^{12}\right)\right]\notag\\
		&=-6N^{-2}\sum_{i,j,k=1}^{m,n,p}\mathbb{E}\left[b_j^4(Q_{ii}^{11})^2(Q_{kk}^{33})^2+c_k^4(Q_{ii}^{11})^2(Q_{jj}^{22})^2\right]+\mathrm{O}(\eta_0^{-4}N^{-1})\notag\\
		&=-6\sum_{k=2}^3\Vert\bba^{(5-k)}\Vert_4^4N^{-2}\mathbb{E}\left[\operatorname{tr}(\bbQ^{11}(z)\circ\bbQ^{11}(z))\operatorname{tr}(\bbQ^{jj}(z)\circ\bbQ^{jj}(z))\right]+\mrO(\eta_0^{-4}N^{-1}).\notag
	\end{align}
	By Lemma \ref{Thm of Entrywise almost sure convergence} and Theorem \ref{Thm of entrywise law d=3}, we know that
	$$\Cov(N^{-1}\tr(\bbQ^{ii}(z)\circ\bbQ^{ii}(z)),N^{-1}\tr(\bbQ^{jj}(z)\circ\bbQ^{jj}(z)))=\mrO(\eta_0^{-12}N^{-2\omega})$$
	and
	$$\big|N^{-1}\tr(\bbQ^{ii}(z_1)\circ\bbQ^{ii}(z_2))-\mfc_i^{-1}g_i(z_1)g_i(z_2)\big|\prec\mrO(C_{\eta_0}N^{-\omega}).$$
	By (\ref{Eq of Ap L4 d=3}), for \(1\leq i\leq 3\), let
	\begin{align}
		H_{i,N}^{(3)}(z_1,z_2):=\mfc_i^{-1}g_i(z_1)g_k(z_2)\sum_{l\neq i}^3\mcB_{(4)}^{(k,l)}\mfc_l^{-1}g_l(z_1)g_l(z_2).\label{Eq of H3 d=3}
	\end{align}
{\bf Remainders:} By Lemma \ref{Lem of minor terms}, we rewrite (\ref{Eq of mean majors}) as follows:
	$$N^{-1/2}\sum_{i,j,k=1}^{m,n,p}\mathscr{D}\left(b_j\partial_{ijk}^{(l)}Q_{ik}^{13}+c_k\partial_{ijk}^{(l)}Q_{ij}^{12}\right)=N^{-5/2}\sum_{i,j,k=1}^{m,n,p}(a_i)^{n_a}(b_j)^{n_b}(c_k)^{n_c}(Q_{ii}^{11})^{n_{11}}(Q_{jj}^{22})^{n_{22}}(Q_{kk}^{33})^{n_{33}},$$
	where \(n_a+n_b+n_c=5\). Hence, at least one of \(n_a,n_b,n_c\) is equal or greater than 2 and the above sum is bounded by \(N^{-3/2}\Vert\boldsymbol{Q}\Vert^5\). Combined with (\ref{Eq of off diagonal minor}), we conclude that
	\begin{align}
	    &N^{-1/2}\sum_{i,j,k=1}^{m,n,p}|\epsilon_{ijk}^{(4)}|\leq C_{\kappa_5}N^{-1/2}\sum_{i,j,k=1}^{m,n,p}\big|(\mathscr{D}+\mathscr{O})\big(b_j\partial_{ijk}^{(l)}Q_{ik}^{13}+c_k\partial_{ijk}^{(l)}Q_{ij}^{12}\big)\big|\leq\mrO(\eta_0^{-5}N^{-3/2}).\label{Eq of minor remainder}
	\end{align}
    Now, let us  prove Theorem \ref{Thm of Mean function}.
\begin{proof}[Proof of Theorem \ref{Thm of Mean function}]
	Based on previous discussions, we obtain
	\begin{align*}
	    &z\mathbb{E}[\tr(\boldsymbol{Q}^{11}(z))]=-N^{-1}\mathbb{E}\big[\tr(\boldsymbol{Q}^{11})\tr(\boldsymbol{Q}^{22})+\tr(\boldsymbol{Q}^{11})\tr(\boldsymbol{Q}^{33})\big]-\big(V_{12,N}^{(3)}(z,z)+V_{13,N}^{(3)}(z,z)\big)\\
        &-N^{-1}\mathbb{E}\big[2\bbb'\bbQ^{23}\bbc\tr(\bbQ^{11})+\boldsymbol{a}'\boldsymbol{Q}^{12}\boldsymbol{b}\tr(\boldsymbol{Q}^{33})+\boldsymbol{a}'\boldsymbol{Q}^{13}\boldsymbol{c}\tr(\boldsymbol{Q}^{22})\big]\\
        &-m+2\kappa_3G_N^{(3)}(z)-\kappa_4H_{1,N}^{(3)}(z,z)+\mrO(C_{\eta_0}N^{-\omega}).
	\end{align*}
	By Lemma \ref{Thm of Entrywise almost sure convergence} and Theorem \ref{Thm of Variance}, for \(1\leq i,j\leq3\), we have \(|(\bba'\bbQ^{12}\bbb)^c|,|(\bba'\bbQ^{13}\bbc)^c|,|(\bbb'\bbQ^{23}\bbc)^c|\prec\mrO(\eta_0^{-5}N^{-\omega})\) and \(|N^{-1}\tr(\bbQ^{12}\bbQ^{21})^c|,|N^{-1}\tr(\bbQ^{13}\bbQ^{31})^c|\prec\mrO(\eta_0^{-6}N^{-\omega})\). Combined with the conclusion \(\Cov(\tr(\bbQ^{11},\tr(\bbQ^{22})))\leq\mrO(\eta_0^{-10})\) in \eqref{Eq of covariance bound d=3}, it implies that
	$$\big|N^{-1}\mathbb{E}\left[\tr(\boldsymbol{Q}^{11}(z))\tr(\boldsymbol{Q}^{22}(z))\right]-\mfm_2(z)\mathbb{E}[\tr(\boldsymbol{Q}^{11}(z))]\big|\leq\mrO(\eta_0^{-10}N^{-1})$$
	and
	$$\big|\operatorname{Cov}(\rho_l(z),(\bba^{(i)})'\bbQ^{ij}\bba^{(j)})\big|\leq\mrO(\eta_0^{-10}N^{-2\omega}),$$
	where \(\rho_i(z)=N^{-1}\tr(\bbQ^{ii}(z))\) and \(\mfm_i(z)=\mbE[\rho_i(z)]\). Hence, recall the definition of \(W_{ij,N}^{(3)}(z)\) in (\ref{Eq of W}), we can obtain
	\begin{align}
		(z+\mfm_2+\mfm_3)\mathbb{E}\left[\tr(\boldsymbol{Q}^{11})\right]&=-(m+2\mfm_1W_{23,N}^{(3)}+\mfm_2W_{13,N}^{(3)}+\mfm_3W_{12,N}^{(3)}+V_{12,N}^{(3)}+V_{13,N}^{(3)})\notag\\
		&+2\kappa_3G_N^{(3)}(z)-\kappa_4H_{1,N}^{(3)}(z,z)+\mrO(C_{\eta_0}N^{-\omega})\notag\\
		&:=-(\mfc_1 N+M_{1,N}^{(3)}(z))+\mrO(C_{\eta_0}N^{-\omega}),\label{Eq of trQ11}
	\end{align}
	where we omit \((z,z)\) in \(V_{ij,N}^{(3)}(z,z)\) and \((z)\) in \(W_{ij,N}^{(3)}(z)\) for convenience. Moreover, for proofs of \(W_{st,N}^{(3)}(z)\to W_{st}^{(3)}(z)\) and \(V_{st,N}^{(3)}(z_1,z_2)\to V_{st}^{(3)}(z_1,z_2)\), readers can find details in \eqref{Eq of solve W} and \eqref{Eq of bbV limit d=3} in \S\ref{Sec of majors}. Next, we can repeat previous arguments to obtain the similar results for \(\mathbb{E}[\tr(\boldsymbol{Q}^{22})]\) and \(\mathbb{E}[\tr(\boldsymbol{Q}^{33})]\) as follows:
	\begin{align}
		&(z+\mfm_j(z)+\mfm_k(z))\mathbb{E}\left[\tr(\bbQ^{ii}(z))\right]=-(\mfc_iN+M_{i,N}^{(3)}(z))+\mrO(C_{\eta_0}N^{-\omega}),\notag
	\end{align}
	where \(\{i,j,k\}=\{1,2,3\}\), i.e.
	$$(z+\mfm_j(z)+\mfm_k(z))\mfm_i(z)=-\mfc_i-N^{-1}M_{i,N}^{(3)}(z)+\mrO(C_{\eta_0}N^{-\omega}).$$
	Next, let \(\bbh(z):=N(\bbm(z)-\bbg(z)),h(z):=N\sum_{i=1}^3(\mfm_i(z)-g_i(z))=N(\mfm(z)-g(z))\), then we have
	\begin{align}
		(z+\mfm_j(z)+\mfm_k(z))h_i(z)&=-\mfc_iN-M_{i,N}^{(3)}(z)-N(z+\mfm_j(z)+\mfm_k(z))g_i(z)+\mrO(C_{\eta_0}N^{-\omega})\notag\\
		&=\mfc_iN\left(\frac{z+\mfm_j(z)+\mfm_k(z)}{z+g_j(z)+g_k(z)}-1\right)-M_{i,N}^{(3)}(z)+\mrO(C_{\eta_0}N^{-\omega})\notag\\
		&=\mfc_i\frac{h(z)-h_i(z)}{z+g_j(z)+g_k(z)}-M_{i,N}^{(3)}(z)+\mrO(C_{\eta_0}N^{-\omega})\notag\\
		&=-g_i(z)(h(z)-h_i(z))-M_{i,N}^{(3)}(z)+\mrO(C_{\eta_0}N^{-\omega}).\label{Eq of mean 1}
	\end{align}
	Let \(\overrightarrow{M}_N^{(3)}(z)=(M_{1,N}^{(3)}(z),\cdots,M_{3,N}^{(3)}(z))'\), we have
	\begin{align}
		{\rm (\ref{Eq of mean 1})}\ \Rightarrow\ &(z+\mfm_j(z)+\mfm_i(z))h_i(z)=-g_i(z)(\bbS_d\bbh(z))_i-M_{i,N}^{(3)}(z)+\mrO(C_{\eta_0}N^{-\omega})\notag\\
		\Leftrightarrow\ &\bbTheta_N^{(3)}(z,z)\tilde{\bbh}(z)=-\overrightarrow{M}_N^{(3)}(z)+\mrO(C_{\eta_0}N^{-\omega}),\notag
	\end{align}
	where \(\bbTheta_N^{(3)}(z,z)\) is defined in (\ref{Eq of bbTheta d=3}) and it is invertible such that
	$$\lim_{N\to\infty}\Vert\bbTheta_N^{(3)}(z,z)^{-1}+\bbPi^{(3)}(z,z)^{-1}\diag(\mfc^{-1}\circ\bbg(z))\Vert=0,$$
	where \(\bbPi^{(3)}(z,z)\) is defined in (\ref{Eq of invertible 2}) and it is also invertible by Remark \ref{Rem of invertible submatrix}, so we obtain
	\begin{align}
		\lim_{N\to\infty}\Vert N(\bbm(z)-\bbg(z))-\boldsymbol{\Pi}^{(3)}(z,z)^{-1}\diag(\boldsymbol{\mfc}^{-1}\circ\bbg(z))\overrightarrow{M}_N^{(3)}(z)\Vert_{\infty}=0,\label{Eq of mean function d=3}
	\end{align}
	and \(\lim_{N\to\infty}\mbE[\tr(\bbQ(z))]-Ng(z)-\mu_N^{(3)}(z)=0\).
\end{proof}
\subsection{System equations for \texorpdfstring{\eqref{Eq of W}}{(E.40)}, \texorpdfstring{\eqref{Eq of V}}{(F.30)}, \texorpdfstring{\eqref{Eq of capital V}}{(F.9)} and \texorpdfstring{\eqref{Eq of capital W}}{(F.18)}}\label{Sec of majors}
We now derive the asymptotic forms of the major terms appearing in \(\mu_N^{(3)}(z_1)\) and the covariance function \(\mcC_N^{(3)}(z_1,z_2)\) for \(z_1,z_2\in\mathcal{S}_{\eta_0}\). Similar to proofs of Theorems \ref{Thm of Variance} and \ref{Thm of Mean function}, the calculation procedures rely on the cumulant expansions (\ref{Eq of cumulant expansion}). For brevity, we omit the detailed calculations of minor terms, such as the remainders \(\epsilon_{ijk}^{(2)}\).

\vspace{5mm}
\noindent
{\bf System equations for \(W_{st,N}^{(3)}(z)\) in (\ref{Eq of W}):} By the cumulant expansion (\ref{Eq of cumulant expansion}) and (\ref{Eq of Qii11}), we can obtain
	\begin{align}
		\mbE[(\bba^{(s)})'\bbQ^{st}(z)\bba^{(t)}(z+\rho(z)-\rho_s(z))]=-\delta_{st}-\sum_{l\neq s}^3\mbE[(\rho(z)-\rho_s(z)-\rho_l(z))(\bba^{(l)})'\bbQ^{lt}(z)\bba^{(t)}]+\mrO(\eta_0^{-3}N^{-1/2}),\notag
	\end{align}
	where \(\rho_i(z)\) is defined in (\ref{Eq of mi}). By Lemma \ref{Thm of Entrywise almost sure convergence}, we have
	$$W_{st,N}^{(3)}(z)(z+\mfm(z)-\mfm_s(z))=-\delta_{st}-\sum_{l\neq s}^3(\mfm(z)-\mfm_s(z)-\mfm_l(z))W_{lt,N}^{(3)}(z)+\mrO(\eta_0^{-6}N^{-\omega}).$$
	Since all \(|W_{il,N}^{(3)}(z)|\leq\mrO(\eta_0^{-1})\) and \(\Vert\bbg(z)-\bbm(z)\Vert_{\infty}=\mrO(\eta_0^{-15}N^{-2\omega})\) by Theorem \ref{Thm of approximation}, we have
	$$W_{st,N}^{(3)}(z)(z+g(z)-g_s(z))=-\delta_{st}-\sum_{l\neq s}^3(g(z)-g_s(z)-g_l(z))W_{lt,N}^{(3)}(z)+\mrO(\eta_0^{-16}N^{-\omega}),$$
	Therefore, define \(\bbW_N^{(3)}(z)=[W_{st,N}^{(3)}(z)]_{3\times3}\) and
	\begin{align*}
		\bbGa^{(3)}(z):=(z+g(z))\bbI_3-\diag(\bbg(z))+g(z)\bbS_3-\diag(g(z))\bbS_3-\bbS_3\diag(g(z)),
	\end{align*}
	we can obtain that
	$$\bbGa^{(3)}(z)\bbW_N^{(3)}(z)=-\bbI_3+\mrO(\eta_0^{-16}N^{-\omega})\boldsymbol{1}_{3\times3}.$$
	For the invertibility of \(\bbGa^{(3)}(z)\), readers can refer to Lemma \ref{Lem of invertible bbGa} later, so we can derive the limiting expression of \(\bbW_N^{(3)}(z)\) as follows:
	\begin{align}
		\bbW^{(3)}(z):=\lim_{N\to\infty}\bbW_N^{(3)}(z)=-\bbGa^{(3)}(z)^{-1},\quad\Vert\bbW^{(3)}(z)-\bbW_N^{(3)}(z)\Vert\leq\mrO(\eta_0^{-17}N^{-\omega}).\label{Eq of solve W}
	\end{align}
{\bf System equations for \(V_{st,N}^{(3)}(z_1,z_2)\) in (\ref{Eq of V}):} By the cumulant expansion (\ref{Eq of cumulant expansion}), Lemma \ref{Lem of minor terms} and Theorem \ref{Thm of approximation}, we obtain
	\begin{align*}
		&z_1V_{st,N}^{(3)}(z_1,z_2)=\frac{1}{N^{3/2}}\sum_{l\neq s}^3\sum_{i,j,k=1}^{m,n,p}\mbE\big[X_{ijk}\mcA_{ijk}^{(s,l)}Q_{\tilde{l}\cdot}^{lt}(z_1)Q_{\cdot \tilde{s}}^{ts}(z_2)\big]-\delta_{st}\mfm_s(z_2)=-\frac{1}{N^2}\sum_{l\neq s}^3\sum_{r_1\neq r_2}^3\\
		&\Big(\sum_{i,j,k=1}^{m,n,p}\mbE\big[\mcA_{ijk}^{(s,l)}\mcA_{ijk}^{(r_1,r_2)}(Q_{\tilde{l}\tilde{r}_1}^{lr_1}(z_1)Q_{\tilde{r}_2\cdot}^{r_2t}(z_1)Q_{\cdot \tilde{s}}^{ts}(z_2)+Q_{\tilde{l}\cdot}^{lt}(z_1)Q_{\cdot \tilde{r}_1}^{tr_1}(z_2)Q_{\tilde{r}_2\tilde{s}}^{r_2s}(z_2))\big]+\epsilon_{ijk}^{(2)}\Big)-\delta_{st}\mfm_s(z_2)\\
		&=-\frac{1}{N^2}\sum_{l\neq s}^3\sum_{i,j,k=1}^{m,n,p}\mbE\big[(\mcA_{ijk}^{(s,l)})^2(Q_{\tilde{l}\tilde{l}}^{ll}(z_1)Q_{\tilde{s}\cdot}^{st}(z_1)Q_{\cdot \tilde{s}}^{ts}(z_2)+Q_{\tilde{l}\cdot}^{lt}(z_1)Q_{\cdot \tilde{l}}^{tl}(z_2)Q_{\tilde{s}\tilde{s}}^{ss}(z_2))\big]-\delta_{st}\mfm_s(z_2)+\mrO(C_{\eta_0}N^{-1/2})\\
		&=-\sum_{l\neq s}^3\mfm_l(z_1)V_{st,N}^{(3)}(z_1,z_2)-\sum_{l\neq s}^3\mfm_s(z_2)V_{lt,N}^{(3)}(z_1,z_2)-\delta_{st}\mfm_s(z_2)+\mrO(C_{\eta_0}N^{-\omega})\\
		&=-V_{st,N}^{(3)}(z_1,z_2)\sum_{l\neq s}^3g_l(z_1)-g_s(z_2)\sum_{l\neq s}^3V_{lt,N}^{(3)}(z_1,z_2)-\delta_{st}g_s(z_2)+\mrO(C_{\eta_0}N^{-\omega}),
	\end{align*}
	i.e.
	$$V_{st,N}^{(3)}(z_1,z_2)=\mfc_s^{-1}g_s(z_1)g_s(z_2)\Big(\delta_{st}+\sum_{l\neq s}^3V_{lt,N}^{(3)}(z_1,z_2)\Big)+\mrO(C_{\eta_0}N^{-\omega}).$$
	Here, define
	\begin{align}
		\bbV_N^{(3)}(z_1,z_2)=[V_{st,N}^{(3)}(z_1,z_2)]_{3\times3},\label{Eq of bbV d=3}
	\end{align}
	and we have
	\begin{align}
		\bbV_N^{(3)}(z_1,z_2)=\bbPi^{(3)}(z_1,z_2)^{-1}\diag(\mfc^{-1}\circ\bbg(z_1)\circ\bbg(z_2))+\mro(\boldsymbol{1}_{3\times3}),\notag
	\end{align}
	where \(\bbPi^{(3)}(z_1,z_2)\) is defined in (\ref{Eq of invertible 2}). Hence, we have
	\begin{align}
		\bbV^{(3)}(z_1,z_2):=\lim_{N\to\infty}\bbV_N^{(3)}(z_1,z_2)=\bbPi^{(3)}(z_1,z_2)^{-1}\diag(\mfc^{-1}\circ\bbg(z_1)\circ\bbg(z_2)).\label{Eq of bbV limit d=3}
	\end{align}
{\bf System equations for \(\mathcal{V}_{ij,N}^{(3)}(z_1,z_2)\) in (\ref{Eq of capital V}):} First, for \(i,j,k\in\{1,2,3\}\), define
	\begin{align}
		V_{ijk,N}^{(3)}(z_1,z_2):=N^{-1}\mathbb{E}[\tr(\boldsymbol{Q}^{ij}(z_1)\boldsymbol{Q}^{jk}(z_2)\boldsymbol{Q}^{ki}(z_2))],
	\end{align}
	Since \(\mcV_{ij,N}^{(3)}(z_1,z_2)=\sum_{l\neq i}^3V_{ilj,N}^{(3)}(z_1,z_2)\), it suffices to calculate all \(V_{ilj,N}^{(3)}(z_1,z_2)\). By the cumulant expansion (\ref{Eq of cumulant expansion}), we have
	\begin{align*}
		&z_1V_{stl,N}^{(3)}(z_1,z_2)=\frac{1}{N^{3/2}}\sum_{r\neq s}^3\sum_{i,j,k=1}^{m,n,p}\mbE[X_{ijk}\mcA_{ijk}^{(s,r)}Q_{\tilde{r}\cdot}^{rl}(z_1)\bbQ^{lt}(z_2)Q_{\cdot\tilde{s}}^{ts}(z_2)]-\delta_{sl}V_{ts,N}^{(3)}(z_2,z_2)\\
		&=\frac{1}{N^{3/2}}\Big(\sum_{r\neq s}^3\sum_{i,j,k=1}^{m,n,p}\mbE[\mcA_{ijk}^{(s,r)}\partial_{ijk}^{(1)}\{Q_{\tilde{r}\cdot}^{rl}(z_1)\bbQ^{lt}(z_2)Q_{\cdot\tilde{s}}^{ts}(z_2)\}]+\epsilon_{ijk}^{(2)}\Big)-\delta_{sl}V_{ts,N}^{(3)}(z_2,z_2),
	\end{align*}
	where
	\begin{align*}
		&\frac{1}{N^{3/2}}\sum_{r\neq s}^3\sum_{i,j,k=1}^{m,n,p}\mbE[\mcA_{ijk}^{(s,r)}\partial_{ijk}^{(1)}\{Q_{\tilde{r}\cdot}^{rl}(z_1)\bbQ^{lt}(z_2)Q_{\cdot\tilde{s}}^{ts}(z_2)\}]=\\
		&-\frac{1}{N^2}\sum_{r\neq s}^3\sum_{w_1\neq w_2}^3\sum_{i,j,k=1}^{m,n,p}\mbE\big[\mcA_{ijk}^{(s,r)}Q_{\tilde{r}\tilde{w}_1}^{rw_1}(z_1)\mcA_{ijk}^{(w_1,w_2)}Q_{\tilde{w}_2\cdot}^{w_2l}(z_1)\bbQ^{lt}(z_2)Q_{\cdot\tilde{s}}^{ts}(z_2)\big]\\
		&-\frac{1}{N^2}\sum_{r\neq s}^3\sum_{w_1\neq w_2}^3\sum_{i,j,k=1}^{m,n,p}\mbE\big[\mcA_{ijk}^{(s,r)}Q_{\tilde{r}\cdot}^{rl}(z_1)Q_{\cdot\tilde{w}_1}^{lw_1}(z_2)\mcA_{ijk}^{(w_1,w_2)}Q_{\tilde{w}_2\cdot}^{w_2t}(z_2)Q_{\cdot\tilde{s}}^{ts}(z_2)\big]\\
		&-\frac{1}{N^2}\sum_{r\neq s}^3\sum_{w_1\neq w_2}^3\sum_{i,j,k=1}^{m,n,p}\mbE\big[\mcA_{ijk}^{(s,r)}Q_{\tilde{r}\cdot}^{rl}(z_1)\bbQ^{lt}(z_2)Q_{\cdot\tilde{w}_1}^{tw_1}(z_2)\mcA_{ijk}^{(w_1,w_2)}Q_{\tilde{w}_2\tilde{s}}^{w_2s}(z_2)\big],
	\end{align*}
	then by Lemmas \ref{Lem of minor terms} and \ref{Thm of Entrywise almost sure convergence}, we have
	\begin{align*}
		&\frac{1}{N^2}\sum_{r\neq s}^3\sum_{w_1\neq w_2}^3\sum_{i,j,k=1}^{m,n,p}\mbE\big[\mcA_{ijk}^{(s,r)}Q_{\tilde{r}\tilde{w}_1}^{rw_1}(z_1)\mcA_{ijk}^{(w_1,w_2)}Q_{\tilde{w}_2\cdot}^{w_2l}(z_1)\bbQ^{lt}(z_2)Q_{\cdot\tilde{s}}^{ts}(z_2)\big]\\
		&=\frac{1}{N^2}\sum_{r\neq s}^3\sum_{i,j,k=1}^{m,n,p}\mbE\big[(\mcA_{ijk}^{(s,r)})^2Q_{\tilde{r}\tilde{r}}^{rr}(z_1)Q_{\tilde{s}\cdot}^{sl}(z_1)\bbQ^{lt}(z_2)Q_{\cdot\tilde{s}}^{ts}(z_2)\big]+\mrO(C_{\eta_0}N^{-\omega})\\
		&=V_{stl,N}^{(3)}(z_1,z_2)\sum_{r\neq s}^3\mfm_r(z_1)+\mrO(C_{\eta_0}N^{-\omega}),
	\end{align*}
	and
	\begin{align*}
		&\frac{1}{N^2}\sum_{r\neq s}^3\sum_{w_1\neq w_2}^3\sum_{i,j,k=1}^{m,n,p}\mbE\big[\mcA_{ijk}^{(s,r)}Q_{\tilde{r}\cdot}^{rl}(z_1)Q_{\cdot\tilde{w}_1}^{lw_1}(z_2)\mcA_{ijk}^{(w_1,w_2)}Q_{\tilde{w}_2\cdot}^{w_2t}(z_2)Q_{\cdot\tilde{s}}^{ts}(z_2)\big]\\
		&=\frac{1}{N^2}\sum_{r\neq s}^3\sum_{i,j,k=1}^{m,n,p}\mbE\big[(\mcA_{ijk}^{(s,r)})^2Q_{\tilde{r}\cdot}^{rl}(z_1)Q_{\cdot\tilde{r}}^{lr}(z_2)Q_{\tilde{s}\cdot}^{st}(z_2)Q_{\cdot\tilde{s}}^{ts}(z_2)\big]+\mrO(C_{\eta_0}N^{-\omega})\\
		&=V_{st,N}^{(3)}(z_2,z_2)\sum_{r\neq s}^3V_{rl,N}^{(3)}(z_1,z_2)+\mrO(C_{\eta_0}N^{-\omega}),
	\end{align*}
	and
	\begin{align*}
		&\frac{1}{N^2}\sum_{r\neq s}^3\sum_{w_1\neq w_2}^3\sum_{i,j,k=1}^{m,n,p}\mbE\big[\mcA_{ijk}^{(s,r)}Q_{\tilde{r}\cdot}^{rl}(z_1)\bbQ^{lt}(z_2)Q_{\cdot\tilde{w}_1}^{tw_1}(z_2)\mcA_{ijk}^{(w_1,w_2)}Q_{\tilde{w}_2\tilde{s}}^{w_2s}(z_2)\big]\\
		&=\frac{1}{N^2}\sum_{r\neq s}^3\sum_{i,j,k=1}^{m,n,p}\mbE\big[(\mcA_{ijk}^{(s,r)})^2Q_{\tilde{r}\cdot}^{rl}(z_1)\bbQ^{lt}(z_2)Q_{\cdot\tilde{r}}^{tr}(z_2)Q_{\tilde{s}\tilde{s}}^{ss}(z_2)\big]+\mrO(C_{\eta_0}N^{-\omega})\\
		&=\mfm_s(z_2)\sum_{r\neq s}^3V_{rtl,N}^{(3)}(z_1,z_2)+\mrO(C_{\eta_0}N^{-\omega})
	\end{align*}
	and \(\sum_{i,j,k=1}^{m,n,p}|\epsilon_{ijk}^{(2)}|=\mrO(\eta_0^{-5}N^{-1/2})\), so we have
	\begin{align*}
		&(z_1+\mfm(z_1)-\mfm_s(z_1))V_{stl,N}^{(3)}(z_1,z_2)=-\mfm_s(z_2)\sum_{r\neq s}^3V_{rtl,N}^{(3)}(z_1,z_2)-\delta_{st}V_{st,N}^{(3)}(z_2,z_2)\\
		&-V_{st,N}^{(3)}(z_2,z_2)\sum_{r\neq s}^3V_{rl,N}^{(3)}(z_1,z_2)+\mrO(C_{\eta_0}N^{-\omega}),
	\end{align*}
	combined with Theorem \ref{Thm of approximation}, we have
	\begin{small}
	\begin{align*}
		&V_{stl,N}^{(3)}(z_1,z_2)=\mfc_s^{-1}g_s(z_1)\Big(\delta_{st}V_{st,N}^{(3)}(z_2,z_2)+g_s(z_2)\sum_{r\neq s}^3V_{rtl,N}^{(3)}(z_1,z_2)+V_{st,N}^{(3)}(z_2,z_2)\sum_{r\neq s}^3V_{rl,N}^{(3)}(z_1,z_2)\Big)+\mrO(C_{\eta_0}N^{-\omega}).
	\end{align*}
	\end{small}\noindent
	Now, for fixed \(1\leq l\leq3\), define
	\begin{align*}
		\bbV_{l,N}^{(3)}(z_1,z_2):=[V_{stl,N}^{(3)}(z_1,z_2)]_{3\times3},
	\end{align*}
	then we have
	\begin{align}
		\lim_{N\to\infty}\bbV_{l,N}^{(3)}(z_1,z_2):&=\bbPi^{(3)}(z_1,z_2)^{-1}\diag(\mfc^{-1}\circ\bbg(z_1))\diag(\bbV_{\cdot l}^{(3)}(z_2,z_2))(\bbI_3+\bbS_3\bbV^{(3)}(z_1,z_2)).\label{Eq of bbV_l d=3}
	\end{align}
	where \(\bbV_{\cdot l}^{(3)}(z_1,z_2)\) is the \(l\)-th column of \(\bbV^{(3)}(z_1,z_2)\) defined in (\ref{Eq of bbV limit d=3}). Once we obtain the limiting value of all \(V_{stl,N}^{(3)}(z_1,z_2)\), the limiting expression for \(\mcV_{st,N}^{(3)}(z_1,z_2)\) follows:
	\begin{align}
		\mcV_{st}^{(3)}(z_1,z_2)=\sum_{l\neq s}^3V_{slt}^{(3)}(z_1,z_2).\label{Eq of mcV limiting d=3}
	\end{align}
{\bf System equations for \(\widetilde{\mcW}_{st,N}^{(3)}(z_1,z_2)\) in (\ref{Eq of capital W}):} By Lemma \ref{Thm of Entrywise almost sure convergence}, we know that 
	$$\Cov(N^{-1}\tr(\bbQ^{ss}(z_1)\circ\bbQ^{ss}(z_2)),N^{-1}\tr(\bbQ^{tt}(z_1)\circ(\bbQ^{tr}(z_2)\bbQ^{rt}(z_2))))=\mrO(C_{\eta_0}N^{-\omega}),$$
	where \(s,t,r\in\{1,2,3\}\). Hence, it is enough to compute \(N^{-1}\mbE[\tr(\bbQ^{tt}(z_1)\circ(\bbQ^{tr}(z_2)\bbQ^{rt}(z_2)))]\) and \(N^{-1}\mbE[\tr(\bbQ^{ss}(z_1)\circ\bbQ^{ss}(z_2))]\), respectively. By (\ref{Eq of H3 d=3}), we know that \(\lim_{N\to\infty}N^{-1}\mbE[\tr(\bbQ^{ss}(z_1)\circ\bbQ^{ss}(z_2))]=\mfc_s^{-1}g_s(z_1)g_s(z_2)\). Next, let us  define
	$$\mathring{V}_{st,N}^{(3)}(z_1,z_2)=N^{-1}\mathbb{E}\big[\operatorname{tr}(\bbQ^{ss}(z_1)\circ(\bbQ^{st}(z_2)\bbQ^{ts}(z_2)))\big].$$
	Similarly, by the cumulant expansion (\ref{Eq of cumulant expansion}), we have
	\begin{align*}
		&z_1\mathring{V}_{st,N}^{(3)}(z_1,z_2)=\frac{1}{N^{3/2}}\sum_{r\neq s}^3\sum_{i,j,k=1}^{m,n,p}\mbE\big[X_{ijk}\mcA_{ijk}^{(s,r)}Q_{\tilde{r}\tilde{s}}^{rs}(z_1)Q_{\tilde{s}\cdot}^{st}(z_2)Q_{\cdot\tilde{s}}^{ts}(z_2)\big]-V_{st,N}^{(3)}(z_1,z_2)\\
		&=\frac{1}{N^{3/2}}\Big(\sum_{r\neq s}^3\sum_{i,j,k=1}^{m,n,p}\mbE\big[\mcA_{ijk}^{(s,r)}\partial_{ijk}^{(1)}\{Q_{\tilde{r}\tilde{s}}^{rs}(z_1)Q_{\tilde{s}\cdot}^{st}(z_2)Q_{\cdot\tilde{s}}^{ts}(z_2)\}\big]+\epsilon_{ijk}^{(2)}\Big)-V_{st,N}^{(3)}(z_1,z_2)\\
		&=-\frac{1}{N^2}\sum_{r\neq s}^3\sum_{w_1\neq w_2}^3\sum_{i,j,k=1}^{m,n,p}\mbE\big[\mcA_{ijk}^{(s,r)}Q_{\tilde{r}\tilde{w}_1}^{rw_1}(z_1)\mcA_{ijk}^{(w_1,w_2)}Q_{\tilde{w}_2\tilde{s}}^{w_2s}(z_1)Q_{\tilde{s}\cdot}^{st}(z_2)Q_{\cdot\tilde{s}}^{ts}(z_2)\big]-V_{st,N}^{(3)}(z_1,z_2)\\
		&-\frac{2}{N^2}\sum_{r\neq s}^3\sum_{w_1\neq w_2}^3\sum_{i,j,k=1}^{m,n,p}\mbE\big[\mcA_{ijk}^{(s,r)}Q_{\tilde{r}\tilde{s}}^{rs}(z_1)Q_{\tilde{s}\tilde{w}_1}^{sw_1}\mcA_{ijk}^{(w_1,w_2)}Q_{\tilde{w}_2\cdot}^{w_2t}(z_2)Q_{\cdot\tilde{s}}^{ts}(z_2)\big],
	\end{align*}
	by Lemma \ref{Lem of minor terms} and Theorem \ref{Thm of approximation}, we have
	\begin{align*}
		&\frac{1}{N^2}\sum_{r\neq s}^3\sum_{w_1\neq w_2}^3\sum_{i,j,k=1}^{m,n,p}\mbE\big[\mcA_{ijk}^{(s,r)}Q_{\tilde{r}\tilde{w}_1}^{rw_1}(z_1)\mcA_{ijk}^{(w_1,w_2)}Q_{\tilde{w}_2\tilde{s}}^{w_2s}(z_1)Q_{\tilde{s}\cdot}^{st}(z_2)Q_{\cdot\tilde{s}}^{ts}(z_2)\big]\\
		&=\frac{1}{N^2}\sum_{r\neq s}^3\sum_{i,j,k=1}^{m,n,p}\mbE\big[(\mcA_{ijk}^{(s,r)})^2Q_{\tilde{r}\tilde{r}}^{rr}(z_1)Q_{\tilde{s}\tilde{s}}^{ss}(z_1)Q_{\tilde{s}\cdot}^{st}(z_2)Q_{\cdot\tilde{s}}^{ts}(z_2)\big]+\mrO(C_{\eta_0}N^{-\omega})\\
		&=\mathring{V}_{st,N}^{(3)}(z_1,z_2)\sum_{r\neq s}^3\mfm_r(z_1)+\mrO(C_{\eta_0}N^{-\omega})=\mathring{V}_{st,N}^{(3)}(z_1,z_2)\sum_{r\neq s}^3g_r(z_1)+\mrO(C_{\eta_0}N^{-\omega}),
	\end{align*}
	and
	\begin{align*}
		&\frac{1}{N^2}\sum_{r\neq s}^3\sum_{w_1\neq w_2}^3\sum_{i,j,k=1}^{m,n,p}\mbE\big[\mcA_{ijk}^{(s,r)}Q_{\tilde{r}\tilde{s}}^{rs}(z_1)Q_{\tilde{s}\tilde{w}_1}^{sw_1}\mcA_{ijk}^{(w_1,w_2)}Q_{\tilde{w}_2\cdot}^{w_2t}(z_2)Q_{\cdot\tilde{s}}^{ts}(z_2)\big]=\mrO(C_{\eta_0}N^{-1})
	\end{align*}
	i.e.
	$$\mathring{V}_N^{(3)}(z_1,z_2)=\mfc_s^{-1}g_s(z_1)V_{st,N}^{(3)}(z_1,z_2)+\mrO(C_{\eta_0}N^{-\omega}).$$
	Define \(\mathring{\bbV}_{st,N}^{(3)}(z_1,z_2):=[\mathring{V}_{st,N}^{(3)}(z_1,z_2)]_{3\times3}\), then we can conclude that
	\begin{align}
		\lim_{N\to\infty}\Vert\mathring{\bbV}_N^{(3)}(z_1,z_2)-\diag(\mfc^{-1}\circ\bbg(z_1))\bbV^{(3)}(z_1,z_2)\Vert=0,\label{Eq of bbV circ d=3}
	\end{align}
	where \(\bbV^{(3)}(z_1,z_2)\) is given in (\ref{Eq of bbV limit d=3}). Hence, by (\ref{Eq of capital W}) and (\ref{Eq of Ap L4 d=3}), \(\mcW_{st,N}^{(3)}(z_1,z_2)\) is given as
    \begin{small}
    \begin{align}
		&\mcW_{st}^{(3)}(z_1,z_2):=\mfc_s^{-1}g_s(z_1)g_s(z_2)\sum_{l\neq s}^3\mcB_{(4)}^{(s,l)}\mathring{V}_{lt}^{(3)}(z_1,z_2)+\mathring{V}_{st}^{(3)}(z_1,z_2)\sum_{l\neq s}^3\mcB_{(4)}^{(s,l)}\mfc_l^{-1}g_l(z_1)g_l(z_2).\label{Eq of mcW limiting d=3}
	\end{align}
    \end{small}

\section{CLT for the LSS when \texorpdfstring{$d=3$}{d=3}}\label{Sec of CLT}
\setcounter{equation}{0}
\def\theequation{\thesection.\arabic{equation}}
\setcounter{subsection}{0}
In this section, we establish the central limit theorem for the linear spectral statistics of \(\bbM\) in (\ref{Eq of N and Q}) when \(d=3\). Precisely, we consider the following family of functions:
\begin{align}
    \mathfrak{F}_3:=\{f(z):f\text{ is analytic on an open set containing the interval }[-\max\{\zeta,\mfv_3\},\max\{\zeta,\mfv_3\}]\},\label{Eq of analytic function d=3}
\end{align}
where $\zeta$ \eqref{Eq of support boundary} is the boundary of LSD $\nu$ and $\mfv_3$ is defined in Theorem \ref{Thm of Extreme eigenvalue N d=3}. For any \(f\in\mathfrak{F}_3\), the LSS of \(\bbM\) is defined as follows:
\begin{align}
	\mathcal{L}_{\bbM}(f):=\frac{1}{N}\sum_{l=1}^Nf(\lambda_l)=\int_{\mbR}f(x)\nu_N(dx),\label{Eq of LSS d=3}
\end{align}
where \(\lambda_1\geq\lambda_2\geq\cdots\geq\lambda_N\) are the eigenvalues of \(\bbM\) and \(\nu_N=N^{-1}\sum_{l=1}^N\delta_{\lambda_l}\) is the ESD of \(\bbM\). By Theorem \ref{Thm of approximation}, we know that the ESD \(\nu_N\) converges to the LSD \(\nu\) in Theorem \ref{Lem of finite support} almost surely, so let
\begin{align}
	G_N(f):=N\int_{-\infty}^{\infty}f(x)(\nu_N(dx)-\nu(dx))=N\Big(\mcL_{\bbM}(f)-\int_{-\infty}^{\infty}f(x)\nu(dx)\Big),\label{Eq of Gf d=3}
\end{align}
we establish that
\begin{thm}\label{Thm of CLT LSS d=3}
	Under Assumptions {\rm \ref{Ap of general noise}} and {\rm \ref{Ap of dimension}}, when \(d=3\), let \(\mfC_{1}\) and \(\mfC_2\) be two disjoint rectangular contours with vertices of \(\pm E_{1}\pm{\rm i}\eta_{1}\) and \(\pm E_2\pm{\rm i}\eta_2\), respectively, such that \(E_{1},E_2\geq\max\{\zeta,\mfv_3\}+t\), where \(t>0\) is fixed constant, and \(\eta_{1},\eta_2>0\) are sufficiently small. Then for any \(f\in\mathfrak{F}_3\) in \eqref{Eq of analytic function d=3}, we have
    $$(G_N(f)-\xi_N^{(3)})/\sigma_N^{(d)}\overset{d}{\longrightarrow}\mcN(0,1),$$
    where
	\begin{align*}
		\xi_N^{(3)}&:=-\frac{1}{2\pi{\rm i}}\oint_{\mfC_1}f(z)\mu_N^{(3)}(z)dz,\\
        (\sigma_N^{(3)})^2&:=-\frac{1}{4\pi^2}\oint_{\mfC_1}\oint_{\mfC_2}f(z_1)f(z_2)\mcC_N^{(3)}(z_1,z_2)dz_1dz_2,
	\end{align*}
	and the mean function \(\mu_N^{(3)}(z)\) and the covariance function \(\mcC_N^{(3)}(z_1,z_2)\) are defined in {\rm (\ref{Eq of mean function d=3})} and {\rm (\ref{Eq of covariance function d=3})}.
\end{thm}
By Theorem \ref{Thm of Extreme eigenvalue N d=3}, for any fixed \(t>0\), we have \(\mbP(\Vert\bbM\Vert>\mfv_3+t)\leq\mro(N^{-l})\) for any \(l>0\), then \(G_N(f)1_{\Vert\bbM\Vert\leq\mfv_3+t}\overset{\mbP}{\longrightarrow}G_N(f)\). Thus, conditional on \(\Vert\bbM\Vert\leq\mfv_3+t\), by the Cauchy integration theorem, we have
\begin{align*}
    G_N(f)=-\frac{1}{2\pi{\rm i}}\oint_{\mathfrak{C}}f(z)\{\tr(\bbQ(z))-Ng(z)\}dz,
\end{align*}
where \(\mathfrak{C}\) is a rectangle contour with vertices of \(\pm E_0\pm{\rm i}\eta_0\) such that \(E_0\geq\max\{\zeta,\mfv_3\}+t\), where \(t>0\) is a fixed constant and \(\eta_0>0\) is sufficiently small, \(\zeta\) is the boundary of the LSD defined in (\ref{Eq of support boundary}). Consequently, to establish the CLT for \(G_N(f)\), it is enough to establish the CLT for \(\tr(\bbQ(z))-Ng(z)\). The proof proceeds in two steps: first, we show that the process \(\tr(\bbQ(z))-\mbE[\tr(\bbQ(z))]\) is tight on \(\mcS_{\eta_0}\), then, we prove that the joint characteristic function of the real part and imaginary part of \(\tr(\bbQ(z))-\mbE[\tr(\bbQ(z))]\) converges to the characteristic function of a normal vector.
\subsection{Tightness}\label{Sec of Tightness}
\begin{thm}\label{Thm of Tightness}
	Under Assumptions {\rm \ref{Ap of general noise}} and {\rm \ref{Ap of dimension}}, for any \(\eta_0>0\), \(\tr(\boldsymbol{Q}(z))-\mathbb{E}[\tr(\boldsymbol{Q}(z))]\) is tight in $\mathcal{S}_{\eta_0}$, i.e. 
	$$\sup_{\substack{z_1,z_2\in\mathcal{S}_{\eta_0}\\z_1\neq z_2}}\frac{\mathbb{E}\left[|\tr(\boldsymbol{Q}(z_1)-\boldsymbol{Q}(z_2))-\mathbb{E}[\tr(\boldsymbol{Q}(z_1)-\boldsymbol{Q}(z_2))|^2\right]}{|z_1-z_2|^2}<C_{\eta_0}.$$
\end{thm}
Similar to the proof of Theorem \ref{Thm of Variance}, several major terms will appear, e.g., \(\mcV_{st,N}^{(3)}(z_1,z_2)\) and \(\widetilde{\mcW}_{st,N}^{(3)}\) in \eqref{Eq of capital V} and \eqref{Eq of capital W}. For simplicity, we first define the following terms: for any \(s_1,s_2,t_1,t_2\in\{1,2,3\}\) and \(z_1,z_2\in\mcS_{\eta_0}\), let
    \begin{align}
		\mcC_{s_1t_1,s_2t_2,N}^{(3)}(z_1,z_2):=\Cov\big(\tr(\bbQ^{s_1t_1}(z_1)\bbQ^{t_1s_1}(z_2)),\tr(\bbQ^{s_2t_2}(z_1)\bbQ^{t_2s_2}(z_2))\big),\label{Eq of mcC tightness d=3}
	\end{align}
    and
    \begin{align}
		\left\{\begin{array}{c}
			\mcC_{s_1t_1,s_2,N}^{(3)}(z_1,z_2):=\Cov(\tr(\bbQ^{s_1t_1}(z_1)\bbQ^{t_1s_1}(z_2)),\tr(\bbQ^{s_2s_2}(z_1))),\\
			\mcC_{s_2,s_1t_1,N}^{(3)}(z_1,z_2):=\Cov(\tr(\bbQ^{s_2s_2}(z_1)),\tr(\bbQ^{s_1t_1}(z_1)\bbQ^{t_1s_1}(z_2))).
		\end{array}\right.\label{Eq of mcC tightness 2 d=3}
	\end{align}
Moreover, for \(k_1,k_2,l_1,l_2\in\{1,2,3\}\), we define
\begin{align}
	&\mcV_{k_1l_1,k_2l_2,N}^{(3)}(z_1,z_2):=\frac{1}{N}\sum_{s\neq k_1}^3\sum_{r=1}^2\mbE\big[\tr\big(\bbQ^{k_1k_2}(\bar{z}_{3-r})\bbQ^{k_2l_2}(\bar{z}_r)\bbQ^{l_2s}(\bar{z}_{3-r})\bbQ^{sl_1}(z_1)\bbQ^{l_1k_1}(z_2)\big)\big],\label{Eq of mcV tightness d=3}\\
    &\mcV_{k_1l_1,k_2,N}^{(3)}(z_1,z_2):=N^{-1}\sum_{r\neq k_1}^3\mathbb{E}\left[\tr(\bbQ^{rl_1}(z_1)\bbQ^{l_1k_1}(z_2)\bbQ^{k_1k_2}(\bar{z}_1)\bbQ^{k_2r}(\bar{z}_2))\right].\label{Eq of mcV tightness d=3 2 vs 1}
\end{align}
By notations in \eqref{Eq of vector notation d=3}, we further define
\begin{small}
\begin{align}
    &\mcW_{11,11,N}^{(3)}(z,z)\label{Eq of mcW tightness d=3}\\
    :&=\frac{1}{N^2}\sum_{r=2}^3\sum_{w=1}^2\sum_{t_1,t_2}^{(1,r)}\Vert\bba^{(5-r)}\Vert_4^4\mbE\big[\tr(\bbQ^{t_2t_2}(z_2)\circ(\bbQ^{t_21}(\bar{z}_w)\bbQ^{11}(\bar{z}_{3-w})\bbQ^{1t_2}(\bar{z}_w)))\cdot\tr(\bbQ^{t_1t_1}(\bar{z}_w)\circ(\bbQ^{t_11}(z_1)\bbQ^{1t_1}(z_2)))\big]\notag\\
	&+\frac{1}{N^2}\sum_{r=2}^3\sum_{w=1}^2\sum_{t_1,t_2}^{(1,r)}\Vert\bba^{(5-r)}\Vert_4^4\mbE\big[\tr((\bbQ^{t_21}(z_1)\bbQ^{1t_2}(z_2))\circ(\bbQ^{t_21}(\bar{z}_w)\bbQ^{11}(\bar{z}_{3-w})\bbQ^{1t_2}(\bar{z}_w)))\cdot\tr(\bbQ^{t_1t_1}(\bar{z}_w)\circ\bbQ^{t_1t_1}(z_1))\big]\notag\\
	&+\frac{1}{N^2}\sum_{r=2}^3\sum_{w=1}^2\sum_{t_1,t_2}^{(1,r)}\Vert\bba^{(5-r)}\Vert_4^4\mbE\big[\tr(\bbQ^{11}(z_2)\circ(\bbQ^{11}(\bar{z}_w)\bbQ^{11}(\bar{z}_w)))\cdot\tr((\bbQ^{r1}(z_1)\bbQ^{1r}(z_2))\circ(\bbQ^{r1}(\bar{z}_{3-w})\bbQ^{1r}(\bar{z}_{3-w})))\big]\notag\\
	&+\frac{1}{N^2}\sum_{r=2}^3\sum_{w=1}^2\sum_{t_1,t_2}^{(1,r)}\Vert\bba^{(5-r)}\Vert_4^4\mbE\big[\tr((\bbQ^{11}(z_1)\bbQ^{11}(z_2))\circ(\bbQ^{11}(\bar{z}_w)\bbQ^{11}(\bar{z}_w)))\cdot\tr(\bbQ^{rr}(z_1)\circ(\bbQ^{r1}(\bar{z}_{3-w})\bbQ^{1r}(\bar{z}_{3-w})))\big],\notag
\end{align}
\end{small}\noindent
and
\begin{small}
\begin{align}
    &\mcW_{11,1,N}^{(3)}(z_1,z_2)\label{Eq of mcW tightness d=3 2 vs 1}\\
    :&=\frac{1}{N^2}\sum_{r=2}^3\Vert\bba^{(5-r)}\Vert_4^4\mbE\big[\tr(\bbQ^{11}(z_1)\circ\bbQ^{11}(\bar{z}_2))\cdot\tr((\bbQ^{r1}(z_1)\bbQ^{1r}(z_2))\circ(\bbQ^{r1}(\bar{z}_2)\bbQ^{1r}(\bar{z}_2)))\big]\notag\\
	&+\frac{1}{N^2}\sum_{r=2}^3\Vert\bba^{(5-r)}\Vert_4^4\mbE\big[\tr(\bbQ^{11}(z_1)\circ(\bbQ^{11}(\bar{z}_2)\bbQ^{11}(\bar{z}_2)))\cdot\tr((\bbQ^{r1}(z_1)\bbQ^{1r}(z_2))\circ\bbQ^{rr}(\bar{z}_2))\big]\notag\\
	&+\frac{1}{N^2}\sum_{r=2}^3\Vert\bba^{(5-r)}\Vert_4^4\mbE\big[\tr((\bbQ^{11}(z_1)\bbQ^{11}(z_2))\circ\bbQ^{11}(\bar{z}_2))\cdot\tr(\bbQ^{rr}(z_2)\circ(\bbQ^{r1}(\bar{z}_2)\bbQ^{1r}(\bar{z}_2)))\big]\notag\\
	&+\frac{1}{N^2}\sum_{r=2}^3\Vert\bba^{(5-r)}\Vert_4^4\mbE\big[\tr((\bbQ^{11}(z_1)\bbQ^{11}(z_2))\circ(\bbQ^{11}(\bar{z}_2)\bbQ^{11}(\bar{z}_2)))\cdot\tr(\bbQ^{rr}(z_2)\circ\bbQ^{rr}(\bar{z}_2))\big],\notag
\end{align}
\end{small}\noindent
where \(t_1,t_2\in\{1,2,3\}\) and the notation $\sum_{t_1,t_2}^{(1,r)}$ means that the summation of \(t_1\) and \(t_2\) are over \(\{1,2,3\}\backslash\{1,r\}\).
\begin{proof}
    By (II.19) in \cite{khorunzhy1996asymptotic}, we know that \(\bbQ(z_1)-\bbQ(z_2)=(z_1-z_2)\bbQ(z_1)\bbQ(z_2)\), then \((z_1-z_2)^{-1}\tr\left(\boldsymbol{Q}(z_1)-\boldsymbol{Q}(z_2)\right)=\tr\left(\boldsymbol{Q}(z_1)\boldsymbol{Q}(z_2)\right)\) for \(z_1\neq z_2\). Note that
	\begin{align}
		&\tr\left(\boldsymbol{Q}(z_1)\boldsymbol{Q}(z_2)\right)=\sum_{i=1}^3\tr\left(\boldsymbol{Q}^{ii}(z_1)\boldsymbol{Q}^{ii}(z_2)\right)+2\sum_{1\leq i<j\leq3}\tr\left(\boldsymbol{Q}^{ij}(z_1)\boldsymbol{Q}^{ji}(z_2)\right),\notag
	\end{align}
	so establishing tightness of \(\tr(\boldsymbol{Q}(z))-\mathbb{E}[\tr(\boldsymbol{Q}(z))]\) reduces to showing that
	$$\mcC_{s_1t_1,s_2t_2,N}^{(3)}(z_1,z_2)=\Cov\big(\tr(\bbQ^{s_1t_1}(z_1)\bbQ^{t_1s_1}(z_2)),\tr(\bbQ^{s_2t_2}(z_1)\bbQ^{t_2s_2}(z_2))\big)<C_{\eta_0}$$
	for any \(z_1,z_2\in\mcS_{\eta_0}\) and \(s_1,s_2,t_1,t_2\in\{1,2,3\}\). Similar to proof of Theorem \ref{Thm of Variance}, we will derive a system equation for all \(\mcC_{s_1t_1,s_2t_2,N}^{(3)}(z_1,z_2)\). For convenience, we only present the detailed calculation procedures of \(\mcC_{11,11,N}^{(3)}\) and omit the \((z_1,z_2)\) behind it, further assume \(\mcC_{11,11,N}^{(3)}\geq1\), otherwise it is bounded. By \(\boldsymbol{Q}(z)\bbM-z\boldsymbol{Q}(z)=\boldsymbol{I}_N\), we obtain that
	\begin{align}
		&z_1\mcC_{11,11,N}^{(3)}(z_1,z_2)=\mathbb{E}\left[\tr(\boldsymbol{Q}^{11}(z_1)\boldsymbol{Q}^{11}(z_2))\left\{\tr(\boldsymbol{Q}^{11}(\bar{z}_1)\boldsymbol{Q}^{11}(\bar{z}_2))\right\}^c\right]\notag\\
		&=\frac{1}{\sqrt{N}}\sum_{i,j,k=1}^{m,n,p}\mathbb{E}\left[X_{ijk}(c_kQ_{j\cdot}^{21}(z_1)Q_{\cdot i}^{11}(z_2)+b_jQ_{k\cdot}^{31}(z_1)Q_{\cdot i}^{11}(z_2))\tr(\boldsymbol{Q}^{11}(\bar{z}_1)\boldsymbol{Q}^{11}(\bar{z}_2))^c\right]\label{Eq of tightness major 1}\\
		&-\mcC_{1,11,N}^{(3)}(z_1,z_2),\label{Eq of tightness major 2}
	\end{align}
    where $\mcC_{1,11,N}^{(3)}(z_1,z_2)={\rm Cov}\left(\tr(\boldsymbol{Q}^{11}(z_2)),\tr(\boldsymbol{Q}^{11}(z_1)\boldsymbol{Q}^{11}(z_2))\right)$ by \eqref{Eq of mcC tightness 2 d=3}. Note that \(\mcC_{s_1t_1,s_2,N}^{(3)}(z_1,z_2)\) and \(\mcC_{s_2,s_1t_1,N}^{(3)}(z_1,z_2)\) are conjugate, to solve \(\mcC_{s_1t_1,s_2t_2,N}^{(3)}(z_1,z_2)\), we also need to obtain the system equation of \(\mcC_{s_1t_1,s_2,N}^{(3)}(z_1,z_2)\). Thus, we will calculate (\ref{Eq of tightness major 1}) and derive the system of equations for (\ref{Eq of tightness major 2}), respectively.

    \vspace{5mm}
    \noindent
    {\bf (\ref{Eq of tightness major 1}):} First, let 
    $$G_{ijk}^1(z_1,z_2):=c_kQ_{j\cdot}^{21}(z_2)Q_{\cdot i}^{11}(z_1)+b_jQ_{k\cdot}^{31}(z_2)Q_{\cdot i}^{11}(z_1).$$
    By the cumulant expansion (\ref{Eq of cumulant expansion}), we have
	\begin{align}
		&\eqref{Eq of tightness major 1}=\frac{1}{\sqrt{N}}\sum_{i,j,k=1}^{m,n,p}\mathbb{E}\left[X_{ijk}G_{ijk}^1(z_1,z_2)\{\tr(\boldsymbol{Q}^{11}(\bar{z}_1)\boldsymbol{Q}^{11}(\bar{z}_2))\}^c\right]=\notag\\
		&\frac{1}{\sqrt{N}}\sum_{i,j,k=1}^{m,n,p}\left(\sum_{\alpha=1}^3\frac{\kappa_{\alpha+1}}{\alpha!}\mathbb{E}\left[\partial_{ijk}^{(\alpha)}\big\{G_{ijk}^1(z_1,z_2)\{\tr(\boldsymbol{Q}^{11}(\bar{z}_1)\boldsymbol{Q}^{11}(\bar{z}_2))\}^c\big\}\right]+\epsilon_{ijk}^{(4)}\right).\notag
	\end{align}
    {\bf First derivatives:} When \(\alpha=1\), similar to the proofs for \(l=1\) in Theorem \ref{Thm of Variance}, by direct calculations, we can show the followings by Lemma \ref{Thm of Entrywise almost sure convergence}:
		\begin{align}
			&N^{-1/2}\sum_{i,j,k=1}^{m,n,p}\mathbb{E}\left[\partial_{ijk}^{(1)}\big\{G_{ijk}^1(z_1,z_2)\{\tr(\boldsymbol{Q}^{11}(\bar{z}_1)\boldsymbol{Q}^{11}(\bar{z}_2))\}^c\big\}\right]=\notag\\
			&-N^{-1}\mathbb{E}\left[\tr(\bbQ^{11}(z_2))\tr(\bbQ^{12}(z_1)\bbQ^{21}(z_2)+\bbQ^{13}(z_1)\bbQ^{31}(z_2))\{\tr(\bbQ^{11}(\bar{z}_1)\bbQ^{11}(\bar{z}_2))\}^c\right]\notag\\
			&-N^{-1}\mathbb{E}\left[\tr(\bbQ^{11}(z_1)\bbQ^{11}(z_2))\tr(\bbQ^{22}(z_1)+\bbQ^{33}(z_1))\tr(\bbQ^{11}(\bar{z}_1)\bbQ^{11}(\bar{z}_2))^c\right]+\mrO(C_{\eta_0}N^{-\omega})\notag\\
			&=-(V_{12,N}^{(3)}(z_1,z_2)+V_{13,N}^{(3)}(z_1,z_2))\mcC_{1,11,N}^{(3)}(z_1,z_2)-\mfm_1(z_2)(\mcC_{12,11,N}^{(3)}+\mcC_{13,11,N}^{(3)})+\mrO(C_{\eta_0}N^{-\omega})\notag\\
			&-V_{11,N}^{(3)}(z_1,z_2)(\mcC_{2,11,N}^{(3)}(z_1,z_2)+\mcC_{3,11,N}^{(3)}(z_1,z_2))-(\mfm_2(z_1)+\mfm_3(z_1)+\mrO(C_{\eta_0}N^{-\omega}))\mcC_{11,11,N}^{(3)}\notag
		\end{align}
		where \(V_{ij,N}^{(3)}(z_1,z_2)\) are defined in \eqref{Eq of V} and we use the same trick as (\ref{Eq of covariance trick}). Since
		\begin{align}
			&\partial_{ijk}^{(1)}\tr(\bbQ^{11}(z_1)\bbQ^{11}(z_2))=\sum_{s,t=1}^{m,m}\partial_{ijk}^{(1)}\big\{Q_{st}^{11}(z_1)Q_{st}^{11}(z_2)\big\}=\sum_{s,t=1}^{m,m}Q_{st}^{11}(z_2)\partial_{ijk}^{(1)}Q_{st}^{11}(z_1)+Q_{st}^{11}(z_1)\partial_{ijk}^{(1)}Q_{st}^{11}(z_2)\notag\\
			&=-N^{-1/2}\sum_{r_1\neq r_2}^3\mcA_{ijk}^{(r_1,r_2)}Q_{\tilde{t}_1\cdot}^{t_21}(z_1)\bbQ^{11}(z_2)Q_{\cdot \tilde{t}_3}^{1t_3}(z_1)-N^{-1/2}\sum_{r_1\neq r_2}^3\mcA_{ijk}^{(r_1,r_2)}Q_{\tilde{s}_1\cdot}^{s_21}(z_2)\bbQ^{11}(z_1)Q_{\cdot \tilde{s}_3}^{1s_3}(z_2),\notag
		\end{align}
		then by Lemma \ref{Lem of minor terms}, we have
		\begin{align}
			&N^{-1/2}\sum_{i,j,k=1}^{m,n,p}\mathbb{E}\left[G_{ijk}^1(z_1,z_2)\partial_{ijk}^{(1)}\tr(\boldsymbol{Q}^{11}(\bar{z}_1)\boldsymbol{Q}^{11}(\bar{z}_2))\right]=\mrO(\eta_0^{-5}N^{-1/2})\notag\\
			&-2N^{-1}\mathbb{E}\left[\tr(\bbQ^{11}(\bar{z}_1)\bbQ^{11}(\bar{z}_2)[\bbQ^{12}(\bar{z}_1)\bbQ^{21}(z_1)\bbQ^{11}(z_2)+\bbQ^{13}(\bar{z}_1)\bbQ^{31}(z_1)\bbQ^{11}(z_2)])\right]\notag\\
			&-2N^{-1}\mathbb{E}\left[\tr(\bbQ^{11}(\bar{z}_2)\bbQ^{11}(\bar{z}_1)[\bbQ^{12}(\bar{z}_2)\bbQ^{21}(z_1)\bbQ^{11}(z_2)+\bbQ^{13}(\bar{z}_2)\bbQ^{31}(z_1)\bbQ^{11}(z_2)])\right].\notag
		\end{align}
		For simplicity, we define
		\begin{align}
			&\mcV_{k_1l_1,k_2l_2,N}^{(3)}(z_1,z_2):=\frac{1}{N}\sum_{s\neq k_1}^3\sum_{r=1}^2\mbE\big[\tr\big(\bbQ^{k_1k_2}(\bar{z}_{3-r})\bbQ^{k_2l_2}(\bar{z}_r)\bbQ^{l_2s}(\bar{z}_{3-r})\bbQ^{sl_1}(z_1)\bbQ^{l_1k_1}(z_2)\big)\big],\notag
		\end{align}
		then we have
		\begin{align}
			&N^{-1/2}\sum_{i,j,k=1}^{m,n,p}\mathbb{E}\left[\partial_{ijk}^{(1)}\big\{G_{ijk}^1(z_1,z_2)\{\tr(\boldsymbol{Q}^{11}(\bar{z}_1)\boldsymbol{Q}^{11}(\bar{z}_2))\}^c\big\}\right]=\label{Eq of tightness d=3 1}\\
			&-2\mcV_{11,11,N}^{(3)}(z_1,z_2)-(1+V_{12,N}^{(3)}(z_1,z_2)+V_{13,N}^{(3)}(z_1,z_2))\mcC_{1,11,N}^{(3)}-\mfm_1(z_2)(\mcC_{12,11,N}^{(3)}+\mcC_{13,11,N}^{(3)})\notag\\
			&-V_{11,N}^{(3)}(z_1,z_2)(\mcC_{2,11,N}^{(3)}+\mcC_{3,11,N}^{(3)})-(\mfm_2(z_1)+\mfm_3(z_1)+\mrO(C_{\eta_0}N^{-\omega}))\mcC_{11,11,N}^{(3)}+\mrO(C_{\eta_0}N^{-\omega}),\notag
		\end{align}
		where we omit \((z_1,z_2)\) behind \(\mcC_{st,11,N}^{(3)}(z_1,z_2)\) and \(\mcC_{s,11,N}^{(3)}(z_1,z_2)\) for simplicity.

    \vspace{5mm}
    \noindent
    {\bf Second derivatives:} When \(\alpha=2\), similar to the proofs for \(l=2\) in Theorem \ref{Thm of Variance}, we claim that there is no major terms. Since \(|G_{ijk}^1(z_1,z_2)|\leq\mrO(\eta_0^{-2})\) and \(|\partial_{ijk}^{(1)}\tr(\bbQ^{11}(\bar{z}_1)\bbQ^{11}(\bar{z}_2))|\leq\mrO(\eta_0^{-3}N^{-1/2})\), we can show the sum over all \(i,j,k\) of following terms are minor by Lemma \ref{Lem of minor terms}:
		$$G_{ijk}^1(z_1,z_2)\mathscr{O}\big\{\partial_{ijk}^{(2)}\tr(\bbQ^{11}(\bar{z}_1)\bbQ^{11}(\bar{z}_2))\big\}\quad{\rm and}\quad\partial_{ijk}^{(1)}\tr(\bbQ^{11}(\bar{z}_1)\bbQ^{11}(\bar{z}_2))\mathscr{O}\big\{\partial_{ijk}^{(1)}G_{ijk}^1(z_1,z_2)\big\},$$
		where \(\msO\) is defined in (\ref{Eq of operator O}). Otherwise, consider
		$$N^{-1/2}\sum_{i,j,k=1}^{m,n,p}G_{ijk}^1(z_1,z_2)\mathscr{D}\big\{\partial_{ijk}^{(2)}\tr(\bbQ^{11}(\bar{z}_1)\bbQ^{11}(\bar{z}_2))\big\},$$
		since 
		\begin{align}
			&\partial_{ijk}^{(2)}\{\tr(\bbQ^{11}(z_1)\bbQ^{11}(z_2))\}\notag\\
			&=\frac{2}{N}\sum_{\substack{t_1\neq t_2\\t_3\neq t_4}}^3\sum_{r=1}^2\mcA_{ijk}^{(t_1,t_2)}Q_{\tilde{t}_2\tilde{t}_3}^{t_2t_3}(z_r)\mcA_{ijk}^{(t_3,t_4)}Q_{\tilde{t}_4\cdot}^{t_41}(z_r)\bbQ^{11}(z_{3-r})Q_{\cdot \tilde{t}_1}^{1t_1}(z_r)\notag\\
			&+\frac{2}{N}\sum_{\substack{t_1\neq t_2\\t_3\neq t_4}}^d\mcA_{ijk}^{(t_1,t_2)}Q_{\tilde{t}_2\cdot}^{t_21}(z_1)Q_{\cdot \tilde{t}_3}^{1t_3}(z_1)\mcA_{ijk}^{(t_3,t_4)}Q_{\tilde{t}_4\cdot}^{t_41}(z_2)Q_{\cdot \tilde{t}_1}^{1t_1}(z_2)\notag,
		\end{align}
		then \(\mathscr{D}\big\{\partial_{ijk}^{(2)}\tr(\bbQ^{11}(\bar{z}_1)\bbQ^{11}(\bar{z}_2))\big\}\) contains (e.g.) \(N^{-1}c_k^2Q_{ii}^{11}(\bar{z}_1)Q_{j\cdot}^{21}(\bar{z}_1)\bbQ^{11}(\bar{z}_2)Q_{\cdot j}^{12}(\bar{z}_1)\) and
		\begin{align}
			&N^{-3/2}\sum_{i,j,k=1}^{m,n,p}c_k^3Q_{i\cdot}^{11}(z_2)Q_{\cdot j}^{12}(z_1)Q_{ii}^{11}(\bar{z}_1)Q_{j\cdot}^{21}(\bar{z}_1)\bbQ^{11}(\bar{z}_2)Q_{\cdot j}^{12}(\bar{z}_1)\notag\\
			&=N^{-3/2}\tr(\bbQ^{11}(\bar{z}_1)\circ(\bbQ^{11}(z_2)\bbQ^{12}(z_1)\bbQ^{21}(\bar{z}_1)\bbQ^{11}(\bar{z}_2)\bbQ^{11}(\bar{z}_1)))=\mrO(\eta_0^{-6}N^{-1/2}),\notag
		\end{align}
		and we can also show other situations are also minor. Finally, since \(N^{-1/2}\sum_{i,j,k=1}^{m,n,p}\partial_{ijk}^{(2)}G_{ijk}^1(z_1,z_2)\) contains the following major term:
		\begin{align*}
			&N^{-3/2}\sum_{i,j,k=1}^{m,n,p}a_ib_jc_kQ_{i\cdot}^{11}(z_2)Q_{\cdot i}^{11}(z_1)Q_{kk}^{33}(z_1)Q_{jj}^{22}(z_1)=N^{-3/2}\boldsymbol{1}_m'\bbQ^{11}(z_1)\bbQ^{11}(z_2)\bba\cdot\boldsymbol{1}_n'\bbQ^{22}(z_1)\bbb\cdot\boldsymbol{1}_p'\bbQ^{22}(z_1)\bbc,
		\end{align*}
		so by Lemma \ref{Thm of Entrywise almost sure convergence}, it yields that
		\begin{align}
			&\Big|\mathbb{E}\left[\tr(\bbQ^{11}(\bar{z}_1)\bbQ^{11}(\bar{z}_2))^c\big\{N^{-3/2}\boldsymbol{1}_m'\bbQ^{11}(z_1)\bbQ^{11}(z_2)\bba\cdot\boldsymbol{1}_n'\bbQ^{22}(z_1)\bbb\cdot\boldsymbol{1}_p'\bbQ^{22}(z_1)\big\}^c\right]\Big|\notag\leq\mrO(C_{\eta_0}N^{-\omega})\mcC_{11,11,N}^{(3)},\notag
		\end{align}
		where we apply the same trick as the part \(l=2\) of Theorem \ref{Thm of Variance}. For other terms, the results are same, so we obtain
		$$N^{-1/2}\sum_{i,j,k=1}^{m,n,p}\mathbb{E}\left[\tr(\boldsymbol{Q}^{11}(\bar{z}_1)\boldsymbol{Q}^{11}(\bar{z}_2))^c\partial_{ijk}^{(2)}G_{ijk}^1(z_1,z_2)\right]\leq\mrO(C_{\eta_0}N^{-\omega})\mcC_{11,11,N}^{(3)}.$$
		Therefore, combined with previous results, we have
		\begin{align}
			N^{-1/2}\sum_{i,j,k=1}^{m,n,p}\mathbb{E}\left[\partial_{ijk}^{(2)}\big\{G_{ijk}^1(z_1,z_2)\tr(\boldsymbol{Q}^{11}(\bar{z}_1)\boldsymbol{Q}^{11}(\bar{z}_2))\big\}\right]\leq\mrO(C_{\eta_0}N^{-\omega})\mcC_{11,11,N}^{(3)}.\label{Eq of tightness d=3 2}
		\end{align}
    {\bf Third derivatives:} When \(\alpha=3\), by the same proofs as those for \(l=3\) in Theorem \ref{Thm of Variance}, we can conclude that for \(\alpha=0,1,3\)
		$$N^{-1/2}\Big|\sum_{i,j,k=1}^{m,n,p}\partial_{ijk}^{(\alpha)}\tr(\bbQ^{11}(\bar{z}_1)\bbQ^{11}(\bar{z}_2))\partial_{ijk}^{(3-\alpha)}G_{ijk}^1(z_1,z_2)\Big|\leq\mrO(C_{\eta_0}N^{-\omega})\mcC_{11,11,N}^{(3)},$$
		here we omit the details for clarity. Now, we only focus on the case of \(\alpha=2\). Since 
		\begin{align}
			&\partial_{ijk}^{(2)}\{\tr(\bbQ^{11}(z_1)\bbQ^{11}(z_2))\}\notag\\
			&=\frac{2}{N}\sum_{\substack{t_1\neq t_2\\t_3\neq t_4}}^3\sum_{r=1}^2\mcA_{ijk}^{(t_1,t_2)}Q_{\tilde{t}_2\tilde{t}_3}^{t_2t_3}(z_r)\mcA_{ijk}^{(t_3,t_4)}Q_{\tilde{t}_4\cdot}^{t_41}(z_r)\bbQ^{11}(z_{3-r})Q_{\cdot \tilde{t}_1}^{1t_1}(z_r)\notag\\
			&+\frac{2}{N}\sum_{\substack{t_1\neq t_2\\t_3\neq t_4}}^d\mcA_{ijk}^{(t_1,t_2)}Q_{\tilde{t}_2\cdot}^{t_21}(z_1)Q_{\cdot \tilde{t}_3}^{1t_3}(z_1)\mcA_{ijk}^{(t_3,t_4)}Q_{\tilde{t}_4\cdot}^{t_41}(z_2)Q_{\cdot \tilde{t}_1}^{1t_1}(z_2)\notag,
		\end{align}
		and
		\begin{align*}
			\partial_{ijk}^{(1)}\{G_{ijk}^1(z_1,z_2)\}=-\frac{1}{\sqrt{N}}\sum_{r\neq1}^3\mcA_{ijk}^{(1,r)}\mcA_{ijk}^{(t_1,t_2)}\big(Q_{i\tilde{t}_1}^{1t_1}(z_2)Q_{\tilde{t}_2\cdot}^{t_21}(z_2)Q_{\cdot\tilde{r}}^{1r}(z_1)+Q_{i\cdot}^{11}(z_2)Q_{\cdot\tilde{t}_1}^{1t_1}(z_1)Q_{\tilde{t}_2\tilde{r}}^{t_2r}(z_1)\big),
		\end{align*}
		then by Lemmas \ref{Cor of minor terms} and \ref{Lem of minor terms}, we have
		\begin{small}
		\begin{align}
			&\frac{1}{\sqrt{N}}\sum_{i,j,k=1}^{m,n,p}\sum_{\substack{t_1\neq t_2\\t_3\neq t_4}}^3\mbE\big[\partial_{ijk}^{(1)}\{G_{ijk}^1(z_1,z_2)\}\partial_{ijk}^{(2)}\{\tr(\bbQ^{11}(z_1)\bbQ^{11}(z_2))\}\big]=\mrO(C_{\eta_0}N^{-1/2})\notag\\
			&-\frac{2}{N^2}\sum_{i,j,k=1}^{m,n,p}\sum_{r\neq1}^3\sum_{t_1,t_2}^{(1,r)}\sum_{w=1}^2\mbE\big[(\mcA_{ijk}^{(1,r)})^4Q_{ii}^{11}(z_2)Q_{\tilde{r}\cdot}^{r1}(z_1)Q_{\cdot\tilde{r}}^{1r}(z_2)Q_{\tilde{t}_1\tilde{t}_1}^{t_1t_1}(\bar{z}_w)Q_{\tilde{t}_2\cdot}^{t_21}(\bar{z}_w)\bbQ^{11}(\bar{z}_{3-w})Q_{\cdot\tilde{t}_2}^{1t_2}(\bar{z}_w)\big]\notag\\
			&-\frac{2}{N^2}\sum_{i,j,k=1}^{m,n,p}\sum_{r\neq1}^3\sum_{w=1}^2\sum_{t_1,t_2}^{(1,r)}\mbE\big[(\mcA_{ijk}^{(1,r)})^4Q_{i\cdot}^{11}(z_2)Q_{\cdot i}^{11}(z_1)Q_{\tilde{r}\tilde{r}}^{rr}(z_1)Q_{\tilde{t}_1\tilde{t}_1}^{t_1t_1}(\bar{z}_w)Q_{\tilde{t}_2\cdot}^{t_21}(\bar{z}_w)\bbQ^{11}(\bar{z}_{3-w})Q_{\cdot\tilde{t}_2}^{1t_2}(\bar{z}_w)\big]\notag\\
			&-\frac{2}{N^2}\sum_{i,j,k=1}^{m,n,p}\sum_{r\neq1}^3\sum_{t_1,t_2}^{(1,r)}\mbE\big[(\mcA_{ijk}^{(1,r)})^4Q_{ii}^{11}(z_2)Q_{\tilde{r}\cdot}^{r1}(z_1)Q_{\cdot\tilde{r}}^{1r}(z_2)Q_{\tilde{t}_1\cdot}^{t_11}(\bar{z}_1)Q_{\cdot \tilde{t}_1}^{1t_1}(\bar{z}_1)Q_{\tilde{t}_2\cdot}^{t_21}(\bar{z}_2)Q_{\cdot\tilde{t}_2}^{1t_2}(\bar{z}_2)\big]\notag\\
			&-\frac{2}{N^2}\sum_{i,j,k=1}^{m,n,p}\sum_{r\neq1}^3\sum_{t_1,t_2}^{(1,r)}\mbE\big[(\mcA_{ijk}^{(1,r)})^4Q_{i\cdot}^{11}(z_2)Q_{\cdot i}^{11}(z_1)Q_{\tilde{r}\tilde{r}}^{rr}(z_1)Q_{\tilde{t}_1\cdot}^{t_11}(\bar{z}_1)Q_{\cdot \tilde{t}_1}^{1t_1}(\bar{z}_1)Q_{\tilde{t}_2\cdot}^{t_21}(\bar{z}_2)Q_{\cdot\tilde{t}_2}^{1t_2}(\bar{z}_2)\big]=\mrO(C_{\eta_0}N^{-1/2})\notag\\
			&-\frac{2}{N^2}\sum_{r=2}^3\sum_{w=1}^2\sum_{t_1,t_2}^{(1,r)}\Vert\bba^{(5-r)}\Vert_4^4\mbE\big[\tr(\bbQ^{t_2t_2}(z_2)\circ(\bbQ^{t_21}(\bar{z}_w)\bbQ^{11}(\bar{z}_{3-w})\bbQ^{1t_2}(\bar{z}_w)))\cdot\tr(\bbQ^{t_1t_1}(\bar{z}_w)\circ(\bbQ^{t_11}(z_1)\bbQ^{1t_1}(z_2)))\big]\notag\\
			&-\frac{2}{N^2}\sum_{r=2}^3\sum_{w=1}^2\sum_{t_1,t_2}^{(1,r)}\Vert\bba^{(5-r)}\Vert_4^4\mbE\big[\tr((\bbQ^{t_21}(z_1)\bbQ^{1t_2}(z_2))\circ(\bbQ^{t_21}(\bar{z}_w)\bbQ^{11}(\bar{z}_{3-w})\bbQ^{1t_2}(\bar{z}_w)))\cdot\tr(\bbQ^{t_1t_1}(\bar{z}_w)\circ\bbQ^{t_1t_1}(z_1))\big]\notag\\
			&-\frac{2}{N^2}\sum_{r=2}^3\sum_{w=1}^2\sum_{t_1,t_2}^{(1,r)}\Vert\bba^{(5-r)}\Vert_4^4\mbE\big[\tr(\bbQ^{11}(z_2)\circ(\bbQ^{11}(\bar{z}_w)\bbQ^{11}(\bar{z}_w)))\cdot\tr((\bbQ^{r1}(z_1)\bbQ^{1r}(z_2))\circ(\bbQ^{r1}(\bar{z}_{3-w})\bbQ^{1r}(\bar{z}_{3-w})))\big]\notag\\
			&-\frac{2}{N^2}\sum_{r=2}^3\sum_{t_1,t_2}^{(1,r)}\sum_{w=1}^2\Vert\bba^{(5-r)}\Vert_4^4\mbE\big[\tr((\bbQ^{11}(z_1)\bbQ^{11}(z_2))\circ(\bbQ^{11}(\bar{z}_w)\bbQ^{11}(\bar{z}_w)))\cdot\tr(\bbQ^{rr}(z_1)\circ(\bbQ^{r1}(\bar{z}_{3-w})\bbQ^{1r}(\bar{z}_{3-w})))\big]\notag\\
			&=-2\mcW_{11,11,N}^{(3)}(z,z)+\mrO(C_{\eta_0}N^{-1/2}),\notag
		\end{align}
		\end{small}\noindent
    where the notation $\sum_{t_1,t_2}^{(1,r)}$ means the summation of $t_1$ and $t_2$ are over $\{1,2,3\}\backslash\{1,r\}$, $\mcW_{11,11,N}^{(3)}(z,z)$ is defined in \eqref{Eq of mcW tightness d=3}.
    
    \vspace{5mm}
    \noindent
    {\bf Remainders:} When \(\alpha=4\), we can repeat the same proof argument as those for the part \(l=4\) of Theorem \ref{Thm of Variance} to show that the sum over all \(i,j,k\) of \(\epsilon_{ijk}^{(4)}\) is a minor term, the details are omitted for brevity.

    \vspace{5mm}
    \noindent
	Now, combining (\ref{Eq of tightness d=3 1}), (\ref{Eq of tightness d=3 2}) and (\ref{Eq of mcW tightness d=3}), we obtain
    \begin{small}
    \begin{align}
		&(z_1+\mfm(z_1)-\mfm_1(z_1)+\mrO(C_{\eta_0}N^{-\omega}))\mcC_{11,11,N}^{(3)}=-2\mcV_{11,11,N}^{(3)}(z_1,z_2)-\kappa_4\mcG_{11,11,N}^{(3)}(z_1,z_2)-\mfm_1(z_2)(\mcC_{12,11,N}^{(3)}+\mcC_{13,11,N}^{(3)})\notag\\
		&-(1+V_{12,N}^{(3)}(z_1,z_2)+V_{13,N}^{(3)}(z_1,z_2))\mcC_{1,11,N}^{(3)}-V_{11,N}^{(3)}(z_1,z_2)(\mcC_{2,11,N}^{(3)}+\mcC_{3,11,N}^{(3)})+\mrO(C_{\eta_0}N^{-\omega}).\notag
	\end{align}
    \end{small}\noindent
	Similarly, for any \(s,t\in\{1,2,3\}\) and \(\mcC_{st,11,N}^{(3)}(z,z)\), we have
	\begin{small}
	\begin{align}
		&(z_1+\mfm(z_1)-\mfm_s(z_1))\mcC_{st,11,N}^{(3)}=-2\mcV_{st,11,N}^{(3)}(z_1,z_2)-\kappa_4\mcW_{st,11,N}^{(3)}(z_1,z_2)-\mfm_s(z_2)\sum_{l\neq s}^3\mcC_{lt,11,N}^{(3)}\label{Eq of system equation tightness d=3}\\
		&-\Big(\delta_{st}+\sum_{l\neq s}^3V_{lt,N}^{(3)}(z_1,z_2)\Big)\mcC_{s,11,N}^{(3)}-V_{st,N}^{(3)}(z_1,z_2)\sum_{l\neq s}^3\mcC_{l,11,N}^{(3)}+\mrO(C_{\eta_0}N^{-\omega})\mcC_{11,11,N}^{(3)}+\mrO(C_{\eta_0}N^{-\omega}).\notag
	\end{align}
	\end{small}\noindent
    {\bf (\ref{Eq of tightness major 2}):} Next, we will derive the system equations for all \(\mcC_{l,11,N}^{(3)}\). Here, we only present the detailed calculation procedure for \(\mcC_{11,1}^{(3)}(z_1,z_2)\), since the arguments for the others are the same. By the cumulant expansion (\ref{Eq of cumulant expansion}), we can obtain
	\begin{align}
		&z_1\mcC_{11,1,N}^{(3)}(z_1,z_2)=\frac{1}{\sqrt{N}}\sum_{i,j,k=1}^{m,n,p}\Big(\sum_{\alpha=1}^3\frac{\kappa_{\alpha+1}}{\alpha!}\mathbb{E}\left[\partial_{ijk}^{(\alpha)}\big\{G_{ijk}^1(z_1,z_2)\tr(\bbQ(\bar{z}_2))^c\big\}\right]+\epsilon_{ijk}^{(4)}\Big)-\mathcal{C}_{1,1,N}^{(3)}(z_1,z_2),\notag
	\end{align}
	where \(\mathcal{C}_{1,1,N}^{(3)}(z_1,z_2)\) is defined in (\ref{Eq of mcC d=3}) and it is already bounded by \(C_{\eta_0}\). Actually, we can repeat the proofs in Theorem \ref{Thm of Variance}, and major terms only appear in cases of \(\alpha=1\) and \(3\), so we omit other cases for ease of presentation.
 
    \vspace{5mm}
    \noindent
    {\bf First derivatives:} When \(\alpha=1\), by Lemma \ref{Lem of minor terms} and Theorem \ref{Thm of Entrywise almost sure convergence}, we can show that
		\begin{align}
			&N^{-1/2}\sum_{l=1}^m\mathbb{E}\left[\partial_{ijk}^{(1)}\{G_{ijk}^1(z_1,z_2)\}\tr(\bbQ^{11}(\bar{z}_2))^c\right]=\notag\\
			&-N^{-1}\mathbb{E}\left[\tr(\bbQ^{11}(z_2))\tr(\bbQ^{12}(z_2)\bbQ^{21}(z_1)+\bbQ^{13}(z_2)\bbQ^{31}(z_1))\tr(\bbQ^{11}(\bar{z}_2))^c\right]\notag\\
			&-N^{-1}\mathbb{E}\left[\tr(\bbQ^{11}(z_1)\bbQ^{11}(z_2))\tr(\bbQ^{22}(z_1)+\bbQ^{33}(z_1))\tr(\bbQ^{11}(\bar{z}_2))^c\right]+\mrO(C_{\eta_0}N^{-1/2})\notag\\
			&=-(V_{12,N}^{(3)}(z_1,z_2)+V_{13,N}^{(3)}(z_1,z_2))\mathcal{C}_{1,1,N}^{(3)}(z_2,z_2)-V_{11}(z_1,z_2)(\mathcal{C}_{2,1,N}^{(3)}(z_1,z_2)+\mathcal{C}_{3,1,N}^{(3)}(z_1,z_2))\notag\\
			&-\mfm_1(z_2)(\mcC_{12,1,N}^{(3)}(z_1,z_2)+\mcC_{13,1}^{(3)}(z_1,z_2))-(\mfm_2(z_1)+\mfm_3(z_1))\mcC_{11,1,N}^{(3)}(z_1,z_2)+\mrO(C_{\eta_0}N^{-\omega})\notag
		\end{align}
		and
		\begin{align}
			&N^{-1/2}\sum_{i,j,k=1}^{m,n,p}\mathbb{E}\left[G_{ijk}^1(z_1,z_2)\partial_{ijk}^{(1)}\{\tr(\bbQ^{11}(\bar{z}_2))^c\}\right]=\notag\\
			&-2N^{-1}\mathbb{E}\left[\tr(\bbQ^{11}(z_2)\bbQ^{12}(z_1)\bbQ^{21}(\bar{z}_2)\bbQ^{11}(\bar{z}_2)+\bbQ^{11}(z_2)\bbQ^{13}(z_1)\bbQ^{31}(\bar{z}_2)\bbQ^{11}(\bar{z}_2))\right]+C_{\eta_0}N^{-1/2}.\notag
		\end{align}
		For simplicity, we denote
		\begin{align}
			\mcV_{k_1l_1,k_2,N}^{(3)}(z_1,z_2):=N^{-1}\sum_{r\neq k_1}^3\mathbb{E}\left[\tr(\bbQ^{rl_1}(z_1)\bbQ^{l_1k_1}(z_2)\bbQ^{k_1k_2}(\bar{z}_1)\bbQ^{k_2r}(\bar{z}_2))\right],\notag
		\end{align}
		where \(k_1,l_1,k_2\in\{1,2,3\}\), then we have
		\begin{align}
			&N^{-1/2}\sum_{i,j,k=1}^{m,n,p}\mathbb{E}\left[G_{ijk}^1(z_1,z_2)\partial_{ijk}^{(1)}\{\tr(\bbQ^{11}(\bar{z}_2))^c\}\right]=-2\mcV_{11,1,N}^{(3)}(z_1,z_2)\label{Eq of tightness d=3 1 2 vs 1}\\
			&-(V_{12,N}^{(3)}(z_1,z_2)+V_{13,N}^{(3)}(z_1,z_2))\mathcal{C}_{1,1,N}^{(3)}(z_2,z_2)-V_{11}(z_1,z_2)(\mathcal{C}_{2,1,N}^{(3)}(z_1,z_2)+\mathcal{C}_{3,1,N}^{(3)}(z_1,z_2))\notag\\
			&-\mfm_1(z_2)(\mcC_{12,1,N}^{(3)}(z_1,z_2)+\mcC_{13,1}^{(3)}(z_1,z_2))-(\mfm_2(z_1)+\mfm_3(z_1))\mcC_{11,1,N}^{(3)}(z_1,z_2)+\mrO(C_{\eta_0}N^{-\omega}).\notag
		\end{align}
    {\bf Second derivatives:} The calculations of second derivatives are similar to those in proofs of Theorem \ref{Thm of Variance}, we can also show that
    \begin{align}
		N^{-1/2}\Bigg|\sum_{i,j,k=1}^{m,n,p}\mathbb{E}\left[\partial_{ijk}^{(2)}\{G_{ijk}^1(z_1,z_2)\tr(\bbQ^{11}(\bar{z}_2))^c\}\right]\Bigg|\leq\mrO(C_{\eta_0}N^{-\omega}),\notag
	\end{align}
    we omit the details.
    
    \vspace{5mm}
    \noindent
    {\bf Third derivatives:} When \(\alpha=3\), similar to previous arguments for \(\mcC_{11,11,N}^{(3)}\), the major terms will only appear in
		\begin{align}
			N^{-1/2}\sum_{i,j,k=1}^{m,n,p}\mathbb{E}\left[\partial_{ijk}^{(1)}G_{ijk}^1(z_1,z_2)\partial_{ijk}^{(2)}\tr(\bbQ^{11}(\bar{z}_2))^c\right],\notag
		\end{align}
		where
		\begin{align*}
			\partial_{ijk}^{(2)}\tr(\bbQ^{11}(z))=2N^{-1}\sum_{t_1\neq t_2}^3\sum_{t_3\neq t_4}^3\mcA_{ijk}^{(t_1,t_2)}Q_{\tilde{t}_2\tilde{t}_3}^{t_2t_3}(z)\mcA_{ijk}^{(t_3,t_4)}Q_{\tilde{t}_4\cdot}^{t_41}(z)Q_{\cdot\tilde{t}_1}^{1t_1}(z),
		\end{align*}
		and
		\begin{align*}
			\partial_{ijk}^{(1)}\{G_{ijk}^1(z_1,z_2)\}=-\frac{1}{\sqrt{N}}\sum_{r\neq1}^3\mcA_{ijk}^{(1,r)}\mcA_{ijk}^{(t_1,t_2)}\big(Q_{i\tilde{t}_1}^{1t_1}(z_1)Q_{\tilde{t}_2\cdot}^{t_21}(z_1)Q_{\cdot\tilde{r}}^{1r}(z_2)+Q_{i\cdot}^{11}(z_1)Q_{\cdot\tilde{t}_1}^{1t_1}(z_2)Q_{\tilde{t}_2\tilde{r}}^{t_2r}(z_2)\big),
		\end{align*}
		then we have
		\begin{align}
			&N^{-1/2}\sum_{i,j,k=1}^{m,n,p}\mathbb{E}\left[\partial_{ijk}^{(1)}G_{ijk}^1(z_1,z_2)\partial_{ijk}^{(2)}\tr(\bbQ^{11}(\bar{z}_2))^c\right]=\mrO(C_{\eta_0}N^{-\omega})\notag\\
			&-\frac{2}{N^2}\sum_{i,j,k=1}^{m,n,p}\sum_{r\neq1}^3\mbE\big[(\mcA_{ijk}^{(1,r)})^4Q_{ii}^{11}(z_1)Q_{\tilde{r}\cdot}^{r1}(z_1)Q_{\cdot\tilde{r}}^{1r}(z_2)Q_{ii}^{11}(\bar{z}_2)Q_{\tilde{r}\cdot}^{r1}(\bar{z}_2)Q_{\cdot\tilde{r}}^{1r}(\bar{z}_2)\big]\notag\\
			&-\frac{2}{N^2}\sum_{i,j,k=1}^{m,n,p}\sum_{r\neq1}^3\mbE\big[(\mcA_{ijk}^{(1,r)})^4Q_{ii}^{11}(z_1)Q_{\tilde{r}\cdot}^{r1}(z_1)Q_{\cdot\tilde{r}}^{1r}(z_2)Q_{\tilde{r}\tilde{r}}^{rr}(\bar{z}_2)Q_{i\cdot}^{11}(\bar{z}_2)Q_{\cdot i}^{11}(\bar{z}_2)\big]\notag\\
			&-\frac{2}{N^2}\sum_{i,j,k=1}^{m,n,p}\sum_{r\neq1}^3\mbE\big[(\mcA_{ijk}^{(1,r)})^4Q_{i\cdot}^{11}(z_1)Q_{\cdot i}^{11}(z_2)Q_{\tilde{r}\tilde{r}}^{rr}(z_2)Q_{ii}^{11}(\bar{z}_2)Q_{\tilde{r}\cdot}^{r1}(\bar{z}_2)Q_{\cdot\tilde{r}}^{1r}(\bar{z}_2)\big]\notag\\
			&-\frac{2}{N^2}\sum_{i,j,k=1}^{m,n,p}\sum_{r\neq1}^3\mbE\big[(\mcA_{ijk}^{(1,r)})^4Q_{i\cdot}^{11}(z_1)Q_{\cdot i}^{11}(z_2)Q_{\tilde{r}\tilde{r}}^{rr}(z_2)Q_{\tilde{r}\tilde{r}}^{rr}(\bar{z}_2)Q_{i\cdot}^{11}(\bar{z}_2)Q_{\cdot i}^{11}(\bar{z}_2)\big]\notag\\
			&=-\frac{2}{N^2}\sum_{r=2}^3\Vert\bba^{(5-r)}\Vert_4^4\mbE\big[\tr(\bbQ^{11}(z_1)\circ\bbQ^{11}(\bar{z}_2))\cdot\tr((\bbQ^{r1}(z_1)\bbQ^{1r}(z_2))\circ(\bbQ^{r1}(\bar{z}_2)\bbQ^{1r}(\bar{z}_2)))\big]\notag\\
			&-\frac{2}{N^2}\sum_{r=2}^3\Vert\bba^{(5-r)}\Vert_4^4\mbE\big[\tr(\bbQ^{11}(z_1)\circ(\bbQ^{11}(\bar{z}_2)\bbQ^{11}(\bar{z}_2)))\cdot\tr((\bbQ^{r1}(z_1)\bbQ^{1r}(z_2))\circ\bbQ^{rr}(\bar{z}_2))\big]\notag\\
			&-\frac{2}{N^2}\sum_{r=2}^3\Vert\bba^{(5-r)}\Vert_4^4\mbE\big[\tr((\bbQ^{11}(z_1)\bbQ^{11}(z_2))\circ\bbQ^{11}(\bar{z}_2))\cdot\tr(\bbQ^{rr}(z_2)\circ(\bbQ^{r1}(\bar{z}_2)\bbQ^{1r}(\bar{z}_2)))\big]\notag\\
			&-\frac{2}{N^2}\sum_{r=2}^3\Vert\bba^{(5-r)}\Vert_4^4\mbE\big[\tr((\bbQ^{11}(z_1)\bbQ^{11}(z_2))\circ(\bbQ^{11}(\bar{z}_2)\bbQ^{11}(\bar{z}_2)))\cdot\tr(\bbQ^{rr}(z_2)\circ\bbQ^{rr}(\bar{z}_2))\big]+\mrO(C_{\eta_0}N^{-1/2})\notag\\
			&=-2\mcW_{11,1,N}^{(3)}(z_1,z_2)+\mrO(C_{\eta_0}N^{-1/2}),\notag
		\end{align}
    where $\mcW_{11,1,N}^{(3)}(z_1,z_2)$ is defined in \eqref{Eq of mcW tightness d=3 2 vs 1}. 

    \vspace{5mm}
    \noindent
    {\bf Remainders:} We can repeat the same proof argument as those for the part \(l=4\) of Theorem \ref{Thm of Variance} to show that the sum over all \(i,j,k\) of \(\epsilon_{ijk}^{(4)}\) is a minor term, we omit the details.
    
    Now, combining (\ref{Eq of tightness d=3 1 2 vs 1}) and (\ref{Eq of mcW tightness d=3 2 vs 1}), we obtain that
	\begin{align}
		&(z_1+\mfm_2(z_1)+\mfm_3(z_1))\mcC_{11,1,N}^{(3)}(z_1,z_2)=-2\mcV_{11,1,N}^{(3)}(z_1,z_2)-\kappa_4\mcW_{11,1,N}^{(3)}(z_1,z_2)\notag\\
		&-(V_{12,N}^{(3)}(z_1,z_2)+V_{13,N}^{(3)}(z_1,z_2))\mathcal{C}_{1,1,N}^{(3)}(z_2,z_2)-V_{11,N}^{(3)}(z_1,z_2)(\mathcal{C}_{2,1,N}^{(3)}(z_1,z_2)+\mathcal{C}_{3,1,N}^{(3)}(z_1,z_2))\notag\\
		&-\mfm_1(z_2)(\mcC_{12,1,N}^{(3)}(z_1,z_2)+\mcC_{13,1,N}^{(3)}(z_1,z_2))+\mrO(C_{\eta_0}N^{-\omega})\notag\\
		:&=-\mfm_1(z_2)(\mcC_{12,1,N}^{(3)}(z_1,z_2)+\mcC_{13,1,N}^{(3)}(z_1,z_2))-\mcF_{11,1,N}^{(3)}(z_1,z_2)+\mrO(C_{\eta_0}N^{-\omega})\notag
	\end{align}
	Define
	\begin{align*}
		\bbC_{1,N}^{(3)}(z_1,z_2):=[\mcC_{st,1,N}^{(3)}(z_1,z_2)]_{3\times3}\quad{\rm and}\quad\bbF_{1,N}^{(3)}(z_1,z_2):=[\mcF_{st,1,N}^{(3)}(z_1,z_2)]_{3\times3}.
	\end{align*}
	By system equations in \S\ref{Sec of majors} and Theorem \ref{Thm of Variance}, \(\mcV_{st,1,N}^{(3)}(z_1,z_2),\mcW_{st,1,N}^{(3)}(z_1,z_2),V_{st,N}^{(3)}(z_1,z_2)\) and \(\mcC_{s,t,N}^{(3)}(z_1,z_2)\) are all bounded by \(C_{\eta_0}\), so \(\Vert\bbF_{1,N}^{(3)}(z_1,z_2)\Vert<C_{\eta_0}\). Moreover, since
	\begin{align*}
		\bbTheta_N^{(3)}(z_1,z_2)\bbC_{1,N}^{(3)}(z_1,z_2)=-\bbF_{1,N}^{(3)}(z_1,z_2)+\mro(\boldsymbol{1}_{3\times3})
	\end{align*}
	where \(\bbTheta_N^{(3)}(z_1,z_2)\) is defined in (\ref{Eq of bbTheta d=3}). By Theorem \ref{Thm of Variance}, we have \(\lim_{N\to\infty}\Vert\bbTheta_N^{(3)}(z_1,z_2)^{-1}+\diag(\mfc^{-1}\circ\bbg(z))\bbPi^{(3)}(z_1,z_2)^{-1}\Vert=0\), so \(\Vert\bbTheta_N^{(3)}(z_1,z_2)^{-1}\Vert\leq C_{\mfc}\eta_0^{-5}\), which implies that 
    \begin{align*}
        &\lim_{N\to\infty}\Vert\bbC_{1,N}^{(3)}(z_1,z_2)\Vert\leq\lim_{N\to\infty}\Vert\bbTheta_N^{(3)}(z_1,z_2)^{-1}\Vert\cdot\Vert\bbF_{1,N}^{(3)}(z_1,z_2)\Vert\leq C_{\eta_0},
    \end{align*}
    i.e. all \(|\mcC_{st,1,N}^{(3)}(z_1,z_2)|\leq C_{\eta_0}\) for \(1\leq s,t\leq3\). Similarly, we can repeat the previous procedures to show that \(|\mcC_{st,l,N}^{(3)}(z_1,z_2)|\leq C_{\eta_0}\) for \(1\leq s,t,l\leq3\). Finally, let us back to (\ref{Eq of system equation tightness d=3}), define
	\begin{align*}
		&\mcF_{st,11,N}^{(3)}(z_1,z_2):=2\mcV_{st,11,N}^{(3)}(z_1,z_2)+\kappa_4\mcW_{st,11,N}^{(3)}(z_1,z_2)\\
		&+\Big(\delta_{st}+\sum_{l\neq s}^3V_{lt,N}^{(3)}(z_1,z_2)\Big)\mcC_{s,11,N}^{(3)}(z_1,z_2)+V_{st,N}^{(3)}(z_1,z_2)\sum_{l\neq s}^3\mcC_{l,11,N}^{(3)}(z_1,z_2),
	\end{align*}
	and
	\begin{align*}
		\bbC_{11,N}^{(3)}(z_1,z_2):=[\mcC_{st,11,N}^{(3)}(z_1,z_2)]_{3\times3}\quad{\rm and}\quad\bbF_{11,N}^{(3)}(z_1,z_2):=[\mcF_{st,11,N}^{(3)}(z_1,z_2)]_{3\times3},
	\end{align*}
	then write (\ref{Eq of system equation tightness d=3}) into matrix notations, i.e.
	$$\bbTheta_N^{(3)}(z_1,z_2)\bbC_{11,N}^{(3)}(z_1,z_2)=-\bbF_{11,N}^{(3)}(z_1,z_2)+\mro(\boldsymbol{1}_{3\times3})+\mrO(C_{\eta_0}N^{-\omega})\mcC_{11,11,N}^{(3)}(z_1,z_2)\boldsymbol{1}_{3\times3}.$$
	Since we have shown \(|\mcC_{l,11,N}^{(3)}(z_1,z_2)|<C_{\eta_0}\), and \(|\mcV_{st,11,N}^{(3)}(z_1,z_2)|,|\mcW_{st,11,N}^{(3)}(z_1,z_2)|\leq C_{\eta_0}\) by definitions in (\ref{Eq of mcV tightness d=3}) and (\ref{Eq of mcW tightness d=3}), so \(|\mcF_{st,11,N}^{(3)}(z_1,z_2)|\leq C_{\eta_0}\) and \(\Vert\bbF_{11,N}^{(3)}(z_1,z_2)\Vert\leq C_{\eta_0}\). Thus,
	$$\lim_{N\to\infty}\Vert\bbC_{11,N}^{(3)}(z_1,z_2)\Vert\leq\lim_{N\to\infty}\Vert\bbTheta_N^{(3)}(z_1,z_2)^{-1}\Vert\cdot\Vert\bbF_{11,N}^{(3)}(z_1,z_2)\Vert\leq C_{\eta_0}.$$
	For other \(\mcC_{s_1t_1,s_2t_2,N}^{(3)}(z_1,z_2)\), we can repeat the previous procedures to derive the system of equations for \(\bbC_{s_2t_2,N}^{(3)}(z_1,z_2):=[\mcC_{s_1t_1,s_2t_2,N}^{(3)}(z_1,z_2)]_{3\times3}\) for each fixed \(1\leq s_2,t_2\leq 3\), since the arguments are analogous, we omit the details here.
\end{proof}
\subsection{Characteristic function}\label{Sec of CLT sub}
\begin{thm}\label{Thm of CLT}
	Under Assumptions {\rm \ref{Ap of general noise}} and {\rm \ref{Ap of dimension}}, when \(d=3\), \(\tr(\boldsymbol{Q}(z))-\mathbb{E}[\tr(\boldsymbol{Q}(z))]\) converges weakly to a Gaussian random process in \(\mathcal{S}_{\eta_0}\) {\rm (\ref{Eq of stability region})}.
\end{thm}
\begin{proof}
	First, we define
	$$\gamma_l(z):=\tr(\boldsymbol{Q}^{ll}(z))-\mathbb{E}[\tr(\boldsymbol{Q}^{ll}(z))],\ \ (\mathfrak{a}(\tau),\mathfrak{b}(\tau)):=\left\{\begin{array}{ll}
		(1/2,1/2)&\tau=1\\(1/2{\rm i},-1/2{\rm i})&\tau={\rm i}
	\end{array}\right.$$
	where \(l=1,2,3\) and \(\gamma(z):=\sum_{l=1}^3\gamma_l(z)\). Besides, let 
	\begin{align}
		e_q:=e_q(\bbt_q,\bbtau_q,\bbz_q):=\exp\left({\rm i}\sum_{s=1}^qt_s\left(\mathfrak{a}(\tau_s)\gamma(z_s)+\mathfrak{b}(\tau_s)\gamma(\bar{z}_s)\right)\right)\quad{\rm for\ }q\in\mathbb{N}^+,\label{Eq of eq}
	\end{align}
	where \(\bbt_q:=(t_1,\cdots,t_q)',\bbtau_q:=(\tau_1,\cdots,\tau_q),\bbz_q=(z_1,\cdots,z_q)\) and \(\tau_s\in\{1,{\rm i}\},z_s\in\mathcal{S}_{\eta_0}\). Notice that
	$$\frac{\partial}{\partial t_s}\mathbb{E}[e_q]={\rm i}\mathbb{E}\left[e_q\left(\mathfrak{a}(\tau_s)\gamma(z_s)+\mathfrak{b}(\tau_s)\gamma(\bar{z}_s)\right)\right],$$
	and we will show that there exists a set of covariance coefficients \(A_{st},s,t=1,\cdots,q\) such that for each fixed \(T_q\)
	\begin{align}
		\lim_{N\to\infty}\Big|\mathbb{E}\left[e_q\left(\mathfrak{a}(\tau_s)\gamma(z_s)+\mathfrak{b}(\tau_s)\gamma(\bar{z}_s)\right)\right]+\mathbb{E}[e_q]\sum_{w=1}^qt_wA_{sw}\Big|=0.\label{Eq of characteristic function d=3}
	\end{align}
	Since
	$$\mathbb{E}[e_q\gamma_l(z)]=\frac{z^{-1}}{\sqrt{N}}\sum_{i,j,k=1}^{m,n,p}\mathbb{E}\big[X_{ijk}e_q^cF_{ijk}^l(z)\big],\ {\rm where\ }F_{ijk}^l(z):=\left\{\begin{array}{ll}
		c_kQ_{ij}^{12}(z)+b_jQ_{ik}^{13}(z)&l=1\\c_kQ_{ij}^{12}(z)+a_iQ_{jk}^{23}(z)&l=2\\a_iQ_{jk}^{23}(z)+b_jQ_{ik}^{13}(z)&l=3
	\end{array}\right.$$
	for \(z,z_1,\cdots,z_q\in\mathcal{S}_{\eta_0}\). Next, we only compute \(\mathbb{E}[e_q\gamma_1(z)]\) in detail, since the arguments for $\mathbb{E}[e_q\gamma_2(z)]$ and $\mathbb{E}[e_q\gamma_3(z)]$ are analogous. For convenience, we define
	\begin{align}
		\mcC_{l,e,N}^{(3)}(z;\bbt_q,\bbtau_q,\bbz_q):=\mcC_{l,e,N}^{(3)}:={\rm Cov}(\tr(\boldsymbol{Q}^{ll}(z)),e_q)\quad{\rm for\ }l=1,2,3.\label{Eq of Ci}
	\end{align}
	By the cumulant expansion (\ref{Eq of cumulant expansion}), we have
	\begin{align}
		&z\mbE[e_q\gamma_l(z)]=z\Cov(\gamma_1(z),e_q)=\mcC_{1,e,N}^{(3)}(z;\bbt_q,\bbtau_q,\bbz_q)\notag\\
		&=\frac{1}{\sqrt{N}}\sum_{i,j,k=1}^{m,n,p}\left(\sum_{l=0}^3\frac{\kappa_{l+1}}{l!}\mathbb{E}\big[\partial_{ijk}^{(l)}\left\{F_{ijk}^1(z)e_q^c\right\}\big]+\epsilon_{ijk}^{(4)}\right).\notag
	\end{align}
	Similar to proofs of Theorem \ref{Thm of Variance}, we can show that only \(\partial_{ijk}^{(1)}\{F_{ijk}^1(z)e_q^c\}\) and \(\partial_{ijk}^{(3)}\{F_{ijk}^1(z)e_q^c\}\) contain major terms, the details are omitted for brevity and only present the final results:

    \vspace{5mm}
    \noindent
    {\bf First derivatives:} When \(l=1\), since \(|e_q|\leq1\) and by Lemma \ref{Thm of Entrywise almost sure convergence}, we can obtain
		$$\mathbb{E}[\partial_{ijk}^{(1)}\{F_{ijk}^1(z)\}e_q^c]=-N^{-1}\mbE[\tr(\bbQ^{11}(z))\tr(\bbQ^{22}(z)+\bbQ^{33}(z))e_q^c]+\mrO(C_{\eta_0}N^{-\omega}).$$
		By the same trick as (\ref{Eq of covariance trick}), we have
		\begin{align}
			&{\rm Cov}(N^{-1}\tr(\boldsymbol{Q}^{11}(z))\tr(\boldsymbol{Q}^{22}(z)),e_q)=\mathbb{E}\Big[\big(N^{-1}\tr(\boldsymbol{Q}^{11}(z))\tr(\boldsymbol{Q}^{22}(z))-N^{-1}\mathbb{E}[\tr(\boldsymbol{Q}^{11}(z))]\tr(\boldsymbol{Q}^{22}(z))\notag\\
			&+N^{-1}\mathbb{E}[\tr(\boldsymbol{Q}^{11}(z))]\tr(\boldsymbol{Q}^{22}(z))-N^{-1}\mathbb{E}[\tr(\boldsymbol{Q}^{11}(z))]\mathbb{E}[\tr(\boldsymbol{Q}^{22}(z))]+N^{-1}\mathbb{E}[\tr(\boldsymbol{Q}^{11}(z))]\mathbb{E}[\tr(\boldsymbol{Q}^{22}(z))]\notag\\
			&-N^{-1}\mathbb{E}[\tr(\boldsymbol{Q}^{11}(z))\tr(\boldsymbol{Q}^{22}(z))]\big)e_q^c\Big]=\mfm_1(z)\mcC_{2,e,N}^{(3)}+\mathbb{E}\left[N^{-1}\tr(\boldsymbol{Q}^{22}(z))\tr(\boldsymbol{Q}^{11}(z))^ce_q^c\right].\notag
		\end{align}
		where \(\mfm_l(z)=N^{-1}\mathbb{E}[\tr(\boldsymbol{Q}^{ll}(z))]\). According to Lemma \ref{Thm of Entrywise almost sure convergence}, we have
		\begin{align}
			&\big|\mathbb{E}\left[N^{-1}\tr(\boldsymbol{Q}^{22}(z))\tr(\boldsymbol{Q}^{11}(z))^c\overline{e_q^c}\right]-m_2(z){\rm Cov}(\tr(\boldsymbol{Q}^{11}(z)),e_q)\big|\notag\\
			&\leq\mrO(\eta_0^{-5}N^{-\omega})\mathbb{E}\left[|\tr(\boldsymbol{Q}^{11}(z))^c\overline{e_q^c}|\right]+\mathbb{E}\left[|N^{-1}\tr(\boldsymbol{Q}^{22}(z))\tr(\boldsymbol{Q}^{11}(z))^c\overline{e_q^c}|1_{|\tr(\boldsymbol{Q}^{22}(z))^c|>\eta_0^{-5}N^{1-\omega}}\right]\notag\\
			&\leq\mrO(\eta_0^{-5}N^{-\omega})\sqrt{{\rm Var}(\tr(\boldsymbol{Q}^{11}(z))){\rm Var}(e_q)}+N^2\exp(-CN^{1-2\omega})=\mrO(C_{\eta_0}N^{-\omega}),\notag
		\end{align}
		where we use the fact that \(\Var(e_q)\leq\mbE[|e_q|^2]\leq1\). As a result, we conclude that
		$${\rm Cov}(N^{-1}\tr(\boldsymbol{Q}^{11}(z))\tr(\boldsymbol{Q}^{22}(z)),e_q)=\mfm_1(z)\mcC_{2,e,N}^{(3)}+\mfm_2(z)\mcC_{1,e,N}^{(3)}+\mrO(C_{\eta_0}N^{-\omega}).$$
		Similarly, the following equation is also valid:
		$${\rm Cov}(N^{-1}\tr(\boldsymbol{Q}^{11}(z))\tr(\boldsymbol{Q}^{33}(z)),e_q)=\mfm_1(z)\mcC_{3,e,N}^{(3)}+\mfm_3(z)\mcC_{1,e,N}^{(3)}+\mrO(C_{\eta_0}N^{-\omega}).$$
		Moreover, since 
		\begin{align*}
			\partial_{ijk}^{(1)}\{e_q\}=-\frac{{\rm i} e_q}{\sqrt{N}}\sum_{q=1}^s\sum_{w=1}^3\sum_{s_1\neq s_2}^3t_s\mcA_{ijk}^{(s_1,s_2)}\big[\mfa(\tau_s)Q_{\tilde{s}_1\cdot}^{s_1w}(z_s)Q_{\cdot\tilde{s}_2}^{ws_2}(z_s)+\mfb(\tau_s)Q_{\tilde{s}_1\cdot}^{s_1w}(\bar{z}_s)Q_{\cdot\tilde{s}_2}^{ws_2}(\bar{z}_s)\big],
		\end{align*}
		by Lemmas \ref{Lem of minor terms} and \ref{Cor of minor terms}, we have
		\begin{align*}
			&N^{-1/2}\sum_{i,j,k=1}^{m,n,p}\mathbb{E}[\partial_{ijk}^{(1)}\{e_q\}F_{ijk}^1(z)]=\mrO(C_{\eta_0}N^{-1/2})\\
			&-\frac{2{\rm i}}{N}\sum_{s=1}^q\sum_{l\neq1}^3\sum_{w=1}^3t_s\mbE[\tr(\bbQ^{1l}(z)(\mfa(\tau_s)\bbQ^{lw}(z_s)\bbQ^{w1}(z_s)+\mfb(\tau_s)\bbQ^{lw}(\bar{z}_s)\bbQ^{w1}(\bar{z}_s)))\cdot e_q].
		\end{align*}
		By Lemma \ref{Thm of Entrywise almost sure convergence} and the fact of \(|e_q|\leq1\), we have
		\begin{align*}
			&N^{-1/2}\sum_{i,j,k=1}^{m,n,p}\mathbb{E}[\partial_{ijk}^{(1)}\{e_q\}F_{ijk}^1(z)]=\mrO(C_{\eta_0}N^{-\omega})\\
			&-\frac{2{\rm i}}{N}\sum_{s=1}^q\sum_{l\neq1}^3\sum_{w=1}^3t_s\mbE[\tr(\bbQ^{1l}(z)(\mfa(\tau_s)\bbQ^{lw}(z_s)\bbQ^{w1}(z_s)+\mfb(\tau_s)\bbQ^{lw}(\bar{z}_s)\bbQ^{w1}(\bar{z}_s)))]\cdot\mbE[e_q].
		\end{align*}
		Hence, for \(i\in\{1,2,3\}\), we define
        \begin{align*}
            \mcV_{i,e,N}^{(3)}(z,z_s):=\frac{1}{N}\sum_{l\neq i}^3\sum_{w=1}^3\mbE[\bbQ^{il}(z)\bbQ^{lw}(z_s)\bbQ^{wi}(z_s)],
        \end{align*}
        we can obtain that
		\begin{align}
			&N^{-1/2}\sum_{i,j,k=1}^{m,n,p}\mathbb{E}[\partial_{ijk}^{(1)}\{F_{ijk}^1(z)e_q^c\}]=-\mfm_1(z)\big(\mcC_{2,e,N}^{(3)}+\mcC_{3,e,N}^{(3)}\big)-\mcC_{1,e,N}^{(3)}(\mfm_2(z)+\mfm_3(z))\notag\\
			&-2{\rm i}\mbE[e_q]\sum_{s=1}^qt_s\big(\mfa(\tau_s)\mcV_{1,e,N}^{(3)}(z,z_s)+\mfb(\tau_s)\mcV_{1,e,N}^{(3)}(z,\bar{z}_s)\big)+\mrO(C_{\eta_0}N^{-\omega}).\label{Eq of characteristic 1 d=3}
		\end{align}
    {\bf Second derivatives:} The calculations for the second derivatives are the same as those in Theorem \ref{Thm of Variance}, we can show that
    \begin{align*}
		N^{-1/2}\Bigg|\sum_{i,j,k=1}^{m,n,p}\mbE\big[\partial_{ijk}^{(2)}\{F_{ijk}^1(z)e_q^c\}\big]\Bigg|\leq\mrO(C_{\eta_0}N^{-\omega}),
	\end{align*}
    we omit the details here.

    \vspace{5mm}
    \noindent
    {\bf Third derivatives:} When \(l=3\), only the following one contains the major terms:
		\begin{align*}
			N^{-1/2}\sum_{i,j,k=1}^{m,n,p}\mbE\big[\partial_{ijk}^{(2)}\{e_q^c\}\partial_{ijk}^{(1)}\{F_{ijk}^1(z)\}\big],
		\end{align*}
		where
		\begin{align*}
			&\partial_{ijk}^{(2)}\{e_q^c\}=-e_qA_1-2{\rm i} e_qA_2,
		\end{align*}
        and
        \begin{align*}
            A_1:&=\frac{1}{N}\Big(\sum_{q=1}^s\sum_{w=1}^3\sum_{s_1\neq s_2}^3t_s\mcA_{ijk}^{(s_1,s_2)}\big[\mfa(\tau_s)Q_{\tilde{s}_1\cdot}^{s_1w}(z_s)Q_{\cdot\tilde{s}_2}^{ws_2}(z_s)+\mfb(\tau_s)Q_{\tilde{s}_1\cdot}^{s_1w}(\bar{z}_s)Q_{\cdot\tilde{s}_2}^{ws_2}(\bar{z}_s)\big]\Big)^2,\\
            A_2:&=\frac{1}{N}\sum_{q=1}^s\sum_{w=1}^3\sum_{s_1\neq s_2}^3t_s\mcA_{ijk}^{(s_1,s_2)}\mcA_{ijk}^{(s_3,s_4)}\big[\mfa(\tau_s)Q_{\tilde{s}_1\tilde{s}_3}^{s_1s_3}(z_s)Q_{\tilde{s}_4\cdot}^{s_4w}(z_s)Q_{\cdot\tilde{s}_2}^{ws_2}(z_s)+\mfb(\tau_s)Q_{\tilde{s}_1\tilde{s}_3}^{s_1s_3}(\bar{z}_s)Q_{\tilde{s}_4\cdot}^{s_4w}(\bar{z}_s)Q_{\cdot\tilde{s}_2}^{ws_2}(\bar{z}_s)\big].
        \end{align*}
		and
		\begin{align*}
			&\partial_{ijk}^{(1)}\{F_{ijk}^1(z)\}=-\frac{1}{\sqrt{N}}\sum_{s_1\neq s_2}^3\sum_{l\neq1}^3\mcA_{ijk}^{(1,l)}\mcA_{ijk}^{(s_1,s_2)}Q_{i\tilde{s}_1}^{1s_1}(z)Q_{\tilde{s}_2\tilde{l}}^{s_2l}(z).
		\end{align*}
		For the \(A_1\) in \(\partial_{ijk}^{(2)}\{e_q^c\}\), it is easy to see it only contains the off-diagonal terms since \(s_1\neq s_2\), so by Lemma \ref{Cor of minor terms}, if it associates with \(\partial_{ijk}^{(1)}\{F_{ijk}^1(z)\}\), the summation over all \(i,j,k\) will be minor with order of \(C_{\eta_0}N^{-1/2}\). In fact, although \(A_1\) is a square of the summation of off-diagonal terms, we can use Cauchy's inequality  to transform it as the summation of square of off-diagonal terms, then we can claim it is a minor term by Lemma \ref{Cor of minor terms}. Next, by Lemma \ref{Lem of minor terms}, we have
		\begin{align*}
			&N^{-1/2}\sum_{i,j,k=1}^{m,n,p}\mbE\big[A_2e_q\partial_{ijk}^{(1)}\{F_{ijk}^1(z)\}\big]\\
			&=-\frac{1}{N^2}\sum_{i,j,k=1}^{m,n,p}\sum_{q=1}^s\sum_{w=1}^3\sum_{l\neq1}^3\sum_{s_1,s_2}^{(1,l)}t_s(\mcA_{ijk}^{(1,l)})^4\mbE\big[(Q_{ii}^{11}(z)Q_{\tilde{l}\tilde{l}}^{ll}(z))\cdot\big(\mfa(\tau_s)Q_{\tilde{s}_1\tilde{s}_1}^{s_1s_1}(z_s)Q_{\tilde{s}_2\cdot}^{s_2w}(z_s)Q_{\cdot\tilde{s}_2}^{ws_2}(z_s)\\
			&+\mfb(\tau_s)Q_{\tilde{s}_1\tilde{s}_1}^{s_1s_1}(\bar{z}_s)Q_{\tilde{s}_2\cdot}^{s_2w}(\bar{z}_s)Q_{\cdot\tilde{s}_2}^{ws_2}(\bar{z}_s)\big)\cdot e_q\big]+\mrO(C_{\eta_0}N^{-1/2})\\
			&=-\frac{1}{N^2}\sum_{q=1}^s\sum_{w=1}^3\sum_{l\neq1}^3\sum_{s_1,s_2}^{(1,l)}t_s\Vert\bba^{(5-l)}\Vert_4^4\mbE\big[\big(\mfa(\tau_s)\tr(\bbQ^{s_1s_1}(z)\circ\bbQ^{s_1s_1}(z_s))\tr(\bbQ^{s_2s_2}(z)\circ(\bbQ^{s_2w}(z_s)\bbQ^{ws_2}(z_s)))\\
			&+\mfb(\tau_s)\tr(\bbQ^{s_1s_1}(z)\circ\bbQ^{s_1s_1}(\bar{z}_s))\tr(\bbQ^{s_2s_2}(z)\circ(\bbQ^{s_2w}(\bar{z}_s)\bbQ^{ws_2}(\bar{z}_s)))\big)\cdot e_q\big]+\mrO(C_{\eta_0}N^{-1/2}).
		\end{align*}
        For simplicity, for $i\in\{1,2,3\}$, we define
        \begin{align*}
            \mcW_{i,e,N}^{(3)}(z,z_s):=\frac{1}{N^2}\sum_{w=1}^3\sum_{l\neq i}^3\sum_{s_1,s_2}^{(1,i)}\Vert\bba^{(5-l)}\Vert_4^4\mbE[\tr(\bbQ^{s_1s_1}(z)\circ\bbQ^{s_1s_1}(z_s))\tr(\bbQ^{s_2s_2}(z)\circ(\bbQ^{s_2w}(z_s)\bbQ^{ws_2}(z_s))]),
        \end{align*}
        where the notation $\sum_{s_1,s_2}^{(i,l)}$ means that the summation of $s_1$ and $s_2$ are over $\{1,2,3\}\backslash\{i,r\}$. Then by Lemma \ref{Thm of Entrywise almost sure convergence}, we can further obtain
		\begin{align}
			&N^{-1/2}\sum_{i,j,k=1}^{m,n,p}\mbE\big[A_2e_q\partial_{ijk}^{(1)}\{F_{ijk}^1(z)\}\big]\notag\\
            &=\mbE[e_q]\sum_{q=1}^st_s\big(\mfa(\tau_s)\mcW_{1,e,N}^{(3)}(z,z_s)+\mfb(\tau_s)\mcW_{1,e,N}^{(3)}(z,\bar{z}_s)\big)+\mrO(C_{\eta_0}N^{-1/2}).\label{Eq of characteristic 3 d=3}
		\end{align}
	Now, let 
    $$\mcF_{l,e,N}^{(3)}(z,z_s):=\mcV_{l,e,N}^{(3)}(z,z_s)+\kappa_4\mcW_{l,e,N}^{(3)}(z,z_s),$$
    then combining (\ref{Eq of characteristic 1 d=3}) and (\ref{Eq of characteristic 3 d=3}), we have
	\begin{align*}
		&(z+\mfm(z)-\mfm_1(z))\mcC_{1,e,N}^{(3)}\\
		&=-\mfm_1(z)\big(\mcC_{2,e,N}^{(3)}+\mcC_{3,e,N}^{(3)}\big)-{\rm i}\mbE[e_q]\sum_{s=1}^qt_s\big[\mfa(\tau_s)\mcF_{1,e,N}^{(3)}(z,z_s)+\mfb(\tau_s)\mcF_{1,e,N}^{(3)}(z,\bar{z}_s)\big]+\mrO(C_{\eta_0}N^{-\omega}).
	\end{align*}
	Similarly, for other \(r\in\{1,2,3\}\), we have
	\begin{align*}
		&(z+\mfm(z)-\mfm_r(z))\mcC_{r,e,N}^{(3)}=-\mfm_r(z)\sum_{l\neq r}^3\mcC_{l,e,N}^{(3)}+\mrO(C_{\eta_0}N^{-\omega})\\
		&-{\rm i}\mbE[e_q]\sum_{s=1}^qt_s\big[\mfa(\tau_s)\mcF_{r,e,N}^{(3)}(z,z_s)+\mfb(\tau_s)\mcF_{r,e,N}^{(3)}(z,\bar{z}_s)\big].
	\end{align*}
	Next, define
	\begin{align*}
		\bbC_{e,N}^{(3)}(z):=[\mcC_{r,e,N}^{(3)}(z;\bbt_q,\bbtau_q,\bbz_q)]_{3\times1}\quad{\rm and}\quad\bbF_{e,N}^{(3)}(z,z_s):=[\mcF_{r,e,N}^{(3)}(z;z_s)]_{3\times1},
	\end{align*}
	so 
	\begin{align*}
		\bbTheta_N^{(3)}(z,z)\bbC_{e,N}^{(3)}(z)=-{\rm i}\mbE[e_q]\sum_{s=1}^qt_s\big[\mfa(\tau_s)\bbF_{e,N}^{(3)}(z,z_s)+\mfb(\tau_s)\bbF_{e,N}^{(3)}(z,\bar{z}_s)\big]+\mro(\boldsymbol{1}_{3\times3}),
	\end{align*}
	where \(\bbTheta_N^{(3)}(z,z)\) defined in (\ref{Eq of bbTheta d=3}) is invertible, and we have shown that \(\lim_{N\to\infty}\Vert\bbTheta_N^{(3)}(z,z)^{-1}+\bbPi^{(3)}(z,z)^{-1}\diag(\mfc^{-1}\circ\bbg(z))\Vert=0\). Thus, we obtain that
	\begin{align*}
		\lim_{N\to\infty}\Bigg\Vert\bbC_{e,N}^{(3)}(z)-{\rm i}\mbE[e_q]\sum_{s=1}^qt_s\bbPi^{(3)}(z,z)^{-1}\diag(\mfc^{-1}\circ\bbg(z))\big[\mfa(\tau_s)\bbF_{e,N}^{(3)}(z,z_s)+\mfb(\tau_s)\bbF_{e,N}^{(3)}(z,\bar{z}_s)\big]\Bigg\Vert=0.
	\end{align*}
	As a result,
	\begin{align*}
		&\frac{\partial}{\partial t_s}\mathbb{E}[e_q]={\rm i}\mathbb{E}\left[e_q\left(\mathfrak{a}(\tau_s)\gamma(z_s)+\mathfrak{b}(\tau_s)\gamma(\bar{z}_s)\right)\right]={\rm i}\sum_{r=1}^3\big(\mfa(\tau_s)\mcC_{r,e,N}^{(3)}(z_s)+\mfa(\tau_s)\mcC_{r,e,N}^{(3)}(\bar{z}_s)\big)\\
		&=-\mbE[e_q]\sum_{w=1}^qt_w\Big(\mfa(\tau_s)\boldsymbol{1}_3'\bbPi^{(3)}(z_s,z_s)^{-1}\diag(\mfc^{-1}\circ\bbg(z_s))\big[\mfa(\tau_w)\bbF_{e,N}^{(3)}(z_s,z_w)+\mfb(\tau_w)\bbF_{e,N}^{(3)}(z_s,\bar{z}_w)\big]\\
		&+\mfb(\tau_s)\boldsymbol{1}_3'\bbPi^{(3)}(\bar{z}_s,\bar{z}_s)^{-1}\diag(\mfc^{-1}\circ\bbg(\bar{z}_s))\big[\mfa(\tau_w)\bbF_{e,N}^{(3)}(\bar{z}_s,z_w)+\mfb(\tau_w)\bbF_{e,N}^{(3)}(\bar{z}_s,\bar{z}_w)\big]\Big)+\mro(1),
	\end{align*}
	which concludes (\ref{Eq of characteristic function d=3}). Hence, combined with Theorem \ref{Thm of Tightness}, we can conclude that \(\tr(\bbQ(z))-\mbE[\tr(\bbQ(z))]\) converges weakly to a Gaussian process in \(\mcS_{\eta_0}\).
\end{proof}
\subsection{Proof of Theorem \ref{Thm of CLT LSS d=3}}\label{sec of proof CLT d=3}
Having established tightness in \S\ref{Sec of Tightness} and convergence of finite-dimensional distributions in \S\ref{Sec of CLT sub}, we now complete the proof of Theorem \ref{Thm of CLT LSS d=3}.
\begin{proof}[Proof of Theorem \ref{Thm of CLT LSS d=3}]
	First, since
	\begin{align}
		&G_N(f)\overset{\mbP}{\longrightarrow}-\frac{1}{2\pi{\rm i}}\oint_{\mathfrak{C}}f(z)\{\tr(\bbQ(z))-Ng(z)\}dz\notag\\
		&=-\frac{1}{2\pi{\rm i}}\oint_{\mathfrak{C}}f(z)\{\tr(\bbQ(z))-\mathbb{E}[\tr(\bbQ(z))]\}dz-\frac{1}{2\pi{\rm i}}\oint_{\mathfrak{C}}f(z)\{\mathbb{E}[\tr(\bbQ(z))]-Ng(z)\}dz.\notag
	\end{align}
    Here, let us decompose \(\mathfrak{C}\) into \(\mfC^v\cup\mfC^h\), where
	$$\mfC^v:=\left\{z=\pm E_0+{\rm i}\eta\in\mathbb{C}:|\eta|\in[0,\eta_0]\right\}\quad{\rm and}\quad\mfC^h:=\left\{z=E\pm{\rm i}\eta_0\in\mathbb{C}:|E|\in[0,E_0]\right\}.$$
	First, by Theorems \ref{Thm of CLT}, \ref{Thm of Variance} and \ref{Thm of Mean function}, we know that $\tr(\bbQ(z))-Ng(z)$ converges weakly to a Gaussian process in \(\mathcal{S}_{\eta_0}\) (\ref{Eq of stability region}) with mean \(\mu_N^{(3)}(z)\) (\ref{Eq of Mean function}) and variance \(\mcC_N^{(3)}(z,z)\) (\ref{Eq of covariance function d=3}). Hence, since \(\mfC^v\subset\mcS_{\eta_0}\), we can conclude that
	$$(\sigma_N^{(3)})^{-1}\left(\frac{1}{2\pi{\rm i}}\oint_{\mathfrak{C}^h}f(z)\{\tr(\bbQ(z))-Ng(z)\}dz-\xi_N^{(3)}\right)\overset{d}{\longrightarrow}\mcN(0,1),$$
	where
	\begin{align}
		&\xi_N^{(3)}:=\frac{1}{2\pi{\rm i}}\oint_{\mathfrak{C}}f(z)\mu_N^{(3)}(z)dz,\label{Eq of CLT mean d=3}\\
		&(\sigma_N^{(3)})^2:=-\frac{1}{4\pi^2}\oint_{\mfC_1}\oint_{\mfC_2}f(z_1)f(z_2)\mcC_N^{(3)}(z_1,z_2)dz_1dz_2,\label{Eq of CLT variance d=3}
	\end{align}
	where \(\mfC_{1,2}\) are two disjoint rectangular contours with vertices of \(\pm E_{1,2}\pm{\rm i}\eta_{1,2}\) such that \(E_{1,2}\geq\max\{\zeta,\mfv_3\}+t\) and \(\eta_{1,2}>0\) are sufficiently small. Next, we will show that
	$$\lim_{\eta_0\downarrow0}\limsup_{N\to\infty}\frac{1}{2\pi{\rm i}}\oint_{\mathfrak{C}^v}f(z)\{\tr(\bbQ(z))-Ng(z)\}dz\overset{\mbP}{\longrightarrow}0,$$
	it is enough to prove
	\begin{align}
		&\lim_{\eta_0\downarrow0}\limsup_{N\to\infty}\mathbb{E}\Big|\int_{\mathfrak{C}^v}1_{\mcE_{\bbM}}f(z)\{\tr(\bbQ(z))-\mathbb{E}[\tr(\bbQ(z))]\}dz\Big|^2=0,\label{Eq of CLT minor 1 d=3}\\
		&\lim_{\eta_0\downarrow0}\limsup_{N\to\infty}\Big|\int_{\mathfrak{C}^v}1_{\mcE_{\bbM}}f(z)\{\mathbb{E}[\tr(\bbQ(z))]-Ng(z)\}dz\Big|^2=0,\label{Eq of CLT minor 2 d=3}
	\end{align}
    where the event \(\mcE_{\bbM}:=\{\Vert\bbM\Vert\leq\max\{\mfv_3,\zeta\}+t\}\), \(\mfv_3,\zeta\) are defined in Theorem \ref{Thm of Extreme eigenvalue N d=3} and \eqref{Eq of support boundary}, respectively, and $t>0$ is a fixed constant. By Theorem \ref{Thm of Extreme eigenvalue N d=3}, we know that \(\mbP(\mcE_{\bbM})\geq1-\mro(N^{-l})\) for any \(l>0\). Now, let us first prove (\ref{Eq of CLT minor 2 d=3}). By the definition of \(\mathfrak{C}^v\), we know that \(\operatorname{dist}(z,[-\mfv_3,\mfv_3])>t\) conditional on \(\mcE_{\bbM}\), so \(\Vert\bbQ(z)\Vert\leq t^{-1}\) for any \(z\in\mathfrak{C}^v\). Hence, we can use the same proofs of Theorem \ref{Thm of Mean function} to conculde that
	\begin{align*}
		\mathbb{E}[\tr(\bbQ(z))]-Ng(z)&=\mu_N^{(3)}(z)+\mrO(C_tN^{-\omega}),
	\end{align*}
	where the error term \(\mrO(C_tN^{-\omega})\) is independent of \(\eta_0\). Moreover, by Lemma \ref{Lem of analytic} later, we know that \(\mu_N^{(3)}(z)\) is analytic on \(\mfC^v\). And \(f(z)\) is also analytic on \(\mfC^v\) due to $f\in\mathfrak{F}_3$ in \eqref{Eq of analytic function d=3}. Therefore, conditional on \(\mcE_{\bbM}\), we have
	\begin{align*}
		\eqref{Eq of CLT minor 2 d=3}&\leq\lim_{\eta_0\downarrow0^+}\limsup_{N\to\infty}\Big(\oint_{\mfC^v}\big|f(z)(\mathbb{E}[\tr(\bbQ(z))]-Ng(z))\big|dz\Big)^2\\
		&\leq\lim_{\eta_0\downarrow0^+}\limsup_{N\to\infty}\Big(\oint_{\mfC^v}\big|f(z)\mu_N^{(3)}(z)\big|dz\Big)^2+2\lim_{\eta_0\downarrow0^+}\limsup_{N\to\infty}\eta_0C_tN^{-\omega}\sup_{z\in\mfC^v}|f(z)|=0,
	\end{align*}
	where we use the fact that \(\lim_{\eta_0\downarrow0^+}\oint_{\mfC^v}\big|f(z)\mu_N^{(3)}(z)\big|dz=0\) due to \(f(z)\mu_N^{(3)}(z)\) is analytic on \(\mfC^v\) and the length of \(\mfC^v\) tends to \(0\). Next, for (\ref{Eq of CLT minor 1 d=3}), conditional on \(\mcE_{\bbM}\), we have
	\begin{align*}
		&\eqref{Eq of CLT minor 1 d=3}\leq\lim_{\eta_0\downarrow0^+}\limsup_{N\to\infty}\oint_{\mfC^v}\mbE\big[\big|f(z)(\tr(\bbQ(z))-\mbE[\tr(\bbQ(z))])\big|^2\big]dz\\
    &=\lim_{\eta_0\downarrow0^+}\limsup_{N\to\infty}\oint_{\mfC^v}|f(z)|^2\Var(\tr(\bbQ(z)))dz.
	\end{align*}
	By the same proofs of Theorem \ref{Thm of Variance}, we know that
	\begin{align*}
		\Var(\tr(\bbQ(z)))=\mcC_N^{(3)}(z,z)+\mrO(C_tN^{-\omega}).
	\end{align*}
	Similarly, since \(\mcC_N^{(3)}(z,z)\) is analytic by Lemma \ref{Lem of analytic}, so we have
	$$(\ref{Eq of CLT minor 1 d=3})\leq\lim_{\eta_0\downarrow0^+}\limsup_{N\to\infty}\oint_{\mfC^v}|f(z)|^2\Var(\tr(\bbQ(z)))dz=0.$$
	Hence, we conclude (\ref{Eq of CLT minor 1 d=3}) and (\ref{Eq of CLT minor 2 d=3}), which completes our proof.
\end{proof}
\begin{lem}\label{Lem of analytic}
    The mean function \(\mu_N(z)\) in {\rm (\ref{Eq of Mean function})} and the covariance function \(\mcC_N^{(3)}(z_1,z_2)\) in {\rm (\ref{Eq of covariance function d=3})} are analytic for \(z,z_1,z_2\in\mathfrak{C}^v\).
\end{lem}
\begin{proof}
	We first need to ensure that \(\boldsymbol{\Pi}^{(3)}(z_1,z_2)\) is still invertible when \(z\in\mathfrak{C}^v\). Based on the definition of \(\boldsymbol{\Pi}^{(3)}(z_1,z_2)\) in \eqref{Eq of invertible 2} and the proof of Proposition \ref{Pro of invertible matrices}, it is enough to show that \(\operatorname{diag}(e^{-2{\rm i}\boldsymbol{q}})-\bbF^{(3)}(z)\) is invertible, where \(\boldsymbol{q},\bbF^{(3)}(z)\) are defined in (\ref{Eq of MDE taking image}) and (\ref{Eq of operator F}). For convenience, we simplify \(\bbF^{(3)}(z)\) by $\bbF$. let us first show that \(\Vert\bbF(z)\Vert<1\) for all \(z\in\mathfrak{C}^v\). Since
	$$\Vert\bbF\Vert=1-\frac{\Im(z)\langle\boldsymbol{f},\boldsymbol{\mfc}^{-1}\circ|\bbg|\rangle}{\langle\boldsymbol{f},\sin\boldsymbol{q}\rangle},$$
	where \(\sin\boldsymbol{q}=\frac{\Im(\bbg)}{|\bbg|}\) and \(\boldsymbol{f}\) is the unit eigenvector of \(\bbF\) with eigenvalue of \(\Vert\bbF\Vert\). Denote \(\nu_i\) be the measure deduced by \(g_i(z)\) for \(i=1,\cdots,d\), whose support is bounded by \(\zeta\) defined in (\ref{Eq of support boundary}), then
	$$\Im(g_i(z))=\int_{-\xi}^{\xi}\frac{\eta}{(E_0-x)^2+\eta^2}\nu_i(dx),$$
	where \(z=E_0+{\rm i}\eta\in\mathfrak{C}^v\) and \(|E_0|-\xi=t>0,\eta\in[0,\eta_0]\). Hence,
	$$\sup_{\eta\in[0,\eta_0]}\eta^{-1}\Im(g_i(z))\leq\int_{-\xi}^{\xi}\frac{1}{(E_0-x)^2}\nu_i(dx)<t^{-2}\mfc_i.$$
    On the other hand, we have
    \begin{align*}
        &\inf_{\eta\in[0,\eta_0]}|g_i(E_0+{\rm i}\eta)|\geq\inf_{\eta\in[0,\eta_0]}|\Re(g_i(E_0+{\rm i}\eta))|\geq\int_{-\xi}^{\xi}\frac{|E_0-x|}{|E_0-x|^2+\eta_0^2}\nu_i(dx)\\
        &>\frac{t\mfc_i}{t^2+\eta_0^2}>\mfc_i/(2t),
    \end{align*}
    so we have \(\sup_{z\in\mfC^v}\frac{\Im(z)^{-1}\Im(g_i(z))}{|g_i(z)|}\leq2t^{-1}\) and $\Im(z)^{-1}\langle\boldsymbol{f},\sin\boldsymbol{q}\rangle<2t^{-1}\langle\boldsymbol{f},1\rangle$. In addition, since
	\begin{align}
		&\inf_{\eta\in[0,\eta_0]}\mfc_i^{-1}|g_i(E_0+{\rm i}\eta)|=\mfc_i^{-1}\inf_{\eta\in[0,\eta_0]}|\Re(g_i(E_0+{\rm i}\eta))|>\frac{t}{t^2+\eta_0^2},\notag
	\end{align}
	then 
	$$\langle\boldsymbol{f},\boldsymbol{\mfc}^{-1}\circ|\bbg|\rangle>\frac{t}{t^2+\eta_0^2}\langle\boldsymbol{f},\boldsymbol{1}\rangle.$$
	Therefore, we conclude that
	$$\sup_{\eta\in[0,\eta_0]}\Vert\bbF(E+{\rm i}\eta)\Vert\leq1-\frac{t^2}{2(t^2+\eta_0^2)}<1,$$
    combined with Lemma \ref{Lem of spectral gap}, we can further conclude that \({\rm diag}(e^{-2{\rm i}\boldsymbol{q}})-\bbF(z)\) is invertible for $z\in\mfC^v$, as does \(\boldsymbol{\Pi}^{(3)}(z,z)\). Moreover, by the same proofs of Proposition \ref{Pro of invertible matrices}, we can show that \(\boldsymbol{\Pi}^{(3)}(z_1,z_2)\) is invertible for $z_1,z_2\in\mfC^v$. Since \(g_i(z)\) are analytic on \(\mathfrak{C}^v\), then the entries of \(\boldsymbol{\Pi}^{-1}(z,z)\) are also analytic; further based on the system equations in \S\ref{Sec of majors}, \(W_{ij}^{(3)}(z)\) \eqref{Eq of solve W} and \(V_{ij}^{(3)}(z,z)\) \eqref{Eq of bbV d=3} are all analytic, so the mean function \(\mu_N(z)\) in (\ref{Eq of Mean function}) is also analytic. Similarly, by Theorem \ref{Thm of Variance} and system equations in \S\ref{Sec of majors}, \(\mcC_N^{(3)}(z_1,z_2)\) is analytic on \(\mfC^v\) due to the system of equations for \(\mcV_{ij}^{(3)}(z_1,z_2)\) \eqref{Eq of mcV limiting d=3} and \(\mcW_{ij,N}^{(3)}(z_1,z_2)\) \eqref{Eq of mcW limiting d=3} only depend on \(\bbg(z),\bbPi^{(3)}(z_1,z_2)\), so the covariance function \(\mcC_N^{(3)}(z_1,z_2)\) in (\ref{Eq of covariance function d=3}) is also analytic, which completes our proof.
\end{proof}

\section{General cases}\label{Sec of General cases}
\setcounter{equation}{0}
\def\theequation{\thesection.\arabic{equation}}
\setcounter{subsection}{0}
In this section, we extend all results from \S\ref{Sec of entrywise law}, \S\ref{Sec of mean and covariance} and \S\ref{Sec of CLT} to general \(d\geq3\). Since the proof procedures for $d\geq3$ are analogous to those for $d=3$, we present only the key calculations to highlight the differences. Moreover, in proofs of these generalized results, we only present the key calculations to highlight the differences. Before presenting the details, we establish some notation. Recall the blockwise tensor contraction mapping \(\bbPhi_d\) defined in (\ref{Eq of tensor contraction}), denote
\begin{align*}
	\bbM=\frac{1}{\sqrt{N}}\bbPhi_d(\bbX,\bba^{(1)},\cdots\bba^{(d)})\quad{\rm and}\quad\bbQ(z)=(\bbM-z\bbI_N)^{-1},
\end{align*}
where \(\bba^{(i)}\in\mbS^{(n_i-1)},i=1\cdots,d\) are \(d\) fixed unit deterministic vectors with bounded $L^2$ norms and \(N=\sum_{i=1}^dn_i\), the dimensions \(n_1,\cdots,n_d\) satisfy Assumption \ref{Ap of dimension}, \(\bbX=[X_{i_1\cdots i_d}]_{n_1\times\cdots\times n_d}\in\mbR^{n_1\times\cdots\times n_d}\) is the random tensor such that \(X_{i_1\cdots i_d}\) are i.i.d. satisfying Assumption \ref{Ap of general noise}. Similar to \eqref{Eq of stability region}, for any sufficiently small \(\eta_0>0\), we define
\begin{align}
    \mcS_{\eta_0}:=\{z\in\mbC^+:{\rm dist}(z,[-\max\{\mfv_d,\zeta\},\max\{\mfv_d,\zeta\}])\geq\eta_0,|\Re(z)|,|\Im(z)|\leq\eta_0^{-1}\}.\label{Eq of stability region general d}
\end{align}
Unless otherwise stated, $\omega\in(1/2-\delta,1/2)$, where $\delta>0$ is a sufficient small constant. For any matrix $A$, $A_{i\cdot}$ and $A_{\cdot j}$ means the $i$-th row and $j$-th column of $A$, respectively. 
\subsection{Preliminary Lemmas}
First, we will extend Lemmas \ref{Thm of Entrywise almost sure convergence}, \ref{Lem of minor terms} and \ref{Cor of minor terms} to general \(d\geq3\). Similar to (\ref{Eq of partial N}), we have
\begin{align*}
	\partial_{i_1\cdots i_d}^{(1)}\bbM=\frac{1}{\sqrt{N}}\left(\begin{array}{cccc}
		\boldsymbol{0}_{n_1\times n_1}&\mcA_{i_1\cdots i_d}^{(1,2)}\bbe_{i_1}^{n_1}(\bbe_{i_2}^{n_2})'&\cdots&\mcA_{i_1\cdots i_d}^{(1,d)}\bbe_{i_1}^{n_1}(\bbe_{i_d}^{n_d})'\\
		\mcA_{i_1\cdots i_d}^{(2,1)}\bbe_{i_2}^{n_2}(\bbe_{i_1}^{n_1})'&\boldsymbol{0}_{n_2\times n_2}&\cdots&\mcA_{i_1\cdots i_d}^{(2,d)}\bbe_{i_2}^{n_2}(\bbe_{i_d}^{n_d})'\\
		\vdots&\cdots&\ddots&\vdots\\
		\mcA_{i_1\cdots i_d}^{(d,1)}\bbe_{i_d}^{n_d}(\bbe_{i_1}^{n_1})'&\cdots&\mcA_{i_1\cdots i_d}^{(d,d-1)}\bbe_{i_{d-1}}^{n_{d-1}}(\bbe_{i_d}^{n_d})'&\boldsymbol{0}_{n_d\times n_d}
	\end{array}\right),
\end{align*}
where \(\bbe_{i_k}^{n_k}\in\mbR^{n_k}\) such that its \(i_k\)-th entry is 1 while others are 0, and \(\partial_{i_1\cdots i_d}^{(l)}:=\frac{\partial^l}{\partial X_{i_1\cdots i_d}^l}\) is the generalization of (\ref{Eq of partial operator}). Since we also split \(\bbQ(z)=[\bbQ^{st}(z)]_{d\times d}\) into \(d\times d\) blocks such that \(\bbQ^{st}(z)\in\mbC^{n_s\times n_t}\) for \(s,t\in\{1,\cdots,d\}\), we say \(\bbQ^{st}(z)\) comes from the off-diagonal block if \(s\neq t\), otherwise it belongs to the diagonal block. Before proving Lemma \ref{Thm of Entrywise almost sure convergence} for general \(d\geq3\), we need extend Lemma \ref{Lem of minor terms 1} first:
\begin{lem}\label{Rem of extension minor 1}
	For any \(K\in\mbN^+\) and \(z\in\mbC^+\), let \(\bbx,\bby\in\mbC^N\) be two deterministic vectors with bounded $L^2$ norms, then for \(l=1,2,3\), we have
	$$\sum_{i_1\cdots i_d}^{n_1\cdots n_d}\big|\boldsymbol{x}'\partial_{i_1\cdots i_d}^{(l)}\Big(\prod_{k=1}^K\bbQ(z)\Big)\boldsymbol{y}\big|^2<\left\{\begin{array}{ll}
	     C_l\Vert\bbQ(z)\Vert^{2(l+K)}N^{-1}&l=1,2,  \\
	     C_l\Vert\bbQ(z)\Vert^{2(l+K)}N^{-2}&l=3. 
	\end{array}\right.$$
\end{lem}
\begin{proof}
	Note that \(\partial_{i_1\cdots i_d}^{(l)}\bbQ=(-1)^ll!(\bbQ\partial_{i_1\cdots i_d}^{(1)}\bbM)^l\bbQ\) for \(l\in\mbN^+\). We consider the following cases:

    \vspace{5mm}
    \noindent
    {\bf First derivatives:} When \(l=1\), we have
		\begin{align*}
			&\sum_{i_1\cdots i_d}^{n_1\cdots n_d}|(\bbx^{(l_0)})'\partial_{i_1\cdots i_d}^{(1)}\bbQ^{s_{l_0}s_{l_0+1}}\bby^{(l_0)}|^2\leq N^{-1}C_d\sum_{i_1\cdots i_d}^{n_1\cdots n_d}\sum_{t_1\neq t_2}^d|(\bbx^{(l_0)})'Q_{\cdot i_{t_1}}^{s_{l_0}t_1}\mcA_{i_1\cdots i_d}^{(t_1,t_2)}Q_{i_{t_2}\cdot}^{t_2s_{l_0+1}}\bby^{(l_0)}|^2\\
			&\leq C_dN^{-1}\sum_{i_1\cdots i_d}^{n_1\cdots n_d}\sum_{t_1\neq t_2}^d\prod_{l\neq t_1,t_2}^d|(\bbx^{(l_0)})'Q_{\cdot i_{t_1}}^{s_{l_0}t_1}a_{i_l}^{(l)}Q_{i_{t_2}\cdot}^{t_2s_{l_0+1}}\bby^{(l_0)}|^2\\
			&\leq C_dN^{-1}\sum_{t_1\neq t_2}^d\Vert\bbQ^{t_1s_{l_0}}\bbx^{(l_0)}\Vert_2^2\cdot\Vert\bbQ^{t_2s_{l_0+1}}\bby^{(l_0)}\Vert_2^2\leq C_dN^{-1}\Vert\bbQ\Vert^{2(K+1)}.
		\end{align*}
    {\bf Second derivatives:} When \(l=2\), we have
		\begin{align*}
			&\sum_{i_1\cdots i_d}^{n_1\cdots n_d}|(\bbx^{j_1})'\partial_{i_1\cdots i_d}^{(2)}\bbQ^{j_1j_2}\bby^{j_2}|^2\leq C_d\sum_{s_1\neq s_2}^d\sum_{s_3\neq s_4}^dN^{-2}\sum_{i_1\cdots i_d}^{n_1\cdots n_d}|(\bbx^{j_1})'Q_{\cdot i_{s_1}}^{j_1s_1}\mcA_{i_1\cdots i_d}^{(s_1,s_2)}Q_{i_{s_2}i_{s_3}}^{s_2s_3}\mcA_{i_1\cdots i_d}^{(s_3,s_4)}Q_{i_{s_4}\cdot}^{s_4j_2}\bby^{j_2}|^2\\
			:&=C_d\sum_{s_1\neq s_2}^d\sum_{s_3\neq s_4}^d\mcP_{s_1\cdots s_4}(j_1,j_2)
		\end{align*}
		By the proof of Lemma \ref{Lem of minor terms 1}, \(\mcP_{s_1\cdots s_4}(j_1,j_2)\) will have the maximal order of \(N\) when \(s_2=s_3\) and \(s_1=s_4\)
		\begin{align*}
			&N^{-2}\sum_{i_1\cdots i_d}^{n_1\cdots n_d}|(\bbx^{j_1})'Q_{\cdot i_{s_1}}^{j_1s_1}\mcA_{i_1\cdots i_d}^{(s_1,s_2)}Q_{i_{s_2}i_{s_3}}^{s_2s_3}\mcA_{i_1\cdots i_d}^{(s_3,s_4)}Q_{i_{s_4}\cdot}^{s_4j_2}\bby^{j_2}|^2\\
			&\leq N^{-2}\tr(|\bbQ^{s_2s_2}|^{\circ2})\cdot(|\bbx^{j_1}|^{\circ2})'|\bbQ^{j_1s_1}|^{\circ2}|\bbQ^{s_1j_2}|^{\circ2}(|\bby^{j_2}|^{\circ2})\leq N^{-1}\Vert\bbQ\Vert^{2(K+2)},
		\end{align*}
		and \(s_2\neq s_3\) and \(s_1=s_4\)
		\begin{align*}
			&N^{-2}\sum_{i_1\cdots i_d}^{n_1\cdots n_d}|(\bbx^{j_1})'Q_{\cdot i_{s_1}}^{j_1s_1}\mcA_{i_1\cdots i_d}^{(s_1,s_2)}Q_{i_{s_2}i_{s_3}}^{s_2s_3}\mcA_{i_1\cdots i_d}^{(s_3,s_4)}Q_{i_{s_4}\cdot}^{s_4j_2}\bby^{j_2}|^2\\
			&\leq N^{-2}\tr(|\bbQ^{s_2s_2}|^{\circ2})\cdot(|\bbx^{j_1}|^{\circ2})'|\bbQ^{j_1s_1}|^{\circ2}|\bba^{(s_1)}|^{\circ2}\cdot(|\bby^{j_2}|^{\circ2})'|\bbQ^{j_2s_3}|^{\circ2}|\bba^{(s_3)}|^{\circ2}\leq N^{-1}\Vert\bbQ\Vert^{2(K+2)},
		\end{align*}
		for other situations, \(\mcP_{s_1\cdots s_4}(j_1,j_2)\leq N^{-2}\Vert\bbQ\Vert^{2(k+2)}\).

    \vspace{5mm}
    \noindent
    {\bf Third derivatives:} \(l=3\): similarly, it is enough to bound
		\begin{align*}
			\mcP_{s_1\cdots s_6}(j_1,j_2):=N^{-3}\sum_{i_1\cdots i_d}^{n_1\cdots n_d}|(\bbx^{j_1})'Q_{\cdot i_{s_1}}^{j_1s_1}\mcA_{i_1\cdots i_d}^{(s_1,s_2)}Q_{i_{s_2}i_{s_3}}^{s_2s_3}\mcA_{i_1\cdots i_d}^{(s_3,s_4)}Q_{i_{s_4}i_{s_5}}^{s_4s_5}\mcA_{i_1\cdots i_d}^{(s_5,s_6)}Q_{i_{s_6}\cdot}^{s_6j_2}\bby^{j_2}|^2,
		\end{align*}
		and \(\mcP_{s_1\cdots s_6}(j_1,j_2)\) will have the maximal order of \(N\) when \(s_2=s_3=s_6,s_1=s_4=s_5\)
		\begin{align*}
			&N^{-3}\sum_{i_1\cdots i_d}^{n_1\cdots n_d}|(\bbx^{j_1})'Q_{\cdot i_{s_1}}^{j_1s_1}\mcA_{i_1\cdots i_d}^{(s_1,s_2)}Q_{i_{s_2}i_{s_2}}^{s_2s_2}\mcA_{i_1\cdots i_d}^{(s_2,s_1)}Q_{i_{s_1}i_{s_1}}^{s_1s_1}\mcA_{i_1\cdots i_d}^{(s_1,s_2)}Q_{i_{s_2}\cdot}^{s_2j_2}\bby^{j_2}|^2\\
			&\leq N^{-3}(|\bbx^{j_1}|^{\circ2})'|\bbQ^{j_1s_1}|^{\circ2}|\bbQ^{s_1s_1}|^{\circ2}\boldsymbol{1}_{n_{s_1}}\cdot(|\bby^{j_2}|^{\circ2})'|\bbQ^{j_2s_2}|^{\circ2}|\bbQ^{s_2s_2}|^{\circ2}\boldsymbol{1}_{n_{s_2}}\leq N^{-2}\Vert\bbQ\Vert^{2(K+3)},
		\end{align*}
	which completes the proof of Lemma \ref{Rem of extension minor 1}.
\end{proof}
Now, based on the above lemma, we can further prove that
\begin{lem}\label{Thm of a.s. convergence}
	Under Assumptions {\rm \ref{Ap of general noise}} and {\rm \ref{Ap of dimension}}, for any \(K\in\mathbb{N}^+\) and \(z_1,\cdots,z_K\in\mcS_{\eta_0}\) in {\rm (\ref{Eq of stability region general d})}, let \(s_i\in\{1,\cdots,d\}\) for \(1\leq i\leq K\) such that \(s_{2j}\neq s_{2j+1}\) and \(\bbx\in\mbC^{n_{s_1}},\bby\in\mbC^{n_{s_{K+1}}}\) be two deterministic vectors with bounded $L^2$ norms, then for any \(\omega\in(1/2-\delta,1/2)\), where $\delta>0$ is a sufficiently small number, we have
	\begin{align}
		\Big|\boldsymbol{x}'\prod_{i=1}^K\boldsymbol{Q}^{s_i s_{i+1}}(z_i)\boldsymbol{y}-\mathbb{E}\Big[\boldsymbol{x}'\prod_{i=1}^K\boldsymbol{Q}^{s_i s_{i+1}}(z_i)\boldsymbol{y}\Big]\Big|\prec C_K\eta_0^{-(K+4)}N^{-\omega}.\notag
	\end{align}
	In particular, if \(s_1=s_K\), we further have
	\begin{align}
		\Big|\boldsymbol{x}'{\rm diag}\Big(\prod_{i=1}^K\boldsymbol{Q}^{s_i s_{i+1}}(z_i)\Big)\boldsymbol{y}-\mathbb{E}\Big[\boldsymbol{x}'{\rm diag}\Big(\prod_{i=1}^K\boldsymbol{Q}^{s_i s_{i+1}}(z_i)\Big)\boldsymbol{y}\Big]\Big|\prec C_K\eta_0^{-(K+4)}N^{-\omega}.\notag
	\end{align}
\end{lem}
By the proofs of Lemma \ref{Thm of Entrywise almost sure convergence}, it is easy to see the condition \(d=3\) is not essential, so we omit the detailed proofs of Lemma \ref{Thm of a.s. convergence}. Next, let us  give the extension of Lemma \ref{Lem of minor terms}:
\begin{lem}\label{Rem of extension minor 2}
	For any \(z\in\mbC^+\) and \(l\in\mbN^+,1\leq l\leq4\), let \(t_1,\cdots,t_{2(l+1)}\in\{1,\cdots,d\}\) such that \(t_{2\alpha}\neq t_{2\alpha+1}\) and \(t_1\neq t_{2(l+1)}\) for \(1\leq\alpha\leq l\), let
	\begin{align}
		\sum_{i_1\cdots i_d}^{n_1\cdots n_d}\mcA_{i_1\cdots i_d}^{(t_1,t_{2l+2})}Q_{i_{t_1}i_{t_2}}^{t_1t_2}\Big(\prod_{\alpha=1}^l\mcA_{i_1\cdots i_d}^{(t_{2\alpha},t_{2\alpha+1})}Q_{i_{t_{2\alpha+1}}i_{t_{2\alpha+2}}}^{t_{2\alpha+1}t_{2\alpha+2}}\Big),\label{Eq of minor terms general}
	\end{align}
	where \(Q_{i\cdot},Q_{\cdot i}\) means the \(i\)-th row and column of \(\bbQ\), \(\mcA_{i_1\cdots i_d}^{(t_{2l},t_{2l+1})}\) is defined in {\rm (\ref{Eq of nij})}. If there is at least one terms in
	$$\left\{Q_{i_{t_{2\alpha-1}}i_{t_{2\alpha}}}^{t_{2\alpha-1}t_{2\alpha}}:\alpha=1,\cdots,l+1\right\},$$
    coming from the off-diagonal blocks, then the norms of {\rm (\ref{Eq of minor terms general})} are bounded by \(\mrO(\Vert\boldsymbol{Q}\Vert^{l+1} N)\).
\end{lem}
\begin{proof}
	Without loss of generality, we only give the detail proofs for \(l=4\), i.e.
	\begin{align}
		&\sum_{i_1\cdots i_d}^{n_1\cdots n_d}\mcA_{i_1\cdots i_d}^{(t_1,t_{2l+2})}Q_{i_{t_1}i_{t_2}}^{t_1t_2}\Big(\prod_{\alpha=1}^4\mcA_{i_1\cdots i_d}^{(t_{2\alpha},t_{2\alpha+1})}Q_{i_{t_{2\alpha+1}}i_{t_{2\alpha+2}}}^{t_{2\alpha+1}t_{2\alpha+2}}\Big),\label{Eq of d sum}
	\end{align}
	Denote
	\begin{center}
		\(n_{\bba}^{(r)}\) and \(n_{r_1,r_2}\) to be the number of \(a_{i_r}^{(r)}\) and \(Q_{i_{r_1}i_{r_2}}^{r_1r_2}\) in (\ref{Eq of d sum}) respectively.
	\end{center}
	By the definition of \(\mcA_{i_1\cdots i_d}^{(t_{2\alpha-1},t_{2\alpha})}\), there are at most two \(r_1\neq r_2\) such that \(n_{\bba}^{(r_1)}=n_{\bba}^{(r_2)}=0\). Consider the following three cases:

    \vspace{5mm}
    \noindent
    {\bf Case 1:} If there exists \(r_1,r_2\in\{1,\cdots,d\}\) such that \(r_1<r_2\) and \(n_{\bba}^{(r_1)}=n_{\bba}^{(r_2)}=0\), then all \(\mcA_{i_1\cdots i_d}^{(t_{2\alpha-1},t_{2\alpha})}\) should be equal, i.e. \(t_{2\alpha-1}=r_1,t_{2\alpha}=r_2\) or \(t_{2\alpha-1}=r_2,t_{2\alpha}=r_1\). Hence, for all \(Q_{i_{t_{2\alpha}}i_{t_{2\alpha+1}}}^{t_{2\alpha}t_{2\alpha+1}}\), it must equal to \(Q_{i_{r_1}i_{r_1}}^{r_1r_1},Q_{i_{r_2}i_{r_2}}^{r_2r_2}\) or \(Q_{\tilde{r}_1\tilde{r}_2}^{r_1r_2}\). Since we have at least one off-diagonal term, then we have at least one \(Q_{i_{r_1}i_{r_2}}^{r_1r_2}\), i.e. \(n_{r_1,r_2}\geq1\). Hence if \(n_{r_1,r_1},n_{r_2,r_2}\geq1\),
		\begin{align*}
			&|(\ref{Eq of d sum})|\leq\sum_{i_1\cdots i_d}^{n_1\cdots n_d}|\mcA_{i_1\cdots i_d}^{(r_1,r_2)}|^l|Q_{i_{r_1}i_{r_1}}^{r_1r_1}|^{n_{r_1,r_1}}|Q_{i_{r_2}i_{r_2}}^{r_2r_2}|^{n_{r_2,r_2}}|Q_{i_{r_1}i_{r_2}}^{r_1r_2}|^{n_{r_1,r_2}}\\
			&\leq\boldsymbol{1}'\diag(|\bbQ^{r_1r_1}|^{\circ n_{r_1,r_1}})|\bbQ^{r_1r_2}|^{\circ n_{r_1,r_2}}\diag(|\bbQ^{r_2r_2}|^{\circ n_{r_2,r_2}})\boldsymbol{1}\leq N\Vert\bbQ\Vert^{l+1}.
		\end{align*}
		Otherwise, if at most one of \(n_{r_1,r_1},n_{r_2,r_2}\) is nonzero, then we have \(n_{r_1,r_2}\geq2\)
		\begin{align*}
			|(\ref{Eq of d sum})|\leq\tr\big(|\bbQ^{r_2r_1}|\diag(|\bbQ^{r_1r_1}|^{\circ n_{r_1,r_1}})|\bbQ^{r_1r_2}|^{\circ(n_{r_1,r_2}-1)}\big)\leq N\Vert\bbQ\Vert^{l+1}.
		\end{align*}
    {\bf Case 2:} If there is an \(r_1\in\{1,\cdots,d\}\) such that \(n_{\bba}^{(r_1)}=0\), then for any \(\mcA_{i_1\cdots i_d}^{(t_{2\alpha},t_{2\alpha+1})}\), we have \(t_{2\alpha}=r_1\) or \(t_{2\alpha+1}=r_1\). Without loss of generality, let \(r_1=1\). If there is no diagonal terms, then we have all \(n_{\bba}^{(r)}\geq1\) for \(r>1\) and 
		\begin{align}
			&|(\ref{Eq of d sum})|\leq\sum_{i_1\cdots i_d}^{n_1\cdots n_d}\prod_{j=1}^{l+1}|\mcA^{(1,s_j)}Q_{i_1\tilde{s}_j}^{1s_j}|\leq\sum_{i_1}^{n_1}\prod_{j=1}^{l+1}|Q_{i_1\cdot}^{1s_j}||\bba^{(s_j)}|\leq N\Vert\bbQ\Vert^{l+1},\label{Eq of d sum 1}
		\end{align}
		where \(s_j\neq 1\) and \(s_j\in\{1,\cdots,d\}\). Hence, it is enough to consider there exists diagonal terms, and we claim that (\ref{Eq of d sum}) must contain \(Q_{i_1i_1}^{11}\). Otherwise, suppose (\ref{Eq of d sum}) contains a diagonal term \(Q_{i_{t_1}i_{t_1}}^{t_1t_1}\) such that \(t_1=t_2\neq1\), then \(t_3,t_{10}=1\) and \(t_4\neq1\), since \(\mcA_{i_1\cdots i_d}^{(4,5)}\) does not contain \(a_{i_1}^{(1)}\), then \(t_5=1,t_6\neq1\), otherwise \(Q_{i_{t_5}i_{t_6}}^{t_5t_6}=Q_{i_1i_1}^{11}\). Similarly, since \(\mcA_{i_1\cdots i_d}^{(6,7)}\) does not contain \(a_{i_1}^{(1)}\), then \(t_7=1,t_8\neq1\) and \(t_9=1,t_{10}\neq1\), which is a contradiction. Moreover, since we have at least one off-diagonal term, there are at most three types of diagonal terms. First, if there exists three different kinds of diagonal terms, without loss of generality, let them be \(Q_{i_1i_1}^{11},Q_{i_2i_2}^{22}\) and \(Q_{i_3i_3}^{33}\) the only possible case is as follows: 
		\begin{align*}
			&(\ref{Eq of d sum})=\sum_{i_1\cdots i_d}^{n_1\cdots n_d}\mcA_{i_1\cdots i_d}^{(1,s_1)}Q_{i_1i_1}^{11}\mcA_{i_1\cdots i_d}^{(1,2)}Q_{i_2i_2}^{22}\mcA_{i_1\cdots i_d}^{(2,1)}Q_{i_1i_1}^{11}\mcA_{i_1\cdots i_d}^{(1,3)}Q_{i_3i_3}^{33}\mcA_{i_1\cdots i_d}^{(3,1)}Q_{i_1i_{s_1}}^{1s_1}
		\end{align*}
		where \(s_1\neq1\). Since both \(n_{\bba}^{(2)},n_{\bba}^{(3)}\geq2\). If \(s_1\neq2\) or \(s_1\neq3\), we have
		\begin{align*}
			&|(\ref{Eq of d sum})|\leq\boldsymbol{1}'\diag(|\bbQ^{11}|^{\circ2})|Q^{1s_1}||\bba^{(s_1)}|\cdot|\bba^{(2)}|'\diag(|\bbQ^{22}|)|\bba^{(2)}|\\
			&\cdot|\bba^{(3)}|'\diag(|\bbQ^{33}|)|\bba^{(3)}|\leq N\Vert\bbQ\Vert^{l+1}.
		\end{align*}
		Otherwise, suppose \(s_1=2\), we have
		\begin{align*}
			&|(\ref{Eq of d sum})|\leq\boldsymbol{1}'\diag(|\bbQ^{11}|^{\circ2})|Q^{12}|\diag(|\bbQ^{22}|)|\bba^{(2)}|\cdot|\bba^{(3)}|'\diag(|\bbQ^{33}|)|\bba^{(3)}|\leq N\Vert\bbQ\Vert^{l+1}.
		\end{align*}
		Next, suppose there are only two kinds of diagonal terms, then
		\begin{align*}
			&|(\ref{Eq of d sum})|=\sum_{i_1\cdots i_d}^{n_1\cdots n_d}|\mcA_{i_1\cdots i_d}^{(1,s_3)}Q_{i_1i_1}^{11}\mcA_{i_1\cdots i_d}^{(1,s_1)}Q_{i_{s_1}i_{s_1}}^{s_1s_1}\mcA_{i_1\cdots i_d}^{(s_1,1)}Q_{i_1i_1}^{11}\mcA_{i_1\cdots i_d}^{(1,s_1)}Q_{i_{s_1}i_{s_1}}^{s_1s_1}\mcA_{i_1\cdots i_d}^{(s_1,1)}Q_{i_1i_{s_3}}^{1s_3}|\\
			&\leq\boldsymbol{1}'\diag(|\bbQ^{11}|^{\circ2})|Q^{1s_3}||\bba^{(s_3)}|\cdot\boldsymbol{1}'\diag(|\bbQ^{s_1s_1}|^{\circ2})||\bba^{(s_1)}|\leq N\Vert\bbQ\Vert^{l+1}.
		\end{align*}
		Finally, if there is only one type of diagonal term which is \(Q_{i_1i_1}^{11}\), then
		\begin{align*}
			&|(\ref{Eq of d sum})|\leq\sum_{i_1\cdots i_d}^{n_1\cdots n_d}|\mcA_{i_1\cdots i_d}^{(s_5,1)}Q_{i_1i_1}^{11}\mcA_{i_1\cdots i_d}^{(1,s_1)}Q_{i_{s_1}i_{s_2}}^{s_1s_2}\mcA_{i_1\cdots i_d}^{(s_2,1)}Q_{i_1i_1}^{11}\mcA_{i_1\cdots i_d}^{(1,s_3)}Q_{i_{s_3}i_{s_3}}^{s_3s_4}\mcA_{i_1\cdots i_d}^{(s_4,1)}Q_{i_1i_{s_3}}^{1s_5}|\\
			&\leq\boldsymbol{1}'\diag(|\bbQ|^{11})|\bbQ^{1s_5}||\bba^{(s_5)}|\cdot|\bba^{(s_1)}|'|\bbQ^{s_1s_2}||\bba^{(s_2)}|\cdot|\bba^{(s_3)}|'|\bbQ^{s_3s_4}||\bba^{(s_4)}|\leq N\Vert\bbQ\Vert^{l+1},
		\end{align*}
		where \(s_1\neq s_2,s_3\neq s_4\) and \(s_1,\cdots,s_5\neq1\).

    \vspace{5mm}
    \noindent
    {\bf Case 3:} Suppose all \(n_{(r)}\geq1\). In this case, for any off-diagonal terms, notice that all \(\mcA_{i_1\cdots i_d}^{(t_{2\alpha},t_{2\alpha+1})}\) and \(\mcA_{i_1\cdots i_d}^{(t_1,t_{2l+2})}\) can not be equal, otherwise it is the first situation. Hence, there exists at most two \(r_1<r_2\) such that \(n_{\bba}^{(r_1)}=n_{\bba}^{(r_2)}=1\). Without loss of generality, let \(r_1=1\) and \(r_2=2\), then there are two situations. First, we have four \(\mcA_{i_1\cdots i_d}^{(1,2)}\) and an \(\mcA_{i_1\cdots i_d}^{(s_1,s_2)}\), where \(s_1,s_2\neq1,2\). Then we will have two off-diagonal terms \(Q_{i_{s_1}i_1}^{s_11},Q_{i_{s_2}i_2}^{s_22}\) and
		\begin{align}
			&|(\ref{Eq of d sum})|\leq\sum_{i_1\cdots i_d}^{n_1\cdots n_d}|(\mcA_{i_1\cdots i_d}^{(1,2)})^4(Q_{i_1i_1}^{11})^{n_{1,1}}(Q_{i_2i_2}^{22})^{n_{2,2}}(Q_{i_1i_2}^{12})^{n_{1,2}}\mcA_{i_1\cdots i_d}^{(1,2)}Q_{i_{s_1}i_1}^{s_11}Q_{i_{s_2}i_2}^{s_22}|\notag\\
			&\leq|\bba^{(s_1)}|'|\bbQ^{s_11}|\diag(|\bbQ^{11}|^{\circ n_{1,1}})|\bbQ^{12}|^{\circ n_{1,2}}\diag(|\bbQ^{22}|^{\circ n_{2,2}})|\bbQ^{2s_2}||\bba^{(s_2)}|\leq\Vert\bbQ\Vert^{l+1}.\label{Eq of d sum 2}
		\end{align}
		Otherwise, we have three \(\mcA_{i_1\cdots i_d}^{(1,2)}\), one \(\mcA_{i_1\cdots i_d}^{1,s_1}\) and \(\mcA_{i_1\cdots i_d}^{2,s_2}\), so we will have one \(Q_{i_{s_1}i_{s_2}}^{s_1s_2}\) or two off-diagonal terms like \(Q_{i_{s_1}i_1}^{s_11},Q_{i_{s_2}i_2}^{s_22}\) or \(Q_{i_{s_1}i_1}^{s_11},Q_{i_{s_2}i_1}^{s_21}\) or \(Q_{i_{s_1}i_2}^{s_12},Q_{i_{s_2}i_2}^{s_22}\). For the case of one \(Q_{i_{s_1}i_{s_2}}^{s_1s_2}\), since it will associate with \(\bba^{(s_1)},\bba^{(s_2)}\) as \(|\bba^{(s_1)}|'|\bbQ^{s_1s_2}||\bba^{(s_2)}|\leq\Vert\bbQ\Vert\), then we can use the same trick in (\ref{Eq of d sum 1}) to conclude our claim. For the case of two \(Q_{i_{s_1}i_1}^{s_11},Q_{i_{s_2}i_2}^{s_22}\), it is the same as (\ref{Eq of d sum 2}). So we only give the case of \(Q_{i_{s_1}i_1}^{s_11},Q_{i_{s_2}i_1}^{s_21}\) as follows (the other one is analogous):
		\begin{align*}
			&|(\ref{Eq of d sum})|\leq\sum_{i_1}^{n_1}|a_{i_1}^{(1)}||Q_{i_1i_1}^{11}|^{n_{1,1}}|Q_{i_1\cdot}^{12}|^{\circ n_{1,2}}\diag(|\bbQ^{22}|^{n_{2,2}})|\bba^{(2)}|\cdot|Q_{i_1\cdot}^{1s_1}||\bba^{(s_1)}|\cdot|Q_{i_1\cdot}^{1s_2}||\bba^{(s_2)}|\leq N\Vert\bbQ\Vert^{l+1}.
		\end{align*}
		Next, assume \(n_{\bba}^{(1)}=1\) and \(n_{\bba}^{(r)}\geq2\) for all \(r\neq1\). In this case, we have at most three \(Q_{i_{s_1}i_{s_2}}^{s_1s_2},Q_{i_{s_3}i_{s_4}}^{s_3s_4},Q_{i_{s_5}i_{s_6}}^{s_5s_6}\) such that \(s_1,\cdots,s_6\neq1\), i.e. 
		\begin{align*}
			\mcA_{i_1\cdots i_d}^{(1,s_1)}Q_{i_{s_1}i_{s_2}}^{s_1s_2}\mcA_{i_1\cdots i_d}^{(s_2,s_3)}Q_{i_{s_3}i_{s_4}}^{s_3s_4}\mcA_{i_1\cdots i_d}^{(s_4,1)}Q_{i_1i_1}^{11}\mcA_{i_1\cdots i_d}^{(1,s_5)}Q_{i_{s_5}i_{s_6}}^{s_5s_6}\mcA_{i_1\cdots i_d}^{(s_6,1)}Q_{i_1i_1}^{11}.
		\end{align*}
		Since \(Q_{i_1i_1}^{11}\) will associate with \(\bba^{(1)}\) as \(\boldsymbol{1}\diag(|\bbQ^{11}|^{\circ2})|\bba^{(1)}|\leq N^{1/2}\Vert\bbQ\Vert^2\), and \(Q_{i_{s_1}i_{s_2}}^{s_1s_2}\) will associate with \(\bba^{(s_1)}\) or other terms like \(Q_{i_{s_3}i_{s_4}}^{s_3s_4}\) (e.g.) \(|\bba^{(s_1)}|'|\bbQ^{(s_1,s_2)}||\bba^{(s_2)}|\) or \(|\bba^{(s_1)}|'|\bbQ^{(s_1,s_2)}||\bbQ^{(s_2,s_3)}||\bba^{(s_2)}|\) or \(|\bba^{(s_1)}|'|\bbQ^{(s_1,s_2)}|\diag(|\bbQ|^{s_2s_2})|\bba^{(s_2)}|\) due to all \(n_{\bba}^{(r)}\geq2\) for \(r\neq1\), then our conclusion is still valid. Next, if we have two \(Q_{i_{s_1}i_{s_2}}^{s_1s_2},Q_{i_{s_3}i_{s_4}}^{s_3s_4}\) such that \(s_1,\cdots,s_4\neq1\). In this case, we must have the off-diagonal terms like \(Q_{i_1i_r}^{1,r}\), i.e. (e.g.)
		\begin{align*}
			\mcA_{i_1\cdots i_d}^{(s_6,1)}Q_{i_1i_{s_5}}^{1s_5}\mcA_{i_1\cdots i_d}^{(s_5,s_1)}Q_{i_{s_1}i_{s_2}}^{s_1s_2}\mcA_{i_1\cdots i_d}^{(s_4,1)}Q_{i_1i_1}^{11}\mcA_{i_1\cdots i_d}^{(1,s_3)}Q_{i_{s_3}i_{s_4}}^{s_3s_4}\mcA_{i_1\cdots i_d}^{(s_4,1)}Q_{i_1i_{s_6}}^{1s_6}.
		\end{align*}
		Therefore, for \(Q_{i_{s_1}i_{s_2}}^{s_1s_2}\), it will associate with \(\bba^{(s_1)},\bba^{(s_2)}\) as \(|\bba^{(s_1)}|'|\bbQ^{s_1s_2}||\bba^{(s_2)}|\leq\Vert\bbQ\Vert\) if \(s_1,s_2\neq s_5,s_6\) or one \(Q_{i_1i_{s_5}}^{1s_5}\) as \(Q_{i_1\cdot}^{1s_1}\bbQ^{s_1s_2}\bba^{(s_2)}\). Whatever which situations, it must have (e.g. \(s_1,\cdots,s_4\neq s_5,s_6\))
		\begin{align*}
			&|(\ref{Eq of d sum})|\leq|\bba^{(s_5)}|'|\bbQ^{s_51}|\diag(|\bbQ|^{11})|\bbQ^{1s_6}||\bba^{(s_6)}|\cdot|\bba^{(s_1)}|'|\bbQ^{s_1s_2}||\bba^{(s_2)}|\cdot|\bba^{(s_3)}|'|\bbQ^{s_3s_4}||\bba^{(s_4)}|\leq\Vert\bbQ\Vert^{l+1}.
		\end{align*}
		Now, if we have one \(Q_{i_{s_1}i_{s_2}}^{s_1s_2}\). Similar to the above result, \(Q_{i_{s_1}i_{s_2}}^{s_1s_2}\) can associate with \(Q_{i_1s_1}^{1s_1}\) as \(Q_{i_1\cdot}^{1s_1}\bbQ^{s_1s_2}\bba^{(s_2)}\), or associate with \(\bba^{(s_1)},\bba^{(s_2)}\) as \(|\bba^{(s_1)}|'|\bbQ^{s_1s_2}||\bba^{(s_2)}|\leq\Vert\bbQ\Vert\). So, we have
		\begin{align*}
			&|(\ref{Eq of d sum})|\leq|\bba^{(s_1)}|'|\bbQ^{s_1s_2}||\bba^{(s_2)}|\sum_{i_1}^{n_1}|a_{i_1}^{(1)}||Q_{i_1i_1}^{11}|^{n_{1,1}}\cdot\prod_{j=1}^{l_0}|Q_{i_1\cdot}^{1s_j}||\bba^{(j)}|\leq N\Vert\bbQ\Vert^{l+1},
		\end{align*}
		where we use the fact that all \(|Q_{i_1\cdot}^{1s_j}||\bba^{(j)}|\leq\Vert\bbQ\Vert\) and \(\boldsymbol{1}'\diag(|\bbQ^{11}|^{\circ n_{1,1}})|\bba^{(1)}|\leq N\Vert\bbQ\Vert^{n_{1,1}}\). Finally, suppose all \(n_{\bba}^{(r)}\geq2\). Since \(l=4\), given an \(r\in\{1,\cdots,d\}\), there will be at most three \(\mcA_{i_1\cdots i_d}^{(t_{2\alpha-1},t_{2\alpha})}\) contain \(r\). As a result, for any diagonal terms \(Q_{i_ri_r}^{rr}\) in (\ref{Eq of d sum}), we have two situations. First, all \((t_{2\alpha},t_{2\alpha+1})\) from \(Q_{i_{t_{2\alpha}}i_{t_{2\alpha+1}}}^{t_{2\alpha}t_{2\alpha+1}}\) does not contain \(r\), then there is only one \(Q_{i_ri_r}^{rr}\) in (\ref{Eq of d sum}), which will associate with \(\bba^{(r)}\) as follows:
		\begin{align*}
			&|\bba^{(r)}|'\diag(|\bbQ^{rr}|)|\bba^{(r)}|^{\circ(n_{\bba}^{(r)}-1)}\leq\Vert\bbQ\Vert^{n_{r,r}}.
		\end{align*}
		Otherwise, there exists only one \(Q_{i_ri_{s_1}}^{rs_1}\) such that \(s_1\neq r\), then \(Q_{i_ri_r}^{rr}\) will associate with \(\bba^{(r)}\) and \(Q_{i_ri_{s_1}}^{rs_1}\) as follows (e.g.):
		\begin{align*}
			&(|\bba^{(r)}|^{\circ n_{(r)}})'\diag(|\bbQ^{rr}|)|\bbQ^{rs_1}||\bba^{(s_1)}|\leq\Vert\bbQ\Vert^2.
		\end{align*}
		Lastly, for the off-diagonal terms, since \(l=4\), there are at most three \(Q_{i_ri_{s_1}}^{rs_1},Q_{i_ri_{s_2}}^{rs_2},Q_{i_ri_{s_3}}^{rs_3}\) such that \(s_1,s_2,s_3\neq r\). For each of them, it will associate with \(\bba^{(s_1)}\) as \(|Q_{i_r\cdot}^{rs_1}||\bba^{(s_1)}|\) or other \(Q_{i_{s_1}i_{s_4}}^{s_1s_4}\) as \(|Q_{i_r\cdot}^{rs_1}||\bbQ^{s_1s_4}||\bba^{(s_4)}|\). No matter which cases, we have (e.g.)
		\begin{align*}
			&|(\ref{Eq of d sum})|\leq\sum_{i_r}^{n_r}|a_{i_r}^{(r)}|^2|Q_{i_r\cdot}^{rs_1}||\bbQ^{s_1s_4}||\bba^{(s_4)}|\cdot|Q_{i_r\cdot}^{rs_2}||\bbQ^{s_2s_5}||\bba^{(s_5)}|\cdot|Q_{i_r\cdot}^{rs_3}||\bba^{(s_3)}|\leq\Vert\bbQ\Vert^{l+1}.
		\end{align*}
	   This completes the proof.
\end{proof}
Finally, we present the extension of Lemma \ref{Cor of minor terms} as follows:
\begin{lem}\label{Rem of extension of Corollary}
	For any \(z\in\mathbb{C}_{\eta}^+\) and \(1\leq l_1,l_2\leq4\), let \(t_i,s_j\in\{1,\cdots,d\}\) for \(1\leq i\leq2l_1+1,1\leq j\leq2l_2+1\) such that \(t_{2\alpha}\neq t_{2\alpha+1}\) and \(s_{2\gamma}\neq s_{2\gamma+1}\) for \(1\leq\alpha\leq l_1,1\leq\gamma\leq l_2\) and \(t_1\neq t_{2l_1+2},s_1\neq s_{2l_2+2}\) the define
	$$\left\{\begin{array}{l}
		P_1(z):=\mcA_{ijk}^{(t_1,t_{2l_1+2})}Q_{i_{t_1}i_{t_2}}^{t_1t_2}(z)\left(\prod_{\alpha=1}^{l_1}\mcA_{i_1\cdots i_d}^{(t_{2\alpha},t_{2\alpha+1})}Q_{i_{t_{2\alpha+1}}i_{t_{2\alpha+2}}}^{t_{2\alpha+1}t_{2\alpha+2}}(z)\right),\\
		P_2(z):=\mcA_{ijk}^{(s_1,s_{2l_2+2})}Q_{i_{s_1}i_{s_2}}^{s_1s_2}(z)\left(\prod_{\gamma=1}^{l_2}\mcA_{i_1\cdots i_d}^{(s_{2\gamma},s_{2\gamma+1})}Q_{i_{s_{2\gamma+1}}i_{s_{2\gamma+2}}}^{s_{2\gamma+1}s_{2\gamma+2}}(z)\right).
	\end{array}\right.$$
	If there are at least one term in
	$$\left\{Q_{i_{t_{2\alpha+1}}i_{t_{2\alpha+2}}}^{t_{2\alpha+1}t_{2\alpha+2}}(z):\alpha=1,\cdots,l_1+1\right\}\quad{\rm or}\quad\left\{Q_{i_{s_{2\gamma+1}}i_{s_{2\gamma+2}}}^{s_{2\gamma+1}s_{2\gamma+2}}(z):\gamma=1,\cdots,l_2+1\right\},$$
    coming from the off-diagonal blocks, the norm of \(\sum_{i,j,k=1}^{m,n,p}P_1(z)P_2(z)\) is bounded by \(\mrO(\Vert\bbQ\Vert^{l_1+l_2+2}N)\).
\end{lem}
\begin{proof}
	Denote
	\begin{center}
		\(n_{r_1,r_2}^{(i)}\) to be the number of \(Q_{i_{r_1}i_{r_3}}^{r_1r_2}\) in \(P_i(z)\) for \(i=1,2\).
	\end{center}
	Since we can apply Cauchy's inequality to show \(\sum_{i_1\cdots i_d}^{n_1\cdots i_d}P_1(z)P_2(z)\) is bounded by \(C\Vert\bbQ\Vert^{l_1+l_2+2}N\), without loss of generality, assume only \(P_2(z)\) contains off-diagonal blocks. Therefore, for each \(Q_{i_{t_{\alpha}}i_{t_{\alpha}}}^{t_{\alpha}t_{\alpha}}(z)\) in \(P_1(z)\), denoting \(n_{t_{\alpha},t_{\alpha}}^{(1)}\) and \(n_{t_{\alpha},t_{\alpha}}^{(2)}\) to be the number of \(Q_{i_{t_{\alpha}}i_{t_{\alpha}}}^{t_{\alpha}t_{\alpha}}(z)\) in \(P_1(z)\) and \(P_2(z)\), respectively, then consider the following two cases:

    \vspace{5mm}
    \noindent
    {\bf Case 1:} If there exists \(Q_{i_{s_{2\gamma+1}}i_{s_{2\gamma+2}}}^{s_{2\gamma+1}s_{2\gamma+2}}(z)\) in \(P_2(z)\) equal to \(Q_{i_{t_{\alpha}}i_{t_{\alpha}}}^{t_{\alpha}t_{\alpha}}(z)\), then all arguments in Remark \ref{Rem of extension minor 2} will be almost unchanged, just replacing the original power \(n_{t_{\alpha},t_{\alpha}}^{(2)}\) by \(n_{t_{\alpha},t_{\alpha}}^{(1)}+n_{t_{\alpha},t_{\alpha}}^{(2)}\).

    \vspace{5mm}
    \noindent
    {\bf Case 2:} Otherwise, if \(Q_{i_{t_{\alpha}}i_{t_{\alpha}}}^{t_{\alpha}t_{\alpha}}(z)\) does not exist in \(P_2(z)\), then all \(\mcA_{i_1\cdots i_d}^{(s_{2\gamma},s_{2\gamma+1})}\) must contain \(a_{i_{t_{\alpha}}}^{(t_{\alpha})}\). Therefore, when \(l_2\geq2\), the number of \(a_{i_{t_{\alpha}}}^{(t_{\alpha})}\) in \(P_2(z)\) must be no smaller than \(2\), then all \(Q_{i_{t_{\alpha}}i_{t_{\alpha}}}^{t_{\alpha}t_{\alpha}}(z)\) will associate with \(a_{i_{t_{\alpha}}}^{(t_{\alpha})}\) by \(|\bba^{(t_{\alpha})}|'\diag(|\bbQ^{t_{\alpha}t_{\alpha}}|^{\circ n_{t_{\alpha},t_{\alpha}}^{(1)}})|\bba^{(t_{\alpha})}|\leq\Vert\bbQ\Vert^{n_{t_{\alpha},t_{\alpha}}^{(1)}}\). Finally, when \(l_2=1\), i.e. \(P_2(z)=\mcA_{i_1\cdots i_d}^{(s_1,s_2)}Q_{i_{s_1}i_{s_2}}^{s_1s_2}(z)\) and \(s_1\neq s_2\), since \(P_1(z)\) only contains diagonal terms, then \(l_1\geq2\) and we have at least two kinds of diagonal terms, consider the following two possible cases:
		\begin{itemize}
			\item \(P_1(z)\) only contains \(Q_{i_{t_1}i_{t_1}}^{t_1t_1}(z)\) and \(Q_{i_{t_2}i_{t_2}}^{t_2t_2}(z)\), i.e.
			$$P_1(z)=\mcA_{i_1\cdots i_d}^{(t_2,t_1)}Q_{i_{t_1}i_{t_1}}^{t_1t_1}(z)\mcA_{i_1\cdots i_d}^{(t_1,t_2)}Q_{i_{t_2}i_{t_2}}^{t_2t_2}(z)\mcA_{i_1\cdots i_d}^{(t_2,t_1)}\cdots.$$
			Since \(s_1,s_2\neq t_1\) and \(s_1,s_2\neq t_2\), then there exists \(a_{i_{t_1}}^{(t_1)},a_{i_{t_2}}^{(t_2)}\) and \(a_{i_{s_1}}^{(s_1)},a_{i_{s_2}}^{(s_2)}\), we can conclude that
			\begin{align*}
				&\sum_{i_1\cdots i_d}^{n_1\cdots n_d}|P_1(z)P_2(z)|\leq|\bba^{(s_1)}|'|\bbQ^{s_1s_2}||\bba^{(s_2)}|\cdot\boldsymbol{1}'\diag(|\bbQ^{t_1t_1}|^{\circ n_{t_1,t_1}^{(1)}})|\bba^{(t_1)}|\cdot\\
				&\boldsymbol{1}'\diag(|\bbQ^{t_2t_2}|^{\circ n_{t_2,t_2}^{(1)}})|\bba^{(t_2)}|\leq\Vert\bbQ\Vert^{l_1+l_2+2}N.
			\end{align*}
			\item \(P_1(z)\) contains at least three different \(Q_{i_{t_1}i_{t_1}}^{t_1t_1}(z),Q_{i_{t_2}i_{t_2}}^{t_2t_2}(z),Q_{i_{t_3}i_{t_3}}^{t_3t_3}(z),\cdots\) coming from diagonal blocks, where \(t_1\neq t_2\neq t_3\cdots\) and all \(t_j\neq s_1,t_j\neq s_2\), then for each \(t_j\), there will be at most one \(a_{i_{t_j}}^{(t_j)}\) not existing in \(P_1(z)\), otherwise we back to the previous situation. Without loss generality, assume \(P_1(z)\) does not contain \(a_{i_{t_1}}^{(t_1)}\), then for each \(a_{i_{t_j}}^{(t_j)}\) where \(j\geq2\), it will appear in \(\mcA_{i_1\cdots i_d}^{(s_1,s_2)}\) and \(P_1(z)\) at least once, so \(Q_{i_{t_j}i_{t_j}}^{t_jt_j}(z)\) will associate with at least two \(a_{i_{t_j}}^{(t_j)}\) by \(|\bba^{(t_j)}|'\diag(|\bbQ^{t_jt_j}|^{\circ n_{t_j,t_j}^{(1)}})|\bba^{(t_j)}|\leq\Vert\bbQ\Vert^{n_{t_j,t_j}^{(1)}}\). For \(Q_{i_{t_1}i_{t_1}}^{t_1t_1}(z)\), it will associate with one \(a_{i_{t_1}}^{(t_1)}\) by \(\boldsymbol{1}'\diag(|\bbQ^{t_1t_1}|^{\circ n_{t_1,t_1}^{(1)}})|\bba^{(t_1)}|\leq N^{1/2}\Vert\bbQ\Vert^{n_{t_1,t_1}^{(1)}}\), which can conclude our claim. Finally, if \(P_1(z)\) contains \(a_{i_{t_1}}^{(t_1)}\), then all \(Q_{i_{t_j}i_{t_j}}^{t_jt_j}(z)\) will associate with at least two \(a_{i_{t_j}}^{(t_j)}\) by \(|\bba^{(t_j)}|'\diag(|\bbQ^{t_jt_j}|^{\circ n_{t_j,t_j}^{(1)}})|\bba^{(t_j)}|\leq\Vert\bbQ\Vert^{n_{t_j,t_j}^{(1)}}\).
		\end{itemize}
	Now, we complete the extension of Lemma \ref{Cor of minor terms} for general \(d\geq3\).
\end{proof}
\subsection{Entrywise law}
\begin{thm}\label{Thm of entrywise law general d}
	Under Assumptions {\rm \ref{Ap of general noise}} and {\rm \ref{Ap of dimension}}, for any \(\eta_0>0\), \(z\in\mcS_{\eta_0}\) in \eqref{Eq of stability region general d} and $\omega\in(1/2-\delta,1/2)$, where $\delta>0$ is a sufficiently small number, let 
    \begin{align}
		\bbW^{(d)}(z)=-((z+g(z))\bbI_d-\diag(\bbg(z))+g(z)\bbS_d-\diag(\bbg(z))\bbS_d-\bbS_d\diag(\bbg(z)))^{-1}.\label{Eq of W general d}
	\end{align}
    For \(s,t\in\{1,\cdots,d\}\), we have 
	\begin{align*}
		\Big|Q_{i_si_t}^{st}(z)-\mathfrak{c}_s^{-1}g_s(z)\Big[\delta_{st}\delta_{i_si_t}+(a_{i_s}^{(s)})^2\sum_{k\neq s}^d(g(z)-g_s(z)-g_k(z))W_{sk}^{(d)}(z)\Big]\Big|\prec\mrO(\eta_0^{-21}N^{-\omega}),
	\end{align*}
	where \(Q_{i_si_t}^{st}(z)\) is the \((i_s,i_t)\)-th entry of \(\bbQ^{st}\) and \(a_{i_s}^{(s)}\) is the \(i\)-th entry of \(\bba^{(s)}\), as does \(W_{sk}^{(d)}(z)\), and \(\bbg(z)=(g_1(z),\cdots,g_d(z))'\) is the solution of {\rm (\ref{Eq of MDE 3 order})}.
\end{thm}
\begin{proof}
    The existence of \(\bbW^{(d)}(z)\) on \(\mcS_{\eta_0}\) is established in Lemma \ref{Lem of invertible bbGa} later. Similar to \eqref{Eq of W}, we define 
    $$W_{st,N}^{(d)}(z)=\mbE\big[(\bba^{(s)})'\bbQ^{st}(z)\bba^{(t)}\big],\quad\text{for }1\leq s,t\leq d.$$
    Recall that \(\Vert\bbQ(z)\Vert\leq\eta_0^{-1}\) for any \(z\in\mcS_{\eta_0}\) in (\ref{Eq of stability region general d}). For convenience, we suppress the argument $(z)$ in what follows. Similar to what we have done in \S\ref{sec of proof entrywise law d=3}, we first prove that for any \(\omega\in(1/2-\delta,1/2)\) 
    \begin{align}
        \sup_{z\in\mcS_{\eta_0}}\Vert\bbg(z)-\bbm(z)\Vert_{\infty}=\mrO(\eta_0^{-15}N^{-2\omega}).\label{Eq of approximation general d}
    \end{align}
    By the identity \(\bbM\bbQ-z\bbQ=\bbI_N\) and the cumulant expansion (\ref{Eq of cumulant expansion}), we have
	\begin{align*}
		z\mbE[Q_{i_si_t}^{st}]&=\frac{1}{\sqrt{N}}\sum_{l\neq s}^d\sum_{i_1\cdots i_d}^{(s,t)}\mbE[X_{i_1\cdots i_d}\mcA_{i_1\cdots i_d}^{(s,l)}Q_{i_li_t}^{lt}]-\delta_{st}\delta_{i_si_t}\\
		&=\frac{1}{\sqrt{N}}\sum_{i_1\cdots i_d}^{(s,t)}\Big(\sum_{l\neq s}^d\mbE[\partial_{i_1\cdots i_d}^{(1)}\{\mcA_{i_1\cdots i_d}^{(s,l)}Q_{i_li_t}^{lt}\}]+\epsilon_{i_1\cdots i_d}^{(2)}\Big)-\delta_{st}\delta_{i_si_t},
	\end{align*}
	where the notation \(\sum_{i_1\cdots i_d}^{(s,t)}\) means the summation is over all \(i_r=1,\cdots,n_r\) {\bf except} \(i_s=1,\cdots,n_s\) and \(i_t=1,\cdots,n_t\). Similar to proofs of (\ref{Eq of ve(2) d=3}) in Theorem \ref{Thm of entrywise law d=3}, we have \(N^{-1/2}|\sum_{i_1\cdots i_d}^{(s,t)}\epsilon_{i_1\cdots i_d}^{(2)}|=\mrO(a_{i_s}^{(s)}\eta_0^{-3}N^{-1/2}+\eta_0^{-3}N^{-1})\). Here we omit the details for convenience, and readers can refer to (\ref{Eq of remainder covariance general d}) in Theorem \ref{Thm of covariance general d} for an example of calculating remainders. Next, by direct calculation, we have
	\begin{align*}
		&\frac{1}{\sqrt{N}}\sum_{l\neq s}^d\sum_{i_1\cdots i_d}^{(s,t)}\mbE\big[\partial_{i_1\cdots i_d}^{(1)}\{\mcA_{i_1\cdots i_d}^{(s,l)}Q_{i_li_t}^{lt}\}\big]=-\frac{1}{N}\sum_{l\neq s}^d\sum_{i_1\cdots i_d}^{(s,t)}\big[\mcA_{i_1\cdots i_d}^{(s,l)}Q_{i_li_{r_1}}^{lr_1}\mcA_{i_1\cdots i_d}^{(r_1,r_2)}Q_{i_{r_2}i_t}^{r_2t}\big]\\
		&=-\sum_{l\neq s}^d\mbE[Q_{i_si_t}^{st}\rho_l(z)]-\sum_{l\neq s}^d\sum_{r\neq s,l}^da_{i_s}^{(s)}\mbE[Q_{i_t\cdot}^{tr}\bba^{(r)}\rho_l(z)]+\mrO(\eta_0^{-2}N^{-1}),
	\end{align*}
	where \(\rho_l(z)=N^{-1}\tr(\bbQ^{ll}(z))\) and \(Q_{i_t\cdot}^{tr}(z)\) is the \(i_t\)-th row of \(\bbQ^{tr}\). By Lemma \ref{Thm of a.s. convergence}, we can derive that
	\begin{align*}
		(z+\mfm(z)-\mfm_s(z))\mbE[Q_{i_si_t}^{st}]=-\delta_{st}\delta_{i_si_t}-\sum_{l\neq s}^d\mfm_l(z)\sum_{r\neq s,l}^da_{i_s}^{(s)}\mbE[Q_{i_t\cdot}^{tr}\bba^{(r)}]+\mrO(\eta_0^{-10}N^{-2\omega}),
	\end{align*}
	where \(\mfm_l(z)=\mbE[\rho_l(z)]\) and \(\mfm(z)=\sum_{l=1}^d\mfm_l(z)\). Thus, we obtain that
    \begin{align*}
        &(z+\mfm(z)-\mfm_s(z))\mfm_s(z)=-\mfc_s+\mrO(\eta_0^{-10}N^{-2\omega})
    \end{align*}
    combining the fact \(|\mfm_s(z)|\geq\eta_0\) and Lemma \ref{Thm of Stability}, we can conclude \eqref{Eq of approximation general d}. Consequently, we obtain that 
    \begin{align*}
		(z+g(z)-g_s(z))\mbE[Q_{i_si_t}^{st}]=-\delta_{st}\delta_{i_si_t}-\sum_{l\neq s}^dg_l(z)\sum_{r\neq s,l}^da_{i_s}^{(s)}\mbE[Q_{i_t\cdot}^{tr}\bba^{(r)}]+\mrO(\eta_0^{-16}N^{-2\omega}),
	\end{align*}
    For \(\mbE[Q_{i_t\cdot}^{tr}\bba^{(r)}]\), by the previous trick, we can obtain that
	\begin{align*}
		&z\mbE[Q_{i_t\cdot}^{tr}\bba^{(r)}]=-\sum_{l\neq t}^d\mbE[Q_{i_t\cdot}^{tr}\bba^{(r)}\rho_l(z)]-a_{i_t}^{(t)}\sum_{l\neq t}^d\sum_{w\neq t,l}^d\mbE[(\bba^{(r)})'\bbQ^{rw}\bba^{(w)}\rho_l(z)]+\mrO(a_{i_t}^{(t)}\eta_0^{-3}N^{-1/2}+\eta_0^{-3}N^{-1}),
	\end{align*}
	further by Lemma \ref{Thm of a.s. convergence} and \eqref{Eq of approximation general d}, it yields that
	\begin{align*}
		&(z+g(z)-g_t(z))\mbE[Q_{i_t\cdot}^{tr}\bba^{(r)}]=-a_{i_t}^{(t)}\sum_{l\neq t}^dg_l(z)\sum_{w\neq t,l}^dW_{rw,N}^{(d)}+\mrO(a_{i_t}^{(t)}\eta_0^{-3}N^{-1/2}+\eta_0^{-3}N^{-1}+\eta_0^{-16}N^{-2\omega}),
	\end{align*}
	i.e.
	\begin{align*}
		&a_{i_t}^{(t)}\mbE[Q_{i_t\cdot}^{tr}\bba^{(r)}]=(a_{i_t}^{(t)})^2\mfc_t^{-1}g_t(z)\sum_{l\neq t}^dg_l(z)\sum_{w\neq t,l}^dW_{rw,N}^{(d)}+a_{i_t}^{(t)}\mrO(a_{i_t}^{(t)}\eta_0^{-4}N^{-1/2}+\eta_0^{-4}N^{-1}+\eta_0^{-17}N^{-2\omega}).
	\end{align*}
	Summing all \(i_t=1,\cdots,n_t\) of above equations, we obtain
	\begin{align*}
		W_{tr,N}^{(d)}=\mfc_t^{-1}g_t(z)\sum_{l\neq t}^dg_l(z)\sum_{w\neq t,l}^dW_{rw,N}^{(d)}+\mrO(\eta_0^{-17}N^{-2\omega+1/2}),
	\end{align*}
	so
	\begin{align*}
		\mbE[Q_{i_t\cdot}^{tr}\bba^{(r)}]=a_{i_t}^{(t)}W_{tr,N}^{(d)}(z)+\mrO(\eta_0^{-17}N^{-2\omega+1/2}),
	\end{align*}
	and
	\begin{align*}
		\mbE[Q_{i_si_t}^{st}]=\mfc_s^{-1}g_s(z)\Big(\delta_{st}\delta_{i_si_t}+a_{i_s}^{(s)}\sum_{l\neq s}^dg_l(z)\sum_{r\neq s,l}^da_{i_t}^{(t)}W_{tr,N}^{(d)}(z)\Big)+\mrO(\eta_0^{-19}N^{-2\omega+1/2}).
	\end{align*}
	Since \(\omega\in(1/2-\delta,1/2)\), then \(2\omega-1/2\in(1/2-2\delta,1/2)\). Finally, combining Lemmas \ref{Thm of a.s. convergence} and \eqref{Eq of bbW general}, we complete our proof.
\end{proof}

\subsection{Mean and covariance functions}
Before establishing the CLT for LSS of \(\bbM\) for general \(d\geq3\), we derive the general forms of the mean function \(\mu_N^{(3)}(z)\) in (\ref{Eq of mean function d=3}) and the covariance function \(\mcC_N^{(3)}(z_1,z_2)\) in (\ref{Eq of covariance function d=3}) as follows:
\subsubsection{Covariance function}
\begin{thm}\label{Thm of covariance general d}
	Under Assumptions {\rm \ref{Ap of general noise}} and {\rm \ref{Ap of dimension}}, for any \(z_1,z_2\in\mathcal{S}_{\eta_0}\) in {\rm (\ref{Eq of stability region general d})}, let
	\begin{align}
		\mcC_{st,N}^{(d)}(z_1,z_2):=\Cov(\tr(\bbQ^{ss}(z_1)),\tr(\bbQ^{tt}(z_2)))\quad{\rm and}\quad\bbC_N^{(d)}(z_1,z_2):=[\mcC_{st,N}^{(d)}(z_1,z_2)]_{d\times d},\label{Eq of mcC covariance}
	\end{align}
	where \(s,t\in\{1,\cdots,d\}\), then we have
	\begin{align}
		\lim_{N\to\infty}\Vert\bbC_N^{(d)}(z_1,z_2)-\bbPi^{(d)}(z_1,z_2)^{-1}\diag(\mfc^{-1}\circ\bbg(z_1))\bbF_N^{(d)}(z_1,z_2)\Vert=0,\label{Eq of bbC}
	\end{align}
	where \(\bbPi^{(d)}(z_1,z_2)\) is defined in {\rm (\ref{Eq of invertible 2})} and
	\begin{align*}
		\bbF_N^{(d)}(z_1,z_2)=[\mcF_{st,N}^{(d)}(z_1,z_2)]_{d\times d}\quad\mcF_{st,N}^{(d)}(z_1,z_2):=2\mcV_{st}^{(d)}(z_1,z_2)+\kappa_4\mcW_{st,N}^{(d)}(z_1,z_2),
	\end{align*}
	and the precise definitions of \(\mcV_{st}^{(d)}(z_1,z_2)\) and \(\mcW_{st,N}^{(d)}(z_1,z_2)\) are postponed to {\rm (\ref{Eq of mcV limit general d})} and {\rm (\ref{Eq of mcW limit general d})}, respectively. Consequently, \(\Var(\tr(\bbQ(z)))\) is bounded by \(C_{\eta_0,\mfc}\) for any \(z\in\mcS_{\eta_0}\) and
    $$\lim_{N\to\infty}|\Cov(\tr(\bbQ(z_1)),\tr(\bbQ(z_2)))-\mcC_N^{(d)}(z_1,z_2)|=0,$$
    where
    \begin{align}
    \mcC_N^{(d)}(z_1,z_2):=\boldsymbol{1}_d'\bbPi^{(d)}(z_1,z_2)^{-1}\diag(\mfc^{-1}\circ\bbg(z_1))\bbF_N^{(d)}(z_1,z_2)\boldsymbol{1}_d.\label{Eq of covariance function general d}
\end{align}
\end{thm}
\begin{proof}
	Let us  first show that \(\mcC_{k_1k_2,N}^{(d)}(z,z)\) is bounded by \(C_{\eta_0,\mfc,d}\). Without loss of generality, we assume that \(|\mcC_{k_1k_2,N}^{(d)}(z,z)|>1\), otherwise, they are already bounded. For convenience, we omit \((z,z)\) behind \(\mcC_{k_1k_2,N}^{(d)}(z,z)\). By the cumulant expansion (\ref{Eq of cumulant expansion}), we have
	\begin{align*}
		&z\mcC_{k_1k_2,N}^{(d)}=z\Cov(\tr(\bbQ^{k_1k_1}(z)),\tr(\bbQ^{k_2k_2}(z)))\\
		&=\frac{1}{\sqrt{N}}\sum_{l\neq k_1}^d\sum_{i_1\cdots i_d}^{n_1\cdots n_d}\mbE\left[X_{i_1\cdots i_d}\mcA_{i_1\cdots i_d}^{(k_1,l)}Q_{i_{k_1}i_l}^{k_1l}(z)\tr(\bbQ^{k_2k_2}(\bar{z}))^c\right]\\
		&=\frac{1}{\sqrt{N}}\sum_{l\neq k_1}^d\sum_{i_1\cdots i_d}^{n_1\cdots n_d}\Bigg(\sum_{\alpha=1}^3\frac{\kappa_{\alpha+1}}{\alpha!}\mbE\left[\mcA_{i_1\cdots i_d}^{(k_1,l)}\partial_{i_1\cdots i_d}^{(\alpha)}\big\{Q_{i_{k_1}i_l}^{k_1l}(z)\tr(\bbQ^{k_2k_2}(\bar{z}))^c\big\}\right]+\epsilon_{i_1\cdots i_d}^{(4)}\Bigg),
	\end{align*}
    where the remainder satisfies that \(|\epsilon_{i_1\cdots i_d}^{(4)}|\leq C_{\kappa_5}\sup_{z\in\mcS_{\eta_0}}|\partial_{i_1\cdots i_d}^{(4)}\big\{Q_{i_{k_1}i_l}^{k_1l}(z)\tr(\bbQ^{k_2k_2}(\bar{z}))^c\big\}|\). As in the proof of Theorem \ref{Thm of Variance}, we claim that the major terms only appear in the first and third derivatives. Since the proofs are the same as those for Theorem \ref{Thm of Variance}, here we demonstrate only that \(N^{-1/2}|\sum_{i_1\cdots i_d}^{n_1\cdots n_d}\epsilon_{i_1\cdots i_d}^{(4)}|\to0\), as this serves as a representative example. By Lemmas \ref{Rem of extension minor 2} and \ref{Rem of extension of Corollary}, it is enough to show that for each \(l\neq k_1\) and \(0\leq\gamma\leq4\)
    \begin{align}
        &N^{-1/2}\Big|\sum_{i_1\cdots i_d}^{n_1\cdots n_d}\mcA_{i_1\cdots i_d}^{(k_1,l)}\mathscr{D}\big(\partial_{i_1\cdots i_d}^{(\gamma)}\{Q_{i_{k_1}i_l}^{k_1l}(z)\}\big)\mathscr{D}\big(\partial_{i_1\cdots i_d}^{(4-\gamma)}\{\tr(\bbQ^{k_2k_2}(\bar{z}))^c\}\big)\Big|\to0.\label{Eq of remainder covariance general d}
    \end{align}
    For \(\gamma=0\), it is easy to see \(\mathscr{D}\big(\partial_{i_1\cdots i_d}^{(0)}\{Q_{i_{k_1}i_l}^{k_1l}(z)\}\big)=0\). For \(\gamma=3\), since
    \begin{align}
		\partial_{i_1\cdots i_d}^{(1)}\tr(\bbQ^{k_2k_2}(\bar{z}))=\sum_{j=1}^{n_{k_2}}\partial_{i_1\cdots i_d}^{(1)}Q_{jj}^{k_2k_2}(\bar{z})=-\frac{1}{\sqrt{N}}\sum_{t_1\neq t_2}^{d,d}\mcA_{i_1\cdots i_d}^{(t_1,t_2)}Q_{i_{t_1}\cdot}^{t_1k_2}(\bar{z})Q_{\cdot i_{t_2}}^{k_2t_2}(\bar{z}),\label{Eq of partial trace d}
	\end{align}
    then \(\mathscr{D}\big(\partial_{i_1\cdots i_d}^{(1)}\{\tr(\bbQ^{k_2k_2}(\bar{z}))^c\}\big)=0\). Next, for \(\gamma=4\), note that 
    \begin{align}
        &\mcA_{i_1\cdots i_d}^{(k_1,l)}\mathscr{D}\big(\partial_{i_1\cdots i_d}^{(4)}\{Q_{i_{k_1}i_l}^{k_1l}(z)\}\big)=N^{-2}\mcA_{i_1\cdots i_d}^{(k_1,l)}Q_{i_{k_1}i_{k_1}}^{k_1k_1}(z)Q_{i_li_l}^{ll}(z)\times\notag\\
        &\sum_{s_1\neq k_1}^d\sum_{s_2\neq s_1}^d\sum_{s_3\neq s_2,l}^d\mcA_{i_1\cdots i_d}^{(k_1,s_1)}Q_{i_{s_1}i_{s_1}}^{s_1s_1}(z)\mcA_{i_1\cdots i_d}^{(s_1,s_2)}Q_{i_{s_2}i_{s_2}}^{s_2s_2}(z)\mcA_{i_1\cdots i_d}^{(s_2,s_3)}Q_{i_{s_3}i_{s_3}}^{s_3s_3}(z)\mcA_{i_1\cdots i_d}^{(s_3,l)},\label{Eq of remainder covariance general d 1}
    \end{align}
    then (\ref{Eq of remainder covariance general d 1}) contains at least two different $Q_{i_{s_l}i_{s_l}}^{s_ls_l}$ coming from diagonal blocks, denote \(n_k\) to be the number \(a_{i_k}^{(k)}\) appearing in (\ref{Eq of remainder covariance general d 1}), so \(n_r\geq1\) for \(1\leq r\leq d\). It is easy to see that \(|\sum_{i_1\cdots i_d}^{n_1\cdots n_d}(\ref{Eq of remainder covariance general d 1})|\leq\mrO(N^{-2}\Vert\bbQ(z)\Vert^5)\) if all \(n_r\geq2\), so that \((\ref{Eq of remainder covariance general d})\leq\mrO(N^{-3/2}\Vert\bbQ(z)\Vert^6)\). Otherwise, there are at most two different \(n_{r_1}=1\) and \(n_{r_2}=1\). If there is only one \(n_r=1\), we can conclude that \(|\sum_{i_1\cdots i_d}^{n_1\cdots n_d}(\ref{Eq of remainder covariance general d 1})|\leq\mrO(N^{-3/2}\Vert\bbQ(z)\Vert^5)\), so that \((\ref{Eq of remainder covariance general d})\leq\mrO(N^{-1}\Vert\bbQ(z)\Vert^6)\). If there are two different \(n_{r_1}=1\) and \(n_{r_2}=1\), then the only possible situation is \(n_{k_1}=n_l=1\), and we can conclude that \(|\sum_{i_1\cdots i_d}^{n_1\cdots n_d}(\ref{Eq of remainder covariance general d 1})|\leq\mrO(N^{-1}\Vert\bbQ(z)\Vert^5)\), so that \((\ref{Eq of remainder covariance general d})\leq\mrO(N^{-1/2}\Vert\bbQ(z)\Vert^6)\). For \(\gamma=1\), we have \(\mcA_{i_1\cdots i_d}^{(k_1,l)}\mathscr{D}\big(\partial_{i_1\cdots i_d}^{(1)}\{Q_{i_{k_1}i_l}^{k_1l}(z)\}\big)=N^{-1/2}(\mcA_{i_1\cdots i_d}^{(k_1,l)})^2Q_{i_{k_1}i_{k_1}}^{k_1k_1}(z)Q_{i_li_l}^{ll}(z)\) and
    \begin{align*}
        &\mathscr{D}\big(\partial_{i_1\cdots i_d}^{(3)}\tr(\bbQ^{k_2k_2}(z))\big)\\
        &=-N^{-3/2}\sum_{s_1=1}^d\sum_{s_2\neq s_1}^d\sum_{s_3\neq s_1,s_2}^dQ_{i_{s_1}\cdot}^{s_1k_2}(z)Q_{\cdot i_{s_1}}^{k_2s_1}(z)\mcA_{i_1\cdots i_d}^{(s_1,s_2)}Q_{i_{s_2}i_{s_2}}^{s_2s_2}(z)\mcA_{i_1\cdots i_d}^{(s_2,s_3)}Q_{i_{s_3}i_{s_3}}^{s_3s_3}(z)\mcA_{i_1\cdots i_d}^{(s_3,s_1)}.
    \end{align*}
    Similar to previous arguments, if \(\min_{1\leq r\leq d}n_r=2\) or 1, we can show that \((\ref{Eq of remainder covariance general d})\leq\mrO(N^{-1}\Vert\bbQ(z)\Vert^6)\). If \(\min_{1\leq r\leq d}n_r=0\), the only possible cases are that \(s_1=s_3=k_1,s_2=l\) and \(s_1=s_3=l,s_2=k_1\), then we can conclude that \((\ref{Eq of remainder covariance general d})\leq\mrO(N^{-1/2}\Vert\bbQ(z)\Vert^6)\). For \(\gamma=3\), the proofs are the same as those for \(\gamma=2\), the details are omitted for brevity. Next, we only will present the detailed calculations for \(\alpha=1\) and 3.
    
    \vspace{5mm}
    \noindent
    {\bf First derivatives:} When \(\alpha=1\), let
		\begin{align}
			\mcV_{k_1k_2,N}^{(d)}(z_1,z_2)=\frac{1}{N}\sum_{l\neq k}^d\mbE[\tr(\bbQ^{k_1k_2}(\bar{z}_2)\bbQ^{k_2l}(\bar{z}_2)\bbQ^{lk_1}(z_1))],\label{Eq of mcV}
		\end{align}
        and readers can refer to \S\ref{Subsec of general Major terms} for proofs of \(\mcV_{k_1k_2,N}^{(d)}(z_1,z_2)\to \mcV_{k_1k_2}^{(d)}(z_1,z_2)\) \eqref{Eq of mcV limit general d}. By directly calculation as in the proof of Theorem \ref{Thm of Variance}, we have
		\begin{align*}
			&\frac{1}{\sqrt{N}}\sum_{i_1\cdots i_d}^{n_1\cdots n_d}\mcA_{i_1\cdots i_d}^{(k_1,l)}Q_{i_{k_1}i_l}^{k_1l}(z)\partial_{i_1\cdots i_d}^{(1)}\tr(\bbQ^{k_2k_2}(\bar{z}))\\
			&=-\frac{1}{N}\sum_{i_1\cdots i_d}^{n_1\cdots n_d}\sum_{t_1\neq t_2}^{d,d}\mcA_{i_1\cdots i_d}^{(k_1,l)}Q_{i_{k_1}i_l}^{k_1l}(z)\mcA_{i_1\cdots i_d}^{(t_1,t_2)}Q_{i_{t_1}\cdot}^{t_1k_2}(\bar{z})Q_{\cdot i_{t_2}}^{k_2t_2}(\bar{z})\\
			&=-\frac{2}{N}\sum_{i_1\cdots i_d}^{n_1\cdots n_d}\mcA_{i_1\cdots i_d}^{(k_1,l)}Q_{i_{k_1}i_l}^{k_1l}(z)\mcA_{i_1\cdots i_d}^{(k_1,l)}Q_{i_{k_1}\cdot}^{k_1k_2}(\bar{z})Q_{\cdot i_l}^{k_2l}(\bar{z})+\mrO(\eta_0^{-3}N^{-1/2})\\
			&=-\frac{2}{N}\tr\big(\bbQ^{k_1k_2}(\bar{z})\bbQ^{k_2l}(\bar{z})\bbQ^{lk_1}(z)\big)+\mrO(\eta_0^{-3}N^{-1/2}),
		\end{align*}
		and
		\begin{align*}
			&\frac{1}{\sqrt{N}}\sum_{i_1\cdots i_d}^{n_1\cdots n_d}\mcA_{i_1\cdots i_d}^{(k_1,l)}\partial_{i_1\cdots i_d}^{(1)}Q_{i_{k_1}i_l}^{k_1l}(z)=-\Big(N\rho_{k_1}(z)\rho_l(z)+N^{-1}\tr\big(\bbQ^{k_1l}(z)\bbQ^{lk_1}(z)\big)\\
			&+\rho_{k_1}(z)\sum_{j\neq k_1,l}(\bba^{(j)})'\bbQ^{jl}(z)\bba^{(l)}+\rho_l(z)\sum_{j\neq k,l}(\bba^{(j)})'\bbQ^{jk_1}(z)\bba^{(k_1)}\Big)+\mrO(\eta_0^{-2}N^{-1/2}),
		\end{align*}
		where \(\rho_l(z)=N^{-1}\tr(\bbQ^{ll}(z))\). By Lemma \ref{Thm of a.s. convergence}
		$$\Cov(N^{-1}\tr\big(\bbQ^{k_1l}(z)\bbQ^{lk_1}(z)\big),\tr(\bbQ^{k_2k_2}(z)))\leq\mrO(C_{\eta_0}N^{-\omega})\mcC_{k_2k_2,N},$$
		as does others like \(\rho_{k_1}(z)(\bba^{(j)})'\bbQ^{jl}(z)\bba^{(l)}\), then we have
		\begin{align*}
			&\frac{1}{\sqrt{N}}\sum_{l\neq k}^d\sum_{i_1\cdots i_d}^{n_1\cdots n_d}\mbE\left[\mcA_{i_1\cdots i_d}^{(k_1,l)}\partial_{i_1\cdots i_d}^{(1)}\big\{Q_{i_{k_1}i_l}^{k_1l}(z)\tr(\bbQ^{k_2k_2}(\bar{z}))^c\big\}\right]\\
			&=-\sum_{l\neq k_1}^d\Cov(N\rho_{k_1}(z)\rho_l(z),\tr(\bbQ^{k_2k_2}(z)))-2\mcV_{k_1k_2,N}^{(3)}(z,z)+\mrO(C_{\eta_0}N^{-\omega})\mcC_{k_2k_2,N}.
		\end{align*}
		Next, we can repeat the trick in (\ref{Eq of covariance trick}) to obtain that
		\begin{align*}
			&\Cov(N\rho_{k_1}(z)\rho_l(z),\tr(\bbQ^{k_2k_2}(z)))\\
			&=\frac{1}{N}\mbE\big[\big(\tr(\bbQ^{k_1k_1}(z))\tr(\bbQ^{ll}(z))-\mbE[\tr(\bbQ^{k_1k_1}(z))]\tr(\bbQ^{ll}(z))+\mbE[\tr(\bbQ^{k_1k_1}(z))]\tr(\bbQ^{ll}(z))\\
			&-\mbE[\tr(\bbQ^{k_1k_1}(z))]\mbE[\tr(\bbQ^{ll}(z))]+\mbE[\tr(\bbQ^{k_1k_1}(z))]\mbE[\tr(\bbQ^{ll}(z))]-\mbE[\tr(\bbQ^{k_1k_1}(z))\tr(\bbQ^{ll}(z))]\big)\\
			&\times\big(\tr(\bbQ^{k_2k_2}(\bar{z}))-\mbE[\tr(\bbQ^{k_2k_2}(\bar{z}))]\big)\big]\\
			&=\mfm_l(z)\Cov(\tr(\bbQ^{k_1k_1}(z)),\tr(\bbQ^{k_2k_2}(z)))+\mfm_{k_1}(z)\Cov(\tr(\bbQ^{ll}(z)),\tr(\bbQ^{k_2k_2}(z)))+\mrO(C_{\eta_0}N^{-\omega})\mcC_{k_2k_2,N}.
		\end{align*}
		As a result, we obtain
		\begin{align*}
			&\frac{1}{\sqrt{N}}\sum_{l\neq k_1}^d\sum_{i_1\cdots i_d}^{n_1\cdots n_d}\mbE\left[\mcA_{i_1\cdots i_d}^{(k_1,l)}\partial_{i_1\cdots i_d}^{(1)}\big\{Q_{i_{k_1}i_l}^{k_1l}(z)\tr(\bbQ^{k_2k_2}(\bar{z}))^c\big\}\right]\\
			&=-\mcC_{k_1k_2,N}^{(d)}\sum_{l\neq k_1}^d\mfm_l(z)-\mfm_{k_1}(z)\sum_{l\neq k_1}^d\mcC_{lk_2,N}^{(d)}-2\mcV_{k_1k_2,N}^{(d)}(z,z)+\mrO(C_{\eta_0}N^{-\omega})\mcC_{k_2k_2,N}.
		\end{align*}
    {\bf Third derivatives:} When \(\alpha=3\), consider
		\begin{align*}
			\frac{1}{\sqrt{N}}\sum_{l\neq k_1}^d\sum_{i_1\cdots i_d}^{n_1\cdots n_d}\mbE\left[\mcA_{i_1\cdots i_d}^{(k_1,l)}\partial_{i_1\cdots i_d}^{(1)}Q_{i_{k_1}i_l}^{k_1l}(z)\partial_{i_1\cdots i_d}^{(2)}\tr(\bbQ^{k_2k_2}(\bar{z}))\right],
		\end{align*}
		where
		\begin{align*}
			\partial_{i_1\cdots i_d}^{(2)}\tr\bbQ^{k_2k_2}(\bar{z})=\frac{2}{N}\sum_{t_1\neq t_2,t_3\neq t_4}\mcA_{i_1\cdots i_d}^{(t_1,t_2)}\mcA_{i_1\cdots i_d}^{(t_3,t_4)}Q_{i_{t_2}i_{t_3}}^{t_2t_3}(\bar{z})Q_{i_{t_4}\cdot}^{t_4k_2}(\bar{z})Q_{\cdot i_{t_1}}^{k_2t_1}(\bar{z}),
		\end{align*}
		then by Lemma \ref{Rem of extension minor 2} and (\ref{Eq of mcB}), we have
		\begin{align*}
			&\frac{1}{\sqrt{N}}\sum_{i_1\cdots i_d}^{n_1\cdots n_d}\mbE\left[\mcA_{i_1\cdots i_d}^{(k_1,l)}\partial_{i_1\cdots i_d}^{(1)}Q_{i_{k_1}i_l}^{k_1l}(z)\partial_{i_1\cdots i_d}^{(2)}\tr(\bbQ^{k_2k_2}(\bar{z}))\right]\\
			&=-\frac{2}{N^2}\sum_{\substack{t_1\neq t_2\\t_3\neq t_4}}^d\sum_{s_1\neq s_2}^d\sum_{i_1\cdots i_d}^{n_1\cdots n_d}\mbE\left[\mcA_{i_1\cdots i_d}^{(k_1,l)}Q_{i_{k_1}i_{s_1}}^{k_1s_1}(z)\mcA_{i_1\cdots i_d}^{(s_1,s_2)}Q_{i_{s_2}i_l}^{s_2l}(z)\mcA_{i_1\cdots i_d}^{(t_1,t_2)}\mcA_{i_1\cdots i_d}^{(t_3,t_4)}Q_{i_{t_2}i_{t_3}}^{t_2t_3}(\bar{z})Q_{i_{t_4}\cdot}^{t_4k_2}(\bar{z})Q_{\cdot i_{t_1}}^{k_2t_1}(\bar{z})\right]\\
			&=-\frac{2}{N^2}\sum_{t_1\neq t_2}^d\sum_{i_1\cdots i_d}^{n_1\cdots n_d}\mbE\left[(\mcA_{i_1\cdots i_d}^{(k_1,l)})^2Q_{i_{k_1}i_{k_1}}^{k_1k_1}(z)Q_{i_li_l}^{ll}(z)(\mcA_{i_1\cdots i_d}^{(t_1,t_2)})^2Q_{i_{t_2}i_{t_2}}^{t_2t_2}(\bar{z}_2)Q_{i_{t_1}\cdot}^{t_1k_2}(\bar{z})Q_{\cdot i_{t_1}}^{k_2t_1}(\bar{z})\right]+\mrO(\eta_0^{-5}N^{-1/2})\\
			&=-\frac{2}{N^2}\sum_{i_1\cdots i_d}^{n_1\cdots n_d}\mbE\left[(\mcA_{i_1\cdots i_d}^{(k_1,l)})^4Q_{i_{k_1}i_{k_1}}^{k_1k_1}(z)Q_{i_li_l}^{ll}(z)Q_{i_{k_1}i_{k_1}}^{k_1k_1}(\bar{z})Q_{i_{l}\cdot}^{lk_2}(\bar{z})Q_{\cdot i_l}^{k_2l}(\bar{z})\right]\\
			&-\frac{2}{N^2}\sum_{i_1\cdots i_d}^{n_1\cdots n_d}\mbE\left[(\mcA_{i_1\cdots i_d}^{(k_1,l)})^4Q_{i_{k_1}i_{k_1}}^{k_1k_1}(z)Q_{i_li_l}^{ll}(z)Q_{i_li_l}^{ll}(\bar{z})Q_{i_{k_1}\cdot}^{k_1k_2}(\bar{z})Q_{\cdot i_{k_1}}^{k_2k_1}(\bar{z})\right]+\mrO(\eta_0^{-5}N^{-1/2})\\
			&=-\frac{2\mcB_{(4)}^{(k_1,l)}}{N^2}\mbE\big[\tr\big(\bbQ^{k_1k_1}(z)\circ\bbQ^{k_1k_1}(\bar{z})\big)\cdot\tr\big(\bbQ^{ll}(z)\circ(\bbQ^{lk_2}(\bar{z})\bbQ^{k_2l}(\bar{z}))\big)\big]\\
			&-\frac{2\mcB_{(4)}^{(k,l)}}{N^2}\mbE\big[\tr\big(\bbQ^{ll}(z)\circ\bbQ^{ll}(\bar{z})\big)\cdot\tr\big(\bbQ^{k_1k_1}(z)\circ(\bbQ^{k_1k_2}(\bar{z})\bbQ^{k_2k_1}(\bar{z}))\big)\big]+\mrO(\eta_0^{-5}N^{-1/2}),
		\end{align*}
        where $\mcB_{(4)}^{(k,l)}$ is defined in \eqref{Eq of mcB}. For simplicity, denote
		\begin{align}
			&\widetilde{\mcW}_{k_1k_2,N}^{(d)}(z_1,z_2):=N^{-2}\sum_{l\neq k_1}^d\mcB_{(4)}^{(k_1,l)}\mbE[\tr(\bbQ^{k_1k_1}(z_1)\circ\bbQ^{k_1k_1}(\bar{z}_2))\cdot\tr(\bbQ^{ll}(z_1)\circ(\bbQ^{lk_2}(\bar{z}_2)\bbQ^{k_2l}(\bar{z}_2)))]\notag\\
			&+N^{-2}\sum_{l\neq k_1}^d\mcB_{(4)}^{(k_1,l)}\mbE[\tr(\bbQ^{ll}(z_1)\circ\bbQ^{ll}(\bar{z}_2))\cdot\tr(\bbQ^{k_1k_1}(z_1)\circ(\bbQ^{k_1k_2}(\bar{z}_2)\bbQ^{k_2k_1}(\bar{z}_2)))],\label{Eq of mcW}
		\end{align}
        and readers can refer to \S\ref{Subsec of general Major terms} for proofs of \(\widetilde{\mcW}_{k_1k_2,N}^{(d)}(z_1,z_2)\to\mcW_{k_1k_2,N}^{(d)}(z_1,z_2)\) \eqref{Eq of mcW limit general d}. As a result, we obtain
		\begin{align*}
			&\frac{1}{\sqrt{N}}\sum_{l\neq k_1}^d\sum_{i_1\cdots i_d}^{n_1\cdots n_d}\frac{\kappa_4}{2}\mbE\left[\mcA_{i_1\cdots i_d}^{(k_1,l)}\partial_{i_1\cdots i_d}^{(1)}Q_{i_{k_1}i_l}^{k_1l}(z)\partial_{i_1\cdots i_d}^{(2)}\tr(\bbQ^{k_2k_2}(\bar{z}))\right]=-\kappa_4\mcW_{k_1k_2,N}^{(d)}(z,z)+\mrO(\eta_0^{-5}N^{-1/2}).
		\end{align*}
	In summary, we obtain that
	\begin{align*}
		z\mcC_{k_1k_2,N}^{(d)}=&-\mcC_{k_1k_2,N}^{(d)}\Big(\sum_{l\neq k_1}^d\mfm_l(z)+\mrO(N^{-\omega})\Big)-\mfm_{k_1}(z)\sum_{l\neq k_1}^d\mcC_{lk_2,N}^{(d)}-2\mcV_{k_1k_2,N}^{(d)}(z,z)\\
		&-\kappa_4\mcW_{k_1k_2,N}^{(d)}(z,z)+\mrO(C_{\eta_0}N^{-\omega})+\mrO(C_{\eta_0}N^{-\omega})\mcC_{k_2k_2,N}^{(d)},
	\end{align*}
	i.e.
	\begin{align*}
		&(z+\mfm(z)-\mfm_{k_1}(z))\mcC_{k_1k_2,N}^{(d)}=-\mfm_{k_1}(z)\sum_{l\neq k_1}^d\mcC_{lk_2,N}^{(d)}-\mcF_{k_1k_2,N}^{(d)}(z,z)+\mrO(C_{\eta_0}N^{-\omega})\mcC_{k_2k_2,N}^{(d)}+\mrO(C_{\eta_0}N^{-\omega}),
	\end{align*}
	where
	\begin{align}
		\mcF_{k_1k_2,N}^{(d)}(z_1,z_2):=2\mcV_{k_1k_2,N}^{(d)}(z_1,z_2)+\kappa_4\mcW_{k_1k_2,N}^{(d)}(z_1,z_2).\label{Eq of mcF}
	\end{align}
	Now, in matrix notations, we obtain that
	\begin{align}
		\bbTheta_N^{(d)}(z,z)\bbC_N^{(d)}=-\bbF_N^{(d)}+\mrO(C_{\eta_0}N^{-\omega})\boldsymbol{1}_{d\times d}+\mrO(C_{\eta_0}N^{-\omega})\boldsymbol{1}_{d\times d}\diag(\bbC_N^{(d)}).\label{Eq of bbTheta general d}
	\end{align}
	Similar to (\ref{Eq of bbTheta d=3}), define
	\begin{align}
		\bbTheta_N^{(d)}(z,z):=(z+\mfm(z))\bbI_d-\diag(\bbm(z))+\diag(\bbm(z))\bbS_d,\label{Eq of bbPi N}
	\end{align}
	where \(\bbS_d\) is defined in (\ref{Eq of bbS d}). According to \eqref{Eq of approximation general d}, we have 
	$$\Vert\bbTheta_N^{(d)}(z,z)+\diag(\mfc^{-1}\circ\bbg(z))^{-1}\bbPi^{(d)}(z,z)\Vert_{\infty}\leq\mrO(C_{\eta_0}N^{-\omega}),$$
	which implies that \(\bbPi_N^{(d)}(z,z)\) is invertible and \(|\Vert\bbTheta_N^{(d)}(z,z)^{-1}-\bbPi^{(d)}(z,z)^{-1}\diag(\mfc^{-1}\circ\bbg(z))\Vert|\leq\mrO(C_{\eta_0}N^{-\omega})\), so we have
    \begin{align*}
		\lim_{N\to\infty}\Vert\bbC_N^{(d)}(z,z)-\bbPi^{(d)}(z,z)^{-1}\diag(\mfc^{-1}\circ\bbg(z))\bbF_N^{(d)}(z,z)\Vert=0.
	\end{align*}
    Moreover, by the definitions of \(\mcV_{st,N}^{(d)}(z_1,z_2)\) and \(\mcW_{st,N}^{(d)}(z_1,z_2)\) in (\ref{Eq of mcV}) and (\ref{Eq of mcW}), we know that \(|\mcF_{st,N}^{(d)}(z_1,z_2)|\leq C_{\eta_0}\), so
	\begin{align*}
		&\bbC_N^{(d)}=\bbPi^{(d)}(z,z)^{-1}\diag(\mfc^{-1}\circ\bbg(z))\bbF_N^{(d)}+\mro(\boldsymbol{1}_{d\times d})+\mro(\boldsymbol{1}_{d\times d})\diag(\bbC_N^{(d)})\\
		\Rightarrow\quad&\Vert\bbC_N^{(d)}\Vert\leq\Vert\bbPi^{(d)}(z,z)^{-1}\diag(\mfc^{-1}\circ\bbg(z))\bbF_N^{(d)}\Vert+\mro(1)+\mro(1)\Vert\bbC_N^{(d)}\Vert,
	\end{align*}
	which ensures that \(\Vert\bbC_N^{(d)}\Vert\) is bounded by \(C_{\eta_0,\mfc,d}\). For \(z_1\neq z_2\in\mcS_{\eta_0}\), by the previous arguments, we can still derive the following system equation for \(\mcC_{k_1k_2,N}^{(d)}(z_1,z_2)\):
	\begin{align*}
		&(z_1+\mfm(z_1)-\mfm_{k_1}(z_1))\mcC_{k_1k_2,N}^{(d)}(z_1,z_2)=-\mfm_{k_1}(z_2)\sum_{l\neq k_1}^d\mcC_{lk_2,N}^{(d)}(z_1,z_2)-\mcF_{k_1k_2,N}^{(d)}(z_1,z_2)\\
		&+\mrO(C_{\eta_0}N^{-\omega})\mcC_{k_2k_2,N}^{(d)}(z_2,z_2)+\mrO(C_{\eta_0}N^{-\omega}),
	\end{align*}
	since we have shown that \(\mcC_{k_2k_2,N}^{(d)}(z_2,z_2)\) is bounded by \(C_{\eta_0,\mfc,d}\), then the above equation can be transformed into the following matrix forms:
	\begin{align*}
		&\lim_{N\to\infty}\Vert\bbC_N^{(d)}(z_1,z_2)-\bbPi^{(d)}(z_1,z_2)^{-1}\diag(\mfc^{-1}\circ\bbg(z_1))\bbF_N^{(d)}(z_1,z_2)\Vert=0,
	\end{align*}
	which completes our proof.
\end{proof}
\subsubsection{Mean function}
\begin{thm}\label{Thm of mean general d}
	Under Assumptions {\rm \ref{Ap of general noise}} and {\rm \ref{Ap of dimension}}, for any \(z\in\mathcal{S}_{\eta_0}\) in {\rm (\ref{Eq of stability region general d})}, let 
	$$\overrightarrow{M}_N^{(d)}(z)=(M_{1,N}^{(d)}(z),\cdots,M_{d,N}^{(d)}(z))'$$
	such that
	\begin{align}
		&M_{i,N}^{(d)}(z):=g_i(z)\sum_{r\neq i}^d\sum_{w\neq i,r}^dW_{rw}^{(d)}(z)-2\kappa_3 G_{i,N}^{(d)}(z)+\kappa_4 H_{i,N}^{(d)}(z,z)\notag\\
		&+\sum_{l\neq i}^d\big[(g(z)-g_i(z)-g_l(z))W_{il}^{(d)}(z)+V_{il}^{(d)}(z,z)\big],\label{Eq of Mi general d}
	\end{align}
	where \(W_{jk}^{(d)}(z),V_{ij}^{(d)}(z,z),G_{i,N}^{(d)}(z),H_{i,N}^{(d)}(z,z)\) are defined in {\rm (\ref{Eq of bbW general}), (\ref{Eq of bbV}), (\ref{Eq of H2 general d}), (\ref{Eq of H3 general d})}. Then we have
    $$\lim_{N\to\infty}|\mbE[\tr(\bbQ(z)]-Ng(z)-\mu_N^{(d)}(z)|=0,$$
    where
	\begin{align}
		\mu_N^{(d)}(z):=\boldsymbol{1}_d'\boldsymbol{\Pi}^{(d)}(z,z)^{-1}\diag(\boldsymbol{\mfc}^{-1}\circ\bbg(z))\overrightarrow{M}_N^{(d)}(z),\label{Eq of mean function general d}
	\end{align}
	and \(\boldsymbol{\Pi}^{(d)}(z,z)\) is defined in {\rm (\ref{Eq of invertible 2})}.
\end{thm}
\begin{proof}
For simplicity, we will omit \((z)\) behind \(\bbQ(z)\), as does \(\rho_k=N^{-1}\tr(\bbQ^{kk}),\mfm_k=\mbE[\rho_k],W_{kl,N}^{(d)}=\mbE[(\bba^{(k)})'\bbQ^{kl}\bba^{(l)}]\) and \(V_{kl,N}^{(d)}=V_{kl,N}^{(d)}(z,z)=\mbE[\tr(\bbQ^{kl}(z)\bbQ^{lk}(z))]\). Moreover, for proofs of \(W_{kl,N}^{(d)}(z)\to W_{kl}(z)\) and \(V_{kl,N}^{(d)}(z_1,z_2)\to V_{kl}^{(d)}(z_1,z_2)\), readers can refer to \S\ref{Subsec of general Major terms}. Note that \(\bbM\bbQ-\bbI=z\bbQ\), we have
	$$z\mbE[\tr(\bbQ^{kk})]=\frac{1}{\sqrt{N}}\sum_{l\neq k}^d\sum_{i_1\cdots i_d}^{n_1\cdots n_d}\mbE[X_{i_1\cdots i_d}\mcA_{i_1\cdots i_d}^{(k,l)}Q_{i_ki_l}^{kl}]-n_k,$$
	where \(\mcA_{i_1\cdots i_d}^{(k,l)}=\prod_{j\neq k,l}^d a_{i_j}^{(j)}\). By the cumulant expansion, we have
	$$\sum_{l\neq k}^d\mbE[X_{i_1\cdots i_d}\mcA_{i_1\cdots i_d}^{(k,l)}Q_{i_ki_l}^{kl}]=\sum_{\alpha=1}^3\sum_{l\neq k}^d\frac{\kappa_{\alpha+1}}{\alpha!}\mbE[\mcA_{i_1\cdots i_d}^{(k,l)}\partial_{i_1\cdots i_d}^{(\alpha)}Q_{i_ki_l}^{kl}]+\epsilon_{i_1\cdots i_d}^{(4)}.$$
    {\bf First derivatives:} When \(\alpha=1\), since 
		$$\partial_{i_1\cdots i_d}^{(1)}Q_{i_ki_l}^{kl}=-\frac{1}{\sqrt{N}}\sum_{t_1\neq t_2}^dQ_{i_ki_{t_1}}^{kt_1}\mcA_{i_1\cdots i_d}^{(t_1,t_2)}Q_{i_{t_2}i_l}^{t_2l}$$
		then by direct calculation, we have
		\begin{align*}
			&\frac{1}{\sqrt{N}}\sum_{i_1\cdots i_d}^{n_1\cdots n_d}\sum_{l\neq k}^d\mbE[\mcA_{i_1\cdots i_d}^{(k,l)}\partial_{i_1\cdots i_d}^{(1)}Q_{i_ki_l}^{kl}]=-\frac{1}{N}\sum_{i_1\cdots i_d}^{n_1\cdots n_d}\sum_{l\neq k}^d\sum_{t_1\neq t_2}^d\mbE[\mcA_{i_1\cdots i_d}^{(k,l)}Q_{i_ki_{t_1}}^{kt_1}\mcA_{i_1\cdots i_d}^{(t_1,t_2)}Q_{i_{t_2}i_l}^{t_2l}]\\
			&=-\sum_{l\neq k}^d\Big(N\mfm_k\mfm_l+V_{kl,N}^{(d)}+\mfm_k\sum_{j\neq k,l}W_{jl,N}^{(d)}+\mfm_l\sum_{j\neq k,l}W_{jk,N}^{(d)}\Big)+\mrO(C_{\eta_0}N^{-\omega}),
		\end{align*}
		where we use Lemma \ref{Thm of a.s. convergence} and Theorem \ref{Thm of covariance general d} to conclude that \(\Cov(\rho_k,W_{kl,N})\leq\mrO(C_{\eta_0}N^{-\omega})\) and \(\Var(N\rho_k)\leq C_{\eta_0\,\mfc,d}\), respectively.

    \vspace{5mm}
    \noindent
    {\bf Second derivatives:} When \(\alpha=2\), by Lemma \ref{Rem of extension minor 2} and (\ref{Eq of mcB}), we have
		\begin{align*}
			&\frac{1}{\sqrt{N}}\sum_{i_1\cdots i_d}^{n_1\cdots n_d}\sum_{l\neq k}^d\mbE\big[\mcA_{i_1\cdots i_d}^{(k,l)}\partial_{i_1\cdots i_d}^{(2)}Q_{i_ki_l}^{kl}\big]\\
			&=\frac{2}{N^{3/2}}\sum_{l\neq k}^d\sum_{t\neq k,l}^d\sum_{i_1\cdots i_d}^{n_1\cdots n_d}\mbE\big[\mcA_{i_1\cdots i_d}^{(k,l)}Q_{i_ki_k}^{kk}\mcA_{i_1\cdots i_d}^{(k,t)}Q_{i_ti_t}^{tt}\mcA_{i_1\cdots i_d}^{(t,l)}Q_{i_li_l}^{ll}\big]+\mrO(\eta_0^{-3}N^{-1/2})\\
			&=2\sum_{l\neq k}^d\sum_{t\neq k,l}^d\mcB_{(3)}^{(k,l,t)}\mbE[\boldsymbol{1}'\diag(\bbQ^{kk})\bba^{(k)}\cdot\boldsymbol{1}'\diag(\bbQ^{tt})\bba^{(t)}\cdot\boldsymbol{1}'\diag(\bbQ^{ll})\bba^{(l)}]+\mrO(\eta_0^{-3}N^{-1/2}),
		\end{align*}
		where \(\mcB_{(3)}^{(k,l,t)}=\sum_{i_1\cdots i_d}^{n_1\cdots n_d}(\mcA_{i_1\cdots i_d}^{(k,l,t)})^3\). Similar to (\ref{Eq of H2 d=3}), we can further imply that
		\begin{align}
			&\frac{1}{\sqrt{N}}\sum_{i_1\cdots i_d}^{n_1\cdots n_d}\sum_{l\neq k}^d\mbE\big[\mcA_{i_1\cdots i_d}^{(k,l)}\partial_{i_1\cdots i_d}^{(2)}Q_{i_ki_l}^{kl}\big]\notag\\
			&=2\sum_{l\neq k}^d\sum_{t\neq k,l}^d\mcB_{(3)}^{(k,l,t)}(\mfc_k\mfc_t\mfc_l)^{-1}g_k(z)g_t(z)g_l(z)\mfb_k^{(1)}\mfb_t^{(1)}\mfb_l^{(1)}+\mrO(C_{\eta_0}N^{-\omega})\notag\\
			:&=2G_{k,N}^{(d)}(z)+\mrO(C_{\eta_0}N^{-\omega}).\label{Eq of H2 general d}
		\end{align}
    {\bf Third derivatives:} When \(\alpha=3\), similar to (\ref{Eq of H3 d=3}), by Lemmas \ref{Thm of a.s. convergence} and \ref{Rem of extension minor 2}, we have
		\begin{align*}
			&\frac{1}{\sqrt{N}}\sum_{i_1\cdots i_d}^{n_1\cdots n_d}\sum_{l\neq k}^d\mbE\big[\mcA_{i_1\cdots i_d}^{(k,l)}\partial_{i_1\cdots i_d}^{(3)}Q_{i_ki_l}^{kl}\big]=-\frac{6}{N^2}\sum_{i_1\cdots i_d}^{n_1\cdots n_d}\sum_{l\neq k}^d\mbE\big[(\mcA_{i_1\cdots i_d}^{(k,l)})^4(Q_{i_ki_k}^{kk}Q_{i_li_l}^{ll})^2\big]+\mrO(\eta_0^{-4}N^{-1/2})\\
			&=-6N^{-2}\sum_{l\neq k}^d\mcB_{(4)}^{(k,l)}\mbE[\tr(\bbQ^{kk}\circ\bbQ^{kk})\tr(\bbQ^{ll}\circ\bbQ^{ll})]+\mrO(C_{\eta_0}N^{-\omega})\\
			&=-6\mfc_k^{-1}g_k(z)g_k(z)\sum_{l\neq k}^d\mcB_{(4)}^{(k,l)}\mfc_l^{-1}g_l(z)g_l(z)+\mrO(C_{\eta_0}N^{-\omega}):=-6H_{3,k}^{(d)}(z,z)+\mrO(C_{\eta_0}N^{-\omega}),
		\end{align*}
		where $\mcB_{(4)}^{(k,l)}$ is defined in \eqref{Eq of mcB} and
		\begin{align}
			H_{k,N}^{(d)}(z_1,z_2):=\mfc_k^{-1}g_k(z_1)g_k(z_2)\sum_{l\neq k}^d\mcB_{(4)}^{(k,l)}\mfc_l^{-1}g_l(z_1)g_l(z_2).\label{Eq of H3 general d}
		\end{align}
    {\bf Remainders:} When \(\alpha=4\), similar to what we have done in \S\ref{ssec of mean function} for \(l=4\), by Lemma \ref{Rem of extension minor 2}, it is enough to consider 
		\begin{align*}
			&\frac{1}{\sqrt{N}}\sum_{i_1\cdots i_d}^{n_1\cdots n_d}\sum_{l\neq k}^d\mbE\big[\mcA_{i_1\cdots i_d}^{(k,l)}\msD\{\partial_{i_1\cdots i_d}^{(4)}Q_{i_ki_l}^{kl}\}\big]=\frac{4}{N^{5/2}}\sum_{i_1\cdots i_d}^{n_1\cdots n_d}\sum_{l\neq k}^d\sum_{s_1\neq k}^d\sum_{s_2\neq s_1}^d\sum_{s_3\neq s_2,l}^d\\
			&\mbE\big[\mcA_{i_1\cdots i_d}^{(k,l)}Q_{i_ki_k}^{kk}\mcA_{i_1\cdots i_d}^{(k,s_1)}Q_{i_{s_1}i_{s_1}}^{s_1s_1}\mcA_{i_1\cdots i_d}^{(s_1,s_2)}Q_{i_{s_2}i_{s_2}}^{s_2s_2}\mcA_{i_1\cdots i_d}^{(s_2,s_3)}Q_{i_{s_3}i_{s_3}}^{s_3s_3}\mcA_{i_1\cdots i_d}^{(s_3,l)}Q_{i_li_l}^{ll}\big],
		\end{align*}
		where \(\msD\) is defined in (\ref{Eq of operator D}). The above equation contains at least three different types of diagonal terms, so that \(\min_{1\leq r\leq d}n_r\geq1\), where \(n_r\) is the number of \(\mfa_r^{(r)}\) appears in the above equation, then
		we can show that \(N^{-1/2}|\sum_{i_1\cdots i_d}^{n_1\cdots n_d}\epsilon_{i_1\cdots i_d}^{(4)}|=\mrO(\eta_0^{-5}N^{-1/2})\) by the same arguments as those for (\ref{Eq of remainder covariance general d}) later, here we omit the details for convenience.

    \vspace{5mm}
    \noindent
	As a result, we obtain
	\begin{align*}
		&zN\mfm_k=-\sum_{l\neq k}^d\Big(N\mfm_k\mfm_l+V_{kl,N}^{(d)}+\mfm_k\sum_{j\neq k,l}W_{jl,N}^{(d)}+\mfm_l\sum_{j\neq k,l}W_{jk,N}^{(d)}\Big)\\
		&+2\kappa_3H_{2,k}^{(d)}(z)-\kappa_4H_{3,k}^{(d)}(z)-n_k+\mrO(C_{\eta_0}N^{-\omega}):=-N\mfm_k\sum_{l\neq k}^d\mfm_l-n_k-M_{k,N}^{(d)}+\mrO(C_{\eta_0}N^{-\omega}),
	\end{align*}
	i.e. \(\big(z+\sum_{l\neq k}^d\mfm_l\big)N\mfm_k=-\mfc_kN-M_{k,N}^{(d)}+\mrO(C_{\eta_0}N^{-\omega})\). Let \(h_k:=N(\mfm_k-g_k)\) and $h=\sum_{k=1}^dh_k$, recall that \(g_k=-\frac{\mfc_k}{z+g-g_k}\), then
	\begin{align*}
		\Big(z+\sum_{l\neq k}^d\mfm_l\Big)h_k&=-\mfc_k N+M_{k,N}^{(d)}-\Big(z+\sum_{l\neq k}^d\mfm_l\Big)g_k+\mrO(C_{\eta_0}N^{-\omega})\\
		&=\frac{\mfc_k(h-h_k)}{z+g-g_k}-M_{k,N}^{(d)}+\mrO(C_{\eta_0}N^{-\omega})\\
		&=-g_k(h-h_k)-M_{k,N}^{(d)}+\mrO(C_{\eta_0}N^{-\omega}).
	\end{align*}
	Therefore, we obtain
	$$N\bbTheta_N^{(d)}(z,z)(\bbg(z)-\bbm(z))=-\overrightarrow{M}_N^{(d)}+\mrO(C_{\eta_0}N^{-\omega}),$$
	where \(\bbTheta_N^{(d)}(z,z)\) is defined in (\ref{Eq of bbTheta general d}) and we have shown that 
    $$\lim_{N\to\infty}\Vert\bbTheta_N^{(d)}(z,z)^{-1}+\bbPi^{(d)}(z,z)^{-1}\diag(\mfc^{-1}\circ\bbg(z))\Vert=0.$$
    Consequently, we conclude that
	$$\lim_{N\to\infty}|\mbE[\tr(\bbQ(z))]-Ng(z)-\boldsymbol{1}_d'\bbPi^{(d)}(z,z)^{-1}\diag(\mfc^{-1}\circ\bbg(z))\overrightarrow{M}_N^{(d)}(z)|=0,$$
	which completes our proof.
\end{proof}

\subsubsection{System of equations for major terms in mean function \texorpdfstring{$\mu_N^{(d)}(z)$}{ } and variance function \texorpdfstring{$\mcC_N^{(d)}(z_1,z_2)$}{ }}\label{Subsec of general Major terms}
We will extend all system equations in \S\ref{ssec of mean function} for general \(d\geq3\). The key method is to use the cumulant expansion (\ref{Eq of cumulant expansion}). In fact, we can use the same method as in Theorem \ref{Thm of entrywise law d=3} to show that only the first derivatives will generate major terms, so we only present the detailed calculation procedures of the first derivatives and omit others.

\vspace{5mm}
\noindent
{\bf System equations for \(W_{kl,N}^{(d)}(z)=\mbE[(\bba^{(k)})'\bbQ(z)\bba^{(l)}]\):} 

By the cumulant expansion (\ref{Eq of cumulant expansion}) and directly calculations, we can obtain
	\begin{align*}
		&zW_{kl,N}^{(d)}=\frac{1}{\sqrt{N}}\sum_{t\neq k}^d\sum_{i_0\cdots i_d}^{n_1\cdots n_d}\mbE\big[X_{i_1\cdots i_d}a_{i_k}^{(k)}a_{i_0}^{(l)}\mcA_{i_1\cdots i_d}^{(k,t)}Q_{i_ti_0}^{tl}\big]-\delta_{kl}\\
		&=-\frac{1}{N}\sum_{t\neq k}^d\sum_{s\neq r}^d\sum_{i_0\cdots i_d}^{n_1\cdots n_d}\mbE\big[a_{i_0}^{(l)}a_{i_k}^{(k)}\mcA_{i_1\cdots i_d}^{(t,k)}Q_{i_ti_s}^{ts}\mcA_{i_1\cdots i_d}^{(s,r)}Q_{i_ri_0}^{rl}\big]-\delta_{kl}+\mrO(C_{\eta_0}N^{-1/2})\\
		&=-\frac{1}{N}\sum_{t\neq k}^d\sum_{r\neq t}^d\sum_{i_0\cdots i_d}^{n_1\cdots n_d}\mbE\big[a_{i_0}^{(l)}a_{i_k}^{(k)}\mcA_{i_1\cdots i_d}^{(t,k)}Q_{i_ti_t}^{tt}\mcA_{i_1\cdots i_d}^{(t,r)}Q_{i_ri_0}^{rl}\big]-\delta_{kl}+\mrO(C_{\eta_0}N^{-1/2})\\
		&=-W_{kl,N}^{(d)}\sum_{t\neq k}^d\mfm_t-\sum_{t\neq k}^d\sum_{r\neq k,t}^d\mfm_tW_{rl,N}^{(d)}-\delta_{kl}+\mrO(\eta_0^{-6}N^{-\omega}),
	\end{align*}
	where we use Lemma \ref{Thm of a.s. convergence} in the last equation. According to \eqref{Eq of approximation general d}, we can further obtain
	\begin{align}
		(z+g-g_k)W_{kl,N}^{(d)}=-\sum_{t\neq k}(g-g_k-g_t)W_{tl,N}^{(d)}-\delta_{kl}+\mrO(\eta_0^{-16}N^{-\omega}).\label{Eq of bbGa 1}
	\end{align}
    Next, we will show that
    \begin{lem}\label{Lem of invertible bbGa}
        Given \(\bbg(z)\) and \(\bbS_d\) in {\rm (\ref{Eq of MDE 3 order})} and {\rm (\ref{Eq of bbS d})}, define
	   \begin{align}
		  \bbGa^{(d)}(z):=(z+g(z))\bbI_d-\diag(\bbg(z))+g(z)\bbS_d-\diag(\bbg(z))\bbS_d-\bbS_d\diag(\bbg(z)),\label{Eq of bbGa}
	   \end{align}
        then \(\bbGa^{(d)}(z)\) is invertible and \(\Vert\bbGa^{(d)}(z)^{-1}\Vert\leq\mrO(\eta_0^{-1})\) for any \(z\in\mcS_{\eta_0}\) in \eqref{Eq of stability region general d}.
    \end{lem}
    \begin{proof}
        By (\ref{Eq of bbGa 1}), we have
        \begin{align}
		  \bbGa^{(d)}(z)\bbW_N^{(d)}(z)=-\bbI_d+\mrO(\eta_0^{-16}N^{-\omega}),\label{Eq of system equation W}
	   \end{align}
        where \(\bbW_N^{(d)}(z):=[W_{kl,N}^{(d)}(z)]_{d\times d}\). For simplicity, let 
	$$\overrightarrow{W}_{d,N}:=\overrightarrow{W}_{d,N}(z):=\big((\bba^{(1)})'\bbQ^{1d}(z)\bba^{(d)},\cdots,(\bba^{(d)})'\bbQ^{dd}(z)\bba^{(d)}\big)'$$
    be the \(d\)-th column of \(\bbW_N^{(d)}(z)\). Suppose \(\bbGa^{(d)}\) is not invertible, then there exists a nonzero vector \(\bbr:=\bbr(z)=(r_1(z),\cdots,r_d(z))'\) such that \(\bbGa^{(d)}\bbr=\boldsymbol{0}_{d\times1}\), so \(\bbr'\bbGa^{(d)}=\boldsymbol{0}_{1\times d}\) due to \(\bbGa(z)\) is symmetric. Note that
	$$\bbGa^{(d)}\overrightarrow{W}_{d,N}=-\bbdel^{(d)}+\mrO(\eta_0^{-16}N^{-\omega})\boldsymbol{1}_{d\times 1}\quad\Rightarrow\quad0=\bbr'\bbGa^{(d)}\overrightarrow{W}_{d,N}=-r_d+\mrO(\eta_0^{-16}N^{-\omega}),$$
    where \(\bbdel_d\) is the \(d\)-th column of \(\bbI_d\) and \(d\) is a fixed integer, it implies that \(|r_d|=\mrO(\eta_0^{-16}N^{-\omega})\). By \(\bbGa^{(d)}\bbr=\boldsymbol{0}\), we have \(\bbGa_{k\cdot}^{(d)}\bbr=0\), where \(\bbGa_{k\cdot}^{(d)}\) is the \(k\)-th row of \(\bbGa^{(d)}\), then we have \((z+g-g_k)r_k+\sum_{l\neq k}^d(g-g_l-g_k)r_l=0\), i.e.
	\begin{align}
		(z+g_k)r_k=(g_k-g)\sum_{l=1}^dr_l+\langle\bbg,\bbr\rangle,\label{Eq of bbr}
	\end{align}
	where \(1\leq k\leq d\). In particular, when \(d=k\), since \(|r_d|=\mrO(\eta_0^{-16}N^{-\omega})\) and \(|z+g_k|\leq\mrO(\eta_0^{-1})\) for \(z\in\mcS_{\eta_0}\) in \eqref{Eq of stability region general d}, we have
    $$(g-g_d)\sum_{l=1}^dr_l=\langle\bbg,\bbr\rangle+\mrO(\eta_0^{-17}N^{-\omega}).$$
    Replacing \(\langle\bbg,\bbr\rangle=(g-g_d)\sum_{l=1}^dr_l+\mrO(\eta_0^{-17}N^{-\omega})\) in \eqref{Eq of bbr}, we have 
    \begin{align}
        (z+g_k)r_k=(g_k-g_d)\sum_{l=1}^dr_l+\mrO(\eta_0^{-17}N^{-\omega}),\quad 1\leq k\leq d.\label{Eq of bbr minor}
    \end{align}
    Summing all \(d\) above equations (\(d\) is a fixed integer), it yields that
	$$z\sum_{l=1}^dr_l+\langle\bbg,\bbr\rangle=(g-dg_d)\sum_{l=1}^dr_l+\mrO(\eta_0^{-17}N^{-\omega}),$$
	replacing \(\langle\bbg,\bbr\rangle=(g-g_d)\sum_{l=1}^dr_l+\mrO(\eta_0^{-17}N^{-\omega})\) again, we have
	$$z\sum_{l=1}^dr_l+(g-g_d)\sum_{l=1}^dr_l=(g-dg_d)\sum_{l=1}^dr_l+\mrO(\eta_0^{-17}N^{-\omega})$$
	i.e. \((z+(d-1)g_d)\sum_{l=1}^dr_l=\mrO(\eta_0^{-17}N^{-\omega})\). Since \(\Im(z+(d-1)g_d)\geq\eta_0\) for all \(z\in\mcS_{\eta_0}\), it implies that \(\sum_{l=1}^dr_l=\mrO(\eta_0^{-18}N^{-\omega})\). Since \(\Vert\bbr\Vert_2=1\) and \(d\) is a fixed integer, there exists \(1\leq k_0\leq d\) such that \(r_{k_0}\neq 0\), so \eqref{Eq of bbr minor} deduces that \(|z+g_k|=\mrO(\eta_0^{-19}N^{-\omega})\), which is a contradiction as \(N\to\infty\) since \(\Im(z+g_k)>\eta_0\) for any \(z\in\mcS_{\eta_0}\). Therefore, \(\bbGa^{(d)}\) must be invertible for all \(z\in\mcS_{\eta_0}\). Finally, by \eqref{Eq of bbGa 1} again, note that
    $$\bbGa^{(d)}(z)^{-1}(\bbI_d+\mrO(\eta_0^{-16}N^{-\omega})\boldsymbol{1}_{d\times d})=\bbW_N^{(d)}(z)\Longrightarrow\bbGa^{(d)}(z)^{-1}=\bbW_N^{(d)}(z)(\bbI_d+\mrO(\eta_0^{-16}N^{-\omega})\boldsymbol{1}_{d\times d})^{-1},$$
    where \(\bbI_d+\mrO(\eta_0^{-16}N^{-\omega})\boldsymbol{1}_{d\times d}\) is invertible for sufficiently large \(N\), so we have
    $$\Vert\bbGa^{(d)}(z)^{-1}\Vert\leq\mrO(\Vert\bbW_N^{(d)}(z)\Vert)\leq\mrO(\eta_0^{-1}),$$
    where we use the fact that \(|W_{st,N}^{(d)}(z)|=|\mbE[(\bba^{(s)})'\bbQ(z)\bba^{(t)}]|\leq\eta_0^{-1}\).
    \end{proof}
	Based on (\ref{Eq of system equation W}), there exists a \(\bbW^{(d)}(z)\) such that
	\begin{align}
		\bbW^{(d)}(z)=[W_{st}^{(d)}(z)]_{d\times d}=-\bbGa^{(d)}(z)^{-1},\quad\Vert\bbW_N^{(d)}(z)-\bbW^{(d)}(z)\Vert_{\infty}\leq\mrO(\eta_0^{-17}N^{-\omega}).\label{Eq of bbW general}
	\end{align}
    {\bf System equations for \(V_{kl,N}^{(d)}(z_1,z_2)=N^{-1}\mbE[\tr(\bbQ^{kl}(z_1)\bbQ^{lk}(z_2))]\):}
    
    Define \(\bbV_N^{(d)}(z_1,z_2)=[V_{kl,N}^{(d)}(z_1,z_2)]_{d\times d}\), where \(z_1,z_2\in\mcS_{\eta_0}\). Since
	\begin{align*}
		&z_1V_{kl,N}(z_1,z_2)=N^{-3/2}\sum_{i_1\cdots i_d}^{n_1\cdots n_d}\sum_{t\neq k}^d\mbE\big[X_{i_1\cdots i_d}\mcA_{i_1\cdots i_d}^{(k,t)}Q_{i_ti_l}^{tl}(z_1)Q_{i_ki_l}^{kl}(z_2)\big]-\delta_{kl}\mfm_k(z_2),
	\end{align*}
	by the cumulant expansion (\ref{Eq of cumulant expansion}) and Lemma \ref{Thm of a.s. convergence}, we have
	\begin{align*}
		&N^{-3/2}\sum_{i_1\cdots i_d}^{n_1\cdots n_d}\sum_{t\neq k}^d\mbE\big[\partial_{i_1\cdots i_d}^{(1)}\{\mcA_{i_1\cdots i_d}^{(k,t)}Q_{i_ti_l}^{tl}(z_1)Q_{i_ki_l}^{kl}(z_2)\}\big]\\
		&=-N^{-2}\sum_{i_1\cdots i_d}^{n_1\cdots n_d}\sum_{t\neq k}^d\sum_{s_1\neq s_2}^d\mbE\big[\mcA_{i_1\cdots i_d}^{(k,t)}Q_{i_ti_{s_1}}^{ts_1}(z_1)\mcA_{i_1\cdots i_d}^{(s_1,s_2)}Q_{i_{s_2}i_l}^{s_2l}(z_1)Q_{i_ki_l}^{kl}(z_2)\big]\\
		&-N^{-2}\sum_{i_1\cdots i_d}^{n_1\cdots n_d}\sum_{t\neq k}^d\sum_{s_1\neq s_2}^d\mbE\big[\mcA_{i_1\cdots i_d}^{(k,t)}Q_{i_ti_l}^{tl}(z_1)Q_{i_ki_{s_1}}^{ks_1}(z_2)\mcA_{i_1\cdots i_d}^{(s_1,s_2)}Q_{i_{s_2}i_l}^{s_2l}(z_2)\big]+\mrO(\eta_0^{-4}N^{-1/2})\\
		&=-N^{-2}\sum_{i_1\cdots i_d}^{n_1\cdots n_d}\sum_{t\neq k}^d\mbE\big[(\mcA_{i_1\cdots i_d}^{(k,t)})^2\big(Q_{i_ti_t}^{tt}(z_1)Q_{i_ki_l}^{kl}(z_1)Q_{i_ki_l}^{kl}(z_2)+Q_{i_ki_k}^{kk}(z_2)Q_{i_ti_l}^{tl}(z_1)Q_{i_ti_l}^{tl}(z_2)\big)\big]+\mrO(\eta_0^{-4}N^{-1/2})\\
		&=-V_{kl,N}^{(d)}(z_1,z_2)\sum_{t\neq k}^d\mfm_t(z_1)-\mfm_k(z_2)\sum_{t\neq k}^dV_{tl,N}^{(d)}(z_1,z_2)+\mrO(C_{\eta_0}N^{-\omega})\\
		&=-V_{kl,N}^{(d)}(z_1,z_2)\sum_{t\neq k}^dg_t(z_1)-g_k(z_2)\sum_{t\neq k}^dV_{tl,N}^{(d)}(z_1,z_2)+\mrO(C_{\eta_0}N^{-\omega}),
	\end{align*}
	where we use Theorem \eqref{Eq of approximation general d} in the last step. Now, we obtain that
	\begin{align*}
		&(z_1+g(z_1)-g_k(z_1))V_{kl,N}^{(d)}(z_1,z_2)=-g_k(z_2)\Big(\delta_{kl}+\sum_{t\neq k}^dV_{tl,N}^{(d)}(z_1,z_2)\Big)+\mrO(C_{\eta_0}N^{-\omega}),
	\end{align*}
	which implies that
	\begin{align*}
		&\bbPi^{(d)}(z_1,z_2)\bbV_N^{(d)}(z_1,z_2)=\diag(\mfc^{-1}\circ\bbg(z_1)\circ\bbg(z_2))+\mro(\boldsymbol{1}_{d\times d}).
	\end{align*}
	Since we have shown that \(\bbPi^{(d)}(z_1,z_2)\) is invertible in Remark \ref{Rem of invertible submatrix}, then we can derive that 
	\begin{align}
		&\bbV^{(d)}(z_1,z_2):=\lim_{N\to\infty}\bbV_N(z_1,z_2)=\bbPi^{(d)}(z_1,z_2)^{-1}\diag(\mfc^{-1}\circ\bbg(z_1)\circ\bbg(z_2)).\label{Eq of bbV}
	\end{align}
    {\bf System equations for \(\mcV_{k_1k_2,N}^{(d)}(z_1,z_2)=N^{-1}\sum_{l\neq k_1}^d\mbE[\tr(\bbQ^{k_1k_2}(\bar{z}_2)\bbQ^{k_2l}(\bar{z}_1)\bbQ^{lk_1}(z_1))]\):} 
    
    It is enough to find the limiting value of the following terms:
	\begin{align}
		V_{klr,N}^{(d)}(z_1,z_2):=\frac{1}{N}\mbE[\tr\bbQ^{kl}(z_1)\bbQ^{lr}(z_1)\bbQ^{rk}(z_2)],\label{Eq of V of mcV}
	\end{align}
	where \(z_1,z_2\in\mcS_{\eta_0}\) and \(k,l,r\in\{1,\cdots,d\}\). Similarly, for any {\bf fixed} \(r\), define
	\begin{align}
		\bbV_{r,N}^{(d)}(z_1,z_2):=[V_{klr,N}^{(d)}(z_1,z_2)]_{d\times d},
	\end{align}
	By the cumulant expansion (\ref{Eq of cumulant expansion}), since
	\begin{align*}
		&z_1V_{klr,N}^{(d)}(z_1,z_2)=\frac{z_1}{N}\sum_{i_1=1}^{n_1}\mbE[Q_{i_1\cdot}^{kl}(z_1)\bbQ^{lr}(z_1)Q_{\cdot i_1}^{rk}(z_2)]\\
		&=\frac{1}{N^{3/2}}\sum_{j\neq k}^d\sum_{i_1\cdots i_d}^{n_1\cdots n_d}\mbE[X_{i_1\cdots i_d}\mcA_{i_1\cdots i_d}^{(k,j)}Q_{i_j\cdot}^{jl}(z_1)\bbQ^{lr}(z_1)Q_{\cdot i_1}^{rk}(z_2)]-\delta_{kl}V_{kr}(z_1,z_2)\\
		&=\frac{1}{N^{3/2}}\sum_{j\neq k}^d\sum_{i_1\cdots i_d}^{n_1\cdots n_d}\mbE[\mcA_{i_1\cdots i_d}^{(k,j)}\partial_{i_1\cdots i_d}^{(1)}\{Q_{i_j\cdot}^{jl}(z_1)\bbQ^{lr}(z_1)Q_{\cdot i_1}^{rk}(z_2)\}]-\delta_{kl}V_{kr}(z_1,z_2)+\mrO(C_{\eta_0}N^{-1/2}),
	\end{align*}
	where
	\begin{align*}
		&\frac{1}{N^{3/2}}\sum_{j\neq k}^d\sum_{s,t}^{n_l,n_r}\sum_{i_1\cdots i_d}^{n_1\cdots n_d}\mbE[\mcA_{i_1\cdots i_d}^{(k,j)}\partial_{i_1\cdots i_d}^{(1)}\{Q_{i_js}^{jl}(z_1)\}Q_{st}^{lr}(z_1)Q_{ti_k}^{rk}(z_2)]\\
		&=-\frac{1}{N^2}\sum_{j\neq k}^d\sum_{s,t}^{n_l,n_r}\sum_{i_1\cdots i_d}^{n_1\cdots n_d}\sum_{p\neq q}^d\mbE[\mcA_{i_1\cdots i_d}^{(k,j)}Q_{i_ji_p}^{jp}(z_1)\mcA_{i_1\cdots i_d}^{(p,q)}Q_{i_qs}^{ql}(z_1)Q_{st}^{lr}(z_1)Q_{ti_k}^{rk}(z_2)]\\
		&=-\frac{1}{N^2}\sum_{j\neq k}^d\sum_{s,t}^{n_l,n_r}\sum_{i_1\cdots i_d}^{n_1\cdots n_d}\sum_{q\neq j}^d\mbE[\mcA_{i_1\cdots i_d}^{(k,j)}Q_{i_ji_j}^{jj}(z_1)\mcA_{i_1\cdots i_d}^{(j,q)}Q_{i_qs}^{ql}(z_1)Q_{st}^{lr}(z_1)Q_{ti_k}^{rk}(z_2)]\\
		&=-V_{klr,N}^{(d)}(z_1,z_2)\sum_{j\neq k}^d\mfm_j(z_1)+\mrO(C_{\eta_0}N^{-\omega}),
	\end{align*}
	and
	\begin{align*}
		&\frac{1}{N^{3/2}}\sum_{j\neq k}^d\sum_{s,t}^{n_l,n_r}\sum_{i_1\cdots i_d}^{n_1\cdots n_d}\mbE[\mcA_{i_1\cdots i_d}^{(k,j)}Q_{i_js}^{jl}(z_1)\partial_{i_1\cdots i_d}^{(1)}\{Q_{st}^{lr}(z_1)\}Q_{ti_k}^{rk}(z_2)]\\
		&=-\frac{1}{N^2}\sum_{j\neq k}^d\sum_{s,t}^{n_l,n_r}\sum_{i_1\cdots i_d}^{n_1\cdots n_d}\sum_{p\neq q}^d\mbE[\mcA_{i_1\cdots i_d}^{(k,j)}Q_{i_js}^{jl}(z_1)Q_{si_p}^{lp}(z_1)\mcA_{i_1\cdots i_d}^{(p,q)}Q_{i_qt}^{qr}(z_1)Q_{ti_k}^{rk}(z_2)]\\
		&=-V_{kr,N}^{(d)}(z_1,z_2)\sum_{j\neq k}^dV_{jl,N}^{(d)}(z_1,z_1)+\mrO(C_{\eta_0}N^{-\omega}),\notag
	\end{align*}
	and
	\begin{align*}
		&\frac{1}{N^{3/2}}\sum_{j\neq k}^d\sum_{s,t}^{n_l,n_r}\sum_{i_1\cdots i_d}^{n_1\cdots n_d}\mbE[\mcA_{i_1\cdots i_d}^{(k,j)}Q_{i_js}^{jl}(z_1)Q_{st}^{lr}(z_1)\partial_{i_1\cdots i_d}^{(1)}\{Q_{ti_k}^{rk}(z_2)\}]\\
		&=-\frac{1}{N^2}\sum_{j\neq k}^d\sum_{s,t}^{n_l,n_r}\sum_{i_1\cdots i_d}^{n_1\cdots n_d}\sum_{p\neq q}^d\mbE[\mcA_{i_1\cdots i_d}^{(k,j)}Q_{i_js}^{jl}(z_1)Q_{st}^{lr}(z_1)Q_{ti_p}^{rp}(z_2)\mcA_{i_1\cdots i_d}^{(p,q)}Q_{i_qi_k}^{qk}(z_2)]\\
		&=-\mfm_k(z_2)\sum_{j\neq k}^dV_{jlr,N}^{(d)}(z_1,z_2)+\mrO(C_{\eta_0}N^{-\omega}).
	\end{align*}
	In summary, we can conclude that
	\begin{small}
	\begin{align*}
		&V_{klr,N}^{(d)}(z_1,z_2)=\mfc_k^{-1}g_k(z_1)\Big(\delta_{kl}V_{kr}^{(d)}(z_1,z_2)+V_{kr}^{(d)}(z_1,z_2)\sum_{j\neq k}^dV_{jl}^{(d)}(z_1,z_2)+g_k(z_2)\sum_{j\neq k}^dV_{jlr,N}^{(d)}(z_1,z_2)\Big)+\mrO(C_{\eta_0}N^{-\omega}).
	\end{align*}
	\end{small}\noindent
	i.e.
    \begin{small}
    \begin{align*}
		&\lim_{N\to\infty}\Vert\bbV_{r,N}^{(d)}(z_1,z_2)-\bbPi^{(d)}(z_1,z_2)^{-1}\diag(\mfc^{-1}\circ\bbg(z_1))\diag(\bbV_{\cdot r}^{(d)}(z_1,z_2))\big(\bbI_d+\bbS_d\bbV^{(d)}(z_1,z_2)\big)\Vert=0,
	\end{align*}
    \end{small}\noindent
	where \(\bbV_{\cdot r}^{(d)}(z_1,z_2)\) is the \(r\)-th column of \(\bbV^{(d)}(z_1,z_2)\). Thus, the limit value \(\bbV_{r,N}^{(d)}(z_1,z_2)\) is given as
	\begin{small}
	\begin{align}
		\bbV_r^{(d)}(z_1,z_2):&=\bbPi^{(d)}(z_1,z_2)^{-1}\diag(\mfc^{-1}\circ\bbg(z_1))\diag(\bbV_{\cdot r}^{(d)}(z_1,z_2))\big(\bbI_d+\bbS_d\bbV^{(d)}(z_1,z_2)\big),\label{Eq of bbV r}
	\end{align}
	\end{small}\noindent
	Once we solve \(V_{klr,N}^{(d)}(z_1,z_2)\), by (\ref{Eq of mcV}), the limiting expression of \(\mcV_{k_1k_2,N}^{(d)}(z_1,z_2)\) is given as
	\begin{align}
		\lim_{N\to\infty}\mcV_{k_1k_2,N}^{(d)}(z_1,z_2)=\mcV_{k_1k_2}^{(d)}(z_1,z_2):=\sum_{l\neq k_1}^dV_{k_1k_2l}^{(d)}(z_1,z_2).\label{Eq of mcV limit general d}
	\end{align}
    {\bf System equations for \(\widetilde{\mcW}_{k_1k_2,N}^{(d)}(z_1,z_2)\) in (\ref{Eq of mcW}):} By Theorem \ref{Thm of a.s. convergence}, we can show that
	$$\Cov(N^{-1}\tr(\bbQ^{k_1k_1}(z_1)\circ\bbQ^{k_1k_1}(\bar{z}_2)),N^{-1}\tr(\bbQ^{k_1k_1}(z_1)\circ(\bbQ^{k_1k_2}(\bar{z}_2)\bbQ^{k_2k_1}(\bar{z}_2))))=\mrO(C_{\eta_0}N^{-2\omega}).$$
	Consequently, by (\ref{Eq of mcW}), we only need to compute the limiting values of \(N^{-1}\mbE[\tr(\bbQ^{k_1k_1}(z_1)\circ\bbQ^{k_1k_1}(\bar{z}_2))]\) and \(N^{-1}\mbE[\tr(\bbQ^{k_1k_1}(z_1)\circ(\bbQ^{k_1k_2}(\bar{z}_2)\bbQ^{k_2k_1}(\bar{z}_2)))]\) respectively. For the first term, by Theorem \ref{Thm of entrywise law general d}, we can obtain
	\begin{align}
		\lim_{N\to\infty}N^{-1}\mbE[\tr(\bbQ^{k_1k_1}(z_1)\circ\bbQ^{k_1k_1}(\bar{z}_2))]=\mfc_{k_1}^{-1}g_{k_1}(z_1)g_{k_1}(\bar{z}_2).\label{Eq of V circ}
	\end{align}
	For the second term, we define 
    $$\mathring{V}_{kl,N}^{(d)}(z_1,z_2):=\frac{1}{N}\mbE[\tr(\bbQ^{kk}(z_1)\circ(\bbQ^{kl}(z_2)\bbQ^{lk}(z_2)))],$$
    by the cumulant expansion (\ref{Eq of cumulant expansion}) again, we have
	\begin{align*}
		&z_1\mathring{V}_{kl,N}^{(d)}(z_1,z_2)=\frac{z_1}{N}\sum_{i_k=1}^{n_k}\mbE[Q_{i_ki_k}^{kk}(z_1)Q_{i_k\cdot}^{kl}(z_2)Q_{\cdot i_k}^{lk}(z_2)]\\
		&=\frac{1}{N^{3/2}}\sum_{t\neq k}^d\sum_{i_1\cdots i_d}^{n_1\cdots n_d}\mbE\big[X_{i_1\cdots i_d}\mcA_{i_1\cdots i_d}^{(k,t)}Q_{i_ti_k}^{tk}(z_1)Q_{i_k\cdot}^{kl}(z_2)Q_{\cdot i_k}^{lk}(z_2)\big]-V_{kl,N}^{(d)}(z_2,z_2)\\
		&=\frac{1}{N^{3/2}}\sum_{t\neq k}^d\sum_{i_1\cdots i_d}^{n_1\cdots n_d}\mbE\big[\mcA_{i_1\cdots i_d}^{(k,t)}\partial_{i_1\cdots i_d}^{(1)}\{Q_{i_ti_k}^{tk}(z_1)Q_{i_k\cdot}^{kl}(z_2)Q_{\cdot i_k}^{lk}(z_2)\}\big]-V_{kl,N}^{(d)}(z_2,z_2)
	\end{align*}
	where
	\begin{align*}
		&\frac{1}{N^{3/2}}\sum_{t\neq k}^d\sum_{i_1\cdots i_d}^{n_1\cdots n_d}\sum_{s=1}^{n_l}\mbE\big[\mcA_{i_1\cdots i_d}^{(k,t)}\partial_{i_1\cdots i_d}^{(1)}\{Q_{i_ti_k}^{tk}(z_1)\}Q_{i_ks}^{kl}(z_2)Q_{si_k}^{lk}(z_2)\big]\\
		&=-\frac{1}{N^2}\sum_{t\neq k}^d\sum_{i_1\cdots i_d}^{n_1\cdots n_d}\sum_{s=1}^{n_l}\sum_{p\neq q}^d\mbE\big[\mcA_{i_1\cdots i_d}^{(k,t)}Q_{i_ti_p}^{tp}(z_1)\mcA_{i_1\cdots i_d}^{(p,q)}Q_{i_qi_k}^{qk}(z_1)Q_{i_ks}^{kl}(z_2)Q_{si_k}^{lk}(z_2)\big]\\
		&=-\mathring{V}_{kl,N}^{(d)}(z_1,z_2)\sum_{t\neq k}^d\mfm_t(z_1)+\mrO(C_{\eta_0}N^{-\omega})
	\end{align*}
	and
	\begin{align*}
		&\frac{1}{N^{3/2}}\sum_{t\neq k}^d\sum_{i_1\cdots i_d}^{n_1\cdots n_d}\sum_{s=1}^{n_l}\mbE\big[\mcA_{i_1\cdots i_d}^{(k,t)}Q_{i_ti_k}^{tk}(z_1)\partial_{i_1\cdots i_d}^{(1)}\{Q_{i_ks}^{kl}(z_2)Q_{si_k}^{lk}(z_2)\}\big]\\
		&=-\frac{2}{N^2}\sum_{t\neq k}^d\sum_{p\neq q}^d\sum_{i_1\cdots i_d}^{n_1\cdots n_d}\sum_{s=1}^{n_l}\mbE\big[\mcA_{i_1\cdots i_d}^{(k,t)}Q_{i_ti_k}^{tk}(z_1)Q_{i_ks}^{kl}(z_2)Q_{i_ki_p}^{kp}(z_2)\mcA_{i_1\cdots i_d}^{(p,q)}Q_{i_qs}^{ql}(z_2)\big]=\mrO(C_{\eta_0}N^{-\omega}).
	\end{align*}
	Hence, we obtain that
	$$(z_1+\mfm(z_1)-\mfm_k(z_1))\mathring{V}_{kl,N}^{(d)}(z_1,z_2)=V_{kl,N}^{(d)}(z_2,z_2)+\mrO(C_{\eta_0}N^{-\omega}),$$
	in matrix notations
	\begin{align*}
		\mathring{\bbV}_N^{(d)}(z_1,z_2)=\diag(\mfc^{-1}\circ\bbg(z_1))\bbV^{(d)}(z_2,z_2)+\mro(\boldsymbol{1}_{d\times d}),
	\end{align*}
	where \(\mathring{\bbV}_N^{(d)}(z_1,z_2)=[\mathring{V}_{kl,N}^{(d)}(z_1,z_2)]_{d\times d}\). So it concludes that
	\begin{align}
		\mathring{\bbV}^{(d)}(z_1,z_2):=[\mathring{V}_{st}^{(d)}(z_1,z_2)]_{d\times d}:=\lim_{N\to\infty}\mathring{\bbV}_N^{(d)}(z_1,z_2)=\diag(\mfc^{-1}\circ\bbg(z_1))\bbV^{(d)}(z_2,z_2).\label{Eq of V mathring}
	\end{align}
	Now, by (\ref{Eq of mcW}), we obtain that
	\begin{align}
		&\widetilde{\mcW}_{k_1k_2,N}^{(3)}(z_1,z_2)=\mcW_{k_1k_2,N}^{(d)}(z_1,z_2)+\mrO(C_{\eta_0}N^{-\omega})\notag,
	\end{align}
    where
    \begin{small}
    \begin{align}
        \mcW_{k_1k_2,N}^{(d)}(z_1,z_2):=\mfc_{k_1}^{-1}g_{k_1}(z_1)g_{k_1}(\bar{z}_2)\sum_{l\neq k_1}^d\mcB_{(4)}^{(k_1,l)}\mathring{V}_{lk_2}^{(d)}(z_1,z_2)+\mathring{V}_{k_1k_2}^{(d)}(z_1,z_2)\sum_{l\neq k_1}^d\mcB_{(4)}^{(k_1,l)}\mfc_l^{-1}g_l(z_1)g_l(\bar{z}_2).\label{Eq of mcW limit general d}
    \end{align}
    \end{small}
\subsection{CLT for the LSS}
Consider the following family of functions: 
\begin{align}
    \mathfrak{F}_d:=\{f(z):f\text{ is analytic on an open set containing the interval }[-\max\{\zeta,\mfv_d\},\max\{\zeta,\mfv_d\}]\},\label{Eq of analytic function general d}
\end{align}
where $\zeta$ \eqref{Eq of support boundary} is the boundary of LSD $\nu$ and $\mfv_d$ is defined in Theorem \ref{Thm of Extreme eigenvalue N d=3}. For any \(f\in\mathfrak{F}_d\), the LSS of \(\bbM\) is given as
\begin{align}
	\mcL_{\bbM}(f):=\frac{1}{N}\sum_{l=1}^Nf(\lambda_l).\notag
\end{align}
where \(\lambda_1\geq\cdots\geq\lambda_N\) be are eigenvalues of \(\bbM\). Similar to Theorem \ref{Thm of CLT LSS d=3}, we will establish the CLT of
\begin{align}
	G_N(f):=N\int_{-\infty}^{\infty}f(x)(\nu_N(dx)-\nu(dx))=N\Big(\mcL_{\bbM}(f)-\int_{-\infty}^{\infty}f(x)\nu(dx)\Big),\label{Eq of GNf general d}
\end{align}
where \(\nu_N\) and \(\nu\) are the ESD and LSD of \(\bbM\) respectively. Precisely, we will show that
\begin{thm}\label{Thm of general CLT LSS}
	Under Assumptions {\rm \ref{Ap of general noise}} and {\rm \ref{Ap of dimension}}, let \(\mfC_{1}\) and \(\mfC_{2}\) be two disjoint rectangular contours with vertices of \(\pm E_{1}\pm\eta_{1}\) and \(\pm E_{2}\pm\eta_{2}\) such that \(E_1,E_2\geq\max\{\zeta,\mfv_d\}+t\) for any \(t>0\), where \(\zeta\) and \(\mfv_d\) are defined in {\rm (\ref{Eq of support boundary})} and {\rm (\ref{Eq of Extreme eigenvalues d=3})}, then we have
    $$(G_N(f)-\xi_N^{(d)})/\sigma_N^{(d)}\overset{d}{\longrightarrow}\mcN(0,1)$$
    where
	\begin{align*}
		\xi_N^{(d)}&:=-\frac{1}{2\pi{\rm i}}\oint_{\mfC_1}f(z)\mu_N^{(d)}(z)dz,\\
        (\sigma_N^{(d)})^2&:=-\frac{1}{4\pi^2}\oint_{\mfC_1}\oint_{\mfC_2}f(z_1)f(z_2)\mcC_N^{(d)}(z_1,z_2)dz_1dz_2.
	\end{align*}
    and the mean function \(\mu_N^{(d)}(z)\) and the covariance function \(\mcC_N^{(d)}(z_1,z_2)\) are defined in {\rm (\ref{Eq of mean function general d})} and {\rm (\ref{Eq of covariance function general d})}, respectively.
\end{thm}
The basic outlines are the same as those in \S\ref{Sec of CLT}. 
\subsubsection{Tightness}
\begin{thm}\label{Thm of Tightness general d}
	Under Assumptions {\rm \ref{Ap of general noise}} and {\rm \ref{Ap of dimension}}, \(\tr(\boldsymbol{Q}(z))-\mathbb{E}[\tr(\boldsymbol{Q}(z))]\) is a tight sequence in $\mathcal{S}_{\eta_0}$ in {\rm (\ref{Eq of stability region general d})}, i.e. 
	$$\sup_{\substack{z_1,z_2\in\mathcal{S}_{\eta_0}\\z_1\neq z_2}}\frac{\mathbb{E}\left[|\tr(\boldsymbol{Q}(z_1)-\boldsymbol{Q}(z_2))-\mathbb{E}[\tr(\boldsymbol{Q}(z_1)-\boldsymbol{Q}(z_2))|^2\right]}{|z_1-z_2|^2}<C_{\eta_0}.$$
\end{thm}
\begin{proof}
	For any \(z\in\mcS_{\eta_0}\), the tightness of the process \(\tr(\bbQ(z))-\mbE[\tr(\bbQ(z))]\) is equivalent to
	\begin{align*}
		&\Var\big(\tr(\bbQ^{kl}(z_1)\bbQ^{lk}(z_2))\big)\leq C_{\eta_0,d,\mfc},
	\end{align*}
	where \(z_1,z_2\in\mcS_{\eta_0}\) and \(k,l\in\{1,\cdots,d\}\). Define 
	\begin{align}
		\mcC_{k_1l_1,k_2l_2,N}^{(d)}(z_1,z_2)=\Cov\big(\tr(\bbQ^{k_1l_1}(z_1)\bbQ^{l_1k_1}(z_2)),\tr(\bbQ^{k_2l_2}(z_1)\bbQ^{l_2k_2}(z_2))\big),
	\end{align}
	it is enough to show that \(|\mcC_{k_1l_1,k_2l_2,N}^{(d)}(z_1,z_2)|\leq C_{\eta_0,d,\mfc}\) for any \(k_1,k_2,l_1,l_2\in\{1,\cdots,d\}\). Similar to what we have done in \S\ref{Sec of Tightness}, let us  derive a system equation for all \(\mcC_{k_1l_1,k_2l_2,N}^{(d)}(z_1,z_2)\). We omit \((z_1,z_2)\) behind \(\mcC_{k_1l_1,k_2l_2,N}^{(d)}(z_1,z_2)\), then
	\begin{align}
		&z_1\mcC_{k_1l_1,k_2l_2,N}^{(d)}=\frac{1}{\sqrt{N}}\sum_{i_1\cdots i_d}^{n_1\cdots n_d}\sum_{s\neq k_1}^d\mbE\big[X_{i_1\cdots i_d}\mcA_{i_1\cdots i_d}^{(k_1,s)}Q_{i_s\cdot}^{sl_1}(z_1)Q_{\cdot i_k}^{l_1k_1}(z_2)\tr(\bbQ^{k_2l_2}(\bar{z}_1)\bbQ^{l_2k_2}(\bar{z}_2))^c\big]\label{Eq of tightness major 1 general d}\\
		&-\delta_{k_1l_1}\Cov\big(\tr(\bbQ^{k_1k_1}(z_1)),\tr(\bbQ^{k_2l_2}(z_1)\bbQ^{l_2k_2}(z_2))\big),\label{Eq of tightness major 2 general d}
	\end{align}
	and we only need to show both of above two terms are bounded by \(C_{\eta_0}\). 

    \vspace{5mm}
    \noindent
    {\bf Calculations of (\ref{Eq of tightness major 2 general d}):} Define
	\begin{align}
		\mcC_{k_1l_1,k_2,N}^{(d)}(z_1,z_2):=\Cov\big(\tr(\bbQ^{k_1l_1}(z_1)\bbQ^{l_1k_1}(z_2),\tr(\bbQ^{k_2k_2}(z_1)))\big),\label{Eq of mcC 2 1 general d}
	\end{align}
	here, we still omit the \((z_1,z_2)\) behind \(\mcC_{k_1l_1,k_2,N}^{(d)}(z_1,z_2)\). By the cumulant expansion (\ref{Eq of cumulant expansion}), we have
	\begin{align*}
		&z_1\mcC_{k_1l_1,k_2,N}^{(d)}=\frac{1}{\sqrt{N}}\sum_{i_1\cdots i_d}^{n_1\cdots n_d}\sum_{s\neq k_1}^d\mbE\big[X_{i_1\cdots i_d}\mcA_{i_1\cdots i_d}^{(k_1,s)}Q_{i_s\cdot}^{sl_1}(z_1)Q_{\cdot i_{k_1}}^{l_1k_1}(z_2)\tr(\bbQ^{k_2k_2}(\bar{z}_1))^c\big]\\
		&=\frac{1}{\sqrt{N}}\sum_{i_1\cdots i_d}^{n_1\cdots n_d}\Big(\sum_{\alpha=0}^3\sum_{s\neq k}^d\mbE\big[\partial_{i_1\cdots i_d}^{(\alpha)}\{\mcA_{i_1\cdots i_d}^{(k_1,s)}Q_{i_s\cdot}^{sl_1}(z_1)Q_{\cdot i_{k_1}}^{l_1k_1}(z_2)\tr(\bbQ^{k_2k_2}(\bar{z}_1))^c\}\big]+\epsilon_{i_1\cdots i_d}^{(4)}\Big).
	\end{align*} 
    For convenience, we omit the detailed calculation for minor terms since the proofs are the same as those in Theorems \ref{Thm of Tightness}, \ref{Thm of mean general d} and \ref{Thm of covariance general d}. Actually, only \(\alpha=1,3\) will have the major terms:

    \vspace{5mm}
    \noindent
    {\bf First derivatives:} When \(\alpha=1\), by (\ref{Eq of partial trace d}), we have 
    \begin{align*}
        \partial_{i_1\cdots i_d}^{(1)}\tr(\bbQ^{k_2k_2}(\bar{z}))=-N^{-1/2}\sum_{t_1\neq t_2}^{d,d}\mcA_{i_1\cdots i_d}^{(t_1,t_2)}Q_{i_{t_1}\cdot}^{t_1k_2}(\bar{z})Q_{\cdot i_{t_2}}^{k_2t_2}(\bar{z}),
    \end{align*}
		and by Lemma \ref{Rem of extension minor 2},
		\begin{align*}
			&\frac{1}{\sqrt{N}}\sum_{i_1\cdots i_d}^{n_1\cdots n_d}\sum_{i_t}^{n_{l_1}}\mcA_{i_1\cdots i_d}^{(k_1,s)}\partial_{i_1\cdots i_d}^{(1)}\{Q_{i_si_t}^{sl_1}(z_1)Q_{i_{k_1}i_t}^{k_1l_1}(z_2)\}\\
			&=-N^{-1}\sum_{i_1\cdots i_d}^{n_1\cdots n_d}\mcA_{i_1\cdots i_d}^{(k_1,s)}\big[Q_{i_si_{s_1}}^{ss_1}(z_1)\mcA_{i_1\cdots i_d}^{(s_1,s_2)}Q_{i_{s_2}\cdot}^{s_2l_1}(z_1)Q_{\cdot i_{k_1}}^{l_1k_1}(z_2)+Q_{i_{k_1}i_{s_1}}^{k_1s_1}(z_2)\mcA_{i_1\cdots i_d}^{(s_1,s_2)}Q_{i_{s_2}\cdot}^{s_2l_1}(z_2)Q_{\cdot i_s}^{l_1s}(z_1)\big]\\
			&=-N^{-1}\sum_{i_1\cdots i_d}^{n_1\cdots n_d}(\mcA_{i_1\cdots i_d}^{(k_1,s)})^2[Q_{i_si_s}^{ss}(z_1)Q_{i_{k_1}\cdot}^{k_1l_1}(z_1)Q_{\cdot i_{k_1}}^{l_1k_1}(z_2)+Q_{i_{k_1}i_{k_1}}^{k_1k_1}(z_2)Q_{i_s\cdot}^{sl_1}(z_2)Q_{\cdot i_s}^{l_1s}(z_1)]+\mrO(\eta_0^{-3}N^{-1/2})\\
			&=-N^{-1}[\tr(\bbQ^{ss}(z_1))\tr(\bbQ^{k_1l_1}(z_1)\bbQ^{l_1k_1}(z_2))+\tr(\bbQ^{k_1k_1}(z_2))\tr(\bbQ^{sl_1}(z_1)\bbQ^{l_1s}(z_2))]+\mrO(\eta_0^{-3}N^{-1/2}),
		\end{align*}
		then by direct calculation, we can obtain that
		\begin{align*}
			&\frac{1}{\sqrt{N}}\sum_{i_1\cdots i_d}^{n_1\cdots n_d}\sum_{s\neq k_1}^d\mbE\big[\mcA_{i_1\cdots i_d}^{(k_1,s)}Q_{i_s\cdot}^{sl_1}(z_1)Q_{\cdot i_{k_1}}^{l_1k_1}(z_2)\partial_{i_1\cdots i_d}^{(1)}\{\tr(\bbQ^{k_2k_2}(\bar{z}_1))\}\big]\\
			&=-\frac{1}{N}\sum_{i_1\cdots i_d}^{n_1\cdots n_d}\sum_{s\neq k_1}^d\sum_{t_1\neq t_2}^d\mbE\big[\mcA_{i_1\cdots i_d}^{(k_1,s)}Q_{i_s\cdot}^{sl_1}(z_1)Q_{\cdot i_{k_1}}^{l_1k_1}(z_2)\mcA_{i_1\cdots i_d}^{(t_1,t_2)}Q_{i_{t_1}\cdot}^{t_1k_2}(\bar{z}_1)Q_{\cdot i_{t_2}}^{k_2t_2}(\bar{z}_1)\big]\\
			&=-\frac{2}{N}\sum_{s\neq k_1}^d\mbE[\tr(\bbQ^{sl_1}(z_1)\bbQ^{l_1k_1}(z_2)\bbQ^{k_1k_2}(\bar{z}_1)\bbQ^{k_2s}(\bar{z}_1))]+\mrO(C_{\eta_0}N^{-1/2}),
		\end{align*}
		and
		\begin{align*}
			&\frac{1}{\sqrt{N}}\sum_{i_1\cdots i_d}^{n_1\cdots n_d}\sum_{s\neq k_1}^d\mbE\big[\partial_{i_1\cdots i_d}^{(1)}\{\mcA_{i_1\cdots i_d}^{(k_1,s)}Q_{i_s\cdot}^{sl_1}(z_1)Q_{\cdot i_{k_1}}^{l_1k_1}(z_2)\}\tr(\bbQ^{k_2k_2}(\bar{z}_1))^c\big]\\
			&=-N^{-1}\sum_{s\neq k_1}^d\Cov(\tr(\bbQ^{ss}(z_1))\tr(\bbQ^{k_1l_1}(z_1)\bbQ^{l_1k_1}(z_2)),\tr(\bbQ^{k_2k_2}(z_1)))\\
			&-N^{-1}\sum_{s\neq k_1}^d\Cov(\tr(\bbQ^{k_1k_1}(z_2))\tr(\bbQ^{sl_1}(z_1)\bbQ^{l_1s}(z_2)),\tr(\bbQ^{k_2k_2}(z_1)))+\mrO(C_{\eta_0}N^{-\omega}),
		\end{align*}
		similar to (\ref{Eq of covariance trick}), we can show that
		\begin{align*}
			&N^{-1}\Cov(\tr(\bbQ^{ss}(z_1))\tr(\bbQ^{k_1l_1}(z_1)\bbQ^{l_1k_1}(z_2)),\tr(\bbQ^{k_2k_2}(z_1)))\\
			&=\mfm_s(z_1)\mcC_{k_1l_1,k_2,N}^{(d)}+V_{k_1l_1,N}^{(d)}(z_1,z_2)\mcC_{sk_2,N}^{(d)}(z_1,z_1)+\mrO(C_{\eta_0}N^{-\omega}),
		\end{align*}
		where $V_{k_1l_1,N}^{(d)}(z_1,z_2)=N^{-1}\mbE[\tr(\bbQ^{k_1l_1}(z_1)\bbQ^{l_1k_1}(z_2))]$ and \(\mcC_{sk_2,N}^{(d)}(z_1,z_2)\) is defined in (\ref{Eq of mcC covariance}). For simplicity, we define
        $$\mcV_{k_1l_1,k_2,N}^{(d)}(z_1,z_2):=\frac{1}{N}\sum_{s\neq k_1}^d\mbE[\tr(\bbQ^{sl_1}(z_1)\bbQ^{l_1k_1}(z_2)\bbQ^{k_1k_2}(\bar{z}_1)\bbQ^{k_2s}(\bar{z}_1))]+\mrO(C_{\eta_0}N^{-1/2}),$$
        then we obtain
		\begin{align}
			&\frac{1}{\sqrt{N}}\sum_{i_1\cdots i_d}^{n_1\cdots n_d}\sum_{s\neq k_1}^d\mbE\big[\partial_{i_1\cdots i_d}^{(1)}\{\mcA_{i_1\cdots i_d}^{(k_1,s)}Q_{i_s\cdot}^{sl_1}(z_1)Q_{\cdot i_{k_1}}^{l_1k_1}(z_2)\tr(\bbQ^{k_2k_2}(\bar{z}_1))^c\}\big]\notag\\
			&=-\mcC_{k_1l_1,k_2,N}^{(d)}\sum_{s\neq k_1}^d\mfm_s(z_1)-V_{k_1l_1,N}^{(d)}(z_1,z_2)\sum_{s\neq k_1}^d\mcC_{sk_2,N}^{(d)}(z_1,z_2)-\mfm_{k_1}(z_2)\sum_{s\neq k_1}^d\mcC_{sl_1,k_2,N}^{(d)}\notag\\
			&-\mcC_{k_1k_2,N}^{(d)}(z_2,z_1)\sum_{s\neq k_1}^dV_{sl_1,N}^{(d)}(z_1,z_2)-2\mcV_{k_1l_1,k_2,N}^{(d)}(z_1,z_2)+\mrO(C_{\eta_0}N^{-\omega}).\label{Eq of tightness C21 1}
		\end{align}
    {\bf Third derivatives:} When \(\alpha=3\), similar to proofs in Theorem \ref{Thm of Tightness} for \(\alpha=3\), only the following case contains the major terms:
		\begin{align*}
			&\frac{1}{\sqrt{N}}\sum_{i_1\cdots i_d}^{n_1\cdots n_d}\sum_{s\neq k_1}^d\mbE\big[\partial_{i_1\cdots i_d}^{(2)}\{\tr(\bbQ^{k_2k_2}(\bar{z}_1))^c\}\partial_{i_1\cdots i_d}^{(1)}\{\mcA_{i_1\cdots i_d}^{(k_1,s)}Q_{i_s\cdot}^{sl_1}(z_1)Q_{\cdot i_{k_1}}^{l_1k_1}(z_2)\}\big],
		\end{align*}
		where
		\begin{align*}
			\partial_{i_1\cdots i_d}^{(2)}\tr\bbQ^{k_2k_2}(\bar{z})=\frac{2}{N}\sum_{t_1\neq t_2,t_3\neq t_4}^d\mcA_{i_1\cdots i_d}^{(t_1,t_2)}\mcA_{i_1\cdots i_d}^{(t_3,t_4)}Q_{i_{t_2}i_{t_3}}^{t_2t_3}(\bar{z})Q_{i_{t_4}\cdot}^{t_4k_2}(\bar{z})Q_{\cdot i_{t_1}}^{k_2t_1}(\bar{z}),
		\end{align*}
		then by Lemmas \ref{Rem of extension minor 2}, \ref{Rem of extension of Corollary} and (\ref{Eq of mcB}), we have
		\begin{align}
			&\frac{1}{\sqrt{N}}\sum_{i_1\cdots i_d}^{n_1\cdots n_d}\sum_{s\neq k_1}^d\mbE\big[\partial_{i_1\cdots i_d}^{(2)}\{\tr(\bbQ^{k_2k_2}(\bar{z}_1))^c\}\partial_{i_1\cdots i_d}^{(1)}\{\mcA_{i_1\cdots i_d}^{(k_1,s)}Q_{i_s\cdot}^{sl_1}(z_1)Q_{\cdot i_{k_1}}^{l_1k_1}(z_2)\}\big]=\mrO(C_{\eta_0}N^{-1/2})\notag\\
			&-\frac{2}{N^2}\sum_{i_1\cdots i_d}^{n_1\cdots n_d}\sum_{s\neq k_1}^d\sum_{t_1,t_2}^{(s,k_1)}\mbE\big[(\mcA_{i_1\cdots i_d}^{(t_1,t_2)})^2Q_{i_{t_2}i_{t_2}}^{t_2t_2}(\bar{z}_1)Q_{i_{t_1}\cdot}^{t_1k_2}(\bar{z}_1)Q_{\cdot i_{t_1}}^{k_2t_1}(\bar{z}_1)(\mcA_{i_1\cdots i_d}^{(k_1,s)})^2Q_{i_si_s}^{ss}(z_1)Q_{i_{k_1}\cdot}^{k_1l_1}(z_1)Q_{\cdot i_{k_1}}^{l_1k_1}(z_2)\big]\notag\\
			&-\frac{2}{N^2}\sum_{i_1\cdots i_d}^{n_1\cdots n_d}\sum_{s\neq k_1}^d\sum_{t_1,t_2}^{(s,k_1)}\mbE\big[(\mcA_{i_1\cdots i_d}^{(t_1,t_2)})^2Q_{i_{t_2}i_{t_2}}^{t_2t_2}(\bar{z}_1)Q_{i_{t_1}\cdot}^{t_1k_2}(\bar{z}_1)Q_{\cdot i_{t_1}}^{k_2t_1}(\bar{z}_1)(\mcA_{i_1\cdots i_d}^{(k_1,s)})^2Q_{i_s\cdot}^{sl_1}(z_1)Q_{\cdot i_s}^{l_1s}(z_2)Q_{i_{k_1}i_{k_1}}^{k_1k_1}(z_2)\big]\notag
		\end{align}
    where the notation $\sum_{t_1,t_2}^{(s,k_1)}$ means that the summation of $t_1$ and $t_2$ are over $\{1,\cdots,d\}\backslash\{s,k_1\}$. For simplicity, we define 
    \begin{align*}
        &\mcW_{k_1l_1,k_2,N}^{(d)}(z_1,z_2)\label{Eq of tightness C21 3}\\
        :&=\frac{1}{N^2}\sum_{s\neq k_1}^d\mcB_{(4)}^{(k_1,s)}\mbE\big[\tr(\bbQ^{ss}(z_1)\circ\bbQ^{ss}(\bar{z}_1))\cdot\tr((\bbQ^{k_1l_1}(z_1)\bbQ^{l_1k_1}(z_2))\circ(\bbQ^{k_1k_2}(\bar{z}_1)\bbQ^{k_2k_1}(\bar{z}_1)))\big]\notag\\
			&+\frac{1}{N^2}\sum_{s\neq k_1}^d\mcB_{(4)}^{(k_1,s)}\mbE\big[\tr(\bbQ^{ss}(z_1)\circ(\bbQ^{sk_2}(\bar{z}_1)\bbQ^{k_2s}(\bar{z}_1)))\cdot\tr((\bbQ^{k_1l_1}(z_1)\bbQ^{l_1k_1}(z_2))\circ\bbQ^{k_1k_1}(\bar{z}_1))\big]\notag\\
			&+\frac{1}{N^2}\sum_{s\neq k_1}^d\mcB_{(4)}^{(k_1,s)}\mbE\big[\tr((\bbQ^{sl_1}(z_1)\bbQ^{l_1s}(z_2))\circ\bbQ^{ss}(\bar{z}_1))\cdot\tr(\bbQ^{k_1k_1}(z_1)\circ(\bbQ^{k_1k_2}(\bar{z}_1)\bbQ^{k_2k_1}(\bar{z}_1)))\big]\notag\\
			&+\frac{1}{N^2}\sum_{s\neq k_1}^d\mcB_{(4)}^{(k_1,s)}\mbE\big[\tr((\bbQ^{sl_1}(z_1)\bbQ^{l_1s}(z_1))\circ(\bbQ^{sk_2}(\bar{z}_1)\bbQ^{k_2s}(\bar{z}_1)))\cdot\tr(\bbQ^{k_1k_1}(z_2)\circ\bbQ^{k_1k_1}(\bar{z}_1))\big]\notag,
    \end{align*}
    where $\mcB_{(4)}^{(k_1,s)}$ is defined in \eqref{Eq of mcB}. Then combining (\ref{Eq of tightness C21 1}) and (\ref{Eq of tightness C21 3}), we obtain
	\begin{align*}
		&(z_1+\mfm(z_1)-\mfm_{k_1}(z_1))\mcC_{k_1l_1,k_2,N}^{(d)}=-V_{k_1l_1,N}^{(d)}(z_1,z_2)\sum_{s\neq k_1}^d\mcC_{sk_2,N}^{(d)}(z_1,z_2)-\mfm_{k_1}(z_2)\sum_{s\neq k_1}^d\mcC_{sl_1,k_2,N}^{(d)}\\
		&-\mcC_{k_1k_2,N}^{(d)}(z_2,z_1)\sum_{s\neq k_1}^dV_{sl_1,N}^{(d)}(z_1,z_2)-2\mcV_{k_1l_1,k_2,N}^{(d)}(z_1,z_2)-\kappa_4\mcW_{k_1l_1,k_2,N}^{(d)}(z_1,z_2)+\mrO(C_{\eta_0}N^{-\omega})\\
		:&=-\mfm_{k_1}(z_2)\sum_{s\neq k_1}^d\mcC_{sl_1,k_2,N}^{(d)}-\mcF_{k_1l_1,k_2,N}^{(d)}(z_1,z_2)+\mrO(C_{\eta_0}N^{-\omega}).
	\end{align*}
	Hence, for any {\bf fixed} \(k_2\in\{1,\cdots,d\}\), define
	\begin{align}
		&\bbC_{k_2,N}^{(d)}(z_1,z_2)=[\mcC_{kl,k_2,N}^{(d)}(z_1,z_2)]_{d\times d}\quad{\rm and}\quad\bbF_{k_2,N}^{(d)}(z_1,z_2)=[\mcF_{kl,k_2,N}^{(d)}(z_1,z_1)]_{d\times d},
	\end{align}
	then we obtain that
	\begin{align*}
		\bbTheta_N^{(d)}(z_1,z_2)\bbC_{k_2,N}^{(d)}(z_1,z_2)=-\bbF_{k_2,N}^{(d)}(z_1,z_2)+\mro(\boldsymbol{1}_{d\times d}),
	\end{align*}
	where \(\bbTheta_N^{(d)}(z_1,z_2)\) is defined in (\ref{Eq of bbPi N}) and it is invertible, so we can further derive 
	\begin{align}
		\lim_{N\to\infty}\Vert\bbC_{k_2,N}^{(d)}(z_1,z_2)-\bbPi^{(d)}(z_1,z_2)^{-1}\diag(\mfc^{-1}\circ\bbg(z_1))\bbF_{k_2,N}^{(d)}(z_1,z_2)\Vert=0,\label{Eq of bbC 2 vs 1 general d}
	\end{align}
	which suggests that all entries of \(\bbC_{k_2,N}^{(d)}(z_1,z_2)\) are bounded by \(C_{\eta_0,\mfc,d}\). 

    \vspace{5mm}
    \noindent
    {\bf Calculations of (\ref{Eq of tightness major 1 general d}):} By the cumulant expansion (\ref{Eq of cumulant expansion}) again, we have
	\begin{align*}
		&\frac{1}{\sqrt{N}}\sum_{i_1\cdots i_d}^{n_1\cdots n_d}\sum_{s\neq k_1}^d\mbE\big[X_{i_1\cdots i_d}\mcA_{i_1\cdots i_d}^{(k_1,s)}Q_{i_s\cdot}^{sl_1}(z_1)Q_{\cdot i_{k_1}}^{l_1k_1}(z_2)\tr(\bbQ^{k_2l_2}(\bar{z}_1)\bbQ^{l_2k_2}(\bar{z}_2))^c\big]\\
		&=\frac{1}{\sqrt{N}}\sum_{i_1\cdots i_d}^{n_1\cdots n_d}\Big(\sum_{s\neq k_1}^d\sum_{\alpha=0}^3\frac{\kappa_{\alpha+1}}{\alpha!}\mbE\big[\mcA_{i_1\cdots i_d}^{(k_1,s)}\partial_{i_1\cdots i_d}^{(\alpha)}\{Q_{i_s\cdot}^{sl}(z_1)Q_{\cdot i_{k_1}}^{l_1k_1}(z_2)\tr(\bbQ^{k_2l_2}(\bar{z}_1)\bbQ^{l_2k_2}(\bar{z}_2))^c\}\big]+\epsilon_{i_1\cdots i_d}^{(4)}\Big).
	\end{align*}
    Here, we still omit the details for calculating minors. 

    \vspace{5mm}
    \noindent
    {\bf First derivatives:} When \(\alpha=1\), since
		\begin{align*}
			&\partial_{i_1\cdots i_d}^{(1)}\tr(\bbQ^{kl}(z_1)\bbQ^{lk}(z_2))=-N^{-1/2}\sum_{s_1\neq s_2}^d\sum_{j=1}^2\mcA_{i_1\cdots i_d}^{(s_1,s_2)}Q_{i_{s_1}\cdot}^{s_1k}(z_j)\bbQ^{kl}(z_{3-j})Q_{\cdot i_{s_2}}^{ls_2}(z_j),
		\end{align*}
		and
		\begin{align*}
			&\frac{1}{\sqrt{N}}\sum_{i_1\cdots i_d}^{n_1\cdots n_d}\sum_{i_t}^{n_{l_1}}\mcA_{i_1\cdots i_d}^{(k_1,s)}\partial_{i_1\cdots i_d}^{(1)}\{Q_{i_si_t}^{sl_1}(z_1)Q_{i_{k_1}i_t}^{k_1l_1}(z_2)\}\\
			&=-N^{-1}\sum_{i_1\cdots i_d}^{n_1\cdots n_d}\mcA_{i_1\cdots i_d}^{(k_1,s)}\mcA_{i_1\cdots i_d}^{(s_1,s_2)}\big[Q_{i_si_{s_1}}^{ss_1}(z_1)Q_{i_{s_2}\cdot}^{s_2l_1}(z_1)Q_{\cdot i_{k_1}}^{l_1k_1}(z_2)+Q_{i_{k_1}i_{s_1}}^{k_1s_1}(z_2)Q_{i_{s_2}\cdot}^{s_2l_1}(z_2)Q_{\cdot i_s}^{l_1s}(z_1)\big]\\
			&=-N^{-1}\sum_{i_1\cdots i_d}^{n_1\cdots n_d}(\mcA_{i_1\cdots i_d}^{(k_1,s)})^2[Q_{i_si_s}^{ss}(z_1)Q_{i_{k_1}\cdot}^{k_1l_1}(z_1)Q_{\cdot i_{k_1}}^{l_1k_1}(z_2)+Q_{i_{k_1}i_{k_1}}^{k_1k_1}(z_2)Q_{i_s\cdot}^{sl_1}(z_2)Q_{\cdot i_s}^{l_1s}(z_1)]+\mrO(\eta_0^{-3}N^{-1/2})\\
			&=-N^{-1}[\tr(\bbQ^{ss}(z_1))\tr(\bbQ^{k_1l_1}(z_1)\bbQ^{l_1k_1}(z_2))+\tr(\bbQ^{k_1k_1}(z_2))\tr(\bbQ^{sl_1}(z_1)\bbQ^{l_1s}(z_2))]+\mrO(\eta_0^{-3}N^{-1/2}).
		\end{align*}
        For simplicity, we define
        \begin{small}
        $$\mcV_{k_1l_1,k_2l_2,N}^{(d)}(z_1,z_2)=\frac{1}{N}\sum_{s\neq k_1}^d\sum_{j=1}^2\mbE\big[\tr\big(\bbQ^{sl_1}(z_1)\bbQ^{l_1k_1}(z_2)[\bbQ^{k_1k_2}(\bar{z}_j)\bbQ^{k_2l_2}(\bar{z}_{3-j})\bbQ^{l_2s}(\bar{z}_j)+\bbQ^{k_1l_2}(\bar{z}_j)\bbQ^{l_2k_2}(\bar{z}_{3-j})\bbQ^{k_2s}(\bar{z}_j)]\big)\big],$$
        \end{small}\noindent
		then by direct calculations, we have
		\begin{align}
			&\frac{1}{\sqrt{N}}\sum_{i_1\cdots i_d}^{n_1\cdots n_d}\sum_{s\neq k_1}^d\mbE\big[\mcA_{i_1\cdots i_d}^{(k_1,s)}Q_{i_s\cdot}^{sl_1}(z_1)Q_{\cdot i_{k_1}}^{l_1k_1}(z_2)\partial_{i_1\cdots i_d}^{(1)}\{\tr(\bbQ^{k_2l_2}(\bar{z}_1)\bbQ^{l_2k_2}(\bar{z}_2))^c\}\big]\notag\\
			&=-\frac{2}{N}\sum_{i_1\cdots i_d}^{n_1\cdots n_d}\sum_{s\neq k_1}^d\sum_{s_1\neq s_2}^d\sum_{j=1}^2\mbE\big[\mcA_{i_1\cdots i_d}^{(k_1,s)}Q_{i_s\cdot}^{sl_1}(z_1)Q_{\cdot i_{k_1}}^{l_1k_1}(z_2)\mcA_{i_1\cdots i_d}^{(s_1,s_2)}Q_{i_{s_1}\cdot}^{s_1k_2}(\bar{z}_j)\bbQ^{k_2l_2}(\bar{z}_{3-j})Q_{\cdot i_{s_2}}^{l_2s_2}(\bar{z}_j)\big]\notag\\
			&=-\frac{2}{N}\sum_{s\neq k_1}^d\sum_{j=1}^2\mbE\big[\tr\big(\bbQ^{sl_1}(z_1)\bbQ^{l_1k_1}(z_2)[\bbQ^{k_1k_2}(\bar{z}_j)\bbQ^{k_2l_2}(\bar{z}_{3-j})\bbQ^{l_2s}(\bar{z}_j)\notag\\
			&+\bbQ^{k_1l_2}(\bar{z}_j)\bbQ^{l_2k_2}(\bar{z}_{3-j})\bbQ^{k_2s}(\bar{z}_j)]\big)\big]+\mrO(C_{\eta_0}N^{-1/2})=-2\mcV_{k_1l_1,k_2l_2,N}^{(d)}(z_1,z_2)+\mrO(C_{\eta_0}N^{-1/2})\label{Eq of mcV tightness general d}
		\end{align}
		and
		\begin{align}
			&\frac{1}{\sqrt{N}}\sum_{i_1\cdots i_d}^{n_1\cdots n_d}\sum_{s\neq k_1}^d\mbE\big[\partial_{i_1\cdots i_d}^{(1)}\{\mcA_{i_1\cdots i_d}^{(k_1,s)}Q_{i_s\cdot}^{sl_1}(z_1)Q_{\cdot i_{k_1}}^{l_1k_1}(z_2)\}\tr(\bbQ^{k_2l_2}(\bar{z}_1)\bbQ^{l_2k_2}(\bar{z}_2))^c\big]\notag\\
			&=-\frac{1}{N}\sum_{s\neq k_1}^d\Cov\big(\tr(\bbQ^{ss}(z_1))\tr(\bbQ^{k_1l_1}(z_1)\bbQ^{l_1k_1}(z_2)),\tr(\bbQ^{k_2l_2}(z_1)\bbQ^{l_2k_2}(z_2))\big)\notag\\
			&-\frac{1}{N}\sum_{s\neq k_1}^d\Cov\big(\tr(\bbQ^{k_1k_1}(z_2))\tr(\bbQ^{sl_1}(z_1)\bbQ^{l_1s}(z_2)),\tr(\bbQ^{k_2l_2}(z_1)\bbQ^{l_2k_2}(z_2)\big)+\mrO(C_{\eta_0}N^{-\omega})\notag\\
			&=-\mcC_{k_1l_1,k_2l_2,N}^{(d)}\sum_{s\neq k_1}^d\mfm_s(z_1)-V_{k_1l_1,N}^{(d)}(z_1,z_2)\sum_{s\neq k_1}^d\mcC_{s,k_2l_2,N}^{(d)}\notag\\
			&-\mfm_{k_1}(z_2)\sum_{s\neq k_1}^d\mcC_{sl_1,k_2l_2,N}^{(d)}-\mcC_{k_1,k_2l_2,N}^{(d)}\sum_{s\neq k_1}^dV_{sl_1,N}^{(d)}(z_1,z_2)+\mrO(C_{\eta_0}N^{-\omega})\mcC_{k_2l_2,k_2l_2,N}^{(d)}\label{Eq of tightness 1 general d}
		\end{align}
		where we use the same trick as (\ref{Eq of covariance trick}) and $V_{sl_1,N}^{(d)}(z_1,z_2)=N^{-1}\mbE[\tr(\bbQ^{sl_1}(z_1)\bbQ^{l_1s}(z_2))]$.

    \vspace{5mm}
    \noindent
    {\bf Third derivatives:} When \(\alpha=3\), similar to proofs for \(\alpha=3\) in Theorem \ref{Thm of Tightness}, the only one contains major terms is
		\begin{align}
			&\frac{1}{\sqrt{N}}\sum_{i_1\cdots i_d}^{n_1\cdots n_d}\sum_{s\neq k_1}^d\mbE\big[\partial_{i_1\cdots i_d}^{(1)}\{\mcA_{i_1\cdots i_d}^{(k_1,s)}Q_{i_s\cdot}^{sl_1}(z_1)Q_{\cdot i_{k_1}}^{l_1k_1}(z_2)\}\partial_{i_1\cdots i_d}^{(2)}\{\tr(\bbQ^{k_2l_2}(\bar{z}_1)\bbQ^{l_2k_2}(\bar{z}_2))^c\}\big]\notag\\
			&=-2\mcW_{k_1l_1,k_2l_2,N}^{(d)}(z_1,z_2)+\mrO(C_{\eta_0}N^{-\omega}),\label{Eq of mcW tightness general d}
		\end{align}
		where
		\begin{align*}
			\mcW_{k_1l_1,k_2l_2,N}^{(d)}(z_1,z_2):=\mcW_{k_1l_1,k_2l_2,N}^{(1,d)}(z,\bar{z})+\mcW_{k_1l_1,k_2l_2,N}^{(2,d)}(z,\bar{z})+\mcW_{k_1l_1,k_2l_2,N}^{(3,d)}(z,\bar{z})
		\end{align*}
		and \(\mcW_{k_1l_1,k_2l_2,N}^{(1,d)}(z,\bar{z}),\mcW_{k_1l_1,k_2l_2,N}^{(2,d)}(z,\bar{z}),\mcW_{k_1l_1,k_2l_2,N}^{(3,d)}(z,\bar{z})\) will be defined in (\ref{Eq of mcW 22 1}), (\ref{Eq of mcW 22 3}) and (\ref{Eq of mcW 22 2}), respectively. Notice that
		\begin{align}
			&\partial_{i_1\cdots i_d}^{(2)}\{\tr(\bbQ^{k_2l_2}(\bar{z}_1)\bbQ^{l_2k_2}(\bar{z}_2))\}\notag\\
			&=\frac{2}{N}\sum_{\substack{t_1\neq t_2\\t_3\neq t_4}}^d\mcA_{i_1\cdots i_d}^{(t_1,t_2)}Q_{i_{t_2}i_{t_3}}^{t_2t_3}(\bar{z}_1)\mcA_{i_1\cdots i_d}^{(t_3,t_4)}Q_{i_{t_4}\cdot}^{t_4l_2}(\bar{z}_1)\bbQ^{l_2k_2}(\bar{z}_2)Q_{\cdot i_{t_1}}^{k_2t_1}(\bar{z}_1)\label{Eq of partial d 3 1}\\
			&+\frac{2}{N}\sum_{\substack{t_1\neq t_2\\t_3\neq t_4}}^dQ_{i_{t_4}\cdot}^{t_4k_2}(\bar{z})\bbQ^{k_2l_2}(\bar{z})Q_{\cdot i_{t_1}}^{l_2t_1}(\bar{z})\mcA_{i_1\cdots i_d}^{(t_1,t_2)}Q_{i_{t_2}i_{t_3}}^{t_2t_3}(\bar{z})\mcA_{i_1\cdots i_d}^{(t_3,t_4)}\label{Eq of partial d 3 2}\\
			&+\frac{2}{N}\sum_{\substack{t_1\neq t_2\\t_3\neq t_4}}^d\mcA_{i_1\cdots i_d}^{(t_1,t_2)}Q_{i_{t_2}\cdot}^{t_2l_2}(\bar{z}_1)Q_{\cdot i_{t_3}}^{l_2t_3}(\bar{z}_1)\mcA_{i_1\cdots i_d}^{(t_3,t_4)}Q_{i_{t_4}\cdot}^{t_4k_2}(\bar{z}_2)Q_{\cdot i_{t_1}}^{k_2t_1}(\bar{z}_2)\label{Eq of partial d 3 3},
		\end{align}
		then we combining (\ref{Eq of partial d 3 1}) with \(\partial_{i_1\cdots i_d}^{(1)}\{\mcA_{i_1\cdots i_d}^{(k_1,s)}Q_{i_s\cdot}^{sl_1}(z_1)Q_{\cdot i_{k_1}}^{l_1k_1}(z_2)\}\), by Lemmas \ref{Rem of extension minor 2}, \ref{Rem of extension of Corollary} and (\ref{Eq of mcB}), we have
		\begin{small}
			\begin{align}
				&\frac{1}{\sqrt{N}}\sum_{i_1\cdots i_d}^{n_1\cdots n_d}\sum_{\substack{s\neq k_1\\t_1\neq t_2\\t_3\neq t_4}}^d\mbE\big[\partial_{i_1\cdots i_d}^{(1)}\{\mcA_{i_1\cdots i_d}^{(k_1,s)}Q_{i_s\cdot}^{sl_1}(z_1)Q_{\cdot i_{k_1}}^{l_1k_1}(z_2)\}\mcA_{i_1\cdots i_d}^{(t_1,t_2)}Q_{i_{t_2}i_{t_3}}^{t_2t_3}(\bar{z}_1)\mcA_{i_1\cdots i_d}^{(t_3,t_4)}Q_{i_{t_4}\cdot}^{t_4l_2}(\bar{z}_1)\bbQ^{l_2k_2}(\bar{z}_2)Q_{\cdot i_{t_1}}^{k_2t_1}(\bar{z}_1)\big]\notag\\
				&=-\frac{1}{N^2}\sum_{i_1\cdots i_d}^{n_1\cdots n_d}\sum_{\substack{s\neq k_1\\t_1\neq t_2}}^d\mbE\big[(\mcA_{i_1\cdots i_d}^{(k_1,s)})^2Q_{i_si_s}^{ss}(z_1)Q_{i_{k_1}\cdot}^{k_1l_1}(z_1)Q_{\cdot i_{k_1}}^{l_1k_1}(z_2)(\mcA_{i_1\cdots i_d}^{(t_1,t_2)})^2Q_{i_{t_2}i_{t_2}}^{t_2t_2}(\bar{z}_1)Q_{i_{t_1}\cdot}^{t_1l_2}(\bar{z}_1)\bbQ^{l_2k_2}(\bar{z}_2)Q_{\cdot i_{t_1}}^{k_2t_1}(\bar{z}_1)\big]\notag\\
				&-\frac{1}{N^2}\sum_{i_1\cdots i_d}^{n_1\cdots n_d}\sum_{\substack{s\neq k_1\\t_1\neq t_2}}^d\mbE\big[(\mcA_{i_1\cdots i_d}^{(k_1,s)})^2Q_{i_s\cdot}^{sl_1}(z_1)Q_{\cdot i_s}^{l_1s}(z_2)Q_{i_{k_1}i_{k_1}}^{k_1k_1}(z_2)(\mcA_{i_1\cdots i_d}^{(t_1,t_2)})^2Q_{i_{t_2}i_{t_2}}^{t_2t_2}(\bar{z}_1)Q_{i_{t_1}\cdot}^{t_1l_2}(\bar{z}_1)\bbQ^{l_2k_2}(\bar{z}_2)Q_{\cdot i_{t_1}}^{k_2t_1}(\bar{z}_1)\big]+\mrO(C_{\eta_0}N^{-1/2})\notag\\
				&=-\frac{1}{N^2}\sum_{s\neq k_1}^d\mcB_{(4)}^{(k_1,s)}\mbE\big[\tr(\bbQ^{ss}(z_1)\circ\bbQ^{ss}(\bar{z}_2))\tr((\bbQ^{k_1l_1}(z_1)\bbQ^{l_1k_1}(z_2))\circ(\bbQ^{k_1l_2}(\bar{z}_1)\bbQ^{l_2k_2}(\bar{z}_2)\bbQ^{k_2k_1}(\bar{z}_1)))\big]\notag\\
				&-\frac{1}{N^2}\sum_{s\neq k_1}^d\mcB_{(4)}^{(k_1,s)}\mbE\big[\tr(\bbQ^{ss}(z_1)\circ(\bbQ^{sl_2}(\bar{z}_1)\bbQ^{l_2k_2}(\bar{z}_2)\bbQ^{k_2s}(\bar{z}_1)))\tr((\bbQ^{k_1l_1}(z_1)\bbQ^{l_1k_1}(z_2))\circ\bbQ^{k_1k_1}(\bar{z}_1))\big]\notag\\
				&-\frac{1}{N^2}\sum_{s\neq k_1}^d\mcB_{(4)}^{(k_1,s)}\mbE\big[\tr((\bbQ^{sl_1}(z_1)\bbQ^{l_1s}(z_2))\circ\bbQ^{ss}(\bar{z}_1))\tr(\bbQ^{k_1k_1}(z_2)\circ(\bbQ^{k_1l_2}(\bar{z}_1)\bbQ^{l_2k_2}(\bar{z}_2)\bbQ^{k_2k_1}(\bar{z}_1)))\big]\notag\\
				&-\frac{1}{N^2}\sum_{s\neq k_1}^d\mcB_{(4)}^{(k_1,s)}\mbE\big[\tr((\bbQ^{sl_1}(z_1)\bbQ^{l_1s}(z_2))\circ(\bbQ^{sl_2}(\bar{z}_1)\bbQ^{l_2k_2}(\bar{z}_2)\bbQ^{k_2s}(\bar{z}_1)))\tr(\bbQ^{k_1k_1}(z_2)\circ\bbQ^{k_1k_1}(\bar{z}_1))\big]+\mrO(C_{\eta_0}N^{-1/2})\notag\\
				:&=\mcW_{k_1l_1,k_2l_2,N}^{(1,d)}(z_1,z_2)+\mrO(C_{\eta_0}N^{-1/2}),\label{Eq of mcW 22 1}
			\end{align}
		\end{small}\noindent
		Similarly, combining (\ref{Eq of partial d 3 2}) with \(\partial_{i_1\cdots i_d}^{(1)}\{\mcA_{i_1\cdots i_d}^{(k_1,s)}Q_{i_s\cdot}^{sl_1}(z_1)Q_{\cdot i_{k_1}}^{l_1k_1}(z_2)\}\) will obtain the same result, just replace all \(\bar{z}_2,\bar{z}_1\) by \(\bar{z}_1,\bar{z}_2\) respectively, i.e.
		\begin{align}
			&\mcW_{k_1l_1,k_2l_2,N}^{(2,d)}(z_1,z_2)\label{Eq of mcW 22 3}\\
			:&=\frac{1}{N^2}\sum_{s\neq k_1}^d\mcB_{(4)}^{(k_1,s)}\mbE\big[\tr(\bbQ^{ss}(z_1)\circ\bbQ^{ss}(\bar{z}_2))\tr((\bbQ^{k_1l_1}(z_1)\bbQ^{l_1k_1}(z_2))\circ(\bbQ^{k_1l_2}(\bar{z}_2)\bbQ^{l_2k_2}(\bar{z}_1)\bbQ^{k_2k_1}(\bar{z}_2)))\big]\notag\\
			&+\frac{1}{N^2}\sum_{s\neq k_1}^d\mcB_{(4)}^{(k_1,s)}\mbE\big[\tr(\bbQ^{ss}(z_1)\circ(\bbQ^{sl_2}(\bar{z}_2)\bbQ^{l_2k_2}(\bar{z}_1)\bbQ^{k_2s}(\bar{z}_2)))\tr((\bbQ^{k_1l_1}(z_1)\bbQ^{l_1k_1}(z_2))\circ\bbQ^{k_1k_1}(\bar{z}_2))\big]\notag\\
			&+\frac{1}{N^2}\sum_{s\neq k_1}^d\mcB_{(4)}^{(k_1,s)}\mbE\big[\tr((\bbQ^{sl_1}(z_1)\bbQ^{l_1s}(z_2))\circ\bbQ^{ss}(\bar{z}_2))\tr(\bbQ^{k_1k_1}(z_2)\circ(\bbQ^{k_1l_2}(\bar{z}_2)\bbQ^{l_2k_2}(\bar{z}_1)\bbQ^{k_2k_1}(\bar{z}_2)))\big]\notag\\
			&+\frac{1}{N^2}\sum_{s\neq k_1}^d\mcB_{(4)}^{(k_1,s)}\mbE\big[\tr((\bbQ^{sl_1}(z_1)\bbQ^{l_1s}(z_2))\circ(\bbQ^{sl_2}(\bar{z}_1)\bbQ^{l_2k_2}(\bar{z}_1)\bbQ^{k_2s}(\bar{z}_2)))\tr(\bbQ^{k_1k_1}(z_1)\circ\bbQ^{k_1k_1}(\bar{z}_2))\big]\notag
		\end{align}
		For (\ref{Eq of partial d 3 3}), we directly list it as follows:
		\begin{align}
			&\mcW_{k_1l_1,k_2l_2,N}^{(3,d)}(z_1,z_2)\label{Eq of mcW 22 2}\\
			:&=\frac{1}{N^2}\sum_{s\neq k_1}^d\mcB_{(4)}^{(k_1,s)}\mbE\big[\tr(\bbQ^{ss}(z_1)\circ(\bbQ^{sl_2}(\bar{z}_1)\bbQ^{l_2s}(\bar{z}_1)))\tr((\bbQ^{k_1l_1}(z_1)\bbQ^{l_1k_1}(z_2))\circ(\bbQ^{k_1k_2}(\bar{z}_2)\bbQ^{k_2k_1}(\bar{z}_2)))\big]\notag\\
			&+\frac{1}{N^2}\sum_{s\neq k_1}^d\mcB_{(4)}^{(k_1,s)}\mbE\big[\tr(\bbQ^{ss}(z_1)\circ(\bbQ^{sk_2}(\bar{z}_2)\bbQ^{k_2s}(\bar{z}_2)))\tr((\bbQ^{k_1l_1}(z_1)\bbQ^{l_1k_1}(z_2))\circ(\bbQ^{k_1l_2}(\bar{z}_1)\bbQ^{l_2k_1}(\bar{z}_1)))\big]\notag\\
			&+\frac{1}{N^2}\sum_{s\neq k_1}^d\mcB_{(4)}^{(k_1,s)}\mbE\big[\tr((\bbQ^{sl_1}(z_1)\bbQ^{l_1s}(z_2))\circ(\bbQ^{sl_2}(\bar{z}_1)\bbQ^{l_2s}(\bar{z}_1)))\tr(\bbQ^{k_1k_1}(z_2)\circ(\bbQ^{k_1k_2}(\bar{z}_2)\bbQ^{k_2k_1}(\bar{z}_2)))\big]\notag\\
			&+\frac{1}{N^2}\sum_{s\neq k_1}^d\mcB_{(4)}^{(k_1,s)}\mbE\big[\tr((\bbQ^{sl_1}(z_1)\bbQ^{l_1s}(z_2))\circ(\bbQ^{sk_2}(\bar{z}_2)\bbQ^{k_2s}(\bar{z}_2)))\tr(\bbQ^{k_1k_1}(z_2)\circ(\bbQ^{k_1l_2}(\bar{z}_1)\bbQ^{l_2k_1}(\bar{z}_1)))\big].\notag
		\end{align}
	As a result, combining (\ref{Eq of mcV tightness general d}), (\ref{Eq of tightness 1 general d}) and (\ref{Eq of mcW tightness general d}), we obtain that
	\begin{align*}
		&(z_1+\mfm(z_1)-\mfm_{k_1}(z_1))\mcC_{k_1l_1,k_2l_2,N}^{(d)}=-\Big(\delta_{k_1l_1}+\sum_{s\neq k_1}^dV_{sl_1,N}^{(d)}(z_1,z_2)\Big)\mcC_{k_1,k_2l_2,N}^{(d)}-2\mcV_{k_1l_1,k_2l_2,N}^{(d)}(z_1,z_2)+\mrO(C_{\eta_0}N^{-\omega})\\
		&-\kappa_4\mcW_{k_1l_1,k_2l_2,N}^{(d)}(z_1,z_2)-V_{k_1l_1,N}^{(d)}(z_1,z_2)\sum_{s\neq k_1}^d\mcC_{s,k_2l_2,N}^{(d)}-\mfm_{k_1}(z_2)\sum_{s\neq k_1}^d\mcC_{sl_1,k_2l_2,N}^{(d)}+\mrO(C_{\eta_0}N^{-\omega})\mcC_{k_2l_2,k_2l_2,N}^{(d)}\\
		&:=-\mfm_{k_1}(z_2)\sum_{s\neq k_1}^d\mcC_{sl_1,k_2l_2,N}^{(d)}-\mcF_{k_1l_1,k_2l_2,N}^{(d)}(z_1,z_2)+\mrO(C_{\eta_0}N^{-\omega})\mcC_{k_2l_2,k_2l_2,N}^{(d)}+\mrO(C_{\eta_0}N^{-\omega}).
	\end{align*}
	Hence, for any \(k_2,l_2\in\{1,\cdots,d\}\), define
	\begin{align}
		\bbC_{k_2l_2,N}^{(d)}(z_1,z_2):=[\mcC_{k_1l_1,k_2l_2,N}^{(d)}(z_1,z_2)]_{d\times d}\quad{\rm and}\quad\bbF_{k_2l_2,N}^{(d)}(z,\bar{z}):=[\mcF_{k_1l_1,k_2l_2}^N(z,\bar{z})]_{d\times d},
	\end{align}
	we have
	\begin{align*}
		\bbTheta_N^{(d)}(z_1,z_2)\bbC_{k_2l_2,N}^{(d)}(z_1,z_2)=-\bbF_{k_2l_2,N}^{(d)}(z_1,z_2)+\boldsymbol{1}_{d\times d}\mrO(C_{\eta_0}N^{-\omega})\mcC_{k_2l_2,k_2l_2,N}^{(d)},
	\end{align*}
	where \(\bbTheta_N^{(d)}\) is defined in (\ref{Eq of bbPi N}). Notice that 
	$$\Vert\boldsymbol{1}_{d\times d}\mrO(C_{\eta_0}N^{-\omega})\mcC_{k_2l_2,k_2l_2,N}^{(d)}\Vert\leq\sqrt{d}C_{\eta_0}N^{-\omega}\Vert\bbC_{k_2l_2,N}(z_1,z_2)\Vert,$$
	so we can use the same argument as those in Theorem \ref{Thm of covariance general d} to derive that
	$$\lim_{N\to\infty}\Vert\bbC_{k_2l_2,N}^{(d)}(z_1,z_2)\Vert\leq\lim_{N\to\infty}\Vert\bbPi^{(d)}(z_1,z_2)^{-1}\diag(\mfc^{-1}\circ\bbg(z_1))\Vert\cdot\Vert\bbF_{k_2l_2,N}^{(d)}(z_1,z_2)\Vert\leq C_{\eta_0,\mfc,d},$$
	which suggests that all entries of \(\bbC_{k_2l_2,N}^{(d)}(z_1,z_2)\) are bounded by \(C_{\eta_0,\mfc,d}\).
\end{proof}
\subsubsection{Characteristic function}
\begin{thm}\label{Thm of CLT general d}
	Under Assumptions {\rm \ref{Ap of general noise}} and {\rm \ref{Ap of dimension}}, \(\tr(\boldsymbol{Q}(z))-\mathbb{E}[\tr(\boldsymbol{Q}(z))]\) converges weakly to a Gaussian random process on \(z\in\mathcal{S}_{\eta_0}\) in {\rm (\ref{Eq of stability region general d})}.
\end{thm}
\begin{proof}
	Recall the following notations in Theorem \ref{Thm of CLT}:
	\begin{align}
		\gamma(z):=\sum_{l=1}^d\gamma_l(z):=\sum_{l=1}^d\tr(\boldsymbol{Q}^{ll}(z))^c,\quad(\mathfrak{a}(\tau),\mathfrak{b}(\tau)):=\left\{\begin{array}{cc}
			(1/2,1/2)&\tau=1\\(1/2{\rm i},-1/2{\rm i})&\tau={\rm i}
		\end{array}\right.,\notag
	\end{align}
	where \(l\in\{1,\cdots,d\}\). Given \(q\in\mbN^+\), let \(\bbz_q:=(z_1,\cdots,z_q)',\bbtau_q:=(\tau_1,\cdots,\tau_q)'\) and \(\bbt_q:=(t_1,\cdots,t_q)'\), where \(z_s\in\mcS_{\eta_0},\tau_s\in\{1,{\rm i}\}\) and \(t_s\in\mbR\) respectively, define
	\begin{align}
		e_q:=e_q(\bbt_q,\bbtau_q,\bbz_q):=\exp\Big({\rm i}\sum_{s=1}^qt_s[\mfa(\tau_s)\gamma(z_s)+\mfb(\tau_s)\gamma(\bar{z}_s)]\Big),
	\end{align}
	so the characteristic function is \(\mbE[e_q]\). Notice that
	$$\frac{\partial}{\partial t_s}\mathbb{E}[e_q]={\rm i}\mathbb{E}\left[e_q\left(\mathfrak{a}(\tau_s)\gamma(z_s)+\mathfrak{b}(\tau_s)\gamma(\bar{z}_s)\right)\right],$$
	and we will show that there exists a set of covariance coefficients \(A_{rw}\) such that for each fixed \(\bbt_q\)
	$$\lim_{N\to\infty}\Big|\mathbb{E}\left[e_q\left(\mathfrak{a}(\zeta_s)\gamma(z_s)+\mathfrak{b}(\zeta_s)\gamma(\bar{z}_s)\right)\right]+\mathbb{E}[e_q]\sum_{r=1}^qt_rA_{rs,N}\Big|=0.$$
	For any \(z\in\mcS_{\eta_0}\), by the cumulant expansion (\ref{Eq of cumulant expansion}), we have
	\begin{align*}
		&z\mathbb{E}[e_q\gamma_l(z)]=z\Cov(\gamma_l(z),e_q)=\frac{1}{\sqrt{N}}\sum_{i_1\cdots i_d}^{n_1\cdots n_d}\sum_{r\neq l}^d\mathbb{E}\big[X_{i_1\cdots i_d}\mcA_{i_1\cdots i_d}^{(l,r)}Q_{i_ri_l}^{rl}(z)e_q^c\big]\\
		&=\frac{1}{\sqrt{N}}\sum_{i_1\cdots i_d}^{n_1\cdots n_d}\Big(\sum_{r\neq l}^d\sum_{\alpha=0}^3\frac{\kappa_{\alpha+1}}{\alpha!}\mbE\big[\partial_{i_1\cdots i_d}^{(\alpha)}\{\mcA_{i_1\cdots i_d}^{(l,r)}Q_{i_ri_l}^{rl}(z)e_q^c\}\big]+\epsilon_{i_1\cdots i_d}^{(4)}\Big).
	\end{align*}
	Similar to proofs of Theorem \ref{Thm of CLT}, only the cases of \(\alpha=1,3\) contain the major terms:

    \vspace{5mm}
    \noindent
    {\bf First derivatives:} When \(\alpha=1\), since
		\begin{align*}
			&\frac{1}{\sqrt{N}}\sum_{i_1\cdots i_d}^{n_1\cdots n_d}\sum_{r\neq l}^d\mbE\big[\partial_{i_1\cdots i_d}^{(1)}\{\mcA_{i_1\cdots i_d}^{(l,r)}Q_{i_ri_l}^{rl}(z)\}e_q^c\big]\\
			&=-\frac{1}{N}\sum_{i_1\cdots i_d}^{n_1\cdots n_d}\sum_{r\neq l}^d\sum_{s_1\neq s_2}^d\mbE\big[\mcA_{i_1\cdots i_d}^{(l,r)}Q_{i_ri_{s_1}}^{rs_1}(z)\mcA_{i_1\cdots i_d}^{(s_1,s_2)}Q_{i_{s_2}i_l}^{s_2l}(z)e_q^c\big]\\
			&=-\frac{1}{N}\sum_{r\neq l}^d\Cov(\tr(\bbQ^{rr}(z))\tr(\bbQ
			^{ll}(z)),e_q)+\mrO(C_{\eta_0}N^{-1/2})\\
			&=-\Cov(\gamma_l(z),e_q)\sum_{r\neq l}^d\mfm_r(z)-\mfm_l(z)\sum_{r\neq l}^d\Cov(\gamma_r(z),e_q)+\mrO(C_{\eta_0}N^{-\omega}),
		\end{align*}
		where we use the fact that \(|e_q|\leq1\). And
		\begin{align*}
			&\partial_{i_1\cdots i_d}^{(1)}\{e_q\}=-\frac{{\rm i} e_q}{\sqrt{N}}\sum_{s=1}^q\sum_{w=1}^d\sum_{s_1\neq s_2}^dt_s\mcA_{i_1\cdots i_d}^{(s_1,s_2)}\big[\mfa(\tau_s)Q_{i_{s_1}\cdot}^{s_1w}(z_s)Q_{\cdot i_{s_2}}^{ws_2}(z_s)+\mfb(\tau_s)Q_{i_{s_1}\cdot}^{s_1w}(\bar{z}_s)Q_{\cdot i_{s_2}}^{ws_2}(\bar{z}_s)\big],
		\end{align*}
		then
		\begin{align*}
			&\frac{1}{\sqrt{N}}\sum_{i_1\cdots i_d}^{n_1\cdots n_d}\sum_{r\neq l}^d\mbE\big[\mcA_{i_1\cdots i_d}^{(l,r)}Q_{i_ri_l}^{rl}(z)\partial_{i_1\cdots i_d}^{(1)}\{e_q\}\big]\\
			&=-\frac{2{\rm i}}{N}\sum_{s=1}^q\sum_{w=1}^d\sum_{r\neq l}^dt_s\mbE\big[\tr(\bbQ^{rl}(z)[\mfa(\tau_s)\bbQ^{lw}(z_s)\bbQ^{wr}(z_s)+\mfb(\tau_s)\bbQ^{lw}(\bar{z}_s)\bbQ^{wr}(\bar{z}_s)])e_q\big]+\mrO(C_{\eta_0}N^{-1/2}).
		\end{align*}
		Moreover, by Theorem \ref{Thm of a.s. convergence}, we can obtain that
		\begin{align*}
			\Cov(N^{-1}\tr(\bbQ^{rw}(z)\bbQ^{wl}(z_s)\bbQ^{lr}(z_s)),e_q)=\mrO(C_{\eta_0}N^{-\omega}),
		\end{align*}
		so it implies that
		\begin{align}
			&\frac{1}{\sqrt{N}}\sum_{i_1\cdots i_d}^{n_1\cdots n_d}\sum_{r\neq l}^d\mbE\big[\partial_{i_1\cdots i_d}^{(1)}\{\mcA_{i_1\cdots i_d}^{(l,r)}Q_{i_ri_l}^{rl}(z)e_q^c\}\big]\label{Eq of characteristic function d 1}\\
			&=-\Cov(\gamma_l(z),e_q)\sum_{r\neq l}^d\mfm_r(z)-\mfm_l(z)\sum_{r\neq l}^d\Cov(\gamma_r(z),e_q)\notag\\
			&-\frac{2{\rm i}\mbE[e_q]}{N}\sum_{s=1}^q\sum_{w=1}^d\sum_{r\neq l}^dt_s\mbE\big[\tr(\bbQ^{rl}(z)[\mfa(\tau_s)\bbQ^{lw}(z_s)\bbQ^{wr}(z_s)+\mfb(\tau_s)\bbQ^{lw}(\bar{z}_s)\bbQ^{wr}(\bar{z}_s)])\big]+\mrO(C_{\eta_0}N^{-\omega})\notag\\
			:&=-\Cov(\gamma_l(z),e_q)\sum_{r\neq l}^d\mfm_r(z)-\mfm_l(z)\sum_{r\neq l}^d\Cov(\gamma_r(z),e_q)\notag\\
			&-2{\rm i}\mbE[e_q]\sum_{s=1}^qt_s\big[\mfa(\tau_s)\mcV_{l,e,N}^{(d)}(z,z_s)+\mfb(\tau_s)\mcV_{l,e,N}^{(d)}(z,\bar{z}_s)\big]+\mrO(C_{\eta_0}N^{-\omega}).\notag
		\end{align}
    {\bf Third derivatives:} When \(\alpha=3\), the only case contains major terms is 
		\begin{align*}
			&\frac{1}{\sqrt{N}}\sum_{i_1\cdots i_d}^{n_1\cdots n_d}\sum_{r\neq l}^d\mbE\big[\partial_{i_1\cdots i_d}^{(1)}\{\mcA_{i_1\cdots i_d}^{(l,r)}Q_{i_ri_l}^{rl}(z)\}\partial_{i_1\cdots i_d}^{(2)}\{e_q\}\big],
		\end{align*}
		since
		\begin{align*}
			&\partial_{i_1\cdots i_d}^{(2)}\{e_q\}=-\frac{e_q}{N}\Big(\sum_{s=1}^q\sum_{w=1}^d\sum_{s_1\neq s_2}^dt_s\mcA_{i_1\cdots i_d}^{(s_1,s_2)}[\mfa(\tau_s)Q_{i_{s_1}\cdot}^{s_1w}(z_s)Q_{\cdot i_{s_2}}^{ws_2}(z_s)+\mfb(\tau_s)Q_{i_{s_1}\cdot}^{s_1w}(\bar{z}_s)Q_{\cdot i_{s_2}}^{ws_2}(\bar{z}_s)]\Big)^2+\frac{2{\rm i} e_q}{N}\\
			&\times\sum_{s=1}^q\sum_{w=1}^d\sum_{\substack{s_1\neq s_2\\s_3\neq s_4}}^dt_s\mcA_{i_1\cdots i_d}^{(s_1,s_2)}\mcA_{i_1\cdots i_d}^{(s_3,s_4)}[\mfa(\tau_s)Q_{i_{s_1}i_{s_3}}^{s_1s_3}(z_s)Q_{i_{s_4}\cdot}^{s_4w}(z_s)Q_{\cdot i_{s_2}}^{ws_2}(z_s)+\mfb(\tau_s)Q_{i_{s_1}i_{s_3}}^{s_1s_3}(\bar{z}_s)Q_{i_{s_4}\cdot}^{s_4w}(\bar{z}_s)Q_{\cdot i_{s_2}}^{ws_2}(\bar{z}_s)],
		\end{align*}
		for the previous term, since it only contains the off-diagonal terms, by Lemma \ref{Rem of extension of Corollary}, if it associate with \(\partial_{i_1\cdots i_d}^{(1)}\{\mcA_{i_1\cdots i_d}^{(l,r)}Q_{i_ri_l}^{rl}(z)\}\), the summation over all \(i_1\cdots i_d\) will be minor. For the later one, by Lemma \ref{Rem of extension minor 2} and (\ref{Eq of mcB}), we have
		\begin{align*}
			&\frac{1}{N^{3/2}}\sum_{i_1\cdots i_d}^{n_1\cdots n_d}\sum_{w=1}^d\sum_{r\neq l}^d\sum_{\substack{s_1\neq s_2\\s_3\neq s_4}}^d\mbE\big[e_q\partial_{i_1\cdots i_d}^{(1)}\{\mcA_{i_1\cdots i_d}^{(l,r)}Q_{i_ri_l}^{rl}(z)\}\mcA_{i_1\cdots i_d}^{(s_1,s_2)}\mcA_{i_1\cdots i_d}^{(s_3,s_4)}Q_{i_{s_1}i_{s_3}}^{s_1s_3}(z_s)Q_{i_{s_4}\cdot}^{s_4w}(z_s)Q_{\cdot i_{s_2}}^{ws_2}(z_s)\big]\\
			&=-\frac{1}{N^2}\sum_{i_1\cdots i_d}^{n_1\cdots n_d}\sum_{w=1}^d\sum_{r\neq l}^d\sum_{\substack{s_1\neq s_2\\s_3\neq s_4\\s_5\neq s_6}}^d\mbE\big[e_q\mcA_{i_1\cdots i_d}^{(l,r)}Q_{i_ri_{s_5}}^{rs_5}(z)\mcA_{i_1\cdots i_d}^{(s_5,s_6)}Q_{i_{s_6}i_l}^{s_6l}(z)\mcA_{i_1\cdots i_d}^{(s_1,s_2)}\mcA_{i_1\cdots i_d}^{(s_3,s_4)}Q_{i_{s_1}i_{s_3}}^{s_1s_3}(z_s)Q_{i_{s_4}\cdot}^{s_4w}(z_s)Q_{\cdot i_{s_2}}^{ws_2}(z_s)\big]\\
			&=-\frac{1}{N^2}\sum_{w=1}^d\sum_{r\neq l}^d\mcB_{(4)}^{(l,r)}\mbE[e_q\tr(\bbQ^{rr}(z)\circ\bbQ^{rr}(z_s))\tr(\bbQ^{ll}(z)\circ(\bbQ^{lw}(z_s)\bbQ^{wl}(z_s)))]\\
			&-\frac{1}{N^2}\sum_{w=1}^d\sum_{r\neq l}^d\mcB_{(4)}^{(l,r)}\mbE[e_q\tr(\bbQ^{rr}(z)\circ(\bbQ^{rw}(z_s)\bbQ^{wr}(z_s)))\tr(\bbQ^{ll}(z)\circ\bbQ^{ll}(z_s))]+\mrO(C_{\eta_0}N^{-1/2}).
		\end{align*}
		Similarly, by Theorem \ref{Thm of a.s. convergence}, we have
		\begin{align*}
			\Cov(N^{-1}\tr(\bbQ^{rr}(z)\circ\bbQ^{rr}(z_s))N^{-1}\tr(\bbQ^{ll}(z)\circ(\bbQ^{lw}(z_s)\bbQ^{wl}(z_s))),e_q)=\mrO(C_{\eta_0}N^{-\omega}),
		\end{align*}
		and as does the other one. For simplicity, denote
		\begin{align*}
			&\frac{1}{N^2}\sum_{s=1}^qt_s\sum_{w=1}^d\sum_{r\neq l}^d\mcB_{(4)}^{(l,r)}\big(\mfa(\tau_s)\mbE[\tr(\bbQ^{rr}(z)\circ\bbQ^{rr}(z_s))\tr(\bbQ^{ll}(z)\circ(\bbQ^{lw}(z_s)\bbQ^{wl}(z_s)))]\\
			&+\mfb(\tau_s)\mbE[\tr(\bbQ^{rr}(z)\circ\bbQ^{rr}(\bar{z}_s))\tr(\bbQ^{ll}(z)\circ(\bbQ^{lw}(\bar{z}_s)\bbQ^{wl}(\bar{z}_s)))]\\
			&+\mfa(\tau_s)\mbE[\tr(\bbQ^{rr}(z)\circ(\bbQ^{rw}(z_s)\bbQ^{wr}(z_s)))\tr(\bbQ^{ll}(z)\circ\bbQ^{ll}(z_s))]\\
			&+\mfb(\tau_s)\mbE[\tr(\bbQ^{rr}(z)\circ(\bbQ^{rw}(\bar{z}_s)\bbQ^{wr}(\bar{z}_s)))\tr(\bbQ^{ll}(z)\circ\bbQ^{ll}(\bar{z}_s))]\big)\\
			&:=\sum_{s=1}^qt_s\big[\mfa(\tau_s)\mcW_{l,e,N}^{(d)}(z,z_s)+\mfb(\tau_s)\mcW_{l,e,N}^{(d)}(z,\bar{z}_s)\big],
		\end{align*}
		then we obtain that
		\begin{align}
			&\frac{1}{\sqrt{N}}\sum_{i_1\cdots i_d}^{n_1\cdots n_d}\sum_{r\neq l}^d\mbE\big[\partial_{i_1\cdots i_d}^{(1)}\{\mcA_{i_1\cdots i_d}^{(l,r)}Q_{i_ri_l}^{rl}(z)\}\partial_{i_1\cdots i_d}^{(2)}\{e_q\}\big]\notag\\
			&=-2{\rm i}\mbE[e_q]\sum_{s=1}^qt_s\big[\mfa(\tau_s)\mcW_{l,e,N}^{(d)}(z,z_s)+\mfb(\tau_s)\mcW_{l,e,N}^{(d)}(z,\bar{z}_s)\big]+\mrO(C_{\eta_0}N^{-\omega}).\label{Eq of characteristic function d 3}
		\end{align}
	Now, define
	\begin{align*}
		\mcC_{l,e,N}^{(d)}(z):=\Cov(\gamma_l(z),e_q)\quad{\rm and}\quad\mcF_{l,e,N}^{(d)}(z,z_s):=2\mcV_{l,e,N}^{(d)}(z,z_s)+\kappa_4\mcW_{l,e,N}^{(d)}(z,z_s)
	\end{align*}
	and
	\begin{align*}
		&\bbC_{e,N}^{(d)}(z):=(\mcC_{1,e,N}^{(d)}(z),\cdots,\mcC_{d,e,N}^{(d)}(z))',\quad\bbF_{e,N}^{(d)}(z,z_s):=(\mcF_{1,e,N}^{(d)}(z,z_s),\cdots,\mcF_{d,e,N}^{(d)}(z,z_s))'.
	\end{align*}
	Combining (\ref{Eq of characteristic function d 1}) and (\ref{Eq of characteristic function d 3}), we obtain that
	\begin{align*}
		&(z+\mfm(z)-\mfm_l(z))\mcC_{l,e,N}^{(d)}(z)=-\mfm_l(z)\sum_{r\neq l}^d\mcC_{r,e,N}^{(d)}(z)\\
		&-{\rm i}\mbE[e_q]\sum_{s=1}^qt_s\big[\mfa(\tau_s)\mcF_{l,e,N}^{(d)}(z,z_s)+\mfb(\tau_s)\mcF_{l,e,N}^{(d)}(z,\bar{z}_s)\big]+\mrO(C_{\eta_0}N^{-\omega}),
	\end{align*}
	then we obtain that
	\begin{align*}
		\bbTheta_N^{(d)}(z,z)\bbC_{e,N}^{(d)}(z)=-{\rm i}\mbE[e_q]\sum_{s=1}^qt_s\big[\mfa(\tau_s)\bbF_{e,N}^{(d)}(z,z_s)+\mfb(\tau_s)\bbF_{e,N}^{(d)}(z,\bar{z}_s)\big]+\mro(\boldsymbol{1}_d),
	\end{align*}
	where \(\bbTheta_N^{(d)}(z,z)\) is defined in (\ref{Eq of bbPi N}) and it is invertible such that
	$$\lim_{N\to\infty}\Vert\bbTheta_N^{(d)}(z,z)^{-1}+\bbPi^{(d)}(z,z)^{-1}\diag(\mfc^{-1}\circ\bbg(z))\Vert=0.$$
	Hence, we obtain that
	\begin{align*}
		\bbC_{e,N}^{(d)}(z)={\rm i}\mbE[e_q]\sum_{s=1}^qt_s\bbPi^{(d)}(z,z)^{-1}\diag(\mfc^{-1}\circ\bbg(z))\big[\mfa(\tau_s)\bbF_{e,N}^{(d)}(z,z_s)+\mfb(\tau_s)\bbF_{e,N}^{(d)}(z,\bar{z}_s)\big]+\mro(\boldsymbol{1}_d),
	\end{align*}
	According to \S\ref{Subsec of general Major terms}, we can derive the limiting expressions of \(\mcV_{l,e,N}^{(d)}(z,z_s),\mcW_{l,e,N}^{(d)}(z,z_s)\), so  we have
	\begin{align*}
		&\mbE[e_q(\mfa(\tau_s)\gamma(z_s)+\mfb(\tau_s)\gamma(\bar{z}_s))]=\sum_{l=1}^d\mbE[e_q(\mfa(\tau_s)\gamma_l(z_s)+\mfb(\tau_s)\gamma_l(\bar{z}_s))]\\
		&=\sum_{l=1}^d\mfa(\tau_s)\Cov(\gamma_l(z_s),e_q)+\mfb(\tau_s)\Cov(\gamma_l(\bar{z}_s),e_q)\\
		&={\rm i}\mbE[e_q]\sum_{r=1}^qt_r\Big(\mfa(\tau_s)\boldsymbol{1}_d'\bbPi^{(d)}(z_s,z_s)^{-1}\diag(\mfc^{-1}\circ\bbg(z_s))\big[\mfa(\tau_r)\bbF_{e,N}^{(d)}(z_s,z_r)+\mfb(\tau_r)\bbF_{e,N}^{(d)}(z_s,\bar{z}_r)\big]\\
		&+\mfb(\tau_s)\boldsymbol{1}_d'\bbPi^{(d)}(\bar{z}_s,\bar{z}_s)^{-1}\diag(\mfc^{-1}\circ\bbg(\bar{z}_s))\big[\mfa(\tau_r)\bbF_{e,N}^{(d)}(\bar{z}_s,z_r)+\mfb(\tau_r)\bbF_{e,N}^{(d)}(\bar{z}_s,\bar{z}_r)\big]\Big)+\mro(1)\\
		:&={\rm i}\mbE[e_q]\sum_{r=1}^qt_rA_{rs,N}+\mro(1),
	\end{align*}
	which completes our proof.
\end{proof}
\subsubsection{Proof of Theorem \ref{Thm of general CLT LSS}}
	By Theorem \ref{Thm of Extreme eigenvalue N d=3}, it yields that
	\begin{align*}
		G_N(f)\overset{\mbP}{\longrightarrow}G_N(f)1_{\Vert\bbM\Vert\leq\mfv_d+t}=-\frac{1}{2\pi{\rm i}}\oint_{\mfC}f(z)(\tr(\bbQ(z))-Ng(z))dz,
	\end{align*}
	where $t>0$ is a fixed constant. Next, we split \(\mfC:=\mfC(\eta_0):=\mfC^h\cup\mfC^v\) by
	\begin{align*}
		\mfC^h:=\{z=E\pm{\rm i}\eta_0\in\mbC:|E|\leq E_0\},\quad\mfC^v:=\{z=\pm E_0+{\rm i}\eta\in\mbC:|\eta|\leq\eta_0\}.
	\end{align*}
	In other words, \(\mcC\) is a rectangular contour with vertex of \(\pm E_0\pm{\rm i}\eta_0\), where \(E_0>\max\{\mfv_d,\zeta\}+t\) and \(t>0\) is a constant. According to Theorems \ref{Thm of covariance general d}, \ref{Thm of mean general d} and \ref{Thm of CLT general d}, we have shown that $\tr(\bbQ(z))-Ng(z)$ is a Gaussian process with mean of $\mu_N^{(d)}(z)$ and variance of $\mcC_N^{(d)}(z,z)$, which further implies that
	\begin{align*}
		-\frac{1}{2\pi{\rm i}}\oint_{\mfC^h}f(z)(\tr(\bbQ(z))-Ng(z))dz/\sigma_N^{(d)}-\xi_N^{(d)}/\sigma_N^{(d)}\overset{d}{\longrightarrow}\mcN(0,1).
	\end{align*}
	Next, let us  show that
	\begin{align*}
		\lim_{\eta_0\downarrow0}\limsup_{N\to\infty}\oint_{\mfC^v}f(z)(\tr(\bbQ(z))-Ng(z))dz\overset{\mbP}{\longrightarrow}0.
	\end{align*}
	It is enough to show that
	\begin{align*}
		\lim_{\eta_0\downarrow0}\limsup_{N\to\infty}\oint_{\mfC^v}\mbE\big[|f(z)(\tr(\bbQ(z))-Ng(z))|^2\big]dz=0.
	\end{align*}
	Let us  first show that
	\begin{align*}
		\lim_{\eta_0\downarrow0}\limsup_{N\to\infty}\oint_{\mfC^v}\big|f(z)(\mbE[\tr(\bbQ(z))]-Ng(z))\big|^2dz=0.
	\end{align*}
	According to Theorem \ref{Thm of mean general d}, we know that
	\begin{align*}
		\mbE[\tr(\bbQ(z))]-Ng(z)=\mu_N^{(d)}(z)+\mrO(C_tN^{-\omega}).
	\end{align*}
	In fact, by the definition of \(\bbg(z)\) in (\ref{Eq of MDE 3 order}), it is easy to see \(g_i(z)\) are analytic on \(\mfC^v\) for all \(i=1,\cdots,d\), as does the entries \(\bbg(z)\) and \(\bbPi^{(d)}(z,z)\) due to their definitions in (\ref{Eq of MDE 3 order}) and (\ref{Eq of invertible 2}). Moreover, we have shown that \(\bbPi^{(d)}(z,z)\) is invertible in Lemma \ref{Lem of analytic} by (\ref{Eq of Mi general d}). The mean function \(\mu_N^{(d)}(z)\) only depends on \(\bbg(z)\) since \(\bbW^{(d)}(z),\bbV^{(d)}(z,z),G_{i,N}^{(d)}(z)\) and \(H_{i,N}^{(d)}(z)\) also depend on \(\bbg(z)\) by their definitions in (\ref{Eq of bbW general}) (\ref{Eq of bbV}), (\ref{Eq of H2 general d}) and (\ref{Eq of H3 general d}), as does the covariance function \(\mcC_N^{(d)}(z_1,z_2)\) in (\ref{Eq of covariance function general d}), implies that \(\mu_N^{(d)}(z)\) and \(\mcC_N^{(d)}(z_1,z_2)\) are analytic on \(\mfC^v\), thus, combined with the fact \(f\in\mathfrak{F}_d\) is analytic on \(\mfC^v\), we have
	\begin{align*}
		\lim_{\eta_0\downarrow0}\limsup_{N\to\infty}\oint_{\mfC^v}\big|(f(z)\mbE[\tr(\bbQ(z))]-Ng(z))\big|^2dz=0.
	\end{align*}
	According to Theorem \ref{Thm of covariance general d}, we know that \(\Var(\tr(\bbQ(z)))=C_{t,d,\mfc}\) for all \(z\in\mfC^v\) since \(\Re(z)=E_0>\max\{\mfv_d,\zeta\}+t\), then it also conclude that
	\begin{align*}
		\lim_{\eta_0\downarrow0}\limsup_{N\to\infty}\oint_{\mfC^v}|f(z)|^2\Var(\tr\bbQ(z))dz=0,
	\end{align*}
	which completes our proof.

\end{document}